\documentclass[10pt, sigconf, language=english]{acmart}

\usepackage{xcolor}
\usepackage{xspace}
\usepackage[all]{xy}
\usepackage{amsmath}
\usepackage{amsthm}
\usepackage{multicol}
\usepackage{mathpartir}
\usepackage[noabbrev,capitalize]{cleveref}
\usepackage{ifthen}
\usepackage{thm-restate}
\usepackage{hyperref}
\usepackage[normalem]{ulem}
\usepackage{enumitem}
\usepackage{mathbbol}
\newdimen\proofrulebreadth \proofrulebreadth=.05em
\newdimen\proofdotseparation \proofdotseparation=1.25ex
\newdimen\proofrulebaseline \proofrulebaseline=2ex
\newcount\proofdotnumber \proofdotnumber=3
\let\then\relax
\def\hfi{\hskip0pt plus.0001fil}
\mathchardef\squigto="3A3B
\newif\ifinsideprooftree\insideprooftreefalse
\newif\ifonleftofproofrule\onleftofproofrulefalse
\newif\ifproofdots\proofdotsfalse
\newif\ifdoubleproof\doubleprooffalse
\let\wereinproofbit\relax
\newdimen\shortenproofleft
\newdimen\shortenproofright
\newdimen\proofbelowshift
\newbox\proofabove
\newbox\proofbelow
\newbox\proofrulename
\def\shiftproofbelow{\let\next\relax\afterassignment\setshiftproofbelow\dimen0 }
\def\shiftproofbelowneg{\def\next{\multiply\dimen0 by-1 }%
\afterassignment\setshiftproofbelow\dimen0 }
\def\setshiftproofbelow{\next\proofbelowshift=\dimen0 }
\def\setproofrulebreadth{\proofrulebreadth}

\def\prooftree{% NESTED ZERO (\ifonleftofproofrule)
\ifnum  \lastpenalty=1
\then   \unpenalty
\else   \onleftofproofrulefalse
\fi
\ifonleftofproofrule
\else   \ifinsideprooftree
        \then   \hskip.5em plus1fil
        \fi
\fi
\bgroup% NESTED ONE (\proofbelow, \proofrulename, \proofabove,
\setbox\proofbelow=\hbox{}\setbox\proofrulename=\hbox{}%
\let\justifies\proofover\let\leadsto\proofoverdots\let\Justifies\proofoverdbl
\let\using\proofusing\let\[\prooftree
\ifinsideprooftree\let\]\endprooftree\fi
\proofdotsfalse\doubleprooffalse
\let\thickness\setproofrulebreadth
\let\shiftright\shiftproofbelow \let\shift\shiftproofbelow
\let\shiftleft\shiftproofbelowneg
\let\ifwasinsideprooftree\ifinsideprooftree
\insideprooftreetrue
\setbox\proofabove=\hbox\bgroup$\displaystyle % NESTED TWO
\let\wereinproofbit\prooftree
\shortenproofleft=0pt \shortenproofright=0pt \proofbelowshift=0pt
\onleftofproofruletrue\penalty1
}

\def\eproofbit{% NESTED TWO
\ifx    \wereinproofbit\prooftree
\then   \ifcase \lastpenalty
        \then   \shortenproofright=0pt  % 0: some other object, no indentation
        \or     \unpenalty\hfil         % 1: empty hypotheses, just glue
        \or     \unpenalty\unskip       % 2: just had a tree, remove glue
        \else   \shortenproofright=0pt  % eh?
        \fi
\fi
\global\dimen0=\shortenproofleft
\global\dimen1=\shortenproofright
\global\dimen2=\proofrulebreadth
\global\dimen3=\proofbelowshift
\global\dimen4=\proofdotseparation
\global\count255=\proofdotnumber
$\egroup  % NESTED ONE
\shortenproofleft=\dimen0
\shortenproofright=\dimen1
\proofrulebreadth=\dimen2
\proofbelowshift=\dimen3
\proofdotseparation=\dimen4
\proofdotnumber=\count255
}

\def\proofover{% NESTED TWO
\eproofbit % NESTED ONE
\setbox\proofbelow=\hbox\bgroup % NESTED TWO
\let\wereinproofbit\proofover
$\displaystyle
}%
\def\proofoverdbl{% NESTED TWO
\eproofbit % NESTED ONE
\doubleprooftrue
\setbox\proofbelow=\hbox\bgroup % NESTED TWO
\let\wereinproofbit\proofoverdbl
$\displaystyle
}%
\def\proofoverdots{% NESTED TWO
\eproofbit % NESTED ONE
\proofdotstrue
\setbox\proofbelow=\hbox\bgroup % NESTED TWO
\let\wereinproofbit\proofoverdots
$\displaystyle
}%
\def\proofusing{% NESTED TWO
\eproofbit % NESTED ONE
\setbox\proofrulename=\hbox\bgroup % NESTED TWO
\let\wereinproofbit\proofusing
\kern0.3em$
}

\def\endprooftree{% NESTED TWO
\eproofbit % NESTED ONE
  \dimen5 =0pt% spread of hypotheses
\dimen0=\wd\proofabove \advance\dimen0-\shortenproofleft
\advance\dimen0-\shortenproofright
\dimen1=.5\dimen0 \advance\dimen1-.5\wd\proofbelow
\dimen4=\dimen1
\advance\dimen1\proofbelowshift \advance\dimen4-\proofbelowshift
\ifdim  \dimen1<0pt
\then   \advance\shortenproofleft\dimen1
        \advance\dimen0-\dimen1
        \dimen1=0pt
        \ifdim  \shortenproofleft<0pt
        \then   \setbox\proofabove=\hbox{%
                        \kern-\shortenproofleft\unhbox\proofabove}%
                \shortenproofleft=0pt
        \fi
\fi
\ifdim  \dimen4<0pt
\then   \advance\shortenproofright\dimen4
        \advance\dimen0-\dimen4
        \dimen4=0pt
\fi
\ifdim  \shortenproofright<\wd\proofrulename
\then   \shortenproofright=\wd\proofrulename
\fi
\dimen2=\shortenproofleft \advance\dimen2 by\dimen1
\dimen3=\shortenproofright\advance\dimen3 by\dimen4
\ifproofdots
\then
        \dimen6=\shortenproofleft \advance\dimen6 .5\dimen0
        \setbox1=\vbox to\proofdotseparation{\vss\hbox{$\cdot$}\vss}%
        \setbox0=\hbox{%
                \advance\dimen6-.5\wd1
                \kern\dimen6
                $\vcenter to\proofdotnumber\proofdotseparation
                        {\leaders\box1\vfill}$%
                \unhbox\proofrulename}%
\else   \dimen6=\fontdimen22\the\textfont2 % height of maths axis
        \dimen7=\dimen6
        \advance\dimen6by.5\proofrulebreadth
        \advance\dimen7by-.5\proofrulebreadth
        \setbox0=\hbox{%
                \kern\shortenproofleft
                \ifdoubleproof
                \then   \hbox to\dimen0{%
                        $\mathsurround0pt\mathord=\mkern-6mu%
                        \cleaders\hbox{$\mkern-2mu=\mkern-2mu$}\hfill
                        \mkern-6mu\mathord=$}%
                \else   \vrule height\dimen6 depth-\dimen7 width\dimen0
                \fi
                \unhbox\proofrulename}%
        \ht0=\dimen6 \dp0=-\dimen7
\fi
\let\doll\relax
\ifwasinsideprooftree
\then   \let\VBOX\vbox
\else   \ifmmode\else$\let\doll=$\fi
        \let\VBOX\vcenter
\fi
\VBOX   {\baselineskip\proofrulebaseline \lineskip.2ex
        \expandafter\lineskiplimit\ifproofdots0ex\else-0.6ex\fi
        \hbox   spread\dimen5   {\hfi\unhbox\proofabove\hfi}%
        \hbox{\box0}%
        \hbox   {\kern\dimen2 \box\proofbelow}}\doll%
\global\dimen2=\dimen2
\global\dimen3=\dimen3
\egroup % NESTED ZERO
\ifonleftofproofrule
\then   \shortenproofleft=\dimen2
\fi
\shortenproofright=\dimen3
\onleftofproofrulefalse
\ifinsideprooftree
\then   \hskip.5em plus 1fil \penalty2
\fi
}

\colorlet{darkgreen}{green!60!black}
\colorlet{DARKGREEN}{green!60!black}

\newcommand{\lsvsym}{\mathsf{lsv}}

\newcommand{\cf}{\emph{cf.}\xspace}

\newcommand{\CBN}{CBN\xspace}
\newcommand{\CBV}{CBV\xspace}
\newcommand{\stabilized}{rigid\xspace}
\newcommand{\nonuseful}{linear\xspace}
\newcommand{\Nonuseful}{Linear\xspace}
\newcommand{\abstractionframe}{abstraction frame\xspace}
\newcommand{\princo}{{\tt UP1}\xspace}
\newcommand{\princd}{{\tt UP2}\xspace}

\newcommand{\emptyPremise}{\vphantom{{}^@}}
\newcommand{\indrulename}[1]{\textup{\textsc{#1}}}
\newcommand{\indrule}[3]{
\ensuremath{
\begin{array}{c}
  \prooftree #2
    \justifies #3
    \thickness=0.05em
    \using \indrulename{#1}
  \endprooftree
\end{array}}}

\newcommand{\indruleNPos}[4]{
\begin{array}[#1]{c@{}r}
\hspace{-.2cm}
 #3
\hspace{-.2cm}
\vspace{-.1cm}
\\
& \,#2\!\hspace{-.5cm}\vspace{-.2cm} \\
\cline{1-1}\vspace{-.3cm} \\
  #4 \hspace{.5cm}\,
\end{array}
}
\newcommand{\indruleN}[3]{\indruleNPos{b}{#1}{#2}{#3}}

\newcommand{\deriv}{\Phi}

\newcommand{\derivs}[2]{#1 \rhd #2}

\renewcommand{\theenumi}{\arabic{enumi}}
\renewcommand{\theenumii}{\arabic{enumii}}
\renewcommand{\theenumiii}{\arabic{enumiii}}
\renewcommand{\theenumiv}{\arabic{enumiv}}
\makeatletter
\renewcommand\p@enumii{\theenumi.}
\renewcommand\p@enumiii{\theenumi.\theenumii.}
\renewcommand\p@enumiv{\theenumi.\theenumii.\theenumiii.}
\makeatother

\AtEndPreamble{%
  \theoremstyle{acmdefinition}
  \newtheorem{remark}[theorem]{Remark}
}

\newcommand{\arUd}[4]{\ar[d]^-{\Usym}_-{#1,#2,#3,#4}}
\newcommand{\arUr}[4]{\ar[r]^-{\Usym}_-{#1,#2,#3,#4}}
\newcommand{\arUrStable}[4]{\ar[r]^-{\footnotesize\blacktriangle}_-{#1,#2,#3,#4}}
\newcommand{\arsdUd}[4]{\ar@{.>}[d]^-{\Usym}_-{#1,#2,#3,#4}}
\newcommand{\arsdUr}[4]{\ar@{.>}[r]^-{\Usym}_-{#1,#2,#3,#4}}
\newcommand{\arsdUrStable}[4]{\ar@{.>}[r]^-{\footnotesize\blacktriangle}_-{#1,#2,#3,#4}}
\newcommand{\arVd}[1]{\ar[d]^-{\Vsym}_-{#1}}
\newcommand{\arVdd}[1]{\ar[dd]^-{\Vsym}_-{#1}}
\newcommand{\arVr}[1]{\ar[r]^-{\Vsym}_-{#1}}
\newcommand{\arsdVd}[1]{\ar@{.>}[d]^-{\Vsym}_-{#1}}
\newcommand{\arsdVr}[1]{\ar@{.>}[r]^-{\Vsym}_-{#1}}
\newcommand{\arsdVdn}[1]{\ar@{.>}[d]^-{\Vsym}^>>{*}_-{#1}}
\newcommand{\arsdVrn}[1]{\ar@{.>}[r]^-{\Vsym}^>>{*}_-{#1}}

\newcommand{\defn}[1]{{\bf #1}}
\newcommand{\eg}{{\em e.g.}\xspace}
\newcommand{\ie}{{\em i.e.}\xspace}
\newcommand{\ih}{{\it i.h.}\xspace}

\newcommand{\ST}{\ |\ }
\newcommand{\HS}{\hspace{.5cm}}
\newcommand{\sep}{\hspace{.3cm}}

\renewcommand{\emptyset}{\varnothing}

\newcommand{\eqdef}{:=}
\newcommand{\defeq}{\mathrel{:=}}
\newcommand{\eqgram}{\mathrel{::=}}

\newcommand{\set}[1]{\{#1\}}
\renewcommand{\implies}{\Rightarrow}

\newcommand{\lam}[2]{\lambda#1.\,#2}
\newcommand{\sub}[2]{\{#1:=#2\}}
\newcommand{\esub}[2]{[#1\backslash#2]}

\newcommand{\Vsym}{\circ}
\newcommand{\Usym}{\bullet}

\newcommand{\tovv}[1]{\xrightarrow{\raisebox{-0.75ex}[0ex][0ex]{\footnotesize $\Vsym$}}_{#1}}
\newcommand{\tovvn}[1]{\xrightarrow{\raisebox{-0.75ex}[0ex][0ex]{\footnotesize $\Vsym$}}^{*}_{#1}}
\newcommand{\tovveq}[1]{\xrightarrow{\raisebox{-0.75ex}[0ex][0ex]{\footnotesize $\Vsym$}}^{=}_{#1}}
\newcommand{\tovvinv}[1]{\xleftarrow{\raisebox{-0.75ex}[0ex][0ex]{\footnotesize $\Vsym$}}_{#1}}
\newcommand{\tou}{\xrightarrow{\raisebox{-0.75ex}[0ex][0ex]{\footnotesize $\Usym$}}}
\newcommand{\tov}[4]{\xrightarrow{\raisebox{-0.75ex}[0ex][0ex]{\footnotesize $\Usym$}}_{#1,#2,#3,#4}}
\newcommand{\tovinv}[4]{\xleftarrow{\raisebox{-0.75ex}[0ex][0ex]{\footnotesize $\Usym$}}_{#1,#2,#3,#4}}

\newcommand{\tovvTop}{\xrightarrow{\raisebox{-0.75ex}[0ex][0ex]{\footnotesize $\Vsym$}}_{\mathtt{top}}}
\newcommand{\tovTop}[1]{\xrightarrow{\raisebox{-0.75ex}[0ex][0ex]{\footnotesize $\Usym$}}_{\mathtt{top},#1}}

\newcommand{\dbsym}{\mathsf{db}}
\newcommand{\svsym}{\mathsf{sv}}

\newcommand{\inv}[3]{\mathsf{inv}{(#1,#2,#3)}}

\newcommand{\var}{x}
\newcommand{\vartwo}{y}
\newcommand{\varthree}{z}
\newcommand{\varfour}{w}

\newcommand{\tm}{t}
\newcommand{\tmtwo}{u}
\newcommand{\tmthree}{s}
\newcommand{\tmfour}{r}
\newcommand{\tmfive}{p}
\newcommand{\tmsix}{q}

\newcommand{\val}{v}
\newcommand{\valtwo}{w}

\newcommand{\ctxhole}{\Diamond}

\newcommand{\ctxof}[2]{#1\langle#2\rangle}

\newcommand{\gctx}{\mathtt{C}}
\newcommand{\sctx}{\mathtt{L}}
\newcommand{\sctxtwo}{\mathtt{L'}}

\newcommand{\domSctx}[1]{\mathtt{dom}(#1)}

\newcommand{\ruleHVarVar}{\indrulename{hv-var}}
\newcommand{\ruleHVarSubi}{\indrulename{hv-sub$_1$}}
\newcommand{\ruleHVarSubii}{\indrulename{hv-sub$_2$}}
\newcommand{\ruleHAbsVar}{\indrulename{h-var}\xspace}
\newcommand{\ruleHAbsLam}{\indrulename{h-lam}\xspace}
\newcommand{\ruleHAbsSubi}{\indrulename{h-sub$_1$}\xspace}
\newcommand{\ruleHAbsSubii}{\indrulename{h-sub$_2$}\xspace}
\newcommand{\ruleStructVar}{\indrulename{s-var}\xspace}
\newcommand{\ruleStructApp}{\indrulename{s-app}\xspace}
\newcommand{\ruleStructSubi}{\indrulename{s-sub$_1$}\xspace}
\newcommand{\ruleStructSubii}{\indrulename{s-sub$_2$}\xspace}
\newcommand{\ruleStableVar}{\indrulename{stable-var}}
\newcommand{\ruleStableAbs}{\indrulename{stable-abs}}
\newcommand{\ruleStableApp}{\indrulename{stable-app}}
\newcommand{\ruleStableESHAbs}{\indrulename{stable-es-habs}}
\newcommand{\ruleStableESStruct}{\indrulename{stable-es-struct}}
\newcommand{\ruleStableCtxEmpty}{\indrulename{stableCtx-empty}}
\newcommand{\ruleStableCtxHAbs}{\indrulename{stableCtx-habs}}
\newcommand{\ruleStableCtxStruct}{\indrulename{stableCtx-struct}}
\newcommand{\ruleVNFVar}{\indrulename{NF-var$^\Vsym$}\xspace}
\newcommand{\ruleVNFLam}{\indrulename{NF-lam$^\Vsym$}\xspace}
\newcommand{\ruleVNFApp}{\indrulename{NF-app$^\Vsym$}\xspace}
\newcommand{\ruleVNFEsVal}{\indrulename{NF-esVal$^\Vsym$}\xspace}
\newcommand{\ruleVNFEsNonVal}{\indrulename{NF-esNonVal$^\Vsym$}\xspace}
\newcommand{\ruleCtxNFEmpty}{\indrulename{Ctx-NF-empty}}
\newcommand{\ruleCtxNFAddVal}{\indrulename{Ctx-NF-addVal}}
\newcommand{\ruleCtxNFAddNonVal}{\indrulename{Ctx-NF-addNonVal}}
\newcommand{\ruleUNFVar}{\indrulename{NF-var$^{\Usym}$}\xspace}
\newcommand{\ruleUNFLam}{\indrulename{NF-lam$^{\Usym}$}\xspace}
\newcommand{\ruleUNFApp}{\indrulename{NF-app$^{\Usym}$}\xspace}
\newcommand{\ruleUNFEsAbs}{\indrulename{NF-esA$^{\Usym}$}\xspace}
\newcommand{\ruleUNFEsStruct}{\indrulename{NF-esS$^{\Usym}$}\xspace}
\newcommand{\ruleVDb}{\indrulename{db$^\Vsym$}\xspace}
\newcommand{\ruleVSub}{\indrulename{sub$^\Vsym$}\xspace}
\newcommand{\ruleVLsv}{\indrulename{lsv$^\Vsym$}\xspace}
\newcommand{\ruleVAppL}{\indrulename{appL$^\Vsym$}\xspace}
\newcommand{\ruleVAppR}{\indrulename{appR$^\Vsym$}\xspace}
\newcommand{\ruleVEsL}{\indrulename{esL$^\Vsym$}\xspace}
\newcommand{\ruleVEsR}{\indrulename{esR$^\Vsym$}\xspace}
\newcommand{\ruleUDb}{\indrulename{db$^{\Usym}$}\xspace}
\newcommand{\ruleUSub}{\indrulename{sub$^{\Usym}$}\xspace}
\newcommand{\ruleULsv}{\indrulename{lsv$^{\Usym}$}\xspace}
\newcommand{\ruleUAppL}{\indrulename{appL$^{\Usym}$}\xspace}
\newcommand{\ruleUAppR}{\indrulename{appR$^{\Usym}$}\xspace}
\newcommand{\ruleUEsR}{\indrulename{esR$^{\Usym}$}\xspace}
\newcommand{\ruleUEsLAbs}{\indrulename{esLA$^{\Usym}$}\xspace}
\newcommand{\ruleUEsLStruct}{\indrulename{esLS$^{\Usym}$}\xspace}
\newcommand{\ruleUDbStable}{\indrulename{db-stable$^{\Usym}$}}
\newcommand{\ruleEquivRefl}{\indrulename{refl}}
\newcommand{\ruleEquivSym}{\indrulename{sym}}
\newcommand{\ruleEquivTrans}{\indrulename{trans}}
\newcommand{\ruleEquivCongApp}{\indrulename{cong-app}}
\newcommand{\ruleEquivCongES}{\indrulename{cong-es}}
\newcommand{\ruleEquivEsLDist}{\indrulename{es-l-dist}}
\newcommand{\ruleEquivEsRDist}{\indrulename{es-r-dist}}
\newcommand{\ruleEquivEsComm}{\indrulename{es-comm}}
\newcommand{\ruleEquivEsAssoc}{\indrulename{es-assoc}}
\newcommand{\ruleTypVar}{\indrulename{var}\xspace}
\newcommand{\ruleTypAbs}{\indrulename{abs}\xspace}
\newcommand{\ruleTypAppP}{\indrulename{appP}\xspace}
\newcommand{\ruleTypAppC}{\indrulename{appC}\xspace}
\newcommand{\ruleTypES}{\indrulename{es}\xspace}
\newcommand{\ruleTypSctxEmpty}{\indrulename{emptySubsCtx}\xspace}
\newcommand{\ruleTypSctxAdd}{\indrulename{addSubsCtx}\xspace}

\newcommand{\rulename}{\rho}
\newcommand{\RulesV}[1]{\mathcal{R}_{#1}}
\newcommand{\vset}{\mathcal{V}}
\newcommand{\aset}{\mathcal{A}}
\newcommand{\asettwo}{\mathcal{B}}
\newcommand{\sset}{\mathcal{S}}
\newcommand{\ssettwo}{\mathcal{T}}
\newcommand{\appflag}{\mu}
\newcommand{\app}{@}
\newcommand{\nonapp}{{\text{$\not$\!@}}}
\newcommand{\HVar}[1]{\mathsf{HVar}_{#1}}
\newcommand{\HAbs}[1]{\mathsf{HA}_{#1}}
\newcommand{\Struct}[1]{\mathsf{St}_{#1}}
\newcommand{\Stable}[2]{\mathsf{Stable}_{#1,#2}}
\newcommand{\StableCtx}[2]{\mathsf{StableCtx}_{#1,#2}}
\newcommand{\VNF}[2]{\mathsf{NF}^{\Vsym}_{#1,#2}}
\newcommand{\VRed}[1]{\mathsf{Red}^{\Vsym}_{#1}}
\newcommand{\VIrred}[1]{\mathsf{Irred}^{\Vsym}_{#1}}
\newcommand{\NF}[3]{\mathsf{NF}^{\Usym}_{#1,#2,#3}}
\newcommand{\Red}[3]{\mathsf{Red}^{\Usym}_{#1,#2,#3}}
\newcommand{\Irred}[3]{\mathsf{Irred}^{\Usym}_{#1,#2,#3}}
\newcommand{\ruledb}{\mathsf{db}}
\newcommand{\rulelsv}{\mathsf{lsv}}
\newcommand{\rulesub}[2]{\mathsf{sub}_{(#1,#2)}}
\newcommand{\expansion}[2]{#1^{#2}}
\newcommand{\CtxNF}[1]{\mathsf{CtxNF}^{\Vsym}_{#1}}

\newcommand{\fv}[1]{\mathsf{fv}(#1)}

\newcommand{\rv}[1]{\mathsf{rv}(#1)}

\newcommand{\evalas}{\cdot}
\newcommand{\valas}{\sigma}
\newcommand{\tovalas}[1][\valas]{\to_{#1}}
\newcommand{\tonvalas}[1][\valas]{\to^{*}_{#1}}
\newcommand{\unvalas}[2][\valas]{{#2}^{\downarrow#1}}
\newcommand{\unvalasevalas}[1]{{#1}^{\downarrow}}

\newcommand{\compat}[3]{\mathsf{compatible}(#1, #2, #3)}

\newcommand{\meas}[1]{\#(#1)}
\newcommand{\measPhi}[2]{\#^{#2}(#1)}
\newcommand{\measvar}[2]{\#_{#1}(#2)}
\newcommand{\measvarPhi}[3]{\#^{#3}_{#1}(#2)}
\newcommand{\measvalas}[2][\valas]{\#^{#1}(#2)}

\newcommand{\toustable}{\xrightarrow{\raisebox{-0.75ex}{{\footnotesize$\blacktriangle$}}}}
\newcommand{\tostable}[4]{\toustable_{#1,#2,#3,#4}}
\newcommand{\toustablen}[1]{\toustable^{\raisebox{-.1cm}{{\footnotesize{$#1$}}}}}
\newcommand{\disj}{\mathrel{\#}}
\newcommand{\none}{\bot}

\newcommand{\mset}[1]{[#1]}
\newcommand{\emset}{\mset{}}
\newcommand{\iI}{{i \in I}}
\newcommand{\jJ}{{j \in J}}
\newcommand{\kK}{{k \in K}}

\newcommand{\typ}{\tau}
\newcommand{\typtwo}{\sigma}

\newcommand{\nityp}{\mathcal{I}}
\newcommand{\nityptwo}{\mathcal{J}}
\newcommand{\niunion}{\mathrel{\uplus}}

\newcommand{\mleq}{\lhd}
\newcommand{\mtyp}{\mathcal{M}}
\newcommand{\mtyptwo}{\mathcal{N}}
\newcommand{\mtypthree}{\mathcal{L}}
\newcommand{\mtypfour}{\mathcal{O}}
\newcommand{\optmtyp}{\mtyp^?}
\newcommand{\optmtyptwo}{\mtyptwo^?}
\newcommand{\optmtypthree}{\mtypthree^?}

\newcommand{\dom}[1]{\mathsf{dom}(#1)}
\newcommand{\im}[1]{\mathsf{im}(#1)}
\newcommand{\emptyctx}{\emptyset}
\newcommand{\tctx}{\Gamma}
\newcommand{\tctxtwo}{\Delta}
\newcommand{\tctxthree}{\Sigma}
\newcommand{\tctxfour}{\Theta}

\newcommand{\cm}{m}
\newcommand{\ce}{e}
\newcommand{\numarrows}[1]{{\tt ta}(#1)}

\newcommand{\judg}[4][]{#2\vdash^{(#1)}#3:#4}
\newcommand{\judgs}[3]{#1\vdash #2:#3}
\newcommand{\judgSctx}[4][]{#2\Vdash^{(#1)}#3\rhd#4}

\newcommand{\TEnv}[3]{\mathsf{TEnv}(#1,#2,#3)}
\newcommand{\isAppr}[2]{\mathtt{appropriate}_{#1}(#2)}

\newcommand{\abs}[1]{#1 \in \mathsf{Abs}}
\newcommand{\valPred}[1]{#1 \in \mathsf{Val}}
\newcommand{\purePred}[1]{#1 \in \mathsf{Pure}}

\newcommand{\betafireball}{\beta_f}

\newcommand{\id}{{\tt I}}

\newcommand{\positionalFlag}{positional flag\xspace}

\newcommand{\typesystem}{\mathcal{U}}
\newcommand{\relation}{\mathcal{R}}
\newcommand{\rewrite}[1]{\rightarrow_{#1}}
\newcommand{\rewriten}[1]{\rightarrow^*_{#1}}
\newcommand{\rewriteeq}[1]{\rightarrow^=_{#1}}
\newcommand{\rewritenpasos}[2]{\rightarrow^{#2}_{#1}}

\newcommand{\equivC}{\equiv_{c}}

\newcommand{\glamour}{GLAMoUr\xspace}
\newcommand{\usefmaca}[4]{#2 & #3 & #4 & #1}

\newcommand{\admsym}{{\mathtt c}}
\newcommand{\esym}{{\mathtt e}}
\newcommand{\msym}{{\mathtt m}}
\newcommand{\usym}{{\mathtt u}}
\newcommand{\ctxholep}[1]{\langle #1\rangle}
\newcommand{\ectx}{\ctxhole}

\newcommand{\evctx}{\genevctx}

\newcommand{\pair}[2]{(#1,#2)}

\newcommand{\genevctx}{F}

\newcommand{\glamourst}[4]{(#1,#2,#3,#4)}
\renewcommand{\state}{s}
\newcommand{\statetwo}{s'}
\newcommand{\statei}{s_0}
\newcommand{\stater}{s_r}

\newcommand{\genv}{E}
\newcommand{\genvtwo}{E'}
\newcommand{\decgenv}{\decode{\genv}}
\renewcommand{\dump}{D}
\newcommand{\dumptwo}{D'}
\newcommand{\estack}{\epsilon}
\newcommand{\cons}{:}
\newcommand{\stack}{\pi}
\newcommand{\stacktwo}{\pi'}

\newcommand{\stackitem}{\phi}
\newcommand{\stackitemtwo}{\psi}
\newcommand{\stackiteml}[1][\lab]{\stackitem^{#1}}
\newcommand{\stackitemtwol}[1][\lab]{\stackitemtwo^{#1}}

\newcommand{\tocode}[1]{\mathbf{#1}}
\newcommand{\code}{\mathbf{\tm}}
\newcommand{\codetwo}{\mathbf{\tmtwo}}
\newcommand{\codethree}{\mathbf{\tmthree}}

\newcommand{\tomachhole}[1]{\leadsto_{#1}}
\newcommand{\tomachcone}{\tomachhole{\admsym_1}}
\newcommand{\tomachum}{\tomachhole{\usym\msym}}
\newcommand{\tomachctwo}{\tomachhole{\admsym_2}}
\newcommand{\tomachcthree}{\tomachhole{\admsym_3}}
\newcommand{\tomachcfour}{\tomachhole{\admsym_4}}
\newcommand{\tomachcfive}{\tomachhole{\admsym_5}}
\newcommand{\tomachue}{\tomachhole{\usym\esym}}
\newcommand{\tomachume}{\tomachhole{\usym\msym,\usym\esym}}

\newcommand{\lab}{l}
\newcommand{\alive}{\mathbb{A}}
\newcommand{\dead}{\mathbb{S}}
\newcommand{\herval}[1]{#1^\alive}
\newcommand{\rename}[1]{#1^\alpha}
\newcommand{\decode}[1]{\{\!\! \{ #1 \} \! \! \}}
\newcommand{\decodep}[2]{\decode{#1}\ctxholep{#2}}

\newcommand{\tobetafireball}{\to_{\betafireball}}
\newcommand{\unfold}[1]{#1^\Downarrow}

\newcommand{\iin}{\! \in \!}
\newcommand{\notiin}{\! \notin \!}
\newcommand{\VSC}{\textsc{vsc}\xspace}

\newcommand{\LOCBV}{\ensuremath{\textup{\textsc{lcbv}}^\Vsym}\xspace}
\newcommand{\UOCBV}{\ensuremath{\textup{\textsc{ucbv}}^\Usym}\xspace}

\newcommand{\tight}{\mathbb{t}}
\newcommand{\tightN}{\,\mathbb{s}}
\renewcommand{\typ}{\alpha}
\renewcommand{\typtwo}{\beta}

\renewcommand{\nityp}{\mathcal{M}}
\renewcommand{\nityptwo}{\mathcal{N}}
\renewcommand{\mtyp}{\mathcal{T}}
\renewcommand{\mtyptwo}{\mathcal{S}}
\renewcommand{\mtypthree}{\mathcal{R}}
\renewcommand{\mtypfour}{\mathcal{Q}}

\newcommand{\hiddenproof}[2]{\begin{proof}#1\end{proof}}
\newcommand{\HIDDENFRAGMENT}[2]{#1}

\allowdisplaybreaks[1]

\AtBeginDocument{%
  }

\begin{document}

\title{Useful Call-by-Value: Syntax and Semantics \\ (Technical Report)}

% \author{Anonymous authors}
\author{Pablo Barenbaum}
\email{pbarenbaum@dc.uba.ar}
\affiliation{%
  \institution{Universidad Nacional de Quilmes (CONICET), and Instituto de Ciencias de la Computación, UBA}
  \country{Argentina}
}

\author{Delia Kesner}
\email{kesner@irif.fr}
\affiliation{%
  \institution{Université de Paris, CNRS, IRIF}
  \country{France}
}

\author{Mariana Milicich}
\email{milicich@irif.fr}
\affiliation{%
  \institution{Université Paris Cité, CNRS, IRIF}
  \country{France}
}% used to denote shared contribution to the research.

\renewcommand{\shortauthors}{Barenbaum, Kesner, and Milicich}

%\input{abstract}
\onecolumn

\begin{abstract}
Useful evaluation, introduced by Accattoli and Dal Lago, is an
optimised evaluation mechanism for functional programming languages.
It relies on representing programs with sharing and imposing a 
restricted notion of \emph{useful substitutions}, that intuitively 
disallows copying subterms that do not contribute to the progress of 
the computation.

Initially defined in the framework of call-by-name, 
\emph{useful evaluation} has since been extended to call-by-value 
(\CBV), where it is shown that useful \CBV is an optimisation of 
standard \CBV that preserves its original semantics.
This preservation result has been shown by means of syntactical rewriting
techniques, requiring ad-hoc proofs.
Consequently, the applicability of such optimisations to other 
models of computation remains limited.

This work provides the first inductive definition of useful \CBV 
evaluation. 
For that, we first restrict the substitution operation in the Value 
Substitution Calculus to be linear, yielding the \LOCBV strategy. 
We then further restrict substitution in \LOCBV, so that substitution 
contributes to the progress of the computation.
This optimisation is the \UOCBV strategy, and its notion of 
substitution is sensitive to the surrounding evaluation context, so 
it is non-trivial to capture it inductively. 
Moreover, we show that \UOCBV is a sound and complete implementation 
of \LOCBV, optimised to implement useful evaluation.
As a further contribution, we show that an existing notion of usefulness
in the literature, namely the \glamour abstract machine, implements
the \UOCBV strategy with polynomial overhead in time. 
This establishes that \UOCBV is time-invariant, \ie, that the number 
of reduction steps to normal form in \UOCBV can be used as a measure 
of time complexity.

Defining \UOCBV leads us to the first semantic model of useful \CBV 
evaluation through system $\typesystem$, a non-idempotent intersection 
type system.
Our main result is a characterisation of termination for useful \CBV 
evaluation via system $\typesystem$:
a term is typable in system $\typesystem$ if and only if it terminates 
in \UOCBV. 
Additionally, system $\typesystem$ provides a quantitative 
interpretation for \UOCBV, offering exact step-count information for 
program evaluation. 
Even though the specification of the operational semantics of 
\UOCBV is highly complex, system $\typesystem$ is notably simple. 
As far as we know, system $\typesystem$ is one of the scarce 
quantitative type systems capturing exactly the substitution step-count 
for a call-by-value strategy.
\end{abstract}

%    <concept>
%        <concept_id>10003752.10003753.10003754.10003733</concept_id>
%        <concept_desc>Theory of computation~Lambda calculus</concept_desc>
%        <concept_significance>500</concept_significance>
%        </concept>
%    <concept>
%        <concept_id>10003752.10003790.10011740</concept_id>
%        <concept_desc>Theory of computation~Type theory</concept_desc>
%        <concept_significance>500</concept_significance>
%        </concept>
%    <concept>
%        <concept_id>10003752.10010124.10010131.10010134</concept_id>
%        <concept_desc>Theory of computation~Operational semantics</concept_desc>
%        <concept_significance>500</concept_significance>
%        </concept>
%  </ccs2012>

% \ccsdesc[500]{Theory of computation~Lambda calculus}
% \ccsdesc[500]{Theory of computation~Type theory}
% \ccsdesc[500]{Theory of computation~Operational semantics}

% \keywords{
%   lambda calculus, evaluation strategies, call-by-value, useful evaluation, 
%   intersection types,
%   quantitative models
% }


% To hide proofs : 
\newcommand{\hinput}[1]{\input{#1}}
% To show proofs : \newcommand{\hinput}[1]{\input{#1}}
% \newcommand{\hinput}[1]{\input{#1}}


\maketitle

\setcounter{tocdepth}{2}

\section{Introduction}
\label{sec:introduction}

The $\lambda$-calculus is the foundational model behind functional 
programming languages and most proof assistants. 
Implementing them \emph{efficiently} requires designing an evaluation 
mechanism that closely adheres to the operational semantics, while 
optimising the key operations for resource efficiency, both for time 
and space. 
This creates a significant gap between the original model and its 
optimised version, making it crucial to guarantee that both notions 
have the same \emph{observable behaviour}.

In the $\lambda$-calculus, evaluation strategies vary, with 
\emph{call-by-name} (\CBN) and \emph{call-by-value} (\CBV) being two 
prominent examples. 
In both cases, efficient implementations rely on term representations 
with \emph{sharing}. 
For example, implementations typically use environments to bind a 
variable to a shared subexpression, rather than textually substituting 
the variable by many copies of the subexpression. 
Formally, $\lambda$-terms with sharing can be represented with 
\emph{explicit substitutions} 
(ESs)~\footnote{See~\cite{Kesner2009} for a survey on ESs.}. 
An ES binding a variable $\var$ to an expression $\tmtwo$ is written
$\esub\var\tmtwo$, and a term $\tm$ affected by such ES is written 
$\tm\esub\var\tmtwo$ (which some authors note 
``$\mathsf{let}\ \var = \tmtwo\ \mathsf{in}\ \tm$''). 
The expression $\tm\esub\var\tmtwo$ means that all the free occurrences
of $\var$ in $\tm$ are bound to $\tmtwo$, and thus $\tmtwo$ is 
\emph{shared}. 
While the $\beta$-reduction rule of the $\lambda$-calculus is based 
on the meta-level substitution operation $\tm\sub\var\tmtwo$ that 
replaces all the free occurrences of $\var$ in $\tm$ at once, some 
calculi with ESs allow replacing a single occurrence of a variable, 
such as $(\var \, \var)\esub\var\tm \to (\tm \, \var)\esub\var\tm$. 
This finer-grained substitution operation is called \defn{linear substitution}.

Accattoli and Dal Lago~\cite{AccattoliL14} have proposed an optimised 
evaluation mechanism, called \defn{useful evaluation}, which 
represents terms with ESs so that copy of shared subterms is 
restricted to avoid \emph{size explosion}.
This mechanism relies on two key \emph{usefulness principles} 
---that we call \princo and \princd---. 
Specifically, these two principles ensure that an occurrence of a 
variable may be linearly substituted by an expression only if this 
contributes to creating a \emph{redex}\footnote{\textbf{Red}ucible \textbf{ex}pression.}.
Principle \princo (\emph{Sharing Structures}) ensures that terms 
headed by a free variable, or more precisely \emph{structures}, such
as $\var\,\vartwo$, always remain shared in ESs, and are never
copied. 
Principle \princd (\emph{Substituting Abstractions for Progress}) 
ensures that an occurrence of a variable bound to a 
$\lambda$-abstraction may be substituted only if the variable is
applied to an argument. 
See \cref{sec:preliminary_notions} for more details.

Useful evaluation has been extended to 
\emph{open \CBV}~\cite{AccattoliC15, AccattoliG17}, \ie, \CBV over 
terms that may contain free variables.
Unlike typical strategies that are defined by context closure rules
that are \emph{local} in nature, useful evaluation is context-sensitive,
in the sense that steps involve side conditions of a \emph{global} nature.
This dependency complicates inductive reasoning, as non-useful steps 
can become useful when placed within specific contexts. 
This makes it challenging to define useful evaluation by means of 
inductive rules, as we discuss in the following sections.

\paragraph{\bf Quantitative Interpretations.}
The model of useful evaluation we define in this work is based on the
\emph{non-idempotent} flavour of \emph{intersection types} (IT). 
IT extend simple types with an intersection type constructor $\cap$,
allowing a program $\tm$ to be typed with $\typ \cap \typtwo$ if 
$\tm$ types with both $\typ$ and $\typtwo$ independently. 
Initially introduced as \emph{models} to capture \emph{qualitative} 
computational properties of functional programs~\cite{CoppoDezani78}, 
IT systems are able in particular to characterise termination of 
evaluation strategies: a program $\tm$ terminates in a given 
evaluation strategy if and only if $\tm$ is typable in an associated 
type system.

In the original formulation of intersection types, the type
constructor $\cap$ is considered to be \defn{idempotent}, meaning
that $\sigma \cap \sigma = \sigma$. 
So, an intersection $\typ_1 \cap \hdots \cap \typ_n$ can be 
represented as the \emph{set} $\set{\typ_1,\hdots,\typ_n}$. 
More recent works have adopted a \defn{non-idempotent} intersection 
type constructor~\cite{Gardner94,deCarvalho2007}. 
Here, an intersection $\typ_1 \cap \hdots \cap \typ_n$ can be 
represented as the \emph{multiset} $\mset{\typ_1,\hdots,\typ_n}$. 
Like their idempotent precursors, non-idempotent IT still allow 
characterising operational properties of programs by means of 
typability~\cite{Gardner94,deCarvalho2007} but also yield a 
substantial improvement: they provide \emph{quantitative} measures 
about these properties.
For instance, from a typing derivation, it is possible to extract an
\emph{upper bound} ---or even an \emph{exact measure}--- for the
number of steps which is necessary to reach the normal form.

More precisely, exact measures can be obtained in the so-called
\emph{tight} quantitative type systems~\cite{AccattoliGK18}, by
decorating the type judgements with counters, and by imposing
a \emph{tightness} condition on typing derivations to ensure that
they are \emph{minimal}.
More specifically, typing judgements are enriched with integer
\emph{counters} specifying some quantitative information; for
instance, in a typing judgement $\tctx \vdash^m \tm : \typ$, the
counter $m$ is used to capture the fact that $\tm$ evaluates in
exactly $m$ reduction steps to a normal form.

To (quantitatively) characterise termination for a given evaluation
strategy in such type systems, one relies on the results of
\emph{soundness} and \emph{completeness} of the type system with
respect to the corresponding evaluation strategy. 
Here, \emph{quantitative} soundness means that for any \emph{tight} 
type derivation of a program $\tm$ with counter $\cm$, the program 
$\tm$ evaluates to normal form in exactly $\cm$ steps; this idea is
generalised for steps of many possible kinds with counters
$\cm_1,\hdots,\cm_n$. 
Conversely, quantitative \defn{completeness} means that every 
reduction sequence to normal form of a given length corresponds to a 
\emph{tight} typing derivation with appropriate counters.

\emph{Quantitative types} based on non-idempotent IT have been
applied to various evaluation strategies in the $\lambda$-calculus
for obtaining upper bounds and exact measures~\cite
{BernadetGL13,BucciarelliKV17,AccattoliGK18,deCarvalho18,KesnerViso22},
as well as in the contexts of classical calculi~\cite
{KesnerVial17,KesnerVial20}, call-by-value~\cite
{Ehrhard12,AccattoliGuerrieri18,Guerrieri18,AccattoliGuerrieri22},
call-by-need~\cite
{BalabonskiBBK17,BarenbaumBM18,AccattoliGL19,Leberle21},
call-by-push-value~\cite{BucciarelliKRV20}, etc. 
These type systems also serve as semantic interpretations, akin to 
\emph{relational models} in the usual sense of linear logic~\cite
{Girard88,BucciarelliE01}.

\paragraph{\bf Contributions.}
This work is split in two parts.
The \textbf{first part} defines the first inductive specification of 
useful evaluation, in particular in the framework of open \CBV.
Although most functional programming languages are restricted to 
closed terms (forbidding free variables), working in an open setting
(allowing free variables) broadens the applicability of our results.
Furthermore, being able to treat the open \CBV case is a prerequisite 
to deal in the future with other cases.
For example, Gr\'egoire and Leroy study \emph{strong} \CBV evaluation 
in~\cite{GregoireL02}, motivated by the implementation of proof 
assistants based on dependent type theory.

In contrast to existing specifications of useful evaluation in the
literature, our calculus relies on \emph{inductive} rules. 
To define the new strategy \UOCBV, we proceed in two stages. 
First, we recall in \textbf{\cref{sec:opencbv}} the 
\emph{Value Substitution Calculus} (\VSC)~\cite{AccattoliP12}, which 
is the starting point. 
We then introduce the \emph{\nonuseful} open \CBV strategy (\LOCBV), 
which refines the \VSC with \emph{linear substitution}.
Moreover, \LOCBV implements principle \princo but still not \princd. 
Second, we introduce in \textbf{\cref{sec:usefulcbv}} the \UOCBV 
strategy, by refining evaluation in \LOCBV to also implement principle 
\princd. 
Our inductive approach to usefulness has been inspired by the 
definition of strong call-by-need evaluation in~\cite{BalabonskiLM23}. 
We also prove key operational properties of \UOCBV; in particular, it 
enjoys the diamond property ---and thus confluence. 
In \textbf{\cref{sec:relating}}, we show that \UOCBV is indeed a 
useful implementation of \LOCBV. 
The relationship between the two strategies is established 
syntactically, by means of rewriting techniques, and relying on a
\emph{partial unfolding} operation that performs \emph{all} the 
substitutions assigning values to variables. 
The main result is that \UOCBV preserves the semantics of \LOCBV, \ie,
if \LOCBV reaches a result from a starting term then, \UOCBV reaches 
the same result, up to partial unfolding.

Finally, in \textbf{\cref{sec:usefulcbv_invariant}}, we show that \UOCBV 
can in turn be implemented by a lower-level abstract machine 
implementing useful \CBV, namely the \glamour~\cite{AccattoliC15}.
This result not only connects \UOCBV with a preexisting notion of usefulness,
but it also entails that \UOCBV is \emph{time-invariant}. 
This means that the length of reductions to normal form can be taken 
as a measure of time complexity. 
More precisely, \UOCBV can be simulated by a \emph{reasonable} cost 
model of computation in the sense of~\cite{SlotvEB84} (such as a 
Turing machine or a RAM) with at most polynomial overhead in time.

The \textbf{second part} defines the first semantic model of useful
\CBV evaluation, addressing the lack of a high-level, semantic 
characterisation of this strategy.
In \textbf{\cref{sec:typing}}, we define a type system based on non-idempotent 
intersection types called system $\typesystem$, which is shown to
characterise termination in \UOCBV.
We show that system $\typesystem$ turns out to be the first semantic 
model for useful \CBV evaluation, and moreover, it provides 
quantitative information on the evaluation strategy.
To do so, we equip type judgements of system $\typesystem$ with 
counters, and we define a notion of \emph{tightness}, intuitively 
capturing minimality of derivations. 
The counters are meant to capture the exact step-count for program 
evaluation. 
More precisely, we show that a term $\tm$ is \emph{tightly} typable 
with counters $\cm$ and $\ce$ in system $\typesystem$ if and only if 
$\tm$ terminates in \UOCBV in exactly $\cm$ function application steps 
and $\ce$ substitution steps
(\cref{coro:quantitative_characterisation_termination}).

Our contribution is novel, as previous definitions of useful 
evaluation (including useful \CBV) lack semantic models and current
quantitative interpretations of \CBV do not consider \emph
{usefulness}.
Moreover, except~\cite{KesnerViso22}, entirely based 
on the \emph{persistent/consuming} paradigm, this is the only 
quantitative type system for \CBV that is able to count \emph
{substitution} steps \emph{exactly}.

\section{Preliminary notions}
\label{sec:preliminary_notions}

In this section, we define the shared notions for both strategies 
\LOCBV and \UOCBV.
Next, we explain the two principles \princo and \princd that underlie
the notion of useful evaluation.
Moreover, we recall some background notions of rewriting theory.

\paragraph{\bf Syntax.}
Given a denumerable set of \defn{variables} ($\var,\vartwo,\varthree,\hdots$), 
the sets of \defn{terms} ($\tm,\tmtwo,\tmthree,\hdots$),
\defn{substitution contexts} ($\sctx,\sctxtwo,\hdots$), and 
\defn{values} ($\val,\valtwo,\hdots$) are given by the following grammars:
\[
  \tm \eqgram \var \mid \lam\var\tm \mid \tm\,\tm \mid \tm\esub\var\tm 
  \HS
  \sctx \eqgram \ctxhole \mid \sctx\esub\var\tm
  \HS
  \val \eqgram \var \mid \lam\var\tm
\]
The set of terms includes \defn{variables}, \defn{abstractions},
\defn{applications}, and \defn{closures} $\tm\esub\var\tmtwo$ representing 
an \defn{explicit substitution} (ES) $\esub\var\tmtwo$ on a term
$\tm$. 
\defn{Free} and \defn{bound occurrences} of variables are defined as 
usual, where free occurrences of $\var$ in $\tm$ are bound in 
$\tm\esub\var\tmtwo$.
We write $\fv\tm$ and $\fv\sctx$ for the set of free variables of a 
term and of a context, respectively.
Terms are considered up to \emph{$\alpha$-renaming} of bound variables. 
We write $\tm\sctx$ for the \emph{variable-capturing} \defn{replacement} 
of the hole $\ctxhole$ in $\sctx$ by $\tm$, keeping the standard 
notation $\ctxof\gctx\tm$ for other kinds of contexts.
We write $\abs\tm$ if $\tm$ is of the form $(\lam\var\tm)\sctx$ and
$\tm\sub\var\tmtwo$ stands for the \emph{capture-avoiding} 
\defn{substitution} of the free occurrences of $\var$ with $\tmtwo$ in $\tm$.
The set of \defn{reachable variables} of a term $\tm$ is written 
$\rv\tm$ and defined as:
\[
  \rv\var                 \eqdef \set\var
  \HS\HS
  \rv{\lam\var\tm}        \eqdef \emptyset
  \HS\HS
  \rv{\tm \, \tmtwo}      \eqdef \rv\tm \cup \rv\tmtwo
  \HS\HS
  \rv{\tm\esub\var\tmtwo} \eqdef (\rv\tm \setminus \set\var) \cup \rv\tmtwo
\]

Some recurring terms are the identity function $\id \eqdef \lam\var\var$ 
and the operator $\omega \eqdef \lam\var{\var \, \var}$.

\paragraph{\bf Principles of Useful Evaluation.}
As mentioned in the introduction, useful evaluation is motivated by
the two following key \emph{usefulness principles}:

\indent \emph{Sharing Structures} (\princo).
  A term is a \textbf{structure}\footnote{This name is borrowed
  from~\cite{BalabonskiBBK17}, but the terminology \emph{inert term} 
  is also used~\cite{AccattoliC15}.} if its \emph{(full) unfolding} 
  ---the result of performing all the explicit substitutions by using 
  the capture-avoiding substitution--- results in an application 
  headed by a variable, \ie, of the form $\var\,\tm_1\hdots\tm_n$, 
  where $(n \geq 0)$.
  For example,
  $(\var\,\var)\esub\var{\vartwo\,\id}\esub\vartwo{\varthree\,\varthree}$
  is a structure, as it (fully) unfolds to 
  $\varthree\,\varthree\,\id\,(\varthree\,\varthree\,\id)$.
  Structures must remain \emph{shared} in useful evaluation, as 
  substituting a variable by a structure does not create a 
  \emph{function application redex}.
  Thus, a term like 
  $(\var\,\var)\esub\var{\vartwo\,\id}\esub\vartwo{\varthree\,\varthree}$
  is a normal form in useful evaluation.
  The notion of structure just defined is context-independent.
  Actually, we shall need a subtler \emph{context-dependent} notion of
  structure relative to a ``structure frame'' (defined in \cref{sec:usefulcbv}).
  For instance $\var\,\vartwo$ should \emph{not} be treated as a structure
  under a context that binds $\var$ to an abstraction.

\indent \emph{Substituting Abstractions for Progress} (\princd).
  In useful evaluation, abstractions are substituted only if they  
  contribute to creating a function application redex,
  thus ensuring progress in the computation. 
  For example,  a  reduction step like
  $\var\esub\var\id \, \vartwo \to \id\esub\var\id \, \vartwo$ is 
  useful, while the steps 
  $\var\esub\var\id \to \id\esub\var\id$ and 
  $(\tm\,\var)\esub\var\id \to (\tm\,\id)\esub\var\id$
  are not, as they do not contribute to creating function application redexes.
  Some of these ideas can also be found in the literature on 
  \emph{optimal reduction} (see \eg, \cite{Yoshida93}).

\paragraph{\bf Background on Rewriting Theory.}
We now introduce some general notions of reduction that will be used 
all along the document. 
Given a \defn{reduction system} $\relation$, we denote by 
$\rewrite\relation$ the (one-step) reduction relation associated to 
system $\relation$. 
We write $\rewriteeq\relation$ and $\rewriten\relation$ for the 
reflexive and reflexive-transitive closure of $\rewrite\relation$, 
and $\rewritenpasos\relation{n}$ for the composition of $n$-steps of 
$\rewrite\relation$. 
A term $\tm$ is said to be \defn{$\relation$-reducible} if there is 
$\tmtwo$ such that $\tm \rewrite\relation \tmtwo$, and $\tm$ is said 
to be \defn{$\relation$-irreducible}, or in \defn{$\relation$-normal
form}, written $\tm \not \rewrite\relation$, if there is no $\tmtwo$ 
such that $\tm \rewrite\relation \tmtwo$. 
A term $\tm$ is said to be \defn{$\relation$-terminating} if there is 
no infinite $\relation$-sequence starting at $\tm$. 
A term $\tm$ is \defn{$\relation$-diamond} (or enjoys the 
$\relation$-diamond property) if $\tm \rewrite\relation \tm_0$ and 
$\tm \rewrite\relation \tm_1$ with $\tm_0 \neq \tm_1$ imply there is 
$\tm'$ such that $\tm_0 \rewrite\relation \tm'$ and 
$\tm_1 \rewrite\relation \tm'$. 
A term $\tm$ is  \defn{$\relation$-locally confluent} 
(resp. \defn{$\relation$-confluent}) if $\tm \rewrite\relation \tm_0$
and $\tm \rewrite\relation \tm_1$ (resp. $\tm \rewriten\relation \tm_0$ 
and $\tm \rewriten\relation \tm_1$) imply there is $\tm'$ such that
$\tm_0 \rewriten\relation \tm'$ and $\tm_1 \rewriten\relation \tm'$. 
A relation $\relation$ is \defn{terminating} (resp. \defn{diamond}, 
\defn{locally confluent}, \defn{confluent}) if and only if every term 
is $\relation$-terminating (resp. $\relation$-diamond, 
$\relation$-locally confluent, $\relation$-confluent). 
Any relation verifying the diamond property is in particular confluent,
and any relation verifying termination and local confluence is
confluent~\cite{Terese03}. 
Moreover, if $\tm$ is confluent, then its $\relation$-normal form, 
if it exists, is unique~\cite{Terese03}.

\section{\Nonuseful Call-by-Value}
\label{sec:opencbv}

In this section, we define the \textsc{L}inear \textsc{\CBV} strategy 
(\LOCBV), a first step towards useful evaluation.
\LOCBV fulfils principle \princo by keeping structures always shared, 
but it does not fulfil \princd, as it allows substituting variables 
by values unrestrictedly (see~\cref{sec:preliminary_notions} for the 
definitions of \princo and \princd).

Furthermore, \LOCBV adapts the substitution operation in the Value
Substitution Calculus (\VSC) from~\cite{AccattoliP12} to be 
\emph{linear} by allowing the substitution of one occurrence of a 
variable at a time.
We begin by explaining the difference between the (non-linear) 
original substitution mechanism of the \VSC, and the linear 
substitution of \LOCBV.
We then formally define \LOCBV and conclude by stating its main 
properties.

\paragraph{\bf The Value Substitution Calculus.}
Avoiding the substitution of variables by structures is a sufficient 
mechanism to obtain an invariant implementation of open \CBV (see 
for instance~\cite[Lemma 6]{AccattoliG17}).
One calculus implementing this mechanism for open terms is the \VSC, 
which consists of two kinds of reduction steps:

\emph{Distant beta} ($\dbsym$) steps are given by the rule
$(\lam\var\tm)\sctx\,\tmtwo \to \tm\esub\var\tmtwo\sctx$, where 
$\sctx$ denotes a list of ES, called a \emph{substitution context}.
A $\dbsym$-step performs a \emph{function application} by creating
an ES.
For example,
$(\lam\var\id)\,(\vartwo\, \vartwo)\,\varthree
 \to \id\esub\var{\vartwo \, \vartwo}\,\varthree
 \to \varfour\esub{\varfour}{\varthree}\esub\var{\vartwo \, \vartwo}$.
After the first step, the ES $\esub\var{\vartwo \, \vartwo}$
appears to be blocking the interaction between $\id$ and its argument
$\varthree$.
However, the $\dbsym$-rule allows an arbitrary list of ESs to 
appear between an abstraction and its argument, thus allowing to 
fire the second reduction step in the example.

\emph{Substitution} steps ($\svsym$) are given by the rule 
$\tm\esub\var{\val\sctx} \to \tm\sub\var\val\sctx$, which substitutes 
a variable $\var$ by a value $\val$, \emph{pushing outside} the 
substitution context $\sctx$ originally accompanying the value. 
For example, $(\var\, \var)\esub\var{\omega\esub\vartwo\id} \to 
(\omega\, \omega)\esub\vartwo\id$, where we recall that 
$\omega \eqdef \lam\varthree{\varthree\,\varthree}$.
The \VSC implements principle \princo because it only allows 
substituting variables by values (so that structures remain shared).
For example, the step
$(\var\,\var)\esub\var{\varthree\,\varthree} \to
\varthree\,\varthree\,(\varthree\,\varthree)$ is not allowed. 

\paragraph{\bf Towards Linear Substitution.}
\LOCBV refines the \VSC by implementing \emph{linear substitution},
which proceeds by \emph{micro-steps}, replacing one occurrence of a 
variable at a time.
Linear substitution is a prerequisite ---but not sufficient--- to
fulfil the principle \princd, \ie, to \emph{fully} implement useful 
evaluation.
Specifically, when a variable occurs multiple times in a term,
substituting it with an abstraction may create a function application 
in some cases but not in others. For example, in the step
$(\varthree\,(\underline\var\,\vartwo)\,\overline\var)\esub\var\id
  \to
  \varthree\,(\id\,\vartwo)\,\id$,
substituting the underlined occurrence of $\var$ by $\id$ creates
a $\dbsym$-redex (and thus it is useful) while substituting the 
overlined occurrence of $\var$ by $\id$ does not contribute to 
creating a $\dbsym$-redex.

\paragraph{\bf The Linear \CBV Calculus.}
We now define the Linear \CBV strategy \LOCBV, which refines the \VSC
by implementing linear substitution. 
We adopt a formulation of rules based on the style of~\cite{BalabonskiLM23},
which we also follow to formulate our inductive notion of useful evaluation,
\UOCBV, in~\cref{sec:usefulcbv}. 

Formally, we define a family of binary relations $\tovv\rulename$,
where $\rulename \in \set{\ruledb, \rulelsv, \rulesub\var\val}$
distinguishes the \defn{step kind}, with $\var$ being a variable and 
$\val$ a value such that $\var \notin \fv\val$. 
The set of \defn{free variables} of a step kind $\rulename$ is given 
by $\fv\ruledb \eqdef \emptyset$, $\fv\rulelsv \eqdef \emptyset$, and
$\fv{\rulesub\var\val} \eqdef \set\var \cup \fv\val$.
The \defn{\nonuseful} relation $\tovv\rulename$ of \LOCBV is defined 
inductively as follows:
\label{def:non_useful_open_CBV}
\[
  \inferrule{
  }{
    (\lam\var\tm)\sctx \, \tmtwo \tovv\ruledb
    \tm\esub\var\tmtwo\sctx
  }\ruleVDb
  \HS
  \inferrule{
  }{
    \var \tovv{\rulesub\var\val} \val
  }\ruleVSub
  \HS
  \inferrule{
    \tm \tovv{\rulesub\var\val} \tm'
  }{
    \tm\esub\var{\val\sctx} \tovv\rulelsv
    \tm'\esub\var\val\sctx
  }\ruleVLsv
\]
\[ 
  \inferrule{
    \tm \tovv\rulename \tm'
  }{
    \tm \, \tmtwo \tovv\rulename \tm' \, \tmtwo
  }\ruleVAppL
  \HS
  \inferrule{
    \tmtwo \tovv\rulename \tmtwo'
  }{
    \tm\,\tmtwo \tovv\rulename \tm \, \tmtwo'
  }\ruleVAppR
  \HS
  \inferrule{
    \tm \tovv\rulename \tm'
    \sep
    \var \notin \fv\rulename
  }{
    \tm\esub\var\tmtwo \tovv\rulename \tm'\esub\var\tmtwo
  }\ruleVEsL
  \HS
  \inferrule{
    \tmtwo \tovv\rulename \tmtwo'
  }{
    \tm\esub\var\tmtwo \tovv\rulename \tm\esub\var{\tmtwo'}
  }\ruleVEsR
\]
Rules \ruleVDb, \ruleVSub, and \ruleVLsv define the three kinds of 
reduction steps, whereas \ruleVAppL, \ruleVAppR, \ruleVEsL, and 
\ruleVEsR are congruence rules, and propagate any step kind. 
Note that evaluation is \emph{weak} and does not proceed within abstractions.
A step of the form $\tm \tovv\ruledb \tmtwo$ represents a 
\emph{distant beta} step.
The reduction steps $\tm \tovv{\rulesub\var\val} \tmtwo$ and 
$\tm \tovv\rulelsv \tmtwo$ constitute two kinds of \emph{substitution} steps.
In the first kind, $\tm \tovv{\rulesub\var\val} \tmtwo$, \emph{one} 
free occurrence of $\var$ in $\tm$ is substituted with $\val$.
In the second kind, $\tm \tovv\rulelsv \tmtwo$, \emph{one} bound 
occurrence of $\var$ is substituted with $\val$, provided $\var$ is 
bound to a term of the form $\val\sctx$ by an ES; \ruleVLsv is the 
only rule that introduces an $\lsvsym$-step.
Moreover, each step substituting a bound variable ($\tovv\rulelsv$) 
relies internally on a step that substitutes a free variable 
($\tovv{\rulesub\var\val}$).
These substitutions steps focus on a \emph{single} variable 
occurrence, thus we say the substitution operation is \emph{linear}.
One noteworthy remark is that rule \ruleVSub allows substitution of a 
variable for \emph{any} value, thus the reduction steps 
$\var \tovv{\rulesub\var{\val_1}} \val_1$ and
$\var \tovv{\rulesub\var{\val_2}} \val_2$ are valid, so \emph{confluence} 
of $\tovv{\rulesub{\_}{\_}}$ makes no sense: these steps are only an 
auxiliary mechanism to define $\rulelsv$-steps, and in particular to 
define the notion of \emph{linear} substitution of a single occurrence
of a variable by a value.
At the end of this section, we prove the property of confluence on 
the \defn{top-level} \LOCBV reduction, defined as 
$\tovvTop \eqdef \tovv\ruledb \cup \tovv\rulelsv$.
Moreover, the relation $\tovv\rulelsv$ is shown to be terminating
(\cref{cor:tolsv-terminating}).

Along with rule \ruleVDb, an application can be evaluated using rules 
\ruleVAppL and \ruleVAppR, which reduce within the left and right subterms, 
respectively.
Since rules \ruleVDb, \ruleVAppL, and \ruleVAppR overlap, reduction 
is \emph{non-deterministic}. 
For example:
\[
  \var\esub\var\id \, (\id\,\id) \tovvinv\ruledb
  \id\,\id \, (\id\,\id) \tovv\ruledb \id\,\id\,\var\esub\var\id
\]

The congruence rules \ruleVEsL and \ruleVEsR allow reduction on the 
left side of a term and within the argument of an ES, respectively. 
These rules overlap as well.
A key restriction in rule \ruleVEsL is that the variable $\var$ bound 
by the ES $\esub\var\tmtwo$ must not occur free in the step kind $\rulename$.
This restriction is to avoid variable capture; for example, it ensures 
that a ``pathological'' step like
$\vartwo\esub\var\varthree \tovv{\rulesub\vartwo\var} \var\esub\var\varthree$ 
cannot be derived from the valid step 
$\vartwo \tovv{\rulesub\vartwo\var} \var$.

\begin{example}
\label{example:evaluation_in_linear_cbv}
The following is a reduction sequence to normal form in \LOCBV:
\[
  \begin{array}{r@{\,\,}l@{\,\,}l@{\,\,}l@{\,\,}l@{\,\,}l@{\,\,}l}
    & (\lam\var{\varthree \, \var \, (\var\,\vartwo)}) \, \id
    & \tovv\ruledb 
    & (\varthree\,\var\,(\var\,\vartwo))\esub\var\id
  \\
      \tovv{\rulelsv} 
    & (\varthree\,\id\,(\var\,\vartwo))\esub\var\id
    & \tovv\rulelsv
    & (\varthree\,\id\,(\id\,\vartwo))\esub\var\id 
  \\
      \tovv\ruledb 
    & (\varthree\,\id\,(\varfour\esub\varfour\vartwo))\esub\var\id
    & \tovv\rulelsv 
    & (\varthree\,\id\,(\vartwo\esub\varfour\vartwo))\esub\var\id
  \end{array}
\]
\end{example}

\paragraph{\bf Confluence}
Despite the overlaps of some reduction rules, it is straightforward 
to show that top-level $\tovv{}$ reduction is confluent:
\begin{proposition}
$\tovvTop$ is confluent.
\end{proposition}

\section{Useful Call-by-Value}
\label{sec:usefulcbv}

In this section, we refine the notion of \LOCBV reduction introduced in
\cref{sec:opencbv}, the resulting strategy is called \UOCBV. 
While preserving principle \princo, already present in \LOCBV, \UOCBV 
evaluation also implements principle \princd, which is to be captured
\emph{inductively}.
To achieve useful \CBV evaluation, we define a family of reduction 
relations indexed by certain \emph{parameters} representing the 
essential information from the surrounding evaluation contexts. 
This is the minimal data necessary to decide whether a substitution
step is useful or not.

After defining the new strategy \UOCBV, we give an inductive 
characterisation of its normal forms (\cref{thm:uocbv_characterization_of_normal_forms}).
Following this result, we show that \UOCBV satisfies the Diamond 
Property (\cref{thm:diamond-property-top-level}), which ensures that 
the length of a reduction sequence to normal form does not depend on 
the particular order of redexes chosen to reduce terms.

\paragraph{\bf Towards the Useful \CBV Strategy.}
%\label{sec:useful}
To implement principle \princd in \UOCBV it is necessary to ensure
that substitution contributes to the \emph{progress} of the evaluation.
For example, the substitution step
$(\var\,\tm)\esub\var\vartwo\esub\vartwo\id
 \to (\vartwo\,\tm)\esub\var\vartwo\esub\vartwo\id$
is (indirectly) useful because the occurrence of $\var$ is applied to 
an argument $\tm$, and substituting $\var$ by $\vartwo$ then 
contributes to creating the underlined redex
$(\underline{\id \, \tm})\esub\var\vartwo\esub\vartwo\id$ after one 
more substitution step.
This motivates the need to identify \emph{hereditary abstractions},
which intuitively include abstractions and variables hereditarily bound to abstractions,
such as $\vartwo$ in the example.
Formally, let $\aset$ be a set of variables, called an \defn{abstraction frame}.
The set of \defn{hereditary abstractions} under $\aset$ is
written $\HAbs\aset$ and defined below:
\[
  \inferrule{
  }{
    \lam\var\tm \iin \HAbs\aset
  }\ruleHAbsLam
  \HS
  \inferrule{
    \var\iin\aset
  }{
    \var\iin\HAbs\aset
  }\ruleHAbsVar
  \HS
  \inferrule{
    \tm \iin \HAbs\aset
    \sep
    \var \notiin \aset
  }{
    \tm\esub\var\tmtwo \iin \HAbs\aset
  }\ruleHAbsSubi
  \HS
  \inferrule{
    \tm \iin \HAbs{\aset \cup \set\var}
    \sep
    \var \notiin \aset
    \sep
    \tmtwo \iin \HAbs\aset
  }{
    \tm\esub\var\tmtwo \iin \HAbs\aset
  }\ruleHAbsSubii
\]
The abstraction frame $\aset$ keeps track of variables that are bound
to hereditary abstractions.
Every abstraction is a hereditary abstraction (by the first rule),
while a variable is a hereditary abstraction if it appears in the abstraction frame
(by the second rule).
In the case of terms with ESs, the crucial point is that the bound variable $\var$ may be
added to the abstraction frame only if it is bound to a hereditary abstraction
(third and fourth rule).
For example, $\var\esub\var{\lam\vartwo\vartwo} \in \HAbs\aset$ for 
any $\aset$, and 
$\var\esub\varthree\varfour\esub\var\vartwo \in \HAbs\aset$ if and 
only if $\vartwo \in \aset$.
Note also that applications are never hereditary abstractions.

On the other hand, it is necessary to identify \emph{structures},
which intuitively include both applied variables and variables 
hereditarily bound to structures.
Formally, let $\sset$ be a set of variables, called a \defn{structure frame}.
The set of \defn{structures under $\sset$} is written $\Struct\sset$
and defined below:
\[
  \inferrule{
    \var \iin \sset
  }{
    \var \iin \Struct\sset
  }\ruleStructVar
  \HS
  \inferrule{
    \tm \iin \Struct\sset
  }{
    \tm\,\tmtwo \iin \Struct\sset
  }\ruleStructApp
  \HS
  \inferrule{
    \tm \iin \Struct\sset
    \sep
    \var\notiin\sset
  }{
    \tm\esub\var\tmtwo \iin \Struct\sset
  }\ruleStructSubi
  \HS
  \inferrule{
    \tm \iin \Struct{\sset\cup\set\var}
    \sep
    \var\notiin\sset
    \sep
    \tmtwo \iin \Struct\sset
  }{
    \tm\esub\var\tmtwo \iin \Struct\sset
  }\ruleStructSubii
\] 
The structure frame $\sset$ keeps track of variables that are bound
to structures.
Note that no abstraction belongs to the set of structures, for any $\sset$. 
An application is a structure if its left subterm is itself a structure,
thus excluding $\dbsym$-redexes from the set of structures under any $\sset$.
Variables in the structure frame are structures as well. 
For terms with ESs, the intuitions behind the third and fourth rules 
are analogous to the corresponding ones for the predicate $\HAbs{}$, respectively.
For example, performing the substitution on the left subterm in
$\var\esub\var{\vartwo_1 \, \vartwo_2} \, \varthree$ does not create 
any $\dbsym$-redex, so $\var\esub\var{\vartwo_1 \, \vartwo_2}$ is 
considered a structure.
Similarly, $(\var \, \vartwo)\esub\vartwo\id$ is a structure, even if 
$\vartwo$ is bound to an abstraction.
\medskip

We summarise below some key (but easy) properties of hereditary 
abstractions and structures:
\begin{remark}
\label{rem:habs_st}
(1) If $\aset \subseteq \aset'$, then $\HAbs\aset \subseteq \HAbs{\aset'}$; 
(2) A term of the form $\tm \, \tmtwo$ is never in $\HAbs\aset$; 
(3) If $\tm \in \HAbs\aset$, then $\tm = \val\sctx$; 
(4) If $\sset \subseteq \sset'$, then $\Struct\sset \subseteq \Struct{\sset'}$;
(5) A term of the form $(\lam\var\tm)\sctx$ is never in $\Struct\sset$; 
(6) For any $\sctx$, one has that $(\lam\var\tm)\sctx \in \HAbs\aset$.
\end{remark}

\paragraph{\bf The Useful \CBV Reduction.}
We can now formally define the strategy \UOCBV specified as a family 
of binary relations $\tov\rulename\aset\sset\appflag$, where 
$\rulename \in \set{\ruledb, \rulelsv, \rulesub\var\val}$ is a step 
kind, $\aset$ is an \abstractionframe, $\sset$ is a structure frame, 
and $\appflag \in \set{\app, \nonapp}$ is a \positionalFlag. 
The \defn{reduction relation} $\tov\rulename\aset\sset\appflag$ of 
\UOCBV is inductively defined as follows:
\[
  \inferrule{
  }{
    (\lam\var\tm)\sctx\,\tmtwo
    \tov\ruledb\aset\sset\appflag
    \tm\esub\var\tmtwo\sctx
  }\ruleUDb
  \HS
  \inferrule{
  }{
    \var \tov{\rulesub\var\val}{\aset \cup \set\var}\sset\app \val
  }\ruleUSub
  \HS
  \inferrule{
    \tm
    \tov{\rulesub\var\val}{\aset\cup\set\var}\sset\appflag
    \tm'
    \sep
    \var \notin \aset\cup\sset
    \sep
    \val\sctx \in \HAbs\aset
  }{
    \tm\esub\var{\val\sctx}
    \tov\rulelsv\aset\sset\appflag
    \tm'\esub\var\val\sctx
  }\ruleULsv
\]
\[
  \inferrule{
    \tm \tov\rulename\aset\sset\app \tm'
  }{
    \tm \, \tmtwo \tov\rulename\aset\sset\appflag \tm' \, \tmtwo
  }\ruleUAppL
  \HS
  \inferrule{
    \tm \in \Struct\sset
    \sep
    \tmtwo \tov\rulename\aset\sset\nonapp \tmtwo'
  }{
    \tm \, \tmtwo \tov\rulename\aset\sset\appflag \tm \, \tmtwo'
  }\ruleUAppR
  \HS
  \inferrule{
    \tmtwo \tov\rulename\aset\sset\nonapp \tmtwo'
  }{
    \tm\esub\var\tmtwo 
    \tov\rulename\aset\sset\appflag
    \tm\esub\var{\tmtwo'}
  }\ruleUEsR
\]
\[
  \inferrule{
    \tm \tov\rulename{\aset\cup\set\var}\sset\appflag \tm'
    \sep
    \tmtwo \in \HAbs\aset
    \sep
    \var \notin \aset \cup \sset
    \sep
    \var \notin \fv\rulename
  }{
    \tm\esub\var\tmtwo
    \tov\rulename\aset\sset\appflag
    \tm'\esub\var\tmtwo
  }\ruleUEsLAbs
\]
\[
  \inferrule{
    \tm \tov\rulename\aset{\sset\cup\set\var}\appflag \tm'
    \sep
    \tmtwo \in \Struct\sset
    \sep
    \var \notin \aset\cup\sset
    \sep
    \var\notin\fv\rulename
  }{
    \tm\esub\var\tmtwo
    \tov\rulename\aset\sset\appflag
    \tm'\esub\var\tmtwo
  }\ruleUEsLStruct
\]
Note that each \UOCBV step is also a \LOCBV step, \ie,
$\tov\rulename\aset\sset\appflag \ \subseteq \ \tovv\rulename$. 
Also, the reduction relation defined above is \defn{non-erasing}: if
$\tm\tov\rulename\aset\sset\appflag \tm'$ with
$\rulename \in \set{\ruledb, \rulelsv}$, then $\fv\tm = \fv{\tm'}$. 
As in the previous section, rules \ruleUDb, \ruleUSub, and \ruleULsv 
introduce the three possible step kinds, while all other cases 
(\ruleUAppL, \ruleUAppR, \ruleUEsLAbs, \ruleUEsLStruct, and \ruleUEsR) 
are congruence rules for steps of an arbitrary step kind $\rulename$. 
Reduction is \emph{weak} since there are no rules to evaluate under 
abstractions.

Rule \ruleUDb performs a function application step, identical to rule 
\ruleVDb in \LOCBV. 
Steps of the form $\tm \tov{\rulesub\var\val}\aset\sset\appflag \tmtwo$ 
and $\tm \tov\rulelsv\aset\sset\appflag \tmtwo$ represent two different 
kinds of \emph{substitution steps}, as in \LOCBV: the former substitutes 
\emph{one} free occurrence of a variable, while the latter substitutes
\emph{one} bound occurrence.
Only rule \ruleULsv creates an $\lsvsym$-step.
Each application of rule \ruleULsv relies internally on a 
$\rulesub\var\val$-step, where $\val$ must be a hereditary abstraction,
a restriction not required in rule \ruleVLsv rule in \LOCBV, but that
it is crucial to implement principle \princd
(\emph{substituting abstractions for progress}).

Substituting a free variable in a term $\tm$ using rule \ruleUSub is 
only possible when $\tm$ is in an applied position because this 
ensures the progress of the computation.
The following example shows an instance of the \ruleUSub rule:
\[
  \inferrule{
    \inferrule* [Right = \ruleULsv]{ 
      \inferrule* [Right = \ruleUSub]{
      }{
        \var
        \tov{\rulesub\var\id}{\set\var}{\set\varthree}\app
        \id
      }{
      \HS\HS
      \id\esub\vartwo\varthree \in \HAbs\emptyset
      }
    }{
      \var\esub\var{\id\esub\vartwo\varthree}
      \tov\rulelsv\emptyset{\set\varthree}\app
      \id\esub\var\id\esub\vartwo\varthree
    }
  }{
    \var\esub\var{\id\esub\vartwo\varthree} \, \varfour
    \tov\rulelsv\emptyset{\set\varthree}\nonapp
    \id\esub\var\id\esub\vartwo\varthree \varfour
  }\ruleUAppL
\]

The congruence rules \ruleUAppL and \ruleUAppR perform reduction on the 
application constructor.
Rule \ruleUAppR additionally requires the left subterm of the application 
to be a structure, avoiding a possible overlap with rule \ruleUDb.
Still, rules \ruleUAppL and \ruleUAppR overlap, so reduction is 
non-deterministic, like in \LOCBV. 
For example:
\[
  \var\,\vartwo\esub\vartwo\id\,(\id\,\id)
  \tovinv\ruledb\emptyset{\set\var}\nonapp
  \var\,(\id\,\id)\,(\id\,\id)
  \tov\ruledb\emptyset{\set\var}\nonapp
  \var\,(\id\,\id)\,\vartwo\esub\vartwo\id
\]

The congruence rules \ruleUEsR, \ruleUEsLAbs, and \ruleUEsLStruct apply
to closures.
Rule \ruleUEsR allows reducing the argument of any ESs, while rules 
\ruleUEsLAbs and \ruleUEsLStruct allow reduction of the left side of 
the closure under specific conditions:
\ruleUEsLAbs (resp. \ruleUEsLStruct) applies only if the substitution 
argument is a hereditary abstraction (resp. a structure).
Note that rules \ruleUEsLAbs and \ruleUEsLStruct force reducing the
argument $\tmtwo$ of a closure $\tm\esub\var\tmtwo$ until it becomes 
``stuck'', after which evaluation can proceed to the body $\tm$.
This behaviour aligns with a \CBV strategy, covering all cases in 
terminating terms, as every terminating term is eventually either a 
hereditary abstraction or a structure 
(\cf \cref{lem:nf_in_HAbs_or_Struct} in \cref{app:usefulcbv}).
These rules also enforce that the variable $\var$ bound by the ES 
$\esub\var\tmtwo$ does not occur free in the step kind $\rulename$,
as in rule \ruleVEsL in \cref{sec:opencbv}.
This restriction prevents variable capture in steps like 
$\vartwo\esub\var\id \tov{\rulesub\vartwo\var}{\set\vartwo}\emptyset\app \var\esub\var\id$.
Note that there is a possible overlap between rule 
\ruleUEsR and rules \ruleUEsLAbs, and \ruleUEsLStruct.

A term $\tm$ is \defn{$(\rulename,\aset,\sset,\appflag)$-reducible}
if there exists a term $\tm'$ such that 
$\tm \tov\rulename\aset\sset\appflag \tm'$.
A term $\tm$ belongs to the set $\Irred\aset\sset\appflag$ if 
$\tm$ is not $(\rulename,\aset,\sset,\appflag)$-reducible.

Up to this point, hereditary abstractions and structures show distinct 
and disjoint behaviours within \UOCBV, which are put in evidence
through the distinction between abstraction frames and structure frames.
An \abstractionframe $\aset$ and a structure frame $\sset$ are said 
to verify the \defn{correctness property} for $\tm$, written
$\inv\aset\sset\tm$, if $\aset \cap \sset = \emptyset$ and
$\fv\tm \subseteq \aset \cup \sset$.
This property is implicitly assumed in theorems and lemmas, \ie, 
whenever we write $\tm \tov\rulename\aset\sset\appflag$, we always 
keep $\inv\aset\sset\tm$ as an \emph{invariant}.
Note that $\inv\emptyset{\fv\tm}\tm$ always holds, so for a top-level 
term $\tm$, we may take $\aset \eqdef \emptyset$ and $\sset \eqdef \fv\tm$.

\begin{example}
\label{example:evaluation_in_useful_cbv}
The following is an \UOCBV reduction sequence to normal form:
\[
  \begin{array}{r@{\,\,}l@{\,\,}l@{\,\,}l@{\,\,}l}
    & (\lam\var{\varthree\,\var\,(\var\,\vartwo)})\,\id
    & \tov\ruledb\emptyset{\set{\varthree,\vartwo}}\nonapp 
    & (\varthree\,\var\,(\var\,\vartwo))\esub\var\id
  \\
      \tov\rulelsv\emptyset{\set{\varthree,\vartwo}}\nonapp
    & (\varthree\,\var\,(\id\,\vartwo))\esub\var\id
    & \tov\ruledb\emptyset{\set{\varthree,\vartwo}}\nonapp 
    & (\varthree\,\var\,(\varfour\esub\varfour\vartwo))\esub\var\id
  \end{array}
\]
Compare this reduction sequence with the corresponding evaluation of 
the same term in \LOCBV (\cref{example:evaluation_in_linear_cbv}).
In \UOCBV, the leftmost occurrence of $\var$ and the occurrence of 
$\varfour$ are not substituted, because these substitutions would not
satisfy principle \princd: they do not contribute to creating
a $\dbsym$-redex.
\end{example}

\begin{example}
\label{ex:reduction-for-tight}
The following is another example of a \UOCBV reduction sequence to normal form:
\[ 
  \begin{array}{l@{\,\,}l@{\,\,}l@{\,\,}l}
    & (\var \, \varthree)\esub\var{\vartwo\esub\vartwo\id}
    & \tov\rulelsv\emptyset{\set\varthree}\nonapp
    & (\vartwo \, \varthree)\esub\var\vartwo\esub\vartwo\id 
  \\
      \tov\rulelsv\emptyset{\set\varthree}\nonapp 
    & (\id \, \varthree)\esub\var\vartwo\esub\vartwo\id
    & \tov\ruledb\emptyset{\set\varthree}\nonapp 
    & \var_1\esub{\var_1}\varthree\esub\var\vartwo\esub\vartwo\id
  \end{array} 
\]
\end{example}

As discussed in~\cref{sec:introduction}, useful evaluation is 
\emph{context-sensitive}, as the context surrounding a term 
determines whether a step is useful or not.
For instance, the term $\tm = \var\esub\var\id$ is already in normal 
form in \UOCBV, as substituting $\var$ by $\id$ is not useful, 
because it does not create a $\dbsym$-redex.
However, if $\tm$ is located to the left of an application, such as 
in $\var\esub\var\id \, (\vartwo \, \id)$, the substitution of $\var$ 
by $\id$ becomes useful, and yields the following reduction sequence:
\[
  \hfill
    \var\esub\var\id \, (\vartwo \, \id)
    \tov\rulelsv\emptyset{\set{\vartwo}}\nonapp
    \id\esub\var\id \, (\vartwo \, \id)
    \tov\ruledb\emptyset{\set{\vartwo}}\nonapp
    \varthree\esub\varthree{\vartwo \, \id}\esub\var\id
  \hfill
\]
Conversely, a useful step like 
$(\var\,\tmtwo)\esub\var\id 
\tov\rulelsv\emptyset\emptyset\nonapp (\id\,\tmtwo)\esub\var\id$ may 
not be allowed if the term is located below a context such as
$\lam\vartwo\ctxhole$; in fact, note that 
$\lam\vartwo{(\var\,\tmtwo)\esub\var\id}$ cannot be reduced.

\paragraph{\bf Operational Properties.}

We state two theorems about
the operational properties of \UOCBV.
\cref{thm:uocbv_characterization_of_normal_forms} provides an inductive
characterisation of the set of normal forms, while 
\cref{thm:diamond-property-top-level} shows that it enjoys a very 
strong form of confluence, namely the diamond property.
See \cref{app:usefulcbv} for more details and complete proofs.

\begin{theorem}[Characterisation of normal forms]
\label{thm:uocbv_characterization_of_normal_forms}
The set of irreducible terms $\Irred\aset\sset\appflag$
is exactly the set $\NF\aset\sset\appflag$ defined inductively as below:
\[
  \inferrule{
    \var \in \aset \implies \appflag=\nonapp
  }{
    \var \in \NF\aset\sset\appflag
  }\ruleUNFVar
  \HS
  \inferrule{
  }{ 
    \lam\var\tm \in \NF\aset\sset\nonapp
  }\ruleUNFLam
  \HS
  \inferrule{
    \tm \in \NF\aset\sset\app
    \sep
    \tmtwo \in \NF\aset\sset\nonapp
  }{
    \tm\,\tmtwo \in \NF\aset\sset\appflag
  }\ruleUNFApp
\]
\[
  \inferrule{
    \tm \in \NF{\aset\cup\set\var}\sset\appflag
    \sep
    \tmtwo \in \NF\aset\sset\nonapp
    \sep
    \tmtwo \in \HAbs\aset
  }{
    \tm\esub\var\tmtwo
      \in \NF\aset\sset\appflag
  }\ruleUNFEsAbs
  \HS
  \inferrule{
    \tm \in \NF\aset{\sset\cup\set\var}\appflag
    \sep
    \tmtwo \in \NF\aset\sset\nonapp
    \sep
    \tmtwo \in \Struct\sset
  }{
    \tm\esub\var\tmtwo
      \in \NF\aset\sset\appflag
  }\ruleUNFEsStruct
\]
\end{theorem}

\begin{restatable}[Diamond Property]{theorem}{diamondproperty}
\label{thm:diamond-property-top-level}
Let $\tm \tov{\rulename_1}\emptyset\sset\nonapp \tm_1$
and $\tm \tov{\rulename_2}\emptyset\sset\nonapp \tm_2$,
where $\tm_1 \neq \tm_2$ and 
$\rulename_1,\rulename_2 \in \set{\ruledb,\rulelsv}$
and $\sset = \fv\tm$.
Then, there exists $\tm'$ such that
$\tm_1 \tov{\rulename_2}\emptyset\sset\nonapp \tm'$ and
$\tm_2 \tov{\rulename_1}\emptyset\sset\nonapp \tm'$.
\end{restatable}

The two announced consequences of the previous result follow:
\begin{corollary}
\label{coro:all_reductions_NF_same_length}
Any two reduction sequences to normal form in \UOCBV have the same
number of $\dbsym$ and $\lsvsym$ steps.
\end{corollary}

As with \LOCBV, the \emph{confluence} of \UOCBV holds only for the 
\emph{top-level} step kinds $\dbsym$ and $\lsvsym$.
Let $\tovTop\sset {\eqdef} \tov\dbsym\emptyset\sset\nonapp \cup \tov\lsvsym\emptyset\sset\nonapp$
denote the top-level \UOCBV reduction, then:
\begin{corollary}
The reduction $\tovTop\sset$ is confluent.
\end{corollary}

\section{Relating \Nonuseful and Useful Call-by-Value}
\label{sec:relating}

This section aims to establish a connection between \LOCBV and \UOCBV,
defined in \cref{sec:opencbv,sec:usefulcbv} respectively. 
We show that usefulness in \UOCBV is a (complete) restriction of 
\LOCBV: \UOCBV evaluates to equivalent normal forms as in \LOCBV 
while omitting certain substitution steps, specifically those that do 
not contribute to the creation of $\dbsym$-redexes.

Simulating \LOCBV with \UOCBV is straightforward, as each \UOCBV step 
is also a \LOCBV step, as we mention in \cref{sec:usefulcbv}. 
However, relating both strategies in the opposite direction is much 
more delicate, since not all \LOCBV steps are useful: \UOCBV disallows
some substitution steps that \LOCBV allows. 
For example, the step 
$\tm = (\var\,\vartwo)\esub\vartwo\id \tovv{} 
(\var \, \id)\esub\vartwo\id = \tmtwo$ is not useful. 
In fact, $\tm$ is a normal form in \UOCBV (details in 
\cref{prop:TOPLEVEL_useful_nf_unfold_to_nf}), whereas in \LOCBV, 
$\tm$ is a redex whose normal form is $\tmtwo$.

Although normal forms differ between the two formalisms, they are 
\emph{structurally equivalent}, as \LOCBV evaluates terms further than 
\UOCBV. 
To relate them precisely, we restrict the standard notion of unfolding 
to a subtler one.
Recall that in calculi with ESs, the (full) unfolding of a term is an 
operation that performs \emph{all} substitutions. 
For example, the (full) unfolding of
$\var\esub\var{\vartwo\,\vartwo}\esub\vartwo{\varthree\,\varthree}$ is 
$\varthree\,\varthree\,(\varthree\,\varthree)\esub\var{\vartwo\,\vartwo}\esub\vartwo{\varthree\,\varthree}$. 
In this section, however, the notion unfolding must be defined in a
controlled way. 
For instance, unfolding the term $\var\esub\var\id$, which is a 
normal form in \UOCBV, yields its corresponding normal form 
$\id\esub\vartwo\id$ in \LOCBV.
Yet the new unfolding operation does not need to unfold \emph{all} 
ESs: for example, in \LOCBV, $(\lam\var\vartwo)\esub\vartwo\id$ is 
\emph{not} unfolded to $(\lam\var\id)\esub\vartwo\id$ (due to weak evaluation),
and $\var\esub\var{\vartwo\,\vartwo}$ is not unfolded to
$(\vartwo\,\vartwo)\esub\var{\vartwo\,\vartwo}$ (since variables are 
never substituted by structures). 
Intuitively, the new notion of unfolding selectively performs only 
those substitutions on \emph{reachable} variables bound to values,
thus we name it \defn{partial unfolding}.

The remainder of this section is organised as follows.
We first define the partial unfolding of a term with respect to a 
value assignment.
We then relate \LOCBV and \UOCBV via two technical results 
(\cref{prop:TOPLEVEL_useful_nf_unfold_to_nf}), showing that (1) the 
partial unfolding of normal forms in \UOCBV always yields a normal 
form in \LOCBV, and (2) redexes in \UOCBV remain reducible in \LOCBV.

\paragraph{\bf Partial Unfolding}
To formalise our approach, we define a \defn{value assignment}, 
written $\valas$, as a partial function that maps variables to values.
The \defn{domain} and \defn{image} of a value assignment $\valas$
are denoted by $ \dom\valas$ and $\im\valas$, respectively.
To ensure idempotence, we require that $\dom\valas$ and the free 
variables in $\im\valas$ to be disjoint, \ie, 
$\dom\valas \cap \fv{\im\valas} = \emptyset$.
The empty value assignment is denoted with $\evalas$.

We now define the \defn{partial unfolding} of a term $\tm$ under a 
value assignment $\valas$, written $\unvalas\tm$, as follows:
\[ 
  \begin{array}[t]{r@{\,\,}c@{\,\,}l}
      \unvalas\var
    & \eqdef
    & \left\{
        \begin{array}{ll}
          \valas(\var) & \text{if } \var \in \dom\valas
        \\
          \var         & \text{otherwise}
        \end{array}
      \right.
  \\
      \unvalas{(\lam\var\tm)}
    & \eqdef
    & \lam\var\tm
  \\
      \unvalas{(\tm \, \tmtwo)}
    & \eqdef
    & \unvalas\tm \, \unvalas\tmtwo 
  \\
      \unvalas{\tm\esub\var\tmtwo}
    & \eqdef
    & \left\{
        \begin{array}{ll}
          \unvalas[\valas \cup (\var \mapsto \val)]
            \tm\esub\var{\val}\sctx
            & \text{if } \unvalas\tmtwo = \val\sctx , \var \in \rv\tm
        \\
          \unvalas\tm\esub\var{\unvalas\tmtwo}
            & \text{otherwise}
        \end{array}
      \right.
  \end{array} 
\]
This definition has two base cases.
First, the partial unfolding causes no effect on abstractions, 
reflecting the fact that evaluation is weak. 
Second, when partially unfolding a variable $\var$, if 
$\var \in \dom\valas$, the definition mimics a substitution step of 
the form $\var \tovv{\rulesub\var{\valas(\var)}} \val$; otherwise, 
$\var$ remains unchanged, meaning it will not be substituted by a value.

Expanding $\valas$ with $\var \mapsto \val$ still produces a value 
assignment. 
By $\alpha$-conversion, we assume $\var \notin \dom\valas$,
$\var \notin \fv\val$, and for any $\val' \in \im\valas$, 
$\var \notin \fv{\val'}$.
Partial unfolding at top-level is always performed under the empty 
value assignment. 
Some examples of such partial unfoldings follow:
\[
  \unvalasevalas{(\lam\varthree\var)\esub\var\vartwo}
  =
  (\lam\varthree\var)\esub\var\vartwo
  \HS\HS
  \unvalasevalas{(\var \, \vartwo)\esub\vartwo\id}
  =
  (\var \, \id)\esub\vartwo\id
  \HS\HS
  \unvalasevalas{\var\esub\var{\vartwo\esub\varthree\id}}
  =
  \vartwo\esub\var\vartwo\esub\varthree\id
\]

In some cases, the left subterm of a closure is partially unfolded 
under the original value assignment, as in the first example,
where $\var \notin \rv{\lam\varthree\var}$. 
However, in other cases, the value assignment must be extended.
In the last example, partially unfolding the subterm $\var$ extends 
the value assignment to $(\var \mapsto \vartwo)$.

At top-level, the partial unfolding of any term $\tm$ is the (unique) 
$\tovv\lsvsym$ normal form of $\tm$; see \cref{sec:rewriting-tovalas} 
for details.

\paragraph{\bf Relating Reduction Steps and Normal Forms}
We now formally state the existing relation between \LOCBV and \UOCBV. 
Our goal is to show that \UOCBV simulates \LOCBV while reaching 
equivalent normal forms, up to partial unfolding. 
To do this, we state the following result, consisting of two parts: 
(1) the partial unfolding of a useful normal form is a \nonuseful 
normal form; (2) the partial unfolding of a reducible term in \UOCBV 
is either a $\dbsym$-redex in \LOCBV or an abstraction in an applied 
position.
This second condition justifies the name ``useful'' in \UOCBV because 
it implies that any substitution step must contribute to the creation 
of a $\dbsym$-redex (see \cref{sec:introduction}).

Recall that $\tovvTop$ and $\tovTop\sset$ denote the top-level 
\LOCBV and \UOCBV reduction, respectively.
The following proposition links \LOCBV and \UOCBV reduction at the 
top-level case:

\begin{proposition}
\label{prop:TOPLEVEL_useful_nf_unfold_to_nf}
\label{prop:TOPLEVEL_unfold_reducible_is_reducible}
Let $\tm$ be a term and $\sset = \fv\tm$. 
Then:
\begin{enumerate}
\item
  \label{prop:TOPLEVEL_useful_nf_unfold_to_nf:part}
  If $\tm$ is $\tovTop\sset$-irreducible,
  then $\unvalasevalas\tm$ is $\tovvTop$-irreducible.
\item
  \label{prop:TOPLEVEL_unfold_reducible_is_reducible:part}
  If $\tm \tovTop\sset \tm'$, then there exists $\tm''$
  such that $\unvalasevalas\tm \tovv\dbsym \tm''$.
\end{enumerate}
\end{proposition}
This statement is generalised for arbitrary parameters 
($\aset$, $\sset$, $\appflag$, etc.) in the appendix 
(\cref{prop:useful_nf_unfold_to_nf}), which has some subtleties.

To illustrate statement \ref{prop:TOPLEVEL_useful_nf_unfold_to_nf:part}, let 
us take the term $(\var\,\vartwo)\esub\vartwo\id$, which is 
$\tovTop{\set\var}$-irreducible.
Its partial unfolding is $(\var\,\id)\esub\vartwo\id$, which is 
$\tovvTop$-irreducible.
As an example of statement 
\ref{prop:TOPLEVEL_unfold_reducible_is_reducible:part},
consider the $\lsvsym$-step
$\tm = (\var\,\vartwo)\esub\var\id \tovTop{\set\vartwo}
(\id \, \vartwo)\esub\var\id = \tm' = \unvalasevalas\tm$;
note that $\tm' \tovv\dbsym \var_1\esub{\var_1}\var\esub\vartwo\id$.

From these propositions, we conclude:
\begin{corollary}
A term $\tm$ is $\tovTop{\fv\tm}$-irreducible if and only if
$\unvalasevalas\tm$ is $\tovvTop$-irreducible.
\end{corollary}

\section{\UOCBV is Invariant}
\label{sec:usefulcbv_invariant}

In this section, we relate our \UOCBV strategy to the \glamour 
abstract machine introduced in~\cite{AccattoliC15}, which is an 
existing implementation of usefulness.
As a consequence, we yield that \UOCBV is \emph{time-invariant}.
That is, the number of reduction steps to normal form in \UOCBV can
be used as a measure of time complexity.

We proceed in two stages, following~\cite{AccattoliC15}.
On one hand, we prove a \textbf{high-level implementation} theorem
(\cref{thm:high_level_implementation}),
stating that reduction in the fireball calculus~\cite{AccattoliC15} 
can be simulated by reduction in \UOCBV, with \emph{quadratic} 
overhead in time.
On the other hand, we prove a \textbf{low-level implementation} theorem
(\cref{thm:low_level_implementation}),
stating that reduction in \UOCBV can be implemented by the \glamour 
abstract machine, with \emph{linear} overhead in time.
By composing these results, we obtain that \UOCBV is an invariant
implementation of open \CBV (\cref{coro:ucbv_invariant}).
The key property in this section is a simulation result 
(\cref{glamour_simulation}), which embeds the \glamour abstract machine 
into \UOCBV. 

We state the main results, while most of the technical details are 
in~\cref{app:usefulcbv_invariant}.
The definitions and proofs in this section follow well-known 
methodologies~(\eg, \cite{AccattoliBM14,AccattoliC15,AccattoliG17}).

\paragraph{\bf Stability Notions}
Given an abstraction frame $\aset$ and a structure frame $\sset$, a
term $\tm$ is said to be \defn{\stabilized} under $(\aset,\sset)$ if 
$\tm \in \HAbs\aset$ or $\tm \in \Struct\sset$.
In \UOCBV, evaluating a closure $\tm\esub\var\tmtwo$ proceeds in the 
body $\tm$ only when the argument $\tmtwo$ is rigid.
To better align with low-level abstract machines, we define for any
abstraction frame $\aset$ and structure frame $\sset$ the set of 
\defn{stable terms} under $(\aset,\sset)$ as those in which every ES 
argument is rigid under $(\aset,\sset)$.
We introduce a reduction relation we dub \defn{stable reduction} and
is written ${\tostable\rulename\aset\sset\appflag}$, such that 
${\tostable\rulename\aset\sset\appflag} \subseteq {\tov\rulename\aset\sset\appflag}$.
Sometimes, we write $\toustable$ if the parameters are clear from the context.
Stable reduction forces arguments of applications to be \stabilized 
reduction proceeds.
Therefore, in this setting, rule \ruleUAppL allows to reduce the head 
$\tm$ of an application $\tm\,\tmtwo$ only if $\tmtwo$ is \stabilized,
and rule \ruleUDb allows to contract a $\dbsym$-redex only if its
argument is rigid. Stable reduction preserves stable terms and enjoys 
the diamond property.

\paragraph{{\bf Embedding the \glamour into useful Open \CBV} }
We start by briefly recalling the syntax of the \glamour.
Full details on the \glamour transitions and the remaining definitions 
can be found in~\cite{AccattoliC15,AccattoliG17} and \cref{app:usefulcbv_invariant}.

The syntax of the \glamour is given by the following
grammar, defining
\defn{states} $(\state, \statetwo, \hdots)$,
\defn{dumps} $(\dump, \dumptwo, \hdots)$,
\defn{stacks} $(\stack, \stacktwo, \hdots)$, 
\defn{stack items} $(\stackitem, \stackitemtwo, \hdots)$, and
\defn{global environments} $(\genv, \genvtwo, \hdots)$:
\[
  \begin{array}{rclcrcl}
    \state     & \eqgram & \glamourst\dump\code\stack\genv
  & \HS &
    \stackitem & \eqgram & \code   \mid \pair\code\stack
  \\
    \dump      & \eqgram & \estack \mid \dump \cons \pair\code\stack
  & \HS &
    \genv      & \eqgram & \estack \mid \esub\var\stackiteml \cons \genv
  \\
    \stack     & \eqgram & \estack \mid \stackiteml \cons \stack
  \end{array}
\]
where $\lab \in \set{\alive,\dead}$ are \defn{labels} decorating stack items.
Here, \defn{codes} $(\code, \codetwo, \hdots)$ are terms with no ESs, 
but they are \textbf{not} considered up to $\alpha$-equivalence.
By convention, labels indicate the shape of $\stackitem$. 
In particular, $\lab = \alive$ if and only if $\stackitem$ is of the 
form $\code$, and $\lab = \dead$ if and only if $\stackitem$ is of the 
form $\pair{\code}\stack$.
Intuitively, these labels indicate whether a stack item unfolds to 
a hereditary abstraction ($\alive$) or a structure ($\dead$).
In a state $\state = \glamourst\dump\code\stack\genv$,
the code $\code$ is called the \defn{focus} of $\state$.
A state $\state_0$ is \defn{initial} if it is of the form 
$\state_0 = \glamourst\estack\code\estack\estack$.

The decode function $\decode{\_}$ operates over the syntactic elements
of the \glamour and returns terms in the syntax of \UOCBV. 
The function is defined recursively as:
\[
  \begin{array}{rclcrcl}
    \decode{\glamourst{\dump}{\code}{\stack}{\genv}} & \eqdef & \decodep{\genv}{\decodep{\dump}{\decodep{\stack}{\tm}}}
  & \HS &
    \decode{\esub{\var}{\stackiteml} \cons \genv}    & \eqdef & \decodep{\genv}{\ectx\esub{\var}{\decode{\stackiteml}}}
  \\
    \decode{\estack}                                 & \eqdef & \ectx
  & \HS &
    \decode{\pair{\code}{\stack}^\lab}               & \eqdef & \decodep{\stack}{\tm}
  \\
    \decode{\dump \cons \pair{\code}{\stack}}        & \eqdef & \decodep{\dump}{\decodep{\stack}{\tm\ectx}}
  & \HS &
    \decode{\code^\lab}                              & \eqdef & \tm 
  \\
    \decode{\stackiteml \cons \stack}                & \eqdef & \decodep{\stack}{\ectx\decode{\stackiteml}}
  & \HS &
    \decode{\code}                                   & \eqdef & \tm
  \end{array}
\]
In the two last lines on the right column, $\tm$ denotes a term which 
is $\alpha$-equivalent to $\code$.

As is typical when relating abstract machines and calculi with ESs 
(see \eg,~\cite{AccattoliBM14}), the \glamour abstract machine can be 
simulated in \UOCBV up to \defn{structural equivalence} 
of terms~\cite{AccattoliG16}, denoted $\equiv$,
defined as the least equivalence relation closed by weak contexts
(\ie, contexts that do not enter inside abstractions) of the four following
axioms:
\[
  \begin{array}{rcl@{\HS}rcl}
  \tm\esub\var\tmtwo\esub\vartwo\tmthree & \equiv & \tm\esub\vartwo\tmthree\esub\var\tmtwo
  &
  \tm\esub\var\tmtwo\esub\vartwo\tmthree & \equiv & \tm\esub\var{\tmtwo\esub\vartwo\tmthree}
  \\
  (\tm \, \tmtwo)\esub\var\tmthree & \equiv & \tm\esub\var\tmthree \, \tmtwo
  &
  (\tm \, \tmtwo)\esub\var\tmthree & \equiv & \tm \, \tmtwo\esub\var\tmthree
  \end{array}
\]
All these axioms assume that variable-capture is avoided,
for example, the first equation assumes that $\var \notin \fv\tmthree$ 
and $\vartwo \notin \fv\tmtwo$.

There are seven kinds of transitions in the \glamour:
transitions ($\state \tomachum \state'$) for function applications,
transitions ($\state \tomachue \state'$) for linear substitutions,
and \emph{administrative} transitions
$\state \tomachhole{\admsym_i} \state'$ with $i \in \set{1..5}$,
for changing the evaluation focus without computing.
Formally, the transitions of the \glamour abstract machine are:
\[
  {\small
    \begin{array}{c|c|c|ccc|c|c|ccc}
      \usefmaca\genv\dump{\code \, \codethree}\stack
    & \tomachcone &
     \usefmaca\genv{\dump\cons\pair\code\stack}\codethree\estack
    \\
      \usefmaca\genv\dump{\tocode{\lam\var\tm}}{\stackiteml\cons\stack}
    & \tomachum &
      \usefmaca{\esub\var\stackiteml \cons \genv}\dump\code\stack
    \\
      \usefmaca\genv{\dump \cons \pair\code\stack}{\tocode{\lam\var\tmthree}}\estack
    & \tomachctwo &
      \usefmaca\genv\dump\code{\herval{(\tocode{\lam\var\tmthree})}\cons\stack}
    \\
      \usefmaca\genv{\dump\cons\pair\code\stack}{\tocode\var}\stacktwo
    & \tomachcthree &
      \usefmaca\genv\dump\code{\pair{\tocode\var}\stacktwo^\dead\cons\stack}
    \\
      \usefmaca{\genv_1 \cons \esub\var{\phi^\dead} \cons \genv_2}{\dump\cons\pair\code\stack}{\tocode\var}\stacktwo
    & \tomachcfour &
      \usefmaca{\genv_1 \cons \esub\var{\phi^\dead} \cons \genv_2}\dump\code{\pair{\tocode\var}\stacktwo^\dead\cons\stack}
    \\
      \usefmaca{\genv_1 \cons \esub\var{\herval\codethree} \cons \genv_2}{\dump\cons\pair\code\stack}{\tocode\var}\estack
    & \tomachcfive &
      \usefmaca{\genv_1 \cons \esub\var{\herval\codethree} \cons \genv_2}\dump\code{\herval{\tocode\var}\cons\stack}
    \\
      \usefmaca{\genv_1 \cons \esub\var{\herval\codethree} \cons \genv_2}\dump{\tocode\var}{\stackiteml\cons\stack}
    & \tomachue &
      \usefmaca{\genv_1 \cons \esub\var{\herval\codethree} \cons \genv_2}\dump{\rename\codethree}{\stackiteml\cons\stack}
    \end{array}
  }
\]
where $\rename\codetwo$ is any code $\alpha$-equivalent to
$\codetwo$ that preserves well-naming of the machine. 
Moreover, in the transition $\tomachcthree$, the variable $\tocode\var$
does not belong to the domain of the environment $\genv$. 
The domain of an environment $\genv$ is written $\domSctx\genv$ and 
is defined recursively by $\domSctx\estack \eqdef \emptyset$
and $\domSctx{\esub\var\stackiteml \cons \genvtwo} \eqdef \set\var \cup \domSctx\genvtwo$.

The main technical lemma is:

\begin{restatable}[\glamour simulation]{lemma}{glamourSimulation}
\label{glamour_simulation}
Let $\state$ be a state reachable from an initial state
whose focus is $\tm_0$ and let $\sset_0 \eqdef \fv{\tm_0}$. Then:
\begin{enumerate}
\item
  If $\state \tomachum \state'$,
  then $\decode\state \toustable_\ruledb\equiv \decode{\state'}$.
\item
  If $\state \tomachue \state'$, 
  then $\decode\state \toustable_\rulelsv\equiv \decode{\state'}$.
\item
  If $\state \tomachhole{\admsym_i} \state'$,
  then $\decode\state = \decode{\state'}$, for all $i \in \set{1..5}$.
\item
  If $\state$ is $\tomachhole{}$-irreducible,
  then $\decode\state$ is $\toustable$-irreducible (progress property).
\end{enumerate}
\end{restatable}

As a consequence, any \emph{sequence} of \glamour transitions
corresponds to a \emph{sequence} of (stable) \UOCBV reduction steps 
interleaved with equivalences.
To \emph{postpone} all the intermediate structural equivalence steps
to obtain 
$\decode{\state_1} \toustable\toustable \hdots \toustable \equiv \decode{\state_n}$,
we need the following lemma, stating that $\equiv$ is a 
\emph{strong bisimulation} with respect to $\toustable$, thus 
obtaining in particular a postponement property:

\begin{restatable}[Strong bisimulation]{lemma}{structeqBisimulation}
\label{structeq_bisimulation}
Let $\tm_0, \tmtwo_0$ be stable terms such that
$\tmtwo_0 \equiv \tm_0$.
If $\tm_0 \tostable\rulename\aset\sset\appflag \tm_1$, then there 
exists $\tmtwo_1$ such that
$\tmtwo_0 \tostable\rulename\aset\sset\appflag \tmtwo_1 \equiv \tm_1$.
\end{restatable}

This result relies crucially on the stability notion introduced 
earlier in this section. 
For instance, suppose $\tmthree$ is not \stabilized. 
Then, the equivalence $\tm\esub\var\tmtwo\esub\vartwo\tmthree \equiv
\tm\esub\vartwo\tmthree\esub\var\tmtwo$ cannot be postponed after the 
(non-stable) step $\tm\esub\vartwo\tmthree\esub\var\tmtwo \tou
\tm\esub\vartwo\tmthree\esub\var{\tmtwo'}$ because the rules
\ruleUEsLAbs and \ruleUEsLStruct cannot be applied to derive a step
$\tm\esub\var\tmtwo\esub\vartwo\tmthree \tou 
\tm\esub\var{\tmtwo'}\esub\vartwo\tmthree$, given that $\tmthree$ is 
not \stabilized.

\paragraph{{\bf High and Low-Level Implementation}}
We can now show the main results of this section. 
For that, we first say that a term is \defn{pure} if it contains no ESs.

Relying on the previous results, it can now be shown that reduction 
in the fireball calculus ($\tobetafireball$) can be implemented 
through \UOCBV ($\tou$) with quadratic overhead in time.
More precisely, if we write $\unfold\tm$ for the \emph{full} unfolding 
of a term $\tm$, performing \emph{all} the pending ESs in $\tm$:
\begin{restatable}[High-Level Implementation]{theorem}{highLevelImplementation}
\label{thm:high_level_implementation}
Let $\tm$ be a pure term and $\sset = \fv\tm$.
If $\tm \tobetafireball^n \tm'$, then there exists $\tmtwo$ such that 
$\tm \toustablen{k} \tmtwo$ where $\unfold\tmtwo = \tm'$ and 
$k \in O(|\tm| \cdot (n^2 + 1))$.
\end{restatable}

Now, we show that \UOCBV ($\tou$) can be implemented in the \glamour 
($\tomachhole{}$) with bilinear overhead in time, as a function of 
the size of the starting term and the length of the reduction sequence 
to normal form:
\begin{restatable}[Low-Level Implementation]{theorem}{lowLevelImplementation}
\label{thm:low_level_implementation}
Let $\tm$ be a pure term.
If $\tm \toustablen{n} \tm'$, with $\tm'$ in normal form
and $\state$ is an initial state such that $\decode\state = \tm$,
then $\state \tomachhole{}^k \statetwo$, where $\decode\statetwo$ is 
structurally equivalent to $\tm'$ and $k \in O(|\tm| \cdot (n + 1))$.
\end{restatable}

This result entails that \UOCBV is time-invariant, \ie, that it can 
be simulated with at most polynomial overhead by a time-invariant 
cost model (such as Turing machines or RAMs):
\begin{corollary}
\label{coro:ucbv_invariant}
The \UOCBV strategy is an invariant implementation of open \CBV.
\end{corollary}

\section{A Quantitative Interpretation}
\label{sec:typing}

In this section, we present a non-idempotent intersection type system 
called $\typesystem$ (for $\typesystem$seful). 
We show that it is sound (\cref{thm:soundness_typing}) and complete 
(\cref{thm:completeness_typing}) with respect to reduction in \UOCBV.
Therefore, we yield in \cref{coro:quantitative_characterisation_termination} 
that system $\typesystem$ characterises termination in $\UOCBV$: 
a term $\tm$ is \defn{typable} in system $\typesystem$ (or 
\defn{$\typesystem$-typable}) if and only if $\tm$ terminates in \UOCBV.
In this sense, we say that system $\typesystem$ is a \emph{quantitative 
interpretation} of \UOCBV.

\subsection{Defining System $\typesystem$}
\label{sec:type_system_definition}

In quantitative type systems, one key feature is that a single 
occurrence of an expression in the (static) source code of the program 
may be (dynamically) used many times during runtime, each time playing 
a different role, expressed by a different type. 
For example, in some programs an occurrence
of the identity function $\id$ may be applied twice, while in other
programs it may be applied only once.  In each case, the type of the
subexpression $\id$ should change to reflect this quantitative
difference.  Hence \emph{terms do not have a unique type}.  In fact,
as usual in intersection type systems, there is no notion of
\emph{principal type} in $\typesystem$.

To be able to capture quantitative information about evaluation,
the type of each $\lambda$-abstraction
is a \emph{multiset} (rather than a \emph{set}) whose cardinality
corresponds exactly to the number of times that the abstraction is
applied to some argument during the whole
evaluation process.  In general, the type of
an abstraction is a multiset of the form $\mtyp
= \mset{\typ_1,\hdots,\typ_n}$, where each of the $\typ_i$ is
an \emph{arrow type}.  For example, the underlined identity function
in the expression
$(\lam{f}{\var\,(f\,\vartwo)\,(f\,\varthree)})\,\underline{\id}$ takes
part in \emph{two} function applications. Hence, one possible type
assignment for that subexpression is $\vdash \id :
[(\mtyp\to\mtyp),(\mtyptwo\to\mtyptwo)]$, meaning that the identity
function is applied \emph{twice} during evaluation,
once to an argument of type $\mtyp$ and once to an argument
of type $\mtyptwo$. However, some abstractions may never be
applied. For example, the identity function in
the program $\var\esub{\var}{\id}$ does not take part in any function
application.  These abstractions are typed with the \emph{empty} multiset
$\mset{}$.

As studied in \cref{sec:usefulcbv},
all (terminating) terms evaluate to either a \emph{variable}, 
an \emph{abstraction}, or a \emph{structure}.
Abstractions, as well as variables bound to abstractions,
are assigned finite multisets of arrow types, as we just said, called \emph{arrow multi-types}.
Structures, and variables bound to structures, on the other hand,
are always given a distinguished type, just written $\tightN$.

Formally, the types of system $\typesystem$ are given by the 
following grammar:
\[
  \begin{array}{rlcl@{\sep\sep}rlcll}
    \text{\defn{(Arrow Types)}}       & \typ, \typtwo & \eqgram & \optmtyp \to \mtyp 
  \\
    \text{\defn{(Arrow Multi-Types)}} & \nityp, \nityptwo & \eqgram & \mset{\typ_i}_\iI & \text{where $I$ is a finite set}
  \\ 
    \text{\defn{(Types)}}             & \mtyp, \mtyptwo & \eqgram & \tightN \mid \nityp 
  \\
    \text{\defn{(Optional Types)}}    & \optmtyp, \optmtyptwo & \eqgram & \none\ \mid \mtyp
  \end{array}
\] 

We distinguish a set of \defn{tight constants} given by 
$\tight \eqgram \tightN \mid \emset$, where $\tightN$ is assigned to 
terms evaluating to structures, while $\emset$ is assigned to terms 
evaluating to abstractions that are not going to be applied.
Unlike in previous non-idempotent intersection type systems for \CBV, 
we distinguish the \emph{optional} type $\none$ from $\emset$, the 
former meaning that \emph{no typing information} is available, and 
the latter meaning that it is typing an abstraction in a non-applied 
position.

\defn{Typing environments} $\tctx, \tctxtwo, \ldots$ are functions 
mapping variables to optional types, assigning $\none$ to all but 
finitely many variables. 
The \defn{domain} of an environment $\tctx$ is defined as
$\dom\tctx \eqdef \set{\var \mid \tctx(\var) \neq \none}$, and
$\emptyctx$ denotes the empty typing environment, mapping every
variable to $\none$.
\defn{Type assumptions} are denoted $\var : \optmtyp$, meaning that the 
environment assigns $\optmtyp$ to $\var$, and $\none$ to any other variable. 

The \defn{union of arrow multi-types}, written 
$\nityp_1 \niunion \nityp_2$, is a multiset of types defined as
expected, where $\emset$ is the neutral element, and $\niunion$ denotes
the union of multisets. 
The \defn{union of types}, written $\mtyp_1 + \mtyp_2$, is the 
(associative) \emph{partial} operation on types given by
$\tightN + \tightN \eqdef \tightN$ and 
$\nityp_1 + \nityp_2 \eqdef \nityp_1 \niunion \nityp_2$, with all
other cases undefined. 
The \defn{union of optional types} is given by
$\none + \none \eqdef \none$, $\none + \mtyp \eqdef \mtyp$, and
$\mtyp + \none \eqdef \mtyp$, so that $\none$ is the neutral element. 
For typing environments $(\tctx_i)_\iI$, we write $+_\iI \tctx_i$ for 
the environment mapping each variable $\var$ to
$+_\iI \tctx_i(\var)$.
Additionally, $\tctx + \tctxtwo$ and $\tctx+_\jJ \tctxtwo_j$ are 
particular instances of the general notation. 
When $\dom\tctx \cap \dom\tctxtwo = \emptyset$ we write 
$\tctx;\tctxtwo$ instead of $\tctx + \tctxtwo$ to emphasise that
the domains of $\tctx$ and $\tctxtwo$ are disjoint. 
As a consequence, $\tctx;\var:\none$ is identical to $\tctx$.

The \defn{counting} function $\numarrows{\_}$ returns the number of
{\bf t}op-level {\bf a}rrows in a (optional) type, and it is defined 
by $\numarrows\tightN \eqdef 0$, $\numarrows\none \eqdef 0$, and
$\numarrows{\mset{\typ_i}_\iI} \eqdef |I|$.
Note that $\numarrows{\mtyp_1 + \mtyp_2} = \numarrows{\mtyp_1} +
\numarrows{\mtyp_2}$. 

We define a \defn{controlled weakening} relation $\mleq$ between 
optional types and types by declaring that $\none \mleq \tight$ and 
$\mtyp \mleq \mtyp$ hold.

\defn{Typing judgements} in system $\typesystem$ are of the form 
$\judg[\cm,\ce]\tctx\tm\mtyp$, where the usual sequent notation of typing 
judgements is now decorated with natural numbers $\cm$ and $\ce$ 
called \defn{counters}.
Under appropriate conditions, these counters correspond \emph{exactly} 
to the number of $\dbsym$-steps ($\cm$) and $\lsvsym$-steps ($\ce$) 
in \UOCBV that are required to reduce terms to normal form. 
\defn{Typing rules} of system $\typesystem$ are:
\[
  \inferrule{
    n = \numarrows\mtyp
  }{
    \judg[0,n]{\var:\mtyp}\var\mtyp
  }\ruleTypVar
  \HS
  \inferrule{
    \left(
      \judg[\cm_i,\ce_i]{\tctx_i;\var:\optmtyp_i}\tm{\mtyptwo_i}
    \right)_\iI
  }{
    \judg[+_\iI \cm_i, +_\iI \ce_i]
      {+_\iI\tctx_i}{\lam\var\tm}{\mset{\optmtyp_i \to \mtyptwo_i}_\iI
    }
  }\ruleTypAbs
  \HS
  \inferrule{
    \judg[\cm,\ce]\tctx\tm\tightN
    \sep
    \judg[\cm',\ce']\tctxtwo\tmtwo\tight
  }{
    \judg[\cm + \cm', \ce + \ce']{\tctx + \tctxtwo}{\tm \, \tmtwo}\tightN
  }\ruleTypAppP
\]
\[  
  \inferrule{
    \judg[\cm,\ce]\tctx\tm{\mset{\optmtyp \to \mtyptwo}}
    \sep
    \optmtyp \mleq \mtyp
    \sep
    \judg[\cm',\ce']\tctxtwo\tmtwo\mtyp
  }{
    \judg[1+ \cm + \cm', \ce + \ce']{\tctx + \tctxtwo}{\tm \, \tmtwo}\mtyptwo
  }\ruleTypAppC
  \HS
  \inferrule{
    \judg[\cm,\ce]{\tctx;\var:\optmtyp}\tm\mtyptwo
    \HS
    \optmtyp \mleq \mtyp
    \HS
    \judg[\cm',\ce']\tctxtwo\tmtwo\mtyp
  }{
    \judg[\cm + \cm',\ce + \ce']{\tctx + \tctxtwo}{\tm\esub\var\tmtwo}{\mtyptwo}
  }\ruleTypES
\]

A \defn{(type) derivation} is a tree built by applying the previous 
(inductive) typing rules.
We write $\derivs\deriv{\judg[\cm,\ce]\tctx\tm\mtyp}$ to denote a 
type derivation $\deriv$ ending in the typing judgement 
$\judg[\cm,\ce]\tctx\tm\mtyp$.
Sometimes we may write $\derivs{}{\judg[\cm,\ce]\tctx\tm\mtyp}$, 
omitting the name $\deriv$.

Like most non-idempotent intersection type systems, system 
$\typesystem$ is inspired by Linear Logic~\cite{Girard88}, and 
specifically by the relational semantics for \CBV in~\cite{Ehrhard12}.
Indeed, a multiset $\mset{\typ_1,\hdots,\typ_n}$ can be understood as 
a multiplicative conjunction $\typ_1 \otimes \hdots \otimes \typ_n$.
Note also that there are no general \emph{weakening} and \emph{contraction} 
rules, though $\none \mleq \tight$ and $\tight + \tight = \tight$ can
be understood respectively as controlled forms of weakening and contraction.

Rule \ruleTypAbs is a standard rule in non-idempotent intersection 
type systems for \CBV, and reflects the \emph{weak} nature of reduction
in \UOCBV since reduction does not proceed within abstractions.
For example, the derivation $\derivs{}{\judgs\emptyctx{\lam\var\Omega}\emset}$
is valid even if $\Omega$ is non-terminating.

Rule \ruleTypES can be easily derived from rules \ruleTypAbs and \ruleTypAppC. 
Rules for applications follow the \emph{consuming/persistent paradigm}~\cite{KesnerVial20, KesnerViso22} 
and thus deserve some discussion. 
Indeed, system $\typesystem$ tracks different natures of the 
application constructors involved in the evaluation process of a program: 
a term constructor is \defn{consuming} if it is \emph{destroyed} during 
evaluation, whereas a \defn{persistent} constructor \emph{remains} as 
a syntactical subterm of the normal form. 
For example, given the reduction to normal form
$\var ((\lam\vartwo\vartwo) \varthree) \tou \var (\vartwo\esub\vartwo\varthree)$,
the leftmost application constructor is persistent, while the rightmost is consuming. 
Hence, rule \ruleTypAppP (resp. \ruleTypAppC) is used to type persistent
(resp. consuming) applications.
More specifically, \ruleTypAppP requires the left-hand side of the 
application to have type $\tightN$, ensuring that no redex will be created, 
so the (typed) application constructor is persistent. 
Rule \ruleTypAppC requires the left-hand side of the application to type
with an arrow multi-type, ensuring a $\dbsym$-redex is created at some
point during evaluation, and thus the first counter is incremented.

As just explained, system $\typesystem$ uses the consuming/persistent 
paradigm exclusively to type applications.
This approach contrasts with more radical uses for \CBV, such as~\cite{KesnerViso22},
which apply the same paradigm to type \emph{all} constructors, thereby
increasing the number of typing rules and making the typing system more
difficult to use.

We now discuss the use of the controlled weakening relation in the 
typing rules \ruleTypAppC and \ruleTypES.
In \CBV, the arguments of applications and closures are always 
evaluated. Since system $\typesystem$ intends to characterise 
termination in \UOCBV, all arguments must necessarily be typed.
For example, in $(\lam\var\id)\,(\varthree\,\varthree)$ the 
subexpression $\varthree\,\varthree$ is a structure of type $\tightN$,
so one would like to type the body of the abstraction
as $\judgs{\var:\tightN}\id\mtyp$.
But it is impossible to derive this judgement because $\var$ does not 
occur free in $\id$.
This is why we introduce the controlled weakening relation as a premise 
in rules \ruleTypAppC and \ruleTypES.
In our example, one can give a ``stronger'' type to $\lam\var\id$, 
namely $\mset{\emset \to \mtyp}$, using the fact that $\none \mleq \tightN$.
A similar use of controlled weakening for the rule \ruleTypES is shown 
below, where it should be recalled that the empty typing environment 
$\emptyctx$ is identical to $\emptyctx; \var : \none$.

Assuming that $\none \mleq \tightN$,
one typing derivation that illustrates the use of subsumption is the following:
\[
  \inferrule{
    \inferrule*[Right = \ruleTypAbs]{
    }{
      \judg[0,0]{\emptyctx;\var : \none}\id\emset
    }
    \HS
    \none \mleq \tightN
    \HS
    \inferrule*[Right = \ruleTypAppP]{
      \inferrule*[Right = \ruleTypVar]{
      }{
        \judg[0,0]{\varthree : \tightN}\varthree\tightN
      }
      \HS
      \inferrule*[Right = \ruleTypVar]{
      }{
        \judg[0,0]{\varthree : \tightN}\varthree\tightN
      }
    }{
      \judg[0,0]{\varthree : \tightN}{\varthree \, \varthree}\tightN
    }
  }{
    \judg[0,0]
      {\varthree : \tightN}
      {\id\esub\var{\varthree \, \varthree}}
      \tightN
  }\ruleTypES
\]

We end this section by mentioning some basic properties and notions
of the type system $\typesystem$.

\paragraph{\bf Relevance.}
The fact that axioms in system $\typesystem$ are not weakened and that
rules are \emph{multiplicative} yields that system $\typesystem$ is
\emph{relevant}. Formally:

\begin{restatable}[Relevance]{lemma}{relevance}
\label{lem:relevance}
If $\derivs\deriv{\judg[\cm,\ce]\tctx\tm\mtyp}$,
then $\rv\tm \subseteq \dom\tctx \subseteq \fv\tm$.
\end{restatable}
Note that $\fv\tm \subseteq \dom\tctx$ does not always hold. 
For instance, $\derivs{}{\judgs\emptyctx{\lam\var\vartwo}\emset}$ but 
$\set\vartwo \nsubseteq \emptyset$.

\paragraph{\bf Appropriateness.}
To reason inductively about typing derivations, it is sometimes 
necessary to ensure invariants for the types of the free variables in 
a term. 
For example, for a term like $\tm\esub\var\id$, we may need to track 
that $\var$ is bound to an abstraction when applying the inductive 
hypothesis for the subterm $\tm$.
Specifically, we say that a typing environment $\tctx$ is 
\defn{appropriate} with respect to an \abstractionframe $\aset$,
written $\isAppr\aset\tctx$, if for each $\var \in \aset$,
we have $\tctx(\var) \neq \tightN$.
In other words, $\tctx(\var)$ must be $\none$ or an arrow multi-type $\nityp$.
For instance, $\isAppr{\set\var}{\var : \emset, \vartwo : \tightN}$ 
holds, while $\isAppr{\set\var}{\var : \tightN, \vartwo : \tightN}$ 
does not.

\paragraph{\bf Tightness.}
Typing in system $\typesystem$ not only characterises termination in 
\UOCBV, but it is also capable of providing \emph{exact measures} for 
the length of reduction sequences to normal form.
To state this formally, we define a subset of typing derivations 
called \defn{tight derivations}, as follows.
A type is \defn{tight} if it is a tight constant $\tight$. 
An optional type $\optmtyp$ is \defn{tight} if it is either $\none$ or 
a tight type.
A typing environment $\tctx$ is \defn{tight} if $\tctx(\var)$ is 
tight for every variable $\var$.
A typing judgement $\judg[\cm,\ce]\tctx\tm\mtyp$ is \defn{tight} if 
both $\tctx$ and $\mtyp$ are tight.
A derivation $\deriv$ of a typing judgement is \defn{tight} if the 
judgement $\deriv$ derives is tight.
For example, the derivable judgement $\judg[0,0]{\var:\emset}\var\emset$ 
is tight: the counters $(0,0)$ indicate that the evaluating $\var$ 
requires no function application or substitution steps, meaning that 
$\var$ is already in normal form. 
In contrast, the derivable judgement
$\judg[0,1]{\var:\mset{\mtyp \to \mtyptwo}}\var{\mset{\mtyp \to \mtyptwo}}$
is \emph{not} tight, and the counters $(0,1)$ provide an upper bound 
on the number of steps required for evaluating $\var$.

\begin{example}
\label{ex:derivation-for-tight}
Let
$\tm = (\var \, \varthree)\esub{\var}{\vartwo\esub{\vartwo}{\id}}$.
Taking $\mtyp \eqdef \mset{\tightN \to \tightN}$,
the following typing derivation $\deriv$ turns out to be tight:
\[
  \indruleN{\ruleTypES}{
    \inferrule{
      \inferrule{
      }{
        \judg[0,1]{\var : \mtyp}\var\mtyp
      }\ruleTypVar
      \HS
      \inferrule{
      }{
        \judg[0,0]{\varthree : \tightN}\varthree\tightN
      }\ruleTypVar
    }{
      \judg[1,1]
        {\varthree : \tightN, \var : \mtyp}{\var \, \varthree}\tightN
    }\ruleTypAppC
    \HS
    \inferrule{
      \inferrule{
      }{
        \judg[0,1]
          {\vartwo : \mtyp}\vartwo\mtyp
      }\ruleTypVar
      \HS
      \inferrule{
        \inferrule{
        }{
          \judg[0,0]{\varfour : \tightN}\varfour\tightN
        }\ruleTypVar
      }{
        \judg[0,0]\emptyctx\id\mtyp
      }\ruleTypAbs
    }{
      \judg[0,1]\emptyctx{\vartwo\esub\vartwo\id}\mtyp
    }\ruleTypES
  }{
    \judg[1,2]
      {\varthree : \tightN}
      {(\var \, \varthree)\esub\var{\vartwo\esub\vartwo\id}}
      \tightN
  }
\]
where $\tightN \mleq \tightN$ and 
$\mtyp = \mset{\tightN \to \tightN} \mleq \mset{\tightN \to \tightN} = \mtyp$
hold.
Note that all the other subderivations of $\deriv$ are not tight, 
except for the derivation corresponding to the judgement
$\judg[0,1]{\varthree : \tightN}{\varthree}{\tightN}$.
\end{example}

Tightness is essential for proving that the system $\typesystem$ is 
sound and complete with respect to \UOCBV (\cref{thm:soundness_typing,thm:completeness_typing}).

  \subsection{Soundness and Completeness of System $\typesystem$}
  \label{sec:soundness_and_completeness}
  
We first address the \defn{soundness} of system $\typesystem$ with 
respect to \UOCBV, \ie, typability in system $\typesystem$ implies 
\UOCBV-termination.
Additionally, our soundness result provides exact quantitative 
information about the reduction steps to normal form in \UOCBV.
More precisely, we show that when a term $\tm$ is tightly typable in
system $\typesystem$ with counters $(\cm,\ce)$, then $\tm$ evaluates 
in \UOCBV to a normal form $\tmtwo$ in exactly $\cm$ $\dbsym$-steps 
and exactly $\ce$ $\lsvsym$-steps.

The proof of soundness follows well-understood techniques~\cite{BucciarelliKV17}: 
it requires a subject reduction property, based in turn on a 
substitution lemma (omitted here; see \cref{app:typing}). 
Due to the \emph{contextual} nature of \UOCBV, these lemmas are not 
trivial, as they require formulating complex invariants on the 
contextual parameters $\aset$, $\sset$, and $\appflag$ that define 
\UOCBV.

We begin by showing that any \UOCBV step preserves typing and 
correctly decrements the counters.
\begin{restatable}[Subject Reduction]{proposition}{subjectreduction}
\label{prop:subject_reduction}
Let $\tm \tov\rulename\aset\sset\appflag \tm'$, where 
$\rulename \in \set{\ruledb,\rulelsv}$
and $\derivs\deriv{\judg[\cm,\ce]\tctx\tm\mtyp}$, where 
$\isAppr\aset\tctx$.
Suppose moreover that if $\appflag = \app$, then either 
$\mtyp = \tightN$ or $\mtyp$ is a singleton, \ie, of the form $\mset\typ$.
Then, there exists a typing derivation $\deriv'$ such that
$\derivs{\deriv'}{\judg[\cm',\ce']\tctx{\tm'}\mtyp}$,
where if $\rulename = \ruledb$, then $\cm > 0$ and $(\cm',\ce') = (\cm-1,\ce)$,
and if $\rulename = \rulelsv$, then $\ce > 0$ and $(\cm',\ce') = (\cm,\ce-1)$.
\end{restatable}

We can now state the result of quantitative soundness: a tight 
derivation of a term $\tm$ with counters $(\cm,\ce)$ guarantees 
weakly terminating in \UOCBV, with exactly $\cm$ $\dbsym$-steps 
and $\ce$ $\lsvsym$-steps.

\begin{restatable}[Quantitative Soundness of System $\typesystem$]{theorem}{soundnesstyping}
\label{thm:soundness_typing}
Let $\sset = \fv\tm$, and
let $\derivs\deriv{\judg[\cm,\ce]\tctx\tm\mtyp}$ be a tight 
derivation in system $\typesystem$.
Then, there exists a $\tovTop\sset$-irreducible term $\tmtwo$ such 
that $\tm \tovTop\sset^{\cm + \ce} \tmtwo$, where $\cm$ and $\ce$
are respectively the number of $\dbsym$-steps and $\lsvsym$-steps
in the reduction sequence.
\end{restatable}

To illustrate this property, consider the tight derivation $\deriv$ 
for $\tm = (\var \, \varthree)\esub\var{\vartwo\esub\vartwo\id}$ in 
\cref{ex:derivation-for-tight}, where the counter in the conclusion 
of $\deriv$ is $(1,2)$. A reduction sequence from $\tm$ terminates in 
$\var_1\esub{\var_1}\varthree\esub\var\vartwo\esub\vartwo\id \in \NF\emptyset{\set\varthree}\nonapp$
(see \cref{ex:reduction-for-tight}). 
The first counter, 1, matches the number of $\dbsym$-steps in the 
reduction sequence, while the second counter, 2, corresponds to the 
number of $\lsvsym$-steps.

Note that non-tight derivations may not provide exact information 
about the length of reduction sequences to normal form. 
For example, the subderivation 
$\derivs{}{\judg[1,1]{\varthree : \tightN; \var : \mtyp}{\var \, \varthree}\tightN}$ 
in \cref{ex:derivation-for-tight} is not tight, as $\mtyp$ is not 
tight, and the counters $(1,1)$ do not correspond to the number of 
steps needed to reach the normal form of $\var \, \varthree$ which is 
already in normal form.

We now address the \defn{completeness} of system $\typesystem$
with respect to \UOCBV: termination in \UOCBV implies 
$\typesystem$-typability.
Additionally, our result provides exact quantitative information 
about the reduction steps to normal form in \UOCBV.

Completeness of quantitative type systems can also be proved by 
following standard techniques, requiring a subject expansion property 
based in turn on an anti-substitution lemma (omitted here; see~\cref{app:typing}).
These statements are quite technical because the properties must be 
generalised appropriately to reason inductively.

First, we guarantee that normal forms are tightly typable in system 
$\typesystem$. 
For that, we construct tight environments for each normal form
$\tm \in \NF\aset\sset\appflag$, by typing the \emph{reachable} 
variables in $\aset$ with $\emset$ and those in $\sset$ with $\tightN$. 
More precisely, given $\tm$, $\aset$, and $\sset$ such that
$\inv\aset\sset\tm$, the \defn{tight environment} for $\tm$ under 
$(\aset,\sset)$ is written $\TEnv\aset\sset\tm$ and defined as the 
environment $\tctx$ such that $\tctx(\var) = \emset$ when 
$\var \in \aset \cap \rv\tm$, $\tctx(\var) = \tightN$ when 
$\var \in \sset \cap \rv\tm$, and $\tctx(\var) = \none$ otherwise.

Tight environments allow us to type normal forms \emph{tightly}:
\begin{restatable}[Normal Forms are Tight Typable]{proposition}{nftighttypable}
\label{prop:nfs_are_tight_typable}
Let $\tm$ be a term in $\NF\aset\sset\appflag$ and $\inv\aset\sset\tm$.
Then, there exists a tight type $\tight$ such that
$\derivs\deriv{\judg[0,0]{\TEnv\aset\sset\tm}\tm\tight}$.
Moreover, if $\tm \in \HAbs\aset$ then $\tight = \emset$, and if 
$\tm \in \Struct\sset$, then $\tight = \tightN$.
\end{restatable}

Subject Expansion is analogous to \nameref{prop:subject_reduction}, 
but applies in the \emph{reverse direction}: given a typable term 
$\tm'$ and a reduction step $\tm \tov\rulename\aset\sset\appflag \tm'$,
the term $\tm$ is typable as well, with the same type and typing environment.
\begin{restatable}[Subject Expansion]{proposition}{subjectexpansion}
\label{prop:subject_expansion}
Let $\inv\aset\sset\tm$ hold, and suppose 
$\tm \tov\rulename\aset\sset\appflag \tm'$
where $\rulename \in \set{\ruledb,\rulelsv}$ and 
$\derivs{\deriv'}{\judg[\cm', \ce']\tctx{\tm'}\mtyp}$, with $\isAppr\aset\tctx$.
Furthermore, assume that if $\appflag = \app$, then either 
$\mtyp = \tightN$ or $\mtyp$ is a singleton, \ie, of the form $\mset\typ$.
Then, there exists a derivation $\deriv$ such that 
$\derivs\deriv{\judg[\cm, \ce]\tctx\tm\mtyp}$, where
if $\rulename = \ruledb$, then $(\cm,\ce) = (\cm'+1,\ce')$, and
if $\rulename = \rulelsv$, then $(\cm,\ce) = (\cm', \ce'+1)$.
\end{restatable}

We now state a quantitative version of completeness of system $\typesystem$. 
That is, given a weakly terminating \UOCBV reduction sequence from 
$\tm$ containing exactly $\cm$ $\ruledb$-steps and exactly $\ce$ 
$\rulelsv$-steps, there exists a tight derivation of $\tm$ with 
counters $(\cm,\ce)$.
\begin{restatable}[Quantitative Completeness of System $\typesystem$]{theorem}{completenesstyping}
\label{thm:completeness_typing}
Let $\sset = \fv\tm$ and consider a reduction sequence 
$\tm \tovTop\sset^n \tmtwo$, where $\tmtwo$ is $\tovTop\sset$-irreducible.
Let $n = \cm + \ce$, where $\cm$ and $\ce$ are respectively the number 
of $\dbsym$ and $\lsvsym$ steps in the sequence.
Then, there exists a tight derivation $\deriv$, a tight environment 
$\tctx$, and a tight type $\tight$ such that 
$\derivs\deriv{\judg[\cm,\ce]\tctx\tm\tight}$.
\end{restatable}

To illustrate this result, take the reduction sequence starting from
$\tm = (\var \, \varthree)\esub\var{\vartwo\esub\vartwo\id}$ and 
ending in $\var_1\esub{\var_1}\varthree\esub\var\vartwo\esub\vartwo\id
\in \NF{\emptyset}{\set\varthree}\nonapp$ (see \cref{ex:reduction-for-tight}).
The length of this sequence is $3 = \cm + \ce$, with $\cm = 1$ 
$\ruledb$-step and $\ce = 2$ $\rulelsv$-steps.
\cref{ex:derivation-for-tight} shows a tight derivation $\deriv$, 
a tight environment $\varthree : \tightN$, and a tight type $\tightN$
such that $\derivs\deriv{\judg[1,2]{\varthree : \tightN}\tm\tightN}$.

Combining \cref{thm:soundness_typing,thm:completeness_typing} we get 
that a term $\tm$ is typable in system $\typesystem$ if and only if 
it is \emph{weakly} terminating in \UOCBV.
The characterisation of termination through typability follows from 
these theorems and \cref{coro:all_reductions_NF_same_length}.
Moreover, this characterisation is quantitative:
\begin{corollary}[Quantitative Characterisation of Termination]
\label{coro:quantitative_characterisation_termination}
A term $\tm$ is tight typable with counters $(\cm,\ce)$ in system
$\typesystem$ if and only if $\tm$ terminates in \UOCBV after $\cm$
$\dbsym$-steps and $\ce$ $\lsvsym$-steps.
\end{corollary}

\section{Conclusions}
\label{sec:conclusions}

This work contributes to the study of useful call-by-value in two 
different ways.
At a \emph{syntactic} level, we propose an inductive specification of 
usefulness for open \CBV evaluation called \UOCBV, which contrasts 
with previous notions of usefulness (both for \CBN and \CBV) that are 
defined by inductive means.
Our specification is inspired by the specification of strong 
call-by-need evaluation in~\cite{BalabonskiBBK17}. 
Although these strategies are very different, one point in common is 
that they crucially depend on the information of the surrounding 
evaluation context, reflected in the form of parameters.
We think that this technique can scale to other evaluation strategies, 
and in particular that it could be applied to provide an inductive 
formulation of useful call-by-name evaluation, originally formulated
non-inductively in~\cite{AccattoliL14}. 
Moreover, we also show that \UOCBV and \LOCBV are formally related 
through a suitable unfolding operation. 
Our last contribution at a syntactic level is establishing a relation 
between \UOCBV and an existing definition of usefulness in the 
literature, the \glamour abstract machine from~\cite{AccattoliC15}.

At a \emph{semantic} level, we propose the first model of useful 
evaluation in the literature, based on a non-idempotent intersection 
type system we dub system $\typesystem$. 
Even though the specification of the operational semantics of useful 
\CBV evaluation is highly complex, system $\typesystem$ is notably 
simple. 
This highlights the ability of semantic approaches based on 
non-idempotent intersection types to capture intricate operational 
details with streamlined intuitive methods.
System $\typesystem$ characterises termination of the \UOCBV, and 
provides exact quantitative information for the length of reduction 
sequences to normal form. This is why system $\typesystem$ is a
\emph{quantitative} interpretation of \UOCBV.
Our contribution is novel, as previous definitions of useful
evaluation (including useful \CBV) in the literature lack semantic
models, and current quantitative interpretations of \CBV do not
consider usefulness. 
Moreover, with the scarce exception of~\cite{KesnerViso22}, entirely 
based on the (verbose) persistent/consuming paradigm, this is the 
only ``simple'' and ``concise'' quantitative type system for \CBV 
that is able to count substitution steps exactly.

Several complementary properties are worth studying. 
A challenging question is how to extend our inductive
specification of usefulness to \emph{strong} \CBV so that evaluation
is also allowed inside abstractions. 
This is relevant for the implementation of proof assistants based on 
dependent type theory, in which type checking requires deciding the
definitional equality of type expressions up to full $\beta$-conversion, 
thus requiring strong evaluation.
Moreover, we would apply the semantic techniques used for open \CBV 
and yield a quantitative type system for useful strong \CBV.
Even more challenging would be to adapt all this technology to 
call-by-need, and to call-by-push-value. 
Moreover, we would like to understand the flexibility of system 
$\typesystem$, particularly its potential to capture 
\emph{other optimisations} of \CBV, such as those in~\cite{AccattoliG19}.

\bibliographystyle{plainnat}
\bibliography{biblio}

\newpage
\appendix
\onecolumn

\section*{Supplementary Material}

In this appendix we use the following complementary notations:

If $X,Y$ are sets, we write $X \disj Y$ to mean that $X$ and $Y$ are 
\defn{disjoint}, \ie, $X \cap Y = \emptyset$.

We write $\domSctx\sctx$ for the \defn{domain} of $\sctx$, \ie,
$\domSctx{\esub{\var_1}{\tm_1}\hdots\esub{\var_n}{\tm_n}} \eqdef
\set{\var_1,\hdots,\var_n}$.

\section{Proofs of Section~\ref{sec:opencbv} ``\Nonuseful Call-by-Value''}
\label{app:opencbv}

In this section, we present the proofs for the lemmas characterising 
normal forms in \LOCBV.

Drawing inspiration from Balabonski, Lanco, and Melquiond's work 
in~\cite{BalabonskiLM23}, we obtain an \emph{inductive} syntactic 
characterisation of normal forms in \LOCBV.
This inductive definition uses two parameters: a \defn{value frame} 
$\vset$, which is a set of variables, and a \defn{\positionalFlag} 
$\appflag \in \set{\app,\nonapp}$.
Both parameters encode information about the evaluation context in 
which a term is considered a normal form. 
Specifically, the set $\vset$ tracks variables bound to a value in 
the context, while $\appflag$ identifies whether subterms appear in
applied or non-applied positions within the evaluation context.
For example, $\var$ appears in an applied position in the term 
$\var \, \vartwo$, while $\vartwo$ appears in a non-applied position.
For top-level terms, the positional flag is always $\nonapp$.
The set of \defn{normal forms} under $\vset$ and $\appflag$, written
$\VNF\vset\appflag$, is inductively defined as follows:
\[
  \inferrule{
    \var \notin \vset
  }{
    \var \in \VNF\vset\appflag
  }\ruleVNFVar
  \HS
  \inferrule{
  }{ 
    \lam\var\tm \in \VNF\vset\nonapp
  }\ruleVNFLam
  \HS
  \inferrule{
    \tm \in \VNF\vset\app
    \sep
    \tmtwo \in \VNF\vset\nonapp
  }{
    \tm\,\tmtwo \in \VNF\vset\appflag
  }\ruleVNFApp
\]
\[
  \inferrule{
    \tm \in \VNF{\vset\cup\set\var}\appflag
    \sep
    \tmtwo \in \VNF\vset\nonapp
    \sep
    \valPred\tmtwo
  }{
    \tm\esub\var\tmtwo \in \VNF\vset\appflag
  }\ruleVNFEsVal
  \HS
  \inferrule{
    \tm \in \VNF\vset\appflag
    \sep
    \tmtwo \in \VNF\vset\nonapp
    \sep
    \neg(\valPred\tmtwo)
  }{
    \tm\esub\var\tmtwo \in \VNF\vset\appflag
  }\ruleVNFEsNonVal
\]
where the predicate $\valPred\tm$ states whether $\tm$ is of the form
$\val\sctx$ or not.

Value frames are sets of variables that are bound to values, so in 
rule \ruleVNFEsVal, we extend the value frame of the left premise with
the variable bound by the ES. 
Accordingly, in rule \ruleVNFVar, the variable $\var$ must not be in 
the value frame; otherwise it would mean that it must be substituted 
by some value during computation.
As an example, 
$\var\esub\var\vartwo \tovv\rulelsv \vartwo\esub\var\vartwo$
can be derived using the inductive rule \ruleVLsv because
$\var \tovv{\rulesub\var\vartwo} \vartwo$, so intuitively
$\var\esub\var\vartwo \notin \VNF\emptyset\nonapp$.
However, $\vartwo\esub\var\vartwo \in \VNF\emptyset\nonapp$ as we 
state below:
\[
  \inferrule{
    \inferrule{
    }{
      \vartwo \in \VNF{\set\var}\nonapp
    }\ruleVNFVar
    \inferrule{
    }{
      \vartwo \in \VNF\emptyset\nonapp
    }\ruleVNFVar
  }{
    \vartwo\esub\var\vartwo \in \VNF\emptyset\nonapp
  }\ruleVNFEsVal
\]
If we wanted to derive a judgement
$\var\esub{\var}{\vartwo} \in \VNF{\emptyset}{\nonapp}$, we would end
up with a premise stating $\var \notin \set{\var}$, which is false.
In the preceeding derivation tree, the predicate
$\vartwo \in \VNF{\set\var}\nonapp$ indicates that $\vartwo$ is in 
normal form with respect to a value frame containing $\var$.
This dependency comes from the reduction step 
$\var \tovv{\rulesub\var\vartwo} \vartwo$.

Just as normal forms are parameterised by a value frame $\vset$,
representing the set of variables that are bound to values, evaluation 
must also be generalised by parameterising it with value frame, to 
establish a precise relation between reduction and normal forms.
Consequently, we define the set of 
\textbf{reduction rules related to a value frame $\vset$} as follows:
\[
  \RulesV\vset 
  \eqdef
  \set{\ruledb, \rulelsv} \cup \set{\rulesub\var\val \ST \var \in \vset}
\]
A term $\tm$ belongs to the set $\VRed\vset$ of 
\textbf{reducible terms under a value frame $\vset$} if there is a 
step kind $\rulename \in \RulesV\vset$ and a term $\tm'$ such that 
$\tm \tovv\rulename \tm'$; $\tm$ belongs to the set $\VIrred\vset$ of 
\textbf{irreducible terms under $\vset$} if $\tm \notin \VRed\vset$.

\begin{lemma}[Weakening lemma]
\label{lem:vnf_weakening}
Let $\tm \in \VNF\vset\appflag$ and let $\vset'$ be such that
$\vset' \disj \rv\tm$.
Then, $\tm \in \VNF{\vset\cup\vset'}\appflag$.
\end{lemma}
\hiddenproof{
  The proof is straightforward by induction on the derivation of $\tm \in \VNF{\vset}{\appflag}$.
}{
  % Label: lem:vnf_weakening

\begin{proof}
By induction on the derivation of $\tm \in \VNF\vset\appflag$.
\begin{enumerate}
\item \ruleVNFVar.
  Then $\tm = \var \in \VNF\vset\appflag$ with $\var \notin \vset$.
  Given $\vset'$ such that $\vset' \disj \set{\var}$, then 
  $\var \notin \vset \cup \vset'$, so applying rule \ruleVNFVar we 
  obtain $\var \in \VNF{\vset \cup \vset'}\appflag$.
\item \ruleVNFLam.
  Then $\tm = \lam\var\tmtwo \in \VNF\vset\nonapp$, where 
  $\appflag = \nonapp$.
  Given $\vset'$ such that $\vset' \disj \emptyset$, we obtain
  $\lam\var\tmtwo \in \VNF{\vset \cup \vset'}\nonapp$
  by applying rule \ruleVNFLam.
\item \ruleVNFApp.
  Then $\tm = \tmtwo \, \tmthree$, and 
  $\tmtwo \, \tmthree \in \VNF\vset\appflag$ is derived from
  $\tmtwo \in \VNF\vset\app$ and $\tmthree \in \VNF\vset\nonapp$.
  Given $\vset'$ such that $\vset' \disj (\rv\tmtwo \cup \rv\tmthree)$,
  note that $\vset' \disj \rv\tmtwo$ and $\vset' \disj \rv\tmthree$ 
  hold.
  We obtain $\tmtwo \in \VNF{\vset \cup \vset'}\app$
  by the \ih on $\tmtwo$, and 
  $\tmthree \in \VNF{\vset \cup \vset'}\nonapp$ by the \ih on 
  $\tmthree$.
  Applying rule \ruleVNFApp we conclude
  $\tmtwo \, \tmthree \in \VNF{\vset \cup \vset'}\appflag$.
\item \ruleVNFEsVal.
  Then $\tm = \tmtwo\esub\var\tmthree \in \VNF\vset\appflag$ is 
  derived from $\tmtwo \in \VNF{\vset \cup \set\var}\appflag$ and
  $\tmthree \in \VNF\vset\nonapp$, with $\valPred\tmthree$.
  Let $\vset'$ be such that $\vset' \disj 
  ((\rv\tmtwo \setminus \set\var) \cup \rv\tmthree)$; by 
  $\alpha$-conversion we assume $\var \notin \vset'$.
  Then, $\vset' \disj \rv\tmtwo$ and $\vset' \disj \rv\tmthree$.
  By the \ih on $\tmtwo$ and $\tmthree$,
  we yield $\tmtwo \in \VNF{\vset \cup \set\var \cup \vset'}\appflag$,
  and $\tmthree \in \VNF{\vset \cup \vset'}\nonapp$, respectively.
  Applying rule \ruleVNFEsVal we conclude with
  $\tmtwo\esub\var\tmthree \in \VNF{\vset \cup \vset'}\appflag$.
\item \ruleVNFEsNonVal.
  Analogous to the previous case.
  % Then $\tm = \tmtwo\esub\var\tmthree \in \VNF\vset\appflag$ is 
  % derived from $\tmtwo \in \VNF\vset\appflag$ and 
  % $\tmthree \in \VNF\vset\nonapp$, with $\neg(\valPred\tmthree)$.
  % Let $\vset'$ be such that 
  % $\vset' \disj ((\rv\tmtwo \setminus \set\var) \cup \rv\tmthree)$;
  % by $\alpha$-conversion we assume $\var \notin \vset'$.
  % Then, $\vset' \disj \rv\tmtwo$ and $\vset' \disj \rv\tmthree$.
  % By the \ih on $\tmtwo$ and $\tmthree$,
  % we yield $\tmtwo \in \VNF{\vset \cup \vset'}\appflag$, and
  % $\tmthree \in \VNF{\vset \cup \vset'}\nonapp$, respectively.
  % Applying rule \ruleVNFEsNonVal we conclude with
  % $\tmtwo\esub\var\tmthree \in \VNF{\vset \cup \vset'}\appflag$.
\end{enumerate}
\end{proof}
  
}

% \subsection{Characterisation of Normal Forms for \Nonuseful Open \CBV}
% \label{sec:normal-forms-non-useful}

We can now show the two main results.
First, we show soundness of the syntactic characterisation of normal 
forms with respect to the parametrised reduction rules.
More specifically, we show that given any value frame $\vset$ and any 
\positionalFlag $\appflag$, a term in normal form under $\vset$ and 
$\appflag$ is in the set $\VIrred\vset$.

\begin{lemma}
\label{lem:absVNF}
For any $\vset$, $(\lam\var\tm)\sctx \notin \VNF\vset\app$.
\end{lemma}
\hiddenproof{
  The proof is straightforward by induction on $\sctx$.
}{
  % Label: lem:absVNF

\begin{proof}
By induction on $\sctx$.
\begin{itemize}
\item $\sctx = \ctxhole$.
  Immediate, since $\lam\var\tm \in \VNF\vset\app$ cannot be derived 
  by any rule.
\item $\sctx = \sctx'\esub\vartwo\tmtwo$.
  We reason by contradiction.
  Suppose that 
  $(\lam\var\tm)\sctx'\esub\vartwo\tmtwo \in \VNF\vset\app$, which
  can be derived by two different reduction rules, \ruleVNFEsVal and
  \ruleVNFEsNonVal. Both cases are analogous, so we analyze when we
  use rule \ruleVNFEsVal.
  In particular, $(\lam\var\tm)\sctx' \in \VNF{\vset \cup \set\vartwo}\app$,
  but we reach a contradiction since 
  $(\lam\var\tm)\sctx' \notin \VNF{\vset \cup \set\vartwo}\app$ by 
  the \ih on $\sctx'$.
  Therefore, it must be the case that
  $(\lam\var\tm)\sctx'\esub\vartwo\tmtwo \notin \VNF\vset\app$, so
  we are done.
\end{itemize}
\end{proof}
  
}

\begin{proposition}[Soundness of \nonuseful normal forms]
If $\tm \in \VNF\vset\appflag$, then $\tm \in \VIrred\vset$.
\end{proposition}

\begin{proof}
By induction on the derivation of $\tm \in \VNF\vset\appflag$.
\begin{enumerate}
\item \ruleVNFVar.
  Then, $\tm = \var$ and
  \[
    \indrule{\ruleVNFVar}{
      \var \notin \vset
    }{
      \var \in \VNF\vset\appflag
    }
  \]
  The only rule that would allow us to reduce $\var$ is \ruleVSub, as 
  follows: $\var \tovv{\rulesub\var\val} \val$ with $\var \in \vset$.
  But it is impossible since the premise of rule \ruleVNFVar states 
  that $\var \notin \vset$.
  We then conclude with $\var \in \VIrred\vset$.
\item \ruleVNFLam.
  Then, $\tm = \lam\var\tmtwo \in \VNF\vset\nonapp$,
  with $\appflag = \nonapp$.
  There are no rules to reduce abstractions, so it is immediate to 
  conclude with $\lam\var\tmtwo \in \VIrred\vset$.
\item \ruleVNFApp.
  Then, $\tm= \tmtwo \, \tmthree$ and 
  \[
    \indrule{\ruleVNFApp}{
      \tmtwo \in \VNF\vset\app
      \sep
      \tmthree \in \VNF\vset\nonapp
    }{
      \tmtwo \, \tmthree \in \VNF\vset\appflag
    }
  \]
  Three rules would allow us to reduce $\tmtwo \, \tmthree$:
  \ruleVDb, \ruleVAppL, and \ruleVAppR.
  We show that none of these apply:
  \begin{itemize}
  \item
    We cannot apply rule \ruleVDb to conclude 
    $\tmtwo\,\tmthree \in \VRed\vset$, because it would imply 
    $\abs\tmtwo$ at the same time that 
    $\tmtwo \in \VNF\vset\app$, which is impossible by 
    \cref{lem:absVNF}.
  \item
    By the \ih on $\tmtwo$, we yield $\tmtwo \in \VIrred\vset$. 
    Then, we cannot apply rule \ruleVAppL to conclude with 
    $\tmtwo\,\tmthree \in \VRed\vset$.
  \item
    By the \ih on $\tmthree$, we yield $\tmthree \in \VIrred\vset$. 
    Then, we cannot apply rule \ruleVAppR to conclude with 
    $\tmtwo\,\tmthree \in \VRed\vset$.
  \end{itemize}
  And we are done since $\tmtwo\,\tmthree \in \VIrred\vset$ by 
  definition.
\item \ruleVNFEsVal.
  Then, $\tm = \tmtwo\esub\var\tmthree$ and
  \[
    \indrule{\ruleVNFEsVal}{
      \tmtwo \in \VNF{\vset \cup \set\var}\appflag
      \sep
      \tmthree \in \VNF\vset\nonapp
      \sep
      \valPred\tmthree
    }{
      \tmtwo\esub\var\tmthree \in \VNF\vset\appflag
    }
  \]
  Three rules would allow us to reduce $\tmtwo\esub\var\tmthree$:
  \ruleVLsv, \ruleVEsL, and \ruleVEsR.
  We show that none of these apply:
  \begin{itemize}
  \item
    By the \ih on $\tmtwo$, we yield 
    $\tmtwo \in \VIrred{\vset \cup \set\var}$.
    Then, we cannot apply rules \ruleVLsv and \ruleVEsL to conclude
    with $\tmtwo\esub\var\tmthree \in \VRed\vset$.
  \item
    By the \ih on $\tmthree$, we yield
    $\tmthree \in \VIrred\vset$.
    Then, we cannot conclude $\tmtwo\esub\var\tmthree \in \VRed\vset$.
  \end{itemize}
  And we are done since $\tmtwo\esub\var\tmthree \in \VIrred\vset$ by 
  definition.
\item \ruleVNFEsNonVal.
  Then, $\tm = \tmtwo\esub\var\tmthree$ and
  \[
    \indrule{\ruleVNFEsNonVal}{
      \tmtwo \in \VNF\vset\appflag
      \sep
      \tmthree \in \VNF\vset\nonapp
      \sep
      \neg(\valPred\tmthree)
    }{
      \tmtwo\esub\var\tmthree \in \VNF\vset\appflag
    }
  \]
  Three rules would allow us to reduce $\tmtwo\esub\var\tmthree$:
  \ruleVLsv, \ruleVEsL, and \ruleVEsR.
  We show that case \ruleVLsv does not apply, as the remaining cases
  are analogous to the previous item.

  Since $\neg(\valPred\tmthree)$, then we cannot conclude 
  $\tmtwo\esub\var\tmthree \in \VRed\vset$ by rule \ruleVLsv.
\end{enumerate}
\end{proof}

Completeness intuitively states that any term in $\VIrred{\vset}$ is 
in normal form under the same $\vset$ and any 
\positionalFlag $\appflag$.
However, we need to consider the set $\VIrred\vset$ with caution, 
since given an irreducible abstraction in an applied position, the 
normal form predicate would fail, as it means that the computation of 
the whole term can continue with a $\dbsym$-step. Formally,
\begin{lemma}
\label{lem:nu_sub_change_of_values}
If $\tm \tovv{\rulesub{\var}{\val}} \tm'$
then for every value $\valtwo$ there exists $\tm''$
such that $\tm \tovv{\rulesub{\var}{\valtwo}} \tm''$. 
\end{lemma}
\hiddenproof{
  The proof is straightforward by induction on the derivation of 
  $\tm \tovv{\rulesub\var\val} \tm'$.  
}{
  % Label: lem:nu_sub_change_of_values

\begin{proof}
By induction on the derivation of $\tm \tovv{\rulesub\var\val} \tm'$.
\begin{itemize}
\item \ruleVSub.
  Then $\tm = \var \tovv{\rulesub\var\val} \val = \tm'$.
  Taking $\tm'' = \valtwo$ with $\valtwo$ any value
  and applying rule \ruleVSub, we derive
  $\var \tovv{\rulesub\var\valtwo} \valtwo$.
\item \ruleVAppL.
  Then 
  \[
    \indrule{\ruleVAppL}{
      \tmtwo \tovv{\rulesub\var\val} \tmtwo'
    }{
      \tm = \tmtwo \, \tmthree \tovv{\rulesub\var\val} \tmtwo' \, \tmthree = \tm'
    }
  \]

  By the \ih on $\tmtwo$, for every value $\valtwo$ there exists 
  $\tmtwo''$ such that $\tmtwo \tovv{\rulesub\var\valtwo} \tmtwo''$.
  Applying rule \ruleVAppL, we derive
  $\tmtwo \, \tmthree \tovv{\rulesub\var\valtwo} \tmtwo'' \, \tmthree$,
  where $\tm'' = \tmtwo'' \, \tmthree$.
\item \ruleVAppR.
  Analogous to the previous case.
  % Then 
  % \[
  %   \indrule{\ruleVAppR}{
  %     \tmthree \tovv{\rulesub\var\val} \tmthree'
  %   }{
  %     \tm = \tmtwo \, \tmthree \tovv{\rulesub\var\val} \tmtwo \, \tmthree' = \tm'
  %   }
  % \]

  % By the \ih on $\tmthree$, for every value $\valtwo$ there exists 
  % $\tmthree''$ such that $\tmthree \tovv{\rulesub\var\valtwo} \tmthree''$.
  % Applying rule \ruleVAppR, we derive
  % $\tmtwo \, \tmthree \tovv{\rulesub\var\valtwo} \tmtwo \, \tmthree''$,
  % where $\tm'' = \tmtwo \, \tmthree''$.
\item \ruleVEsL.
  Then
  \[
    \indrule{\ruleVEsL}{
      \tmtwo \tovv{\rulesub\var\val} \tmtwo'
      \sep
      \vartwo \notin \fv{\rulesub\var\val}
    }{
      \tm = \tmtwo\esub\vartwo\tmthree \tovv{\rulesub\var\val} \tmtwo'\esub\vartwo\tmthree = \tm'
    }
  \]

  By the \ih on $\tmtwo$, for every value $\valtwo$ there exists 
  $\tmtwo''$ such that $\tmtwo \tovv{\rulesub\var\valtwo} \tmtwo''$.
  By $\alpha$-conversion, we assume $\vartwo \notin \valtwo$, hence 
  $\vartwo \notin \fv{\rulesub\var\valtwo}$.
  Applying rule \ruleVEsL, we derive
  $\tmtwo\esub\vartwo\tmthree \tovv{\rulesub\var\valtwo} 
  \tmtwo''\esub\vartwo\tmthree$, where 
  $\tm'' = \tmtwo''\esub{\vartwo}{\tmthree}$, so we are done.
\item \ruleVEsR. % Análogo al caso \ruleVAppR.
  Analogous to the previous case.
  % Then
  % \[
  %   \indrule{\ruleVEsR}{
  %     \tmthree \tovv{\rulesub\var\val} \tmthree'
  %   }{
  %     \tm = \tmtwo\esub\vartwo\tmthree \tovv{\rulesub\var\val} \tmtwo\esub\vartwo{\tmthree'} = \tm'
  %   }
  % \]

  % By the \ih on $\tmthree$, for every value $\valtwo$ there exists 
  % $\tmthree''$ such that $\tmthree \tovv{\rulesub\var\valtwo} \tmthree''$.
  % Applying rule \ruleVEsR, we derive
  % $\tmtwo\esub\vartwo\tmthree \tovv{\rulesub\var\valtwo} 
  % \tmtwo\esub\vartwo{\tmthree''}$, where 
  % $\tm'' = \tmtwo\esub\vartwo{\tmthree''}$, so we are done.
\end{itemize}
\end{proof}

}

\begin{lemma}
\label{lem:nu_term_reducible_arg_esub_reducible}
Let $\var \notin \vset$.
If $\tm \in \VRed{\vset'}$, then $\tm\esub\var\tmtwo \in \VRed\vset$,
where $\vset' = \vset \cup \set\var$ if $\valPred\tmtwo$, and 
$\vset' = \vset$ otherwise.
\end{lemma}
% Label: lem:nu_term_reducible_arg_esub_reducible

\begin{proof}
By definition, there exists a step kind 
$\rulename \in \RulesV{\vset'}$ and a term $\tm'$ such that 
$\tm \tovv{\vset'} \tm'$. We have two cases, depending on whether 
$\var \in \fv\rulename$ or not:
\begin{itemize}
\item 
  If $\var \in \fv\rulename$, then $\rulename = \rulesub\vartwo\valtwo$, 
  with $\vartwo \in \vset'$.
  We have two subcases:
  \begin{itemize}
  \item 
    If $\var = \vartwo$, then $\tm \tovv{\rulesub\var\valtwo} \tm'$,
    and we have to analyze the form of $\tmtwo$:
    \begin{itemize}
    \item 
      If $\valPred\tmtwo$, then we can apply rule \ruleVLsv and 
      conclude with $\tm\esub\var\tmtwo \in \VRed\vset$.
    \item 
      If $\neg(\valPred\tmtwo)$, then $\vset' = \vset$.
      This case is impossible since $\var \notin \vset$ by hypothesis
      but at the same time $\var = \vartwo \in \vset$.
    \end{itemize}
  \item 
    If $\var \neq \vartwo$, then $\var \in \fv\valtwo$.
    Let $\valtwo'$ be a value such that $\var \notin \fv{\valtwo'}$.
    We can then apply \cref{lem:nu_sub_change_of_values}, yielding
    $\tm''$ such that $\tm \tovv{\rulesub\vartwo{\valtwo'}} \tm''$.
    Applying rule \ruleVEsL, we derive 
    $\tm\esub\var\tmtwo \tovv{\rulesub\vartwo{\valtwo'}} \tm''\esub\var\tmtwo$,
    with $\vartwo \in \vset$,
    so that $\tm\esub\var\tmtwo \in \VRed\vset$.
  \end{itemize}
\item 
  If $\var \notin \fv\rulename$, then we can apply rule \ruleVEsL to
  reduce $\tm\esub\var\tmtwo$, so we conclude with
  $\tm\esub\var\tmtwo \in \VRed\vset$.
\end{itemize}
\end{proof}

\begin{proposition}[Completeness of \nonuseful normal forms]
If $\tm \in \VIrred\vset$ and $(\abs\tm \implies \appflag = \nonapp)$,
then $\tm \in \VNF\vset\appflag$.
\end{proposition}

\begin{proof}
By induction on $\tm$.
\begin{enumerate}
\item $\tm = \var$.
  The only rule that would reduce $\var$ is \ruleVSub.
  By hypothesis, $\var \in \VIrred\vset$ so that $\var \notin \vset$
  by the definition of $\VRed\vset$.
  Then, by rule \ruleVNFVar, we derive $\var \in \VNF\vset\appflag$.
\item $\tm = \lam\var\tmtwo$.
  Then $\appflag = \nonapp$ by the hypothesis, so applying rule 
  \ruleVNFLam we derive $\lam\var\tmtwo \in \VNF\vset\nonapp$.
\item $\tm = \tmtwo \, \tmthree$.
  We necessarily have $\tmtwo \in \VIrred\vset$, since otherwise 
  $\tmtwo \, \tmthree$ would by in $\VRed\vset$ by rule \ruleVAppL.
  Moreover, $\neg\abs\tmtwo$, because otherwise $\tmtwo \, \tmthree$ 
  would reduce by the rule \ruleVDb.
  By the \ih on $\tmtwo$, we yield (1) $\tmtwo \in \VNF\vset\app$.
  
  Likewise, we necessarily have $\tmthree \in \VIrred\vset$, since 
  otherwise $\tmtwo \, \tmthree$ would by in $\VRed\vset$ by rule 
  \ruleVAppR.
  Then, (2) $\tmthree \in \VNF\vset\nonapp$ by the \ih on $\tmthree$.
  Applying rule \ruleVNFApp with (1) and (2) as premises, we derive 
  $\tmtwo \, \tmthree \in \VNF\vset\appflag$.
\item $\tm = \tmtwo\esub\var\tmthree$.
  By $\alpha$-conversion, we assume $\var \notin \vset$.
  We necessarily have $\tmthree \in \VIrred\vset$, since otherwise 
  $\tmtwo\esub\var\tmthree \in \VRed\vset$ by rule \ruleVEsR.
  Then, (1) $\tmthree \in \VNF\vset\nonapp$ by the \ih on $\tmthree$.
  We have to analyze two cases, depending on whether $\valPred\tmtwo$
  holds or not. Both cases are analogous, so we only focus on the case
  where $\valPred\tmtwo$ holds.
  Since $\tmtwo\esub\var\tmthree \in \VIrred\vset$, then 
  $\tmtwo \in \VIrred{\vset \cup \set\var}$ by contraposition of 
  \cref{lem:nu_term_reducible_arg_esub_reducible}.
  Moreover, $\appflag = \nonapp$, so the implication in the 
  hypothesis trivially holds.
  Therefore, we can apply the \ih on $\tmtwo$, yielding
  (2) $\tmtwo \in \VNF{\vset \cup \set\var}\appflag$. 
  Applying rule \ruleVNFEsVal to (1) and (2), we derive
  $\tmtwo\esub\var\tmthree \in \VNF\vset\appflag$.
\end{enumerate}
\end{proof}

The following corollary combines the results of soundness and 
completeness, so that the set of terms that are in normal form 
according to the inductive predicate $\VNF\vset\appflag$ is exactly 
the set $\VIrred\vset$.
Recall that the parameters $\vset$ and $\appflag$ are used in the
definition of \LOCBV to define evaluation \emph{inductively}.
However, we are actually interested in the \emph{top-level} case, \ie, 
in the evaluation of an \emph{isolated} term. 
In the case of top-level term, the value frame $\vset$ is empty, 
because there is no surrounding context binding any variable, and the 
\positionalFlag $\appflag$ is taken to be non-applied ($\nonapp$), 
because an isolated term is never applied.

\begin{corollary}[Characterisation of \nonuseful normal forms]
\label{thm:characterization_of_non_useful_normal_forms}
$\tm \in \VNF\emptyset\nonapp$ iff $\tm \in \VIrred\emptyset$.
\end{corollary}

An example of this characterisation is the term $\vartwo\esub\var\vartwo$
we showed at the beginning of the section: it is in 
$\VNF\emptyset\nonapp$, as shown above, and it is in 
$\VIrred\emptyset$, since none of the rules in $\RulesV\emptyset$ are 
applicable to it.

\section{Proofs of Section~\ref{sec:usefulcbv} ``Useful Call-by-Value''}
\label{app:usefulcbv}
This section of the appendix is organised in two subsections:
in \cref{sec:normal-forms-useful}, we give an inductive 
characterisation of normal forms.
In \cref{sec:diamond}, we show that \UOCBV enjoys the diamond property.

\subsection{Normal Forms for Useful Open Call-by-Value}
\label{sec:normal-forms-useful}

This section provides a syntactic characterisation of
normal forms of \UOCBV,
together with their corresponding soundness and completeness results
(\cref{thm:characterization_of_useful_normal_forms}).

Characterising normal forms for useful evaluation is not simple~\cite{AccattoliL22}.
Our inductive characterisation is similar to the one given for the 
\LOCBV calculus, in that it uses parameters.
Specifically, we use an \abstractionframe $\aset$, a structure frame
$\sset$, and a \positionalFlag $\appflag \in \set{\app, \nonapp}$.
These three parameters give information about the evaluation context
in which the (sub)term is considered to be a normal form.
As explained before in~\cref{sec:usefulcbv}, the frames $\aset$ and 
$\sset$ track variables that are hereditary abstractions or 
structures, respectively, whereas the \positionalFlag tracks applied
positions of subterms with respect to the context. 
For the top-level terms, the positional element is always in a 
non-applied position (\ie, $\nonapp$), as in the characterisation
of normal forms of \LOCBV. 
The set of \defn{normal forms} of \UOCBV under $\aset$, $\sset$, 
and $\appflag$, also called \defn{$(\aset,\sset,\appflag)$-normal forms}, 
is written $\NF\aset\sset\appflag$, and is defined inductively as:
\[
  \inferrule{
    \var \in \aset \implies \appflag = \nonapp
  }{
    \var \in \NF\aset\sset\appflag
  }\ruleUNFVar
  \HS
  \inferrule{
  }{ 
    \lam\var\tm \in \NF\aset\sset\nonapp
  }\ruleUNFLam
  \HS
  \inferrule{
    \tm \in \NF\aset\sset\app
    \sep
    \tmtwo \in \NF\aset\sset\nonapp
  }{
    \tm\,\tmtwo \in \NF\aset\sset\appflag
  }\ruleUNFApp
\]
\[
  \inferrule{
    \tm \in \NF{\aset \cup \set\var}\sset\appflag
    \sep
    \tmtwo \in \NF\aset\sset\nonapp
    \sep
    \tmtwo \in \HAbs\aset
  }{
    \tm\esub\var\tmtwo \in \NF\aset\sset\appflag
  }\ruleUNFEsAbs
  \HS
  \inferrule{
    \tm \in \NF\aset{\sset \cup \set\var}\appflag
    \sep
    \tmtwo \in \NF\aset\sset\nonapp
    \sep
    \tmtwo \in \Struct\sset
  }{
    \tm\esub\var\tmtwo \in \NF\aset\sset\appflag
  }\ruleUNFEsStruct
\]

If a variable is an hereditary abstraction ($\var \in \aset$) and it 
is applied ($\appflag = \app$), substituting the variable contributes 
to creating a $\dbsym$-redex, so the variable is not in normal form. 
Otherwise, it is in normal form, according to rule \ruleUNFVar.
For example, the term $\var\esub\var\id$ is in normal form if 
non-applied ($\appflag = \nonapp$), whereas it reduces to 
$\id\esub\var\id$ if applied ($\appflag = \app$).
To derive $\var\esub\var\id \in \NF\emptyset\emptyset\app$, we would
end up with a premise stating 
$\var \in \set\var \implies \app = \nonapp$, which does not hold.

Abstractions are in normal form if they are not applied.
In rule \ruleUNFEsAbs (resp. \ruleUNFEsStruct), the abstraction 
(resp. structure) frame of the left premise is extended with the 
bound variable in the ES, given that it is bound to an hereditary 
abstraction (resp. structure), exactly as in the left premises of the 
reduction rule \ruleUEsLAbs (resp. \ruleUEsLStruct).

We can now present the two main results of this section.
The first one is soundness of the syntactic characterisation of 
normal forms with respect to the parametrized reduction rules.
More specifically, we state that given any abstraction and structure 
frames $\aset$ and $\sset$ and a \positionalFlag $\appflag$, a term 
in normal form under these parameters belongs to the set
$\Irred\aset\sset\appflag$ of irreducible terms.

\begin{lemma}
\label{lem:absnf}
For any $\aset$ and $\sset$, 
$(\lam\var\tm)\sctx \notin \NF\aset\sset\app$.
\end{lemma}
\hiddenproof{
  The proof is straightforward by induction on $\sctx$.
}{
  % Label: lem:absnf

\begin{proof}
By induction on $\sctx$.
\begin{itemize}
\item $\sctx = \ctxhole$.
  Immediate, since $\lam\var\tm \in \NF\aset\sset\app$ cannot be 
  derived by any rule.
\item $\sctx = \sctx'\esub\vartwo\tmtwo$. 
  We reason by contradiction.
  Suppose 
  $(\lam\var\tm)\sctx'\esub\vartwo\tmtwo \in \NF\aset\sset\app$,
  which can be derived by two different reduction rules.
  Both cases are analogous, so we analyze when we use rule \ruleUNFEsAbs.
  In particular, $(\lam\var\tm)\sctx' \in \NF{\aset \cup \set\vartwo}\sset\app$, 
  but we reach a contradiction since 
  $(\lam\var\tm)\sctx' \notin \NF{\aset \cup \set\vartwo}\sset\app$
  by the \ih on $\sctx'$.
  Therefore, it must be the case that
  $(\lam\var\tm)\sctx'\esub\vartwo\tmtwo \notin \NF\aset\sset\app$,
  so we are done.
\end{itemize}
\end{proof}

}

\begin{lemma}
\label{lem:disjunction}
If $\aset \disj \sset$,
then either $\tm \notin \HAbs\aset$ or $\tm \notin \Struct\sset$.
\end{lemma}
% Label: lem:disjunction

\begin{proof}
By induction on $\tm$.
\begin{itemize}
\item $\tm = \var$.
  Since $\aset \disj \sset$,
  $\var$ cannot be simultaneously in $\aset$ and $\sset$,
  so either $\var \notin \aset$ and $\var \notin \HAbs\aset$,
  or $\var \notin \sset$ and $\var \notin \Struct\sset$.
\item $\tm = \lam\var\tmtwo$.
  Immediate, since $\lam\var\tmtwo \in \Struct\sset$ cannot be 
  derived by any rule, thus $\lam\var\tmtwo \notin \Struct\sset$.
\item $\tm = \tmtwo\,\tmthree$.
  Immediate, since $\tmtwo\,\tmthree \in \HAbs\aset$ cannot be 
  derived by any rule, thus $\tmtwo\,\tmthree \notin \HAbs\aset$.
\item $\tm = \tmtwo\esub\var\tmthree$.
  To show that $\tm \notin \HAbs\aset$ or $\tm \notin \Struct\sset$,
  we assume $\tm \in \HAbs\aset$, and argue that $\tm \notin \Struct\sset$.
  We can assume $\var \notin \aset \cup \sset$ by $\alpha$-conversion.
  We have the following subcases, depending on the rule used to 
  derive $\tm \in \HAbs\aset$:
  \begin{itemize}
  \item \ruleHAbsSubi.
    Then $\tmtwo \in \HAbs\aset$ and $\var \notin \aset$.
    By the \ih on $\tmtwo$, we yield $\tmtwo \notin \Struct\sset$.
    Therefore, $\tmtwo\esub\var\tmthree \in \Struct\sset$ cannot be 
    derived by rule \ruleStructSubi.
    Similarly, since $\aset \disj (\sset \cup \set\var)$, then 
    $\tmtwo \notin \Struct{\sset \cup \set\var}$ by the \ih on $\tmtwo$.
    Therefore, $\tmtwo\esub\var\tmthree \in \Struct\sset$ cannot be 
    derived using rule \ruleStructSubii.
  \item \ruleHAbsSubii.
    Then $\tmtwo \in \HAbs{\aset \cup \set\var}$, $\var \notin \aset$, 
    and $\tmthree \in \HAbs\aset$ hold.
    Since $(\aset \cup \set\var) \disj \sset$, then 
    $\tmtwo \notin \Struct\sset$ by the \ih on $\tmtwo$.
    Therefore, $\tmtwo\esub\var\tmthree \in \Struct\sset$ cannot be 
    derived using rule \ruleStructSubi.

    On the other hand, $\tmthree \notin \Struct\sset$ by the \ih on 
    $\tmthree$, so $\tmtwo\esub\var\tmthree \in \Struct\sset$ cannot 
    be derived using rule \ruleStructSubii.
  \end{itemize}
  We thus conclude with $\tmtwo\esub\var\tmthree \notin \Struct\sset$.
\end{itemize}
\end{proof}

\begin{proposition}[Soundness of Useful Normal Forms]
\label{prop:soundness-characterization-normal-forms}
If $\inv\aset\sset\tm$ and $\tm \in \NF\aset\sset\appflag$,
then $\tm \in \Irred\aset\sset\appflag$.
\end{proposition}
% Label: prop:soundness-characterization-normal-forms

\begin{proof}
By induction on the derivation of $\tm \in \NF\aset\sset\appflag$.
\begin{itemize}
\item \ruleUNFVar.
  Then $\var \in \NF\aset\sset\appflag$, with premise
  $\var \in \aset \implies \appflag = \nonapp$.
  The only rule for reducing variables is \ruleUSub,
  with the condition that $\var \in \aset$.
  This implies by hypothesis that $\appflag$ must be $\nonapp$, contrary 
  to the request of rule \ruleUSub that $\appflag = \app$.
  Hence, no reduction rule applies and we are done.
\item \ruleUNFLam. 
  Immediate, since $\tm$ is an abstraction, and there are no rules 
  for reducing terms of this form.
\item \ruleUNFApp. 
  Then, $\tm = \tmtwo\,\tmthree$ and
  \[
    \inferrule{
      \tmtwo \in \NF\aset\sset\app
      \sep
      \tmthree \in \NF\aset\sset\nonapp
    }{
      \tmtwo\,\tmthree \in \NF\aset\sset\appflag
    }\ruleUNFApp
  \]

  Moreover, $\inv\aset\sset{\tmtwo\,\tmthree }$ implies 
  $\inv\aset\sset\tmtwo$ and $\inv\aset\sset\tmthree$.
  There are three rules for reducing applications: \ruleUDb, \ruleUAppL,
  and \ruleUAppR.
  We show that none of these apply:
  \begin{itemize}
    \item
      We cannot apply rule \ruleUDb to conclude 
      $\tmtwo\,\tmthree \in \Red\aset\sset\appflag$, because it would 
      imply that $\abs\tmtwo$ at the same time that 
      $\tmtwo \in \NF\aset\sset\app$, which contradicts \cref{lem:absnf}.
    \item
      By the \ih on $\tmtwo$, we yield $\tmtwo \in \Irred\aset\sset\app$.
      Then, we cannot apply rule \ruleUAppL to conclude with
      $\tmtwo\,\tmthree \in \Red\aset\sset\appflag$.
    \item
      By the \ih on $\tmthree$, we yield $\tmthree \in \Irred\aset\sset\nonapp$.
      Then, we cannot apply rule \ruleUAppR to conclude with
      $\tmtwo\,\tmthree \in \Red\aset\sset\appflag$.
  \end{itemize}
  And we are done since $\tmtwo\,\tmthree \in \Irred\aset\sset\appflag$ 
  by definition.
\item \ruleUNFEsAbs. 
  Then, $\tm = \tmtwo\esub\var\tmthree$ and
  \[
    \inferrule{
      \tmtwo \in \NF{\aset \cup \set\var}\sset\appflag
      \sep
      \tmthree \in \NF\aset\sset\nonapp
      \sep
      \tmthree \in \HAbs\aset
    }{
      \tmtwo\esub\var\tmthree \in \NF\aset\sset\appflag
    }\ruleUNFEsAbs
  \]
  We assume $\var \notin \aset\cup\sset$ by $\alpha$-conversion.
  Moreover, $\inv\aset\sset{\tmtwo\esub\var\tmthree}$ implies 
  $\inv{\aset \cup \set\var}\sset\tmtwo$ and $\inv\aset\sset\tmthree$.
  There are four rules for reducing closures: \ruleULsv, \ruleUEsR,
  \ruleUEsLAbs, and \ruleUEsLStruct.
  We show that none of these apply:
  \begin{itemize}
    \item 
      By the \ih on $\tmtwo$, we yield 
      $\tmtwo \in \Irred{\aset \cup \set\var}\sset\appflag$.
      Then, we cannot apply rules \ruleULsv and \ruleUEsLAbs to 
      conclude with $\tmtwo\esub\var\tmthree \in \Red\aset\sset\appflag$.
    \item 
      By the \ih on $\tmthree$, we yield 
      $\tmthree \in \Irred\aset\sset\nonapp$.
      Then, we cannot apply rule \ruleUEsR to conclude with 
      $\tmtwo\esub\var\tmthree \in \Red\aset\sset\appflag$.
    \item 
      Since $\aset \disj \sset$ holds by the invariant,
      and $\tmthree \in \HAbs\aset$ holds by the hypothesis,
      then $\tmthree \notin \Struct\sset$ by \cref{lem:disjunction}.
      Hence, we cannot apply rule \ruleUEsLStruct to conclude with 
      $\tmtwo\esub\var\tmthree \in \Red\aset\sset\appflag$. 
  \end{itemize}
  And we are done since $\tmtwo\esub\var\tmthree \in \Irred\aset\sset\appflag$ 
  by definition.
\item \ruleUNFEsStruct. 
  Then, $\tm = \tmtwo\esub\var\tmthree$ and
  \[
    \inferrule{
      \tmtwo \in \NF\aset{\sset \cup \set\var}\appflag
      \sep
      \tmthree \in \NF\aset\sset\nonapp
      \sep
      \tmthree \in \Struct\sset
    }{
      \tmtwo\esub\var\tmthree \in \NF\aset\sset\appflag
    }\ruleUNFEsStruct
  \]
  We assume $\var \notin \aset\cup\sset$ by $\alpha$-conversion.
  Moreover, $\inv\aset\sset{\tmtwo\esub\var\tmthree}$ implies 
  $\inv\aset{\sset \cup \set\var}\tmtwo$ and $\inv\aset\sset\tmthree$.
  There are four rules for reducing closures: \ruleULsv, \ruleUEsR,
  \ruleUEsLAbs, and \ruleUEsLStruct.
  We show that none of these apply:
  \begin{itemize}
    \item 
      We cannot apply rule \ruleULsv to conclude 
      $\tmtwo\esub\var\tmthree \in \Red\aset\sset\appflag$, because 
      it would imply that $\tmthree = \val\sctx \in \HAbs\sset$, while
      at the same time $\tmthree \in \Struct\sset$ by hypothesis,
      which contradicts \cref{lem:disjunction}, as $\aset \disj \sset$
      holds by the invariant.
    \item 
      By the \ih on $\tmthree$, we yield 
      $\tmthree \in \Irred\aset\sset\nonapp$. 
      Then, we cannot apply rule \ruleUEsR to conclude with 
      $\tmtwo\esub\var\tmthree \in \Red\aset\sset\appflag$.
    \item 
      By the \ih on $\tmtwo$, we yield 
      $\tmtwo \in \Irred\aset{\sset \cup \set\var}\appflag$.
      Then, we cannot apply rule \ruleUEsLStruct to conclude with 
      $\tmtwo\esub\var\tmthree \in \Red\aset\sset\appflag$.
    \item 
      Since $\aset \disj \sset$ holds, and $\tmthree \in \Struct\sset$ 
      holds by the hypothesis, then $\tmthree \notin \HAbs\aset$ by 
      \cref{lem:disjunction}.
      Hence, we cannot apply rule \ruleUEsLAbs to conclude with 
      $\tmtwo\esub\var\tmthree \in \Red\aset\sset\appflag$.
  \end{itemize}
  And we are done since $\tmtwo\esub\var\tmthree \in \Irred\aset\sset\appflag$ 
  by definition.
\end{itemize}
\end{proof}

We now yield the lemmas necessary to achieve completeness of the
inductive definition of normal forms (\cref{prop:completeness_nf}).

\begin{lemma}
\label{lem:nf_in_HAbs_or_Struct}
Let $\inv\aset\sset\tm$.
If $\tm \in \NF\aset\sset\appflag$, then 
$\tm \in \HAbs\aset \cup \Struct\sset$.
Furthermore, if $\appflag = \app$, then $\tm \in \Struct\sset$.
\end{lemma}
% Label: lem:nf_in_HAbs_or_Struct

\begin{proof}
By induction on the derivation of $\tm \in \NF\aset\sset\appflag$.
\begin{enumerate}
  \item \ruleUNFVar. 
    Then
    \[
      \inferrule{
        \var \in \aset \implies \appflag = \nonapp
      }{
        \var \in \NF\aset\sset\appflag
      }\ruleUNFVar
    \]

    Given $\inv\aset\sset\var$, we analyse two cases depending on 
    whether $\var \in \aset$ or $\var \in \sset$.
    If $\var \in \aset$, then $\var \in \HAbs\aset$ by rule 
    \ruleHAbsVar, and $\appflag = \nonapp$ by premise of the rule 
    \ruleUNFVar.
    If $\var \in \sset$, then $\var \in \Struct\sset$ by rule 
    \ruleStructVar.
    Thus, we conclude with $\var \in \HAbs\aset \cup \Struct\sset$.
  \item \ruleUNFLam.
    Then, $\lam\var\tmtwo \in \NF\aset\sset\nonapp$.
    By rule \ruleHAbsLam, we derive $\lam\var\tmtwo \in \HAbs\aset$,
    so we conclude with 
    $\lam\var\tmtwo \in \HAbs\aset \cup \Struct\sset$.
  \item \ruleUNFApp.
    Then
    \[
      \inferrule{
        \tmtwo \in \NF\aset\sset\app
        \sep
        \tmthree \in \NF\aset\sset\nonapp
      }{
        \tmtwo \, \tmthree \in \NF\aset\sset\appflag
      }\ruleUNFApp
    \]
    
    Since $\inv\aset\sset{\tmtwo \, \tmthree}$ implies in particular 
    $\inv\aset\sset\tmtwo$, we can apply the \ih on $\tm$, yielding
    $\tmtwo \in \Struct\sset$, as the positional flag is $\app$.
    Then, $\tmtwo \, \tmthree \in \Struct\sset$ by rule 
    \ruleStructApp, which implies
    $\tmtwo \, \tmthree \in \HAbs\aset \cup \Struct\sset$.
  \item \ruleUNFEsAbs.
    Then
    \[
      \inferrule{
        \tmtwo \in \NF{\aset \cup \set\var}\sset\appflag
        \sep
        \tmthree \in \NF\aset\sset\nonapp
        \sep
        \tmthree \in \HAbs\aset
      }{
        \tmtwo\esub\var\tmthree \in \NF\aset\sset\appflag
      }\ruleUNFEsAbs
    \]
    We assume $\var \notin \aset \cup \sset$ by $\alpha$-conversion,
    so that $\inv\aset\sset{\tmtwo\esub\var\tmthree}$ implies in 
    particular $\inv{\aset \cup \set\var}\sset\tmtwo$.
    We can then apply the \ih on $\tmtwo$, yielding 
    $\tmtwo \in \HAbs{\aset \cup \set\var} \cup \Struct\sset$.
    We have two possible cases, depending on whether 
    $\tmtwo \in \HAbs{\aset \cup \set\var}$ or $\tmtwo \in \Struct\sset$.
    Moreover, if $\appflag = \app$, then the latter is the only 
    possible case.
    
    If $\tmtwo \in \HAbs{\aset \cup \set\var}$, then 
    $\tmtwo\esub\var\tmthree \in \HAbs\aset$ by rule \ruleHAbsSubii.
    If $\tmtwo \in \Struct\sset$, then 
    $\tmtwo\esub\var\tmthree \in \Struct\sset$ by rule \ruleStructSubi.
   \item \ruleUNFEsStruct.
    Analogous to the previous case.
    % Then
    % \[
    %   \inferrule{
    %     \tmtwo \in \NF\aset{\sset \cup \set\var}\appflag
    %     \sep
    %     \tmthree \in \NF\aset\sset\nonapp
    %     \sep
    %     \tmthree \in \Struct\sset
    %   }{
    %     \tmtwo\esub\var\tmthree \in \NF\aset\sset\appflag
    %   }\ruleUNFEsStruct
    % \]
    % We assume $\var \notin \aset \cup \sset$ by $\alpha$-conversion,
    % so that $\inv\aset\sset{\tmtwo\esub\var\tmthree}$ implies in 
    % particular $\inv\aset{\sset \cup \set\var}\tmtwo$.
    % We can then apply the \ih on $\tmtwo$, yielding 
    % $\tmtwo \in \HAbs\aset \cup \Struct{\sset \cup \set\var}$.
    % We have two possible cases, depending on whether 
    % $\tmtwo \in \HAbs\aset$ or $\tmtwo \in \Struct{\sset \cup \set\var}$.
    % Moreover, if $\appflag = \app$, then the latter is the only 
    % possible case.

    % If $\tmtwo \in \HAbs\aset$, then 
    % $\tmtwo\esub\var\tmthree \in \HAbs\aset$ by rule \ruleHAbsSubi.
    % If $\tmtwo \in \Struct{\sset \cup \set\var}$, then
    % $\tmtwo\esub\var\tmthree \in \Struct\sset$ by rule \ruleStructSubii.
\end{enumerate}
\end{proof}

\begin{remark}
\label{rem:t_reduces_with_subvar_var_occurs_free_in_t}
If $\tm \tov{\rulesub\var\val}\aset\sset\appflag \tm'$,
then $\var \in \fv\tm$.
\end{remark}

\begin{lemma}
\label{lem:sub_var_in_aset}
Let $\inv\aset\sset\tm$.
If $\tm \tov{\rulesub\var\val}\aset\sset\appflag \tm'$, 
then $\var \in \aset$.
\end{lemma}
% Label: lem:sub_var_in_aset

\begin{proof}
By induction on the derivation of 
$\tm \tov{\rulesub\var\val}\aset\sset\appflag \tm'$.
The interesting case is the \ruleUSub rule.
Then, $\var \tov{\rulesub\var\val}{\aset' \cup \set\var}\sset\appflag \val$,
where $\aset = \aset' \cup \set\var$, so $\var \in \aset$ trivially.

The remaining cases are straightforward by resorting to the \ih.
For example, in the case of rule \ruleUEsLAbs, we have:
\[
  \inferrule{
    \tmtwo \tov{\rulesub\var\val}{\aset \cup \set\vartwo}\sset\appflag \tmtwo'
    \sep
    \tmthree \in \HAbs\aset
    \sep
    \vartwo \notin \aset \cup \sset
    \sep
    \vartwo \notin \fv{\rulesub\var\val}
  }{
    \tmtwo\esub\vartwo\tmthree
    \tov{\rulesub\var\val}\aset\sset\appflag
    \tmtwo'\esub\vartwo\tmthree
  }\ruleUEsLAbs
\] 

Note that $\inv\aset\sset{\tmtwo\esub\vartwo\tmthree}$ implies in 
particular $\inv{\aset \cup \set\vartwo}\sset\tmtwo$, so 
$\var \in \aset \cup \set\vartwo$ by the \ih on $\tmtwo$.
Furthermore, we assume $\vartwo \neq \var$ by $\alpha$-conversion, so 
$\var \in \aset$, as required.
\end{proof}

\begin{lemma}
\label{lem:sub_change_of_values}
Let $\inv\aset\sset\tm$.
If $\tm \tov{\rulesub\var\val}\aset\sset\appflag \tm'$,
then for every value $\valtwo$ there exists $\tm''$ such that 
$\tm \tov{\rulesub\var\valtwo}\aset\sset\appflag \tm''$.
\end{lemma}
% Label: lem:sub_change_of_values

\begin{proof}
By induction on the derivation of 
$\tm \tov{\rulesub\var\val}\aset\sset\appflag \tm'$.
The interesting case is the \ruleUSub rule.
Then, $\var \tov{\rulesub\var\val}{\aset' \cup \set\var}\sset\app \val$,
so $\var \in \aset$. 
Hence, taking any value $\valtwo$, we yield
$\var \tov{\rulesub\var\valtwo}{\aset' \cup \set\var}\sset\app \valtwo 
= \tm''$ by rule \ruleUSub.

The remaining cases are straightforward by resorting to the \ih.
For example, in the case of the \ruleUEsLAbs rule, we have:
\[
  \inferrule{
    \tmtwo \tov{\rulesub\var\val}{\aset \cup \set\vartwo}\sset\appflag \tmtwo'
    \sep
    \tmthree \in \HAbs\aset
    \sep
    \vartwo \notin \aset \cup \sset
    \sep
    \vartwo \notin \fv{\rulesub\var\val}
  }{
    \tmtwo\esub\vartwo\tmthree
    \tov{\rulesub\var\val}\aset\sset\appflag
    \tmtwo'\esub\vartwo\tmthree
  }\ruleUEsLAbs
\]

Note that $\inv\aset\sset{\tmtwo\esub\var\tmthree}$ implies in 
particular $\inv{\aset \cup \set\vartwo}\sset\tmtwo$, so we can apply 
the \ih on $\tmtwo$, yielding $\tmtwo''$ such that
$\tmtwo \tov{\rulesub\var\valtwo}{\aset \cup \set\vartwo}\sset\appflag 
\tmtwo''$.
Moreover, we assume $\vartwo \notin \fv\valtwo$ by $\alpha$-conversion,
so we can apply rule \ruleUEsLAbs to conclude with
$\tmtwo\esub\vartwo\tmthree 
\tov{\rulesub\var\valtwo}\aset\sset\appflag \tmtwo''\esub\vartwo\tmthree 
= \tm''$.
\end{proof}

\begin{lemma}
\label{lem:term_reducible_arg_abstraction_esub_reducible}
Let $\inv{\aset \cup \set\var}\sset\tm$ and $\tmtwo \in \HAbs\aset$,
with $\var \notin \aset$.
If $\tm \in \Red{\aset \cup \set\var}\sset\appflag$,
then $\tm\esub\var\tmtwo \in \Red\aset\sset\appflag$.
\end{lemma}
% Label: lem:term_reducible_arg_abstraction_esub_reducible

\begin{proof}
By definition, there exist $\rulename$ and $\tm'$ such that
$\tm \tov\rulename{\aset \cup \set\var}\sset\appflag \tm'$.
Moreover, $\var \notin \aset \cup \sset$ holds by the hypothesis.
There are two cases, depending on whether $\var \in \fv\rulename$ or not:
\begin{itemize}
\item 
  If $\var \in \fv\rulename$,
  then $\rulename = \rulesub\vartwo\val$.
  We have two subcases, depending on whether $\var = \vartwo$ or not:
  \begin{itemize}
  \item $\var = \vartwo$.
    Since $\tmtwo \in \HAbs\aset$, then $\tmtwo$ is of the form 
    $\valtwo\sctx$ by \cref{rem:habs_st}.
    Applying \cref{lem:sub_change_of_values},
    we yield $\tm''$ such that 
    $\tm \tov{\rulesub\var\valtwo}{\aset \cup \set\var}\sset\appflag \tm''$.
    We can apply rule \ruleULsv, yielding 
    $\tm\esub\var{\valtwo\sctx} \tov\rulelsv\aset\sset\appflag 
     \tm'\esub\var\valtwo\sctx$.
  \item $\var \neq \vartwo$.
    Let $\valtwo$ be a value such that $\var \notin \fv{\valtwo}$.
    There exists $\tm''$ such that 
    $\tm \tov{\rulesub{\vartwo}{\valtwo}}{\aset \cup \set{\var}}{\sset}{\appflag} \tm''$
    by \cref{lem:sub_change_of_values}.
    We can apply rule \ruleUEsLAbs, yielding 
    $\tm\esub\var\tmtwo \tov{\rulesub\vartwo\valtwo}\aset\sset\appflag 
    \tm''\esub\var\tmtwo$.
  \end{itemize}
\item 
  If $\var \notin \fv\rulename$, we can apply rule \ruleUEsLAbs, 
  yielding $\tm\esub\var\tmtwo \tov\rulename\aset\sset\appflag 
  \tm'\esub\var\tmtwo$.
\end{itemize}
In either case, we conclude with
$\tm\esub\var\tmtwo \in \Red\aset\sset\appflag$.
\end{proof}

\begin{lemma}
\label{lem:term_reducible_arg_struct_esub_reducible}
Let $\inv\aset{\sset \cup \set\var}\tm$ and $\tmtwo \in \Struct\sset$,
with $\var \notin \sset$.
If $\tm \in \Red\aset{\sset \cup \set\var}\appflag$,
then $\tm\esub\var\tmtwo \in \Red\aset\sset\appflag$.
\end{lemma}
% Label: lem:term_reducible_arg_struct_esub_reducible

\begin{proof}
By definition, there exist $\rulename$ and $\tm'$ such that
$\tm \tov\rulename\aset{\sset \cup \set\var}\appflag \tm'$.
Moreover, $\var \notin \aset \cup \sset$ holds by the hypothesis.
There are two cases, depending on whether $\var \in \fv\rulename$ or not:
\begin{itemize}
\item 
  If $\var \in \fv\rulename$, then $\rulename = \rulesub\vartwo\val$.
  We have two subcases, depending on whether $\var = \vartwo$ or not:
  \begin{itemize}
  \item $\var = \vartwo$.
    This case is not possible, since $\var = \vartwo \in \aset$ by 
    \cref{lem:sub_var_in_aset}, and at the same time 
    $\var \notin \aset$ by the hypothesis.
  \item $\var \neq \vartwo$.
    Let $\valtwo$ be a value such that $\var \notin \fv\valtwo$.
    There exists $\tm''$ such that 
    $\tm \tov{\rulesub\vartwo\valtwo}\aset{\sset \cup \set\var}\appflag \tm''$
    by \cref{lem:sub_change_of_values}.
    We can apply rule \ruleUEsLStruct, yielding $\tm\esub\var\tmtwo 
    \tov{\rulesub\vartwo\valtwo}\aset\sset\appflag \tm''\esub\var\tmtwo$.
  \end{itemize}
\item 
  If $\var \notin \fv\rulename$, we can apply rule \ruleUEsLStruct, 
  yielding $\tm\esub\var\tmtwo \tov\rulename\aset\sset\appflag \tm'\esub\var\tmtwo$.
\end{itemize}
In either case, we conclude with 
$\tm\esub\var\tmtwo \in \Red\aset\sset\appflag$.
\end{proof}

Completeness states that an irreducible term 
$\tm \in \Irred\aset\sset\appflag$ is in normal form with respect to 
the same parameters $\aset$, $\sset$, and $\appflag$; that is, 
$\tm \in \NF\aset\sset\appflag$. 
However, there is an exception to this statement, since the context 
surrounding $\tm$ has to be taken into account.
In particular, an irreducible abstraction is \emph{not} considered to 
be a normal form if it occurs in an applied position, because the 
evaluation of the whole term (including the surrounding context) can 
proceed by means of a $\ruledb$-step. 
An hereditary abstraction $\tm$ must be either a term of the form 
$(\lam\var{\tm'})\sctx$, or a term which is reducible in an applied 
position, such as $\var\esub\var\id$.
Thus, an applied hereditary abstraction is \emph{always} reducible.

The first part of the following proposition covers the case in which
an irreducible term is a normal form, while the second and third parts 
cover the case of applied hereditary abstractions, which are not in 
normal form even if they are irreducible.
\begin{proposition}[Completeness of Useful Normal Forms]
\label{prop:completeness_nf}
Let $\inv\aset\sset\tm$.
\begin{enumerate}
\item \label{it:hyp1}
  If $\tm \in \Irred\aset\sset\appflag$
  and $(\tm \in \HAbs\aset \implies \appflag = \nonapp)$,
  then $\tm \in \NF\aset\sset\appflag$.
\item \label{it:hyp2}
  If $\tm \in \HAbs\aset$,
  then either $\abs\tm$ or $\tm \in \Red\aset\sset\app$.
\item \label{it:hyp3}
  If $\tm \in \HAbs\aset$, then $\tm \, \tmtwo \in \Red\aset\sset\appflag$, 
  for any term $\tmtwo$.
 \end{enumerate}
\end{proposition}
% Label: prop:completeness_nf

\begin{proof}
Statement \ref{it:hyp3} is an immediate consequence of statement \ref{it:hyp2},
since if $\abs\tm$, then $\tm \, \tmtwo \tov\dbsym\aset\sset\appflag$-reduces
by rule \ruleUDb.
On the other hand, if $\tm \in \Red\aset\sset\app$, then 
$\tm \, \tmtwo \in \Red\aset\sset\appflag$ by applying rule \ruleUAppL.

Statements \ref{it:hyp1} and \ref{it:hyp2} are proved by simultaneous 
induction on $\tm$.
\begin{enumerate}
\item \quad
  \begin{itemize}
  \item $\tm = \var$.
    If $\var \in \HAbs\aset$, it must be the case that $\var \in \aset$.
    Thus, $\appflag = \nonapp$.
    We can then apply rule \ruleUNFVar, yielding
    $\var \in \NF\aset\sset\appflag$.
  \item $\tm = \lam\var\tmtwo$.
    Since $\lam\var\tmtwo \in \HAbs\aset$, then $\appflag = \nonapp$ 
    by hypothesis. 
    We derive $\lam\var\tmtwo \in \NF\aset\sset\nonapp$ by rule \ruleUNFLam.
  \item $\tm = \tmtwo \, \tmthree$. 
    Note that $\inv\aset\sset{\tmtwo \, \tmthree}$ implies
    $\inv\aset\sset\tmtwo$ and $\inv\aset\sset\tmthree$.

    We necessarily have $\tmtwo \in \Irred\aset\sset\app$, since 
    otherwise $\tmtwo \, \tmthree$ would be in $\Red\aset\sset\appflag$ 
    by rule \ruleUAppL.
    By the \ih (\ref{it:hyp1}) on $\tmtwo$, we yield 
    (a) $\tmtwo \in \NF\aset\sset\app$.

    As a consequence, $\tmtwo \in \Struct{\sset}$ by 
    \cref{lem:nf_in_HAbs_or_Struct}.
    And thus we necessarily have $\tmthree \in \Irred\aset\sset\nonapp$, 
    since otherwise $\tmtwo \, \tmthree$ would be in 
    $\Red\aset\sset\appflag$ by rule \ruleUAppR.
    Then, (b) $\tmthree \in \NF\aset\sset\nonapp$ by the \ih (\ref{it:hyp1}) 
    on $\tmthree$.
    Applying rule \ruleUNFApp with (a) and (b) as premises, we derive 
    $\tmtwo \, \tmthree \in \NF\aset\sset\appflag$.
  \item $\tm = \tmtwo\esub\var\tmthree$.
    We necessarily have $\tmthree \in \Irred\aset\sset\nonapp$, 
    since otherwise $\tmtwo\esub\var\tmthree$ would be in 
    $\Red\aset\sset\appflag$ by rule \ruleUEsR.
    Moreover, $\inv\aset\sset{\tmtwo\esub\var\tmthree}$ implies in 
    particular $\inv\aset\sset\tmthree$.
    We can apply the \ih (\ref{it:hyp1}) on $\tmthree$, yielding
    $\tmthree \in \NF\aset\sset\nonapp$.
    
    % Note that we may assume $\var \notin (\aset \cup \sset)$ by $\alpha$-conversion,
    By \cref{lem:nf_in_HAbs_or_Struct}, we have
    that $\tmthree$ is either in $\HAbs\aset$ or in $\Struct\sset$.
    We analyse both cases:
    \begin{itemize}
    \item $\tmthree \in \HAbs\aset$. 
      Then $\inv\aset\sset{\tmtwo\esub\var\tmthree}$ implies 
      in particular $\inv{\aset \cup \set\var}\sset\tmtwo$.
      Then $\tmtwo \in \Irred{\aset \cup \set\var}\sset\appflag$ by 
      the contraposition of \cref{lem:term_reducible_arg_abstraction_esub_reducible}.
      
      Moreover, if
      $\tmtwo\esub\var\tmthree \in \HAbs\aset$ holds, 
      it can be derived either by (1) \ruleHAbsSubi, meaning that 
      $\tmtwo \in \HAbs\aset$, and thus $\tmtwo \in \HAbs{\aset \cup \set\var}$ 
      by \cref{rem:habs_st}, or it can be derived either by 
      (2) \ruleHAbsSubii, meaning that $\tmtwo \in \HAbs{\aset \cup \set\var}$.
      So in this case, by the hypothesis in (\ref{it:hyp1}), we have 
      $\appflag = \nonapp$.
      We can then apply the \ih (\ref{it:hyp1}) on $\tmtwo$, yielding 
      $\tmtwo \in \NF{\aset \cup \set\var}\sset\appflag$.
      By applying rule \ruleUNFEsAbs we derive 
      $\tmtwo\esub\var\tmthree \in \NF\aset\sset\appflag$.
    \item $\tmthree \in \Struct\sset$.
      Then $\inv\aset\sset{\tmtwo\esub\var\tmthree}$ implies in 
      particular $\inv\aset{\sset \cup \set\var}\tmtwo$.
      Then $\tmtwo \in \Irred\aset{\sset \cup \set\var}\appflag$ by 
      the contraposition of \cref{lem:term_reducible_arg_struct_esub_reducible}.
      
      Since $\aset \disj \sset$ holds by the invariant, then
      $\tmthree \notin \HAbs\aset$ by \cref{lem:disjunction}, as
      $\tmthree \in \Struct\sset$.
      As a consequence, if $\tmtwo\esub\var\tmthree \in \HAbs\aset$,
      we necessarily have $\tmtwo \in \HAbs\aset$ by rule \ruleHAbsSubi,
      and not by rule \ruleHAbsSubii.
      So in this case, by the hypothesis in (\ref{it:hyp1}), we have 
      $\appflag = \nonapp$. 
      We can then apply the \ih (\ref{it:hyp1}) on $\tmtwo$, yielding 
      $\tmtwo \in \NF\aset{\sset \cup \set\var}\appflag$.
      By applying rule \ruleUNFEsStruct we derive 
      $\tmtwo\esub\var\tmthree \in \NF\aset\sset\appflag$.
    \end{itemize}
  \end{itemize}
\item \quad
  \begin{itemize}
  \item $\tm = \var$.
    Then, $\var \in \HAbs\aset$ is derived only by rule \ruleHAbsVar,
    so $\var \in \aset$.
    We can apply rule \ruleUSub, yielding
    $\var \tov{\rulesub\var\val}\aset\sset\app \val$, so we are done
    since $\var \in \Red\aset\sset\app$.
  \item $\tm = \lam\vartwo\tmtwo$. Immediate since $\abs\tm$ by definition.
  \item $\tm = \tmtwo \, \tmthree$. 
    This case is not possible since $\tmtwo \, \tmthree \in \HAbs\aset$ 
    never holds.
  \item $\tm = \tmtwo\esub\var\tmthree$.
    There are two cases depending on the rule we use to derive 
    $\tmtwo\esub\var\tmthree \in \HAbs\aset$:
    \begin{itemize}
    \item \ruleHAbsSubi.
      Then
      \[
        \inferrule{
          \tmtwo \in \HAbs\aset
          \sep
          \var \notin \aset
        }{
          \tmtwo\esub\var\tmthree \in \HAbs\aset
        }\ruleHAbsSubi
      \]

      If $\tmthree \in \Red\aset\sset\nonapp$, then applying rule 
      \ruleUEsR we yield 
      $\tmtwo\esub\var\tmthree \in \Red\aset\sset\appflag$, for any 
      $\appflag$, in particular $\appflag = \app$.

      Otherwise, $\tmthree \in \Irred\aset\sset\nonapp$.
      Moreover, $\inv\aset\sset{\tmtwo\esub\var\tmthree}$ implies in 
      particular $\inv\aset\sset\tmthree$. 
      We can then apply the \ih (\ref{it:hyp1}) on $\tmthree$, 
      yielding $\tmthree \in \NF\aset\sset\nonapp$.
      Then, $\tmthree \in \HAbs\aset \cup \Struct\sset$ by 
      \cref{lem:nf_in_HAbs_or_Struct}.
      We reason by cases:
      \begin{itemize}
      \item 
        If $\tmthree \in \HAbs\aset$, then
        $\inv\aset\sset{\tmtwo\esub\var\tmthree}$ implies in 
        particular $\inv{\aset \cup \set\var}\sset\tmtwo$.
        We can then apply the \ih (\ref{it:hyp2}) on $\tmtwo$, yielding
        either $\abs\tmtwo$ or $\tmtwo \in \Red{\aset \cup \set\var}\sset\app$.
        If the former, we conclude with $\abs{\tmtwo\esub\var\tmthree}$
        by definition.
        If the latter, then 
        $\tmtwo\esub\var\tmthree \in \Red\aset\sset\app$ by 
        \cref{lem:term_reducible_arg_abstraction_esub_reducible}, and 
        we are done.
      \item
        If $\tmthree \in \Struct\sset$, then
        $\inv\aset\sset{\tmtwo\esub\var\tmthree}$ implies in 
        particular $\inv\aset{\sset \cup \set\var}\tmtwo$.
        We can then apply the \ih (\ref{it:hyp2}) on $\tmtwo$, yielding
        either $\abs\tmtwo$ or $\tmtwo \in \Red\aset{\sset \cup \set\var}\app$.
        If the former, we conclude with $\abs{\tmtwo\esub\var\tmthree}$
        by definition.
        If the latter, then 
        $\tmtwo\esub\var\tmthree \in \Red\aset\sset\app$ by
        \cref{lem:term_reducible_arg_struct_esub_reducible}, and we 
        are done.
      \end{itemize}
    \item \ruleHAbsSubii.
      Then
      \[
        \inferrule{
          \tmtwo \in \HAbs{\aset \cup \set\var}
          \sep
          \var \notin \aset
          \sep
          \tmthree \in \HAbs\aset
        }{
          \tmtwo\esub\var\tmthree \in \HAbs\aset
        }\ruleHAbsSubii
      \]

      Moreover, $\inv\aset\sset{\tmtwo\esub\var\tmthree}$ implies in 
      particular $\inv{\aset \cup \set\var}\sset{\tmtwo}$.
      By the \ih (\ref{it:hyp2}) on $\tmtwo$, we yield either 
      $\abs\tmtwo$ or $\tmtwo \in \Red{\aset \cup \set\var}\sset\app$.
      If the former, we conclude with $\abs{\tmtwo\esub\var\tmthree}$
      by definition.
      If the latter, we apply \cref{lem:term_reducible_arg_abstraction_esub_reducible},
      yielding $\tmtwo\esub\var\tmthree \in \Red\aset\sset\app$, so
      we are done.
    \end{itemize}
  \end{itemize}
\end{enumerate}
\end{proof}

The following corollary combines soundness
(\cref{prop:soundness-characterization-normal-forms}) and 
completeness (\cref{prop:completeness_nf}) for a term $\tm$ in top-level 
position, \ie, when the \abstractionframe $\aset$ is empty, the 
structure frame $\sset$ is the set of all free variables of $\tm$,
and the positional flag is $\nonapp$:

\begin{corollary}[Characterisation of useful normal forms]
\label{thm:characterization_of_useful_normal_forms}
$\tm \in \NF{\emptyset}{\fv{\tm}}{\nonapp}$ iff $\tm \in \Irred{\emptyset}{\fv{\tm}}{\nonapp}$.
\end{corollary}

\label{ex:nf-useful}
An example of this result is the term $(\var \, \vartwo)\esub\vartwo\id$.
On one hand, the term is in $\Red\emptyset{\set\var}\nonapp$, since 
it cannot reduce using any reduction rule, and on the other hand it 
belongs to the set of normal forms by the following derivation:
\[
  \inferrule{
    \inferrule{
      \inferrule{
      }{
        \var \in \NF{\set\vartwo}{\set\var}\app
      }\ruleUNFVar
      \inferrule{
      }{
        \vartwo \in \NF{\set\vartwo}{\set\var}\nonapp
      }\ruleUNFVar
    }{
      \var \, \vartwo \in \NF{\set\vartwo}{\set\var}\nonapp
    }\ruleUNFApp
    \sep
    \inferrule{
    }{
      \id \in \NF\emptyset{\set\var}\nonapp
    }\ruleUNFLam
    \sep
    \inferrule{
    }{
      \id \in \HAbs\emptyset
    }\ruleHAbsLam
  }{
    (\var \, \vartwo)\esub\vartwo\id \in \NF\emptyset{\set\var}\nonapp
  }\ruleUNFEsAbs
\]

% -------------------------------------------------------------------

\subsection{Diamond Property}
\label{sec:diamond}

\begin{definition}[Expansion of abstraction and structure frames]
Let $\aset$ be an \abstractionframe.
We inductively define the \defn{expansion} of $\aset$ under $\sctx$,
written $\expansion\aset\sctx$, as follows:
\begin{align*}
    \expansion\aset\ctxhole 
  & \eqdef \aset
\\
    \expansion\aset{\sctxtwo\esub\var\tm} 
  & \eqdef 
    \left\{ 
      \begin{array}{ll}
        \expansion\aset\sctxtwo \cup \set\var & \text{if } \tm \in \HAbs\aset \\
        \expansion\aset\sctxtwo               & \text{otherwise}
      \end{array}
    \right.
\end{align*}

Analogously, let $\sset$ be a structure frame.
We inductively define the \defn{expansion} of $\sset$ under $\sctx$,
written $\expansion\sset\sctx$, as follows:
\begin{align*}
    \expansion\sset\ctxhole 
  & \eqdef \sset
\\
    \expansion\sset{\sctxtwo\esub\var\tm} 
  & \eqdef 
    \left\{ 
      \begin{array}{ll}
        \expansion\sset\sctxtwo \cup \set\var & \text{if } \tm \in \Struct\sset \\
        \expansion\sset\sctxtwo               & \text{otherwise}
      \end{array}
    \right.
\end{align*}
\end{definition}

\begin{lemma}
\label{lem:t_Struct_t_StructExp}
  $\tm\sctx \in \Struct\sset$ if and only if
  $\tm \in \Struct{\expansion\sset\sctx}$.
\end{lemma}
% Label: lem:t_Struct_t_StructExp

\begin{proof}
We prove both implications by induction on the length of $\sctx$.
In the two directions we skip the base case, as it is straightforward.

We first show $\tm\sctx \in \Struct\sset$ implies 
$\tm \in \Struct{\expansion\sset\sctx}$.
Then, we are in the case where $\sctx = \sctxtwo\esub\var\tmtwo$. 
The judgement $\tm\sctxtwo\esub\var\tmtwo \in \Struct\sset$ 
can be derived either by rule \ruleStructSubi or by rule \ruleStructSubii:
\begin{itemize}
\item \ruleStructSubi. 
  Then, $\tm\sctxtwo \in \Struct{\sset}$ and $\var \notin \sset$.
  We apply the \ih on $\sctxtwo$, yielding 
  $\tm \in \Struct{\expansion\sset\sctxtwo}$.
  Since $\expansion\sset\sctxtwo \subseteq \expansion\sset{\sctxtwo\esub\var\tmtwo}$
  by definition of $\expansion\sset{\sctxtwo\esub\var\tmtwo}$, then 
  $\Struct{\expansion\sset\sctxtwo} \subseteq \Struct{\expansion\sset{\sctxtwo\esub\var\tmtwo}}$
  by \cref{rem:habs_st}.
  Therefore, $\tm \in \Struct{\expansion\sset\sctx}$.
\item \ruleStructSubii. 
  Then, $\tm\sctxtwo \in \Struct{\sset \cup \set\var}$, 
  $\var \notin \sset$, and $\tmtwo \in \Struct\sset$.
  We apply the \ih on $\sctxtwo$, yielding 
  $\tm \in \Struct{\expansion\sset\sctxtwo \cup \set\var}$.
  Since $\tmtwo \in \Struct\sset$, then $\expansion\sset\sctx 
  = \expansion\sset\sctxtwo \cup \set\var$.
  We then conclude with $\tm \in \Struct{\expansion\sset\sctx}$.
\end{itemize}

We now show $\tm \in \Struct{\expansion\sset\sctx}$ 
implies $\tm\sctx \in \Struct\sset$.
And we are in the case where $\sctx = \sctxtwo\esub\var\tmtwo$, so 
$\tm \in \Struct{\expansion\sset{\sctxtwo\esub\var\tmtwo}}$ by hypothesis.
We have two cases, by definition of $\expansion\sset{\sctxtwo\esub\var\tmtwo}$.
\begin{itemize}
\item 
  If $\tmtwo \in \Struct\sset$, then 
  $\expansion\sset{\sctxtwo\esub\var\tmtwo} 
  = \expansion\sset\sctxtwo \cup \set\var$, so 
  $\tm \in \Struct{\expansion\sset\sctxtwo \cup \set\var}$, and
  $\tm\sctxtwo \in \Struct{\sset \cup \set\var}$ by the \ih on $\sctxtwo$.
  Moreover, $\var \notin \expansion\sset\sctxtwo$ by $\alpha$-conversion
  so in particular, $\var \notin \sset$, as 
  $\sset \subseteq \expansion\sset\sctxtwo$.
  Applying rule \ruleStructSubii, we yield
  $\tm\sctxtwo\esub\var\tmtwo \in \Struct\sset$.
\item 
  If $\tmtwo \notin \Struct\sset$, then
  $\expansion\sset{\sctxtwo\esub\var\tmtwo} = \expansion\sset\sctxtwo$,
  so $\tm \in \Struct{\expansion\sset\sctxtwo}$.
  Moreover, $\var \notin \expansion\sset\sctxtwo$ by $\alpha$-conversion
  so in particular, $\var \notin \sset$.
  Then, $\tm\sctxtwo \in \Struct\sset$ by the \ih on $\sctxtwo$.
  Applying rule \ruleStructSubi, we yield
  $\tm\sctxtwo\esub\var\tmtwo \in \Struct\sset$.
\end{itemize}
\end{proof}

\begin{lemma}
\label{lem:habs_rulesub_closed_reduction}
Let $\inv{\aset \cup \set\var}\sset\tm$, and 
let $\asettwo$ be a set of variables disjoint from $\aset$.
If $\tm \tov{\rulesub\var\val}{\aset \cup \set\var}\sset\appflag \tm'$
with $\val \in \HAbs{\aset \cup \asettwo}$
and $\tm \in \HAbs{\aset \cup \set\var}$,
then $\tm' \in \HAbs{\aset \cup \set\var \cup \asettwo}$.
\end{lemma}
% Label: lem:habs_rulesub_closed_reduction

\begin{proof}
By induction on the derivation of 
$\tm \tov{\rulesub\var\val}{\aset \cup \set\var}\sset\appflag \tm'$.
Note that cases \ruleUAppL and \ruleUAppR are impossible since 
$\tm = \tmtwo \, \tmthree$, which cannot be an element of 
$\HAbs{\aset \cup \set\var}$ by \cref{rem:habs_st}.
We analyse the remaining cases.
\begin{itemize}
\item \ruleUSub.
  Then
  $\tm = \var \tov{\rulesub\var\val}{\aset \cup \set\var}\sset\app \val = \tm'$,
  with $\appflag = \app$.
  Since $\val \in \HAbs{\aset \cup \asettwo}$ by hypothesis,
  then $\val \in \HAbs{\aset \cup \set\var \cup \asettwo}$ by \cref{rem:habs_st}, 
  as $\aset \cup \asettwo \subseteq \aset \cup \set\var \cup \asettwo$.
\item \ruleUEsR.
  Then
  \[
    \inferrule{
      \tmthree \tov{\rulesub\var\val}{\aset \cup \set\var}\sset\nonapp \tmthree'
    }{
      \tm = \tmtwo\esub\vartwo\tmthree
      \tov{\rulesub\var\val}{\aset \cup \set\var}\sset\appflag 
      \tmtwo\esub\vartwo{\tmthree'} = \tm'
    }\ruleUEsR
  \]

  We consider two cases: whether 
  $\tmtwo\esub\vartwo\tmthree \in \HAbs{\aset \cup \set\var}$ is 
  derived by rule \ruleHAbsSubi or by rule \ruleHAbsSubii.
  For both cases we consider $\vartwo \notin \asettwo$ by $\alpha$-conversion.
  \begin{itemize}
  \item \ruleHAbsSubi.
    Then $\tmtwo \in \HAbs{\aset \cup \set\var}$ 
    and $\vartwo \notin \aset \cup \set\var$.
    We have $\tmtwo \in \HAbs{\aset \cup \set\var \cup \asettwo}$
    by \cref{rem:habs_st}.
    We apply rule \ruleHAbsSubi, yielding
    $\tmtwo\esub\vartwo{\tmthree'} \in \HAbs{\aset \cup \set\var \cup \asettwo}$.
  \item \ruleHAbsSubii.
    Then $\tmtwo \in \HAbs{\aset \cup \set\var \cup \set\vartwo}$,
    $\vartwo \notin \aset \cup \set\var$,
    and $\tmthree \in \HAbs{\aset \cup \set\var}$.
    We have $\tmtwo \in \HAbs{\aset \cup \set\var \cup \set\vartwo \cup \asettwo}$
    and $\tmthree \in \HAbs{\aset \cup \set\var \cup \asettwo}$
    by \cref{rem:habs_st}, so 
    $\tmthree' \in \HAbs{\aset \cup \set\var \cup \asettwo}$ 
    by the \ih on $\tmthree$.
    We apply rule \ruleHAbsSubii, yielding
    $\tmtwo\esub\vartwo{\tmthree'} \in \HAbs{\aset \cup \set\var \cup \asettwo}$.
  \end{itemize}
\item \ruleUEsLAbs.
  Then
  \[
    \inferrule{
      \tmtwo \tov{\rulesub\var\val}{\aset \cup \set\var \cup \set\vartwo}\sset\appflag \tmtwo'
      \sep
      \tmthree \in \HAbs{\aset \cup \set\var}
      \sep
      \vartwo \notin (\aset \cup \set\var) \cup \sset
      \sep
      \vartwo \notin \fv{\rulesub\var\val}
    }{
      \tm = \tmtwo\esub\vartwo\tmthree
      \tov{\rulesub\var\val}{\aset \cup \set\var}\sset\appflag 
      \tmtwo'\esub\vartwo\tmthree = \tm'
    }\ruleUEsLAbs
  \]
  The judgement $\tmtwo\esub\var\tmthree \in \HAbs{\aset \cup \set\var}$
  can be derived either by rule \ruleHAbsSubi or by rule \ruleHAbsSubii.
  The former has $\tmtwo \in \HAbs{\aset \cup \set\var}$ as premise,
  and since $\aset \cup \set\var \subseteq \aset \cup \set\var \cup \set\vartwo$,
  then $\tmtwo \in \HAbs{\aset \cup \set\var \cup \set\vartwo}$
  by \cref{rem:habs_st}.
  And it is also the case that $\tmthree \in \HAbs{\aset \cup \set\var}$
  by premise of the rule \ruleUEsLAbs, so in both cases we proceed as follows.
  Since $\inv{\aset \cup \set\var}\sset{\tmtwo\esub\var\tmthree}$
  then in particular $\inv{\aset \cup \set\var \cup \set\vartwo}\sset\tmtwo$.
  Moreover, $\val \in \HAbs{\aset\cup\set\vartwo\cup\asettwo}$
  by \cref{rem:habs_st}, so we apply the \ih on $\tmtwo$, yielding
  $\tmtwo' \in \HAbs{(\aset \cup \set\vartwo) \cup \set\var \cup \asettwo}$.
  Moreover, $\tmthree \in \HAbs{\aset \cup \set\var \cup \asettwo}$
  by \cref{rem:habs_st}, as 
  $\aset \cup \set\var \subseteq \aset \cup \set\var \cup \asettwo$.
  And we assume $\vartwo \notin \asettwo$ by $\alpha$-conversion,
  so we apply rule \ruleHAbsSubii, yielding
  $\tmtwo'\esub\var\tmthree \in \HAbs{\aset \cup \set\var \cup \asettwo}$.
\item \ruleUEsLStruct.
  Then
  \[
    \inferrule{
      \tmtwo \tov{\rulesub\var\val}{\aset \cup \set\var}{\sset \cup \set\vartwo}\appflag \tmtwo'
      \sep
      \tmthree \in \Struct\sset
      \sep
      \vartwo \notin (\aset \cup \set\var) \cup \sset
      \sep
      \vartwo \notin \fv{\rulesub\var\val}
    }{
      \tm = \tmtwo\esub\vartwo\tmthree 
      \tov{\rulesub\var\val}{\aset \cup \set\var}\sset\appflag 
      \tmtwo'\esub\vartwo\tmthree = \tm'
    }\ruleUEsLStruct
  \]

  We consider two cases depending on whether
  $\tmtwo\esub\var\tmthree \in \HAbs{\aset \cup \set\var}$
  is derived using rule \ruleHAbsSubi or rule \ruleHAbsSubii:
  \begin{itemize}
  \item \ruleHAbsSubi.
    Then $\tmtwo \in \HAbs{\aset \cup \set\var}$ 
    and $\vartwo \notin \aset \cup \set\var$.
    Since $\inv{\aset \cup \set\var}\sset{\tmtwo\esub\var\tmthree}$
    then in particular $\inv{\aset \cup \set\var}{\sset \cup \set\vartwo}\tmtwo$.
    We apply the \ih on $\tmtwo$, yielding
    $\tmtwo' \in \HAbs{\aset \cup \set\var \cup \asettwo}$.
    Applying rule \ruleHAbsSubi we obtain
    $\tmtwo'\esub\var\tmthree \in \HAbs{\aset \cup \set\var \cup \asettwo}$.
  \item \ruleHAbsSubii.
    Then $\tmtwo \in \HAbs{\aset \cup \set\var \cup \set\vartwo}$,
    $\vartwo \notin \aset \cup \set\var$, and $\tmthree \in \HAbs{\aset \cup \set\var}$.
    Moreover, $\tmthree \in \Struct\sset$ by premise of the rule \ruleUEsLStruct.
    Then $(\aset \cup \set\var) \cap \sset \neq \emptyset$
    by contraposition of \cref{lem:disjunction},
    which contradicts $\inv{\aset \cup \set\var}\sset{\tmtwo\esub\var\tmthree}$.
    Hence this case is impossible.
  \end{itemize}
\end{itemize}
\end{proof}

\begin{definition}[Hereditary variables]
The set of \defn{hereditary variables} under a variable $\var$ is
written $\HVar\var$ and is defined inductively as follows:
\[
  \inferrule{
  }{
    \var \in \HVar\var
  }\ruleHVarVar
  \HS
  \inferrule{
    \tm \in \HVar\var
    \sep
    \var \neq \vartwo
  }{
    \tm\esub\vartwo\tmtwo \in \HVar\var
  }\ruleHVarSubi
  \HS
  \inferrule{
    \tm \in \HVar\vartwo
    \sep
    \tmtwo \in \HVar\var
  }{
    \tm\esub\vartwo\tmtwo \in \HVar\var
  }\ruleHVarSubii
\]
\end{definition}

\begin{remark}
\label{rem:t_varx_tL_varx}
If $\tm \in \HVar\var$ and $\var \notin \dom\sctx$,
then $\tm\sctx \in \HVar\var$.
\end{remark}

\begin{lemma}
\label{lem:varx_included_habsxext}
Let $\aset$ be an \abstractionframe and let $\var$ be any variable.
If $\tm \in \HVar\var$, then $\tm \in \HAbs{\aset \cup \set\var}$.
\end{lemma}
\hiddenproof{
  The proof is straightforward by induction on the derivation of the 
  judgement $\tm \in \HVar\var$.
}{
  % Label: lem:varx_included_habsxext

\begin{proof}
By induction on the judgement $\tm \in \HVar\var$.
\begin{enumerate}
\item \ruleHVarVar.
  Then $\tm = \var$.
  Since $\var \in \set\var$,
  then $\var \in \HAbs{\aset \cup \set\var}$
  by rule \ruleHAbsVar.
\item \ruleHVarSubi.
  Then $\tm = \tmtwo\esub\vartwo\tmthree$.
  The judgement $\tmtwo\esub\vartwo\tmthree \in \HVar\var$ is derived 
  from $\tmtwo \in \HVar\var$ and $\var \neq \vartwo$.
  We apply the \ih on $\tmtwo$,
  yielding $\tmtwo \in \HAbs{\aset \cup \set\var}$.
  By $\alpha$-conversion we assume $\vartwo \notin \aset \cup \set\var$.
  Applying rule \ruleHAbsSubi we obtain
  $\tmtwo \esub\var\tmthree \in \HAbs{\aset \cup \set\var}$.
\item \ruleHVarSubii.
  Then $\tm = \tmtwo\esub\vartwo\tmthree$.
  The judgement $\tmtwo\esub\vartwo\tmthree \in \HVar\var$
  is derived from
  $\tmtwo \in \HVar\vartwo$ and $\tmthree \in \HVar\var$.
  We apply the \ih on $\tmtwo$,
  yielding $\tmtwo \in \HAbs{\aset \cup  \set\var \cup \set\vartwo}$,
  where we assume $\vartwo \notin \aset \cup \set\var$
  by $\alpha$-conversion.
  On the other hand,
  $\tmthree \in \HAbs{\aset \cup \set\var}$
  by the \ih on $\tmthree$.
  Applying rule \ruleHAbsSubii we conclude
  $\tmtwo\esub\var\tmthree \in \HAbs{\aset \cup \set\var}$.
\end{enumerate}
\end{proof}
  
}

\begin{lemma}
\label{lem:t_habs_hvar}
Let $\aset$ be an \abstractionframe, and let 
$\aset = \aset_1 \cup \aset_2$ be an arbitrary partition.
If $\tm \in \HAbs\aset$, then either $\tm \in \HAbs{\aset_1}$
or there exists $\var \in \aset_2$ such that $\tm \in \HVar\var$.
\end{lemma}
%Label: lem:t_habs_hvar

\begin{proof}
By induction on the derivation of
$\tm \in \HAbs{\aset_1 \cup \aset_2}$.
\begin{itemize}
\item \ruleHAbsVar.
  Then $\tm = \vartwo$,
  and $\vartwo \in \aset_1 \cup \aset_2$
  by premise of the rule \ruleHAbsVar.
  There are two cases, depending on whether
  $\vartwo \in \aset_1$ or $\vartwo \in \aset_2$.
  If $\vartwo \in \aset_1$, then $\vartwo \in \HAbs{\aset_1}$
  by rule \ruleHAbsVar.
  If $\vartwo \in \aset_2$, then $\vartwo \in \HVar\vartwo$
  by rule \ruleHVarVar.
\item \ruleHAbsLam.
  Then $\tm = \lam\vartwo\tmtwo$.
  We apply rule \ruleHAbsLam, 
  yielding $\lam\vartwo\tmtwo \in \HAbs{\aset_1}$.
\item \ruleHAbsSubi.
  Then $\tm = \tmtwo\esub\vartwo\tmthree$.
  The judgement 
  $\tmtwo\esub\vartwo\tmthree \in \HAbs{\aset_1 \cup \aset_2}$
  is derived from $\tmtwo \in \HAbs{\aset_1 \cup \aset_2}$ 
  and $\vartwo \notin \aset_1 \cup \aset_2$.
  We take $\aset' = \aset_1 \cup (\aset_2 \cup \set{\vartwo})$ and 
  apply the \ih on $\tmtwo$, yielding two possible cases:
  \begin{itemize}
  \item $\tmtwo \in \HAbs{\aset_1}$.
    Since $\vartwo \notin \aset_1$,
    we can apply rule \ruleHAbsSubi, yielding
    $\tmtwo\esub\vartwo\tmthree \in \HAbs{\aset_1}$.
  \item $\tmtwo \in \HVar\var$, for some $\var \in \aset_2$.
    Since $\vartwo \notin \aset_2$, then $\var \neq \vartwo$
    and we can apply rule \ruleHVarSubi,
    yielding $\tmtwo\esub\vartwo\tmthree \in \HVar\vartwo$.
  \end{itemize}
\item \ruleHAbsSubii.
  Then $\tm = \tmtwo\esub\vartwo\tmthree$.
  The judgement 
  $\tmtwo\esub\vartwo\tmthree \in \HAbs{\aset_1 \cup \aset_2}$
  is derived from 
  $\tmtwo \in \HAbs{(\aset_1 \cup \aset_2) \cup \set\vartwo}$, 
  $\vartwo \notin \aset_1 \cup \aset_2$,
  and $\tmthree \in \HAbs{\aset_1 \cup \aset_2}$.
  By the \ih on $\tmtwo$ we have two possible cases:
  \begin{itemize}
  \item $\tmtwo \in \HAbs{\aset_1}$.
    Since $\vartwo \notin \aset_1$,
    we can apply rule \ruleHAbsSubi, yielding
    $\tmtwo\esub\vartwo\tmthree \in \HAbs{\aset_1}$.
  \item
    $\tmtwo \in \HVar\var$ for some 
    $\var \in (\aset_2 \cup \set\vartwo)$.
    We have two cases, depending on whether $\var = \vartwo$ or not:
    \begin{itemize}
    \item
      If $\var = \vartwo$, then by the \ih on $\tmthree$ 
      we have two possible subcases:
      \begin{itemize}
      \item $\tmthree \in \HAbs{\aset_1}$.
        Since
        $\vartwo \notin \aset_1$,
        and $\tmtwo \in \HAbs{\aset_1 \cup \set{\vartwo}}$
        by \cref{lem:varx_included_habsxext},
        then $\tmtwo\esub\vartwo\tmthree \in \HAbs{\aset_1}$
        by rule \ruleHAbsSubii.
      \item $\tmthree \in \HVar\varthree$ for some $\varthree \in \aset_2$.
        Then 
        $\tmtwo\esub\vartwo\tmthree \in \HVar\varthree$
        by rule \ruleHVarSubii.
      \end{itemize}
    \item 
      If $\var \neq \vartwo$, then
      $\tmtwo\esub\vartwo\tmthree \in \HVar\var$
      by rule \ruleHVarSubi.
    \end{itemize}
  \end{itemize}
\end{itemize}
\end{proof}

\begin{lemma}
\label{lem:tL_hAbs_t_hAbsExp}
$\tm\sctx \in \HAbs\aset$ if and only if
$\tm \in \HAbs{\expansion\aset\sctx}$.
\end{lemma}
\hiddenproof{
  The proof is straightforward by induction, analogous to \cref{lem:t_Struct_t_StructExp}.
}{
  % Label: lem:tL_hAbs_t_hAbsExp

\begin{proof}
We prove both implications by induction on the length of $\sctx$.
In the two directions we skip the base case, as it is straightforward.

We first show $\tm\sctx \in \HAbs\aset$ implies 
$\tm \in \HAbs{\expansion\aset\sctx}$.
Then, we are in the case where $\sctx = \sctxtwo\esub\var\tmtwo$.
The judgement $\tm\sctxtwo\esub\var\tmtwo \in \HAbs\aset$ can be 
derived either by rule \ruleHAbsSubi or by rule \ruleHAbsSubii:
\begin{itemize}
\item \ruleHAbsSubi. 
  Then $\tm\sctxtwo \in \HAbs\aset$ and $\var \notin \aset$.
  We apply the \ih on $\sctxtwo$, yielding
  $\tm \in \HAbs{\expansion\aset\sctxtwo}$.
  Since
  $\expansion\aset\sctxtwo \subseteq \expansion\aset{\sctxtwo\esub\var\tmtwo}$
  by definition of
  $\expansion\aset{\sctxtwo\esub\var\tmtwo}$,
  then
  $\HAbs{\expansion\aset\sctxtwo} \subseteq \HAbs{\expansion\aset{\sctxtwo\esub\var\tmtwo}}$
  by \cref{rem:habs_st}.
  Therefore
  $\tm \in \HAbs{\expansion\aset\sctx}$.
\item \ruleHAbsSubii. 
  Then
  $\tm\sctxtwo \in \HAbs{\aset \cup \set\var}$, $\var \notin \aset$, 
  and $\tmtwo \in \HAbs\aset$.
  We apply the \ih on $\sctxtwo$, yielding 
  $\tm \in \HAbs{\expansion\aset\sctxtwo \cup \set\var}$.
  Since $\tmtwo \in \HAbs\aset$,
  then $\expansion\aset\sctx = \expansion\aset\sctxtwo \cup \set\var$.
\end{itemize}

We now show $\tm \in \HAbs{\expansion\aset\sctx}$ implies $\tm\sctx \in \HAbs\aset$.
And we are in the case where $\sctx = \sctxtwo\esub\var\tmtwo$,
so $\tm \in \HAbs{\expansion\aset{\sctxtwo\esub\var\tmtwo}}$.
We have to analyse two cases by definition of 
$\expansion\aset{\sctxtwo\esub\var\tmtwo}$, depending on whether 
$\tmtwo \in \HAbs\aset$ or not:
\begin{itemize}
\item $\tmtwo \in \HAbs\aset$.
  Then
  $\expansion\aset{\sctxtwo\esub\var\tmtwo} 
  = \expansion\aset\sctxtwo \cup \set\var$,
  so $\tm \in \HAbs{\expansion\aset\sctxtwo \cup \set\var}$.
  By $\alpha$-conversion, we assume 
  $\var \notin \expansion\aset\sctxtwo$.
  There are two possible subcases by \cref{lem:t_habs_hvar}:
  \begin{itemize}
  \item 
    If $\tm \in \HAbs{\expansion\aset\sctxtwo}$,
    then $\tm\sctxtwo \in \HAbs\aset$ by the \ih on $\sctxtwo$.
    Applying rule \ruleHAbsSubi, we yield 
    $\tm\sctxtwo\esub\var\tmtwo \in \HAbs\aset$.
  \item 
    If $\tm \in \HVar\var$, 
    then $\tm\sctxtwo \in \HVar{\var}$
    by \cref{rem:t_varx_tL_varx},
    and $\HVar\var \subseteq \HAbs{\aset \cup \set\var}$
    by \cref{lem:varx_included_habsxext},
    so we have $\tm\sctxtwo \in \HAbs{\aset \cup \set\var}$.
    By rule \ruleHAbsSubii we conclude
    $\tm\sctxtwo\esub\var\tmtwo \in \HAbs\aset$.
  \end{itemize}
\item $\tmtwo \notin \HAbs\aset$.
  Then 
  $\expansion\aset{\sctxtwo\esub\var\tmtwo} 
  = \expansion\aset\sctxtwo$,
  so $\tm \in \HAbs{\expansion\aset\sctxtwo}$.
  We apply the \ih on $\sctxtwo$, yielding
  $\tm\sctxtwo \in \HAbs\aset$.
  By $\alpha$-conversion we assume
  $\var \notin \expansion\aset\sctxtwo$,
  and in particular $\var \notin \aset$.
  Applying rule \ruleHAbsSubi
  we conclude $\tm\sctxtwo\esub\var\tmtwo \in \HAbs\aset$.
\end{itemize}
\end{proof}

}

\begin{lemma}[Hereditary abstractions and structures are closed by reduction]
\label{lem:HAbs_Struct_closed_reduction}
Let $\inv\aset\sset\tm$.
\begin{enumerate}
\item
  Let $\tm \in \HAbs\aset$
  and $\tm \tov\rulename\aset\sset\appflag \tm'$
  where $\rulename \in \set{\ruledb, \rulelsv}$.
  Then $\tm' \in \HAbs\aset$.
\item
  If $\tm \in \Struct\sset$
  and $\tm \tov\rulename\aset\sset\appflag \tm'$,
  then $\tm' \in \Struct\sset$.
\end{enumerate}
\end{lemma}
\hiddenproof{
  By induction on the derivation of
  $\tm \tov\rulename\aset\sset\appflag \tm'$.
}{
  % Label: HAbs_Struct_closed_reduction

\begin{proof}
We prove each item independently, by induction on the derivation of 
$\tm \tov\rulename\aset\sset\appflag \tm'.
$\begin{enumerate}
\item 
  Cases \ruleUDb, \ruleUAppL, and \ruleUAppR are impossible since
  $\tm = \tmtwo\,\tmthree$, which cannot be an element of $\HAbs\aset$
  by \cref{rem:habs_st}. We analyse the remaining cases:
  \begin{itemize}
  % \item \ruleUDb.
  %   Then
  %   $\tm = (\lam\var\tmtwo)\sctx \, \tmthree 
  %   \tov\ruledb\aset\sset\appflag \tmtwo\esub\var\tmthree\sctx = \tm'$,
  %   with $\rulename = \ruledb$.
  %   This case does not apply since no applications belong to 
  %   $\HAbs\aset$ by \cref{rem:habs_st}, contradicting the hypothesis
  %   $\tm \in \HAbs\aset$.
  \item \ruleULsv.
    Then
    \[
      \indrule{\ruleULsv}{
        \tmtwo \tov{\rulesub\var\val}{\aset \cup \set\var}\sset\appflag \tmtwo'
        \sep
        \var \notin \aset \cup \sset
        \sep
        \val\sctx \in \HAbs\aset
      }{
        \tm = \tmtwo\esub\var{\val\sctx} 
        \tov\rulelsv\aset\sset\appflag 
        \tmtwo'\esub\var\val\sctx = \tm'
      }
    \]
    where $\rulename = \rulelsv$.
    The judgement $\tmtwo\esub\var{\val\sctx} \in \HAbs\aset$ can be 
    derived either by rule \ruleHAbsSubi or by rule \ruleHAbsSubii.
    The former has $\tmtwo \in \HAbs\aset$ as premise, and since 
    $\aset \subseteq \aset \cup \set\var$, then 
    $\tmtwo \in \HAbs{\aset \cup \set\var}$ by \cref{rem:habs_st}.
    And it is also the case that $\val\sctx \in \HAbs\aset$ by 
    premise of the rule \ruleULsv, so in both cases we proceed as follows.
    Given $\inv\aset\sset{\tmtwo\esub\var\tmthree}$,
    then in particular $\inv{\aset \cup \set\var}\sset\tmtwo$ holds.
    We assume $\var \notin \domSctx\sctx$, and $\aset \disj \domSctx\sctx$
    by $\alpha$-conversion.
    Since $\val\sctx \in \HAbs\aset$, then 
    $\val \in \HAbs{\expansion\aset\sctx}$ 
    by \cref{lem:tL_hAbs_t_hAbsExp}.
    We apply \cref{lem:habs_rulesub_closed_reduction}, taking 
    $\asettwo \eqdef \expansion\aset\sctx \setminus \aset$, which 
    yields $\tmtwo' \in \HAbs{\expansion\aset\sctx \cup \set\var}$.
    Then $\tmtwo'\esub\var\val \in \HAbs{\expansion\aset\sctx}$
    by rule \ruleHAbsSubii, and 
    $\tmtwo'\esub\var\val\sctx \in \HAbs\aset$ by 
    \cref{lem:tL_hAbs_t_hAbsExp}.
  % \item \ruleUAppL.
  %   Then $\tm = \tmtwo \, \tmthree 
  %   \tov\rulename\aset\sset\appflag \tmtwo' \, \tmthree = \tm'$.
  %   This case is impossible since no applications belong to 
  %   $\HAbs\aset$ by \cref{rem:habs_st}, contradicting the 
  %   hypothesis $\tm \in \HAbs\aset$.
  % \item \ruleUAppR.
  %   Analogous to the previous case.
    % Then
    % $ 
    %   \tm = \tmtwo \, \tmthree
    %   \tov\rulename\aset\sset\appflag
    %   \tmtwo \, \tmthree' = \tm'
    % $.
    % This case is not possible since
    % $\tmtwo \, \tmthree$ cannot be an element of $\HAbs\aset$
    % by \cref{rem:habs_st}, contradicting the hypothesis
    % $\tm \in \HAbs\aset$.
  \item \ruleUEsR.
    Then
    \[
      \indrule{\ruleUEsR}{
        \tmthree \tov\rulename\aset\sset\nonapp \tmthree'
      }{
        \tm = \tmtwo\esub\var\tmthree
        \tov\rulename\aset\sset\appflag
        \tmtwo\esub\var{\tmthree'} = \tm'
      }
    \]
    We analyse two cases, 
    depending on which rule we use to derive
    $\tmtwo\esub{\var}{\tmthree} \in \HAbs{\aset}$:
    \begin{itemize}
    \item \ruleHAbsSubi.
      Then, $\tmtwo \in \HAbs\aset$ and $\var \notin \aset$,
      therefore $\tmtwo\esub\var{\tmthree'} \in \HAbs\aset$
      by rule \ruleHAbsSubi.
    \item \ruleHAbsSubii.
      Then, $\tmtwo \in \HAbs{\aset \cup \set\var}, \var \notin \aset$,
      and $\tmthree \in \HAbs\aset$.
      Since $\inv\aset\sset{\tmtwo\esub\var\tmthree}$,
      then $\inv\aset\sset\tmthree$.
      We apply the \ih on $\tmthree$, 
      yielding $\tmthree' \in \HAbs\aset$.
      Therefore,
      $\tmtwo\esub\var{\tmthree'} \in \HAbs\aset$
      by rule \ruleHAbsSubii.
    \end{itemize}
  \item \ruleUEsLAbs.
    Then
    \[
      \indrule{\ruleUEsLAbs}{
        \tmtwo \tov\rulename{\aset \cup \set\var}\sset\appflag \tmtwo'
        \sep
        \tmthree \in \HAbs\aset
        \sep
        \var \notin \aset \cup \sset
        \sep
        \var \notin \fv\rulename
      }{
        \tm = \tmtwo\esub\var\tmthree
        \tov\rulename\aset\sset\appflag
        \tmtwo'\esub\var\tmthree = \tm'
      }
    \]
    The judgement $\tmtwo\esub\var\tmthree \in \HAbs\aset$
    can be derived either by rule \ruleHAbsSubi or by rule \ruleHAbsSubii.
    Note that the former has $\tmtwo \in \HAbs\aset$ as premise, and since
    $\aset \subseteq \aset \cup \set\var$, 
    then $\tmtwo \in \HAbs{\aset \cup \set\var}$ by \cref{rem:habs_st}.
    And it is also the case that $\tmthree \in \HAbs\aset$
    by premise of the rule \ruleUEsLAbs, so in both cases we proceed as follows.
    Since $\inv\aset\sset{\tmtwo\esub\var\tmthree}$,
    then in particular $\inv{\aset \cup \set\var}\sset\tmtwo$.
    We apply the \ih on $\tmtwo$, yielding
    $\tmtwo' \in \HAbs{\aset \cup \set\var}$.
    By rule \ruleHAbsSubii we conclude
    $\tmtwo'\esub\var\tmthree \in \HAbs\aset$.
  \item \ruleUEsLStruct.
    Then
    \[
      \indrule{\ruleUEsLStruct}{
        \tmtwo \tov\rulename\aset{\sset \cup \set\var}\appflag \tmtwo'
        \sep
        \tmthree \in \Struct\sset
        \sep
        \var \notin \aset \cup \sset
        \sep
        \var \notin \fv\rulename
      }{
        \tm = \tmtwo\esub\var\tmthree
        \tov\rulename\aset\sset\appflag
        \tmtwo'\esub\var\tmthree = \tm'
      }
    \]
    Since $\inv\aset\sset{\tmtwo\esub\var\tmthree}$,
    then in particular $\inv\aset{\sset \cup \set\var}\tmtwo$.
    We analyse two cases, depending on which rule we use to derive
    $\tmtwo\esub\var\tmthree \in \HAbs\aset$:
    \begin{itemize}
    \item \ruleHAbsSubi.
      Then, $\tmtwo \in \HAbs\aset$ and $\var \notin \aset$.
      We apply the \ih on $\tmtwo$, yielding $\tmtwo' \in \HAbs\aset$.
      We apply rule \ruleHAbsSubi and conclude
      $\tmtwo'\esub\var\tmthree \in \HAbs\aset$.
    \item \ruleHAbsSubii.
      Then $\tmtwo \in \HAbs{\aset \cup \set\var}, \var \notin \aset$,
      and $\tmthree \in \HAbs\aset$.
      We also have $\tmthree \in \Struct\sset$
      by premise of the rule \ruleUEsLStruct.
      Then, $(\aset \cap \sset) \neq \emptyset$ by 
      contraposition of \cref{lem:disjunction}, which contradicts 
      $\inv\aset\sset{\tmtwo\esub\var\tmthree}$.
      Therefore this case is not possible.
    \end{itemize}
  \end{itemize}
\item \quad
  \begin{itemize}
  \item \ruleUDb.
    Then
    $\tm = (\lam\var\tmtwo)\sctx \, \tmthree 
     \tov\ruledb\aset\sset\appflag
     \tmtwo\esub\var\tmthree\sctx = \tm'$,
    with $\rulename = \ruledb$.
    The judgement $(\lam\var\tmtwo)\sctx \, \tmthree \in \Struct\sset$
    can only be derived from the rule \ruleStructApp,
    so $(\lam\var\tmtwo)\sctx \in \Struct\sset$.
    But by \cref{rem:habs_st}, 
    this term cannot be an element of $\Struct\sset$.
    Therefore this case is impossible.
  \item \ruleUSub.
    Then
    $\tm = \var \tov{\rulesub\var\val}{\aset' \cup \set\var}\sset\app 
    \val = \tm'$,
    with
    $\rulename = \rulesub\var\val$, $\aset = \aset' \cup \set\var$, 
    and $\appflag = \app$.
    Since $\var \in \aset' \cup \set\var$ and
    $\inv{\aset' \cup \set\var}\sset\var$ implies $\var \notin \sset$
    by \cref{lem:disjunction},
    then it is not the case that $\var \in \Struct\sset$,
    as we cannot apply rule \ruleStructVar, the only rule to derive
    such judgement.
    Therefore this case is impossible.
  \item \ruleULsv.
    Then
    \[
      \indrule{\ruleULsv}{
        \tmtwo \tov{\rulesub\var\val}{\aset \cup \set\var}\sset\appflag \tmtwo'
        \sep
        \var \notin \aset \cup \sset
        \sep
        \val\sctx \in \HAbs\aset
      }{
        \tm = \tmtwo\esub\var{\val\sctx} 
        \tov\rulelsv\aset\sset\appflag
        \tmtwo'\esub\var\val\sctx = \tm'
      }
    \]
    where $\rulename = \rulelsv$.
    We analyse two cases, depending on which rule we use to derive
    $\tmtwo\esub\var{\val\sctx} \in \Struct\sset$:
    \begin{itemize}
    \item \ruleStructSubi.
      Then $\tmtwo \in \Struct\sset$ and $\var \notin \sset$, so
      $\tmtwo' \in \Struct\sset$ by the \ih on $\tmtwo$.
      Since $\sset \subseteq \expansion\sset\sctx$
      then $\tmtwo' \in \Struct{\expansion\sset\sctx}$
      by \cref{rem:habs_st}.
      We apply rule \ruleStructSubi, yielding
      $\tmtwo'\esub\var\val \in \Struct{\expansion\sset\sctx}$.
      Then
      $\tmtwo'\esub\var\val\sctx \in \Struct\sset$
      by \cref{lem:t_Struct_t_StructExp}.
    \item \ruleStructSubii.
      Then $\tmtwo \in \Struct{\sset \cup \set\var}, \var \notin \sset$,
      and $\val\sctx \in \Struct\sset$.
      By premise of the rule \ruleULsv we also have
      $\val\sctx \in \HAbs\aset$, so $\aset \cap \sset \neq \emptyset$ 
      by contraposition of \cref{lem:disjunction},
      which contradicts $\inv\aset\sset{\tmtwo\esub\var{\val\sctx}}$.
      Therefore this case is impossible.
    \end{itemize}
  \item \ruleUAppL.
    Then
    \[
      \indrule{\ruleUAppL}{
        \tmtwo \tov\rulename\aset\sset\app \tmtwo'
      }{
        \tm = \tmtwo \, \tmthree
        \tov\rulename\aset\sset\appflag
        \tmtwo' \, \tmthree = \tm'
      }
    \]
    The judgement $\tmtwo \, \tmthree \in \Struct\sset$ 
    can only be derived by the rule \ruleStructApp,
    with premise $\tmtwo \in \Struct\sset$.
    Then $\tmtwo' \in \Struct\sset$
    by the \ih on $\tmtwo$.
    We apply rule \ruleStructApp, yielding
    $\tmtwo' \, \tmthree \in \Struct\sset$.
  \item \ruleUAppR.
    Then
    \[
      \indrule{\ruleUAppR}{
        \tmtwo \in \Struct\sset
        \sep
        \tmthree \tov\rulename\aset\sset\nonapp \tmthree'
      }{
        \tm = \tmtwo \, \tmthree
        \tov\rulename\aset\sset\appflag
        \tmtwo \, \tmthree' = \tm'
      }
    \]
    Since $\tmtwo \in \Struct\sset$ we apply rule \ruleStructApp, 
    yielding $\tmtwo \, \tmthree' \in \Struct\sset$.
  \item \ruleUEsR.
    Then
    \[
      \indrule{\ruleUEsR}{
        \tmthree \tov\rulename\aset\sset\nonapp \tmthree'
      }{
        \tm = \tmtwo\esub\var\tmthree
        \tov\rulename\aset\sset\appflag
        \tmtwo\esub\var{\tmthree'} = \tm'
      }
    \]
    We analyse two cases, depending on which rule we use to derive
    $\tmtwo\esub\var\tmthree \in \Struct\sset$:
    \begin{itemize}
    \item \ruleStructSubi.
      Then $\tmtwo \in \Struct\sset$ and $\var \notin \sset$.
      We apply rule \ruleStructSubi, and conclude
      $\tmtwo\esub\var{\tmthree'} \in \Struct\sset$.
    \item \ruleStructSubii.
      Then $\tmtwo \in \Struct{\sset \cup \set\var}, \var \notin \sset$
      and $\tmthree \in \Struct\sset$.
      We apply the \ih on $\tmthree$, yielding
      $\tmthree' \in \Struct\sset$.
      Therefore
      $\tmtwo\esub\var{\tmthree'} \in \Struct\sset$
      by rule \ruleStructSubii.
    \end{itemize}
  \item \ruleUEsLAbs.
    Then
    \[
      \indrule{\ruleUEsLAbs}{
        \tmtwo \tov\rulename{\aset \cup \set\var}\sset\appflag \tmtwo'
        \sep
        \tmthree \in \HAbs\aset
        \sep
        \var \notin \aset \cup \sset
        \sep
        \var \notin \fv\rulename
      }{
        \tm = \tmtwo\esub\var\tmthree
        \tov\rulename\aset\sset\appflag
        \tmtwo'\esub\var\tmthree = \tm'
      }
    \]
    We analyse two cases, depending on which rule we use to derive
    $\tmtwo\esub\var\tmthree \in \Struct\sset$:
    \begin{itemize}
    \item \ruleStructSubi.
      Then $\tmtwo \in \Struct\sset$ and $\var \notin \sset$.
      We apply the \ih on $\tmtwo$, yielding
      $\tmtwo' \in \Struct\sset$.
      We apply rule \ruleStructSubi and conclude
      $\tmtwo'\esub\var\tmthree \in \Struct\sset$.
    \item \ruleStructSubii.
      Then $\tmtwo \in \Struct{\sset \cup \set\var}, \var \notin \sset$
      and $\tmthree \in \Struct\sset$.
      By premise of the rule \ruleUEsLAbs we also have
      $\tmthree \in \HAbs\aset$.
      Then $\aset \cap \sset \neq \emptyset$ 
      by contraposition of \cref{lem:disjunction},
      which contradicts $\inv\aset\sset{\tmtwo\esub\var\tmthree}$.
      Therefore this case is impossible.
    \end{itemize}
  \item \ruleUEsLStruct.
    Then
    \[
      \indrule{\ruleUEsLStruct}{
        \tmtwo \tov\rulename\aset{\sset \cup \set\var}\appflag \tmtwo'
        \sep
        \tmthree \in \Struct\sset
        \sep
        \var \notin \aset \cup \sset
        \sep
        \var \notin \fv\rulename
      }{
        \tm = \tmtwo\esub\var\tmthree
        \tov\rulename\aset\sset\appflag
        \tmtwo'\esub\var\tmthree = \tm'
      }
    \]
    We have two cases, depending on which rule we use to derive
    $\tmtwo\esub\var\tmthree \in \Struct\sset$:
    \begin{itemize}
    \item \ruleStructSubi.
      Then $\tmtwo \in \Struct\sset$ and $\var \notin \sset$.
      Since $\sset \subseteq \sset \cup \set\var$, then
      $\tmtwo \in \Struct{\sset \cup \set\var}$ by \cref{rem:habs_st}.
      We apply the \ih on $\tmtwo$, yielding
      $\tmtwo' \in \Struct{\sset \cup \set\var}$.
      Since $\tmthree \in \Struct\sset$
      by premise of the rule \ruleUEsLStruct,
      we then apply rule \ruleStructSubii and conclude
      $\tmtwo'\esub\var\tmthree \in \Struct\sset$.
    \item \ruleStructSubii.
      Then $\tmtwo \in \Struct{\sset \cup \set\var}, \var \notin \sset$,
      and $\tmthree \in \Struct\sset$.
      Since $\sset \subseteq \sset \cup \set\var$, then
      $\tmtwo \in \Struct{\sset \cup \set\var}$ by \cref{rem:habs_st}.
      We apply the \ih on $\tmtwo$, yielding
      $\tmtwo' \in \Struct{\sset \cup \set\var}$.
      By rule \ruleStructSubii we conclude
      $\tmtwo'\esub\var\tmthree \in \Struct\sset$.
    \end{itemize}
  \end{itemize}
\end{enumerate}
\end{proof}

}

Given a substitution context $\sctx$, we can use the reduction rules
defined in \cref{sec:usefulcbv} to define the notion of \defn{useful
reduction for substitution contexts}, by just understanding these
elements as terms, where $\ctxhole$ is taken as a free variable.

\begin{remark}
\label{rem:valL_reduces_in_L} 
\mbox{}
\begin{enumerate}
\item
  If $(\lam\var\tm)\sctx \tov\rulename\aset\sset\app \tmtwo$
  and $\ctxhole \notin \aset \cup \sset$,
  then $\tmtwo$ is of the form $(\lam\var\tm)\sctx'$,
  and $\sctx \tov\rulename\aset{\sset \cup \set\ctxhole}\appflag \sctx'$.
\item
  If $\val\sctx \tov{\rulename}{\aset}{\sset}{\nonapp} \tm$
  and $\ctxhole \notin \aset \cup \sset$,
  then there exists $\sctx'$ such that $\tm = \val\sctx'$,
  and $\sctx \tov{\rulename}{\aset}{\sset \cup \set{\ctxhole}}{\appflag} \sctx'$.
\item
  If $\sctx \tov\rulename\aset{\sset \cup \set\ctxhole}\appflag \sctx'$,
  then $\tm\sctx \tov\rulename\aset\sset\appflag \tm\sctx'$.
\end{enumerate}
\end{remark}

\begin{lemma}
\label{lem:weakening_of_reducion_sets}
If $\tm \tov{\rulename}{\aset}{\sset}{\appflag} \tm'$,
then for all sets $\aset'$ and $\sset'$
such that $\aset \subseteq \aset'$ and $\sset \subseteq \sset'$
it holds that $\tm \tov{\rulename}{\aset'}{\sset'}{\appflag} \tm'$.
\end{lemma}
% Label: lem:weakening_of_reducion_sets

\begin{proof}
By induction on the derivation of 
$\tm \tov\rulename\aset\sset\appflag \tm'$.
\begin{itemize}
\item \ruleUDb.
  Then $\tm = (\lam\var\tmtwo)\sctx \, \tmthree 
  \tov\ruledb\aset\sset\appflag \tmtwo\esub\var\tmthree\sctx = \tm'$,
  where $\rulename = \ruledb$.
  By rule \ruleUDb, we conclude
  $(\lam\var\tmtwo)\sctx \, \tmthree 
  \tov\ruledb{\aset'}{\sset'}\appflag \tmtwo\esub\var\tmthree\sctx$.
\item \ruleUSub.
  Then
  $\tm = \var \tov{\rulesub\var\val}{\aset \cup \set\var}\sset\app \val = \tm'$,
  where $\rulename = \rulesub\var\val$ and $\appflag = \app$.
  By rule \ruleUSub, we conclude
  $\var \tov{\rulesub\var\val}{\aset' \cup \set\var}{\sset'}\app \val$.
\item \ruleULsv.
  Then
  \[
    \inferrule{
      \tmtwo 
        \tov{\rulesub\var\val}{\aset \cup \set\var}\sset\appflag
        \tmtwo'
      \sep
      \var \notin \aset \cup \sset
      \sep
      \val\sctx \in \HAbs\aset
    }{
      \tm = \tmtwo\esub\var{\val\sctx} \tov\rulelsv\aset\sset\appflag
      \tmtwo'\esub\var\val\sctx = \tm'
    }\ruleULsv
  \]
  where $\rulename = \rulelsv$.
  If $\aset \subseteq \aset'$ and $\sset \subseteq \sset'$, then in 
  particular $\aset \cup \set\var \subseteq \aset' \cup \set\var$.
  Therefore,
  $\tmtwo \tov{\rulesub\var\val}{\aset' \cup \set\var}{\sset'}\appflag \tmtwo'$
  by the \ih on $\tmtwo$.
  Moreover, $\val\sctx \in \HAbs{\aset'}$ by \cref{rem:habs_st}, and 
  we can always assume $\var \notin \aset' \cup \sset'$.
  Applying rule \ruleULsv, we conclude
  $\tmtwo\esub\var{\val\sctx} \tov\rulelsv{\aset'}{\sset'}\appflag
  \tmtwo'\esub\var\val\sctx$.
\item \ruleUAppL.
  Straightforward by the \ih.
  % Then
  % \[
  %   \inferrule{
  %     \tmtwo \tov\rulename\aset\sset\app \tmtwo'
  %   }{
  %     \tm = \tmtwo \, \tmthree \tov\rulename\aset\sset\appflag \tmtwo' \, \tmthree = \tm'
  %   }\ruleUAppL
  % \]
  % We have $\tmtwo \tov\rulename{\aset'}{\sset'}\app \tmtwo'$ by the 
  % \ih on $\tmtwo$.
  % Applying rule \ruleUAppL, we conclude with
  % $\tmtwo \, \tmthree \tov\rulename{\aset'}{\sset'}\appflag \tmtwo' \, \tmthree$.
\item \ruleUAppR.
  Then
  \[
    \inferrule{
      \tmtwo \in \Struct\sset
      \sep
      \tmthree \tov\rulename\aset\sset\nonapp \tmthree'
    }{
      \tm = \tmtwo \, \tmthree \tov\rulename\aset\sset\appflag \tmtwo \, \tmthree' = \tm'
    }\ruleUAppR
  \]
  We have $\tmthree \tov\rulename{\aset'}{\sset'}\nonapp \tmthree'$
  by the \ih on $\tmthree$.
  Moreover, $\tmtwo \in \Struct{\sset'}$ by \cref{rem:habs_st}, so 
  applying rule \ruleUAppR, we conclude with
  $\tmtwo \, \tmthree \tov\rulename{\aset'}{\sset'}\appflag \tmtwo \, \tmthree'$.
\item \ruleUEsR.
  Straightforward by the \ih.
  % Then
  % \[
  %   \inferrule{
  %     \tmthree \tov\rulename\aset\sset\nonapp \tmthree'
  %   }{
  %     \tm = \tmtwo\esub\var\tmthree \tov\rulename\aset\sset\appflag \tmtwo\esub\var{\tmthree'} = \tm'
  %   }\ruleUEsR
  % \]
  % Then $\tmthree \tov\rulename{\aset'}{\sset'}\nonapp \tmthree'$ by 
  % the \ih on $\tmthree$.
  % We can then apply rule \ruleUEsR, yielding
  % $\tmtwo\esub\var\tmthree \tov\rulename{\aset'}{\sset'}\appflag \tmtwo\esub\var{\tmthree'}$.
\item \ruleUEsLAbs.
  Then
  \[
    \inferrule{
      \tmtwo \tov\rulename{\aset \cup \set\var}\sset\appflag \tmtwo'
      \sep
      \tmthree \in \HAbs\aset
      \sep
      \var \notin \aset \cup \sset
      \sep
      \var \notin \fv\rulename
    }{
      \tm = \tmtwo\esub\var\tmthree \tov\rulename\aset\sset\appflag \tmtwo'\esub\var\tmthree = \tm'
    }\ruleUEsLAbs
  \]
  As $\aset \cup \set{\var} \subseteq \aset' \cup \set{\var}$, we can
  then apply the \ih on $\tmtwo$, yielding
  $\tmtwo \tov\rulename{\aset' \cup \set\var}{\sset'}\appflag \tmtwo'$.
  Moreover, $\tmthree \in \HAbs{\aset'}$ by \cref{rem:habs_st}.
  We assume $\var \notin \aset' \cup \sset'$ also, by $\alpha$-conversion.
  Applying rule \ruleUEsLAbs, we conclude with
  $\tmtwo\esub\var\tmthree \tov\rulename{\aset'}{\sset'}\appflag \tmtwo'\esub\var\tmthree$.
\item \ruleUEsLStruct.
  Analogous to the previous case.
  % Then
  % \[
  %   \inferrule{
  %     \tmtwo \tov\rulename\aset{\sset \cup \set\var}\appflag \tmtwo'
  %     \sep
  %     \tmthree \in \Struct\sset
  %     \sep
  %     \var \notin \aset \cup \sset
  %     \sep
  %     \var \notin \fv\rulename
  %   }{
  %     \tm = \tmtwo\esub\var\tmthree \tov\rulename\aset\sset\appflag \tmtwo'\esub\var\tmthree = \tm'
  %   }\ruleUEsLStruct
  % \]
  % As $\sset \cup \set\var \subseteq \sset' \cup \set\var$, we can 
  % then apply the \ih on $\tmtwo$, yielding
  % $\tmtwo \tov\rulename{\aset'}{\sset' \cup \set\var}\appflag \tmtwo'$.
  % Moreover, $\tmthree \in \Struct{\sset'}$ by \cref{rem:habs_st}.
  % We assume $\var \notin \aset' \cup \sset'$ also, by $\alpha$-conversion.
  % Applying rule \ruleUEsLStruct, we conclude with
  % $\tmtwo\esub\var\tmthree \tov\rulename{\aset'}{\sset'}\appflag \tmtwo'\esub\var\tmthree$.
\end{itemize}
\end{proof}

\begin{restatable}[Towards the Diamond Property]{proposition}{diamond}
\label{prop:diamond-property}
Let $\inv\aset\sset\tm$ and suppose that
$\tm \tov{\rulename_1}\aset\sset\appflag \tm_1$ and
$\tm \tov{\rulename_2}\aset\sset\appflag \tm_2$, with $\tm_1 \neq \tm_2$.
Assume that if $\rulename_1 = \rulesub\var{\val_1}$ and 
$\rulename_2 = \rulesub\var{\val_2}$, then $\val_1 = \val_2$.
Let $\asettwo^1,\asettwo^2$ be sets of variables, disjoint from $\aset$,
such that if $\rulename_i = \rulesub\var{\val_i}$, then 
$\val_i \in \HAbs{\aset \cup \asettwo^i}$, for all $i \in \set{1,2}$.
Then, there exists a term $\tm'$ such that
$\tm_1 \tov{\rulename_2}{\aset \cup \asettwo^1}\sset\appflag \tm'$ and
$\tm_2 \tov{\rulename_1}{\aset \cup \asettwo^2}\sset\appflag \tm'$.
\end{restatable}

\sloppy
\begin{proof}
By induction on $\tm$.
Case $\tm = \lam\var\tmtwo$ is impossible, as there are no rules for 
reducing abstractions.
\begin{enumerate}
\item $\tm = \var$.
  This case does not apply since we would have
  $\var \tov{\rulesub\var{\val_1}}{\aset \cup \set\var}\sset\app 
  \val_1 = \tm_1$ and
  $\var \tov{\rulesub\var{\val_2}}{\aset \cup \set\var}\sset\app 
  \val_2 = \tm_2$, with $\val_1 = \val_2 = \val$ by the hypothesis
  but which contradicts the hypothesis stating $\tm_1 \neq \tm_2$.
\item $\tm = \tmtwo \, \tmthree$.
  Since $\inv\aset\sset{\tmtwo \, \tmthree}$, then 
  $\inv\aset\sset\tmtwo$ and $\inv\aset\sset\tmthree$ hold.
  Applications reduce by rules \ruleUDb, \ruleUAppL, and \ruleUAppR,
  so we consider the following subcases:
  \begin{enumerate}
  \item \ruleUDb-\ruleUDb.
    This case does not apply since it ends up being that 
    $\tm_1 = \tm_2$, which contradicts the hypothesis.
  \item \ruleUDb-\ruleUAppL.
    Then,
    $\tm = (\lam\var{\tmtwo'})\sctx \, \tmthree 
    \tov\ruledb\aset\sset\appflag 
    \tmtwo'\esub\var\tmthree\sctx = \tm_1$, where
    $\tmtwo = (\lam\var{\tmtwo'})\sctx$ and $\rulename_1 = \ruledb$.
    On the other hand,
    $\tm = (\lam\var{\tmtwo'})\sctx \, \tmthree 
    \tov{\rulename_2}\aset\sset\appflag \tmtwo_2 \, \tmthree = \tm_2$
    is derived from 
    $(\lam\var{\tmtwo'})\sctx \tov{\rulename_2}\aset\sset\app \tmtwo_2$,
    where $\tmtwo_2$ is of the form $(\lam\var{\tmtwo'})\sctx'$
    by \cref{rem:valL_reduces_in_L}.
    Then, $\tm_2 = (\lam\var{\tmtwo'})\sctx' \, \tmthree
    \tov\ruledb\aset\sset\appflag
    \tmtwo'\esub\var\tmthree\sctx' = \tm'$ by rule \ruleUDb, and
    $\tm_1 = \tmtwo'\esub\var\tmthree\sctx 
    \tov{\rulename_2}\aset\sset\appflag 
    \tmtwo'\esub\var\tmthree\sctx' = \tm'$ by 
    \cref{rem:valL_reduces_in_L}.
    And we conclude by applying \cref{lem:weakening_of_reducion_sets}
    on both reduction steps to extend $\aset$ to  
    $\aset \cup \asettwo^1$ and $\aset \cup \asettwo^2$, respectively.
    The following diagram summarises the proof:
    \[
      \xymatrixcolsep{3pc}
      \xymatrix{
          \tm = (\lam\var{\tmtwo'})\sctx \, \tmthree 
            \arUr\ruledb\aset\sset\appflag 
            \arUd{\rulename_2}\aset\sset\appflag
        & \tmtwo'\esub\var\tmthree\sctx = \tm_1
            \arsdUd{\rulename_2}{\aset \cup \asettwo^1}\sset\appflag
      \\
          \tm_2 = (\lam\var{\tmtwo'})\sctx' \, \tmthree
            \arsdUr\ruledb{\aset \cup \asettwo^2}\sset\appflag
        & \tmtwo'\esub\var\tmthree\sctx' = \tm'
      }
    \]
  \item \ruleUDb-\ruleUAppR.
    Then,
    $\tm = (\lam\var{\tmtwo'})\sctx \, \tmthree 
    \tov\ruledb\aset\sset\appflag \tmtwo'\esub\var\tmthree\sctx = \tm_1$,
    where
    $\tmtwo = (\lam\var{\tmtwo'})\sctx$ and $\rulename_1 = \ruledb$.
    On the other hand,
    $\tm = (\lam\var{\tmtwo'})\sctx \, \tmthree 
    \tov{\rulename_2}\aset\sset\appflag 
    (\lam\var{\tmtwo'})\sctx \, \tmthree_2 = \tm_2$ is derived from
    $(\lam\var{\tmtwo'})\sctx \in \Struct\sset$ and
    $\tmthree \tov{\rulename_2}\aset\sset\nonapp \tmthree_2$.
    This case is impossible since
    $(\lam\var{\tmtwo'})\sctx$ cannot be an element of $\Struct\sset$
    by \cref{rem:habs_st}.
  \item \ruleUAppL-\ruleUAppR.
    Then,
    $\tm = \tmtwo \, \tmthree \tov{\rulename_1}\aset\sset\appflag \tmtwo_1 \, \tmthree = \tm_1$
    derived from 
    (1) $\tmtwo \tov{\rulename_1}\aset\sset\app \tmtwo_1$.
    On the other hand,
    $\tm = \tmtwo \, \tmthree \tov{\rulename_2}\aset\sset\appflag \tmtwo \, \tmthree_2 = \tm_2$,
    derived from (2) $\tmtwo \in \Struct\sset$ and
    (3) $\tmthree \tov{\rulename_2}\aset\sset\nonapp \tmthree_2$.
    We can apply rule \ruleUAppL with (1) as premise, yielding
    $\tm_2 = \tmtwo \, \tmthree_2 \tov{\rulename_1}\aset\sset\appflag
    \tmtwo_1 \, \tmthree_2 = \tm'$.
    Having (1) and (2), then $\tmtwo_1 \in \Struct\sset$ by 
    \cref{lem:HAbs_Struct_closed_reduction}.
    With this result and (3) we can apply rule \ruleUAppR, yielding
    $\tm_1 = \tmtwo_1 \, \tmthree \tov{\rulename_2}\aset\sset\appflag
    \tmtwo_1 \, \tmthree_2 = \tm'$.
    And we conclude by applying \cref{lem:weakening_of_reducion_sets}
    on both reductions to extend $\aset$ to $\aset \cup \asettwo^1$ 
    and $\aset \cup \asettwo^2$, respectively.
    The following diagram summarises the proof:
    \[
      \xymatrixcolsep{5pc}
      \xymatrix{
          \tm = \tmtwo \, \tmthree 
            \arUr{\rulename_1}\aset\sset\appflag
            \arUd{\rulename_2}\aset\sset\appflag
        & \tmtwo_1 \, \tmthree = \tm_1
            \arsdUd{\rulename_2}{\aset \cup \asettwo^1}\sset\appflag
      \\
          \tm_2 = \tmtwo \, \tmthree_2
            \arsdUr{\rulename_1}{\aset \cup \asettwo^2}\sset\appflag
        & \tmtwo_1 \, \tmthree_2 = \tm'
      }
    \]
  \item \ruleUAppL-\ruleUAppL.
    Then,
    $\tm = \tmtwo \, \tmthree \tov{\rulename_1}\aset\sset\appflag \tmtwo_1 \, \tmthree = \tm_1$
    derived from $\tmtwo \tov{\rulename_1}\aset\sset\app \tmtwo_1$.
    On the other hand,
    $\tm = \tmtwo \, \tmthree \tov{\rulename_2}\aset\sset\appflag \tmtwo_2 \, \tmthree = \tm_2$,
    derived from 
    $\tmtwo \tov{\rulename_2}\aset\sset\app \tmtwo_2$, 
    where $\tmtwo_1 \neq \tmtwo_2$
    since $\tmtwo_1 \, \tmthree \neq \tmtwo_2 \, \tmthree$ by hypothesis.
    We apply the \ih on $\tmtwo$, yielding $\tmtwo'$ such that
    $\tmtwo_1 \tov{\rulename_2}{\aset \cup \asettwo^1}\sset\app \tmtwo'$
    and
    $\tmtwo_2 \tov{\rulename_1}{\aset \cup \asettwo^2}\sset\app \tmtwo'$.
    Applying rule \ruleUAppL to reduce both $\tmtwo_1 \, \tmthree$ and
    $\tmtwo_2 \, \tmthree$, we obtain
    $\tm_1 = \tmtwo_1 \, \tmthree \tov{\rulename_2}{\aset \cup \asettwo^1}\sset\appflag \tmtwo' \, \tmthree = \tm'$
    and
    $\tm_2 = \tmtwo_2 \, \tmthree \tov{\rulename_1}{\aset \cup \asettwo^2}\sset\appflag \tmtwo' \, \tmthree = \tm'$,
    respectively.
    The following diagram summarises the proof:
    \[
      \xymatrixcolsep{5pc}
      \xymatrix{
          \tm = \tmtwo \, \tmthree 
            \arUr{\rulename_1}\aset\sset\appflag 
            \arUd{\rulename_2}\aset\sset\appflag
        & \tmtwo_1 \, \tmthree = \tm_1
            \arsdUd{\rulename_2}{\aset \cup \asettwo^1}\sset\appflag
      \\
          \tm_2 = \tmtwo_2 \, \tmthree
            \arsdUr{\rulename_1}{\aset \cup \asettwo^2}\sset\appflag
        & \tmtwo' \, \tmthree = \tm'
      }
    \]
  \item \ruleUAppR-\ruleUAppR.
    Analogous to the previous case.
  \end{enumerate}
\item $\tm = \tmtwo\esub\var\tmthree$.
  Since $\inv\aset\sset{\tmtwo\esub\var\tmthree}$, then 
  $\inv{\aset \cup \set\var}\sset\tmtwo$, 
  $\inv\aset{\sset \cup \set\var}\tmtwo$, and $\inv\aset\sset\tmthree$.
  Closures reduce by rules \ruleULsv, \ruleUEsR, \ruleUEsLAbs, and 
  \ruleUEsLStruct, so we consider the following cases:
  \begin{enumerate}
  \item \ruleULsv-\ruleULsv.
    Then
    \[
      \inferrule{
        \tmtwo \tov{\rulesub\var\val}{\aset \cup \set\var}\sset\appflag \tmtwo_1
        \sep
        \var \notin \aset \cup \sset
        \sep
        \val\sctx \in \HAbs\aset
      }{
        \tm = \tmtwo\esub\var{\val\sctx}
        \tov\rulelsv\aset\sset\appflag
        \tmtwo_1\esub\var\val\sctx = \tm_1
      }\ruleULsv
    \]
    where $\tmthree = \val\sctx$ and $\rulename_1 = \rulelsv$.
    On the other hand,
    \[
      \inferrule{
        \tmtwo \tov{\rulesub\var\val}{\aset \cup \set\var}\sset\appflag \tmtwo_2
        \sep
        \var \notin \aset \cup \sset
        \sep
        \val\sctx \in \HAbs\aset
      }{
        \tm = \tmtwo\esub\var{\val\sctx}
        \tov\rulelsv\aset\sset\appflag
        \tmtwo_2\esub\var\val\sctx = \tm_2
      }\ruleULsv
    \]
    where $\rulename_2 = \ruleULsv$ and $\tmtwo_1 \neq \tmtwo_2$,
    since $\tmtwo_1\esub\var{\val\sctx} \neq \tmtwo_2\esub\var{\val\sctx}$.
    We apply the \ih on $\tmtwo$, yielding $\tmtwo'$ such that
    $\tmtwo_1 \tov{\rulesub\var\val}{\aset \cup \set\var \cup \asettwo^1}\sset\appflag \tmtwo'$ and
    $\tmtwo_2 \tov{\rulesub\var\val}{\aset \cup \set\var \cup \asettwo^2}\sset\appflag \tmtwo'$.
    Moreover,
    $\tmtwo_1 \tov{\rulesub\var\val}{\expansion\aset\sctx \cup \set\var \cup \asettwo^1}{\expansion\sset\sctx}\appflag \tmtwo'$ and
    $\tmtwo_2 \tov{\rulesub\var\val}{\expansion\aset\sctx \cup \set\var \cup \asettwo^2}{\expansion\sset\sctx}\appflag \tmtwo'$
    by \cref{lem:weakening_of_reducion_sets}.
    Furthermore, the hypothesis $\val\sctx \in \HAbs\aset$ implies
    $\val \in \HAbs{\expansion\aset\sctx}$ by 
    \cref{lem:tL_hAbs_t_hAbsExp}, so 
    $\val \in \HAbs{\expansion\aset\sctx \cup \asettwo^i}$ by
    \cref{lem:weakening_of_reducion_sets}.
    Since 
    $\var \notin \aset \cup \sset$ by premise in rule \ruleULsv, 
    and $\var \notin \domSctx\sctx \cup \asettwo^1 \cup \asettwo^2$
    by $\alpha$-conversion, 
    then $\var \notin \expansion\aset\sctx \cup \asettwo^1 \cup \expansion\sset\sctx$
    and $\var \notin \expansion\aset\sctx \cup \asettwo^2 \cup \expansion\sset\sctx$.
    We can then apply rule \ruleULsv on both 
    $\tmtwo_1\esub\var\val$ and $\tmtwo_2\esub\var\val$, yielding
    $\tmtwo_1\esub\var\val 
    \tov\rulelsv{\expansion\aset\sctx \cup \asettwo^1}{\expansion\sset\sctx}\appflag 
    \tmtwo'\esub\var\val$ and
    $\tmtwo_2\esub\var\val 
    \tov\rulelsv{\expansion\aset\sctx \cup \asettwo^2}{\expansion\sset\sctx}\appflag 
      \tmtwo'\esub\var\val$, respectively.
    To conclude,
    $\tm_1 = \tmtwo_1\esub\var\val\sctx 
    \tov\rulelsv{\aset \cup \asettwo^1}\sset\appflag 
    \tmtwo'\esub\var\val\sctx = \tm'$ and
    $\tm_2 = \tmtwo_2\esub\var\val\sctx 
    \tov\rulelsv{\aset \cup \asettwo^2}\sset\appflag 
    \tmtwo'\esub\var\val\sctx = \tm'$
    by successively applying rule \ruleUEsLAbs or \ruleUEsLStruct
    accordingly.
    The following diagram summarises the proof:
    \[
      \xymatrixcolsep{5pc}
      \xymatrix{
        \tm = \tmtwo\esub\var{\val\sctx} 
          \arUr\rulelsv\aset\sset\appflag 
          \arUd\rulelsv\aset\sset\appflag
      & \tmtwo_1\esub\var\val\sctx = \tm_1
          \arsdUd\rulelsv{\aset \cup \asettwo^2}\sset\appflag
      \\
        \tm_2 = \tmtwo_2\esub\var\val\sctx
          \arsdUr\rulelsv{\aset \cup \asettwo^1}\sset\appflag
      & \tmtwo'\esub\var\val\sctx = \tm'
      }
    \]
  \item \ruleULsv-\ruleUEsR. 
    Then
    \[
      \inferrule{
        (1)\ \tmtwo \tov{\rulesub\var\val}{\aset \cup \set\var}\sset\appflag \tmtwo_1
        \sep
        (2)\ \var \notin \aset \cup \sset
        \sep
        (3)\ \val\sctx \in \HAbs\aset
      }{
        \tm = \tmtwo\esub\var{\val\sctx}
        \tov\rulelsv\aset\sset\appflag
        \tmtwo_1\esub\var\val\sctx = \tm_1
      }\ruleULsv
    \]
    where $\tmthree = \val\sctx$ and $\rulename_1 = \rulelsv$.
    On the other hand,
    $\tm = \tmtwo\esub\var{\val\sctx} 
    \tov{\rulename_2}\aset\sset\appflag 
    \tmtwo\esub\var{\tmthree_2} = \tm_2$, derived from
    $(4)\ \val\sctx \tov{\rulename_2}\aset\sset\nonapp \tmthree_2$.
    Moreover, $\tmthree_2 = \val\sctx'$ by \cref{rem:valL_reduces_in_L}.
    We can then apply accordingly rule \ruleUEsR, \ruleUEsLAbs, or 
    \ruleUEsLStruct to $(\rulename_2, \aset, \sset, \appflag)$-reduce 
    $\tm_1 = \tmtwo_1\esub\var\val\sctx$ to the term 
    $\tmtwo_1\esub\var\val\sctx' = \tm'$.
    And we extend $\aset$ to $\aset \cup \asettwo^1$ by applying 
    \cref{lem:weakening_of_reducion_sets}.

    On the other hand, we analyse two possible cases,
    depending on the form $\rulename_2$ can have:
    \begin{itemize}
    \item $\rulename_2 \in \set{\ruledb, \rulelsv}$.
      Since (3) and (4), then $\val\sctx' \in \HAbs\aset$ 
      by \cref{lem:HAbs_Struct_closed_reduction}.
      We can apply rule \ruleULsv with this result and (1) and (2) as
      premises, yielding
      $\tm_2 = \tmtwo\esub\var{\val\sctx'} 
      \tov\rulelsv\aset\sset\appflag \tmtwo_1\esub\var\val\sctx' = \tm'$.
      Then, we extend $\aset$ to $\aset \cup \asettwo^2$ by 
      \cref{lem:weakening_of_reducion_sets}.
    \item $\rulename_2 = \rulesub{\var_2}{\val_2}$.
      By $\alpha$-conversion, we may assume $\var_2 \neq \var$, and 
      $\var_2 \in \aset$ by \cref{lem:sub_var_in_aset}.
      Moreover, $\val_2 \in \HAbs{\aset \cup \asettwo^2}$ by hypothesis.
      Since (4) and (3), then $\val\sctx' \in \HAbs{\aset \cup \asettwo^2}$
      by \cref{lem:habs_rulesub_closed_reduction}.
      Moreover, taking (1), we yield
      $\tmtwo 
      \tov{\rulesub\var\val}{\aset \cup \set\var \cup \asettwo^1}\sset\appflag 
      \tmtwo_1$ by \cref{lem:weakening_of_reducion_sets}.
      We apply rule \ruleULsv, concluding
      $\tm_2 = \tmtwo\esub\var{\val\sctx'} 
      \tov\rulelsv{\aset \cup \asettwo^2}\sset\appflag 
      \tmtwo\esub\var\val\sctx' = \tm'$.
    \end{itemize}
    The following diagram summarises the proof:
    \[
      \xymatrixcolsep{5pc}
      \xymatrix{
        \tm = \tmtwo\esub\var{\val\sctx}
          \arUr\rulelsv\aset\sset\appflag 
          \arUd{\rulename_2}\aset\sset\appflag
      & \tmtwo_1\esub\var\val\sctx = \tm_1
          \arsdUd{\rulename_2}{\aset \cup \asettwo^1}\sset\appflag
      \\
        \tm_2 = \tmtwo\esub\var{\val\sctx'}
          \arsdUr\rulelsv{\aset \cup \asettwo^2}\sset\appflag
      & \tmtwo_1\esub\var\val\sctx' = \tm'
      }
    \]
  \item \ruleULsv-\ruleUEsLAbs.
    Then
    \[
      \inferrule{
        \tmtwo \tov{\rulesub\var\val}{\aset \cup \set\var}\sset\appflag \tmtwo_1
        \sep
        \var \notin \aset \cup \sset
        \sep
        \val\sctx \in \HAbs\aset
      }{
        \tm= \tmtwo\esub\var{\val\sctx}
        \tov\rulelsv\aset\sset\appflag
        \tmtwo_1\esub\var\val\sctx =\tm_1
      }\ruleULsv
    \]
    where $\tmthree = \val\sctx$ and $\rulename_1 = \rulelsv$.
    On the other hand,
    \[
      \inferrule{
        \tmtwo \tov{\rulename_2}{\aset \cup \set\var}\sset\appflag \tmtwo_2
        \sep
        \val\sctx \in \HAbs\aset
        \sep
        \var \notin \aset \cup \sset
        \sep
        (1)\ \var \notin \fv{\rulename_2}
      }{
         \tm = \tmtwo\esub\var{\val\sctx}
         \tov{\rulename_2}\aset\sset\appflag
         \tmtwo_2\esub\var{\val\sctx} = \tm_2
      }\ruleUEsLAbs
    \]
    We can apply the \ih on $\tmtwo$, yielding $\tmtwo'$ such that
    $\tmtwo_1 
    \tov{\rulename_2}{\aset \cup \set\var \cup \asettwo^1}\sset\appflag 
    \tmtwo'$ and 
    $(2)\ \tmtwo_2 
    \tov{\rulesub\var\val}{\aset \cup \set\var \cup \asettwo^2}\sset\appflag 
    \tmtwo'$.
    Moreover,
    $(3)\ \tmtwo_1 
    \tov{\rulename_2}{\expansion\aset\sctx \cup \set\var \cup \asettwo^1}{\expansion\sset\sctx}\appflag 
    \tmtwo'$ by \cref{lem:weakening_of_reducion_sets}.
    The hypothesis $\val\sctx \in \HAbs\aset$ implies 
    $\val \in \HAbs{\expansion\aset\sctx}$ by \cref{lem:tL_hAbs_t_hAbsExp}
    and $\val \in \HAbs{\expansion\aset\sctx \cup \asettwo^i}$ by
    \cref{lem:weakening_of_reducion_sets}, with $i \in \set{1, 2}$.
    Hence, we have 
    $(4)\ \val \in \HAbs{\expansion\aset\sctx \cup \asettwo^1}$ and 
    $(5)\ \val\sctx \in \HAbs{\aset \cup \asettwo^2}$ by 
    \cref{lem:tL_hAbs_t_hAbsExp}.
    Since $\var \notin \aset \cup \sset$ by premise in rule \ruleULsv
    and $\var \notin \asettwo^1 \cup  \asettwo^2 \cup \domSctx\sctx$
    by $\alpha$-conversion, then 
    $(6)\ \var \notin (\aset \cup \asettwo^2) \cup \sset$ and 
    $(7)\ \var \notin (\expansion\aset\sctx \cup \asettwo^1) \cup \expansion\sset\sctx$.
    Therefore,
    $\tmtwo_1\esub\var\val 
    \tov{\rulename_2}{\expansion\aset\sctx \cup \asettwo^1}{\expansion\sset\sctx}\appflag 
    \tmtwo'\esub\var\val$ by rule \ruleUEsLAbs,
    with (3), (4), (7), and (1) as premises.
    And we conclude with
    $\tm_1 = \tmtwo_1\esub\var\val\sctx 
    \tov{\rulename_2}{\aset \cup \asettwo^1}\sset\appflag 
    \tmtwo'\esub\var\val\sctx = \tm'$ by successively applying rule 
    \ruleUEsLAbs or rule \ruleUEsLStruct accordingly.
    On the other hand, by rule \ruleULsv with (2), (6), and (5) as 
    premises, we obtain
    $\tm_2 = \tmtwo_2\esub\var{\val\sctx} 
    \tov\rulelsv{\aset \cup \asettwo^2}\sset\appflag 
    \tmtwo'\esub\var\val\sctx = \tm'$.
    The following diagram summarises the proof:
    \[
      \xymatrixcolsep{5pc}
      \xymatrix{
        \tm = \tmtwo\esub\var{\val\sctx} 
          \arUr\rulelsv\aset\sset\appflag 
          \arUd{\rulename_2}\aset\sset\appflag
      & \tmtwo_1\esub\var\val\sctx = \tm_1
          \arsdUd{\rulename_2}{\aset \cup \asettwo^1}\sset\appflag
      \\
        \tm_2 = \tmtwo_2\esub\var{\val\sctx}
          \arsdUr\rulelsv{\aset \cup \asettwo^2}\sset\appflag
      & \tmtwo'\esub\var\val\sctx = \tm'
      }
    \]
  \item \ruleULsv-\ruleUEsLStruct.
    Then
    \[
      \inferrule{
        \tmtwo \tov{\rulesub\var\val}{\aset \cup \set\var}\sset\appflag \tmtwo_1
        \sep
        \var \notin \aset \cup \sset
        \sep
        \val\sctx \in \HAbs\aset
      }{
        \tm = \tmtwo\esub\var{\val\sctx}
        \tov\rulelsv\aset\sset\appflag
        \tmtwo_1\esub\var{\val}\sctx = \tm_1
      }\ruleULsv
    \]
    where $\tmthree = \val\sctx$ and $\rulename_1 = \rulelsv$.

    On the other hand,
    $\tm = \tmtwo\esub\var{\val\sctx} 
    \tov{\rulename_2}\aset\sset\appflag \tmtwo_2\esub\var{\val\sctx} = \tm_2$,
    derived from 
    $\tmtwo \tov{\rulename_2}\aset{\sset \cup \set\var}\appflag \tmtwo_2$,
    $\val\sctx \in \Struct\sset$, $\var \notin \aset \cup \sset$, and 
    $\var \notin \fv{\rulename_2}$.
    Since $\inv\aset\sset{\val\sctx}$, then in particular 
    $\aset \disj \sset$.
    Hence,
    $\val\sctx \in \HAbs\aset$ and $\val\sctx \in \Struct\sset$,
    and at the same time
    $\val\sctx \notin \HAbs\aset$ or $\val\sctx \notin \Struct\sset$
    by \cref{lem:disjunction}.
    Therefore we reach a contradiction, so this case is not possible.
  \item \ruleUEsR-\ruleUEsR.
    Then
    $\tm = \tmtwo\esub\var\tmthree \tov{\rulename_1}\aset\sset\appflag
     \tmtwo\esub\var{\tmthree_1} = \tm_1$,
    derived from $\tmthree \tov{\rulename_1}\aset\sset\nonapp \tmthree_1$.

    On the other hand,
    $\tm = \tmtwo\esub\var\tmthree \tov{\rulename_2}\aset\sset\appflag
    \tmtwo\esub\var{\tmthree_2} = \tm_2$, derived from
    $\tmthree \tov{\rulename_2}\aset\sset\nonapp \tmthree_2$,
    where $\tmthree_1 \neq \tmthree_2$
    since $\tmtwo\esub\var{\tmthree_1} \neq \tmtwo\esub\var{\tmthree_2}$
    by hypothesis.
    We can apply the \ih on $\tmthree$, yielding $\tmthree'$ such that
    $\tmthree_1 \tov{\rulename_2}{\aset \cup \asettwo^1}\sset\nonapp \tmthree'$
    and 
    $\tmthree_2 \tov{\rulename_1}{\aset \cup \asettwo^2}\sset\nonapp \tmthree'$.
    Applying rule \ruleUEsR on both $\tmtwo\esub\var{\tmthree_1}$ and
    $\tmtwo\esub\var{\tmthree_2}$, 
    we obtain $\tm_1 = \tmtwo\esub\var{\tmthree_1} 
    \tov{\rulename_2}{\aset \cup \asettwo^1}\sset\appflag 
    \tmtwo\esub\var{\tmthree'} = \tm'$ 
    and $\tm_2 = \tmtwo\esub\var{\tmthree_2} 
    \tov{\rulename_1}{\aset \cup \asettwo^2}\sset\appflag 
    \tmtwo\esub\var{\tmthree'} = \tm'$, respectively.
    The following diagram summarises the proof:
    \[
      \xymatrixcolsep{5pc}
      \xymatrix{
        \tm = \tmtwo\esub{\var}{\tmthree}
          \arUr{\rulename_1}{\aset}{\sset}{\appflag} 
          \arUd{\rulename_2}{\aset}{\sset}{\appflag}
      & \tmtwo\esub{\var}{\tmthree_1} = \tm_1
          \arsdUd{\rulename_2}{\aset \cup \asettwo^1}{\sset}{\appflag}
      \\
        \tm_2 = \tmtwo\esub{\var}{\tmthree_2}
          \arsdUr{\rulename_1}{\aset \cup \asettwo^2}{\sset}{\appflag}
      & \tmtwo\esub{\var}{\tmthree'} = \tm'
      }
    \]
  \item \ruleUEsLAbs-\ruleUEsLAbs. % Analogo al caso EsR-EsR.
    Analogous to the previous case.
  \item \ruleUEsLStruct-\ruleUEsLStruct.
    Analogous to the previous case.
  \item \ruleUEsR-\ruleUEsLAbs.
    Then, 
    $\tm = \tmtwo\esub\var\tmthree 
    \tov{\rulename_1}\aset\sset\appflag 
    \tmtwo\esub\var{\tmthree_1} = \tm_1$ which is derived from
    $(1)\ \tmthree \tov{\rulename_1}\aset\sset\nonapp \tmthree_1$.

    On the other hand, we have
    $\tm = \tmtwo\esub\var\tmthree \tov{\rulename_2}\aset\sset\appflag
    \tmtwo_2\esub\var\tmthree = \tm_2$, derived from
    $(2)\ \tmtwo \tov{\rulename_2}{\aset \cup \set\var}\sset\appflag \tmtwo_2$,
    $(3)\ \tmthree \in \HAbs\aset$, 
    $(4)\ \var \notin \aset \cup \sset$, and 
    $(5)\ \var \notin \fv{\rulename_2}$.
    We can apply rule \ruleUEsLAbs with (2), (3), (4), and (5) as 
    premises, yielding
    $\tm_1 = \tmtwo\esub\var{\tmthree_1} 
    \tov{\rulename_2}\aset\sset\appflag 
    \tmtwo_2\esub\var{\tmthree_1} = \tm'$.
    Then, we can apply rule \ruleUEsR with (1) as premise, yielding
    $\tm_2 = \tmtwo_2\esub\var\tmthree 
    \tov{\rulename_1}\aset\sset\appflag 
    \tmtwo_2\esub\var{\tmthree_1} = \tm'$.
    We conclude by applying \cref{lem:weakening_of_reducion_sets} on 
    both reductions to extend $\aset$ to $\aset \cup \asettwo^1$ and 
    $\aset \cup \asettwo^2$, respectively.
    The following diagram summarises the proof:
    \[
      \xymatrixcolsep{5pc}
      \xymatrix{
        \tm = \tmtwo\esub\var\tmthree
          \arUr{\rulename_1}\aset\sset\appflag 
          \arUd{\rulename_2}\aset\sset\appflag
      & \tmtwo\esub\var{\tmthree_1} = \tm_1
          \arsdUd{\rulename_2}{\aset \cup \asettwo^1}\sset\appflag
      \\
        \tm_2 = \tmtwo_2\esub\var\tmthree
          \arsdUr{\rulename_1}{\aset \cup \asettwo^2}\sset\appflag
      & \tmtwo_2\esub\var{\tmthree_1} = \tm'
      }
    \]
  \item \ruleUEsR-\ruleUEsLStruct.
    Analogous to the previous case.
  \item \ruleUEsLAbs-\ruleUEsLStruct.
    Then,
    $\tm = \tmtwo\esub\var\tmthree
    \tov{\rulename_1}\aset\sset\appflag
    \tmtwo_1\esub\var\tmthree = \tm_1$, which is derived from
    $\tmtwo \tov{\rulename_1}{\aset \cup \set\var}\sset\appflag \tmtwo_1$,
    $\tmthree \in \HAbs\aset$, $\var \notin \aset \cup \sset$, and 
    $\var \notin \fv{\rulename_1}$.

    On the other hand, 
    $\tm = \tmtwo\esub\var\tmthree \tov{\rulename_2}\aset\sset\appflag 
    \tmtwo_2\esub\var\tmthree = \tm_2$, which is derived from
    $\tmtwo \tov{\rulename_2}\aset{\sset \cup \set\var}\appflag \tmtwo_2$,
    $\tmthree \in \Struct\sset$, $\var \notin \aset \cup \sset$, and 
    $\var \notin \fv{\rulename_2}$.
    Since $\inv\aset\sset\tmthree$ then in particular $\aset \disj \sset$
    so $\tmthree \notin \HAbs\aset$ or $\tmthree \notin \Struct\sset$
    by \cref{lem:disjunction}.
    However, at the same time, we yield $\tmthree \in \HAbs\aset$ and 
    $\tmthree \in \Struct\sset$ by premises of rules \ruleUEsLAbs and
    \ruleUEsLStruct, respectively.
    Therefore we reach a contradiction, so this case is not possible.
  \end{enumerate}
\end{enumerate}
\end{proof}

The following theorem is a particular case of \cref{prop:diamond-property}:
\diamondproperty*

\section{Proofs of Section~\ref{sec:relating} ``Relating \Nonuseful and Useful Call-by-Value''}
\label{app:relating}

This section presents the technical details regarding the relation 
between both \LOCBV and \UOCBV.
We start first by discussing in \cref{sec:rewriting-tovalas} the
partial unfolding operation we introduced in~\cref{sec:relating}:
we give in particular a characterisation of the partial unfolding 
operation via rewriting. 
Our main goal here is to show that the relation $\tovalas$ is 
terminating (\cref{cor:tovalas-terminating}). 
For achieving this result, we prove that there is a decreasing measure, 
which we define in this same subsection.
Then, we move on to \cref{sec:relating-locbv-uocbv}, where we show 
the technical results that are necessary to relate \LOCBV and \UOCBV 
(\cref{coro:t_UNF_unfolding_t_VNF}).

\subsection{Characterising the Unfolding Operation via Rewriting}
\label{sec:rewriting-tovalas}

The reduction relation under a value assignment $\valas$, written 
$\tovalas$, is defined as follows:
\[
  \tovalas\, \eqdef\, \tovv\rulelsv \cup_{\var\in\dom\valas} (\tovv{\rulesub\var{\valas(\var)}})
\]

A term $\tm$ is said to be \defn{$\valas$-reducible} if there exists 
$\tmtwo$ such that $\tm \tovalas \tmtwo$. 
Note that $\dom\valas$ plays the role of a value frame, so this notion 
is closely related to the set $\VRed\vset$ (defined in~\cref{app:opencbv}).

We start by proving that $\tovalas$ is \emph{terminating} (\cf \cref{sec:preliminary_notions}).
To prove this, we define a \emph{measure} $\measvalas{\_}$ on terms 
and we show that it is strictly decreasing w.r.t. $\tovalas$ 
(\cref{lem:decreasing-measure}).
The measure $\measvalas{\_}$ is inspired by de Vrijer's direct proof 
of the finite developments theorem~\cite{Vrijer85}, and defined by 
means of intermediate functions $\measvar\var{\_}$ and $\meas{\_}$.

Given a term $\tm$ and a variable $\var$, the 
\defn{potential number of occurences of $\var$ in $\tm$}, written 
$\measvar\var\tm$, is defined as 0 if $\var \notin \fv\tm$, and 
otherwise is defined recursively as follows:
\[
  \begin{array}{rcl@{\hspace{1cm}}rcl}
    \measvar\var\var
  & \eqdef
  & 1
  & 
    \measvar\var{\tm \, \tmtwo}
  & \eqdef
  & \measvar\var\tm + \measvar\var\tmtwo
  \\
    \measvar\var{\lam\vartwo\tm}
  & \eqdef
  & 0
  & \measvar\var{\tm\esub\vartwo\tmtwo}
  & \eqdef
  & \measvar\var\tm + \measvar\var\tmtwo\cdot(1 + \measvar\vartwo\tm)
  \end{array}
\]

This gives an \emph{overapproximation} of the number of free reachable
occurrences of $\var$ in the unfolding of $\tm$.
By this we mean that $\measvar{\var}{\tm}$ counts the \emph{potential} number
of free occurrences of $\var$ in the unfolding of $\tm$,
even if not all the substitutions are 
performed during useful reduction for some reason.
For example, $\measvar{\vartwo}{\var\esub{\var}{\vartwo \, \vartwo}} =
4$, although the substitution of $\var$ by $\vartwo \, \vartwo$
is never executed because it is not useful.
Another example is $\tm_0 \eqdef \var\esub{\var}{\vartwo\esub{\vartwo}{\varthree}}$.
We have $\measvar{\varthree}{\tm_0} = 4$, despite the fact that $\varthree$
only occurs three times in the corresponding unfolding
$\varthree\esub{\var}{\varthree}\esub{\vartwo}{\varthree}$.

We now define the \textbf{measure of $\tm$}, written $\meas{\tm}$, as follows:
\[
  \begin{array}{rcl@{\hspace{1cm}}rcl}
      \meas{\var}                   & \eqdef & 0 
    & \meas{\tm \, \tmtwo}          & \eqdef & \meas{\tm} + \meas{\tmtwo} 
  \\
      \meas{\lam{\var}{\tm}}        & \eqdef & 0
    & \meas{\tm\esub{\var}{\tmtwo}} & \eqdef & \meas{\tm} + \measvar{\var}{\tm} + \meas{\tmtwo}\cdot(1 + \measvar{\var}{\tm})
  \end{array}
\]
This measure can be seen as the number of
steps that the evaluation takes to perform all the
substitutions in the longest reduction sequence starting at $\tm$.
Notice that $ \meas{\val} = 0$ for any value $\val$.

Furthermore, given $\valas$ a value assignment,
we define the \textbf{measure of $\tm$ under $\valas$}, written $\measvalas{\tm}$, as
$\measvalas{\tm} \eqdef \meas{\tm} + \textstyle\sum_{\var \in \dom{\valas}} \measvar{\var}{\tm}$.

For example, take again the term $\tm_0$ from above.
Since $\tovalas$-reduction is non-deterministic,
we consider two different reduction sequences from $\tm_0$,
of lengths $2$ and $3$ respectively,
where the second sequence is the longest possible.
In this example, $\valas$ is the empty value assignment $(\valas = \evalas)$:
\[
  \begin{array}{l@{\,\,}l@{\,\,}l@{\,\,}l@{\,\,}l@{\,\,}l@{\,\,}l}
    \tm_0 = \var\esub{\var}{\vartwo\esub{\vartwo}{\varthree}}
  & \to &
    \tmtwo = \var\esub{\var}{\varthree\esub{\vartwo}{\varthree}}
  & \to &
    \tm_3 = \varthree\esub{\var}{\varthree}\esub{\vartwo}{\varthree}
  \\
    \tm_0 = \var\esub{\var}{\vartwo\esub{\vartwo}{\varthree}}
  & \to &
    \tm_1 = \vartwo\esub{\var}{\vartwo}\esub{\vartwo}{\varthree}
  & \to &
    \tm_2 = \vartwo\esub{\var}{\varthree}\esub{\vartwo}{\varthree} \\
  &&& \to &
    \tm_3 = \varthree\esub{\var}{\varthree}\esub{\vartwo}{\varthree}
  \end{array}
\] 
Then
$\measvar{\varthree}{\tm_0} = 4$, $\measvar{\varthree}{\tm_1} = 3$,
$\measvar{\varthree}{\tm_2}= 3$, $\measvar{\varthree}{\tm_3} = 3$,
and $\measvar{\varthree}{\tmtwo} = 4$.
Note for example that $\measvar{\varthree}{\tm_0} = 4$
is greater than the actual number of free reachable occurrences
of $\varthree$ in its unfolding (which is $3$).

Note also that $\meas{\tm_0} = 3$,
$\meas{\tm_1} = 2$,
$\meas{\tm_2} = 1$, $\meas{\tm_3} = 0$,
and $\meas{\tmtwo} = 1$.
These are upper bounds for the number of substitution steps required to
completely unfold these terms.
For instance, $\meas{\tm_0} = 3$, even though the first reduction sequence
reaches the unfolding in only two steps.
The important point is that the measure strictly decreases at each step
of any reduction sequence, as we prove at the end of this subsection. 

Given a value frame $\vset$,
the set of \textbf{substitution contexts in normal form under $\vset$}
is written $\CtxNF{\vset}$ and is defined inductively as follows:
\[
\indrule{\ruleCtxNFEmpty}{
  \emptyPremise
}{
  \ctxhole \in \CtxNF{\vset}
}
\indrule{\ruleCtxNFAddVal}{
  \sctx \in \CtxNF{\vset \cup \set{\var}}
  \sep
  \tm \in \VNF{\vset}{\nonapp}
  \sep
  \valPred{\tm}
}{
  \sctx\esub{\var}{\tm} \in \CtxNF{\vset}
}
\]
\[
\indrule{\ruleCtxNFAddNonVal}{
  \sctx \in \CtxNF{\vset}
  \sep
  \tm \in \VNF{\vset}{\nonapp}
  \sep
  \neg(\valPred{\tm})
}{
  \sctx\esub{\var}{\tm} \in \CtxNF{\vset}
}
\]

\begin{definition}[Expansion of value frames]
Let $\vset$ be a value frame.
We inductively define the {\em expansion of $\vset$ under $\sctx$},
written $\expansion{\vset}{\sctx}$, as follows:
\[
\begin{array}{rcl}
  \expansion{\vset}{\ctxhole}
& \eqdef
& \vset
\\
  \expansion{\vset}{\sctx'\esub{\var}{\tm}}
& \eqdef
& \left\{ 
    \begin{array}{ll}
      \expansion{\vset}{\sctx'} \cup \set{\var} & \text{if } \valPred{\tm}
    \\
      \expansion{\vset}{\sctx'}                 & \text{otherwise}
    \end{array}
\right.
\end{array}
\]
\end{definition}

\begin{lemma}
\label{lem:splitting_ctx_nf}
The following are equivalent:
\begin{enumerate}
\item
  $\tm\sctx \in \VNF{\vset}{\appflag}$
\item
  $\tm \in \VNF{\expansion{\vset}{\sctx}}{\appflag}$ 
  and $\sctx \in \CtxNF{\vset}$.
\end{enumerate}
\end{lemma}
\hiddenproof{
  Simultaneously by induction on $\sctx$.
}{./proofs/splitting_ctx_nf}

\begin{definition}
Let $\varphi : \mathsf{Var} \to \mathbb{N}$.
Given a substitution context $\sctx$, the
\textbf{potential number of occurrences of $\var$ in $\sctx$ under $\varphi$},
written $\measvarPhi{\var}{\sctx}{\varphi}$, is recursively defined as follows:
\[
  \measvarPhi{\var}{\ctxhole}{\varphi}
    \eqdef 
    \varphi(\var) 
  \sep \sep \sep 
  \measvarPhi{\var}{\sctx'\esub{\vartwo}{\tm}}{\varphi}
    \eqdef 
    \measvarPhi{\var}{\sctx'}{\varphi}
    + \measvar{\var}{\tm}\cdot(1 + \measvarPhi{\vartwo}{\sctx'}{\varphi})
\]

We define the \textbf{measure of $\sctx$ under $\varphi$},
written $\measPhi{\sctx}{\varphi}$, as follows:
\[
    \measPhi{\ctxhole}{\varphi}
 \eqdef 
 0
\sep\sep\sep
    \measPhi{\sctx'\esub{\var}{\tm}}{\varphi}
   \eqdef 
     \measPhi{\sctx'}{\varphi}
    + \measvarPhi{\var}{\sctx'}{\varphi}
    + \meas{\tm}\cdot(1 + \measvarPhi{\var}{\sctx'}{\varphi})
\]
\end{definition}

For any term $\tm$, we define $\varphi_\tm(\var) \eqdef \measvar\var\tm$.
\begin{lemma}[Splitting of measures]
\label{lem:splitting_measVar_meas_tL}
Let $\tm$ be a term and $\sctx$ be a substitution context.
Then:
\begin{enumerate}
\item
  $\measvar\var{\tm\sctx} = \measvarPhi\var\sctx{\varphi_t}$
\item
  $\meas{\tm\sctx} = \meas\tm +  \measPhi\sctx{\varphi_t}$
\end{enumerate}
\end{lemma}
\hiddenproof{
  Simultaneously by induction on $\sctx$.
}{
  % Label: lem:splitting_measVar_meas_tL

\begin{proof}
We prove each item simultaneously, by induction on $\sctx$.
\begin{enumerate}
\item \quad
  \begin{enumerate}
  \item $\sctx = \ctxhole$.
    Then
    $
      \measvar{\var}{\tm}
    = \varphi_\tm(\var)
    = \measvarPhi{\var}{\ctxhole}{\varphi_\tm}
    $.
  \item $\sctx = \sctx'\esub{\vartwo}{\tmtwo}$.
    Then
    \[
    \begin{array}{rcll}
        \measvar{\var}{\tm\sctx'\esub{\vartwo}{\tmtwo}}
      & =
      & \measvar{\var}{\tm\sctx'} + \measvar{\var}{\tmtwo}\cdot(1 + \measvar{\vartwo}{\tm\sctx'})
    \\
      & =
      & \measvarPhi{\var}{\sctx'}{\varphi_\tm} + \measvar{\var}{\tmtwo}\cdot(1 + \measvar{\vartwo}{\tm\sctx'})
      & \text{(By \ih (1) on $\sctx'$)}
    \\
      & =
      & \measvarPhi{\var}{\sctx'}{\varphi_\tm} + \measvar{\var}{\tmtwo}\cdot(1 + \measvarPhi{\vartwo}{\sctx'}{\varphi_\tm})
      & \text{(By \ih (1) on $\sctx'$)}
    \\
      & =
      & \measvarPhi{\var}{\sctx'\esub{\vartwo}{\tmtwo}}{\varphi_\tm}
    \end{array}
    \] 
  \end{enumerate}
\item \quad
  \begin{enumerate}
  \item $\sctx = \ctxhole$.
    Then
    $
      \meas{\tm}
    = \meas{\tm} + 0
    = \meas{\tm} +  \measPhi{\ctxhole}{\varphi_\tm}
    $.
  \item $\sctx = \sctx'\esub{\var}{\tmtwo}$.
    Then
    \[
    \begin{array}{rcll}
        \meas{\tm\sctx'\esub{\var}{\tmtwo}}
      & =
      & \meas{\tm\sctx'} + \measvar{\var}{\tm\sctx'} + \meas{\tmtwo}\cdot(1 + \measvar{\var}{\tm\sctx'})
    \\
      & =
      & \meas{\tm} + \measPhi{\sctx'}{\varphi_\tm} 
        + \measvar{\var}{\tm\sctx'} + \meas{\tmtwo}\cdot(1 + \measvar{\var}{\tm\sctx'})
      & \text{(By \ih (2) on $\sctx'$)}
    \\
      & =
      & \meas{\tm} + \measPhi{\sctx'}{\varphi_\tm} 
        + \measvarPhi{\var}{\sctx'}{\varphi_{\tm}} + \meas{\tmtwo}\cdot(1 + \measvar{\var}{\tm\sctx'})
      & \text{(By \ih (1) on $\sctx'$)}
    \\
      & =
      & \meas{\tm} + \measPhi{\sctx'}{\varphi_\tm} 
        + \measvarPhi{\var}{\sctx'}{\varphi_{\tm}} + \meas{\tmtwo}\cdot(1 + \measvarPhi{\var}{\sctx'}{\varphi_{\tm}})
      & \text{(By \ih (1) on $\sctx'$)}
    \\
      & =
      & \meas{\tm} + \measPhi{\sctx'\esub{\vartwo}{\tmtwo}}{\varphi_\tm}
    \end{array}
    \]
  \end{enumerate}
\end{enumerate}
\end{proof}

}

\begin{lemma}
\label{lem:t_disj_L_measLtL_comp}
Let $\tm$ and $\tmtwo$ be terms and $\sctx$ be a substitution context.
If $\fv\tm \disj \domSctx\sctx$, then:
\begin{enumerate}
\item
  $\measvarPhi\var\sctx{\varphi_{\tm\esub\vartwo\tmtwo}} 
  \leq \measvar\var\tm + \measvarPhi\var\sctx{\varphi_\tmtwo}\cdot(1 + \measvar\vartwo\tm)$ 
  for any variable $\var$
\item
  $\measPhi\sctx{\varphi_{\tm\esub\var\tmtwo}} 
  \leq \measPhi\sctx{\varphi_\tmtwo}\cdot(1 + \measvar\var\tm)$
\end{enumerate}
\end{lemma}
\hiddenproof{
  Simultaneously by induction on $\sctx$.
}{
  % Label: lem:t_disj_L_measLtL_comp

\begin{proof}
We prove both items simultaneously by induction on $\sctx$.
\begin{enumerate}
\item \quad
  \begin{itemize}
  \item $\sctx = \ctxhole$.
    Then
    \begin{align*}
        \measvarPhi\var\ctxhole{\varphi_{\tm\esub\vartwo\tmtwo}}
      = \varphi_{\tm\esub\vartwo\tmtwo}(\var)
      = \measvar\var{\tm\esub\vartwo\tmtwo}
      = \measvar\var\tm + \measvar\var\tmtwo\cdot(1 + \measvar\vartwo\tm)
    \\
      = \measvar\var\tm + \varphi_\tmtwo(\var)\cdot(1 + \measvar\vartwo\tm)
      = \measvar\var\tm + \measvarPhi\var\ctxhole{\varphi_\tmtwo}\cdot(1 + \measvar\vartwo\tm)
    \end{align*}
  \item $\sctx = \sctxtwo\esub{\var_1}{\tm_1}$.
    Then
    \[
      \begin{array}{cll}
        & \measvarPhi\var{\sctxtwo\esub{\var_1}{\tm_1}}{\varphi_{\tm\esub\vartwo\tmtwo}} 
      \\
        = 
        & \measvarPhi\var\sctxtwo{\varphi_{\tm\esub\vartwo\tmtwo}}
          + \measvar\var{\tm_1}\cdot(1 + \measvarPhi{\var_1}\sctxtwo{\varphi_{\tm\esub\vartwo\tmtwo}})
      \\
        \leq
        & \measvar\var\tm + \measvarPhi\var\sctxtwo{\varphi_\tmtwo}\cdot(1 + \measvar\vartwo\tm)
          + \measvar\var{\tm_1}\cdot(1 + \measvarPhi{\var_1}\sctxtwo{\varphi_{\tm\esub\vartwo\tmtwo}})
        & \text{(By \ih (1) on $\sctxtwo$)}
      \\
        \leq
        & \measvar\var\tm + \measvarPhi\var\sctxtwo{\varphi_\tmtwo}\cdot(1 + \measvar\vartwo\tm)
        \\
        & + \measvar\var{\tm_1}\cdot(1 + \measvar{\var_1}\tm + \measvarPhi{\var_1}\sctxtwo{\varphi_\tmtwo}\cdot(1 + \measvar\vartwo\tm))
        & \text{(By \ih (1) on $\sctxtwo$)}
      \\
        = 
        & \measvar\var\tm
          + \measvarPhi\var\sctxtwo{\varphi_\tmtwo}\cdot(1 + \measvar\vartwo\tm)
        \\
        & + \measvar\var{\tm_1}\cdot(1 + 0 + \measvarPhi{\var_1}\sctxtwo{\varphi_\tmtwo}\cdot(1 + \measvar\vartwo\tm))
        & \text{(By hypothesis)}
      \\
        = 
        & \measvar\var\tm 
          + \measvarPhi\var\sctxtwo{\varphi_\tmtwo}\cdot(1 + \measvar\vartwo\tm)
        \\
        & + \measvar\var{\tm_1}\cdot(1 + \measvarPhi{\var_1}\sctxtwo{\varphi_\tmtwo}\cdot(1 + \measvar\vartwo\tm))
      \\
        = 
        & \measvar\var\tm 
          + \measvarPhi\var\sctxtwo{\varphi_\tmtwo}\cdot(1 + \measvar\vartwo\tm)
          + \measvar\var{\tm_1}
        \\
        & + \measvar\var{\tm_1}\cdot\measvarPhi{\var_1}\sctxtwo{\varphi_\tmtwo}\cdot(1 + \measvar\vartwo\tm)
      \\
        \leq
        & \measvar\var\tm 
          + \measvarPhi\var\sctxtwo{\varphi_\tmtwo}\cdot(1 + \measvar\vartwo\tm)
          + \measvar\var{\tm_1}
        \\
        & + \measvar\var{\tm_1}\cdot\measvarPhi{\var_1}\sctxtwo{\varphi_\tmtwo}\cdot(1 + \measvar\vartwo\tm)
          + \measvar\var{\tm_1}\cdot\measvar\vartwo\tm
      \\
        =
        & \measvar\var\tm 
          + \measvarPhi\var\sctxtwo{\varphi_\tmtwo}\cdot(1 + \measvar\vartwo\tm)
        \\
        & + \measvar\var{\tm_1}\cdot(1 + \measvar\vartwo\tm + \measvarPhi{\var_1}\sctxtwo{\varphi_\tmtwo}\cdot(1 + \measvar\vartwo\tm))
      \\
        =
        & \measvar\var\tm 
          + \measvarPhi\var\sctxtwo{\varphi_\tmtwo}\cdot(1 + \measvar\vartwo\tm)
        \\
        & + \measvar\var{\tm_1}\cdot(1 + \measvar\vartwo\tm + \measvarPhi{\var_1}\sctxtwo{\varphi_\tmtwo} + \measvarPhi{\var_1}\sctxtwo{\varphi_\tmtwo}\cdot\measvar\vartwo\tm)
      \\
        =
        & \measvar\var\tm 
          + \measvarPhi\var\sctxtwo{\varphi_\tmtwo}\cdot(1 + \measvar\vartwo\tm)
          + \measvar\var{\tm_1}\cdot(1 + \measvarPhi{\var_1}\sctxtwo{\varphi_\tmtwo})\cdot(1 + \measvar\vartwo\tm)
      \\
        =
        & \measvar\var\tm 
          + (\measvarPhi\var\sctxtwo{\varphi_\tmtwo} + \measvar\var{\tm_1}\cdot(1 + \measvarPhi{\var_1}\sctxtwo{\varphi_\tmtwo}))\cdot(1 + \measvar\vartwo\tm)
      \\
        =
        & \measvar\var\tm + \measvarPhi\var{\sctxtwo\esub{\var_1}{\tm_1}}{\varphi_\tmtwo}\cdot(1 + \measvar\vartwo\tm)
      \end{array}
    \]
  \end{itemize}
\item \quad
  \begin{itemize}
  \item $\sctx = \ctxhole$.
    Then
    $
        \measPhi\ctxhole{\varphi_{\tm\esub\var\tmtwo}} 
      = 0
      = 0\cdot(1 + \measvar\var\tm)
      = \measPhi\ctxhole{\varphi_\tmtwo}\cdot(1 + \measvar\var\tm)
    $.
  \item $\sctx = \sctxtwo\esub{\var_1}{\tm_1}$.
    Then
    \[
    \begin{array}{cll}
      &
        \measPhi{\sctxtwo\esub{\var_1}{\tm_1}}{\varphi_{\tm\esub\var\tmtwo}} 
    \\
      = 
      & \measPhi\sctxtwo{\varphi_{\tm\esub\var\tmtwo}}
        + \measvarPhi{\var_1}\sctxtwo{\varphi_{\tm\esub\var\tmtwo}}
        + \meas{\tm_1}\cdot(1 + \measvarPhi{\var_1}\sctxtwo{\varphi_{\tm\esub\var\tmtwo}})
    \\
      \leq
      & \measPhi\sctxtwo{\varphi_\tmtwo}\cdot(1 + \measvar\var\tm)
        + \measvarPhi{\var_1}\sctxtwo{\varphi_{\tm\esub\var\tmtwo}}
        + \meas{\tm_1}\cdot(1 + \measvarPhi{\var_1}\sctxtwo{\varphi_{\tm\esub\var\tmtwo}})
      & \text{(By \ih (2) on $\sctxtwo$)}
    \\
      \leq
      & \measPhi\sctxtwo{\varphi_\tmtwo}\cdot(1 + \measvar\var\tm)
        + \measvar{\var_1}\tm 
        + \measvarPhi{\var_1}\sctxtwo{\varphi_\tmtwo}\cdot(1 + \measvar\var\tm)
      & \text{(By \ih (1) on $\sctxtwo$)}
      \\
      & + \meas{\tm_1}\cdot(1 + \measvarPhi{\var_1}\sctxtwo{\varphi_{\tm\esub\var\tmtwo}})
    \\
      \leq
      & \measPhi\sctxtwo{\varphi_\tmtwo}\cdot(1 + \measvar\var\tm)
        + \measvar{\var_1}\tm 
        + \measvarPhi{\var_1}\sctxtwo{\varphi_\tmtwo}\cdot(1 + \measvar\var\tm)
      & \text{(By \ih (1) on $\sctxtwo$)}
      \\
      & + \meas{\tm_1}\cdot(1 + \measvar{\var_1}\tm + \measvarPhi{\var_1}\sctxtwo{\varphi_\tmtwo}\cdot(1 + \measvar\var\tm))
    \\
      =
      & \measPhi\sctxtwo{\varphi_\tmtwo}\cdot(1 + \measvar\var\tm) + 0 
        + \measvarPhi{\var_1}\sctxtwo{\varphi_\tmtwo}\cdot(1 + \measvar\var\tm)
      & \text{(By hypothesis)}
      \\
      & + \meas{\tm_1}\cdot(1 + 0 + \measvarPhi{\var_1}\sctxtwo{\varphi_\tmtwo}\cdot(1 + \measvar\var\tm))
    \\
      =
      & \measPhi\sctxtwo{\varphi_\tmtwo}\cdot(1 + \measvar\var\tm)
        + \measvarPhi{\var_1}\sctxtwo{\varphi_\tmtwo}\cdot(1 + \measvar\var\tm)
      \\
      & + \meas{\tm_1}\cdot(1 + \measvarPhi{\var_1}\sctxtwo{\varphi_\tmtwo}\cdot(1 + \measvar\var\tm))
    \\
      =
      & \measPhi\sctxtwo{\varphi_\tmtwo}\cdot(1 + \measvar\var\tm)
        + \measvarPhi{\var_1}\sctxtwo{\varphi_\tmtwo}\cdot(1 + \measvar\var\tm) 
      \\
      & + \meas{\tm_1}\cdot(1 + \measvarPhi{\var_1}\sctxtwo{\varphi_\tmtwo})\cdot(1 + \measvar\var\tm)
    \\
      =
      & (\measPhi\sctxtwo{\varphi_\tmtwo} 
        + \measvarPhi{\var_1}\sctxtwo{\varphi_\tmtwo} 
        + \meas{\tm_1}\cdot(1 + \measvarPhi{\var_1}\sctxtwo{\varphi_\tmtwo})
        )\cdot(1 + \measvar\var\tm)
    \\
      =
      & \measPhi{\sctxtwo\esub{\var_1}{\tm_1}}{\varphi_\tmtwo}\cdot(1 + \measvar\var\tm)
    \end{array}
    \]
  \end{itemize}
\end{enumerate}
\end{proof}

}

\begin{lemma}
\label{lem:t_measvart_geq_measvartp}
Let $\tm$ be a term such that $\tm \tovalas \tm'$ for some term $\tm'$,
and let $\var$ be a variable such that if $\tm$ 
$(\rulesub\vartwo\val)$-reduces to $\tm'$, then $\var \notin \fv\val$.
Then, $\measvar\var\tm \geq \measvar\var{\tm'}$.
Moreover, if $\var = \vartwo$, then $\measvar\var\tm > \measvar\var{\tm'}$.
\end{lemma}
\hiddenproof{
  By induction on the derivation of the judgement $\tm \tovalas \tm'$.
}{
  % Label: lem:t_measvart_geq_measvartp

\begin{proof}
By induction on the derivation of the judgement $\tm \tovalas \tm'$.
\begin{enumerate}
\item \ruleVSub.
  Then, $\tm = \vartwo \tovv{\rulesub\vartwo\val} \val = \tm'$,
  with $\val = \valas(\var)$.
  We analyse two cases, depending on whether $\var = \vartwo$ or not:
  \begin{enumerate}
  \item $\var = \vartwo$.
    By definition, $\vartwo \notin \fv\val$.
    Hence, $\measvar\vartwo\vartwo = 1 > 0 = \measvar\vartwo\val$.
  \item $\var \neq \vartwo$.
    By hypothesis $\var \notin \fv\val$.
    Hence, $\measvar\var\vartwo = 0 = \measvar\var\val$.
\end{enumerate}
\item \ruleVLsv.
  Then 
  \[
    \inferrule{
      \tmtwo \tovv{\rulesub\varthree\val} \tmtwo'
    }{
      \tm = \tmtwo\esub\varthree{\val\sctx} \tovv\ruleVLsv 
      \tmtwo'\esub\varthree\val\sctx = \tm'
    }\ruleVLsv
  \]
  Hence
  \[
    \begin{array}{rcll}
        \measvar\var{\tmtwo\esub\varthree{\val\sctx}}
      & =
      & \measvar\var\tmtwo + \measvar\var{\val\sctx}\cdot(1 + \measvar\varthree\tmtwo)
    \\
      & \geq
      &   \measvar\var{\tmtwo'} 
        + \measvar\var{\val\sctx}\cdot(1 + \measvar\varthree\tmtwo)
      & \text{(By \ih on $\tmtwo$, as $\var \neq \varthree$ by $\alpha$-conversion})
    \\
      & >
      &   \measvar\var{\tmtwo'} 
        + \measvar\var{\val\sctx}\cdot(1 + \measvar\varthree{\tmtwo'})
      & \text{(By \ih on $\tmtwo$, since $\varthree \notin \fv\val$)}
    \\
      & =
      &   \measvar\var{\tmtwo'} 
        + \measvarPhi\var{\sctx}{\varphi_\val}\cdot(1 + \measvar\varthree{\tmtwo'})
      & \text{(By \cref{lem:splitting_measVar_meas_tL} (1))}
    \\
      & \geq
      & \measvarPhi\var{\sctx}{\varphi_{\tmtwo'\esub\varthree\val}}
      & \text{(By \cref{lem:t_disj_L_measLtL_comp} (1))}
    \\
      & =
      & \measvar\var{\tmtwo'\esub\varthree\val\sctx}
      & \text{(By \cref{lem:splitting_measVar_meas_tL} (1))}
    \end{array}
  \]
\item \ruleVAppL.
  Then 
  \[
    \inferrule{
      \tmtwo \tovv\rulename \tmtwo'
    }{
      \tm = \tmtwo \, \tmthree \tovv\rulename
      \tmtwo' \, \tmthree = \tm'
    }\ruleVAppL
  \]
  Hence 
  $\measvar\var{\tmtwo \, \tmthree} 
  = \measvar\var\tmtwo + \measvar\var\tmthree
  \geq \measvar\var{\tmtwo'} + \measvar\var\tmthree
  = \measvar\var{\tmtwo' \, \tmthree}$ by the \ih on $\tmtwo$.
  Note that if $\rulename = \rulesub\var{\valas(\var)}$, 
  it should be the $>$ symbol instead of the $\geq$ one.
\item \ruleVAppR.
  Analogous to the previous case.
\item \ruleVEsL.
  Then 
  \[
    \inferrule{
      \tmtwo \tovv\rulename \tmtwo'
      \sep
      \vartwo \notin \fv\rulename
    }{
      \tm = \tmtwo\esub\vartwo\tmthree \tovv\rulename
      \tmtwo'\esub\vartwo\tmthree = \tm'
    }\ruleVEsL
  \]
  Hence
  \[
    \begin{array}{rcll}
        \measvar\var{\tmtwo\esub\vartwo\tmthree}
      & =
      & \measvar\var\tmtwo + \measvar\var\tmthree\cdot(1 + \measvar\vartwo\tmtwo)
    \\
      & \geq
      & \measvar\var{\tmtwo'} + \measvar\var\tmthree\cdot(1 + \measvar\vartwo\tmtwo)
      & \text{(By \ih on $\tmtwo$)}
    \\
      & \geq
      & \measvar\var{\tmtwo'} + \measvar\var\tmthree\cdot(1 + \measvar\vartwo{\tmtwo'})
      & \text{(By \ih on $\tmtwo$)}
    \\
      & =
      & \measvar\var{\tmtwo'\esub\vartwo\tmthree}
    \end{array}
  \]
  Note that if $\rulename = \rulesub\var{\valas(\var)}$, 
  it should be the $>$ symbol instead of the $\geq$ one.
\item \ruleVEsR.
  Then 
  \[
    \inferrule{
      \tmthree \tovv\rulename \tmthree'
    }{
        \tm = \tmtwo\esub\vartwo\tmthree 
        \tovv\rulename
        \tmtwo\esub\vartwo{\tmthree'} = \tm'
    }\ruleVEsR
  \]
  Hence
  \[
    \begin{array}{rcll}
        \measvar\var{\tmtwo\esub\vartwo\tmthree}
      & =
      & \measvar\var\tmtwo + \measvar\var\tmthree\cdot(1 + \measvar\vartwo\tmtwo)
    \\
      & \geq
      & \measvar\var\tmtwo + \measvar\var{\tmthree'}\cdot(1 + \measvar\vartwo\tmtwo)
      & \text{(By \ih on $\tmthree$)}
    \\
      & =
      & \measvar\var{\tmtwo\esub\vartwo{\tmthree'}}
    \end{array}
  \]
  Note that if $\rulename = \rulesub\var{\valas(\var)}$, 
  it should be the $>$ symbol instead of the $\geq$ one.
\end{enumerate}
\end{proof}

}

Now we can state a lemma that involves the general notion of decreasing measure:
\begin{restatable}{lemma}{decreasingmeasure}
\label{lem:decreasing-measure}
Let $\tm$ be a term. Then:
\begin{itemize}
\item
  If $\tm \tovv\rulename \tm'$, then $\meas\tm \geq \meas{\tm'}$ when 
  $\rulename = \rulesub\var\val$ or $\meas\tm > \meas{\tm'}$ when 
  $\rulename = \rulelsv$.
\item
  If $\tm \tovalas \tm'$, then $\measvalas\tm > \measvalas{\tm'}$.
\end{itemize}
\end{restatable}
\begin{proof}
The proof uses 
\cref{lem:tsub_meast_geq_meastp,lem:tlsv_meast_greater_meastp,lem:t_measvalas_greater_meastp}.
\end{proof}

\begin{lemma}
\label{lem:tsub_meast_geq_meastp}
Let $\tm \tovv{\rulesub\var\val} \tm'$.
Then, $\meas\tm \geq \meas{\tm'}$.
\end{lemma}
\hiddenproof{
  By induction on the derivation of the judgement 
  $\tm \tovv{\rulesub\var\val} \tm'$.
}{
  % Label: lem:tsub_meast_geq_meastp

\begin{proof}
By induction on the derivation of the judgement 
$\tm \tovv{\rulesub\var\val} \tm'$.
\begin{itemize}
\item \ruleVSub.
  Then, $\tm = \var \tovv{\rulesub\var\val} \val = \tm'$.
  Hence, $\meas\var = 0 = \meas\val$.
\item \ruleVAppL.
  Then
  \[
    \inferrule{
      \tmtwo \tovv{\rulesub\var\val} \tmtwo'
    }{
      \tm = \tmtwo \, \tmthree \tovv{\rulesub\var\val} \tmtwo' \, \tmthree = \tm'
    }\ruleVAppL
  \]
  Hence by the \ih on $\tmtwo$ we obtain
  $\meas{\tmtwo \, \tmthree} = \meas\tmtwo + \meas\tmthree 
  \geq \meas{\tmtwo'} + \meas\tmthree = \meas{\tmtwo' \, \tmthree}$.
\item \ruleVAppR.
  Analogous to the previous case.
\item \ruleVEsL.
  Then
  \[
    \inferrule{
      \tmtwo \tovv{\rulesub\var\val} \tmtwo'
      \sep
      \vartwo \notin \fv{\rulesub\var\val}
    }{
      \tm = \tmtwo\esub\vartwo\tmthree \tovv{\rulesub\var\val} \tmtwo'\esub\vartwo\tmthree = \tm'
    }\ruleVEsL
  \]
  Hence
  \[
    \begin{array}{rcll}
        \meas{\tmtwo\esub\vartwo\tmthree}
      & =
      & \meas\tmtwo + \measvar\vartwo\tmtwo + \meas\tmthree\cdot(1 + \measvar\vartwo\tmtwo)
    \\
      & \geq 
      & \meas{\tmtwo'} + \measvar\vartwo\tmtwo + \meas\tmthree\cdot(1 + \measvar\vartwo\tmtwo) 
      & \text{(By \ih on $\tmtwo$)}
    \\
      & \geq
      & \meas{\tmtwo'} + \measvar\vartwo{\tmtwo'} + \meas\tmthree\cdot(1 + \measvar\vartwo\tmtwo) 
      & \text{(By \cref{lem:t_measvart_geq_measvartp})}
    \\
      & \geq
      & \meas{\tmtwo'} + \measvar\vartwo{\tmtwo'} + \meas\tmthree\cdot(1 + \measvar\vartwo{\tmtwo'}) 
      & \text{(By \cref{lem:t_measvart_geq_measvartp})}
    \\
      & =
      & \meas{\tmtwo'\esub\vartwo\tmthree} 
    \end{array}
  \]
\item \ruleVEsR.
  Then
  \[
    \inferrule{
      \tmthree \tovv{\rulesub\var\val} \tmthree'
    }{
      \tm = \tmtwo\esub\vartwo\tmthree \tovv{\rulesub\var\val} \tmtwo\esub\vartwo{\tmthree'} = \tm'
    }\ruleVEsR
  \]
  Hence
  \[
    \begin{array}{rcll}
        \meas{\tmtwo\esub\vartwo\tmthree}
      & =
      & \meas\tmtwo + \measvar\vartwo\tmtwo + \meas\tmthree\cdot(1 + \measvar\vartwo\tmtwo)
    \\
      & \geq 
      & \meas\tmtwo + \measvar\vartwo\tmtwo + \meas{\tmthree'}\cdot(1 + \measvar\vartwo\tmtwo) 
      & \text{(By \ih on $\tmthree$)}
    \\
      & =
      & \meas{\tmtwo\esub\vartwo{\tmthree'}} 
    \end{array}
  \]
\end{itemize}
\end{proof}

}

\begin{lemma}
\label{lem:tlsv_meast_greater_meastp}
Let $\tm \tovv\rulelsv \tm'$.
Then, $\meas\tm > \meas{\tm'}$.
\end{lemma}
\hiddenproof{
  By induction on the derivation of the judgement $\tm \tovv\rulelsv \tm'$.
}{
  % Label: lem:tlsv_meast_greater_meastp

\begin{proof}
By induction on the derivation of the judgement 
$\tm \tovv\rulelsv \tm'$.
\begin{itemize}
\item \ruleVLsv.
  Then
  \[
    \inferrule{
      \tmtwo \tovv{\rulesub\var\val} \tmtwo'
    }{
      \tm = \tmtwo\esub\var{\val\sctx} \tovv\rulelsv \tmtwo'\esub\var\val\sctx = \tm'
    }\ruleVLsv
  \]
  Hence
  \[
    \begin{array}{rcll}
        \meas{\tmtwo\esub\var{\val\sctx}}
      & =
      & \meas\tmtwo + \measvar\var\tmtwo + \meas{\val\sctx}\cdot(1 + \measvar\var\tmtwo)
    \\
      & >
      & \meas\tmtwo + \measvar\var{\tmtwo'} + \meas{\val\sctx}\cdot(1 + \measvar\var\tmtwo)
      & \text{(By \cref{lem:t_measvart_geq_measvartp})}
    \\
      & >
      & \meas\tmtwo + \measvar\var{\tmtwo'} + \meas{\val\sctx}\cdot(1 + \measvar\var{\tmtwo'})
      & \text{(By \cref{lem:t_measvart_geq_measvartp})}
    \\
      & \geq
      & \meas{\tmtwo'} + \measvar\var{\tmtwo'} + \meas{\val\sctx}\cdot(1 + \measvar\var{\tmtwo'})
      & \text{(By \cref{lem:tsub_meast_geq_meastp})}
    \\
      & \geq
      & \meas{\tmtwo'} + \measvar\var{\tmtwo'} + (\meas\val + \measPhi\sctx{\varphi_\val})\cdot(1 + \measvar\var{\tmtwo'})
      & \text{(By \cref{lem:splitting_measVar_meas_tL} (2))}
    \\
      & =
      & \meas{\tmtwo'} + \measvar\var{\tmtwo'} + (0 + \measPhi\sctx{\varphi_\val})\cdot(1 + \measvar\var{\tmtwo'})
    \\
      & =
      & \meas{\tmtwo'} + \measvar\var{\tmtwo'} + \measPhi\sctx{\varphi_\val}\cdot(1 + \measvar\var{\tmtwo'})
    \\
      & \geq
      & \meas{\tmtwo'} + \measvar\var{\tmtwo'} +  \measPhi\sctx{\varphi_{\tmtwo'\esub\var\val}}
      & \text{(By \cref{lem:t_disj_L_measLtL_comp} (2))}
    \\
      & =
      & \meas{\tmtwo'} + \measvar\var{\tmtwo'} + 0\cdot(1 + \measvar\var{\tmtwo'})
      + \measPhi\sctx{\varphi_{\tmtwo'\esub\var\val}}
    \\
      & =
      & \meas{\tmtwo'} + \measvar\var{\tmtwo'} + \meas\val\cdot(1 + \measvar\var{\tmtwo'})
      + \measPhi\sctx{\varphi_{\tmtwo'\esub\var\val}}
    \\
      & =
      & \meas{\tmtwo'\esub\var\val} + \measPhi\sctx{\varphi_{\tmtwo'\esub\var\val}}
    \\
      & =
      & \meas{\tmtwo'\esub\var\val\sctx}
      & \text{(By \cref{lem:splitting_measVar_meas_tL} (2))}
    \end{array}
  \]
\item \ruleVAppL.
  Then
  \[
    \inferrule{
      \tmtwo \tovv\rulelsv \tmtwo'
    }{
      \tm = \tmtwo \, \tmthree \tovv\rulelsv \tmtwo' \, \tmthree = \tm'
    }\ruleVAppL
  \]
  Hence by the \ih on $\tmtwo$, we obtain
  $   \meas{\tmtwo \, \tmthree} 
    = \meas\tmtwo + \meas\tmthree
    > \meas{\tmtwo'} + \meas\tmthree
    = \meas{\tmtwo \, \tmthree}$.
\item \ruleVAppR.
  Analogous to the previous case.
\item \ruleVEsL.
  Then
  \[
    \inferrule{
      \tmtwo \tovv\rulelsv \tmtwo'
    }{
      \tm = \tmtwo\esub\var\tmthree \tovv\rulelsv \tmtwo'\esub\var\tmthree = \tm'
    }\ruleVEsL
  \]
  where $\fv\rulelsv = \emptyset$, so we have $\var \notin \fv\rulelsv$.
  Hence
  \[
    \begin{array}{rcll}
        \meas{\tmtwo\esub\var\tmthree}
      & =
      & \meas\tmtwo + \measvar\var\tmtwo + \meas\tmthree\cdot(1 + \measvar\var\tmtwo)
    \\
      & >
      & \meas{\tmtwo'} + \measvar\var\tmtwo + \meas\tmthree\cdot(1 + \measvar\var\tmtwo)
      & \text{(By \ih on $\tmtwo$)}
    \\
      & \geq
      & \meas{\tmtwo'} + \measvar\var{\tmtwo'} + \meas\tmthree\cdot(1 + \measvar\var{\tmtwo'})
      & \text{(By \cref{lem:t_measvart_geq_measvartp})}
    \\
      & =
      & \meas{\tmtwo'\esub\var\tmthree}
    \end{array}
  \]
\item \ruleVEsR.
  Then
  \[
    \inferrule{
      \tmthree \tovv\rulelsv \tmthree'
    }{
      \tm = \tmtwo\esub\var\tmthree \tovv\rulelsv \tmtwo\esub\var{\tmthree'} = \tm'
    }\ruleVEsR
  \]
  Hence,
  \[
    \begin{array}{rcll}
        \meas{\tmtwo\esub\var\tmthree}
      & =
      & \meas\tmtwo + \measvar\var\tmtwo + \meas\tmthree\cdot(1 + \measvar\var\tmtwo)
    \\
      & >
      & \meas\tmtwo + \measvar\var\tmtwo + \meas{\tmthree'}\cdot(1 + \measvar\var\tmtwo)
      & \text{(By \ih on $\tmthree$)}
    \\
      & =
      & \meas{\tmtwo\esub\var{\tmthree'}}
    \end{array}
  \]
\end{itemize}
\end{proof}

}

\begin{lemma}
\label{lem:t_measvalas_greater_meastp}
Let $\tm \tovalas \tm'$. 
Then, $\measvalas\tm > \measvalas{\tm'}$.
\end{lemma}
\hiddenproof{
  By case analysis on the rule used to derive $\tm \tovalas \tm'$.
}{
  % Label: lem:t_measvalas_greater_meastp

\begin{proof}
We analyse two cases, depending on the rule which is used to derive 
$\tm \tovalas \tm'$.
\begin{itemize}
\item 
  If $\tm \tovv\rulelsv \tm'$, then
  \[
    \begin{array}{rcll}
        \measvalas\tm
      & =
      & \meas\tm + \sum_{\var \in \dom\valas} \measvar\var\tm
    \\
      & >
      & \meas{\tm'} + \sum_{\var \in \dom\valas} \measvar\var\tm
      & \text{(By \cref{lem:tlsv_meast_greater_meastp})}
    \\
      & \geq
      & \meas{\tm'} + \sum_{\var \in \dom\valas} \measvar\var{\tm'}
      & \text{(By \cref{lem:t_measvart_geq_measvartp})}
    \\
      & =
      & \measvalas{\tm'}
    \end{array}
  \]
\item 
  If $\tm \tovv{\rulesub\var{\valas(\var)}} \tm'$, with 
  $\var \in \dom\valas$, then
  \[
    \begin{array}{rcll}
        \measvalas\tm
      & =
      & \meas\tm + \sum_{\var \in \dom\valas} \measvar\var\tm
    \\
      & \geq
      & \meas{\tm'} + \sum_{\var \in \dom\valas} \measvar\var\tm
      & \text{(By \cref{lem:tsub_meast_geq_meastp})}
    \\
      & >
      & \meas{\tm'} + \sum_{\var \in \dom\valas} \measvar\var{\tm'}
      & \text{(By \cref{lem:t_measvart_geq_measvartp})}
    \\
      & =
      & \measvalas{\tm'}
    \end{array}
  \]
\end{itemize}
\end{proof}

}

\begin{lemma}
\label{lem:vL_reduction_tovalas}
Let $\val$ be a value and $\sctx$ a substitution context. Then:
\begin{enumerate}
\item
  If there exists $\tm$ such that $\val\sctx \tovv\rulename \tm$,
  then there exist $\val'$ and $\sctx'$ such that $\tm = \val'\sctx'$.
\item
  If $\val\sctx \tovv\rulename \val'\sctx'$, then 
  $\tm\esub\var\val\sctx \tovv\rulename \tm\esub\var{\val'}\sctx'$,
  for any term $\tm$.
\end{enumerate}
\end{lemma}
% Label lem:vL_reduction_tovalas

\begin{proof}
We prove each item independently.
\begin{enumerate}
\item
  By induction on the derivation of $\val\sctx \tovv\rulename \tm$.
  Note that it is not possible to apply rules \ruleVDb, \ruleVAppL 
  nor \ruleVAppR, given that $\val\sctx$ is not an application.
  Then, we are left to analyse the following cases:
  \begin{itemize}
  \item \ruleVSub.
    Then, $\var \tovv{\rulesub\var\valtwo} \valtwo$, where
    $\val = \var$, $\sctx = \ctxhole$, and $\tm = \valtwo$,
    so we are done, with $\val' = \valtwo$ and $\sctx' = \ctxhole$.
  \item \ruleVLsv.
    Then, $\sctx = \sctx_1\esub\var{\valtwo\sctx_2}$, and thus
    \[
      \inferrule{
        \val\sctx_1 \tovv{\rulesub\var\valtwo} \tm_1
      }{
        \val\sctx_1\esub\var{\valtwo\sctx_2} \tovv\rulelsv
        \tm_1\esub\var\valtwo\sctx_2 = \tm
      }\ruleVLsv
    \]
    We apply the \ih on $\val\sctx_1$, yielding $\val'_1$ and $\sctx'_1$
    such that $\tm_1 = \val'_1\sctx'_1$. 
    And we are done, with $\val' = \val'_1$ and 
    $\sctx' = \sctx'_1\esub\var\valtwo\sctx_2$.
  \item \ruleVEsL.
    Then, $\sctx = \sctx_1\esub\var\tmtwo$, and thus
    \[
      \inferrule{
        \val\sctx_1 \tovv\rulename \tm_1
        \sep
        \var \notin \fv\rulename
      }{
        \val\sctx_1\esub\var\tmtwo \tovv\rulename \tm_1\esub\var\tmtwo
      }\ruleVEsL
    \]

    We apply the \ih on $\val\sctx_1$, yielding $\val'_1$ and $\sctx'_1$
    such that $\tm_1 = \val'_1\sctx'_1$.
    And we are done, with $\val' = \val'_1$ and 
    $\sctx' = \sctx'_1\esub\var\tmtwo$.
  \item \ruleVEsR.
    Then, $\sctx = \sctx_1\esub\var\tmtwo$, thus deriving
    $\val\sctx_1\esub\var\tmtwo \tovv\rulename \val\sctx_1\esub\var{\tmtwo'}$
    from
    $\tmtwo \tovv\rulename \tmtwo'$.
    We are done with $\val' = \val$ and $\sctx' = \sctx_1\esub\var{\tmtwo'}$,
  \end{itemize}  
\item
  By induction on the length of $\sctx$.
  \begin{itemize}
  \item $\sctx = \ctxhole$.
    Then, $\val \tovv\rulename \val'$, which can only be derived by 
    rule \ruleVSub.
    We obtain $\tm\esub\var\val \tovv\rulename \tm\esub\var{\val'}$
    by applying rule \ruleVEsR.
  \item $\sctx = \sctx_1\esub\vartwo\tmtwo$.
    We analyse different cases, depending on the rule used to derive
    $\val\sctx_1\esub\vartwo\tmtwo \tovv\rulename \val'\sctx'$.
    Since it is not possible to apply rules \ruleVSub, \ruleVDb, 
    \ruleVAppL, and \ruleVAppR, given that
    $\val\sctx_1\esub\vartwo\tmtwo$ is neither a variable nor an 
    application, so we are left to analyse the following three cases:
    \begin{itemize}
    \item \ruleVLsv.
      Then $\tmtwo = \valtwo\sctx_2$ and $\rulename = \rulelsv$, deriving
      \[
        \inferrule{
          \val\sctx_1 \tovv{\rulesub\vartwo\valtwo} \tmthree
        }{
          \val\sctx_1\esub\vartwo{\valtwo\sctx_2} \tovv\rulelsv \tmthree\esub\vartwo\valtwo\sctx_2
        }\ruleVLsv
      \]

      Moreover, $\tmthree = \val'\sctx_3$ by point (1), hence we 
      necessarily have $\sctx' = \sctx_3\esub\vartwo\valtwo\sctx_2$. 
      Thus in particular $\val\sctx_1 \tovv{}\val'\sctx_3$.
      We apply the \ih on $\sctx_1$, yielding
      $\tm\esub\var\val\sctx_1 \tovv{\rulesub\vartwo\valtwo} \tm\esub\var{\val'}\sctx_3$.
      Applying rule \ruleVLsv, we obtain
      $\tm\esub\var\val\sctx = 
      \tm\esub\var\val\sctx_1\esub\vartwo{\valtwo\sctx_2} 
      \tovv{\rulesub\vartwo\valtwo} 
      \tm\esub\var{\val'}\sctx_3\esub\vartwo\valtwo\sctx_2 =
      \tm\esub\var{\val'}\sctx'$.
    \item \ruleVEsL.
      Then
      \[
        \inferrule{
          \val\sctx_1 \tovv\rulename \tmthree
          \sep
          \vartwo \notin \fv\rulename
        }{
          \val\sctx_1\esub\vartwo\tmtwo \tovv\rulename \tmthree\esub\vartwo\tmtwo
        }\ruleVEsL
      \]

      Moreover, $\tmthree = \val'\sctx_3$ by point (1), hence we 
      necessarily have $\sctx' = \sctx_3\esub\vartwo\tmtwo$.
      Thus in particular $\val\sctx_1 \tovv{} \val'\sctx_3$.
      We apply the \ih on $\sctx_1$, yielding
      $\tm\esub\var\val\sctx_1 \tovv\rulename \tm\esub\var{\val'}\sctx_3$.
      Given that $\vartwo \notin \fv\rulename$, we can then apply 
      rule \ruleVEsL, yielding
      $\tm\esub\var\val\sctx =
      \tm\esub\var\val\sctx_1\esub\vartwo\tmtwo \tovv\rulename 
      \tm\esub\var{\val'}\sctx_3\esub\vartwo\tmtwo =
      \tm\esub\var{\val'}\sctx'$.
    \item \ruleVEsR.
      Then,
      $\val\sctx_1\esub\vartwo\tmtwo \tovv\rulename \val\sctx_1\esub\vartwo{\tmtwo'}$
      is derived from
      $\tmtwo \tovv\rulename \tmtwo'$.
      Hence, $\val' = \val$ and $\sctx' = \sctx_1\esub\vartwo{\tmtwo'}$.
      By applying \ruleVEsR, we obtain
      $\tm\esub\var\val\sctx =
      \tm\esub\var\val\sctx_1\esub\vartwo\tmtwo \tovv\rulename
      \tm\esub\var\val\sctx_1\esub\vartwo{\tmtwo'} =
      \tm\esub\var{\val'}\sctx'$.
    \end{itemize}
  \end{itemize}
\end{enumerate}
\end{proof}

\begin{lemma}
\label{lem:t_reducessub_t1_vL_reduces_v1L1_t1_reduces_t1p}
If $\tm \tovv{\rulesub\var\val} \tm_1$ and 
$\val\sctx \tovv\rulename \val'\sctx'$, then there exists $\tm'_1$ 
such that $\tm_1 \tovveq\rulename \tm'_1$.
\end{lemma}
% Label: lem:t_reducessub_t1_vL_reduces_v1L1_t1_reduces_t1p

\begin{proof}
By induction on the derivation of $\tm \tovv{\rulesub\var\val} \tm_1$.
The most interesting cases are \ruleVSub and \ruleVEsL, while the 
rest are analogous to the \ruleVEsL case.
\begin{itemize}
\item \ruleVSub.
  Then, $\tm = \var \tovv{\rulesub\var\val} \val = \tm_1$, so we take 
  $\tm'_1 \eqdef \val$, so we are done.
% \item \ruleVAppL.
%   Then,
%   \[
%     \inferrule{
%       \tmtwo \tovv{\rulesub\var\val} \tmtwo_1
%     }{
%       \tm = \tmtwo \, \tmthree \tovv{\rulesub\var\val} \tmtwo_1 \, \tmthree = \tm_1
%     }\ruleVAppL
%   \]
%   We apply the \ih on $\tmtwo$, yielding $\tmtwo'_1$ such that
%   $\tmtwo_1 \tovveq\rulename \tmtwo'_1$.
%   Applying rule \ruleVAppL, we obtain
%   $\tm_1 = \tmtwo_1 \, \tmthree \tovveq\rulename \tmtwo'_1 \, \tmthree = \tm'_1$.
% \item \ruleVAppR.
%   Then
%   \[
%     \inferrule{
%       \tmthree \tovv{\rulesub\var\val} \tmthree_1
%     }{
%       \tm = \tmtwo \, \tmthree \tovv{\rulesub\var\val} \tmtwo \, \tmthree_1 = \tm_1
%     }\ruleVAppR
%   \]
%   We apply the \ih on $\tmthree$, yielding $\tmthree'_1$ such that
%   $\tmthree_1 \tovveq\rulename \tmthree'_1$.
%   Applying rule \ruleVAppR we obtain
%   $\tm_1 = \tmtwo \, \tmthree_1 \tovveq\rulename \tmtwo \, \tmthree'_1 = \tm'_1$.
\item \ruleVEsL.
  Then
  \[
    \inferrule{
      \tmtwo \tovv{\rulesub\var\val} \tmtwo_1
      \sep
      \vartwo \notin \fv{\rulesub\var\val}
    }{
      \tm = \tmtwo\esub\vartwo\tmthree \tovv{\rulesub\var\val} \tmtwo_1\esub\vartwo\tmthree = \tm_1
    }\ruleVEsL
  \]
  We apply the \ih on $\tmtwo$, yielding $\tmtwo'_1$ such that
  $\tmtwo_1 \tovveq\rulename \tmtwo'_1$.
  We may assume $\vartwo \notin \fv\rulename$ by $\alpha$-conversion,
  so we can apply rule \ruleVEsL, yielding
  $\tm_1 = \tmtwo_1\esub\vartwo\tmthree \tovveq\rulename \tmtwo'_1\esub\vartwo\tmthree = \tm'_1$.
% \item \ruleVEsR.
%   Then
%   \[
%     \inferrule{
%       \tmthree \tovv{\rulesub\var\val} \tmthree_1
%     }{
%       \tm = \tmtwo\esub\vartwo\tmthree \tovv{\rulesub\var\val} \tmtwo\esub\vartwo{\tmthree_1} = \tm_1
%     }\ruleVEsR
%   \]
%   We apply the \ih on $\tmthree$, yielding $\tmthree'_1$ such that
%   $\tmthree_1 \tovveq\rulename \tmthree'_1$.
%   Applying rule \ruleVEsR we obtain
%   $\tm_1 = \tmtwo\esub\vartwo{\tmthree_1} \tovveq\rulename \tmtwo\esub\vartwo{\tmthree'_1} = \tm'_1$.
\end{itemize}
\end{proof}

\begin{proposition}[Local Confluence / WCR of $\tovalas$]
\label{lem:tovalas_wcr}
Let $\tm$ be a $\valas$-reducible term such that
$\tm \tovv{\rulename_1} \tm_1$ and
$\tm \tovv{\rulename_2} \tm_2$, with 
$\rulename_1, \rulename_2 \in \tovalas$.
Then, there exists $\tm'$ such that $\tm_1 \tovvn{\rulename_2} \tm'$
and $\tm_2 \tovvn{\rulename_1} \tm'$.
\end{proposition}
% Label: lem:tovalas_wcr

\begin{proof}
By induction on $\tm$.
Case $\tm = \lam\var\tmtwo$ is impossible, since there are no rules 
to reduce abstractions. We analyse the remaining cases.
\begin{itemize}
\item $\tm = \var$.
  The only rule to reduce $\var$ is \ruleVSub, and thus
  $\tm = \var \tovv{\rulesub\var{\dom\valas}} \valas(\var) = \tm_1 = \tm_2$.
  So we are done.
\item $\tm = \tmtwo \, \tmthree$.
  We can $\valas$-reduce $\tm$ with the rules \ruleVAppL and \ruleVAppR,
  so we have the following cases:
  \begin{itemize}
  \item \ruleVAppL-\ruleVAppR.
    Then, $\tm = \tmtwo \, \tmthree \tovv{\rulename_1} 
    \tmtwo_1 \, \tmthree = \tm_1$, which is derived from 
    (1) $\tmtwo \tovv{\rulename_1} \tmtwo_1$, and we have
    $\tm = \tmtwo \, \tmthree \tovv{\rulename_2} 
    \tmtwo \, \tmthree_2 = \tm_2$, which is derived from 
    (2) $\tmthree \tovv{\rulename_2} \tmthree_2$.
    We apply rule \ruleVAppL with (1) as premise, yielding
    $\tm_2 = \tmtwo \, \tmthree_2 \tovv{\rulename_1} 
    \tmtwo_1 \, \tmthree_2 = \tm'$.
    And we apply rule \ruleVAppR with (2) as premise, yielding
    $\tm_1 = \tmtwo_1 \, \tmthree \tovv{\rulename_2} 
    \tmtwo_1 \, \tmthree_2 = \tm'$.
    The following diagram summarises the proof:
    \[
      \xymatrixcolsep{2pc}
      \xymatrix{
          \tm = \tmtwo \, \tmthree 
            \arVr{\rulename_1} 
            \arVd{\rulename_2}
        & \tmtwo_1 \, \tmthree = \tm_1
            \arsdVd{\rulename_2}
      \\
          \tm_2 = \tmtwo \, \tmthree_2
            \arsdVr{\rulename_1}
        & \tmtwo_1 \, \tmthree_2 = \tm'
      }
    \]
  \item \ruleVAppL-\ruleVAppL.
    We have
    $\tm = \tmtwo \, \tmthree \tovv{\rulename_1} 
    \tmtwo_1 \, \tmthree = \tm_1$, which is derived from 
    $\tmtwo \tovv{\rulename_1} \tmtwo_1$, and we have
    $\tm = \tmtwo \, \tmthree \tovv{\rulename_2} 
    \tmtwo_2 \, \tmthree = \tm_2$, which is derived from 
    $\tmtwo \tovv{\rulename_2} \tmtwo_2$.
    We apply the \ih on $\tmtwo$, yielding $\tmtwo'$ such that
    $\tmtwo_1 \tovvn{\rulename_2} \tmtwo'$ and
    $\tmtwo_2 \tovvn{\rulename_1}\ \tmtwo'$.
    Applying rule \ruleVAppL to reduce both $\tmtwo_1 \, \tmthree$ 
    and $\tmtwo_2 \, \tmthree$, we obtain
    $\tm_1 = \tmtwo_1 \, \tmthree \tovvn{\rulename_2} \tmtwo' \, \tmthree = \tm'$
    and
    $\tm_2 = \tmtwo_2 \, \tmthree \tovvn{\rulename_1} \tmtwo' \, \tmthree = \tm'$
    respectively.
    The following diagram summarises the proof:
    \[
      \xymatrixcolsep{2pc}
      \xymatrix{
          \tm = \tmtwo \, \tmthree 
            \arVr{\rulename_1} 
            \arVd{\rulename_2}
        & \tmtwo_1 \, \tmthree = \tm_1
            \arsdVdn{\rulename_2}
      \\
          \tm_2 = \tmtwo_2 \, \tmthree
            \arsdVrn{\rulename_1}
        & \tmtwo' \, \tmthree = \tm'
      }
    \]
  \item \ruleVAppR-\ruleVAppR.
    Analogous to the previous case.
  \end{itemize}
\item $\tm = \tmtwo\esub\var\tmthree$.
  We can $\valas$-reduce $\tm$ with the rules \ruleVLsv, \ruleVEsL, 
  and \ruleVEsR, so we have the following cases:
  \begin{itemize}
  \item \ruleVLsv-\ruleVLsv.
    We have
    \[
      \inferrule{
        \tmtwo \tovv{\rulesub\var\val} \tmtwo_1
      }{
        \tm = \tmtwo\esub\var{\val\sctx}
        \tovv\rulelsv
        \tmtwo_1\esub\var\val\sctx = \tm_1
      }\ruleVLsv
    \]
    where $\tmthree = \val\sctx$ and $\rulename_1 = \rulelsv$,
    and we have
    \[
      \inferrule{
        \tmtwo \tovv{\rulesub\var\val} \tmtwo_2
      }{
        \tm = \tmtwo\esub\var{\val\sctx} \tovv\rulelsv
        \tmtwo_2\esub\var\val\sctx = \tm_2
      }\ruleVLsv
    \]
    where $\rulename_2 = \ruleVLsv$.
    We apply the \ih on $\tmtwo$, yielding $\tmtwo'$ such that
    $\tmtwo_1 \tovvn{\rulesub\var\val} \tmtwo'$ and
    $\tmtwo_2 \tovvn{\rulesub\var\val} \tmtwo'$.
    We can then apply rule \ruleVLsv on both 
    $\tmtwo_1\esub\var\val$ and $\tmtwo_2\esub\var\val$, yielding
    $\tm_1 = \tmtwo_1\esub\var\val \tovvn\rulelsv \tmtwo'\esub\var\val$
    and
    $\tm_2 = \tmtwo_2\esub\var\val \tovvn\rulelsv \tmtwo'\esub\var\val$ resp.
    To conclude,
    $\tmtwo_1\esub\var\val\sctx \tovvn\rulelsv \tmtwo'\esub\var\val\sctx$
    and
    $\tmtwo_2\esub\var\val\sctx \tovvn\rulelsv \tmtwo'\esub\var\val\sctx$
    is derived by successively applying rule \ruleVEsL,
    as $\fv\rulelsv = \emptyset$.
    The following diagram summarises the proof:
    \[
      \xymatrixcolsep{2pc}
      \xymatrix{
          \tm = \tmtwo\esub\var{\val\sctx} 
            \arVr\rulelsv 
            \arVd\rulelsv
        & \tmtwo_1\esub\var\val\sctx = \tm_1
            \arsdVdn\rulelsv
      \\
          \tm_2 = \tmtwo_2\esub\var\val\sctx
            \arsdVrn\rulelsv
        & \tmtwo'\esub\var\val\sctx = \tm'
      }
    \]
  \item \ruleVLsv-\ruleVEsL.
    We have
    \[
      \inferrule{
        \tmtwo \tovv{\rulesub\var\val} \tmtwo_1
      }{
        \tm = \tmtwo\esub\var{\val\sctx}
        \tovv\rulelsv
        \tmtwo_1\esub\var\val\sctx = \tm_1
      }\ruleVLsv
    \]
    where $\tmthree = \val\sctx$ and $\rulename_1 = \rulelsv$,
    and we have
    \[
      \inferrule{
        \tmtwo \tovv{\rulename_2} \tmtwo_2
        \sep
        \var \notin \fv{\rulename_2} \ (1)
      }{
         \tm = \tmtwo\esub\var{\val\sctx}
         \tovv{\rulename_2}
         \tmtwo_2\esub\var{\val\sctx} = \tm_2
      }\ruleVEsL
    \]
    We apply the \ih on $\tmtwo$, yielding $\tmtwo'$ such that
    (2) $\tmtwo_1 \tovvn{\rulename_2} \tmtwo'$ and 
    (3) $\tmtwo_2 \tovvn{\rulesub\var\val} \tmtwo'$.
    Then,
    $\tmtwo_1\esub\var\val \tovvn{\rulename_2} \tmtwo'\esub\var\val$
    by rule \ruleVEsL, with (2) and (1) as premises.
    And we conclude that
    $\tm_1 = \tmtwo_1\esub\var\val\sctx \tovvn{\rulename_2} 
    \tmtwo'\esub\var\val\sctx$ is derived by successively applying 
    rule \ruleVEsL.
    On the other hand,
    $\tm_2 = \tmtwo_2\esub\var{\val\sctx} \tovvn\rulelsv
    \tmtwo'\esub\var\val\sctx$ by rule \ruleVLsv with (3) as premise.
    The following diagram summarises the proof:
    \[
      \xymatrixcolsep{2pc}
      \xymatrix{
          \tm = \tmtwo\esub\var{\val\sctx} 
            \arVr\rulelsv 
            \arVd{\rulename_2}
        & \tmtwo_1\esub\var\val\sctx = \tm_1
            \arsdVdn{\rulename_2}
      \\
          \tm_2 = \tmtwo_2\esub\var{\val\sctx}
            \arsdVrn\rulelsv
        & \tmtwo'\esub\var\val\sctx = \tm'
      }
    \]
  \item \ruleVLsv-\ruleVEsR. 
    We have
    \[
      \inferrule{
        \tmtwo \tovv{\rulesub\var\val} \tmtwo_1 \ (1)
      }{
        \tm = \tmtwo\esub\var{\val\sctx}
        \tovv\rulelsv
        \tmtwo_1\esub\var\val\sctx = \tm_1
      }\ruleVLsv
    \]
    where $\tmthree = \val\sctx$ and $\rulename_1 = \rulelsv$, and 
    we have
    $\tm = \tmtwo\esub\var{\val\sctx} \tovv{\rulename_2}
    \tmtwo\esub\var{\tmthree_2} = \tm_2$, which is derived from
    (2) $\val\sctx \tovv{\rulename_2} \tmthree_2$.
    Moreover, $\tmthree_2 = \val'\sctx'$ by 
    \cref{lem:vL_reduction_tovalas} (1).
    By \cref{lem:nu_sub_change_of_values}, there exists $\tmtwo'_1$ 
    such that $\tmtwo \tovv{\rulesub\var{\val'}} \tmtwo'_1$, hence we 
    can apply rule \ruleVLsv, yielding
    $\tm_2 = \tmtwo\esub\var{\val'\sctx'} \tovv\rulelsv 
    \tmtwo'_1\esub\var{\val'}\sctx' = \tm'$.
    On the other hand, we have
    $\tmtwo_1\esub\var\val\sctx \tovv{\rulename_2} \tmtwo_1\esub\var{\val'}\sctx'$
    by \cref{lem:vL_reduction_tovalas} (2).
    And since $\tmtwo_1 \tovv{\rulename_2} \tmtwo'_1$ 
    by \cref{lem:t_reducessub_t1_vL_reduces_v1L1_t1_reduces_t1p},
    we can derive 
    $\tmtwo_1\esub\var{\val'} \tovv{\rulename_2} \tmtwo'_1\esub\var{\val'}$
    by rule \ruleVEsL, given that $\var \notin \fv{\rulename_2}$ by $\alpha$-conversion.
    Then,
    $\tmtwo_1\esub\var{\val'}\sctx' \tovv{\rulename_2} \tmtwo'_1\esub\var{\val'}\sctx'$
    by applying rule \ruleVEsL (length of $\sctx'$) times.
    The following diagram summarises the proof:
    \[
      \xymatrixcolsep{2pc}
      \xymatrix{
          \tm = \tmtwo\esub\var{\val\sctx}
            \arVr\rulelsv 
            \arVdd{\rulename_2}
        & \tmtwo_1\esub\var\val\sctx = \tm_1
            \arsdVd{\rulename_2}
      \\
        & \tmtwo_1\esub\var{\val'}\sctx'
            \arsdVd{\rulename_2}
      \\
          \tm_2 = \tmtwo\esub\var{\val'\sctx'}
            \arsdVr\rulelsv
        & \tmtwo'_1\esub\var{\val'}\sctx' = \tm'
      }
    \]
  \item \ruleVEsL-\ruleVEsL.
    We have
    $\tm = \tmtwo\esub\var\tmthree \tovv{\rulename_1}
    \tmtwo_1\esub\var\tmthree = \tm_1$ which is derived from
    $\tmtwo \tovv{\rulename_1} \tmtwo_1$ and 
    $\var \notin \fv{\rulename_1}$, and we have
    $\tm = \tmtwo\esub\var\tmthree \tovv{\rulename_2}
    \tmtwo_2\esub\var\tmthree = \tm_2$ which is derived from
    $\tmtwo \tovv{\rulename_2} \tmtwo_2$ and 
    $\var \notin \fv{\rulename_2}$.
    We apply the \ih on $\tmtwo$, yielding $\tmtwo'$ such that
    $\tmtwo_1 \tovvn{\rulename_2} \tmtwo'$ and
    $\tmtwo_2 \tovvn{\rulename_1} \tmtwo'$.
    Applying rule \ruleVEsL, we derive
    $\tm_1 = \tmtwo_1\esub\var\tmthree \tovvn{\rulename_2} \tmtwo'\esub\var\tmthree = \tm'$
    and
    $\tm_2 = \tmtwo_2\esub\var\tmthree \tovvn{\rulename_1} \tmtwo'\esub\var\tmthree = \tm'$
    respectively.
    The following diagram summarises the proof:
    \[
      \xymatrixcolsep{2pc}
      \xymatrix{
          \tm = \tmtwo\esub\var\tmthree
            \arVr{\rulename_1} 
            \arVd{\rulename_2}
        & \tmtwo_1\esub\var\tmthree = \tm_1
            \arsdVdn{\rulename_2}
      \\
          \tm_2 = \tmtwo_2\esub\var\tmthree
            \arsdVrn{\rulename_1}
        & \tmtwo'\esub\var\tmthree = \tm'
      }
    \]
  \item \ruleVEsL-\ruleVEsR.
    We have
    $\tm = \tmtwo\esub\var\tmthree \tovv{\rulename_1}
    \tmtwo_1\esub\var\tmthree = \tm_1$,
    which is derived from (1) $\tmtwo \tovv{\rulename_1} \tmtwo_1$
    and (2) $\var \notin \fv{\rulename_1}$, and we have
    $\tm = \tmtwo\esub\var\tmthree \tovv{\rulename_1}
    \tmtwo\esub\var{\tmthree_2} = \tm_2$,
    which is derived from (3) $\tmthree \tovv{\rulename_2} \tmthree_2$.
    We can apply rule \ruleVEsR with (3) as premise, yielding
    $\tm_1 = \tmtwo_1\esub\var\tmthree \tovv{\rulename_2}
    \tmtwo_1\esub\var{\tmthree_2}$.
    On the other hand, we can apply rule \ruleVEsL with (1) and (2) 
    as premises, yielding
    $\tm_2 = \tmtwo\esub\var{\tmthree_2} \tovv{\rulename_1}
    \tmtwo_1\esub\var{\tmthree_2}$.
    The following diagram summarises the proof:
    \[
      \xymatrixcolsep{2pc}
      \xymatrix{
          \tmtwo\esub\var\tmthree
            \arVr{\rulename_1} 
            \arVd{\rulename_2}
        & \tmtwo_1\esub\var\tmthree = \tm_1
            \arsdVd{\rulename_2}
      \\
          \tm_2 = \tmtwo\esub\var{\tmthree_2}
            \arsdVr{\rulename_1}
        & \tmtwo_1\esub\var{\tmthree_2} = \tm'
      }
    \]
  \item \ruleVEsR-\ruleVEsR. % Analogo al caso esL-esL
    We have
    $
     \tm = \tmtwo\esub\var\tmthree
     \tovv{\rulename_1}
     \tmtwo\esub\var{\tmthree_1} = \tm_1
    $,
    which is derived from
    $\tmthree \tovv{\rulename_1} \tmthree_1$;
    and we have
    $
      \tm = \tmtwo\esub\var\tmthree
      \tovv{\rulename_2}
      \tmtwo\esub\var{\tmthree_2} = \tm_2
    $
    which is derived from
    $\tmthree \tovv{\rulename_2} \tmthree_2$.
    We apply \ih on $\tmthree$, yielding $\tmthree'$ such that
    $\tmthree_1 \tovvn{\rulename_2} \tmthree'$
    and
    $\tmthree_2 \tovvn{\rulename_1} \tmthree'$.
    Applying rule \ruleVEsR we derive
    $
      \tm_1 = \tmtwo\esub\var{\tmthree_1} 
      \tovvn{\rulename_2}
      \tmtwo\esub\var{\tmthree'} = \tm'$
    and
    $
      \tm_2 = \tmtwo\esub\var{\tmthree_2} 
      \tovvn{\rulename_1}
      \tmtwo\esub\var{\tmthree'} = \tm'
    $
    respectively.
    The following diagram summarises the proof:
    \[
      \xymatrixcolsep{2pc}
      \xymatrix{
          \tmtwo\esub\var\tmthree
            \arVr{\rulename_1} 
            \arVd{\rulename_2}
        & \tmtwo\esub\var{\tmthree_1} = \tm_1
            \arsdVdn{\rulename_2}
      \\
          \tm_2 = \tmtwo\esub\var{\tmthree_2}
            \arsdVrn{\rulename_1}
        & \tmtwo\esub\var{\tmthree'} = \tm'
      }
    \]
  \end{itemize}
\end{itemize}
\end{proof}

Note that the second point of \cref{lem:decreasing-measure} provides 
a \emph{decreasing measure} for $\tovalas$ reduction, and the 
previous lemma states that $\tovalas$ is locally confluent.
Hence:
\begin{theorem}
\label{cor:tovalas-terminating}
\label{cor:tovalas_confluent}
\label{cor:uniqueness}
The reduction relation $\tovalas$ is terminating and confluent.
In particular, a term $\tm$ always has a unique $\tovalas$-normal form.
\end{theorem}
\begin{proof}
Termination is a straightforward consequence of 
\cref{lem:decreasing-measure}.
Confluence results from the local confluence of $\tovalas$ 
(\cref{lem:tovalas_wcr} in \cref{app:relating}) and 
Newman's Lemma~(\cf \cref{sec:preliminary_notions}).
Confluence ensures the uniqueness of normal forms.
\end{proof}

Notably, from this theorem we conclude that
\begin{corollary}
\label{cor:tolsv-terminating}
The reduction relation $\tovv\rulelsv$ is terminating.
\end{corollary}

Let $\valas$ be a value assignment. The following corollary relates 
the normal form of the reduction relation $\tovalas$ and the partial 
unfolding of a term under $\valas$.
\begin{corollary}
\label{coro:unfolding_unique_nf}
For any term $\tm$ and any value assigment $\valas$, the (unique) 
$\valas$-normal form of $\tm$ is $\unvalas\tm$.
\end{corollary}
\begin{proof}
By \cref{cor:tovalas-terminating}, $\tovalas$ is terminating, so 
there exists a term $\tmtwo$ such that $\tm \tonvalas \tmtwo$ and 
$\tmtwo$ is in $\valas$-normal form. 
By \cref{lem:unvalas_t_valasnf_t}(\ref{unvalas:dos}) 
(see \cref{app:relating}) we have $\tm \tonvalas \unvalas\tm$. 
Moreover, by \cref{lem:unvalas_t_valasnf_t}(\ref{unvalas:uno}) 
$\unvalas\tm$ is in $\valas$-normal form. 
Therefore by confluence of $\tovalas$ (\cref{cor:tovalas_confluent}), 
we conclude that $\unvalas\tm = \tmtwo$. 
\end{proof}

In particular, in the top-level case, $\unvalasevalas{\tm}$ 
is the (unique) $\rulelsv$-normal form of $\tm$.

\subsection{Relating \LOCBV and \UOCBV}
\label{sec:relating-locbv-uocbv}

\begin{lemma}
\label{lem:tEsvL_NF_if_t_and_vL_NF}
Let $\val\sctx$ be a term in $\valas$-normal form and $\tm$ be a term 
in $(\valas \cup (\var \mapsto \val))$-normal form.
Then, $\tm\esub\var\val\sctx$ is in $\valas$-normal form.
\end{lemma}
% Label: lem:tEsvL_NF_if_t_and_vL_NF

\begin{proof}
By induction on $\sctx$.
\begin{itemize}
\item $\sctx = \ctxhole$.
  There are three rules to reduce $\tm\esub\var\val$.
  We argue that none of them apply.
  \begin{enumerate}
  \item 
    If $\tm\esub\var\val$ reduces by rule \ruleVLsv, then
    \[
      \inferrule{
        \tm \tovv{\rulesub\var\val} \tm'
      }{
        \tm\esub\var\val \tovv\rulelsv \tm'\esub\var\val
      }\ruleVLsv
    \]
    but we reach a contradiction since $\tm$ is in 
    $(\valas \cup (\var \mapsto \val))$-normal form.
  \item 
    If $\tm\esub\var\val$ reduces by rule \ruleVEsL, given a rule 
    name $\rulename$ such that $\tovv\rulename \in \tovalas$, we have
    \[
      \inferrule{
        \tm \tovv\rulename \tm'
        \sep
        \var \notin \fv\rulename
      }{
        \tm\esub\var\val \tovv\rulename \tm'\esub\var\val
      }\ruleVEsL
    \]
    but we reach a contradiction, since $\tm$ is in 
    $(\valas \cup (\var \mapsto \val))$-normal form.
  \item 
    If $\tm\esub\var\val$ reduces by rule \ruleVEsR, given a rule 
    name $\rulename$ such that $\tovv\rulename \in \tovalas$, we have
    \[
      \inferrule{
        \val \tovv\rulename \val'
      }{
        \tm\esub\var\val \tovv\rulename \tm\esub\var{\val'}
      }\ruleVEsR
    \]
    but we reach a contradiction since $\val$ is a value, and so it 
    is in $\valas$-normal form.
  \end{enumerate}
\item $\sctx = \sctxtwo\esub\vartwo\tmtwo$.
  There are three rules to reduce $\tm\esub\var\val\sctxtwo\esub\vartwo\tmtwo$.
  We argue that none of them apply.
  \begin{enumerate}
  \item 
    If $\tm\esub\var\val\sctxtwo\esub\vartwo\tmtwo$ reduces by rule 
    \ruleVLsv, then $\tmtwo = \valtwo\sctx_1$ and
    \[
      \inferrule{
        \tm\esub\var\val\sctxtwo \tovv{\rulesub\vartwo\valtwo} \tmfour
      }{
        \tm\esub\var\val\sctxtwo\esub\vartwo{\valtwo\sctx_1} \tovv\rulelsv \tmfour\esub\vartwo\valtwo\sctx_1
      }\ruleVLsv
    \]
    Since $\val\sctxtwo\esub\vartwo{\valtwo\sctx_1}$ is in 
    $\valas$-normal form, then $\val\sctxtwo$ is in 
    $(\valas \cup (\vartwo \mapsto \valtwo))$-normal form.
    By $\alpha$-conversion, $\vartwo \notin \fv\tm$, so we also 
    have that $\tm$ is in 
    $(\valas \cup (\var \mapsto \val) \cup (\vartwo \mapsto \valtwo))$-normal form.
    Then, $\tm\esub\var\val\sctxtwo$ is in 
    $(\valas \cup (\vartwo \mapsto \valtwo))$-normal form by the \ih 
    on $\sctxtwo$, so we reach a contradiction.
    Hence $\tm\esub\var\val\sctx$ is in $\valas$-normal form.
  \item 
    If $\tm\esub\var\val\sctxtwo\esub\vartwo\tmtwo$ reduces by rule 
    \ruleVEsL, given a rule name $\rulename$ such that 
    $\tovv\rulename \in \tovalas$, we have
    \[
      \inferrule{
        \tm\esub\var\val\sctxtwo \tovv\rulename \tmfour
        \sep
        \vartwo \notin \fv\rulename
      }{
        \tm\esub\var\val\sctxtwo\esub\var\tmtwo \tovv\rulename \tmfour\esub\var\tmtwo
      }\ruleVEsL
    \]
    since $\val\sctx$ is in $\valas$-normal form, then 
    $\val\sctxtwo$ is in $(\valas \cup (\vartwo \mapsto \valtwo))$-normal form.
    By $\alpha$-conversion, $\vartwo \notin \fv\tm$, so we also have 
    that $\tm$ is in 
    $(\valas \cup (\var \mapsto \val) \cup (\vartwo \mapsto \valtwo))$-normal form.
    Then, $\tm\esub\var\val\sctxtwo$ is in 
    $(\valas \cup (\vartwo \mapsto \valtwo))$-normal form by the \ih 
    on $\sctxtwo$, so we reach a contradiction.
    Hence $\tm\esub\var\val\sctx$ is in $\valas$-normal form.
  \item 
    If $\tm\esub\var\val\sctxtwo\esub\vartwo\tmtwo$ reduces by rule 
    \ruleVEsR, given a rule name $\rulename$ such that 
    $\tovv\rulename \in \tovalas$, we have
    \[
      \inferrule{
        \tmtwo \tovv\rulename \tmtwo'
      }{
        \tm\esub\var\val\sctxtwo\esub\vartwo\tmtwo \tovv\rulename 
        \tm\esub\var\val\sctxtwo\esub\vartwo{\tmtwo'}
      }\ruleVEsR
    \]
    but we reach a contradiction since $\val\sctx$ is in 
    $\valas$-normal form by hypothesis.
  \end{enumerate}
\end{itemize}
\end{proof}

\begin{lemma}
\label{lem:unvalas_t_valasnf_t}
Let $\tm$ be a term and $\valas$ a value assignment. Then:
\begin{enumerate}
\item \label{unvalas:uno} $\unvalas{\tm}$ is in $\valas$-normal form
\item \label{unvalas:dos} $\tm \tonvalas \unvalas\tm$
\end{enumerate}
\end{lemma}
% Label: lem:unvalas_t_valasnf_t

\begin{proof}
We prove each item separately, and on both we proceed by induction on $\tm$.
\begin{enumerate}
\item \quad
  \begin{itemize}
  \item $\tm = \var$.
    There are two cases, depending on the form of $\unvalas\var$:
    \begin{itemize}
    \item 
      If $\var \in \dom\valas$, then $\unvalas\var = \valas(\var)$.
      If $\valas(\var)$ is a variable, then $\valas(\var)$ is in 
      $\valas$-normal form as there are no reduction rules in
      $\tovalas$ to reduce $\valas(\var)$: $\valas(\var) \notin
      \dom\valas$ because $\valas$ is idempotent. 
      If $\valas(\var)$ is an abstraction, then $\valas(\var) $ is in
      $\valas$-normal form as there are no reduction rules to reduce
      abstractions.
    \item 
      Otherwise, $\var \notin \dom\valas$ and $\unvalas\var = \var$.
      We conclude $\var$ is in $\valas$-normal form 
      since there are no reduction rules in $\tovalas$ to reduce $\var$.
    \end{itemize}
  \item $\tm = \lam\var\tmtwo$.
    Then, $\unvalas{(\lam\var\tmtwo)} = \lam\var\tmtwo$, which is in 
    $\valas$-normal form since there are no reduction rules in 
    $\tovalas$ to reduce an abstraction.
  \item $\tm = \tmtwo \, \tmthree$.
    Then, $\unvalas{(\tmtwo \, \tmthree)} = \unvalas\tmtwo \, \unvalas\tmthree$.
    There are two rules that allow us to reduce 
    $\unvalas\tmtwo \, \unvalas\tmthree$. We argue that neither applies:
    \begin{itemize}
    \item 
      We can apply the \ih on $\tmtwo$, yielding $\unvalas\tmtwo$ is 
      in $\valas$-normal form, so $\unvalas\tmtwo \, \unvalas\tmthree$ 
      does not reduce by rule \ruleVAppL.
    \item
      We can apply the \ih on $\tmthree$, yielding $\unvalas\tmthree$ 
      is in $\valas$-normal form, so 
      $\unvalas\tmtwo \, \unvalas\tmthree$ does not reduce by rule 
      \ruleVAppR.
    \end{itemize}
  \item 
    $\tm = \tmtwo\esub\var\tmthree$.
    There are two cases, depending on the form of 
    $\unvalas{\tmtwo\esub\var\tmthree}$:
    \begin{itemize}
    \item
      If $\unvalas\tmthree = \val\sctx$ and $\var \in \rv\tmtwo$, 
      then $\unvalas{\tmtwo\esub\var\tmthree} = 
      \unvalas[\valas \cup (\var \mapsto \val)]\tmtwo\esub\var\val\sctx$,
      which is in $\valas$-normal form by the \ih and by 
      \cref{lem:tEsvL_NF_if_t_and_vL_NF}.
    \item
      Otherwise, $\unvalas{\tmtwo\esub\var\tmthree} = 
      \unvalas\tmtwo\esub\var{\unvalas\tmthree}$.
      Let us reason by contradiction, assuming that the term is $\valas$-reducible.
      Firstly,
      if $\unvalas\tmtwo\esub\var{\unvalas\tmthree}$
      reduces by rule $\ruleVEsR$, 
      then it is also the case that $\unvalas\tmthree$ is $\valas$-reducible,
      therefore we reach a contradiction
      since $\unvalas\tmthree$ is in $\valas$-normal form
      by \ih on $\tmthree$.
      Then the only way to reduce the whole term is by first reducing
      $\unvalas\tmtwo$,
      but we reach again a contradiction,
      since $\unvalas\tmtwo$ is in $\valas$-normal form
      by \ih on $\tmtwo$.
      Then $\unvalas\tmtwo\esub\var{\unvalas\tmthree}$
      is in $\valas$-normal form.      
    \end{itemize}
  \end{itemize}
\item \quad
  \begin{itemize}
  \item $\tm = \var$.
    There are two cases depending on the form of $\unvalas\var$:
    \begin{itemize}
    \item
      If $\var \in \dom\valas$, then $\unvalas\var = \valas(\var)$.
      We conclude $\var \tovv{\rulesub\var{\valas(\var)}} \valas(\var)$
      by rule \ruleVSub.
    \item
      Otherwise, $\var \tonvalas \var =  \unvalas\var$ in zero steps, 
      so we are done.
    \end{itemize}
  \item $\tm = \lam\var\tmtwo$.
    Then, 
    $\lam\var\tmtwo \tonvalas \lam\var\tmtwo = \unvalas{(\lam\var\tmtwo)}$
    in zero steps, so we are done.
  \item $\tm = \tmtwo \, \tmthree$.
    Then $\unvalas{(\tmtwo \, \tmthree)} = 
    \unvalas\tmtwo \, \unvalas\tmthree$.
    We can apply the \ih on $\tmtwo$, yielding 
    $\tmtwo \tonvalas \unvalas\tmtwo$.
    Hence, we have the reduction sequence 
    $\tmtwo \, \tmthree \tonvalas \unvalas\tmtwo \, \tmthree$ in 
    which each step is obtained by applying rule \ruleVAppL.
    Analogously, we can apply the \ih on $\tmthree$, yielding 
    $\tmthree \tonvalas \unvalas\tmthree$.
    Hence we have the reduction sequence 
    $\unvalas\tmtwo \, \tmthree \tonvalas \unvalas\tmtwo \, \unvalas\tmthree$
    in which each step is obtained 
    by applying rule \ruleVAppR.
    Then we can conclude $\tmtwo \, \tmthree \tonvalas \unvalas\tmtwo \, \unvalas\tmthree$.
  \item 
    $\tm = \tmtwo\esub\var\tmthree$.
    There are two cases depending on the form of $\unvalas{\tmtwo\esub\var\tmthree}$:
    \begin{itemize}
    \item
      If $\unvalas\tmthree = \val\sctx$ and $\var \in \rv\tmtwo$, 
      then $\unvalas{\tmtwo\esub\var\tmthree} = 
      \unvalas[\valas \cup (\var \mapsto \val)]\tmtwo\esub\var\val\sctx$.
      We can apply \ih on $\tmthree$, yielding 
      $\tmthree \tonvalas \unvalas\tmthree$.
      Hence we have the reduction sequence
      $\tmtwo\esub\var\tmthree \tonvalas \tmtwo\esub\var{\val\sctx}$
      in which each step is obtained by applying rule \ruleVEsR.
      On the other hand, we can apply the \ih on $\tmtwo$, yielding 
      $\tmtwo \tonvalas[\valas \cup (\var \mapsto \val)]
      \unvalas[\valas \cup (\var \mapsto \val)]\tmtwo$;
      we can then write this reduction sequence as
      $\tmtwo \tonvalas \tmtwo'
      \tovalas[\valas \cup (\var \mapsto \val)]
      \unvalas[\valas \cup (\var \mapsto \val)]\tmtwo$, 
      for some $ \tmtwo'$.
      Then
      $\tmtwo\esub\var{\val\sctx} \tonvalas \tmtwo'\esub\var{\val\sctx}$
      by rule \ruleVEsL, and
      $\tmtwo'\esub\var{\val\sctx} \tovalas 
      \unvalas[\valas \cup (\var \mapsto \val)]\tmtwo\esub\var\val\sctx$
      by rule \ruleVLsv.
      Then we can conclude 
      $\tmtwo\esub\var\tmthree \tonvalas 
      \unvalas[\valas \cup (\var \mapsto \val)]\tmtwo\esub\var\val\sctx$.
    \item
      Otherwise, $\unvalas{\tmtwo\esub\var\tmthree} = 
      \unvalas\tmtwo\esub\var{\unvalas\tmthree}$.
      We can apply the \ih on $\tmtwo$, yielding 
      $\tmtwo \tonvalas \unvalas\tmtwo$.
      Hence, we have the reduction sequence
      $\tmtwo\esub\var\tmthree \tonvalas \unvalas\tmtwo\esub\var\tmthree$
      in which each step is obtained by applying rule \ruleVEsL, since 
      $\var$ does not occur free in the rule name by $\alpha$-conversion.
      On the other hand, we can apply the \ih on $\tmthree$, yielding 
      $\tmthree \tonvalas \unvalas\tmthree$.
      Hence, we have the reduction sequence 
      $\unvalas\tmtwo\esub\var\tmthree \tonvalas 
      \unvalas\tmtwo\esub\var{\unvalas\tmthree}$ in which each step 
      is obtained by applying rule \ruleVEsR.
      Therefore, we can conclude 
      $\tmtwo\esub\var\tmthree \tonvalas 
      \unvalas\tmtwo\esub\var{\unvalas\tmthree}$.
    \end{itemize}
  \end{itemize}
\end{enumerate}
\end{proof}

\begin{remark}
\label{rem:unfolding_vL}
If $\valPred\tm$, then $\valPred{\unvalas\tm}$.
\end{remark}

\begin{remark}
\label{rem:abs_valas_abs_valasext}
If $\abs{\unvalas\tm}$, then $\abs{\unvalas[\valas']\tm}$,
with $\dom\valas \subseteq \dom{\valas'}$.
\end{remark}

\begin{lemma}
\label{lem:valasExt_vL_valas_vL}
Let $\valas_1$ and $\valas_2$ be value assignments and $\tm$ be a term.
Then:
\begin{enumerate}
\item
  $\valPred{\unvalas[\valas_1]\tm}$
  if and only if 
  $\valPred{\unvalas[\valas_2]\tm}$.
\item
  Let $\val$ be a value.
  If $\tm \in \Struct\sset$ and ($\unvalas[\valas_1]\var =
  \unvalas[\valas_2]\var$ for all $\var \in \sset$), then there
  exists $\sctx_1$ such that $\unvalas[\valas_1]\tm = \val\sctx_1$ 
  if and only if there exists $\sctx_2$ such that
  $\unvalas[\valas_2]\tm = \val\sctx_2$.
\end{enumerate}
\end{lemma}
% Label: lem:valasExt_vL_valas_vL

\begin{proof}
We prove the two statements simultaneously.
\begin{enumerate}
\item
  We reason by induction on $\tm$, showing only the left-to-right 
  implication since the other is similar.
  \begin{itemize}
  \item $\tm = \var$.
    There are two cases, depending on the form of $\unvalas[\valas_2]\var$.
    If $\unvalas[\valas_2]\var = \valas_2(\var)$, then we conclude 
    $\valPred{\valas_2(\var)}$ by definition.
    Otherwise, $\unvalas[\valas_2]\var = \var$, and we conclude since 
    $\valPred\var$.
  \item $\tm = \lam\var\tmtwo$.
    Immediate, since 
    $\unvalas[\valas_2]{(\lam\var\tmtwo)} = \lam\var\tmtwo$ by definition.
  \item $\tm = \tmtwo \, \tmthree$.
    This case is not possible since there are no $\val_1, \sctx_1$ 
    such that $\unvalas[\valas_1]{(\tm \, \tmtwo)} 
    = \unvalas[\valas_1]\tmtwo \, \unvalas[\valas_1]\tmthree
    = \val_1\sctx_1$.
  \item 
    $\tm = \tmtwo\esub\var\tmthree$.
    We analyse two cases, depending on the form of 
    $\unvalas[\valas_1]{\tmtwo\esub\var\tmthree}$.
    \begin{itemize}
    \item
      If $\unvalas[\valas_1]\tmthree = \val_a\sctx_a$ and 
      $\var \in \rv\tmtwo$, then 
      $\unvalas[\valas_1]{\tmtwo\esub\var\tmthree} = 
      \unvalas[\valas_1 \cup (\var \mapsto \val_a)]\tmtwo\esub\var{\val_a}\sctx_a$.
      Thus by hypothesis $\valPred{\unvalas[\valas_1 \cup (\var \mapsto \val_a)]\tmtwo}$.
      By the \ih (1) on $\tmthree$, there exist $\val_b$ and $\sctx_b$ 
      such that $\unvalas[\valas_2]\tmthree = \val_b\sctx_b$.
      We are then in the case where
      $\unvalas[\valas_2]{\tmtwo\esub\var\tmthree}
      = \unvalas[\valas_2 \cup (\var \mapsto \val_b)]\tmtwo\esub\var{\val_b}\sctx_b$.
      We can apply the \ih (1) on $\tmtwo$, yielding 
      $\valPred{\unvalas[\valas_2 \cup (\var \mapsto \val_b)]\tmtwo}$.
      Therefore, 
      $\valPred{\unvalas[\valas_2]{\tmtwo\esub\var\tmthree}}$.
    \item
      Otherwise, $\unvalas[\valas_1]{\tmtwo\esub\var\tmthree} 
      = \unvalas[\valas_1]\tmtwo\esub\var{\unvalas[\valas_1]\tmthree}$.
      Thus by hypothesis $\valPred{\unvalas[\valas_1]\tmtwo}$.
      If $\var \notin \rv\tmtwo$, then we are in the case $\unvalas[\valas_2]{\tmtwo\esub\var\tmthree} 
      =\unvalas[\valas_2]\tmtwo\esub\var{\unvalas[\valas_2]\tmthree}$.
      If $\neg(\valPred{\unvalas[\valas_1]\tmthree})$, then the 
      \ih (1) on $\tmthree$ states $\neg(\valPred{\unvalas[\valas_2]\tmthree})$, 
      and so we are also in the case $\unvalas[\valas_2]{\tmtwo\esub\var\tmthree}
      = \unvalas[\valas_2]\tmtwo\esub\var{\unvalas[\valas_2]\tmthree}$.
      Hence, we can apply the \ih (1) on $\tmtwo$, yielding 
      $\valPred{\unvalas[\valas_2]\tmtwo}$.
      Therefore, $\valPred{\unvalas[\valas_2]{\tmtwo\esub\var\tmthree}}$.
    \end{itemize}
  \end{itemize}
\item
  By induction on the derivation of the judgement $\tm \in \Struct\sset$, 
  showing only the left-to-right implication since the other is similar.
  \begin{itemize}
  \item \ruleStructVar.
    Then, $\tm = \vartwo \in \Struct\sset$, with $\vartwo \in \sset$.
    By hypothesis
    $\unvalas[\valas_1]\vartwo = \val\sctx_1$ and
    $\unvalas[\valas_1]\vartwo = \unvalas[\valas_2]\vartwo$.
    Hence we conclude $\unvalas[\valas_2]\vartwo = \val\sctx_2$
    with $\sctx_2 \eqdef \sctx_1$.
  \item \ruleStructApp.
    Then, $\tm = \tmtwo \, \tmthree \in \Struct\sset$, which is 
    derived from $\tmtwo \in \Struct\sset$.
    This case is impossible since there are no $\val$, $\sctx_1$ such 
    that $\unvalas[\valas_1]{(\tmtwo \, \tmthree)} = 
    \unvalas[\valas_1]\tmtwo \, \unvalas[\valas_1]\tmthree = \val\sctx_1$.
  \item \ruleStructSubi.
    Then
    \[
      \inferrule{
        \tmtwo \in \Struct\sset
        \sep
        \vartwo \notin \sset
      }{
        \tm = \tmtwo\esub\vartwo\tmthree \in \Struct\sset
      }\ruleStructSubi
    \]
    By hypothesis, $\unvalas[\valas_1]{\tmtwo\esub\vartwo\tmthree}
    = \val\sctx_1$.
    We analyse two cases depending on the form of 
    $\unvalas[\valas_1]{\tmtwo\esub\vartwo\tmthree}$: 
    \begin{itemize}
    \item
      If $\unvalas[\valas_1]\tmthree = \val_a\sctx_a$ and 
      $\vartwo \in \rv\tmtwo$, then 
      $\unvalas[\valas_1]{\tmtwo\esub\vartwo\tmthree} =
      \unvalas[\valas_1 \cup (\vartwo \mapsto \val_a)]\tmtwo\esub\vartwo{\val_a}\sctx_a$.
      Thus, by hypothesis there exists $\sctx_c$ such that
      $\unvalas[\valas_1 \cup (\vartwo \mapsto \val_a)]\tmtwo = \val\sctx_c$,
      where $\sctx_1 = \sctx_c\esub\vartwo{\val_a}\sctx_a$.
      By the \ih (1) on $\tmthree$, there exist $\val_b$, $\sctx_b$ 
      such that $\unvalas[\valas_2]\tmthree = \val_b\sctx_b$.
      Hence, we are in the case 
      $\unvalas[\valas_2]{\tmtwo\esub\vartwo\tmthree} = 
      \unvalas[\valas_2 \cup (\vartwo \mapsto \val_b)]\tmtwo\esub\vartwo{\val_b}\sctx_b$.
      We are still in the case
      $\unvalas[\valas_1 \cup (\vartwo \mapsto \val_a)]\var 
      = \unvalas[\valas_2 \cup (\vartwo \mapsto \val_b)]\var$
      for all $\var \in \sset$, since $\vartwo \notin \sset$.
      We can then apply the \ih (2) on $\tmtwo$, yielding $\sctx_d$ 
      such that $\unvalas[\valas_2 \cup (\vartwo \mapsto \val_b)]\tmtwo 
      = \val\sctx_d$.
      Therefore, 
      $\unvalas[\valas_2]{\tmtwo\esub\vartwo\tmthree} = \val\sctx_2$,
      where $\sctx_2 =\sctx_d\esub\vartwo{\val_b}\sctx_b$.
    \item
      Otherwise,
      $\unvalas[\valas_1]{\tmtwo\esub\vartwo\tmthree} =
      \unvalas[\valas_1]\tmtwo\esub\vartwo{\unvalas[\valas_1]\tmthree}$.
      By hypothesis, there exists $\sctx_c$ such that
      $\unvalas[\valas_1]\tmtwo = \val\sctx_c$, where 
      $\sctx_{1} =  \sctx_c\esub\vartwo{\unvalas[\valas_1]\tmthree}$.
      If $\var \notin \rv\tmtwo$, then 
      $\unvalas[\valas_2]{\tmtwo\esub\var\tmthree} = 
      \unvalas[\valas_2]\tmtwo\esub\var{\unvalas[\valas_2]\tmthree}$.
      If $\neg(\valPred{\unvalas[\valas_1]\tmthree})$, then the 
      \ih (1) on $\tmthree$ states $\neg(\valPred{\unvalas[\valas_2]\tmthree})$,
      and so we are also in the case 
      $\unvalas[\valas_2]{\tmtwo\esub\var\tmthree} = 
      \unvalas[\valas_2]\tmtwo\esub\var{\unvalas[\valas_2]\tmthree}$.
      Hence, we can apply the \ih (2) on $\tmtwo$, yielding $\sctx_d$ 
      such that $\unvalas[\valas_2]\tmtwo = \val\sctx_d$.
      Therefore, 
      $\unvalas[\valas_2]{\tmtwo\esub\vartwo\tmthree} = \val\sctx_2$, 
      where $\sctx_2 = \sctx_d\esub\vartwo{\unvalas[\valas_2]\tmthree}$.
    \end{itemize}
  \item \ruleStructSubii.
    Then
    \[
      \inferrule{
        \tmtwo \in \Struct{\sset \cup \set\vartwo}
        \sep
        \vartwo \notin \sset
        \sep
        \tmthree \in \Struct\sset
      }{
        \tm = \tmtwo\esub\vartwo\tmthree \in \Struct\sset
      }\ruleStructSubii
    \]
    By hypothesis, there exists $\sctx_1$ such that 
    $\unvalas[\valas_1]{\tmtwo\esub\vartwo\tmthree} = \val\sctx_1$.
    Depending on the form of 
    $\unvalas[\valas_1]{\tmtwo\esub\vartwo\tmthree}$, we have two 
    cases to analyse:
    \begin{itemize}
    \item
      If $\unvalas[\valas_1]\tmthree = \val_a\sctx_a$ and 
      $\var \in \rv\tmtwo$, then $\val\sctx_1 
      = \unvalas[\valas_1]{\tmtwo\esub\vartwo\tmthree} 
      = \unvalas[\valas_1 \cup (\vartwo \mapsto \val_a)]\tmtwo\esub\vartwo{\val_a}\sctx_a$.
      Thus by hypothesis it must be the case that 
      $\unvalas[\valas_1 \cup (\vartwo \mapsto \val_a)]\tmtwo = \val\sctx_c$,
      with $\sctx_1 = \sctx_c\esub\vartwo{\val_a}\sctx_a$, for some $\sctx_c$.
      By the \ih (2) on $\tmthree$, there exists $\sctx_b$ such that 
      $\unvalas[\valas_2]\tmthree = \val_a\sctx_b$.
      Moreover, $\var \in \rv\tmtwo$, so 
      $\unvalas[\valas_2]{\tmtwo\esub\vartwo\tmthree} 
      = \unvalas[\valas_2 \cup (\vartwo \mapsto \val_a)]\tmtwo\esub\vartwo{\val_a}\sctx_b$.
      On the other hand, 
      $\unvalas[\valas_1 \cup (\vartwo \mapsto \val_a)]\var = 
      \unvalas[\valas_2 \cup (\vartwo \mapsto \val_a)]\var$ for all 
      $\var \in \sset \cup \set\vartwo$.
      We then apply the \ih (2) on $\tmtwo$, yielding $\sctx_{c'}$ 
      such that 
      $\unvalas[\valas_2 \cup (\vartwo \mapsto \val_a)]\tmtwo = \val\sctx_{c'}$.
      Hence, we conclude
      $\unvalas[\valas_2]\tm 
      =\unvalas[\valas_2 \cup (\vartwo \mapsto \val_a)]\tmtwo\sctx_{c'}\esub\vartwo{\val_a}\sctx_b$,
      where $\sctx_2 = \sctx_{c'}\esub\vartwo{\val_a}\sctx_b$.
    \item
      Otherwise, $\val\sctx_1 
      = \unvalas[\valas_1]{\tmtwo\esub\vartwo\tmthree} 
      = \unvalas[\valas_1]\tmtwo\esub\vartwo{\unvalas[\valas_1]\tmthree}$.
      Thus, by hypothesis it must be the case that 
      $\unvalas[\valas_1]\tmtwo = \val\sctx_c$, with 
      $\sctx_1 = \sctx_c\esub\vartwo{\unvalas[\valas_1]\tmthree}$, 
      for some $\sctx_c$.
      If $\var \notin \rv\tmtwo$, then we are in the case where 
      $\unvalas[\valas_2]{\tmtwo\esub\var\tmthree} = 
      \unvalas[\valas_2]\tmtwo\esub\var{\unvalas[\valas_2]\tmthree}$.
      If $\neg(\valPred{\unvalas[\valas_1]\tmthree})$, then the 
      \ih (1) on $\tmthree$ states 
      $\neg(\valPred{\unvalas[\valas_2]\tmthree})$, and so we are also 
      in the case $\unvalas[\valas_2]{\tmtwo\esub\var\tmthree} = 
      \unvalas[\valas_2]\tmtwo\esub\var{\unvalas[\valas_2]\tmthree}$.
      We then apply the \ih (2) on $\tmtwo$, yielding $\sctx_{c'}$ 
      such that $\unvalas[\valas_2]\tmtwo = \val\sctx_{c'}$.
      Hence, we conclude
      $\unvalas[\valas_2]\tm 
      =\unvalas[\valas_2]\tmtwo\sctx_{c'}\esub\vartwo{\unvalas[\valas_2]\tmthree}$,
      where $\sctx_2 = \sctx_{c'}\esub\vartwo{\unvalas[\valas_2]\tmthree}$.
    \end{itemize}
  \end{itemize}
\end{enumerate}
\end{proof}

\begin{definition}[Compatibility]
Given a value assigment $\valas$, we say $\valas$ is 
\textbf{compatible} with the sets of variables $\aset$ and $\sset$,
written $\compat\valas\aset\sset$ if the following three conditions hold:
\begin{enumerate}
\item
  The value assignment must affect all variables in $\aset$ and some 
  of the variables in $\sset$ 
  (\ie, $\aset \subseteq \dom\valas \subseteq \aset \cup \sset$).
\item
  Variables in $\aset$ must be mapped to abstractions
  (\ie, $\abs{\valas(\var)}$ must hold for every $\var \in \aset$).
\item
  Variables in $\sset$ affected by $\valas$ must be mapped to variables
  (\ie, $\valas(\var)$ must be a variable for every 
  $\var \in \sset \cap \dom{\valas}$).
\end{enumerate}
\end{definition}

\begin{lemma}
\label{lem:compatible_unfold}
Let $\aset$ and $\sset$ be sets of variables, $\tm$ be any term, and 
$\valas$ be a value assignment such that $\inv\aset\sset\tm$ and 
$\compat\valas\aset\sset$ hold.
Then:
\begin{enumerate}
\item 
  If $\tm \in \HAbs\aset$, then $\abs{\unvalas\tm}$.
\item 
  If $\tm \in \Struct\sset$ and $\unvalas\tm = \val\sctx$, then 
  $\val$ is a variable.
\end{enumerate}
\end{lemma}

We prove each item independently.

\begin{enumerate}
\item By induction on the derivation of $\tm \in \HAbs\aset$.
  \begin{itemize}
  \item \ruleHAbsVar.
    Then, $\tm = \var \in \HAbs\aset$ with $\var \in \aset$. 
    Given $\compat\valas\aset\sset$, then $\unvalas\var = \valas(\var)$
    and $\abs{\valas(\var)}$.
  \item \ruleHAbsLam.
    Then, $\tm = \lam\var\tmtwo \in \HAbs\aset$. 
    We are done, since $\unvalas{(\lam\var\tmtwo)} = \lam\var\tmtwo$, 
    and $\abs{\lam\var\tmtwo}$ holds.
  \item \ruleHAbsSubi.
    Then, $\tm = \tmtwo\esub\var\tmthree \in \HAbs\aset$ which is 
    derived from $\tmtwo \in \HAbs\aset$ and $\var \notin \aset$.
    Since $\inv\aset\sset{\tmtwo\esub\var\tmthree}$ holds, then it
    implies in particular $\inv\aset{\sset \cup \set\var}\tmtwo$.
    We proceed by showing that 
    $\compat\valas\aset{\sset \cup \set\var}$ holds:
    \begin{itemize}
    \item[a.]
      Since $\aset \subseteq \dom\valas \subseteq \aset \cup \sset$
      by the hypothesis $\compat\valas\aset\sset$, then 
      $\aset \subseteq \dom\valas \subseteq \aset \cup (\sset \cup \set\var)$.
    \item[b.]
      Let $\vartwo \in \aset$. 
      The premise $\var \notin \aset$ implies $\var \neq \vartwo$. 
      Thus, $\abs{\valas(\vartwo)}$ by the hypothesis $\compat\valas\aset\sset$.
    \item[c.]
      Let $\vartwo \in (\sset \cup \set\var) \cap \dom\valas$. 
      Since $\var \notin \dom\valas$ by $\alpha$-conversion, then 
      $\vartwo \neq \var$. 
      Thus, $\valas(\vartwo)$ is a variable by the hypothesis 
      $\compat\valas\aset\sset$.
    \end{itemize}
    We can now apply the \ih (1) on $\tmtwo$, yielding 
    $\abs{\unvalas\tmtwo}$. We analyse two cases:
    \begin{itemize}
    \item
      If $\unvalas\tmthree = \val\sctx$ and $\var \in \rv\tmtwo$, 
      then $\unvalas{\tmtwo\esub\var\tmthree} = 
      \unvalas[\valas \cup (\var \mapsto \val)]\tmtwo\esub\var\val\sctx$.
      Since $\abs{\unvalas\tmtwo}$ by the \ih (1) on $\tmtwo$, then 
      $\abs{\unvalas[\valas \cup (\var \mapsto \val)]\tmtwo}$ by 
      \cref{rem:abs_valas_abs_valasext}.
      Hence, we conclude with 
      $\abs{\unvalas[\valas \cup (\var \mapsto \val)]\tmtwo\esub\var\val\sctx}$.
    \item
      Otherwise, $\unvalas{\tmtwo\esub\var\tmthree} =
      \unvalas\tmtwo\esub\var{\unvalas\tmthree}$.
      Since $\abs{\unvalas\tmtwo}$ by the \ih (1) on $\tmtwo$, then 
      $\abs{\unvalas\tmtwo\esub\var{\unvalas\tmthree}}$.
    \end{itemize}
  \item \ruleHAbsSubii.
    Then, $\tm = \tmtwo\esub\var\tmthree \in \HAbs\aset$ which is 
    derived from $\tmtwo \in \HAbs{\aset \cup \set\var}$, 
    $\var \notin \aset$, and $\tmthree \in \HAbs\aset$.
    Since $\inv\aset\sset{\tmtwo\esub\var\tmthree}$ holds, then it
    implies both $\inv{\aset \cup \set\var}\sset\tmtwo$ and 
    $\inv\aset\sset\tmthree$.    

    Applying the \ih (1) on $\tmthree$, we yield $\abs{\unvalas\tmthree}$.
    In particular, $\unvalas\tmthree = \val\sctx$.
    We now prove that 
    $\compat{\valas \cup (\var \mapsto \val)}{\aset \cup \set\var}\sset$
    holds:
    \begin{itemize}
    \item[a.]
      $\aset \cup \set\var 
      \subseteq \dom{\valas \cup (\var \mapsto \val)} 
      \subseteq (\aset \cup \set\var) \cup \sset$
      by the hypothesis $\compat\valas\aset\sset$
    \item[b.]
      Let $\vartwo \in (\aset \cup \set\var)$.
      If $\vartwo = \var$, then 
      $(\valas \cup (\var \mapsto \val))(\var) = \val$, and $\abs\val$
      by the \ih (1) on $\tmthree$.
      If $\vartwo \neq \var$, then 
      $(\valas \cup (\var \mapsto \val))(\vartwo) = \valas(\vartwo)$, 
      and $\abs{\valas(\vartwo)}$ by the hypothesis 
      $\compat\valas\aset\sset$.
    \item[c.]
      Let $\vartwo \in (\sset \cap \dom{\valas \cup (\var \mapsto \val)})$.
      Since $\inv{\aset \cup \set\var}\sset\tmtwo$, then $\var \notin \sset$
      and thus $\vartwo \neq \var$.
      We conclude that $(\valas \cup (\var \mapsto \val))(\vartwo) =
      \valas(\vartwo)$ is a variable by the hypothesis 
      $\compat\valas\aset\sset$.
    \end{itemize}

    We analyse two cases, depending on whether $\var \in \rv\tmtwo$ or not:
    \begin{itemize}
    \item
      If $\var \in \rv\tmtwo$, then 
      $\unvalas{\tmtwo\esub\var\tmthree} = 
      \unvalas[\valas \cup (\var \mapsto \val)]\tmtwo\esub\var\val\sctx$.
      Then, 
      $\abs{\unvalas[\valas \cup (\var \mapsto \val)]\tmtwo}$ by the 
      \ih (1) on $\tmtwo$.
      Therefore, $\abs{\unvalas[\valas \cup (\var \mapsto \val)]\tmtwo\esub\var\val\sctx}$.
    \item
      Otherwise, $\unvalas{\tmtwo\esub\var\tmthree} = 
      \unvalas\tmtwo\esub\var{\unvalas\tmthree}$.
      And since $\var \notin \rv\tmtwo$, then
      $\unvalas\tmtwo = \unvalas[\valas \cup (\var \mapsto \val)]\tmtwo$.
      Similarly to the previous case, we obtain 
      $\abs{\unvalas[\valas \cup (\var \mapsto \val)]\tmtwo}$ by the 
      \ih (1) on $\tmtwo$, thus having $\abs{\unvalas\tmtwo}$, so 
      that we conclude with
      $\abs{\unvalas\tmtwo\esub\var{\unvalas\tmthree}}$.
    \end{itemize}
  \end{itemize}
\item By induction on the derivation of $\tm \in \Struct\sset$.
  \begin{itemize}
  \item \ruleStructVar.
    Then, $\tm = \var \in \Struct\sset$ with $\var \in \sset$.
    We analyse two cases depending on whether $\var \in \dom\valas$ or not:
    \begin{itemize}
    \item 
      If $\var \in \dom\valas$, then $\unvalas\var = \valas(\var)$.
      Since $\compat\valas\aset\sset$, then $\valas(\var)$ is a variable.
    \item 
      Otherwise, $\unvalas\var = \var$, so we are done.
    \end{itemize}
  \item \ruleStructApp.
    Then, $\tm = \tmtwo \, \tmthree \in \Struct\sset$.
    This case is not possible, as $\unvalas{(\tmtwo \, \tmthree)} = 
    \unvalas\tmtwo \, \unvalas\tmthree$ is not of the form $\val\sctx$.
  \item \ruleStructSubi.
    Then, $\tm = \tmtwo\esub\var\tmthree \in \Struct\sset$ which is 
    derived from $\tmtwo \in \Struct\sset$ and $\var \notin \sset$.
    We analyse two cases:
    \begin{itemize}
    \item
      If $\unvalas\tmthree = \val_{\tmthree}\sctx_\tmthree$ and 
      $\var \in \rv\tmtwo$, then $\unvalas{\tmtwo\esub\var\tmthree} =
      \unvalas[\valas \cup (\var \mapsto \val_\tmthree)]\tmtwo\esub\var{\val_\tmthree}\sctx_\tmthree$.
      Since $\unvalas\tm$ has the form $\val\sctx$, the same holds for 
      $\unvalas[\valas \cup (\var \mapsto \val_\tmthree)]\tmtwo$.
      That is, $\unvalas[\valas \cup (\var \mapsto \val_\tmthree)]\tmtwo 
      = \val\sctxtwo$.
      
      Since $\var \notin \sset$ by hypothesis, we have 
      $\unvalas[\valas \cup (\var \mapsto \val_\tmthree)]\vartwo = 
      \unvalas\vartwo$ for all $\vartwo \in \sset$.
      We can then apply \cref{lem:valasExt_vL_valas_vL} (2), yielding
      $\sctx_1$ such that $\unvalas\tmtwo = \val\sctx_1$.
      Applying the \ih (2) on $\tmtwo$, we conclude that $\val$ is a 
      variable.
    \item
      Otherwise, $\unvalas{\tmtwo\esub\var\tmthree} =
      \unvalas\tmtwo\esub\var{\unvalas\tmthree}$.
      Since $\unvalas\tm$ is of the form $\val\sctx$, the same holds
      for $\unvalas\tmtwo$.
      That is, $\unvalas\tmtwo = \val\sctxtwo$.
      Following the same reasoning from case \ruleHAbsSubi in the 
      previous item, $\compat\valas\aset{\sset \cup \set\var}$ holds.
      Moreover, $\tmtwo \in \Struct{\sset \cup \set\var}$ by 
      \cref{rem:habs_st}.
      Since $\inv\aset\sset{\tmtwo\esub\var\tmthree}$ implies in 
      particular $\inv\aset{\sset \cup \set\var}\tmtwo$, we can then 
      apply the \ih (2) on $\tmtwo$, yielding that $\val$ is a variable.
    \end{itemize}
  \item \ruleStructSubii.
    Then, $\tm = \tmtwo\esub\var\tmthree \in \Struct\sset$ which is 
    derived from $\tmtwo \in \Struct{\sset \cup \set\var}$,
    $\var \notin \sset$, and $\tmthree \in \Struct\sset$.
    Since $\inv\aset\sset{\tmtwo\esub\var\tmthree}$ holds, then it
    implies in particular $\inv\aset{\sset \cup \set\var}\tmtwo$.
    We analyse two cases, depending on the form 
    $\unvalas{\tmtwo\esub\var\tmthree}$ has.
    \begin{itemize}
    \item 
      If $\unvalas\tmthree = \val_{\tmthree}\sctx_\tmthree$ and
      $\var \in \rv\tmtwo$, then $\unvalas{\tmtwo\esub\var\tmthree} = 
      \unvalas[\valas \cup (\var \mapsto \val_\tmthree)]\tmtwo\esub\var{\val_\tmthree}\sctx_\tmthree$.
      Moreover, since $\unvalas\tm = \val\sctx$, then 
      $\unvalas[\valas \cup (\var \mapsto \val_\tmthree)]\tmtwo$ is
      of the form $\val\sctxtwo$.
      Since $\inv\aset\sset{\tmtwo\esub\var\tmthree}$ holds, then it
      implies in particular $\inv\aset\sset\tmthree$.
      Given that $\tmthree \in \Struct\sset$ and 
      $\unvalas\tmthree = \val_{\tmthree}\sctx_\tmthree$, then 
      $\val_\tmthree$ is a variable by the \ih (2) on $\tmthree$.
      We now prove that 
      $\compat{\valas \cup (\var \mapsto \val_\tmthree)}\aset{\sset \cup \set\var}$
      holds.
      \begin{itemize}
      \item[a.]
        $\aset
        \subseteq \dom{\valas \cup (\var \mapsto \val_\tmthree)} 
        \subseteq \aset \cup (\sset \cup \set\var)$
        by the hypothesis $\compat\valas\aset\sset$
      \item[b.]
        Let $\vartwo \in \aset$.
        Then, $\vartwo \neq \var$ since 
        $\inv\aset{\sset \cup \set\var}\tmtwo$. 
        Thus, $(\valas \cup (\var \mapsto \val_\tmthree))(\vartwo) = \valas(\vartwo)$, 
        and $\abs{\valas(\vartwo)}$ by the hypothesis 
        $\compat\valas\aset\sset$.
      \item[c.]
        Let $\vartwo \in 
        ((\sset \cup \set\var) \cap \dom{\valas \cup (\var \mapsto \val_\tmthree)})$.
        If $\vartwo = \var$, then 
        $(\valas \cup (\var \mapsto \val_\tmthree))(\var) = \val_\tmthree$
        which is a variable by the \ih (2) on $\tmthree$.
        If $\vartwo \neq \var$, then 
        $(\valas \cup (\var \mapsto \val))(\vartwo) = \valas(\vartwo)$,
        which is a variable by the hypothesis $\compat\valas\aset\sset$.
      \end{itemize}
      We can then apply the \ih on $\tmtwo$, yielding that $\val$ is
      a variable.
    \item 
      Otherwise, $\unvalas{\tmtwo\esub\var\tmthree} = 
      \unvalas\tmtwo\esub\var{\unvalas\tmthree}$.
      Moreover, since $\unvalas\tm = \val\sctx$, then 
      $\unvalas\tmtwo$ is of the form $\val\sctxtwo$.
      Since $\inv\aset\sset{\tmtwo\esub\var\tmthree}$ holds, then it
      implies in particular $\inv\aset\sset\tmthree$.
      Given that $\tmthree \in \Struct\sset$ and 
      $\unvalas\tmthree = \val_{\tmthree}\sctx_\tmthree$, then 
      $\val_\tmthree$ is a variable by the \ih (2) on $\tmthree$.
      We now prove that 
      $\compat\valas\aset{\sset \cup \set\var}$
      holds.
      \begin{itemize}
      \item[a.]
        $\aset \subseteq \dom\valas 
        \subseteq \aset \cup (\sset \cup \set\var)$
        by the hypothesis $\compat\valas\aset\sset$
      \item[b.]
        Let $\vartwo \in \aset$.
        Then, $\vartwo \neq \var$ since 
        $\inv\aset{\sset \cup \set\var}\tmtwo$. 
        Thus $\abs{\valas(\vartwo)}$ by the hypothesis 
        $\compat\valas\aset\sset$.
      \item[c.]
        Let $\vartwo \in 
        ((\sset \cup \set\var) \cap \dom\valas)$.
        Then, $\var \notin \dom\valas$ by $\alpha$-conversion, so
        $\vartwo \neq \var$.
        Thus $\valas(\vartwo)$ is a variable by the hypothesis 
        $\compat\valas\aset\sset$.
      \end{itemize}
      We can then apply the \ih on $\tmtwo$, yielding that $\val$ is
      a variable.
    \end{itemize}
  \end{itemize}
\end{enumerate}

\begin{lemma}
\label{lem:esub_valas_compat_ext_compat}
Let $\tm\esub\var\tmtwo$ be any closure, $\aset$ and $\sset$ be sets
of variables, and $\valas$ be any value assignment such that
$\inv\aset\sset{\tm\esub\var\tmtwo}$ and $\compat\valas\aset\sset$ hold.
Then:
\begin{enumerate}
\item 
  If $\tmtwo \in \HAbs\aset$ and $\unvalas\tmtwo = \val\sctx$, then 
  $\compat{\valas \cup (\var \mapsto \val)}{\aset \cup \set\var}\sset$.
\item
  If $\tmtwo \in \Struct\sset$, then 
  $\compat{\valas'}\aset{\sset \cup \set\var}$, where 
  $\valas' = \valas \cup (\var \mapsto \val)$ if $\var \in \rv\tm$ 
  and $\unvalas{\tmtwo}$ is of the form $\val\sctx$ for some $\val$ 
  and $\sctx$, or $\valas' = \valas$ otherwise.
\end{enumerate}
\end{lemma}
% Label: lem:esub_valas_compat_ext_compat

\begin{proof}
Since $\inv\aset\sset{\tm\esub\var\tmtwo}$ holds, then
$\inv{\aset \cup \set\var}\sset\tm$,
$\inv\aset{\sset \cup \set\var}\tm$, and $\inv\aset\sset\tmtwo$.
We prove each item independently.
\begin{enumerate}
\item
  We show that
  $\compat{\valas \cup (\var \mapsto \val)}{\aset \cup \set\var}\sset$ 
  holds by checking the conditions:
  \begin{enumerate}
  \item[a.]
    Since $\aset \subseteq \dom\valas \subseteq \aset \cup \sset$ by 
    the hypothesis $\compat\valas\aset\sset$, then
    $\aset \cup \set\var \subseteq \dom\valas \cup \set\var \subseteq 
    (\aset \cup \set\var) \cup \sset$.
  \item[b.]
    Let $\vartwo \in \aset \cup \set\var$. 
    If $\vartwo \neq  \var$, then 
    $(\valas \cup (\var \mapsto \val))(\vartwo) = \valas(\vartwo)$,
    and $\abs{\valas(\vartwo)}$ holds by the hypothesis 
    $\compat\valas\aset\sset$.
    Otherwise, $\vartwo = \var$ and $\abs{\val\sctx}$ by 
    \cref{lem:compatible_unfold} (1) and thus $\abs\val$ holds.
  \item [c.]
    Let $\vartwo \in \sset \cap \dom{\valas \cup (\var \mapsto \val)}$.
    Since $\inv{\aset \cup \set\var}\sset\tm$, then $\var \notin \sset$
    so $\var \neq \vartwo$ and thus the property holds by the 
    hypothesis $\compat\valas\aset\sset$.
  \end{enumerate}
\item
  We show that $\compat{\valas'}\aset{\sset \cup \set\var}$ holds by 
  checking the conditions:
  \begin{enumerate}
  \item[a.]
    Since $\aset \subseteq \dom\valas \subseteq \aset \cup \sset$ by 
    the hypothesis $\compat\valas\aset\sset$, then 
    $\aset \subseteq \dom{\valas'} \subseteq \aset \cup (\sset \cup \set\var)$.
  \item[b.]
    Let $\vartwo \in \aset$.
    Since $\inv\aset{\sset \cup \set\var}\tm$ implies $\var \notin \aset$,
    then $\var \neq \vartwo$ so $\valas'(\vartwo) = \valas(\vartwo)$ and
    thus $\abs{\valas'(\vartwo)}$ by the hypothesis $\compat\valas\aset\sset$.
  \item[c.]
    Let $\vartwo \in (\sset \cup \set\var) \cap \dom{\valas'}$. 
    If $\vartwo \neq \var$, then $\valas'(\vartwo) = \valas(\vartwo)$ 
    and thus $\valas'(\vartwo)$ is a variable by the hypothesis 
    $\compat\valas\aset\sset$.
    If $\vartwo = \var$, there are two subcases, depending on whether 
    $\var \in \rv\tmtwo$.
    If $\var \in \rv\tmtwo$, then $\valas' = \valas \cup (\var \mapsto \val)$,
    so we have that $\valas'(\var) = \val$, and we are in the case 
    $\unvalas\tmtwo = \val\sctx$, where $\val$ is a variable by 
    \cref{lem:compatible_unfold} (2).
    If $\var \notin \rv\tmtwo$, then $\valas' = \valas$, so 
    $\valas'(\var) = \valas(\var)$.
    Note that we can assume $\var \notin \dom\valas$ by 
    $\alpha$-conversion, so $\valas(\var) = \var$.
  \end{enumerate}
\end{enumerate}
\end{proof}

\begin{proposition}
\label{prop:useful_nf_unfold_to_nf}
\label{prop:unfold_reducible_is_reducible}
Let $\tm$ be a term, $\valas$ a value assignment, and $\aset, \sset$ 
sets of variables.
Suppose $\inv\aset\sset\tm$ and  $\compat\valas\aset\sset$ hold.
\begin{enumerate}
\item
  \label{prop:useful_nf_unfold_to_nf:part}
  If $\tm \in \NF\aset\sset\appflag$,
  then $\unvalas\tm \in \VNF{\dom\valas}\appflag$.
\item
  \label{prop:unfold_reducible_is_reducible:part}
  If $\tm \tov\rulename\aset\sset\appflag \tm'$,
  then either there exists a term $\tm''$ such that 
  $\unvalas\tm \tovv\dbsym \tm''$, or 
  $\abs{\unvalas\tm}$ and $\appflag = \app$.
\end{enumerate}
\end{proposition}
%Label: prop:useful_nf_unfold_to_nf

\begin{proof}
We prove each item independently.
\begin{enumerate}
\item
  By induction on the derivation of the judgement 
  $\tm \in \NF\aset\sset\appflag$.
  The interesting cases are \ruleUNFVar, \ruleUNFEsAbs, and 
  \ruleUNFEsStruct.
  \begin{itemize}
  \item \ruleUNFVar.
    Then 
    \[
      \inferrule{
        \var \in \aset \implies \appflag = \nonapp
      }{
        \var \in \NF\aset\sset\appflag
      }
    \]
    Where $\tm = \var$. Moreover, $\inv\aset\sset\var$ implies 
    $\aset \disj \sset$ and
    $\fv\var = \set\var \subseteq \aset \cup \sset$.
    Thus, either $\var \in \aset$ or $\var \in \sset$. 
    We analyse both cases:
    \begin{itemize}
    \item 
      If $\var \in \aset$, then in particular $\appflag = \nonapp$.
      Since $\compat\valas\aset\sset$ holds, $\var \in \dom\valas$ 
      and $\abs{\unvalas\var}$, and moreover $\unvalas\var$ is of the 
      form $\lam{\var_1}{\tm_1}$ by definition of $\valas(\var)$.
      Applying rule \ruleVNFLam, we yield
      $\unvalas\var \in \VNF{\dom\valas}\nonapp$.
    \item 
      If $\var \in \sset$, then $\var \in \dom\valas$ or not by 
      definition of $\compat\valas\aset\sset$.
      If $\var \in \dom\valas$, then $\var \in \dom\valas \cap \sset$,
      and therefore $\unvalas\var = \vartwo$.
      By idempotency of value assignments, 
      $\vartwo \notin \dom\valas$. 
      So we yield $\vartwo \in \VNF{\dom\valas}\appflag$ by rule 
      \ruleVNFVar.
      Otherwise, $\var \notin \dom\valas$, thus
      $\unvalas\var = \var \in \VNF{\dom\valas}\appflag$ by rule 
      \ruleVNFVar.
    \end{itemize}
  \item \ruleUNFEsAbs.
    Then
    \[
      \inferrule{
        (1)\ \tmtwo \in \NF{\aset \cup \set\var}\sset\appflag
        \sep
        (2)\ \tmthree \in \NF\aset\sset\nonapp
        \sep
        (3)\ \tmthree \in \HAbs\aset
      }{
        \tm = \tmtwo\esub\var\tmthree \in \NF\aset\sset\appflag
      }\ruleUNFEsAbs
    \]
    Since (3), then $\valPred\tmthree$ by \cref{rem:habs_st}, and 
    thus (4) $\unvalas\tmthree = \val\sctx$ by \cref{rem:unfolding_vL},
    for some $\val$ and $\sctx$.
    Moreover, $\inv\aset\sset{\tmtwo\esub\var\tmthree}$ implies both
    $\inv{\aset \cup \set\var}\sset\tmtwo$ and
    $\inv\aset\sset\tmthree$.
    Given (3) and (4), the first item of 
    \cref{lem:esub_valas_compat_ext_compat} gives
    $\compat{\valas \cup (\var \mapsto \val)}{\aset \cup \set\var}\sset$.
    Since (2), we can then apply the \ih on $\tmthree$, yielding
    $\unvalas\tmthree = \val\sctx \in \VNF{\dom\valas}\nonapp$.

    Since (1), we can then apply the \ih on $\tmtwo$, yielding 
    (5) $\unvalas[\valas \cup (\var \mapsto \val)]\tmtwo \in 
    \VNF{\dom\valas \cup \set\var}\appflag$.
    We consider two cases depending on whether $\var \in \rv\tmtwo$ 
    or not:
    \begin{itemize}
    \item \label{it:inrv}
      If $\var \in \rv\tmtwo$, then 
      $\unvalas{\tmtwo\esub\var\tmthree} = 
      \unvalas[\valas \cup (\var \mapsto \val)]\tmtwo\esub\var\val\sctx$.
      We obtain
      (6) $\val \in \VNF{\expansion{\dom\valas}\sctx}\nonapp$ and
      $\sctx \in \CtxNF{\dom\valas}$ by \cref{lem:splitting_ctx_nf}.
      We assume 
      $\domSctx\sctx \disj 
      \fv{\unvalas[\valas \cup (\var \mapsto \val)]\tmtwo}$ by 
      $\alpha$-conversion, so 
      $\domSctx\sctx \disj 
      \rv{\unvalas[\valas \cup (\var \mapsto \val)]\tmtwo}$.
      Hence we can apply \cref{lem:vnf_weakening} to (5), yielding
      (7) $\unvalas[\valas \cup (\var \mapsto \val)]\tmtwo \in 
      \VNF{\expansion{\dom\valas}\sctx \cup \set\var}\appflag$.
      Applying rule \ruleVNFEsVal with (7) and (6) as premises, we get
      $\unvalas[\valas \cup (\var \mapsto \val)]\tmtwo\esub\var\val 
      \in \VNF{\expansion{\dom\valas}\sctx}\appflag$.
      We conclude $\unvalas\tm = 
      \unvalas[\valas \cup (\var \mapsto \val)]\tmtwo\esub\var\val\sctx 
      \in \VNF{\dom\valas}\appflag$ by \cref{lem:splitting_ctx_nf}.
    \item \label{it:esLAbsxnotrv}
      If $\var \notin \rv\tmtwo$, then 
      $\unvalas{\tmtwo\esub\var\tmthree} = 
      \unvalas\tmtwo\esub\var{\unvalas\tmthree}$.
      Given $\var \notin \rv\tmtwo$, then 
      $\unvalas[\valas \cup (\var \mapsto \val)]\tmtwo = \unvalas\tmtwo$.
      We conclude $\unvalas\tm = 
      \unvalas\tmtwo\esub\var{\unvalas\tmthree} \in 
      \VNF{\dom\valas}\appflag$ by rule \ruleVNFEsVal.
    \end{itemize}
  \item \ruleUNFEsStruct.
    Then
    \[
      \inferrule{
        (1)\ \tmtwo \in \NF\aset{\sset \cup \set\var}\appflag
        \sep
        (2)\ \tmthree \in \NF\aset\sset\nonapp
        \sep
        (3)\ \tmthree \in \Struct\sset
      }{
        \tmtwo\esub\var\tmthree \in \NF\aset\sset\appflag
      }\ruleUNFEsStruct
    \]
    We analyse two cases based on the form of 
    $\unvalas{\tmtwo\esub\var\tmthree}$:
    \begin{itemize}
    \item
      If $\unvalas\tmthree = \val\sctx$ and $\var \in \rv\tmtwo$, 
      then $\unvalas{\tmtwo\esub\var\tmthree} = 
      \unvalas[\valas \cup (\var \mapsto \val)]\tmtwo\esub\var\val\sctx$,
      and the proof follows analogously to case \ref{it:inrv} in rule 
      \ruleUNFEsAbs.
    \item
      Otherwise, $\unvalas{\tmtwo\esub\var\tmthree} = 
      \unvalas\tmtwo\esub\var{\unvalas\tmthree}$.
      We further distinguish two subcases:
      \begin{itemize}
      \item 
        If $\valPred{\unvalas\tmthree}$, the proof follows analogously 
        ---yet more simply--- than case \ref{it:esLAbsxnotrv} in rule 
        \ruleUNFEsAbs.
      \item
        Otherwise, $\inv\aset\sset{\tmtwo\esub\var\tmthree}$ implies
        both $\inv\aset{\sset \cup \set\var}\tmtwo$ and
        $\inv\aset\sset\tmthree$.
        Since (2), we can then apply the \ih on $\tmthree$, yielding
        (4) $\unvalas\tmthree \in \VNF{\dom\valas}\nonapp$.
        Moreover, given (3), the second item of 
        \cref{lem:esub_valas_compat_ext_compat} implies
        $\compat\valas\aset{\sset \cup \set\var}$.
        Since (1), we can then apply the \ih on $\tmtwo$, yielding
        (5) $\unvalas\tmtwo \in \VNF{\dom\valas}\appflag$.
        Applying rule \ruleVNFEsNonVal with (4) and (5) as premises,
        we obtain $\unvalas\tm = 
        \unvalas\tmtwo\esub\var{\unvalas\tmthree} \in \VNF{\dom\valas}\appflag$.
      \end{itemize}
    \end{itemize}
  \end{itemize}
\item
  % Label: prop:unfold_reducible_is_reducible
  By induction on the derivation of 
  $\tm \tov\rulename\aset\sset\appflag \tm'$.
  \begin{itemize}
  \item \ruleUSub.
    Then, $\tm = \var 
    \tov{\rulesub\var\val}{\aset' \cup \set\var}\sset\app \val = \tm'$,
    where $\rulename = \rulesub\var\val$, 
    $\aset = \aset' \cup \set\var$, and $\appflag = \app$.
    Given $\var \in \aset' \cup \set{\var}$ and
    $\compat\valas{\aset' \cup \set\var}\sset$, then 
    $\abs{\valas(\var)}$ (\ie, $\abs{\unvalas\var}$).
  \item \ruleUDb.
    Then,
    $\tm = (\lam\var\tmtwo)\sctx \, \tmthree
     \tov\ruledb\aset\sset\appflag
     \tmtwo\esub\var\tmthree\sctx = \tm''$,
    where $\rulename = \ruledb$.
    We have $\unvalas{((\lam\var\tmtwo)\sctx \, \tmthree)} = 
    \unvalas{((\lam\var\tmtwo)\sctx)} \, \unvalas\tmthree$.
    Moreover, $(\lam\var\tmtwo)\sctx \in \HAbs\aset$ by 
    \cref{rem:habs_st}, then $\abs{\unvalas{((\lam\var\tmtwo)\sctx)}}$
    by \cref{lem:compatible_unfold} (1),
    so $\unvalas{((\lam\var\tmtwo)\sctx)}$ by definition is of the form 
    $(\lam\var\tmtwo)\sctxtwo$.
    Then, $(\lam\var\tmtwo)\sctxtwo \, \unvalas\tmthree \tovv\ruledb 
    \tmtwo\esub\var{\unvalas\tmthree}\sctxtwo$ by rule \ruleVDb.
  \item \ruleULsv.
    Then
    \[
      \inferrule{
        \tmtwo
        \tov{\rulesub\var\val}{\aset \cup \set\var}\sset\appflag
        \tmtwo'
        \sep
        \var \notin \aset \cup \sset
        \sep
        \val\sctx \in \HAbs\aset
      }{
        \tm = \tmtwo\esub\var{\val\sctx}
        \tov\rulelsv\aset\sset\appflag
        \tmtwo'\esub\var\val\sctx = \tm'
      }\ruleULsv
    \]
    where $\rulename = \rulelsv$.
    Since $\inv\aset\sset{\tmtwo\esub\var{\val\sctx}}$ holds, then it 
    implies in particular $\inv{\aset \cup \set\var}\sset{\val\sctx}$.
    We then have $\abs{\unvalas{(\val\sctx)}}$ by 
    \cref{lem:compatible_unfold} (1), and $\unvalas{(\val\sctx)}$ is 
    of the form $\val'\sctxtwo$ by \cref{rem:unfolding_vL}.
    Moreover, 
    $\compat{\valas \cup (\var \mapsto \val')}{\aset \cup \set\var}\sset$
    by \cref{lem:esub_valas_compat_ext_compat} (1).
    We can then apply the \ih on $\tmtwo$, yielding two possible cases:
    \begin{itemize}
    \item 
      There exists $\tmtwo''$ such that
      $\unvalas[\valas \cup (\var \mapsto \val')]\tmtwo \tovv\ruledb \tmtwo''$.
      We analyse two subcases, depending on whether 
      $\var \in \rv\tmtwo$ or not:
      \begin{itemize}
      \item 
        If $\var \in \rv\tmtwo$, then
        $\unvalas{\tmtwo\esub\var{\val\sctx}} =
        \unvalas[\valas \cup (\var \mapsto \val')]\tmtwo\esub\var{\val'}\sctxtwo 
        \tovv\ruledb \tmtwo''\esub\var{\val'}\sctxtwo$
        by applying (length of $\sctxtwo + 1$) times rule \ruleVEsL,
        so we are done.
      \item 
        If $\var \notin \rv\tmtwo$, then
        $\unvalas{\tmtwo\esub\var{\unvalas\tmthree}} = 
        \unvalas\tmtwo\esub\var{\unvalas\tmthree} =
        \unvalas[\valas \cup (\var \mapsto \val')]\tmtwo\esub\var{\unvalas{(\val\sctx)}}
        \tovv\ruledb \tmtwo''\esub\var{\unvalas\tmthree} = \tm''$
        by rule \ruleVEsL.
      \end{itemize}
    \item 
      $\abs{\unvalas[\valas \cup (\var \mapsto \val')]\tmtwo}$ and 
      $\appflag = \app$.
      Then, $\abs{\unvalas{\tmtwo\esub\var{\val\sctx}}}$, and 
      $\appflag = \app$, so we are done.
    \end{itemize}
  \item \ruleUAppL.
    Then
    \[
      \inferrule{
        \tmtwo \tov\rulename\aset\sset\app \tmtwo'
      }{
        \tm = \tmtwo \, \tmthree \tov\rulename\aset\sset\appflag 
        \tmtwo' \, \tmthree = \tm'
      }\ruleUAppL
    \]

    Since $\inv\aset\sset{\tmtwo \, \tmthree}$ holds, then it implies 
    in particular $\inv\aset\sset\tmtwo$.
    We analyse two cases by the \ih on $\tmtwo$:
    \begin{itemize}
    \item
      There exists $\tmtwo''$ such that 
      $\unvalas\tmtwo \tovv\ruledb \tmtwo''$.
      Then,
      $\unvalas\tmtwo \, \unvalas\tmthree \tovv\ruledb 
      \tmtwo'' \, \unvalas\tmthree = \tm''$ by rule \ruleVAppL.
    \item
      $\abs{\unvalas\tmtwo}$ and $\appflag = \app$.
      Then, $\unvalas\tmtwo$ is of the form $(\lam\var\tmfour)\sctx$,
      thus $(\lam\var\tmfour)\sctx \, \unvalas\tmthree \tovv\ruledb 
      \tmfour\esub\var{\unvalas\tmthree}\sctx = \tm''$ by rule 
      \ruleVDb.
    \end{itemize}
  \item \ruleUAppR.
    Then
    \[
      \inferrule{
        \tmtwo \in \Struct\sset
        \sep
        \tmthree \tov\rulename\aset\sset\nonapp \tmthree'
      }{
        \tm = \tmtwo \, \tmthree \tov\rulename\aset\sset\appflag 
        \tmtwo \, \tmthree' = \tm'
      }\ruleUAppR
    \]

    Since $\inv\aset\sset{\tmtwo \, \tmthree}$ holds, then it implies 
    in particular $\inv\aset\sset\tmthree$.
    We analyse two cases by the \ih on $\tmthree$:
    \begin{itemize}
    \item
      There exists $\tmthree''$ such that 
      $\unvalas\tmthree \tovv\ruledb \tmthree''$.
      Then,
      $\unvalas\tmtwo \, \unvalas\tmthree \tovv\ruledb 
      \unvalas\tmtwo \, \tmthree'' = \tm''$ by rule \ruleVAppR.
    \item
      $\abs{\unvalas\tmthree}$ and $\appflag = \app$.
      This case is impossible since $\appflag = \nonapp$ by premise 
      of rule \ruleUAppR.
    \end{itemize}
  \item \ruleUEsR.
    Then
    \[
      \inferrule{
        \tmthree \tov\rulename\aset\sset\nonapp \tmthree'
      }{
        \tm = \tmtwo\esub\var\tmthree \tov\rulename\aset\sset\appflag 
        \tmtwo\esub\var{\tmthree'} = \tm'
      }\ruleUEsR
    \]

    Since $\inv\aset\sset{\tmtwo\esub\var\tmthree}$ holds, then it 
    implies in particular $\inv\aset\sset\tmthree$.
    We analyse two cases by the \ih on $\tmthree$:
    \begin{itemize}
    \item
      There exists $\tmthree''$ such that 
      $\unvalas\tmthree \tovv\ruledb \tmthree''$.
      By definition of $\unvalas{\tmtwo\esub\var\tmthree}$, there 
      are two subcases:
      \begin{itemize}
      \item 
        If $\unvalas\tmthree = \val\sctx$ and $\var \in \rv\tmtwo$,
        then $\unvalas{\tmtwo\esub\var\tmthree} = 
        \unvalas[\valas \cup (\var \mapsto \val)]\tmtwo\esub\var\val\sctx$.
        Since $\val\sctx$ must have the form 
        $\val\sctx_1\esub\vartwo\tmfour\sctx_2$ with 
        $\tmfour \tovv\ruledb \tmfour'$.
        Then, 
        $\unvalas[\valas \cup (\var \mapsto \val)]\tmtwo\esub\var\val\sctx_1\esub\vartwo\tmfour\sctx_2
        \tovv\ruledb
        \unvalas[\valas \cup (\var \mapsto \val)]\tmtwo\esub\var\val\sctx_1\esub\vartwo{\tmfour'}\sctx_2 = \tm''$
        by applying rules \ruleVEsL and \ruleVEsR as needed.
      \item 
        Otherwise, $\unvalas{\tmtwo\esub\var\tmthree} = 
        \unvalas\tmtwo\esub\var{\unvalas\tmthree}$.
        Applying rule \ruleVEsR, we yield
        $\unvalas\tmtwo\esub\var{\unvalas\tmthree} \tovv\ruledb
        \unvalas\tmtwo\esub\var{\tmthree''} = \tm''$.
      \end{itemize}
    \item
      $\abs{\unvalas\tmthree}$ and $\appflag = \app$.
      This case is impossible since $\appflag = \nonapp$ by premise 
      of rule \ruleUEsR.
    \end{itemize}
  \item \ruleUEsLAbs.
    Then
    \[
      \inferrule{
        (1)\ \tmtwo \tov\rulename{\aset \cup \set\var}\sset\appflag \tmtwo'
        \sep
        (2)\ \tmthree \in \HAbs\aset
        \sep
        (3)\ \var \notin \aset \cup \sset
        \sep
        (4)\ \var \notin \fv\rulename
      }{
        \tm = \tmtwo\esub\var\tmthree \tov\rulename\aset\sset\appflag 
        \tmtwo'\esub\var\tmthree = \tm'
      }\ruleUEsLAbs
    \]

    Since $\inv\aset\sset{\tmtwo\esub\var\tmthree}$ holds, then it 
    implies in particular $\inv{\aset \cup \set\var}\sset\tmtwo$.
    Given (2), then $\unvalas\tmthree = \val\sctx$ by 
    \cref{rem:unfolding_vL}.
    By the first item of \cref{lem:esub_valas_compat_ext_compat},
    $\compat{\valas \cup (\var \mapsto \val)}{\aset \cup \set\var}\sset$
    holds, so we analyse two cases by the \ih on $\tmtwo$:
    \begin{itemize}
    \item
      There exists $\tmtwo''$ such that 
      $\unvalas[\valas \cup (\var \mapsto \val)]\tmtwo \tovv\ruledb \tmtwo''$.
      By definition of $\unvalas{\tmtwo\esub\var\tmthree}$, there are 
      two subcases:
      \begin{itemize}
      \item 
        If $\unvalas\tmthree = \val\sctx$ and $\var \in \rv\tmtwo$,
        then $\unvalas{\tmtwo\esub\var\tmthree} =
        \unvalas[\valas \cup (\var \mapsto \val)]\tmtwo\esub\var\val\sctx$.
        We conclude
        $\unvalas[\valas \cup (\var \mapsto \val)]\tmtwo\esub\var\val\sctx
        \tovv\ruledb \tmtwo''\esub\var\val\sctx = \tm''$
        by successively applying rule \ruleVEsL.
      \item 
        Otherwise, $\unvalas{\tmtwo\esub\var\tmthree} = 
        \unvalas\tmtwo\esub\var{\unvalas\tmthree}$.
        Then $\unvalas[\valas \cup (\var \mapsto \val)]\tmtwo = \unvalas\tmtwo$
        given $\var \notin \rv\tmtwo$. 
        Applying rule \ruleVEsL, we obtain
        $\unvalas\tmtwo\esub\var{\unvalas\tmthree} \tovv\ruledb
        \tmtwo''\esub\var{\unvalas\tmthree} = \tm''$.
      \end{itemize}
    \item
      $\abs{\unvalas[\valas \cup (\var \mapsto \val)]\tmtwo}$ and 
      $\appflag = \app$.
      It is immediate to conclude $\abs{\unvalas{\tmtwo\esub\var\tmthree}}$
      and $\appflag = \app$.
    \end{itemize}
  \item \ruleUEsLStruct.
    Then
    \[
      \inferrule{
        \tmtwo \tov\rulename\aset{\sset \cup \set\var}\appflag \tmtwo'
        \sep
        \tmthree \in \Struct\sset
        \sep
        \var \notin \aset \cup \sset
        \sep
        \var \notin \fv\rulename
      }{
        \tm = \tmtwo\esub\var\tmthree \tov\rulename\aset\sset\appflag \tmtwo'\esub\var\tmthree = \tm'
      }\ruleUEsLStruct
    \]

    Since $\inv\aset\sset{\tmtwo\esub\var\tmthree}$ holds, then it
    implies in particular $\inv\aset{\sset \cup \set\var}\tmtwo$.
    Moreover, by the second item of \cref{lem:esub_valas_compat_ext_compat},
    given $\valas' = \valas$ or 
    $\valas' = \valas \cup (\var \mapsto \val)$, we yield
    $\compat{\valas'}\aset{\sset \cup \set\var}$.
    We can apply \ih on $\tmtwo$, yielding two cases:
    \begin{itemize}
    \item
      There exists $\tmtwo''$ such that 
      $\unvalas[\valas']\tmtwo \tovv\ruledb \tmtwo''$.
      We analyse two subcases:
      \begin{itemize}
      \item 
        If $\unvalas\tmthree = \val\sctx$ and $\var \in \rv\tmtwo$,
        then $\valas' = \valas \cup (\var \mapsto \val)$ and
        $\unvalas{\tmtwo\esub\var\tmthree} = 
        \unvalas[\valas \cup (\var \mapsto \val)]\tmtwo\esub\var\val\sctx$.
        Thus, 
        $\unvalas[\valas \cup (\var \mapsto \val)]\tmtwo\esub\var\val\sctx
        \tovv\ruledb \tmtwo''\esub\var\val\sctx = \tm''$ by 
        successively applying rule \ruleVEsL.
      \item 
        Otherwise, $\valas' = \valas$ and 
        $\unvalas{\tmtwo\esub\var\tmthree} = 
        \unvalas\tmtwo\esub\var{\unvalas\tmthree}$.
        Applying rule \ruleVEsL, we obtain
        $\unvalas\tmtwo\esub\var{\unvalas\tmthree} \tovv\ruledb
        \tmtwo''\esub\var{\unvalas\tmthree} = \tm''$.
      \end{itemize}
    \item
      $\abs{\unvalas\tmtwo}$ and $\appflag = \app$.
      It is immediate to conclude $\abs{\unvalas{\tmtwo\esub\var\tmthree}}$
      and $\appflag = \app$.
    \end{itemize}
  \end{itemize}
\end{enumerate}
\end{proof}

\begin{corollary}
\label{coro:t_UNF_unfolding_t_VNF}
Let $\tm$ be a term.
Then, $\tm \in \NF\emptyset{\fv\tm}\nonapp$
if and only if $\unvalasevalas\tm \in \VNF\emptyset\nonapp$.
\end{corollary}

\begin{proof}
The `only if' direction is an immediate consequence of 
\cref{prop:useful_nf_unfold_to_nf} (\ref{prop:useful_nf_unfold_to_nf:part}).

For the `if' direction, it suffices to show the contrapositive, 
namely that $\tm \notin \NF\emptyset{\fv\tm}\nonapp$
implies $\unvalasevalas\tm \notin \VNF\emptyset\nonapp$.
By \cref{thm:characterization_of_useful_normal_forms}, this 
implication is equivalent to showing that if 
$\tm \tov\rulename\emptyset{\fv\tm}\nonapp \tm'$, then there exist a 
step kind $\rulename'$ and a term $\tm''$ such that 
$\unvalasevalas\tm \tovv{\rulename'} \tm''$.
This is a consequence of \cref{prop:unfold_reducible_is_reducible} 
(\ref{prop:unfold_reducible_is_reducible:part}), so we are done.
\end{proof}

\section{Proofs of Section~\ref{sec:usefulcbv_invariant} ``\UOCBV is Invariant''}
\label{app:usefulcbv_invariant}
In this section, we develop the technical details that are necessary 
to show that our inductive characterisation of useful evaluation is 
invariant. 
Recall that we proceed in two stages, following the technique used 
in~\cite{AccattoliC15}: on one hand we prove
a high-level implementation theorem (\cref{thm:high_level_implementation}). 
On the other hand, we show a low-level implementation theorem (\cref{thm:low_level_implementation}). 
This second stage requires more development than the first one, given 
that we need to relate our \UOCBV strategy with the \glamour abstract 
machine introduced in~\cite{AccattoliC15}.

\subsection{Low-Level Implementation}

Recall that a term is \defn{pure} if it contains no ESs.
We write $\purePred\tm$ to denote that $\tm$ is pure.

Let $\aset$ be an abstraction frame and $\sset$ a structure frame.
The set of \defn{stable terms} under $(\aset,\sset)$, written 
$\Stable\aset\sset$, is inductively defined as follows:
\[
  \inferrule{
  }{
    \var \in \Stable\aset\sset
  }\ruleStableVar
  \HS
  \inferrule{
    \purePred\tm
  }{
    \lam\var\tm \in \Stable\aset\sset
  }\ruleStableAbs
\]
\[
  \inferrule{
    \tm \in \Stable\aset\sset
    \sep
    \tmtwo \in \Stable\aset\sset
  }{
    \tm \, \tmtwo \in \Stable\aset\sset
  }\ruleStableApp
\]
\[
  \inferrule{
    \tm \in \Stable{\aset \cup \set\var}\sset
    \sep
    \tmtwo \in \Stable\aset\sset
    \sep
    \tmtwo \in \HAbs\aset
    \sep
    \var \notin \aset \cup \sset
  }{
    \tm\esub\var\tmtwo \in \Stable\aset\sset
  }\ruleStableESHAbs
\]
\[
  \inferrule{
    \tm \in \Stable\aset{\sset \cup \set\var}
    \sep
    \tmtwo \in \Stable\aset\sset
    \sep
    \tmtwo \in \Struct\sset
    \sep
    \var \notin \aset \cup \sset
  }{
    \tm\esub\var\tmtwo \in \Stable\aset\sset
  }\ruleStableESStruct
\]

\begin{remark}
\label{rem:pure_implies_stable}
Let $\tm$ be a term, $\aset$ an abstraction frame, and $\sset$ a
structure frame.
If $\purePred\tm$, then $\tm \in \Stable\aset\sset$.
\end{remark}

\begin{remark}
\label{lem:stable_then_stableFramesExtended}
Let $\tm$ be a term, $\aset$ an abstraction frame, and $\sset$ a 
structure frame. 
Then:
\begin{enumerate}
\item
  Let $\asettwo$ be an abstraction frame such that $\asettwo \disj \fv\tm$,
  and $\ssettwo$ be a structure frame such that $\ssettwo \disj \fv\tm$ 
  and $\asettwo \disj \ssettwo$.
  If $\tm \in \Stable\aset\sset$, then $\tm \in \Stable{\aset \cup \asettwo}{\sset \cup \ssettwo}$.
\item
  Let $\asettwo$ be an abstraction frame such that $\aset \disj \asettwo$
  and $\ssettwo$ be a structure frame such that $\sset \disj \ssettwo$,
  and moreover, $\asettwo \disj \ssettwo$.
  If $\tm \in \Stable{\aset}{\sset}$, then $\tm \in \Stable{\aset \cup \asettwo}{\sset \cup \ssettwo}$.
\end{enumerate}
\end{remark}

Let $\aset$ be an abstraction frame and $\sset$ a structure frame.
The set of \defn{stable substitution contexts} under $(\aset,\sset)$,
written $\StableCtx\aset\sset$, is inductively defined as follows:
\[
  \inferrule{
  }{
    \ctxhole \in \StableCtx\aset\sset
  }\ruleStableCtxEmpty
\]
\[
  \inferrule{
    \sctx \in \StableCtx{\aset \cup \set\var}\sset
    \sep
    \tm \in \Stable\aset\sset
    \sep
    \tm \in \HAbs\aset
    \sep
    \var \notin \aset \cup \sset
  }{
    \sctx\esub\var\tm \in \StableCtx\aset\sset
  }\ruleStableCtxHAbs
\]
\[
  \inferrule{
    \sctx \in \StableCtx\aset{\sset \cup \set\var}
    \sep
    \tm \in \Stable\aset\sset
    \sep
    \tm \in \Struct\sset
    \sep
    \var \notin \aset \cup \sset
  }{
    \sctx\esub\var\tm \in \StableCtx\aset\sset
  }\ruleStableCtxStruct
\]

\begin{lemma}
\label{lem:tL_stable_iff_t_stableExtended_L_stable}
Let $\aset$ be an abstraction frame, $\sset$ be a structure frame,
$\tm$ be a term, and $\sctx$ be a substitution context such that 
$\inv\aset\sset{\tm\sctx}$ holds.
Then, $\tm\sctx \in \Stable\aset\sset$ if and only if 
$\tm \in \Stable{\expansion\aset\sctx}{\expansion\sset\sctx}$ and 
$\sctx \in \StableCtx\aset\sset$.
\end{lemma}
% Label: lem:tL_stable_iff_t_stableExtended_L_stable

\begin{proof}
We prove both sides by induction on the length of $\sctx$.

\noindent $1 \Rightarrow 2)$
\begin{itemize}
\item $\sctx = \ctxhole$.
  On the one hand, $\expansion{\aset}{\ctxhole} = \aset$ 
  and $\expansion{\sset}{\ctxhole} = \sset$ hold by definition,
  therefore we conclude $\tm \in \Stable{\aset}{\sset}$ by hypothesis.
  Lastly, we conclude $\ctxhole \in \StableCtx{\aset}{\sset}$ applying rule
  $\ruleStableCtxEmpty$.
\item $\sctx = \sctx'\esub{\var}{\tmtwo}$.
  We can derive $\tm\sctx'\esub{\var}{\tmtwo}$ either by rule $\ruleStableESHAbs$
  or rule $\ruleStableESStruct$. Both cases are analogous, so we will only show
  case $\ruleStableESHAbs$, where we have
  $\tm\sctx' \in \Stable{\aset \cup \set{\var}}{\sset}$ and
  (1) $\tmtwo \in \Stable{\aset}{\sset}$, with (2) $\tmtwo \in \HAbs{\aset}$
  and (3) $\var \notin \aset \cup \sset$.
  We can apply \ih on $\sctx'$, yielding
  $\tm \in \Stable{\expansion{(\aset \cup \set{\var})}{\sctx'}}{\expansion{\sset}{\sctx'}}$
  and (4) $\sctx' \in \StableCtx{\aset \cup \set{\var}}{\sset}$.
  Note that $\expansion{(\aset \cup \set{\var})}{\sctx'} 
  = \expansion{\aset}{\sctx'} \cup \set{\var}
  = \expansion{\aset}{\sctx'\esub{\var}{\tmtwo}}$,
  and since $\inv{\aset}{\sset}{\tm\sctx'\esub{\var}{\tmtwo}}$ implies
  $\inv{\aset}{\sset}{\tmtwo}$, then we have $\tmtwo \notin \Struct{\sset}$
  by \cref{lem:disjunction}.
  Hence by definition
  $\expansion{\sset}{\sctx'} = \expansion{\sset}{\sctx'\esub{\var}{\tmtwo}}$.
  Therefore 
  $\tm \in \Stable{\expansion{\aset}{\sctx'\esub{\var}{\tmtwo}}}{\expansion{\sset}{\sctx'\esub{\var}{\tmtwo}}}$.
  On the other hand, we can apply rule $\ruleStableCtxHAbs$ 
  with (1), (2), (3) and (4) as premises, yielding
  $\sctx'\esub{\var}{\tmtwo} \in \StableCtx{\aset}{\sset}$.
\end{itemize}
\noindent $2 \Rightarrow 1)$
\begin{itemize}
\item $\sctx = \ctxhole$.
  Then $\expansion{\aset}{\ctxhole} = \aset$
  and $\expansion{\sset}{\ctxhole} = \sset$ hold by definition,
  therefore we conclude $\tm\ctxhole = \tm \in \Stable{\aset}{\sset}$ by hypothesis.
\item $\sctx = \sctx'\esub{\var}{\tmtwo}$.
  We have two cases, depending on whether 
  $\tmtwo \in \HAbs{\aset}$ or $\tmtwo \in \Struct{\sset}$.
  Both cases are analogous, so we will only show case $\tmtwo \in \HAbs{\aset}$.
  Given $\inv{\aset}{\sset}{\tm\sctx'\esub{\var}{\tmtwo}}$ then
  $\inv{\aset \cup \set{\var}}{\sset}{\tm\sctx'}$
  and $\inv{\aset}{\sset}{\tmtwo}$ hold, and the last statement 
  implies $\tmtwo \notin \Struct{\sset}$ by \cref{lem:disjunction}.
  Hence by definition
  $\expansion{\aset}{\sctx'\esub{\var}{\tmtwo}} 
  = \expansion{\aset}{\sctx'} \cup \set{\var}
  = \expansion{(\aset \cup \set{\var})}{\sctx'}$,
  and $\expansion{\sset}{\sctx'\esub{\var}{\tmtwo}} = \expansion{\sset}{\sctx'}$.
  Therefore (1) $\tm \in \Stable{\expansion{(\aset \cup \set{\var})}{\sctx'}}{\expansion{\sset}{\sctx'}}$.
  On the other hand, $\sctx'\esub{\var}{\tmtwo}$ can only be derived by rule
  $\ruleStableCtxHAbs$, with (2) $\sctx' \in \StableCtx{\aset \cup \set{\var}}{\sset}$,
  (3) $\tmtwo \in \Stable{\aset}{\sset}$, (4) $\tmtwo \in \HAbs{\aset}$, 
  and (5) $\var \notin \aset \cup \sset$.
  We can apply \ih on $\sctx'$ with (1) and (2) as hypothesis, yielding
  (6) $\tm\sctx' \in \Stable{\aset \cup \set{\var}}{\sset}$.
  We conclude $\tm\sctx'\esub{\var}{\tmtwo} \in \Stable{\aset}{\sset}$
  by applying rule $\ruleStableESHAbs$, with (3), (4), (5) and (6) as hypothesis.
\end{itemize}
\end{proof}

\begin{lemma}
\label{lem:sub_reduction_preserves_stability}
Let $\tm$ be a term, $\aset$, $\asettwo$ be two abstraction frames, 
and $\sset$, $\ssettwo$ be two structure frames such that 
$\aset \disj \asettwo$ and $\sset \disj \ssettwo$ and 
$\asettwo \disj \ssettwo$, and suppose that 
$\inv{\aset \cup \set\var}\sset\tm$ holds.
If $\tm \tov{\rulesub\var\val}{\aset \cup \set\var}\sset\appflag \tm'$ 
with $\val \in \Stable{\aset \cup \asettwo}{\sset \cup \ssettwo}$ and 
$\tm \in \Stable{\aset \cup \set\var}\sset$, then 
$\tm' \in \Stable{\aset \cup \set\var \cup \asettwo}{\sset \cup \ssettwo}$.
\end{lemma}
%Label: lem:sub_reduction_preserves_stability

\begin{proof}
By induction on the derivation of 
$\tm \tov{\rulesub{\var}{\val}}{\aset \cup \set{\var}}{\sset}{\appflag} \tm'$.
\begin{enumerate}
\item $\ruleUSub$.
  Then
  $\tm = \var \tov{\rulesub{\var}{\val}}{\aset \cup \set{\var}}{\sset}{\app} \val = \tm'$,
  with $\appflag = \app$.
  Since $\val \in \Stable{\aset \cup \asettwo}{\sset \cup \ssettwo}$ by hypothesis,
  then $\val \in \Stable{\aset \cup \set{\var} \cup \asettwo}{\sset \cup \ssettwo}$
  by \cref{lem:stable_then_stableFramesExtended}.
\item $\ruleUAppL$.
  Then
  \[
    \indrule{\ruleUAppL}{
      \tmtwo \tov{\rulesub{\var}{\val}}{\aset \cup \set{\var}}{\sset}{\app} \tmtwo'
    }{
      \tm = \tmtwo \, \tmthree 
      \tov{\rulesub{\var}{\val}}{\aset \cup \set{\var}}{\sset}{\appflag} 
      \tmtwo' \, \tmthree = \tm'
    }
  \]
  and $\tmtwo \, \tmthree \in \Stable{\aset \cup \set{\var}}{\sset}$ can only be
  derived from rule $\ruleStableApp$, so that
  $\tmtwo \in \Stable{\aset \cup \set{\var}}{\sset}$ and 
  $\tmthree \in \Stable{\aset \cup \set{\var}}{\sset}$ hold.
  Moreover, $\inv{\aset \cup \set{\var}}{\sset}{\tmtwo \, \tmthree}$ implies
  $\inv{\aset \cup \set{\var}}{\sset}{\tmtwo}$, therefore
  we can apply \ih on $\tmtwo$, yielding 
  $\tmtwo' \in \Stable{\aset \cup \set{\var} \cup \asettwo}{\sset \cup \ssettwo}$.
  By \cref{lem:stable_then_stableFramesExtended}, we have
   $\tmthree \in \Stable{\aset \cup \set{\var} \cup \asettwo}{\sset \cup \ssettwo}$.
  We conclude
  $\tmtwo' \, \tmthree \in \Stable{\aset \cup \set{\var} \cup \asettwo}{\sset \cup \ssettwo}$
  by applying rule $\ruleStableApp$.
\item $\ruleUAppR$.
  Analogous to the previous case.
  % Then
  % \[
  %   \indrule{\ruleUAppR}{
  %     \tmtwo \in \Struct{\sset}
  %     \sep
  %     \tmthree \tov{\rulesub{\var}{\val}}{\aset \cup \set{\var}}{\sset}{\nonapp} \tmthree'
  %   }{
  %     \tm = \tmtwo \, \tmthree 
  %     \tov{\rulesub{\var}{\val}}{\aset \cup \set{\var}}{\sset}{\appflag} 
  %     \tmtwo \, \tmthree' = \tm'
  %   }
  % \]
  % and $\tmtwo \, \tmthree \in \Stable{\aset \cup \set{\var}}{\sset}$ can only be
  % derived from rule $\ruleStableApp$, so that
  % $\tmtwo \in \Stable{\aset \cup \set{\var}}{\sset}$ and 
  % $\tmthree \in \Stable{\aset \cup \set{\var}}{\sset}$ hold.
  % Moreover, $\inv{\aset \cup \set{\var}}{\sset}{\tmtwo \, \tmthree}$ implies
  % $\inv{\aset \cup \set{\var}}{\sset}{\tmthree}$, therefore
  % we can apply \ih on $\tmthree$, yielding 
  % $\tmthree' \in \Stable{\aset \cup \set{\var} \cup \asettwo}{\sset \cup \ssettwo}$.
  % By \cref{lem:stable_then_stableFramesExtended}, we have
  % (1) $\tmthree \in \Stable{\aset \cup \set{\var} \cup \asettwo}{\sset \cup \ssettwo}$.
  % We conclude
  % $\tmtwo' \, \tmthree \in \Stable{\aset \cup \set{\var} \cup \asettwo}{\sset \cup \ssettwo}$
  % by applying rule $\ruleStableApp$.
\item $\ruleUEsR$.
  Then
  \[
    \indrule{\ruleUEsR}{
      \tmthree \tov{\rulesub{\var}{\val}}{\aset \cup \set{\var}}{\sset}{\nonapp} \tmthree'
    }{
      \tm = \tmtwo\esub{\vartwo}{\tmthree} 
      \tov{\rulesub{\var}{\val}}{\aset \cup \set{\var}}{\sset}{\appflag} 
      \tmtwo\esub{\vartwo}{\tmthree'} = \tm'
    }
  \]
  By $\alpha$-conversion we may assume 
  $\vartwo \notin \asettwo \cup \ssettwo$.
  We can derive
  $\tmtwo\esub{\vartwo}{\tmthree} \in \Stable{\aset \cup \set{\var}}{\sset}$
  either by rule $\ruleStableESHAbs$ or by rule $\ruleStableESStruct$.
  Both cases are analogous, so we only show case $\ruleStableESHAbs$.
  Hence $\tmtwo \in \Stable{\aset \cup \set{\var} \cup \set{\vartwo}}{\sset}$,
  $\tmthree \in \Stable{\aset \cup \set{\var}}{\sset}$, 
  $\tmthree \in \HAbs{\aset \cup \set{\var}}$,
  and $\vartwo \notin \aset \cup \set{\var} \cup \sset$,
  so that (1) $\vartwo \notin \aset \cup \set{\var} \cup \asettwo \cup \sset \cup \ssettwo$.
  Given that $\inv{\aset \cup \set{\var}}{\sset}{\tmtwo\esub{\vartwo}{\tmthree}}$
  implies $\inv{\aset \cup \set{\var}}{\sset}{\tmthree}$,
  we can apply \ih on $\tmthree$, yielding 
  (2) $\tmthree' \in \Stable{\aset \cup \set{\var} \cup \asettwo}{\sset \cup \ssettwo}$.
  Moreover, 
  (3) $\tmtwo \in \Stable{\aset \cup \set{\var} \cup \set{\vartwo} \cup \asettwo}{\sset \cup \ssettwo}$ 
  by \cref{lem:stable_then_stableFramesExtended}.
  We have (4) $\tmthree' \in \HAbs{\aset \cup \set{\var} \cup \asettwo}$
  by \cref{lem:habs_rulesub_closed_reduction}.
  We apply rule $\ruleStableESHAbs$ with (1), (2), (3) and (4) as premises, yielding
  $\tmtwo\esub{\vartwo}{\tmthree'} \in \Stable{\aset \cup \set{\var} \cup \asettwo}{\sset \cup \ssettwo}$.
\item $\ruleUEsLAbs$.
  Then
  \[
    \indrule{\ruleUEsLAbs}{
      \tmtwo \tov{\rulesub{\var}{\val}}{\aset \cup \set{\var} \cup \set{\vartwo}}{\sset}{\appflag} \tmtwo'
      \sep
      \tmthree \in \HAbs{\aset \cup \set{\var}}
      \sep
      \vartwo \notin (\aset \cup \set{\var}) \cup \sset
      \sep
      \vartwo \notin \fv{\rulesub{\var}{\val}}
    }{
      \tm = \tmtwo\esub{\vartwo}{\tmthree} 
      \tov{\rulesub{\var}{\val}}{\aset \cup \set{\var}}{\sset}{\appflag} 
      \tmtwo'\esub{\vartwo}{\tmthree} = \tm'
    }
  \]
  Note that the judgement $\tmtwo\esub{\var}{\tmthree} \in \Stable{\aset \cup \set{\var}}{\sset}$
  can be derived only by rule $\ruleStableESHAbs$.
  Then $\tmtwo \in \Stable{\aset \cup \set{\var} \cup \set{\vartwo}}{\sset}$ and
  $\tmthree \in \Stable{\aset \cup \set{\var}}{\sset}$.
  By $\alpha$-conversion, we may assume that $\vartwo \notin \asettwo \cup \ssettwo$.
  Moreover, $\inv{\aset \cup \set{\var}}{\sset}{\tmtwo\esub{\vartwo}{\tmthree}}$
  implies $\inv{\aset \cup \set{\var} \cup \set{\vartwo}}{\sset}{\tmtwo}$, and
  note that $\val \in \Stable{\aset\cup\set{\vartwo}\cup\asettwo}{\sset \cup \ssettwo}$
  by hypothesis and \cref{rem:habs_st}.
  We can apply \ih on $\tmtwo$, yielding
  (1) $\tmtwo' \in \Stable{\aset \cup \set{\var} \cup \set{\vartwo} \cup \asettwo}{\sset \cup \ssettwo}$.
  By \cref{lem:stable_then_stableFramesExtended} we have
  $\tmthree \in \Stable{\aset \cup \set{\var} \cup \asettwo}{\sset \cup \ssettwo}$.
  And we have $\tmthree \in \HAbs{\aset \cup \set{\var} \cup \asettwo}$
  by \cref{rem:habs_st}, and 
  $\vartwo \notin \aset \cup \set{\var} \cup \asettwo \cup \sset \cup \ssettwo$.
  We conclude 
  $\tmtwo'\esub{\vartwo}{\tmthree} \in \Stable{\aset \cup \set{\var} \cup \asettwo}{\sset \cup \ssettwo}$
  by applying rule $\ruleStableESHAbs$.
\item $\ruleUEsLStruct$.
  Analogous to the previous case.
\end{enumerate}
\end{proof}

\begin{lemma}
\label{lem:reduction_preserves_stability}
Let $\tm$ be a term, $\aset$ be an abstraction frame, and $\sset$ be 
a structure frame such that $\inv\aset\sset\tm$ holds.
If $\tm \tov\rulename\aset\sset\appflag \tm'$ with 
$\rulename \in \set{\ruledb,\rulelsv}$ and $\tm \in \Stable\aset\sset$,
then $\tm' \in \Stable\aset\sset$.
\end{lemma}
% Label: reduction_preserves_stability

\begin{proof}
We proceed by induction on $\tm \tov{\rulename}{\aset}{\sset}{\appflag} \tm'$.
\begin{itemize}
\item $\ruleUDb$.
  Then 
  $\tm = (\lam{\var}{\tmtwo})\sctx \, \tmthree \tov{\ruledb}{\aset}{\sset}{\appflag}
  \tmtwo\esub{\var}{\tmthree}\sctx = \tm'$,
  where $\rulename = \ruledb$.
  The hypothesis $(\lam{\var}{\tmtwo})\sctx \, \tmthree \in \Stable{\aset}{\sset}$
  is derived from:
  \[
    \indrule{\ruleStableApp}{
      (\lam{\var}{\tmtwo})\sctx \in \Stable{\aset}{\sset}
      \sep
      \tmthree \in \Stable{\aset}{\sset}
    }{
      (\lam{\var}{\tmtwo})\sctx \, \tmthree \in \Stable{\aset}{\sset}
    }
  \]
  Moreover, $\inv{\aset}{\sset}{(\lam{\var}{\tmtwo})\sctx \, \tmthree}$ implies
  $\inv{\aset}{\sset}{(\lam{\var}{\tmtwo})\sctx}$ and
  $\inv{\aset}{\sset}{\tmthree}$.
  We then have 
  $\lam{\var}{\tmtwo} \in \Stable{\expansion{\aset}{\sctx}}{\expansion{\sset}{\sctx}}$
  and (1) $\sctx \in \StableCtx{\aset}{\sset}$
  by \cref{lem:tL_stable_iff_t_stableExtended_L_stable}, 
  and this judgement can only be derived from rule $\ruleStableAbs$, 
  hence $\purePred{\tmtwo}$.
  By \cref{rem:pure_implies_stable},
  $\tmtwo \in \Stable{\expansion{\aset}{\sctx} \cup \set{\var}}{\expansion{\sset}{\sctx}}$,
  and we have to analyze whether $\tmthree \in \HAbs{\expansion{\aset}{\sctx}}$ 
  or $\tmthree \in \Struct{\expansion{\sset}{\sctx}}$.
  Since both cases are analogous, we will proceed with the first one.
  Since $\domSctx{\sctx} \disj \fv{\tmthree}$ by $\alpha$-conversion,
  then $\tmthree \in \Stable{\expansion{\aset}{\sctx}}{\expansion{\sset}{\sctx}}$ 
  holds by \cref{lem:stable_then_stableFramesExtended},
  and applying rule $\ruleStableESHAbs$ we have that
  (2) $\tmtwo\esub{\var}{\tmthree} \in \Stable{\expansion{\aset}{\sctx}}{\expansion{\sset}{\sctx}}$.
  Given that $\inv{\aset}{\sset}{(\lam{\var}{\tmtwo})\sctx \, \tmthree}$ implies
  $\inv{\aset}{\sset}{\tmtwo\esub{\var}{\tmthree}\sctx}$,
  we can apply \cref{lem:tL_stable_iff_t_stableExtended_L_stable} with (1) and (2),
  to obtain $\tmtwo\esub{\var}{\tmthree}\sctx \in \Stable{\aset}{\sset}$.
\item $\ruleULsv$.
  Then
  $\tm = \tmtwo\esub{\var}{\val\sctx} \tov{\rulelsv}{\aset}{\sset}{\appflag}
  \tmtwo'\esub{\var}{\val}\sctx = \tm'$, where $\rulename = \rulelsv$ and it is
  derived from
  $\tmtwo \tov{\rulesub{\var}{\val}}{\aset \cup \set{\var}}{\sset}{\appflag} \tmtwo'$,
  $\var \notin \aset \cup \sset$, and
  (1) $\val\sctx \in \HAbs{\aset}$.
  The hypothesis $\tmtwo\esub{\var}{\val\sctx} \in \Stable{\aset}{\sset}$
  is derived from rule $\ruleStableESHAbs$ by (1), 
  so (2) $\tmtwo \in \Stable{\aset \cup \set{\var}}{\sset}$ and
  $\val\sctx \in \Stable{\aset}{\sset}$.
  By \cref{lem:tL_stable_iff_t_stableExtended_L_stable}, we have
  (3) $\val \in \Stable{\expansion{\aset}{\sctx}}{\expansion{\sset}{\sctx}}$
  and (4) $\sctx \in \StableCtx{\aset}{\sset}$.
  On the other hand, since $\inv{\aset}{\sset}{\tmtwo\esub{\var}{\tmthree}}$
  then $\inv{\aset \cup \set{\var}}{\sset}{\tmtwo}$ holds
  and we can apply \cref{lem:sub_reduction_preserves_stability} with (2) and (3) as hypothesis,
  yielding 
  (5) $\tmtwo' \in \Stable{\expansion{(\aset \cup \set{\var})}{\sctx}}{\expansion{\sset}{\sctx}}
  = \Stable{\expansion{\aset}{\sctx} \cup \set{\var}}{\expansion{\sset}{\sctx}}$.
  Moreover (6) $\val \in \HAbs{\expansion{\aset}{\sctx}}$ by \cref{lem:tL_hAbs_t_hAbsExp}.
  we may assume $\var \notin \domSctx{\sctx}$ by $\alpha$-conversion, 
  hence we can apply rule $\ruleStableESHAbs$ with (3), (5), (6) and
  $\var \notin \expansion{\aset}{\sctx} \cup \expansion{\sset}{\sctx}$ as premises, thus
  having (7) $\tmtwo'\esub{\var}{\val} \in \Stable{\expansion{\aset}{\sctx}}{\expansion{\sset}{\sctx}}$.
  We conclude $\tmtwo'\esub{\var}{\val}\sctx \in \Stable{\aset}{\sset}$
  by applying \cref{lem:tL_stable_iff_t_stableExtended_L_stable} with
  (4) and (7) as hypothesis.
\item $\ruleUAppL$.
  Then
  $\tm = \tmtwo\,\tmthree \tov{\rulename}{\aset}{\sset}{\appflag} \tmtwo' \, \tmthree = \tm'$,
  derived from
  $\tmtwo \tov{\rulename}{\aset}{\sset}{\app} \tmtwo'$.
  The hypothesis $\tmtwo\,\tmthree \in \Stable{\aset}{\sset}$
  can only be derived from rule $\ruleStableApp$, 
  so (1) $\tmtwo \in \Stable{\aset}{\sset}$ and $\tmthree \in \Stable{\aset}{\sset}$.
  Given that $\inv{\aset}{\sset}{\tmtwo\,\tmthree}$
  implies $\inv{\aset}{\sset}{\tmtwo}$, we can apply \ih on $\tmtwo$,
  yielding (2) $\tmtwo' \in \Stable{\aset}{\sset}$.
  By rule $\ruleStableApp$ with (1) and (2) as premises,
  we conclude $\tmtwo'\,\tmthree \in \Stable{\aset}{\sset}$.
\item $\ruleUAppR$.
  Analogous to the previous case.
\item $\ruleUEsR$.
  Then
  $\tm = \tmtwo\esub{\var}{\tmthree} \tov{\rulename}{\aset}{\sset}{\appflag}
  \tmtwo\esub{\var}{\tmthree'} = \tm'$,
  derived from
  $\tmthree \tov{\rulename}{\aset}{\sset}{\nonapp} \tmthree'$.
  The hypothesis $\tmtwo\esub{\var}{\tmthree} \in \Stable{\aset}{\sset}$
  can be derived either from rule $\ruleStableESHAbs$ or rule $\ruleStableESStruct$, 
  so $\tmthree \in \Stable{\aset}{\sset}$.
  Given that $\inv{\aset}{\sset}{\tmtwo\esub{\var}{\tmthree}}$
  implies $\inv{\aset}{\sset}{\tmthree}$, we can apply \ih on $\tmthree$,
  yielding $\tmthree' \in \Stable{\aset}{\sset}$.
  We can apply either rule $\ruleStableESHAbs$ or $\ruleStableESStruct$
  to conclude $\tmtwo\esub{\var}{\tmthree'} \in \Stable{\aset}{\sset}$.
\item $\ruleUEsLAbs$.
  Then
  $\tm = \tmtwo\esub{\var}{\tmthree} \tov{\rulename}{\aset}{\sset}{\appflag}
  \tmtwo'\esub{\var}{\tmthree} = \tm'$,
  derived from
  $\tmtwo \tov{\rulename}{\aset \cup \set{\var}}{\sset}{\appflag} \tmtwo'$,
  (1) $\tmthree \in \HAbs{\aset}$,
  $\var \notin \aset \cup \sset$, and $\var \notin \fv{\rulename}$.
  The hypothesis $\tmtwo\esub{\var}{\tmthree} \in \Stable{\aset}{\sset}$ can only
  be derived from rule $\ruleStableESHAbs$, 
  so $\tmtwo \in \Stable{\aset \cup \set{\var}}{\sset}$
  and (2) $\tmthree \in \Stable{\aset}{\sset}$.
  Given that $\inv{\aset}{\sset}{\tmtwo\esub{\var}{\tmthree}}$
  implies $\inv{\aset \cup \set{\var}}{\sset}{\tmtwo}$, we can apply \ih on $\tmtwo$,
  yielding (3) $\tmtwo' \in \Stable{\aset \cup \set{\var}}{\sset}$.
  By rule $\ruleStableESHAbs$ with (1), (2) and (3) as premises, we conclude
  $\tmtwo'\esub{\var}{\tmthree} \in \Stable{\aset}{\sset}$.
\item $\ruleUEsLStruct$.
  Analogous to the previous case.
\end{itemize}
\end{proof}

We now restrict the reduction rule $\tov\rulename\aset\sset\appflag$ 
to stable terms.
The \defn{stable reduction} relation, written $\tostable\rulename\aset\sset\appflag$,
is defined by the same reduction rules as the relation $\tov\rulename\aset\sset\appflag$
(see \cref{sec:usefulcbv}), but rule \ruleUDb is replaced by
\[
  \inferrule{
    \tmtwo \in \HAbs\aset \cup \Struct\sset
  }{
    (\lam\var\tm)\sctx \, \tmtwo \tostable\ruledb\aset\sset\appflag
    \tm\esub\var\tmtwo\sctx
  }\ruleUDbStable
\]

The relation of \defn{structural equivalence} defined in 
\cref{sec:usefulcbv_invariant} can be defined alternatively as:
\[
  \inferrule{
    \var \notin \fv\tmthree
    \sep
    \vartwo \notin \fv\tmtwo
  }{
    \tm\esub\var\tmtwo\esub\vartwo\tmthree \equiv 
    \tm\esub\vartwo\tmthree\esub\var\tmtwo
  }\ruleEquivEsComm
  \HS
  \inferrule{
    \vartwo \notin \fv\tm
  }{
    \tm\esub\var\tmtwo\esub\vartwo\tmthree \equiv 
    \tm\esub\var{\tmtwo\esub\vartwo\tmthree}
  }\ruleEquivEsAssoc
\]
\[
  \inferrule{
    \var \notin \fv\tmtwo
  }{
    (\tm \, \tmtwo)\esub\var\tmthree \equiv \tm\esub\var\tmthree \, \tmtwo
  }\ruleEquivEsLDist
  \HS
  \inferrule{
    \var \notin \fv\tm
  }{
    (\tm \, \tmtwo)\esub\var\tmthree \equiv \tm \, \tmtwo\esub\var\tmthree
  }\ruleEquivEsRDist
\]
\[
  \inferrule{
    \tm \equiv \tm'
    \sep
    \tmtwo \equiv \tmtwo'
  }{
    \tm \, \tmtwo \equiv \tm' \, \tmtwo'
  }\ruleEquivCongApp
  \HS
  \inferrule{
    \tm \equiv \tm'
    \sep
    \tmtwo \equiv \tmtwo'
  }{
    \tm\esub\var\tmtwo \equiv \tm'\esub\var{\tmtwo'}
  }\ruleEquivCongES
\]
\[
  \inferrule{
  }{
    \tm \equiv \tm
  }\ruleEquivRefl
  \HS
  \inferrule{
    \tm \equiv \tmtwo
  }{
    \tmtwo \equiv \tm
  }\ruleEquivSym
  \HS
  \inferrule{
    \tm \equiv \tmtwo
    \sep
    \tmtwo \equiv \tmthree
  }{
    \tm \equiv \tmthree
  }\ruleEquivTrans
\]

\begin{remark}
\label{rem:t_equiv_vL_is_vpLp}
Let $\tm$ be a term, $\val$ be a value, and $\sctx$ be a list of ESs. 
If $\val\sctx \equiv \tm$ then $\tm$ is of the form $\val'\sctxtwo$,
and the proof of equivalence necessarily uses rule $\ruleEquivRefl$.
\end{remark}

\begin{remark}[Strengthening of Abstraction and Value Frames]
\label{lem:strengthening_a_s}
Let $\aset$ be an abstraction frame and $\sset$ be a structure set.
Let $\tm$ be a term such that $\var \notin \fv\tm$.
Then:
\begin{enumerate}
\item 
  If $\tm \in \HAbs{\aset \cup \set\var}$ then $\tm \in \HAbs\aset$.
\item 
  If $\tm \in \Struct{\sset \cup \set\var}$ then $\tm \in \Struct\sset$.
\end{enumerate}
\end{remark}

\begin{lemma}[Hereditary Abstractions and Structures are Closed by Structural Equivalence]
\label{lem:habs_structures_closed_by_equivalence}
Let $\tm$ be a term, $\aset$ be an abstraction frame, and $\sset$ be
a structure frame such that $\inv\aset\sset\tm$ holds.
If $\tm \equiv \tmtwo$, then:
\begin{enumerate}
\item \label{lem:habs_structures_closed_by_equivalence-case-one}
  $\tm \in \HAbs\aset$ if and only if $\tmtwo \in \HAbs\aset$.
\item
  $\tm \in \Struct\sset$ if and only if $\tmtwo \in \Struct\sset$.
\end{enumerate}
\end{lemma}
\hiddenproof{
  Each item is proved
  by induction on the derivation of $\tm \equiv \tmtwo$.
}{
  %Label: lem:habs_structures_closed_by_equivalence

\begin{proof}
We prove both items by induction on the derivation of $\tm \equiv \tmtwo$.
Both cases are analogous, so we will only show case \ref{lem:habs_structures_closed_by_equivalence-case-one}. \\
Note that cases $\ruleEquivRefl$, $\ruleEquivSym$ and $\ruleEquivTrans$ are immediate,
and that case $\ruleEquivCongApp$ is not possible since applications cannot
belong to the set of hereditary abstractions by \cref{rem:habs_st}. \\
We first show $\tm \in \HAbs{\aset}$ implies $\tmtwo \in \HAbs{\aset}$.
We analyze the remaining cases.
\begin{itemize}
\item $\ruleEquivEsComm$.
  Then $\tm = \tmthree\esub{\var}{\tmfour}\esub{\vartwo}{\tmfive} \equiv
  \tmthree\esub{\vartwo}{\tmfive}\esub{\var}{\tmfour} = \tmtwo$, 
  with $\var \notin \fv{\tmfive}$ and $\vartwo \notin \fv{\tmfour}$ and $\var \neq \vartwo$.
  The judgment $\tmthree\esub{\var}{\tmfour}\esub{\vartwo}{\tmfive} \in \HAbs{\aset}$
  can be derived either by rule $\ruleHAbsSubi$ or $\ruleHAbsSubii$, so we
  analyze both cases:
  \begin{itemize}
  \item $\ruleHAbsSubi$.
    Then $\tmthree\esub{\var}{\tmfour} \in \HAbs{\aset}$ and $\vartwo \notin \aset$.
    In turn $\tmthree\esub{\var}{\tmfour} \in \HAbs{\aset}$ can be derived
    either by rule $\ruleHAbsSubi$ or $\ruleHAbsSubii$.
    \begin{itemize}
    \item $\ruleHAbsSubi$.
      Then $\tmthree \in \HAbs{\aset}$ and $\var \notin \aset$.
      Applying rule $\ruleHAbsSubi$ we obtain 
      $\tmthree\esub{\vartwo}{\tmfive} \in \HAbs{\aset}$,
      and applying again rule $\ruleHAbsSubi$ we conclude
      $\tmthree\esub{\vartwo}{\tmfive}\esub{\var}{\tmfour} \in \HAbs{\aset}$.
    \item $\ruleHAbsSubii$.
      Then $\tmthree \in \HAbs{\aset \cup \set{\var}}$,
      $\tmfour \in \HAbs{\aset}$ and $\var \notin \aset$.
      Since $\vartwo \notin \aset \cup \set{\var}$, we can apply rule 
      $\ruleHAbsSubi$, yielding
      $\tmthree\esub{\vartwo}{\tmfive} \in \HAbs{\aset \cup \set{\var}}$,
      and applying rule $\ruleHAbsSubii$ we conclude
      $\tmthree\esub{\vartwo}{\tmfive}\esub{\var}{\tmfour} \in \HAbs{\aset}$.
    \end{itemize}
  \item $\ruleHAbsSubii$.
    Then $\tmthree\esub{\var}{\tmfour} \in \HAbs{\aset \cup \set{\vartwo}}$,
    $\tmfive \in \HAbs{\aset}$, and $\vartwo \notin \aset$.
    In turn $\tmthree\esub{\var}{\tmfour} \in \HAbs{\aset \cup \set{\vartwo}}$
    can be derived either by rule $\ruleHAbsSubi$ or $\ruleHAbsSubii$.
    \begin{itemize}
    \item $\ruleHAbsSubi$.
      Then $\tmthree \in \HAbs{\aset \cup \set{\vartwo}}$ and $\var \notin \aset$.
      Applying rule $\ruleHAbsSubii$ we obtain 
      $\tmthree\esub{\vartwo}{\tmfive} \in \HAbs{\aset}$,
      and applying rule $\ruleHAbsSubi$ we conclude
      $\tmthree\esub{\vartwo}{\tmfive}\esub{\var}{\tmfour} \in \HAbs{\aset}$.
    \item $\ruleHAbsSubii$.
      Then $\tmthree \in \HAbs{\aset \cup \set{\vartwo} \cup \set{\var}}$,
      $\tmfour \in \HAbs{\aset \cup \set{\vartwo}}$ 
      and $\var \notin \aset \cup \set{\vartwo}$.
      Since $\vartwo \notin \aset \cup \set{\var}$, we can apply rule 
      $\ruleHAbsSubii$, yielding
      $\tmthree\esub{\vartwo}{\tmfive} \in \HAbs{\aset \cup \set{\var}}$,
      and applying again rule $\ruleHAbsSubii$ we conclude
      $\tmthree\esub{\vartwo}{\tmfive}\esub{\var}{\tmfour} \in \HAbs{\aset}$.
    \end{itemize}
  \end{itemize}
\item $\ruleEquivEsAssoc$.
  Then $\tm = \tmthree\esub{\var}{\tmfour}\esub{\vartwo}{\tmfive} \equiv
  \tmthree\esub{\var}{\tmfour\esub{\vartwo}{\tmfive}} = \tmtwo$, with
  $\vartwo \notin \fv{\tmthree}$.
  The judgment $\tmthree\esub{\var}{\tmfour}\esub{\vartwo}{\tmfive} \in \HAbs{\aset}$
  can be derived either by rule $\ruleHAbsSubi$ or $\ruleHAbsSubii$, so we
  analyze both cases:
  \begin{itemize}
  \item $\ruleHAbsSubi$.
    Then $\tmthree\esub{\var}{\tmfour} \in \HAbs{\aset}$ and (1) $\vartwo \notin \aset$.
    In turn $\tmthree\esub{\var}{\tmfour} \in \HAbs{\aset}$ can be derived
    either by rule $\ruleHAbsSubi$ or $\ruleHAbsSubii$.
    \begin{itemize}
    \item $\ruleHAbsSubi$.
      Then $\tmthree \in \HAbs{\aset}$ and $\var \notin \aset$.
      Applying rule $\ruleHAbsSubi$ to these premises we conclude 
      $\tmthree\esub{\var}{\tmfour\esub{\vartwo}{\tmfive}} \in \HAbs{\aset}$.
    \item $\ruleHAbsSubii$.
      Then (2) $\tmthree \in \HAbs{\aset \cup \set{\var}}$,
      (3) $\tmfour \in \HAbs{\aset}$ and (4) $\var \notin \aset$.
      We can apply rule $\ruleHAbsSubi$ with (3) and (1) as premises, yielding
      (5) $\tmfour\esub{\vartwo}{\tmfive} \in \HAbs{\aset}$,
      and then we apply rule $\ruleHAbsSubii$ with (2), (5) and (4) as premises 
      to conclude
      $\tmthree\esub{\var}{\tmfour\esub{\vartwo}{\tmfive}} \in \HAbs{\aset}$.
    \end{itemize}
  \item $\ruleHAbsSubii$.
    Then $\tmthree\esub{\var}{\tmfour} \in \HAbs{\aset \cup \set{\vartwo}}$,
    $\tmfive \in \HAbs{\aset}$, and $\vartwo \notin \aset$.
    In turn $\tmthree\esub{\var}{\tmfour} \in \HAbs{\aset \cup \set{\vartwo}}$
    can be derived either by rule $\ruleHAbsSubi$ or $\ruleHAbsSubii$.
    \begin{itemize}
    \item $\ruleHAbsSubi$.
      Then (1) $\tmthree \in \HAbs{\aset \cup \set{\vartwo}}$ and $\var \notin \aset$.
      Hence $\tmthree \in \HAbs{\aset}$,
      by \cref{lem:strengthening_a_s}, since (1) and 
      $\vartwo \notin \fv{\tmthree}$ by hypothesis.
      Applying rule $\ruleHAbsSubi$ we obtain 
      $\tmthree\esub{\var}{\tmfour\esub{\vartwo}{\tmfive}} \in \HAbs{\aset}$.
    \item $\ruleHAbsSubii$.
      Then (1) $\tmthree \in \HAbs{\aset \cup \set{\vartwo} \cup \set{\var}}$,
      (2) $\tmfour \in \HAbs{\aset \cup \set{\vartwo}}$ and
      $\var \notin \aset \cup \set{\vartwo}$.
      We can apply rule $\ruleHAbsSubii$ with (2), $\vartwo \notin \aset$ and
      $\tmfour \in \HAbs{\aset}$ as premises, yielding
      (3) $\tmfour\esub{\vartwo}{\tmfive} \in \HAbs{\aset}$.
      Moreover, (4) $\tmthree \in \HAbs{\aset \cup \set{\var}}$,
      by \cref{lem:strengthening_a_s}, since (1) and 
      $\vartwo \notin \fv{\tmthree}$ by hypothesis.
      Since $\var \notin \aset$, we can apply rule 
      $\ruleHAbsSubii$ again with (4), (3) and $\var \notin \aset$ as premises, yielding
      $\tmthree\esub{\var}{\tmfour\esub{\vartwo}{\tmfive}} \in \HAbs{\aset}$.
    \end{itemize}
  \end{itemize}
\item $\ruleEquivEsLDist$.
  Then $\tm = (\tmthree \, \tmfour)\esub{\var}{\tmfive} \equiv
  \tmthree\esub{\var}{\tmfive} \, \tmfour = \tmtwo$, with
  $\var \notin \fv{\tmfour}$.
  The judgment $(\tmthree \, \tmfour)\esub{\var}{\tmfive} \in \HAbs{\aset}$
  can be derived either by rule $\ruleHAbsSubi$ or $\ruleHAbsSubii$,
  and in both cases we have a judgment of the form
  $\tmthree \, \tmfour \in \HAbs{\aset'}$, which is not possible by \cref{rem:habs_st}.
  Hence this case is impossible.
\item $\ruleEquivEsRDist$.
  Analogously to the previous case, this case is not possible.
\item $\ruleEquivCongES$.
  Then $\tm = \tmthree\esub{\var}{\tmfour} \equiv
  \tmthree'\esub{\var}{\tmfour'} = \tmtwo$, derived from
  $\tmthree \equiv \tmthree'$ and $\tmfour \equiv \tmfour'$.
  Moreover, $\inv{\aset}{\sset}{\tmthree\esub{\var}{\tmfour}}$ implies
  $\inv{\aset \cup \set{\var}}{\sset}{\tmthree}$ and
  $\inv{\aset}{\sset}{\tmfour}$.
  The judgment $\tmthree\esub{\var}{\tmfour} \in \HAbs{\aset}$
  can be derived either by rule $\ruleHAbsSubi$ or $\ruleHAbsSubii$, so we
  analyze both cases:
  \begin{itemize}
  \item $\ruleHAbsSubi$.
    Then $\tmthree \in \HAbs{\aset}$ and (1) $\var \notin \aset$.
    By \ih on $\tmthree$, we have that $\tmthree' \in \HAbs{\aset}$.
    Applying rule $\ruleHAbsSubi$ we conclude 
    $\tmthree'\esub{\var}{\tmfour'} \in \HAbs{\aset}$.
  \item $\ruleHAbsSubii$.
    Then $\tmthree \in \HAbs{\aset \cup \set{\var}}$,
    $\tmfour \in \HAbs{\aset}$, and $\var \notin \aset$.
    By \ih on $\tmthree$, we have that $\tmthree' \in \HAbs{\aset \set{\var}}$,
    and by \ih on $\tmfour$, we have that $\tmfour' \in \HAbs{\aset}$.
    Applying rule $\ruleHAbsSubii$ we conclude 
    $\tmthree'\esub{\var}{\tmfour'} \in \HAbs{\aset}$.
  \end{itemize}
\end{itemize}
We now show $\tmtwo \in \HAbs{\aset}$ implies $\tm \in \HAbs{\aset}$:
\begin{itemize}
\item $\ruleEquivEsComm$.
  Then $\tm = \tmthree\esub{\var}{\tmfour}\esub{\vartwo}{\tmfive} \equiv
  \tmthree\esub{\vartwo}{\tmfive}\esub{\var}{\tmfour} = \tmtwo$, 
  with $\var \notin \fv{\tmfive}$ and $\vartwo \notin \fv{\tmfour}$ and $\var \neq \vartwo$.
  This is analogous to case $\ruleEquivEsComm$ of the previous case.
\item $\ruleEquivEsAssoc$.
  Then $\tm = \tmthree\esub{\var}{\tmfour}\esub{\vartwo}{\tmfive} \equiv
  \tmthree\esub{\var}{\tmfour\esub{\vartwo}{\tmfive}} = \tmtwo$, with
  $\vartwo \notin \fv{\tmthree}$; by $\alpha$-conversion we may assume
  $\vartwo \notin \aset$.
  The judgment $\tmthree\esub{\var}{\tmfour\esub{\vartwo}{\tmfive}} \in \HAbs{\aset}$
  can be derived either by rule $\ruleHAbsSubi$ or $\ruleHAbsSubii$, so we
  analyze both cases:
  \begin{itemize}
  \item $\ruleHAbsSubi$.
    Then $\tmthree \in \HAbs{\aset}$ and (1) $\var \notin \aset$.
    Applying rule $\ruleHAbsSubi$ we obtain 
    $\tmthree\esub{\var}{\tmfour} \in \HAbs{\aset}$.
    We can apply rule $\ruleHAbsSubi$ again, and conclude 
    $\tmthree\esub{\var}{\tmfour}\esub{\vartwo}{\tmfive} \in \HAbs{\aset}$.
  \item $\ruleHAbsSubii$.
    Then (1) $\tmthree \in \HAbs{\aset \cup \set{\var}}$ and
    $\tmfour\esub{\vartwo}{\tmfive} \in \HAbs{\aset}$ and (2) $\var \notin \aset$.
    In turn $\tmfour\esub{\vartwo}{\tmfive} \in \HAbs{\aset}$
    can be derived either by rule $\ruleHAbsSubi$ or $\ruleHAbsSubii$.
    \begin{itemize}
    \item $\ruleHAbsSubi$.
      Then (3) $\tmfour \in \HAbs{\aset}$ and $\vartwo \notin \aset$.
      Applying rule $\ruleHAbsSubii$ with (1), (2) and (3) as premises, we obtain
      $\tmthree\esub{\var}{\tmfour} \in \HAbs{\aset}$.
      Applying rule $\ruleHAbsSubi$ we conclude
      $\tmthree\esub{\var}{\tmfour}\esub{\vartwo}{\tmfive} \in \HAbs{\aset}$.
    \item $\ruleHAbsSubii$.
      Then (4) $\tmfour \in \HAbs{\aset \cup \set{\vartwo}}$,
      (5) $\tmfive \in \HAbs{\aset}$ and
      (6) $\vartwo \notin \aset$.
      By \cref{rem:habs_st},
      then (7) $\tmthree \in \HAbs{\aset \cup \set{\var} \cup \set{\vartwo}}$.
      We can apply rule $\ruleHAbsSubii$ with premises (7), (4) and 
      $\var \notin \aset \cup \set{\vartwo}$ since $\var \neq \vartwo$,
      yielding (8) $\tmthree\esub{\var}{\tmfour} \in \HAbs{\aset \cup \set{\vartwo}}$.
      We can apply rule $\ruleHAbsSubii$ again 
      with (8), (5) and (6) as premises, yielding
      $\tmthree\esub{\var}{\tmfour}\esub{\vartwo}{\tmfive} \in \HAbs{\aset}$.
    \end{itemize}
  \end{itemize}
\item $\ruleEquivEsLDist$.
  Then $\tm = (\tmthree \, \tmfour)\esub{\var}{\tmfive} \equiv
  \tmthree\esub{\var}{\tmfive} \, \tmfour = \tmtwo$, with
  $\var\notin \fv{\tmfour}$.
  By \cref{rem:habs_st},
  the term $\tmthree\esub{\var}{\tmfive} \, \tmfour \in \HAbs{\aset}$
  cannot be an hereditary abstraction, hence this case is impossible.
\item $\ruleEquivEsRDist$.
  Analogously to the previous case, this case is not possible.
\item $\ruleEquivCongES$.
  Then $\tm = \tmthree\esub{\var}{\tmfour} \equiv
  \tmthree'\esub{\var}{\tmfour'} = \tmtwo$, derived from
  $\tmthree \equiv \tmthree'$ and $\tmfour \equiv \tmfour'$.
  This is analogous to case $\ruleEquivCongES$ of the previous case.
\end{itemize}
\end{proof}

}

\structeqBisimulation*
% Label: structeq_bisimulation

\begin{proof}
We start by introducing an auxiliary equivalence relation $\equivC$
on terms, recursively defined by:
\[
  \indrule{\ruleEquivEsComm}{
    \var \notin \fv{\tmthree}
    \sep
    \vartwo \notin \fv{\tmtwo}
  }{
    \tm\esub{\var}{\tmtwo}\esub{\vartwo}{\tmthree} \equivC \tm\esub{\vartwo}{\tmthree}\esub{\var}{\tmtwo}
  }
\]
\[
  \indrule{\ruleEquivEsAssoc(1)}{
    \vartwo \notin \fv{\tm}
  }{
    \tm\esub{\var}{\tmtwo}\esub{\vartwo}{\tmthree} \equivC \tm\esub{\var}{\tmtwo\esub{\vartwo}{\tmthree}}
  }
  \indrule{\ruleEquivEsAssoc(2)}{
    \vartwo \notin \fv{\tm}
  }{
    \tm\esub{\var}{\tmtwo\esub{\vartwo}{\tmthree}} \equivC \tm\esub{\var}{\tmtwo}\esub{\vartwo}{\tmthree}
  }
\]
\[
  \indrule{\ruleEquivEsLDist(1)}{
    \var \notin \fv{\tmtwo}
  }{
    (\tm \, \tmtwo)\esub{\var}{\tmthree} \equivC \tm\esub{\var}{\tmthree} \, \tmtwo
  }
  \indrule{\ruleEquivEsLDist(2)}{
    \var \notin \fv{\tmtwo}
  }{
    \tm\esub{\var}{\tmthree} \, \tmtwo \equivC (\tm \, \tmtwo)\esub{\var}{\tmthree}
  }
\]
\[
  \indrule{\ruleEquivEsRDist(1)}{
    \var \notin \fv{\tm}
  }{
    (\tm \, \tmtwo)\esub{\var}{\tmthree} \equivC \tm \, \tmtwo\esub{\var}{\tmthree}
  }
  \indrule{\ruleEquivEsRDist(2)}{
    \var \notin \fv{\tm}
  }{
    \tm \, \tmtwo\esub{\var}{\tmthree} \equivC (\tm \, \tmtwo)\esub{\var}{\tmthree}
  }
\]
\[
  \indrule{\ruleEquivCongApp}{
    \tm \equivC \tm'
    \sep
    \tmtwo \equivC \tmtwo'
  }{
    \tm \, \tmtwo \equivC \tm' \, \tmtwo'
  }
  \indrule{\ruleEquivCongES}{
    \tm \equivC \tm'
    \sep
    \tmtwo \equivC \tmtwo'
  }{
    \tm\esub{\var}{\tmtwo} \equivC \tm'\esub{\var}{\tmtwo'}
  }
\]
Note in particular that
$\equiv$ is the reflexive-transitive closure of $\equivC$.
We divide the proof into two parts:
\begin{enumerate}
\item
  \label{it:bisimulation_ptOne}
  We show that if $\tm_0 \tostable{\rulename}{\aset}{\sset}{\appflag} \tm_1$ and
  $\tm_0 \equivC \tmtwo_0$, then there exists $\tmtwo_1$ such that
  $\tmtwo_0 \tostable{\rulename}{\aset}{\sset}{\appflag} \tmtwo_1$ and
  $\tm_1 \equiv \tmtwo_1$.
\item 
  \label{it:bisimulation_ptTwo}
  Given that $\equiv$ is the reflexive-transitive closure of $\equivC$,
  we show the same result but for the $\equiv$ relation by resorting to 
  \cref{it:bisimulation_ptOne}.
\end{enumerate}
\cref{it:bisimulation_ptTwo} is proved by induction on the reflexive and transitive
closure of $\equivC$ and is immediate.
The proof of \cref{it:bisimulation_ptOne} is by induction on the derivation of
$\tm_0 \tostable{\rulename}{\aset}{\sset}{\appflag} \tm_1$ and case analysis.
\begin{itemize}[leftmargin=*]
\item $\ruleUDbStable$.
  Then 
  $\tm_0 = (\lam{\var}{\tmthree})\sctx \, \tmfour \tostable{\ruledb}{\aset}{\sset}{\appflag}
  \tmthree\esub{\var}{\tmfour}\sctx = \tm_1$,
  with $\tmfour \in \HAbs{\aset} \cup \Struct{\sset}$.
  We also have $\tm_0 = (\lam{\var}{\tmthree})\sctx \, \tmfour \equivC \tmtwo_0$.
  We analyze the different cases according to which rule was used to derive the equivalence.
  Note that cases $\ruleEquivEsComm$, $\ruleEquivEsAssoc$ and $\ruleEquivCongES$
  are impossible due to the form of $\tm_0$.
  The relevant cases are $\ruleEquivEsLDist$, $\ruleEquivEsRDist$, and $\ruleEquivCongApp$:
  \begin{enumerate}[leftmargin=*, label=\alph*.]
  \item $\ruleEquivEsLDist(2)$.
    Then
    $\tm_0 = (\lam{\var}{\tmthree})\sctx'\esub{\vartwo}{\tmfive} \, \tmfour 
    \equivC ((\lam{\var}{\tmthree})\sctx' \, \tmfour)\esub{\vartwo}{\tmfive} = \tmtwo_0$,
    with $\vartwo \notin \fv{\tmfour}$ and $\sctx = \sctx'\esub{\vartwo}{\tmfive}$.
    There are two cases for reducing 
    $((\lam{\var}{\tmthree})\sctx' \, \tmfour)\esub{\vartwo}{\tmfive}$, 
    depending on whether $\tmfive \in \HAbs{\aset}$ or $\tmfive \in \Struct{\sset}$;
    both are analogous, so we show only the first case.
    Then we obtain
    $\tmtwo_0 = ((\lam{\var}{\tmthree})\sctx' \, \tmfour)\esub{\vartwo}{\tmfive} 
    \tostable{\ruledb}{\aset}{\sset}{\appflag}
    \tmthree\esub{\var}{\tmfour}\sctx'\esub{\vartwo}{\tmfive} = \tmtwo_1$
    by applying rules $\ruleUEsLAbs$ and $\ruleUDbStable$.
    Note that $\tmtwo_1 \equiv \tm_1$ by rule $\ruleEquivRefl$.
    The following diagram summarises the proof:
    \[
      \xymatrixrowsep{0.3pc}
      \xymatrix{
        \tm_0 = (\lam{\var}{\tmthree})\sctx'\esub{\vartwo}{\tmfive} \, \tmfour
          \arUrStable{\ruledb}{\aset}{\sset}{\appflag}
      %    \ar@3{-}[d;]_>{c}
      & \tmthree\esub{\var}{\tmfour}\sctx'\esub{\vartwo}{\tmfive} = \tm_1
      %    \ar@3{.}[d]
      \\
        \equivC
      & \equiv
      \\
        \tmtwo_0 = ((\lam{\var}{\tmthree})\sctx' \, \tmfour)\esub{\vartwo}{\tmfive}
          \arsdUrStable{\ruledb}{\aset}{\sset}{\appflag}
      & \tmthree\esub{\var}{\tmfour'}\sctx'\esub{\vartwo}{\tmfive} = \tmtwo_1
      }
    \]
  \item $\ruleEquivEsRDist(2)$.
    Then
    $\tm_0 = (\lam{\var}{\tmthree})\sctx \, (\tmfour'\esub{\vartwo}{\tmfive}) 
    \equivC ((\lam{\var}{\tmthree})\sctx \, \tmfour')\esub{\vartwo}{\tmfive} = \tmtwo_0$,
    with $\vartwo \notin \fv{(\lam{\var}{\tmthree})\sctx}$ 
    and $\tmfour = \tmfour'\esub{\vartwo}{\tmfive}$
    with $\tmfour' \in \HAbs{\aset'} \cup \Struct{\sset'}$
    where $\aset' = \expansion{\aset}{\esub{\vartwo}{\tmfive}}$ and 
    $\sset' = \expansion{\sset}{\esub{\vartwo}{\tmfive}}$.
    In particular (1) $\vartwo \notin \fv{\tmthree} \cup \fv{\sctx}$, 
    since $\var \neq \vartwo$ by $\alpha$-conversion.
    There are two cases for reducing 
    $((\lam{\var}{\tmthree})\sctx \, \tmfour')\esub{\vartwo}{\tmfive}$, 
    depending on whether $\tmfive \in \HAbs{\aset}$ or $\tmfive \in \Struct{\sset}$;
    both are analogous, so we show only the first one.
    Then we obtain
    $\tmtwo_0 = ((\lam{\var}{\tmthree})\sctx \, \tmfour')\esub{\vartwo}{\tmfive} 
    \tostable{\ruledb}{\aset}{\sset}{\appflag}
    \tmthree\esub{\var}{\tmfour'}\sctx\esub{\vartwo}{\tmfive} = \tmtwo_1$
    by applying rules $\ruleUEsLAbs$ and $\ruleUDbStable$,
        since we already know 
        $\tmfour = \tmfour'\esub{\vartwo}{\tmfive} \in \HAbs{\aset} \cup \Struct{\sset}$,
        so it is easy to note in particular 
        $\tmfour' \in \HAbs{\aset \cup \set{\vartwo}} \cup \Struct{\sset}$.
    Then $\tm_1 = \tmthree\esub{\var}{\tmfour'\esub{\vartwo}{\tmfive}}\sctx 
    \equivC \tmthree\esub{\var}{\tmfour'}\esub{\vartwo}{\tmfive}\sctx = \tm'_1$ 
    by applying rules $\ruleEquivCongES$ and $\ruleEquivEsAssoc$, and since
    $\vartwo \notin \fv{\tmthree}$ by (1).
    We conclude $\tm'_1 \equiv \tmtwo_1$ by rules $\ruleEquivCongES$ and 
    $\ruleEquivEsComm$, since $\vartwo \notin \fv{\sctx}$ by (1), 
    and $\domSctx{\sctx} \disj \fv{\tmfive}$ by $\alpha$-conversion, thus
    $\tm_1 \equiv \tmtwo_1$.
    The following diagram summarises the proof:
    \[
      \xymatrixrowsep{0.3pc}
      \xymatrix{
        \tm_0 = (\lam{\var}{\tmthree})\sctx \, \tmfour'\esub{\vartwo}{\tmfive}
          \arUrStable{\ruledb}{\aset}{\sset}{\appflag}
          % \ar@3{-}[dd]_>{c}
      & \tmthree\esub{\var}{\tmfour'\esub{\vartwo}{\tmfive}}\sctx = \tm_1
          % \ar@3{.}[d]
      \\
      & \equiv
      \\
        \equivC
      & \tmthree\esub{\var}{\tmfour'}\esub{\vartwo}{\tmfive}\sctx = \tm'_1
          % \ar@3{.}[d]
      \\
      & \equiv
      \\
        \tmtwo_0 = ((\lam{\var}{\tmthree})\sctx \, \tmfour')\esub{\vartwo}{\tmfive}
          \arsdUrStable{\ruledb}{\aset}{\sset}{\appflag}
      & \tmthree\esub{\var}{\tmfour'}\sctx\esub{\vartwo}{\tmfive} = \tmtwo_1
      }
    \]
  \item $\ruleEquivCongApp$.
    Then $\tm_0 = (\lam{\var}{\tmthree})\sctx \, \tmfour \equivC 
    \tmfive \, \tmfour' = \tmtwo_0$, where 
    (1) $(\lam{\var}{\tmthree})\sctx \equivC \tmfive$, and $\tmfour \equivC \tmfour'$.
    We split into subcases, according to the rule used to derive the equivalence (1).
    Note that the only relevant cases are the axioms of the relation $\equivC$,
    since cases $\ruleEquivEsLDist$ and $\ruleEquivEsRDist$ are impossible due to the 
    form of $(\lam{\var}{\tmthree})\sctx$.
    \begin{itemize}[leftmargin=*]
    \item $\ruleEquivEsComm$.
      Then 
      $(\lam{\var}{\tmthree})\sctx = 
      (\lam{\var}{\tmthree})\sctx'\esub{\vartwo_1}{\tmsix_1}\esub{\vartwo_2}{\tmsix_2}
      \equivC
      (\lam{\var}{\tmthree})\sctx'\esub{\vartwo_2}{\tmsix_2}\esub{\vartwo_1}{\tmsix_1}
      = \tmfive$,
      where $\vartwo_1 \notin \fv{\tmsix_2}$ and $\vartwo_2 \notin \fv{\tmsix_1}$.
      Therefore 
      $\tm_0 = (\lam{\var}{\tmthree})\sctx'\esub{\vartwo_1}{\tmsix_1}\esub{\vartwo_2}{\tmsix_2} \, \tmfour
      \equivC
      (\lam{\var}{\tmthree})\sctx'\esub{\vartwo_2}{\tmsix_2}\esub{\vartwo_1}{\tmsix_1} \, \tmfour'
      = \tmtwo_0$, and $\tmtwo_0$ $(\ruledb,\aset,\sset,\appflag)$-reduces to
      $\tmtwo_1 = \tmthree\esub{\var}{\tmfour'}\sctx'\esub{\vartwo_2}{\tmsix_2}\esub{\vartwo_1}{\tmsix_1}$
      by rule $\ruleUDbStable$, since $\tmfour' \in \HAbs{\aset} \cup \Struct{\sset}$
      by \cref{lem:habs_structures_closed_by_equivalence}
      given that $\tmfour \in \HAbs{\aset} \cup \Struct{\sset}$.
      From $\tm_1 = \tmthree\esub{\var}{\tmfour}\sctx'\esub{\vartwo_1}{\tmsix_1}\esub{\vartwo_2}{\tmsix_2}$,
      we build the derivation $\deriv$ of
      $\tmtwo_1 \equiv
      \tmthree\esub{\var}{\tmfour'}\sctx'\esub{\vartwo_1}{\tmsix_1}\esub{\vartwo_2}{\tmsix_2}
      = \tm'_1$:
      \[
        \deriv \defeq \left(
          \indrule{\ruleEquivCongES}{
            \indrule{\ruleEquivCongES}{
              \deriv'
              \indrule{\ruleEquivRefl}{
                \emptyPremise
              }{
                \tmsix_1 \equiv \tmsix_1
              }
            }{
              \tmthree\esub{\var}{\tmfour}\sctx'\esub{\vartwo_1}{\tmsix_1}
              \equiv
              \tmthree\esub{\var}{\tmfour'}\sctx'\esub{\vartwo_1}{\tmsix_1}
            }
            \indrule{\ruleEquivRefl}{
              \emptyPremise
            }{
              \tmsix_2 \equiv \tmsix_2
            }
          }{
            \tm_1 \equiv \tm'_1
          }
        \right)
      \]
      where
      \[
        \deriv' \defeq \left(
          \indrule{\ruleEquivCongES}{
            \indrule{\ruleEquivCongES}{
              \indrule{\ruleEquivRefl}{
                \emptyPremise
              }{
                \tmthree \equiv \tmthree
              }
              \indrule{}{
                \text{By hypothesis}
              }{
                \tmfour \equiv \tmfour'
              }
            }{
              \tmthree\esub{\var}{\tmfour}
              \equiv
              \tmthree\esub{\var}{\tmfour'}
            }
            \hdots
          }{
            \indrule{\ruleEquivCongES}{
              \hdots
            }{
              \tmthree\esub{\var}{\tmfour}\sctx'
              \equiv
              \tmthree\esub{\var}{\tmfour'}\sctx'
            }
          }
        \right)
      \]
      \sloppy
      and since $\vartwo_1 \notin \fv{\tmsix_2}$ and $\vartwo_2 \notin \fv{\tmsix_1}$,
      we can apply rule $\ruleEquivEsComm$, yielding
      $\tm'_1 \equiv \tmthree\esub{\var}{\tmfour'}\sctx'\esub{\vartwo_2}{\tmsix_2}\esub{\vartwo_1}{\tmsix_1} = \tmtwo_1$,
      thus $\tm_1 \equiv \tmtwo_1$.
      The following diagram summarises the proof:
      \[
        \xymatrixrowsep{0.3pc}
        \xymatrix{
          \tm_0 = (\lam{\var}{\tmthree})\sctx'\esub{\vartwo_1}{\tmsix_1}\esub{\vartwo_2}{\tmsix_2} \, \tmfour
            \arUrStable{\ruledb}{\aset}{\sset}{\appflag}
            % \ar@3{-}[dd]_>{c}
        & \tmthree\esub{\var}{\tmfour}\sctx'\esub{\vartwo_1}{\tmsix_1}\esub{\vartwo_2}{\tmsix_2} = \tm_1
            % \ar@3{.}[d]
        \\
        & \equiv
        \\
          \equivC
        & \tmthree\esub{\var}{\tmfour'}\sctx'\esub{\vartwo_1}{\tmsix_1}\esub{\vartwo_2}{\tmsix_2} = \tm'_1
            % \ar@3{.}[d]
        \\
        & \equiv
        \\
          \tmtwo_0 = (\lam{\var}{\tmthree})\sctx'\esub{\vartwo_2}{\tmsix_2}\esub{\vartwo_1}{\tmsix_1} \, \tmfour'
            \arsdUrStable{\ruledb}{\aset}{\sset}{\appflag}
        & \tmthree\esub{\var}{\tmfour'}\sctx'\esub{\vartwo_2}{\tmsix_2}\esub{\vartwo_1}{\tmsix_1} = \tmtwo_1
        }
      \]
    \item $\ruleEquivEsAssoc(1)$.
      Then 
      $(\lam{\var}{\tmthree})\sctx = 
      (\lam{\var}{\tmthree})\sctx'\esub{\vartwo_1}{\tmsix_1}\esub{\vartwo_2}{\tmsix_2}
      \equivC 
      (\lam{\var}{\tmthree})\sctx'\esub{\vartwo_1}{\tmsix_1\esub{\vartwo_2}{\tmsix_2}}
      = \tmfive$,
      where $\vartwo_2 \notin \fv{(\lam{\var}{\tmthree})\sctx'}$.
      Hence 
      $\tm_0 = (\lam{\var}{\tmthree})\sctx'\esub{\vartwo_1}{\tmsix_1}\esub{\vartwo_2}{\tmsix_2} \, \tmfour
      \equivC (\lam{\var}{\tmthree})\sctx'\esub{\vartwo_1}{\tmsix_1\esub{\vartwo_2}{\tmsix_2}} \, \tmfour' 
      = \tmtwo_0$,
      and $\tmtwo_0$ $(\ruledb,\aset,\sset,\appflag)$-reduces to
      $\tmtwo_1 = \tmthree\esub{\var}{\tmfour'}\sctx'\esub{\vartwo_1}{\tmsix_1\esub{\vartwo_2}{\tmsix_2}}$
      by rule $\ruleUDbStable$, given that $\tmfour' \in \HAbs{\aset} \cup \Struct{\sset}$
      by \cref{lem:habs_structures_closed_by_equivalence}
      since $\tmfour \in \HAbs{\aset} \cup \Struct{\sset}$.
      We have $\tm_1 \equiv
      \tmthree\esub{\var}{\tmfour'}\sctx'\esub{\vartwo_1}{\tmsix_1}\esub{\vartwo_2}{\tmsix_2} = \tm'_1$, 
      obtained from the derivation $\deriv$ from the previous subcase. 
      To conclude $\tm'_1 \equiv \tmtwo_1$ by rule $\ruleEquivEsAssoc$ since
      $\vartwo_2 \notin \fv{\tmthree\esub{\var}{\tmfour'}\sctx'}$: on one hand
      $\vartwo_2 \notin \fv{(\lam{\var}{\tmthree})\sctx'}$ by hypothesis and
      on the other hand
      $\vartwo_2 \notin \fv{\tmfour'}$ by $\alpha$-conversion.
      Thus $\tm_1 \equiv \tmtwo_2$.
      The following diagram summarises the proof:
      \[
        \xymatrixrowsep{0.3pc}
        \xymatrix{
          \tm_0 = (\lam{\var}{\tmthree})\sctx'\esub{\vartwo_1}{\tmsix_1}\esub{\vartwo_2}{\tmsix_2} \, \tmfour
            \arUrStable{\ruledb}{\aset}{\sset}{\appflag}
            % \ar@3{-}[dd]_>{c}
        & \tmthree\esub{\var}{\tmfour}\sctx'\esub{\vartwo_1}{\tmsix_1}\esub{\vartwo_2}{\tmsix_2} = \tm_1
            % \ar@3{.}[d]
        \\
        & \equiv
        \\
          \equivC
        & \tmthree\esub{\var}{\tmfour'}\sctx'\esub{\vartwo_1}{\tmsix_1}\esub{\vartwo_2}{\tmsix_2} = \tm'_1
            % \ar@3{.}[d]
        \\
        & \equiv
        \\
          \tmtwo_0 = (\lam{\var}{\tmthree})\sctx'\esub{\vartwo_1}{\tmsix_1\esub{\vartwo_2}{\tmsix_2}} \, \tmfour'
            \arsdUrStable{\ruledb}{\aset}{\sset}{\appflag}
        & \tmthree\esub{\var}{\tmfour'}\sctx'\esub{\vartwo_1}{\tmsix_1\esub{\vartwo_2}{\tmsix_2}} = \tmtwo_1
        }
      \]
    \item $\ruleEquivEsAssoc(2)$.
      Then 
      $(\lam{\var}{\tmthree})\sctx = 
      (\lam{\var}{\tmthree})\sctx'\esub{\vartwo_1}{\tmsix_1\esub{\vartwo_2}{\tmsix_2}}
      \equivC 
      (\lam{\var}{\tmthree})\sctx'\esub{\vartwo_1}{\tmsix_1}\esub{\vartwo_2}{\tmsix_2}
      = \tmfive$,
      where $\vartwo_2 \notin \fv{(\lam{\var}{\tmthree})\sctx'}$.
      The steps are analogous to the previous case,
      with following diagram summarising the proof:
      \[
        \xymatrixrowsep{0.3pc}
        \xymatrix{
          \tm_0 = (\lam{\var}{\tmthree})\sctx'\esub{\vartwo_1}{\tmsix_1\esub{\vartwo_2}{\tmsix_2}} \, \tmfour
            \arUrStable{\ruledb}{\aset}{\sset}{\appflag}
            % \ar@3{-}[dd]_>{c}
        & \tmthree\esub{\var}{\tmfour}\sctx'\esub{\vartwo_1}{\tmsix_1\esub{\vartwo_2}{\tmsix_2}} = \tm_1
            % \ar@3{.}[d]_>{c}
        \\
        & \equiv
        \\
          \equivC
        & \tmthree\esub{\var}{\tmfour'}\sctx'\esub{\vartwo_1}{\tmsix_1}\esub{\vartwo_2}{\tmsix_2} = \tm'_1
            % \ar@3{.}[d]
        \\
        & \equiv
        \\
          \tmtwo_0 = (\lam{\var}{\tmthree})\sctx'\esub{\vartwo_1}{\tmsix_1}\esub{\vartwo_2}{\tmsix_2} \, \tmfour'
            \arsdUrStable{\ruledb}{\aset}{\sset}{\appflag}
        & \tmthree\esub{\var}{\tmfour'}\sctx'\esub{\vartwo_1}{\tmsix_1}\esub{\vartwo_2}{\tmsix_2} = \tmtwo_1
        }
      \]
    \item $\ruleEquivCongES$.
      Then $(\lam{\var}{\tmthree})\sctx = (\lam{\var}{\tmthree})\sctx'\esub{\vartwo}{\tmfive}$,
      and $(\lam{\var}{\tmthree})\sctx'\esub{\vartwo}{\tmfive} \equivC \tmsix'\esub{\vartwo}{\tmfive'}$,
      with $(\lam{\var}{\tmthree})\sctx' \equivC \tmsix'$ and $\tmfive \equivC \tmfive'$.
      Then this case is not possible since we need to use rule $\ruleEquivRefl$, 
      by \cref{rem:t_equiv_vL_is_vpLp}.
    \end{itemize}
  \end{enumerate}
\item $\ruleULsv$.
  \sloppy
  Then $\tm_0 = \tmthree_0\esub{\var}{\val\sctx} \tostable{\rulelsv}{\aset}{\sset}{\appflag}
  \tmthree_1\esub{\var}{\val}\sctx = \tm_1$, where $\rulename = \rulelsv$,
  and it is derived from
  $\tmthree_0 \tostable{\rulesub{\var}{\val}}{\aset \cup \set{\var}}{\sset}{\appflag} \tmthree_1$
  and $\var \notin \aset \cup \sset$, and $\val\sctx \in \HAbs{\aset}$.
  We also have $\tm_0 = \tmthree\esub{\var}{\val\sctx} \equivC \tmtwo_0$.
  We analyze the different cases according to which rule was used to derive the equivalence.
  Note that case $\ruleEquivCongApp$ is impossible due to the form of $\tm_0$.
  \begin{enumerate}[leftmargin=*, label=\alph*.]
  \item $\ruleEquivEsComm$.
    Then $\tm_0 = \tmthree'_0\esub{\vartwo}{\tmfour}\esub{\var}{\val\sctx} \equivC
    \tmthree'_0\esub{\var}{\val\sctx}\esub{\vartwo}{\tmfour} = \tmtwo_0$, 
    with $\vartwo \notin \fv{\val\sctx}$ and $\var \notin \fv{\tmfour}$.
    The step
    $\tmthree'_0\esub{\vartwo}{\tmfour} 
    \tostable{\rulesub{\var}{\val}}{\aset \cup \set{\var}}{\sset}{\appflag} \tmthree_1$
    can be derived either from rule $\ruleUEsR$, $\ruleUEsLAbs$ or $\ruleUEsLStruct$:
    \begin{itemize}[leftmargin=*]
    \item $\ruleUEsR$.
      Then
      $\tmthree_1 = \tmthree'_0\esub{\vartwo}{\tmfour'}$,
      with
      $\tmfour \tostable{\rulesub{\var}{\val}}{\aset \cup \set{\var}}{\sset}{\nonapp}\tmfour'$.
      By \cref{rem:t_reduces_with_subvar_var_occurs_free_in_t},
      we have $\var \in \fv{\tmfour}$, but at the same time $\var \notin \fv{\tmfour}$
      by hypothesis of rule $\ruleEquivEsComm$.
      Therefore this case is not possible.
    \item $\ruleUEsLAbs$.
      Then $\tmthree_1 = \tmthree''_0\esub{\vartwo}{\tmfour}$, with 
      $\tmthree'_0 
      \tostable{\rulesub{\var}{\val}}{\aset \cup \set{\var} \cup \set{\vartwo}}{\sset}{\appflag} 
      \tmthree''_0$.
      Then
      $\tmtwo_0 = \tmthree'_0\esub{\var}{\val\sctx}\esub{\vartwo}{\tmfour} 
      \tostable{\rulelsv}{\aset}{\sset}{\appflag} 
      \tmthree''_0\esub{\var}{\val}\sctx\esub{\vartwo}{\tmfour} = \tmtwo_1$
      derived from rules $\ruleUEsLAbs$ and $\ruleULsv$.
      Since $\vartwo \notin \fv{\val\sctx}$, we can apply several times
      rules $\ruleEquivCongES$ and $\ruleEquivEsComm$, yielding 
      $\tmtwo_1 \equiv \tm_1$. The following diagram summarises the proof:
      \[
        \xymatrixrowsep{0.3pc}
        \xymatrix{
          \tm_0 = \tmthree'_0\esub{\vartwo}{\tmfour}\esub{\var}{\val\sctx}
            \arUrStable{\rulelsv}{\aset}{\sset}{\appflag}
            % \ar@3{-}[d]_>{c}
        & \tmthree''_0\esub{\vartwo}{\tmfour}\esub{\var}{\val}\sctx = \tm_1
            % \ar@3{.}[d]
        \\
          \equivC
        & \equiv
        \\
          \tmtwo_0 = \tmthree'_0\esub{\var}{\val\sctx}\esub{\vartwo}{\tmfour}
            \arsdUrStable{\rulelsv}{\aset}{\sset}{\appflag}
        & \tmthree''_0\esub{\var}{\val}\sctx\esub{\vartwo}{\tmfour} = \tmtwo_1
        }
      \]
    \item $\ruleUEsLStruct$.
      Analogous to the previous case.
    \end{itemize}
  \item $\ruleEquivEsAssoc(1)$.
    \sloppy
    Then $\tm_0 = \tmthree'_0\esub{\vartwo}{\tmfour}\esub{\var}{\val\sctx}
    \equivC \tmthree'_0\esub{\vartwo}{\tmfour\esub{\var}{\val\sctx}} = \tmtwo_0$,
    with $\var \notin \fv{\tmthree'_0}$.
    The step
    $\tmthree'_0\esub{\vartwo}{\tmfour} 
    \tostable{\rulesub{\var}{\val}}{\aset \cup \set{\var}}{\sset}{\appflag} \tmthree_1$
    can be derived either from rule $\ruleUEsR$, $\ruleUEsLAbs$ or $\ruleUEsLStruct$:
    \begin{itemize}[leftmargin=*]
    \item $\ruleUEsR$.
      Then
      $\tmthree_1 = \tmthree'_0\esub{\vartwo}{\tmfour'}$,
      with
      $\tmfour \tostable{\rulesub{\var}{\val}}{\aset \cup \set{\var}}{\sset}{\nonapp} \tmfour'$.
      An we can perform the reduction step
      $\tmtwo_0 = \tmthree'_0\esub{\vartwo}{\tmfour\esub{\var}{\val\sctx}}
      \tostable{\rulelsv}{\aset}{\sset}{\appflag} 
      \tmthree'_0\esub{\vartwo}{\tmfour'\esub{\var}{\val}\sctx} = \tmtwo_1$,
      derived from rules $\ruleUEsR$ and $\ruleULsv$.
      Then $\tm_1 = \tmthree'_0\esub{\vartwo}{\tmfour'}\esub{\var}{\val}\sctx 
      \equivC \tmthree'_0\esub{\vartwo}{\tmfour'\esub{\var}{\val}}\sctx = \tm'_1$
      by rules $\ruleEquivCongES$ and $\ruleEquivEsAssoc$, since 
      $\var \notin \fv{\tmthree'_0}$ by hypothesis.
      We conclude $\tm'_1 \equiv \tmtwo_1$ by applying several times rules
      $\ruleEquivCongES$ and $\ruleEquivEsAssoc$,
      given that we may assume $\domSctx{\sctx} \disj \fv{\tmthree'_0}$ 
      by $\alpha$-conversion. Thus $\tm_1 \equiv \tmtwo_1$.
      The following diagram summarises the proof:
      \[
        \xymatrixrowsep{0.3pc}
        \xymatrix{
          \tm_0 = \tmthree'_0\esub{\vartwo}{\tmfour}\esub{\var}{\val\sctx}
            \arUrStable{\rulelsv}{\aset}{\sset}{\appflag}
            % \ar@3{-}[dd]_>{c}
        & \tmthree'_0\esub{\vartwo}{\tmfour'}\esub{\var}{\val}\sctx = \tm_1
            % \ar@3{.}[d]
        \\
        & \equiv
        \\
          \equivC
        & \tmthree'_0\esub{\vartwo}{\tmfour'\esub{\var}{\val}}\sctx = \tm'_1
            % \ar@3{.}[d]
        \\
        & \equiv
        \\
          \tmtwo_0 = \tmthree'_0\esub{\vartwo}{\tmfour\esub{\var}{\val\sctx}}
            \arsdUrStable{\rulelsv}{\aset}{\sset}{\appflag}
        & \tmthree'_0\esub{\vartwo}{\tmfour'\esub{\var}{\val}\sctx} = \tmtwo_1
        }
      \]
    \item $\ruleUEsLAbs$.
      Then $\tmthree_1 = \tmthree''_0\esub{\vartwo}{\tmfour}$, with $\tmthree'_0 
      \tostable{\rulesub{\var}{\val}}{\aset \cup \set{\var}}{\sset}{\appflag}
      \tmthree''_0$.
      By \cref{rem:t_reduces_with_subvar_var_occurs_free_in_t}, we have 
      $\var \in \fv{\tmthree'_0}$, but at the same time 
      $\var \notin \fv{\tmthree'_0}$ by hypothesis of rule $\ruleEquivEsAssoc$.
      Therefore this case is not possible.
    \item $\ruleUEsLStruct$.
      This case is analogous to the previous one, hence it is impossible.
    \end{itemize}
  \item $\ruleEquivEsAssoc(2)$.
    Then $\tm_0 = \tmthree_0\esub{\var}{\val\sctx'\esub{\vartwo}{\tmfour}}
    \equivC \tmthree_0\esub{\var}{\val\sctx'}\esub{\vartwo}{\tmfour} = \tmtwo_0$,
    with $\vartwo \notin \fv{\tmthree_0}$ and $\sctx = \sctx'\esub{\vartwo}{\tmfour}$.
    The predicate $\val\sctx'\esub{\vartwo}{\tmfour} \in \HAbs{\aset}$ from the
    premise of rule $\ruleULsv$ can be derived either by rule $\ruleHAbsSubi$ or 
    $\ruleHAbsSubii$, so we have two subcases to derive 
    $\tmtwo_0 \tostable{\rulelsv}{\aset}{\sset}{\appflag} 
    \tmthree_1\esub{\var}{\val}\sctx'\esub{\vartwo}{\tmfour} = \tmtwo_1$.
    Since both are analogous, we only focus on the case where
    $\val\sctx'\esub{\vartwo}{\tmfour} \in \HAbs{\aset}$ is derived by 
    rule $\ruleHAbsSubii$, with $\val\sctx' \in \HAbs{\aset \cup \set{\vartwo}}$
    and $\tmfour \in \HAbs{\aset}$ and $\vartwo \notin \aset$.
    Hence we can build a derivation of
    $\tmtwo_0 \tostable{\rulelsv}{\aset}{\sset}{\appflag} \tmtwo_1$ by applying
    rules $\ruleUEsLAbs$ and $\ruleULsv$. % tambien se usa el lema \cref{lem:weakening_of_reducion_sets}.
    To conclude, we have $\tm_1 \equiv \tmtwo_1$ by applying rule $\ruleEquivRefl$.
    The following diagram summarises the proof:
      \[
        \xymatrixrowsep{0.3pc}
        \xymatrix{
          \tm_0 = \tmthree_0\esub{\var}{\val\sctx'\esub{\vartwo}{\tmfour}}
            \arUrStable{\rulelsv}{\aset}{\sset}{\appflag}
            % \ar@3{-}[d]_>{c}
        & \tmthree_1\esub{\var}{\val}\sctx'\esub{\vartwo}{\tmfour} = \tm_1
            % \ar@3{.}[d]
        \\
          \equivC
        & \equiv
        \\
          \tmtwo_0 = \tmthree_0\esub{\var}{\val\sctx'}\esub{\vartwo}{\tmfour}
            \arsdUrStable{\rulelsv}{\aset}{\sset}{\appflag}
        & \tmthree_1\esub{\var}{\val}\sctx'\esub{\vartwo}{\tmfour} = \tmtwo_1
        }
      \]
  \item $\ruleEquivEsLDist(1)$.
    \sloppy
    Then $\tm_0 = (\tmfour_0 \, \tmfour_1)\esub{\var}{\val\sctx} \equivC
    \tmfour_0\esub{\var}{\val\sctx} \, \tmfour_1 = \tmtwo_0$,
    with $\var \notin \fv{\tmfour_1}$.
    The step $\tmfour_0 \, \tmfour_1 
    \tostable{\rulesub{\var}{\val}}{\aset \cup \set{\var}}{\sset}{\appflag}
    \tmthree_1$ can be derived from rules $\ruleUAppL$ and $\ruleUAppR$:
    \begin{itemize}[leftmargin=*]
    \item $\ruleUAppL$.
      \sloppy
      Then $\tmthree_1 = \tmfour'_0 \, \tmfour_1$,
      where $\tmfour_0 \tostable{\rulesub{\var}{\val}}{\aset \cup \set{\var}}{\sset}{\appflag} \tmfour'_0$
      with $\tm_1 = (\tmfour'_0 \, \tmfour_1)\esub{\var}{\val}\sctx$,
      and $\tmtwo_0 = \tmfour_0\esub{\var}{\val\sctx} \, \tmfour_1
      \tostable{\rulelsv}{\aset}{\sset}{\appflag} 
      \tmfour'_0\esub{\var}{\val}\sctx \, \tmfour_1 = \tmtwo_1$
      by rules $\ruleUAppL$ and $\ruleULsv$.
      Since $\var \notin \fv{\tmfour_1}$ by hypothesis and
      $\domSctx{\sctx} \disj \fv{\tmfour_1}$ by $\alpha$-conversion,
      we can then apply rule $\ruleEquivEsLDist$ to conclude
      $\tm_1 \equiv \tmtwo_1$.
      The following diagram summarises the proof:
      \[
        \xymatrixrowsep{0.3pc}
        \xymatrix{
          \tm_0 = (\tmfour_0 \, \tmfour_1)\esub{\var}{\val\sctx}
            \arUrStable{\rulelsv}{\aset}{\sset}{\appflag}
            % \ar@3{-}[d]_>{c}
        & (\tmfour'_0 \, \tmfour_1)\esub{\var}{\val}\sctx = \tm_1
            % \ar@3{.}[d]
        \\
          \equivC
        & \equiv
        \\
          \tmtwo_0 = \tmfour_0\esub{\var}{\val\sctx} \, \tmfour_1
            \arsdUrStable{\rulelsv}{\aset}{\sset}{\appflag}
        & \tmfour'_0\esub{\var}{\val}\sctx \, \tmfour_1 = \tmtwo_1
        }
      \]
    \item $\ruleUAppR$.
      Then $\tmthree_1 = \tmfour_0 \, \tmfour'_1$,
      derived from $\tmfour_0 \in \Struct{\sset}$ and
      $\tmfour_1
      \tostable{\rulesub{\var}{\val}}{\aset \cup \set{\var}}{\sset}{\nonapp} 
      \tmfour'_1$.
      By \cref{rem:t_reduces_with_subvar_var_occurs_free_in_t}, we have 
      $\var \in \fv{\tmfour_1}$, but at the same time 
      $\var \notin \fv{\tmfour_1}$ by hypothesis of rule $\ruleEquivEsLDist$.
      Therefore this case is not possible.
    \end{itemize}
  \item $\ruleEquivEsRDist(1)$.
    Analogous to the previous case.
  \item $\ruleEquivCongES$.
    Then $\tm_0 = \tmthree_0\esub{\var}{\val\sctx} \equivC 
    \tmthree'_0\esub{\var}{\tmfour} = \tmtwo_0$, derived from
    $\tmthree_0 \equivC \tmthree'_0$ and $\val\sctx \equivC \tmfour$;
    note that $\tmfour = \val'\sctx'$ by \cref{rem:t_equiv_vL_is_vpLp},
    hence it must be the case that we need to use rule $\ruleEquivRefl$,
    so this case is not possible.
  \end{enumerate}
\item $\ruleUSub$.
  Then $\tm_0 = \var \tostable{\rulesub{\var}{\val}}{\aset' \cup \set{\var}}{\sset}{\app}
  \val = \tm_1$, where $\rulename = \rulesub{\var}{\val}$, 
  $\aset = \aset' \cup \set{\var}$ and $\appflag = \app$.
  We also have $\tm_0 = \var \equivC \tmtwo_0$.
  This case is not possible since there are no rules to derive $\var \equivC \tmtwo_0$.
\HIDDENFRAGMENT{
\item The remaining cases are treated in a similar way.
}{
\item $\ruleUAppL$.
  Then $\tm_0 = \tmthree_0 \, \tmthree_1 
  \tostable{\rulename}{\aset}{\sset}{\appflag} \tmthree'_0 \, \tmthree_1 = \tm_1$
  is derived from $\tmthree_0 \tostable{\rulename}{\aset}{\sset}{\app} \tmthree'_0$.
  Moreover, $\tm_0 = \tmthree_0 \, \tmthree_1 \equivC \tmtwo_0$.
  We analyze the different cases according to which rule was used to derive the equivalence.
  Cases $\ruleEquivEsComm$, $\ruleEquivEsAssoc$ and $\ruleEquivCongES$
  are impossible due to the form of $\tm_0$.
  \begin{enumerate}[leftmargin=*, label=\alph*.]
  \item $\ruleEquivEsLDist(2)$.
    Then $\tmthree_0 = \tmfour_0\esub{\var}{\tmfive} 
    \tostable{\rulename}{\aset}{\sset}{\app} \tmthree'_0$,
    and $\tm_0 \equivC (\tmfour_0 \, \tmthree_1)\esub{\var}{\tmfive} = \tmtwo_0$ 
    with $\var \notin \fv{\tmthree_1}$.
    This reduction step can be derived either from rule
    $\ruleULsv$, $\ruleUEsR$, $\ruleUEsLAbs$ or $\ruleUEsLStruct$:
    \begin{itemize}[leftmargin=*]
    \item $\ruleULsv$.
      Then $\tmfive = \val\sctx$ and $\rulename = \rulelsv$.
      The step
      $\tmthree_0 = \tmfour_0\esub{\var}{\val\sctx} \tostable{\rulelsv}{\aset}{\sset}{\app} 
      \tmfour'_0\esub{\var}{\val}\sctx = \tmthree'_0$ is derived from
      $\tmfour_0 \tostable{\rulesub{\var}{\val}}{\aset \cup \set{\var}}{\sset}{\app} 
      \tmfour'_0$.
      On the other hand, we have
      $\tm_0 = \tmfour_0\esub{\var}{\val\sctx} \, \tmthree_1 \equivC 
      (\tmfour_0 \, \tmthree_1)\esub{\var}{\val\sctx} = \tmtwo_0$.
      Hence $\tmtwo_0 \tostable{\rulelsv}{\aset}{\sset}{\appflag} 
      (\tmfour'_0 \, \tmthree_1)\esub{\var}{\val}\sctx = \tmtwo_1$, 
      derived from rules $\ruleULsv$ and $\ruleUAppL$.
      To conclude, 
      $\tm_1 = \tmfour'_0\esub{\var}{\val}\sctx \, \tmthree_1 \equiv \tmtwo_1$
      by applying rules $\ruleEquivEsLDist$ and 
      $\ruleEquivCongES$ with $\ruleEquivEsLDist$, given 
      $\var \notin \fv{\tmthree_1}$ by hypothesis and 
      $\domSctx{\sctx} \disj \fv{\tmthree_1}$ by $\alpha$-conversion.
      The following diagram summarises the proof:
      \[
        \xymatrixrowsep{0.3pc}
        \xymatrix{
          \tm_0 = \tmfour_0\esub{\var}{\val\sctx} \, \tmthree_1
            \arUrStable{\rulelsv}{\aset}{\sset}{\appflag}
            % \ar@3{-}[d]_>{c}
        & \tmfour'_0\esub{\var}{\val}\sctx \, \tmthree_1 = \tm_1
            % \ar@3{.}[d]
        \\
          \equivC
        & \equiv
        \\
          \tmtwo_0 = (\tmfour_0 \, \tmthree_1)\esub{\var}{\val\sctx}
            \arsdUrStable{\rulelsv}{\aset}{\sset}{\appflag}
        & (\tmfour'_0 \, \tmthree_1)\esub{\var}{\val}\sctx = \tmtwo_1
        }
      \]
    \item $\ruleUEsR$.
      The step
      $\tmthree_0 = \tmfour_0\esub{\var}{\tmfive} \tostable{\rulename}{\aset}{\sset}{\app} 
      \tmfour_0\esub{\var}{\tmfive'} = \tmthree'_0$ is derived from
      $\tmfive \tostable{\rulename}{\aset}{\sset}{\nonapp} \tmfive'$.
      On the other hand, we have
      $\tm_0 = \tmfour_0\esub{\var}{\tmfive} \, \tmthree_1 \equivC 
      (\tmfour_0 \, \tmthree_1)\esub{\var}{\tmfive} = \tmtwo_0$.
      Hence $\tmtwo_0 \tostable{\rulename}{\aset}{\sset}{\appflag} 
      (\tmfour_0 \, \tmthree_1)\esub{\var}{\tmfive'} = \tmtwo_1$, 
      derived from rule $\ruleUEsR$.
      To conclude, $\tm_1 = \tmfour_0\esub{\var}{\tmfive'} \equiv \tmtwo_1$
      by applying rule $\ruleEquivEsLDist$, since
      $\var \notin \fv{\tmthree_1}$ by hypothesis.
      The following diagram summarises the proof:
      \[
        \xymatrixrowsep{0.3pc}
        \xymatrix{
          \tm_0 = \tmfour_0\esub{\var}{\tmfive} \, \tmthree_1
            \arUrStable{\rulename}{\aset}{\sset}{\appflag}
            % \ar@3{-}[d]_>{c}
        & \tmfour_0\esub{\var}{\tmfive'} \, \tmthree_1 = \tm_1
            % \ar@3{.}[d]
        \\
          \equivC
        & \equiv
        \\
          \tmtwo_0 = (\tmfour_0 \, \tmthree_1)\esub{\var}{\tmfive}
            \arsdUrStable{\rulename}{\aset}{\sset}{\appflag}
        & (\tmfour_0 \, \tmthree_1)\esub{\var}{\tmfive'} = \tmtwo_1
        }
      \]
    \item $\ruleUEsLAbs$.
      The step
      $\tmthree_0 = \tmfour_0\esub{\var}{\tmfive} \tostable{\rulename}{\aset}{\sset}{\app} 
      \tmfour'_0\esub{\var}{\tmfive} = \tmthree'_0$ is derived from
      $\tmfour_0 \tostable{\rulename}{\aset \cup \set{\var}}{\sset}{\app} \tmfour'_0$
      and $\tmfive \in \HAbs{\aset}$ and $\var \notin \aset \cup \sset$ and
      $\var \notin \fv{\rulename}$.
      On the other hand, we have
      $\tm_0 = \tmfour_0\esub{\var}{\tmfive} \, \tmthree_1 \equivC 
      (\tmfour_0 \, \tmthree_1)\esub{\var}{\tmfive} = \tmtwo_0$.
      Hence $\tmtwo_0 \tostable{\rulename}{\aset}{\sset}{\appflag} 
      (\tmfour'_0 \, \tmthree_1)\esub{\var}{\tmfive} = \tmtwo_1$, 
      derived from rule $\ruleUEsLAbs$.
      To conclude, $\tm_1 = \tmfour'_0\esub{\var}{\tmfive} \, \tmthree_1 \equiv \tmtwo_1$
      by applying rule $\ruleEquivEsLDist$, since
      $\var \notin \fv{\tmthree_1}$ by hypothesis.
      The following diagram summarises the proof:
      \[
        \xymatrixrowsep{0.3pc}
        \xymatrix{
          \tm_0 = \tmfour_0\esub{\var}{\tmfive} \, \tmthree_1
            \arUrStable{\rulename}{\aset}{\sset}{\appflag}
            % \ar@3{-}[d]_>{c}
        & \tmfour'_0\esub{\var}{\tmfive} \, \tmthree_1 = \tm_1
            % \ar@3{.}[d]
        \\
          \equivC
        & \equiv
        \\
          \tmtwo_0 = (\tmfour_0 \, \tmthree_1)\esub{\var}{\tmfive}
            \arsdUrStable{\rulename}{\aset}{\sset}{\appflag}
        & (\tmfour'_0 \, \tmthree_1)\esub{\var}{\tmfive} = \tmtwo_1
        }
      \]
    \item $\ruleUEsLStruct$.
      Analogous to the previous case.
    \end{itemize}
  \item $\ruleEquivEsRDist(2)$.
    Then $\tmthree_0 
    \tostable{\rulename}{\aset}{\sset}{\app} \tmthree'_0$,
    and $\tm_0 = \tmthree_0 \, \tmfour_1\esub{\var}{\tmfive}
    \equivC (\tmthree_0 \, \tmfour_1)\esub{\var}{\tmfive} = \tmtwo_0$ 
    with $\var \notin \fv{\tmthree_0}$.
    Moreover, by $\alpha$-conversion we can assume $\var \notin \fv{\tmthree'_0}$.
    Then we have
    $\tmtwo_0 \tostable{\rulename}{\aset}{\sset}{\appflag} 
    (\tmthree'_0 \, \tmfour_1)\esub{\var}{\tmfive} = \tmtwo_1$, 
    derived from rule $\ruleUEsLAbs$ or $\ruleUEsLStruct$, depending on whether
    $\tmfive \in \HAbs{\aset}$ or $\tmfive \in \Struct{\sset}$, and from rule $\ruleUAppL$.
    To conclude, $\tm_1 = \tmthree'_0 \, \tmfour_1\esub{\var}{\tmfive} \equiv \tmtwo_1$
    by applying rule $\ruleEquivEsLDist$.
    The following diagram summarises the proof:
    \[
      \xymatrixrowsep{0.3pc}
      \xymatrix{
        \tm_0 = \tmthree_0 \, \tmfour_1\esub{\var}{\tmfive}
          \arUrStable{\rulename}{\aset}{\sset}{\appflag}
          % \ar@3{-}[d]_>{c}
      & \tmthree'_0 \, \tmfour_1\esub{\var}{\tmfive} = \tm_1
          % \ar@3{.}[d]
      \\
        \equivC
      & \equiv
      \\
        \tmtwo_0 = (\tmthree_0 \, \tmfour_1)\esub{\var}{\tmfive}
          \arsdUrStable{\rulename}{\aset}{\sset}{\appflag}
      & (\tmthree'_0 \, \tmfour_1)\esub{\var}{\tmfive} = \tmtwo_1
      }
    \]
  \item $\ruleEquivCongApp$.
    Then $\tm_0 \equivC \tmfour_0 \, \tmfour_1 = \tmtwo_0$, derived from
    $\tmthree_0 \equivC \tmfour_0$ and $\tmthree_1 \equivC \tmfour_1$.
    By \ih on $\tmthree_0$, we have that there exists $\tmfour'_0$ such that
    (1) $\tmfour_0 \tostable{\rulename}{\aset}{\sset}{\app} \tmfour'_0$ and
    $\tmthree'_0 \equiv \tmfour'_0$.
    Hence applying rule $\ruleUAppL$ with (1) as premise, we obtain $\tmtwo_0 
    \tostable{\rulename}{\aset}{\sset}{\appflag} \tmfour'_0 \, \tmfour_1 = \tmtwo_1$
    and by rule $\ruleEquivCongApp$ we conclude $\tm_1 \equiv \tmtwo_1$.
    The following diagram summarises the proof:
    \[
      \xymatrixrowsep{0.3pc}
      \xymatrix{
        \tm_0 = \tmthree_0 \, \tmthree_1
          \arUrStable{\rulename}{\aset}{\sset}{\appflag}
          % \ar@3{-}[d]_>{c}
      & \tmthree'_0 \, \tmthree_1 = \tm_1
          % \ar@3{.}[d]
      \\
        \equivC
      & \equiv
      \\
        \tmtwo_0 = \tmfour_0 \, \tmfour_1
          \arsdUrStable{\rulename}{\aset}{\sset}{\appflag}
      & \tmfour'_0 \, \tmfour_1 = \tmtwo_1
      }
    \]
  \end{enumerate}
\item $\ruleUAppR$.
  Analogous to the previous case.
\item $\ruleUEsR$.
  Then $\tm_0 = \tmthree_0\esub{\var}{\tmfour_0} 
  \tostable{\rulename}{\aset}{\sset}{\appflag} \tmthree_0\esub{\var}{\tmfour_1} = \tm_1$,
  which is derived from 
  $\tmfour_0 \tostable{\rulename}{\aset}{\sset}{\nonapp} \tmfour_1$.
  We also have $\tm_0 \equivC \tmtwo_0$, so we analyze the different cases
  according to which rule was used to derive the equivalence.
  Case $\ruleEquivCongApp$ is impossible due to the form of $\tm_0$.
  \begin{enumerate}[leftmargin=*, label=\alph*.]
  \item $\ruleEquivEsComm$.
    Then $\tmthree_0 = \tmthree\esub{\vartwo}{\tmfive}$,
    and $\tm_0 = \tmthree\esub{\vartwo}{\tmfive}\esub{\var}{\tmfour_0} \equivC
    \tmthree\esub{\var}{\tmfour_0}\esub{\vartwo}{\tmfive} = \tmtwo_0$, where
    $\var \notin \fv{\tmfive}$ and $\vartwo \notin \fv{\tmfour_0}$.
    Moreover, we can assume $\vartwo \notin \fv{\tmfour_1}$ by $\alpha$-conversion.
    We can derive $\tmtwo_0 \tostable{\rulename}{\aset}{\sset}{\appflag}
    \tmthree\esub{\var}{\tmfour_1}\esub{\vartwo}{\tmfive} = \tmtwo_1$
    from rule $\ruleUEsLAbs$ or $\ruleUEsLStruct$ depending on whether
    $\tmfive \in \HAbs{\aset}$ or $\tmfive \in \Struct{\sset}$, respectively,
    and from rule $\ruleUEsR$.
    To conclude, $\tm_1 \equiv \tmtwo_1$ by rule $\ruleEquivEsComm$,
    since $\var \notin \fv{\tmfive}$ by hypothesis.
    The following diagram summarises the proof:
    \[
      \xymatrixrowsep{0.3pc}
      \xymatrix{
        \tm_0 = \tmthree\esub{\vartwo}{\tmfive}\esub{\var}{\tmfour_0}
          \arUrStable{\rulename}{\aset}{\sset}{\appflag}
          % \ar@3{-}[d]_>{c}
      & \tmthree\esub{\vartwo}{\tmfive}\esub{\var}{\tmfour_1} = \tm_1
          % \ar@3{.}[d]_>{c}
      \\
        \equivC
      & \equiv
      \\
        \tmtwo_0 = \tmthree\esub{\var}{\tmfour_0}\esub{\vartwo}{\tmfive}
          \arsdUrStable{\rulename}{\aset}{\sset}{\appflag}
      & \tmthree\esub{\var}{\tmfour_1}\esub{\vartwo}{\tmfive} = \tmtwo_1
      }
    \]
  \item $\ruleEquivEsAssoc$ (1).
    Then $\tmthree_0 = \tmthree\esub{\vartwo}{\tmfive}$,
    and $\tm_0 = \tmthree\esub{\vartwo}{\tmfive}\esub{\var}{\tmfour_0} \equivC
    \tmthree\esub{\vartwo}{\tmfive\esub{\var}{\tmfour_0}} = \tmtwo_0$, with
    $\var \notin \fv{\tmthree}$.
    We can derive 
    $\tmtwo_0 \tostable{\rulename}{\aset}{\sset}{\appflag} 
    \tmthree\esub{\vartwo}{\tmfive\esub{\var}{\tmfour_1}} = \tmtwo_1$
    from rule $\ruleUEsR$.
    To conclude, $\tm_1 \equiv \tmtwo_1$ by rule $\ruleEquivEsAssoc$,
    since $\var \notin \fv{\tmthree}$ by hypothesis.
    The following diagram summarises the proof:
    \[
      \xymatrixrowsep{0.3pc}
      \xymatrix{
        \tm_0 = \tmthree\esub{\vartwo}{\tmfive}\esub{\var}{\tmfour_0}
          \arUrStable{\rulename}{\aset}{\sset}{\appflag}
          % \ar@3{-}[d]_>{c}
      & \tmthree\esub{\vartwo}{\tmfive}\esub{\var}{\tmfour_1} = \tm_1
          % \ar@3{.}[d]_>{c}
      \\
        \equivC
      & \equiv
      \\
        \tmtwo_0 = \tmthree\esub{\vartwo}{\tmfive\esub{\var}{\tmfour_0}}
          \arsdUrStable{\rulename}{\aset}{\sset}{\appflag}
      & \tmthree\esub{\vartwo}{\tmfive\esub{\var}{\tmfour_1}} = \tmtwo_1
      }
    \]
  \item $\ruleEquivEsAssoc$ (2).
    Then $\tmfour_0 = \tmfour\esub{\vartwo}{\tmfive}
    \tostable{\rulename}{\aset}{\sset}{\nonapp} \tmfour_1$, and 
    $\tm_0 = \tmthree_0\esub{\var}{\tmfour\esub{\vartwo}{\tmfive}} \equivC
    \tmthree_0\esub{\var}{\tmfour}\esub{\vartwo}{\tmfive} = \tmtwo_0$, 
    with $\vartwo \notin \fv{\tmthree_0}$.
    We can reduce $\tmfour_0$ to $\tmfour_1$ either by rule $\ruleULsv$, 
    $\ruleUEsR$, $\ruleUEsLAbs$ or $\ruleUEsLStruct$:
    \begin{itemize}[leftmargin=*]
    \item $\ruleULsv$.
      Then $\tmfive = \val\sctx$ and $\rulename = \rulelsv$,
      and the reduction step is derived from
      $\tmfour 
      \tostable{\rulesub{\vartwo}{\val}}{\aset \cup \set{\vartwo}}{\sset}{\nonapp}
      \tmfour'$ and $\val\sctx \in \HAbs{\aset}$ and $\vartwo \notin \aset\cup\sset$.
      On the other hand, we have 
      $\tm_0 = \tmthree_0\esub{\var}{\tmfour\esub{\vartwo}{\val\sctx}} \equivC
      \tmthree_0\esub{\var}{\tmfour}\esub{\vartwo}{\val\sctx} = \tmtwo_0$.
      Then we can derive $\tmtwo_0 \tostable{\rulelsv}{\aset}{\sset}{\appflag} 
      \tmthree_0\esub{\var}{\tmfour'}\esub{\vartwo}{\val}\sctx = \tmtwo_1$
      from rules $\ruleULsv$ and $\ruleUEsR$.
      To conclude, $\tm_1 = \tmthree_0\esub{\var}{\tmfour'\esub{\var}{\val}\sctx}
      \equiv \tmthree_0\esub{\var}{\tmfour'\esub{\vartwo}{\val}}\sctx = \tm'_1$
      by rules $\ruleEquivEsAssoc$ and rules $\ruleEquivCongES$ with 
      $\ruleEquivEsAssoc$ since by $\alpha$-conversion we may assume
      $\domSctx{\sctx} \disj \fv{\tmthree_0}$.
      Finally, $\tm'_1 \equiv \tmtwo_1$ by rules $\ruleEquivCongES$ and $\ruleEquivEsAssoc$.
      Thus $\tm_1 \equiv \tmtwo_1$.
      The following diagram summarises the proof:
      \[
        \xymatrixrowsep{0.3pc}
        \xymatrix{
          \tm_0 = \tmthree_0\esub{\var}{\tmfour\esub{\vartwo}{\val\sctx}}
            \arUrStable{\rulelsv}{\aset}{\sset}{\appflag}
            % \ar@3{-}[dd]_>{c}
        & \tmthree_0\esub{\vartwo}{\tmfour'\esub{\var}{\val}\sctx} = \tm_1
            % \ar@3{.}[d]
        \\
        & \equiv
        \\
          \equivC
        & \tmthree_0\esub{\var}{\tmfour'\esub{\vartwo}{\val}}\sctx = \tm'_1
            % \ar@3{.}[d]_>{c}
        \\
        & \equiv
        \\
          \tmtwo_0 = \tmthree_0\esub{\var}{\tmfour}\esub{\vartwo}{\val\sctx}
            \arsdUrStable{\rulelsv}{\aset}{\sset}{\appflag}
        & \tmthree_0\esub{\var}{\tmfour'}\esub{\vartwo}{\val}\sctx = \tmtwo_1
        }
      \]
    \item $\ruleUEsR$.
      The reduction step is derived from
      $\tmfive \tostable{\rulename}{\aset}{\sset}{\nonapp} \tmfive'$.
      On the other hand, we have 
      $\tm_0 = \tmthree_0\esub{\var}{\tmfour\esub{\vartwo}{\tmfive}} \equivC
      \tmthree_0\esub{\var}{\tmfour}\esub{\vartwo}{\tmfive} = \tmtwo_0$.
      Then we can derive $\tmtwo_0 \tostable{\rulename}{\aset}{\sset}{\appflag} 
      \tmthree_0\esub{\var}{\tmfour}\esub{\vartwo}{\tmfive'} = \tmtwo_1$
      from rule $\ruleUEsR$.
      To conclude, $\tm_1 = \tmthree_0\esub{\var}{\tmfour\esub{\vartwo}{\tmfive'}}
      \equiv \tmtwo_1$
      by rule $\ruleEquivEsAssoc$ since $\vartwo \notin \fv{\tmthree_0}$.
      The following diagram summarises the proof:
      \[
        \xymatrixrowsep{0.3pc}
        \xymatrix{
          \tm_0 = \tmthree_0\esub{\var}{\tmfour\esub{\vartwo}{\tmfive}}
            \arUrStable{\rulename}{\aset}{\sset}{\appflag}
            % \ar@3{-}[d]_>{c}
        & \tmthree_0\esub{\var}{\tmfour\esub{\vartwo}{\tmfive'}} = \tm_1
            % \ar@3{.}[d]_>{c}
        \\
          \ 
        & \equiv
        \\
          \tmtwo_0 = \tmthree_0\esub{\var}{\tmfour}\esub{\vartwo}{\tmfive}
            \arsdUrStable{\rulename}{\aset}{\sset}{\appflag}
        & \tmthree_0\esub{\var}{\tmfour}\esub{\vartwo}{\tmfive'} = \tmtwo_1
        }
      \]
    \item $\ruleUEsLAbs$.
      The reduction step is derived from
      $\tmfour \tostable{\rulename}{\aset \cup \set{\vartwo}}{\sset}{\nonapp} \tmfour'$
      and $\tmfive \in \HAbs{\aset}$ and $\vartwo \notin \aset\cup\sset$ and
      $\vartwo \notin \fv{\rulename}$.
      On the other hand, we have 
      $\tm_0 = \tmthree_0\esub{\var}{\tmfour\esub{\vartwo}{\tmfive}} \equivC
      \tmthree_0\esub{\var}{\tmfour}\esub{\vartwo}{\tmfive} = \tmtwo_0$.
      Then we can derive $\tmtwo_0 \tostable{\rulename}{\aset}{\sset}{\appflag} 
      \tmthree_0\esub{\var}{\tmfour'}\esub{\vartwo}{\tmfive} = \tmtwo_1$
      from rules $\ruleUEsLAbs$ and $\ruleUEsR$.
      To conclude, $\tm_1 = \tmthree_0\esub{\var}{\tmfour'\esub{\vartwo}{\tmfive}}
      \equiv \tmtwo_1$ by rule $\ruleEquivEsAssoc$ since 
      $\vartwo \notin \fv{\tmthree_0}$.
      The following diagram summarises the proof:
      \[
        \xymatrixrowsep{0.3pc}
        \xymatrix{
          \tm_0 = \tmthree_0\esub{\var}{\tmfour\esub{\vartwo}{\tmfive}}
            \arUrStable{\rulename}{\aset}{\sset}{\appflag}
            % \ar@3{-}[d]_>{c}
        & \tmthree_0\esub{\var}{\tmfour'\esub{\vartwo}{\tmfive}} = \tm_1
            % \ar@3{.}[d]_>{c}
        \\
          \equivC
        & \equiv
        \\
          \tmtwo_0 = \tmthree_0\esub{\var}{\tmfour}\esub{\vartwo}{\tmfive}
            \arsdUrStable{\rulename}{\aset}{\sset}{\appflag}
        & \tmthree_0\esub{\var}{\tmfour'}\esub{\vartwo}{\tmfive} = \tmtwo_1
        }
      \]
    \item $\ruleUEsLStruct$.
      Analogous to the previous case.
    \end{itemize}
  \item $\ruleEquivEsLDist(1)$.
    Then $\tm_0 = (\tmthree_1 \, \tmthree_2)\esub{\var}{\tmfour_0} \equivC
    \tmthree_1\esub{\var}{\tmfour_0} \, \tmthree_2 = \tmtwo_0$, 
    with $\var \notin \fv{\tmthree_2}$.
    We can derive $\tmtwo_0 \tostable{\rulename}{\aset}{\sset}{\appflag} 
    \tmthree_1\esub{\var}{\tmfour_1}\,\tmthree_2 = \tmtwo_1$ 
    from rules $\ruleUAppL$ and $\ruleUEsR$.
    To conclude, $\tm_1 = (\tmthree_1\,\tmthree_2)\esub{\var}{\tmfour_1} \equiv
    \tmtwo_1$ by rule $\ruleEquivEsLDist$, 
    since $\var \notin \fv{\tmthree_2}$.
    The following diagram summarises the proof:
    \[
      \xymatrixrowsep{0.3pc}
      \xymatrix{
        \tm_0 = (\tmthree_1 \, \tmthree_2)\esub{\var}{\tmfour_0}
          \arUrStable{\rulename}{\aset}{\sset}{\appflag}
          % \ar@3{-}[d]_>{c}
      & (\tmthree_1 \, \tmthree_2)\esub{\var}{\tmfour_1} = \tm_1
          % \ar@3{.}[d]_>{c}
      \\
        \equivC
      & \equiv
      \\
        \tmtwo_0 = \tmthree_1\esub{\var}{\tmfour_0} \, \tmthree_2
          \arsdUrStable{\rulename}{\aset}{\sset}{\appflag}
      & \tmthree_1\esub{\var}{\tmfour_1} \, \tmthree_2 = \tmtwo_1
      }
    \]
  \item $\ruleEquivEsRDist(1)$.
    Analogous to the previous case.
  \item $\ruleEquivCongES$.
    Then $\tm_0 \equivC \tmthree'_0\esub{\var}{\tmfour'_0} = \tmtwo_0$, 
    derived from $\tmthree_0 \equivC \tmthree'_0$ and $\tmfour_0 \equivC \tmfour'_0$.
    By \ih on $\tmfour_0$, there exists $\tmfour'_1$ such that
    $\tmfour'_0 \tostable{\rulename}{\aset}{\sset}{\nonapp} \tmfour'_1$ and
    $\tmfour_1 \equiv \tmfour'_1$.
    Then we can derive $\tmtwo_0 \tostable{\rulename}{\aset}{\sset}{\appflag} 
    \tmthree'_0\esub{\var}{\tmfour'_1} = \tmtwo_1$ 
    from rule $\ruleUEsR$.
    To conclude, $\tm_1 = \tmthree_0\esub{\var}{\tmfour_1} \equiv \tmtwo_1$ 
    by rule $\ruleEquivCongES$, where $\tmthree_0 \equiv \tmthree'_0$ by the
    hypothesis $\tmthree_0 \equivC \tmthree'_0$, and 
    $\tmfour_1 \equiv \tmfour'_1$ by the \ih.
    The following diagram summarises the proof:
    \[
      \xymatrixrowsep{0.3pc}
      \xymatrix{
        \tm_0 = \tmthree_0\esub{\var}{\tmfour_0}
          \arUrStable{\rulename}{\aset}{\sset}{\appflag}
          % \ar@3{-}[d]_>{c}
      & \tmthree_0\esub{\var}{\tmfour_1} = \tm_1
          % \ar@3{.}[d]_>{c}
      \\
        \equivC
      & \equiv
      \\
        \tmtwo_0 = \tmthree'_0\esub{\var}{\tmfour'_0}
          \arsdUrStable{\rulename}{\aset}{\sset}{\appflag}
      & \tmthree'_0\esub{\var}{\tmfour'_1} = \tmtwo_1
      }
    \]
  \end{enumerate}
\item $\ruleUEsLAbs$.
  Then $\tm_0 = \tmthree_0\esub{\var}{\tmfour_0} 
  \tostable{\rulename}{\aset}{\sset}{\appflag} \tmthree_1\esub{\var}{\tmfour_0}
  = \tm_1$,
  which is derived from $\tmthree_0 
  \tostable{\rulename}{\aset \cup \set{\var}}{\sset}{\appflag} \tmthree_1$
  and $\tmfour_0 \in \HAbs{\aset}$
  and $\var \notin \aset \cup \sset$
  and $\var \notin \fv{\rulename}$.
  We also have $\tm_0 \equivC \tmtwo_0$, so we analyze the different cases according to
  which rule was used to derive the equivalence. Case $\ruleEquivCongApp$ is
  impossible due to the form of $\tm_0$.
  \begin{enumerate}[leftmargin=*, label=\alph*.]
  \item $\ruleEquivEsComm$.
    Then $\tmthree_0 = \tmthree\esub{\vartwo}{\tmfive} 
    \tostable{\rulename}{\aset \cup \set{\var}}{\sset}{\appflag} \tmthree_1$,
    and $\tm_0 = \tmthree\esub{\vartwo}{\tmfive}\esub{\var}{\tmfour_0} \equivC
    \tmthree\esub{\var}{\tmfour_0}\esub{\vartwo}{\tmfive}= \tmtwo_0$,
    with $\var \notin \fv{\tmfive}$ and $\vartwo \notin \fv{\tmfour_0}$.
    The reduction step can be derived either from rule $\ruleULsv$, $\ruleUEsR$,
    $\ruleUEsLAbs$ or $\ruleUEsLStruct$:
    \begin{itemize}[leftmargin=*]
    \item $\ruleULsv$.
      Then $\tmfive = \val\sctx$ and $\rulename = \rulelsv$, and the reduction
      step is derived from $\tmthree 
      \tostable{\rulesub{\var}{\val}}{\aset \cup \set{\var} \cup \set{\vartwo}}{\sset}{\appflag}
      \tmthree'$ and $\val\sctx \in \HAbs{\aset \cup \set{\var}}$ and
      $\vartwo \notin \aset \cup \set{\var} \cup \sset$.
      On the other hand, we have 
      $\tm_0 = \tmthree\esub{\vartwo}{\val\sctx}\esub{\var}{\tmfour_0} 
      \equivC \tmthree\esub{\var}{\tmfour_0}\esub{\vartwo}{\val\sctx} = \tmtwo_0$.
      We can derive $\tmtwo_0 \tostable{\rulelsv}{\aset}{\sset}{\appflag}
      \tmthree'\esub{\var}{\tmfour_0}\esub{\vartwo}{\val}\sctx = \tmtwo_1$
      by rules $\ruleULsv$ and $\ruleUEsLAbs$.
      Moreover, 
      $\tm_1 \equiv \tmtwo_1$ by applying several times rules $\ruleEquivCongES$
      and $\ruleEquivEsComm$, since from the derivation of 
      $\tmtwo_0 \tostable{\rulename}{\aset}{\sset}{\appflag} \tmtwo_1$ 
      we may assume $\var \notin \fv{\val\sctx}$ by $\alpha$-conversion,
      and $\domSctx{\sctx} \disj \fv{\tmfour_0}$
      The following diagram summarises the proof:
      \[
        \xymatrixrowsep{0.3pc}
        \xymatrix{
          \tm_0 = \tmthree\esub{\vartwo}{\val\sctx}\esub{\var}{\tmfour_0}
            \arUrStable{\rulelsv}{\aset}{\sset}{\appflag}
            % \ar@3{-}[d]_>{c}
        & \tmthree'\esub{\vartwo}{\val}\sctx\esub{\var}{\tmfour_0} = \tm_1
            % \ar@3{.}[d]
        \\
          \equivC
        & \equiv
        \\
          \tmtwo_0 = \tmthree\esub{\var}{\tmfour_0}\esub{\vartwo}{\val\sctx}
            \arsdUrStable{\rulelsv}{\aset}{\sset}{\appflag}
        & \tmthree'\esub{\var}{\tmfour_0}\esub{\vartwo}{\val}\sctx = \tmtwo_1
        }
      \]
    \item $\ruleUEsR$.
      The reduction step is derived from 
      $\tmfive \tostable{\rulename}{\aset \cup \set{\var}}{\sset}{\nonapp} \tmfive'$.
      On the other hand, we have 
      $\tm_0 = \tmthree\esub{\vartwo}{\tmfive}\esub{\var}{\tmfour_0} 
      \equivC \tmthree\esub{\var}{\tmfour_0}\esub{\vartwo}{\tmfive} = \tmtwo_0$.
      We can derive $\tmtwo_0 \tostable{\rulename}{\aset}{\sset}{\appflag}
      \tmthree\esub{\var}{\tmfour_0}\esub{\vartwo}{\tmfive'} = \tmtwo_1$
      by rule $\ruleUEsR$.
      From the derivation of 
      $\tmtwo_0 \tostable{\rulename}{\aset}{\sset}{\appflag} \tmtwo_1$ 
      we may assume $\var \notin \fv{\tmfive'}$ by $\alpha$-conversion, 
      so that by applying rule $\ruleEquivEsComm$, we conclude
      $\tm_1 \equiv \tmtwo_1$.
      The following diagram summarises the proof:
      \[
        \xymatrixrowsep{0.3pc}
        \xymatrix{
          \tm_0 = \tmthree\esub{\vartwo}{\tmfive}\esub{\var}{\tmfour_0}
            \arUrStable{\rulename}{\aset}{\sset}{\appflag}
            % \ar@3{-}[d]_>{c}
        & \tmthree\esub{\vartwo}{\tmfive'}\esub{\var}{\tmfour_0} = \tm_1
            % \ar@3{.}[d]_>{c}
        \\
          \equivC
        & \equiv
        \\
          \tmtwo_0 = \tmthree\esub{\var}{\tmfour_0}\esub{\vartwo}{\tmfive}
            \arsdUrStable{\rulename}{\aset}{\sset}{\appflag}
        & \tmthree\esub{\var}{\tmfour_0}\esub{\vartwo}{\tmfive'} = \tmtwo_1
        }
      \]
    \item $\ruleUEsLAbs$.
      The reduction step is derived from $\tmthree 
      \tostable{\rulename}{\aset \cup \set{\var} \cup \set{\vartwo}}{\sset}{\appflag}
      \tmthree'$ and $\tmfive \in \HAbs{\aset \cup \set{\var}}$ and
      $\vartwo \notin \aset \cup \set{\var} \cup \sset$ and $\vartwo \notin \fv{\rulename}$.
      On the other hand,
      $\tm_0 = \tmthree\esub{\vartwo}{\tmfive}\esub{\var}{\tmfour_0} 
      \equivC \tmthree\esub{\var}{\tmfour_0}\esub{\vartwo}{\tmfive} = \tmtwo_0$.
      We can derive $\tmtwo_0 \tostable{\rulename}{\aset}{\sset}{\appflag}
      \tmthree'\esub{\var}{\tmfour_0}\esub{\vartwo}{\tmfive} = \tmtwo_1$
      by rule $\ruleUEsLAbs$.
      To conclude, $\tmtwo_1 \equiv \tm_1$ by rule $\ruleEquivEsComm$ since 
      $\var \notin \fv{\tmfive}$ and $\vartwo \notin \fv{\tmfour_0}$ by hypothesis.
      The following diagram summarises the proof:
      \[
        \xymatrixrowsep{0.3pc}
        \xymatrix{
          \tm_0 = \tmthree\esub{\vartwo}{\tmfive}\esub{\var}{\tmfour_0}
            \arUrStable{\rulename}{\aset}{\sset}{\appflag}
            % \ar@3{-}[d]_>{c}
        & \tmthree'\esub{\vartwo}{\tmfive}\esub{\var}{\tmfour_0} = \tm_1
            % \ar@3{.}[d]_>{c}
        \\
          \equivC
        & \equiv
        \\
          \tmtwo_0 = \tmthree\esub{\var}{\tmfour_0}\esub{\vartwo}{\tmfive}
            \arsdUrStable{\rulename}{\aset}{\sset}{\appflag}
        & \tmthree'\esub{\var}{\tmfour_0}\esub{\vartwo}{\tmfive} = \tmtwo_1
        }
      \]
    \item $\ruleUEsLStruct$.
      Analogous to the previous case.
    \end{itemize}
  \item $\ruleEquivEsAssoc$ (1).
    Then $\tmthree_0 = \tmthree\esub{\vartwo}{\tmfive} 
    \tostable{\rulename}{\aset \cup \set{\var}}{\sset}{\appflag} \tmthree_1$,
    and $\tm_0 = \tmthree\esub{\vartwo}{\tmfive}\esub{\var}{\tmfour_0} \equivC
    \tmthree\esub{\vartwo}{\tmfive\esub{\var}{\tmfour_0}} = \tmtwo_0$,
    with $\var \notin \fv{\tmthree}$.
    The reduction step can be derived either from rule $\ruleULsv$, $\ruleUEsR$,
    $\ruleUEsLAbs$ or $\ruleUEsLStruct$:
    \begin{itemize}[leftmargin=*]
    \item $\ruleULsv$.
      Then $\tmfive = \val\sctx$ and $\rulename = \rulelsv$, and the reduction
      step is derived from $\tmthree 
      \tostable{\rulesub{\vartwo}{\val}}{\aset \cup \set{\var} \cup \set{\vartwo}}{\sset}{\appflag}
      \tmthree'$ and $\val\sctx \in \HAbs{\aset \cup \set{\var}}$ and
      $\vartwo \notin \aset \cup \set{\var} \cup \sset$.
      On the other hand, we have 
      $\tm_0 = \tmthree\esub{\vartwo}{\val\sctx}\esub{\var}{\tmfour_0} 
      \equivC \tmthree\esub{\vartwo}{\val\sctx\esub{\var}{\tmfour_0}} = \tmtwo_0$.
      We can derive $\tmtwo_0 \tostable{\rulelsv}{\aset}{\sset}{\appflag}
      \tmthree'\esub{\vartwo}{\val}\sctx\esub{\var}{\tmfour_0} = \tmtwo_1$
      by rule $\ruleULsv$, and $\tmtwo_1 \equiv \tm_1$ by rule $\ruleEquivRefl$.
      The following diagram summarises the proof:
      \[
        \xymatrixrowsep{0.3pc}
        \xymatrix{
          \tm_0 = \tmthree\esub{\vartwo}{\val\sctx}\esub{\var}{\tmfour_0}
            \arUrStable{\rulelsv}{\aset}{\sset}{\appflag}
            % \ar@3{-}[d]_>{c}
        & \tmthree'\esub{\vartwo}{\val}\sctx\esub{\var}{\tmfour_0} = \tm'_1
            % \ar@3{.}[d]
        \\
          \equivC
        & \equiv
        \\
          \tmtwo_0 = \tmthree\esub{\vartwo}{\val\sctx\esub{\var}{\tmfour_0}}
            \arsdUrStable{\rulelsv}{\aset}{\sset}{\appflag}
        & \tmthree'\esub{\vartwo}{\val}\sctx\esub{\var}{\tmfour_0} = \tmtwo_1
        }
      \]
    \item $\ruleUEsR$.
      The reduction step is derived from 
      $\tmfive \tostable{\rulename}{\aset \cup \set{\var}}{\sset}{\nonapp} \tmfive'$.
      On the other hand, we have 
      $\tm_0 = \tmthree\esub{\vartwo}{\tmfive}\esub{\var}{\tmfour_0} 
      \equivC \tmthree\esub{\vartwo}{\tmfive\esub{\var}{\tmfour_0}} = \tmtwo_0$.
      We can derive $\tmtwo_0 \tostable{\rulename}{\aset}{\sset}{\appflag}
      \tmthree\esub{\vartwo}{\tmfive'\esub{\var}{\tmfour_0}} = \tmtwo_1$
      from rules $\ruleUEsR$ and $\ruleUEsLAbs$.
      Applying rule $\ruleEquivEsAssoc$, given $\var \notin \fv{\tmthree}$, 
      we conclude $\tm_1 \equiv \tmtwo_1$.
      The following diagram summarises the proof:
      \[
        \xymatrixrowsep{0.3pc}
        \xymatrix{
          \tm_0 = \tmthree\esub{\vartwo}{\tmfive}\esub{\var}{\tmfour_0}
            \arUrStable{\rulename}{\aset}{\sset}{\appflag}
            % \ar@3{-}[d]_>{c}
        & \tmthree\esub{\vartwo}{\tmfive'}\esub{\var}{\tmfour_0} = \tm_1
            % \ar@3{.}[d]_>{c}
        \\
          \equivC
        & \equiv
        \\
          \tmtwo_0 = \tmthree\esub{\vartwo}{\tmfive\esub{\var}{\tmfour_0}}
            \arsdUrStable{\rulename}{\aset}{\sset}{\appflag}
        & \tmthree\esub{\vartwo}{\tmfive'\esub{\var}{\tmfour_0}} = \tmtwo_1
        }
      \]
    \item $\ruleUEsLAbs$.
      Then the reduction step is derived from $\tmthree 
      \tostable{\rulename}{\aset \cup \set{\var} \cup \set{\vartwo}}{\sset}{\appflag}
      \tmthree'$ and $\tmfive \in \HAbs{\aset \cup \set{\var}}$ and
      $\vartwo \notin \aset \cup \set{\var} \cup \sset$ and $\vartwo \notin \fv{\rulename}$.
      On the other hand, we have 
      $\tm_0 = \tmthree\esub{\vartwo}{\tmfive}\esub{\var}{\tmfour_0} 
      \equivC \tmthree\esub{\vartwo}{\tmfive\esub{\var}{\tmfour_0}} = \tmtwo_0$.
      We can derive $\tmtwo_0 \tostable{\rulename}{\aset}{\sset}{\appflag}
      \tmthree'\esub{\vartwo}{\tmfive\esub{\var}{\tmfour_0}} = \tmtwo_1$
      by rule $\ruleUEsLAbs$.
      To conclude, $\tmtwo_1 \equiv \tm_1$ by rule $\ruleEquivEsAssoc$, 
      since from the derivation of 
      $\tmtwo_0 \tostable{\rulename}{\aset}{\sset}{\appflag} \tmtwo_1$ 
      we may assume $\var \notin \fv{\tmthree'}$ by $\alpha$-conversion.
      The following diagram summarises the proof:
      \[
        \xymatrixrowsep{0.3pc}
        \xymatrix{
          \tm_0 = \tmthree\esub{\vartwo}{\tmfive}\esub{\var}{\tmfour_0}
            \arUrStable{\rulename}{\aset}{\sset}{\appflag}
            % \ar@3{-}[d]_>{c}
        & \tmthree'\esub{\vartwo}{\tmfive}\esub{\var}{\tmfour_0} = \tm_1
            % \ar@3{.}[d]_>{c}
        \\
          \equivC
        & \equiv
        \\
          \tmtwo_0 = \tmthree\esub{\vartwo}{\tmfive\esub{\var}{\tmfour_0}}
            \arsdUrStable{\rulename}{\aset}{\sset}{\appflag}
        & \tmthree'\esub{\vartwo}{\tmfive\esub{\var}{\tmfour_0}} = \tmtwo_1
        }
      \]
    \item $\ruleUEsLStruct$.
      Analogous to the previous case.
    \end{itemize}
  \item $\ruleEquivEsAssoc$ (2).
    Then $\tmfour_0 = \tmfour\esub{\vartwo}{\tmfive}$, and the reduction
    $\tm_0 = \tmthree_0\esub{\var}{\tmfour\esub{\vartwo}{\tmfive}}
    \tostable{\rulename}{\aset}{\sset}{\appflag}
    \tmthree_1\esub{\var}{\tmfour\esub{\vartwo}{\tmfive}} = \tm_1$ is derived from
    $\tmthree_0 \tostable{\rulename}{\aset \cup \set{\var}}{\sset}{\appflag}\tmthree_1$
    and (1) $\tmfour\esub{\vartwo}{\tmfive} \in \HAbs{\aset}$ and 
    $\var \notin \aset \cup \sset$ and $\var \notin \fv{\rulename}$.
    Moreover,
    $\tm_0 = \tmthree_0\esub{\var}{\tmfour\esub{\vartwo}{\tmfive}} \equivC
    \tmthree_0\esub{\var}{\tmfour}\esub{\vartwo}{\tmfive} = \tmtwo_0$,
    with $\vartwo \notin \fv{\tmthree_0}$.
    The judgment (1) can be derived either by rule $\ruleHAbsSubi$ or $\ruleHAbsSubii$.
    Let us only work with the case in which it is derived by the latter; 
    then we obtain the reduction step 
    $\tmtwo_0 \tostable{\rulename}{\aset}{\sset}{\appflag}
    \tmthree_1\esub{\var}{\tmfour}\esub{\vartwo}{\tmfive} = \tmtwo_1$
    by applying rule $\ruleUEsLAbs$.
    To conclude, $\tmtwo_1 \equiv \tm_1$ by rule $\ruleEquivEsAssoc$, 
    since from the derivation of 
    $\tm_0 \tostable{\rulename}{\aset}{\sset}{\appflag} \tm_1$ we may assume 
    $\vartwo \notin \fv{\tmthree_1}$ by $\alpha$-conversion.
    The following diagram summarises the proof:
    \[
      \xymatrixrowsep{0.3pc}
      \xymatrix{
        \tm_0 = \tmthree_0\esub{\var}{\tmfour\esub{\vartwo}{\tmfive}}
          \arUrStable{\rulename}{\aset}{\sset}{\appflag}
          % \ar@3{-}[d]_>{c}
      & \tmthree_1\esub{\var}{\tmfour\esub{\vartwo}{\tmfive}} = \tm_1
          % \ar@3{.}[d]_>{c}
      \\
        \equivC
      & \equiv
      \\
        \tmtwo_0 = \tmthree_0\esub{\var}{\tmfour}\esub{\vartwo}{\tmfive}
          \arsdUrStable{\rulename}{\aset}{\sset}{\appflag}
      & \tmthree_1\esub{\var}{\tmfour}\esub{\vartwo}{\tmfive} = \tmtwo_1
      }
    \]
  \item $\ruleEquivEsLDist(1)$.
    Then the reduction step $\tm_0 \tostable{\rulename}{\aset}{\sset}{\appflag}
    \tm_1$ is derived from $\tmthree_0 = \tmfive_1 \, \tmfive_2 
    \tostable{\rulename}{\aset \cup \set{\var}}{\sset}{\appflag} \tmthree_1$
    and $\tmfour_0 \in \HAbs{\aset}$ and $\var \notin \aset \cup \sset$ and
    $\var \notin \fv{\rulename}$.
    Moreover, $\tm_0 = (\tmfive_1 \, \tmfive_2)\esub{\var}{\tmfour_0} \equivC
    \tmfive_1\esub{\var}{\tmfour_0} \, \tmfive_2 = \tmtwo_0$,
    with $\var \notin \fv{\tmfive_2}$.
    The reduction of $\tmfive_1 \, \tmfive_2$ to $\tmthree_1$ can be derived 
    either from rule $\ruleUDbStable$, $\ruleUAppL$ or $\ruleUAppR$:
    \begin{itemize}[leftmargin=*]
    \item $\ruleUDbStable$.
      Then $\tmfive_1 = (\lam{\vartwo}{\tmfive})\sctx$ and $\rulename = \ruledb$, 
      so we have $(\lam{\vartwo}{\tmfive})\sctx \, \tmfive_2 
      \tostable{\ruledb}{\aset \cup \set{\var}}{\sset}{\appflag} 
      \tmfive\esub{\vartwo}{\tmfive_2}\sctx$
      where $\tmfive_2 \in \HAbs{\aset} \cup \Struct{\sset}$.
      On the other hand,
      $\tm_0 = ((\lam{\vartwo}{\tmfive})\sctx \, \tmfive_2)\esub{\var}{\tmfour_0} 
      \equivC (\lam{\vartwo}{\tmfive})\sctx\esub{\var}{\tmfour_0} \, \tmfive_2 = \tmtwo_0$.
      We can derive $(\lam{\vartwo}{\tmfive})\sctx\esub{\var}{\tmfour_0} \, \tmfive_2 
      \tostable{\ruledb}{\aset}{\sset}{\appflag} 
      \tmfive\esub{\vartwo}{\tmfive_2}\sctx\esub{\var}{\tmfour_0} = \tmtwo_1$
      by rule $\ruleUDbStable$, and $\tmtwo_1 \equiv \tm_1$ by rule $\ruleEquivRefl$.
      The following diagram summarises the proof:
      \[
        \xymatrixrowsep{0.3pc}
        \xymatrix{
          \tm_0 = ((\lam{\vartwo}{\tmfive})\sctx \, \tmfive_2)\esub{\var}{\tmfour_0}
            \arUrStable{\ruledb}{\aset}{\sset}{\appflag}
            % \ar@3{-}[d]_>{c}
        & \tmfive\esub{\vartwo}{\tmfive_2}\sctx\esub{\var}{\tmfour_0} = \tm_1
            % \ar@3{.}[d]
        \\
          \equivC
        & \equiv
        \\
          \tmtwo_0 = (\lam{\vartwo}{\tmfive})\sctx\esub{\var}{\tmfour_0} \, \tmfive_2
            \arsdUrStable{\ruledb}{\aset}{\sset}{\appflag}
        & \tmfive\esub{\vartwo}{\tmfive_2}\sctx\esub{\var}{\tmfour_0} = \tmtwo_1
        }
      \]
    \item $\ruleUAppL$.
      The reduction of $\tmfive_1 \, \tmfive_2$ to $\tmthree_1$ is derived from 
      $\tmfive_1 \tostable{\rulename}{\aset \cup \set{\var}}{\sset}{\app} \tmfive'_1$.
      On the other hand,
      $\tm_0 = (\tmfive_1 \, \tmfive_2)\esub{\var}{\tmfour_0} 
      \equivC \tmfive_1\esub{\var}{\tmfour_0} \, \tmfive_2 = \tmtwo_0$.
      We can derive $\tmtwo_0 \tostable{\rulename}{\aset}{\sset}{\appflag}
      \tmfive'_1\esub{\var}{\tmfour_0} \, \tmfive_2 = \tmtwo_1$
      from rules $\ruleUAppL$ and $\ruleUEsLAbs$.
      Applying rule $\ruleEquivEsLDist$, since $\var \notin \fv{\tmfive_2}$, 
      we conclude $\tm_1 \equiv \tmtwo_1$.
      The following diagram summarises the proof:
      \[
        \xymatrixrowsep{0.3pc}
        \xymatrix{
          \tm_0 = (\tmfive_1 \, \tmfive_2)\esub{\var}{\tmfour_0}
            \arUrStable{\rulename}{\aset}{\sset}{\appflag}
            % \ar@3{-}[d]_>{c}
        & (\tmfive'_1 \, \tmfive_2)\esub{\var}{\tmfour_0} = \tm_1
            % \ar@3{.}[d]_>{c}
        \\
          \equivC
        & \equiv
        \\
          \tmtwo_0 = \tmfive_1\esub{\var}{\tmfour_0} \, \tmfive_2
            \arsdUrStable{\rulename}{\aset}{\sset}{\appflag}
        & \tmfive'_1\esub{\var}{\tmfour_0} \, \tmfive_2 = \tmtwo_1
        }
      \]
    \item $\ruleUAppR$.
      Then the reduction of $\tmfive_1 \, \tmfive_2$ to $\tmthree_1$ is derived from
      $\tmfive_1 \in \Struct{\sset}$ and 
      $\tmfive_2 \tostable{\rulename}{\aset \cup \set{\var}}{\sset}{\nonapp} \tmfive'_2$.
      On the other hand, we have 
      $\tm_0 = (\tmfive_1 \, \tmfive_2)\esub{\var}{\tmfour_0} 
      \equivC \tmfive_1\esub{\var}{\tmfour_0} \, \tmfive_2 = \tmtwo_0$.
      We can derive $\tmtwo_0 \tostable{\rulename}{\aset}{\sset}{\appflag}
      \tmfive_1\esub{\var}{\tmfour_0} \, \tmfive'_2 = \tmtwo_1$
      from rule $\ruleUAppR$, since $\tmfive_1\esub{\var}{\tmfour_0} \in \Struct{\sset}$
      by $\ruleStructSubi$.
      From the derivation of 
      $\tmtwo_0 \tostable{\rulename}{\aset}{\sset}{\appflag} \tmtwo_1$ 
      we may assume $\var \notin \fv{\tmfive'_2}$ by $\alpha$-conversion, 
      so that by applying rule $\ruleEquivEsLDist$, we conclude
      $\tm_1 \equiv \tmtwo_1$.
      The following diagram summarises the proof:
      \[
        \xymatrixrowsep{0.3pc}
        \xymatrix{
          \tm_0 = (\tmfive_1 \, \tmfive_2)\esub{\var}{\tmfour_0}
            \arUrStable{\rulename}{\aset}{\sset}{\appflag}
            % \ar@3{-}[d]_>{c}
        & (\tmfive_1 \, \tmfive'_2)\esub{\var}{\tmfour_0} = \tm_1
            % \ar@3{.}[d]_>{c}
        \\
          \equivC
        & \equiv
        \\
          \tmtwo_0 = \tmfive_1\esub{\var}{\tmfour_0} \, \tmfive_2
            \arsdUrStable{\rulename}{\aset}{\sset}{\appflag}
        & \tmfive_1\esub{\var}{\tmfour_0} \, \tmfive'_2 = \tmtwo_1
        }
      \]
    \end{itemize}
  \item $\ruleEquivEsRDist(1)$.
    Then the reduction step $\tm_0 \tostable{\rulename}{\aset}{\sset}{\appflag}
    \tm_1$ is derived from $\tmthree_0 = \tmfive_1 \, \tmfive_2 
    \tostable{\rulename}{\aset \cup \set{\var}}{\sset}{\appflag} \tmthree_1$
    and $\tmfour_0 \in \HAbs{\aset}$ and $\var \notin \aset \cup \sset$ and
    $\var \notin \fv{\rulename}$.
    Moreover, $\tm_0 = (\tmfive_1 \, \tmfive_2)\esub{\var}{\tmfour_0} \equivC
    \tmfive_1 \, \tmfive_2\esub{\var}{\tmfour_0} = \tmtwo_0$,
    with $\var \notin \fv{\tmfive_1}$.
    The reduction of $\tmfive_1 \, \tmfive_2$ to $\tmthree_1$ can be derived 
    either from rule $\ruleUDbStable$, $\ruleUAppL$ or $\ruleUAppR$.
    The only relevant case is $\ruleUDbStable$, the others being analogous to subcases 
    $\ruleUAppL$ and $\ruleUAppR$ of case $\ruleEquivEsLDist$.
    Hence $\tmfive_1 = (\lam{\vartwo}{\tmfive})\sctx$ and $\rulename = \ruledb$, 
    with $\tmfive_2 \in \HAbs{\aset \cup \set{\var}} \cup \Struct{\sset}$.
    Therefore we have $(\lam{\vartwo}{\tmfive})\sctx \, \tmfive_2 
    \tostable{\ruledb}{\aset \cup \set{\var}}{\sset}{\appflag} 
    \tmfive\esub{\vartwo}{\tmfive_2}\sctx$.
    On the other hand,
    $\tm_0 = ((\lam{\vartwo}{\tmfive})\sctx \, \tmfive_2)\esub{\var}{\tmfour_0} 
    \equivC (\lam{\vartwo}{\tmfive})\sctx \, \tmfive_2\esub{\var}{\tmfour_0} = \tmtwo_0$.
    Given $\tmfive_2 \in \HAbs{\aset \cup \set{\var}} \cup \Struct{\sset}$
    we can derive $\tmfive_2\esub{\var}{\tmfour_0} \in \HAbs{\aset} \cup \Struct{\sset}$,
    so $\tmtwo_0 \tostable{\ruledb}{\aset}{\sset}{\appflag} 
    \tmfive\esub{\vartwo}{\tmfive_2\esub{\var}{\tmfour_0}}\sctx = \tmtwo_1$
    by rule $\ruleUDbStable$. And $\tmtwo_1 \equiv 
    \tmfive\esub{\vartwo}{\tmfive_2}\esub{\var}{\tmfour_0}\sctx = \tmtwo'_1$ 
    by rule $\ruleEquivCongES$, and by rule $\ruleEquivEsAssoc$ 
    on the body of the explicit substitution.
    To conclude, $\tmtwo'_1 \equiv \tm_1$ by applying rule $\ruleEquivEsComm$ and
    rules $\ruleEquivCongES$ and $\ruleEquivEsComm$,
    since $\var \notin \fv{\sctx}$ by hypothesis, and we may assume 
    $\domSctx{\sctx} \disj \fv{\tmfour_0}$ by $\alpha$-conversion.
    The following diagram summarises the proof:
      \[
        \xymatrixrowsep{0.3pc}
        \xymatrix{
          \tm_0 = ((\lam{\vartwo}{\tmfive})\sctx \, \tmfive_2)\esub{\var}{\tmfour_0}
            \arUrStable{\ruledb}{\aset}{\sset}{\appflag}
            % \ar@3{-}[dd]_>{c}
        & \tmfive\esub{\vartwo}{\tmfive_2}\sctx\esub{\var}{\tmfour_0} = \tm_1
            % \ar@3{.}[d]
        \\
        & \equiv
        \\
          \equivC
        & \tmfive\esub{\vartwo}{\tmfive_2}\esub{\var}{\tmfour_0}\sctx = \tmtwo'_1
            % \ar@3{.}[d]_>{c}
        \\
        & \equiv
        \\
          \tmtwo_0 = (\lam{\vartwo}{\tmfive})\sctx \, \tmfive_2\esub{\var}{\tmfour_0}
            \arsdUrStable{\ruledb}{\aset}{\sset}{\appflag}
        & \tmfive\esub{\vartwo}{\tmfive_2\esub{\var}{\tmfour_0}}\sctx = \tmtwo_1
        }
      \]
  \item $\ruleEquivCongES$.
    Then $\tm_0 \equivC \tmthree'_0\esub{\var}{\tmfour'_0} = \tmtwo_0$, derived 
    from $\tmthree_0 \equivC \tmthree'_0$ and $\tmfour_0 \equivC \tmfour'_0$.
    By \ih on $\tmthree_0$, there exists $\tmthree'_1$ such that
    $\tmthree'_0 \tostable{\rulename}{\aset \cup \set{\var}}{\sset}{\appflag} \tmthree'_1$ and
    $\tmthree_1 \equiv \tmthree'_1$.
    Then we can derive $\tmtwo_0 \tostable{\rulename}{\aset}{\sset}{\appflag} 
    \tmthree'_1\esub{\var}{\tmfour'_0} = \tmtwo_1$ 
    from rule $\ruleUEsLAbs$.
    To conclude, $\tm_1 = \tmthree_1\esub{\var}{\tmfour_0} \equiv \tmtwo_1$ 
    by rule $\ruleEquivCongES$, where $\tmthree'_0 \equiv \tmthree'_1$ by the \ih
    and $\tmfour_0 \equiv \tmfour'_0$ by hypothesis.
    The following diagram summarises the proof:
    \[
      \xymatrixrowsep{0.3pc}
      \xymatrix{
        \tm_0 = \tmthree_0\esub{\var}{\tmfour_0}
          \arUrStable{\rulename}{\aset}{\sset}{\appflag}
          % \ar@3{-}[d]_>{c}
      & \tmthree_1\esub{\vartwo}{\tmfour_0} = \tm_1
          % \ar@3{.}[d]
      \\
        \equivC
      & \equiv
      \\
        \tmtwo_0 = \tmthree'_0\esub{\var}{\tmfour'_0}
          \arsdUrStable{\rulename}{\aset}{\sset}{\appflag}
      & \tmthree'_1\esub{\var}{\tmfour'_0} = \tmtwo_1
      }
    \]
  \end{enumerate}
\item $\ruleUEsLStruct$.
  Is analogous to case $\ruleUEsLAbs$.
}
\end{itemize}
\end{proof}

\subsubsection{Simulation of GLAMOUr steps in \UOCBV}
For completeness, we recall the syntax and relevant details of the
\glamour abstract machine.

\begin{definition}[Syntax of the \glamour]
The set of \defn{states} $(\state, \statetwo, \hdots)$,
\defn{dumps} $(\dump, \dumptwo, \hdots)$,
\defn{stacks} $(\stack, \stacktwo, \hdots)$, 
\defn{stack items} $(\stackitem, \stackitemtwo, \hdots)$, and
\defn{global environments} $(\genv, \genvtwo, \hdots)$
are given by the following grammars:
\[
  \begin{array}{rcll}
    \state     & \eqgram & \glamourst{\dump}{\code}{\stack}{\genv}
  \\
    \dump      & \eqgram & \estack \mid \dump \cons \pair{\code}{\stack}
  \\
    \stack     & \eqgram & \estack \mid \stackiteml \cons \stack
      & \text{where $\lab \in \set{\alive,\dead}$}                       
  \\
    \stackitem & \eqgram & \code \mid \pair{\code}{\stack}
  \\
    \genv      & \eqgram & \estack \mid \esub{\var}{\stackiteml} \cons \genv
      & \text{where $\lab \in \set{\alive,\dead}$}                       
  \end{array}
\]
where \defn{codes} $(\code, \codetwo, \hdots)$ are 
terms with no explicit substitutions (\ie, pure terms),
but \emph{they are not} considered up to $\alpha$-equivalence.
We use to decorate stack items with \defn{labels} $\lab \in \set{\alive, \dead}$,
writing $\stackiteml$ rather than $\stackitem$.
By convention this label always indicates the shape of $\stackitem$,
so in particular,
$\lab = \alive$ if and only if $\stackitem$ is of the form $\tm$,
and
$\lab = \dead$ if and only if $\stackitem$ is of the form $\pair{\tm}{\stack}$.
Intuitively, these labels indicate whether a stack item unfolds to 
a hereditary abstraction ($\alive$) or a structure ($\dead$).

Let $\state = \glamourst{\dump}{\code}{\stack}{\genv}$.
A \defn{binding occurrence} for a variable $\var$ in $\state$ is
either the leftmost occurrence of $\var$ in an abstraction $\lam{\var}{\code}$
or the leftmost occurrence of $\var$ in an element $\esub{\var}{\stackitem}$
of the environment.
We say that $\state$ is \defn{well-named}
if the three following conditions hold:
\begin{enumerate}
\item
  Each variable $\var$ has at most one binding occurrence.
  For example,
  $\glamourst{\pair{\lam{\var}{\lam{\vartwo}{\vartwo}}}{\estack}}{\varfour}{\estack}{\esub{\varthree}{\varfour}}$
  is well-named, while
  $\glamourst{\estack}{\lam{\var}{\var}}{\estack}{\esub{\var}{\varfour}}$ and
  $\glamourst{\estack}{\lam{\var}{\lam{\var}{\vartwo}}}{\estack}{\estack}$
  are not.
\item
  If there is some binding occurrence for a variable $\var$
  in an abstraction,then, all occurrences of $\var$ only
  occur inside the body of this abstraction.
  For example,
  $\glamourst{\estack}{\lam{\var}{\lam{\vartwo}{\var}}}{\var}{\estack}$
  is not well-named.
\item
  If there is some binding occurrence for a variable $\var$
  in an environment of the form $\genv_1\cons\esub{\var}{\stackitem}\cons\genv_2$, then there are no other occurrences of $\var$
  in $\genv_2$.
  For example,
  $\glamourst{\pair{\estack}{\estack}}{\vartwo}{\vartwo}{\esub{\var}{\vartwo}\esub{\varthree}{\var}}$
  is not well-named.
\end{enumerate}

\end{definition}

\begin{definition}[Decoding of components]
The decoding of the components of the \glamour abstract machine
into the syntax of our calculus is given by the following function:
\[
  \begin{array}{rcll}
    \decode{\glamourst{\dump}{\code}{\stack}{\genv}} & \defeq & \decodep{\genv}{\decodep{\dump}{\decodep{\stack}{\tm}}}
  \\
    \decode{\estack}                                 & \defeq & \ectx
  \\
    \decode{\dump \cons \pair{\code}{\stack}}        & \defeq & \decodep{\dump}{\decodep{\stack}{\tm\ectx}}
  \\
    \decode{\stackiteml \cons \stack}                & \defeq & \decodep{\stack}{\ectx\decode{\stackiteml}}
  \\
    \decode{\esub{\var}{\stackiteml} \cons \genv}    & \defeq & \decodep{\genv}{\ectx\esub{\var}{\decode{\stackiteml}}}
  \\
    \decode{\pair{\code}{\stack}^\lab}               & \defeq & \decodep{\stack}{\tm}
  \\
    \decode{\code^\lab}                              & \defeq & \tm 
  \\
    \decode{\code}                                   & \defeq & \tm
  \end{array}
\]
In the two last lines, $\tm$ denotes a term which is $\alpha$-equivalent to $\code$.
\end{definition}

\begin{definition}[Transitions of the \glamour abstract machine]
The transitions of the \glamour abstract machine are defined as follows:
\[
  \begin{array}{c|c|c|ccc|c|c|ccc}
    \usefmaca{\genv}{\dump}{\code \, \codethree}{\stack}
  & \tomachcone &
   \usefmaca{\genv}{\dump\cons \pair{\code}{\stack}}{\codethree}{\estack}
  \\
    \usefmaca{\genv}{\dump}{\tocode{\lam{\var}{\tm}}}{\stackiteml\cons\stack}
  & \tomachum &
    \usefmaca{\esub\var{\stackiteml} \cons \genv}{\dump}{\code}{\stack}
  \\
    \usefmaca{\genv}{\dump \cons \pair{\code}{\stack}}{\tocode{\lam{\var}{\tmthree}}}{\estack}
  & \tomachctwo &
    \usefmaca{\genv}{\dump}{\code}{\herval{(\tocode{\lam{\var}{\tmthree}})}\cons\stack}
  \\
    \usefmaca{\genv}{\dump\cons \pair{\code}{\stack}}{\tocode{\var}}{\stacktwo}
  & \tomachcthree &
    \usefmaca{\genv}{\dump}{\code}{\pair{\tocode{\var}}{\stacktwo}^\dead\cons\stack}
  & \tocode{\var} \notin \domSctx{\genv}
  \\
    \usefmaca{\genv_1 \cons \esub\var{\phi^\dead} \cons \genv_2}{\dump\cons\pair{\code}{\stack}}{\tocode{\var}}{\stacktwo}
  & \tomachcfour &
    \usefmaca{\genv_1 \cons \esub\var{\phi^\dead} \cons \genv_2}{\dump}{\code}{\pair{\tocode{\var}}{\stacktwo}^\dead\cons\stack}
  \\
    \usefmaca{\genv_1 \cons \esub\var{\herval{\codethree}} \cons \genv_2}{\dump\cons\pair{\code}{\stack}}{\tocode{\var}}{\estack}
  & \tomachcfive &
    \usefmaca{\genv_1 \cons \esub\var{\herval{\codethree}} \cons \genv_2}{\dump}{\code}{\herval{\tocode{\var}}\cons\stack}
  \\
    \usefmaca{\genv_1 \cons \esub\var{\herval\codethree} \cons \genv_2}{\dump}{\tocode{\var}}{\stackiteml\cons\stack}
  & \tomachue &
    \usefmaca{\genv_1 \cons \esub\var{\herval\codethree} \cons \genv_2}{\dump}{\rename{\codethree}}{\stackiteml\cons\stack}
  \end{array}
\]
where $\rename{\codetwo}$ is any code $\alpha$-equivalent to
$\codetwo$ that preserves well-naming of the machine. Note that in
\cite{AccattoliC15} there are two syntactically different sorts of
variables, corresponding to free and bound variables
respectively. Here we use only one sort of variables,
but this is just a matter of presentation.
\end{definition}

We say that a state $\statei$ is \defn{initial} if and only if it is 
of the form $\statei = \glamourst\estack\code\estack\estack$ for some 
code $\code$.
A state $\state$ is \defn{reachable} if and only if there exists an 
initial state $\statei$ such that $\statei \tomachhole{}^* \state$.

Let $\statei = \glamourst\estack{\code_0}\estack\estack$ be an initial 
state in the \glamour, and $\genv$ a global environment.
We say that a stack item $\stackitem$ is \defn{$\genv$-\stabilized} 
if and only if:
\begin{enumerate}
\item
  $\stackiteml = \code$ 
  implies $\decode{\stackiteml} = \tm \in \HAbs{\aset}$,
  where $\aset = \expansion{\emptyset}{\decgenv}$ is an abstraction frame,
  and $\decgenv$ is the list of substitution contexts resulting
  from the decoding of $\genv$.
  This is equivalent to saying that if $\lab = \alive$ then $\decode{\stackiteml} \in \HAbs{\aset}$.
\item
  $\stackiteml = \pair{\code}{\stack}$
  implies $\decode{\stackiteml} = \decodep{\stack}{\tm} \in \Struct{\sset}$,
  where $\sset = \expansion{\fv{\code_0}}{\decgenv}$ is a structure frame,
  and $\decgenv$ is the list of substitution contexts resulting
  from the decoding of $\genv$.
  This is equivalent to saying that if $\lab = \dead$ then $\decode{\stackiteml} \in \Struct{\sset}$.
\end{enumerate}

\begin{remark}
\label{rem:stable_when_removing_dumps_stacks}
\quad
\begin{enumerate}
\item
  $\decodep{\dump}{\tm} \in \Stable{\aset}{\sset}$
  if and only if $\tm \in \Stable{\aset}{\sset}$
  and $\decode{\dump} \in \Stable{\aset}{\sset}$.
\item
  $\decodep{\stack}{\tm} \in \Stable{\aset}{\sset}$ if and only if
  $\tm \in \Stable{\aset}{\sset}$
  and $\decode{\stack} \in \Stable{\aset}{\sset}$.
\item
  $\decodep{\genv}{\tm} \in \Stable{\aset}{\sset}$ if and only if
  $\tm \in \Stable{\expansion{\aset}{\decode{\genv}}}{\expansion{\sset}{\decode{\genv}}}$
  and $\decode{\genv} \in \Stable{\aset}{\sset}$.
\end{enumerate}
\end{remark}

\begin{lemma}[(New) Invariants for the \glamour]
\label{lem:invariants_glamour}
Let $\stater$ be a reachable state from an initial state
$\glamourst\estack{\code_0}\estack\estack$. 
Then, $\stater$ verifies the following invariant:
\begin{enumerate}
\item
  $\decode\stater \in \Stable\emptyset{\fv{\tm_0}}$
\item
  If $\stater = \glamourst\dump\code{\stack_1\cons\stackiteml\cons\stack_2}\genv$,
  then $\stackiteml$ is $\genv$-\stabilized.
\item
  If $\stater = \glamourst\dump\code\stack{\genv_1\cons\esub\var\stackiteml\cons\genv_2}$,
  then $\stackiteml$ is $\genv_2$-\stabilized.
\end{enumerate}
\end{lemma}
% Label: lem:invariants_glamour

\begin{proof}
We prove that each transition step $\tomachhole{}$ in the \glamour machine
preserves the invariant \ie, if $\state \tomachhole{} \statetwo$ and $\state$ verifies the invariant,
then $\statetwo$ verifies the invariant as well.
We proceed by induction on the transition step $\tomachhole{}$.
\begin{itemize}
\item
  If $\state = \glamourst{\dump}{\code\,\codethree}{\stack}{\genv} \tomachcone 
  \glamourst{\dump\cons\pair{\code}{\stack}}{\codethree}{\estack}{\genv} = \statetwo$
  then
  \[
    \begin{array}{rcl}
      \decode{\glamourst{\dump}{\code\,\codethree}{\stack}{\genv}}
    & = &
      \ctxof{(\decodep{\genv}{\decodep{\dump}{\decode{\stack}}})}{\tm\,\tmthree}
    \\
    & = &
      \decodep{\genv}{\decodep{\dump}{\decodep{\stack}{\tm\,\tmthree}}}
    \\
    & = &
      \ctxof{(\decodep{\genv}{\decodep{\dump}{\decodep{\stack}{\tm\ectx}}})}{\tmthree}
    \\
    & =  &
    \ctxof{(\decodep{\genv}{(\decodep{\dump}{\decodep{\stack}{\tm\ectx}})\ctxholep{\ectx}})}{\tmthree}
    \\
    & = &
      \ctxof{(\decodep{\genv}{(\decodep{\dump}{\decodep{\stack}{\tm\ectx}})\ctxholep{\decode{\estack}}})}{\tmthree}
    \\
    & = &
      \ctxof{(\decodep{\genv}{\decode{\dump\cons\pair{\code}{\stack}}\ctxholep{\decode{\estack}}})}{\tmthree}
    \\
    & = &
      \decode{\glamourst{\dump\cons\pair{\code}{\stack}}{\codethree}{\estack}{\genv}}
    \end{array}
  \]
  Therefore $\decode{\statetwo} \in \Stable{\emptyset}{\fv{\tm_0}}$ holds 
  as the translation of both states is the same.
  Item 2 of the invariant trivially holds
  for $\statetwo$ since it has an empty stack, 
  and item 3 holds for $\statetwo$
  since it holds for $\state$ which has the same environment $\genv$.
\item
  If $\state = \glamourst{\dump}{\lam{\var}{\code}}{\stackiteml\cons\stack}{\genv} \tomachum
  \glamourst{\dump}{\code}{\stack}{\esub\var{\stackiteml} \cons \genv} = \statetwo$,
  then $\state$ verifies the invariant:
  \begin{enumerate}
  \item[1'.]
    $\decode{\glamourst{\dump}{\lam{\var}{\code}}{\stackiteml\cons\stack}{\genv}} 
    = \decodep{\genv}{\decodep{\dump}{\decodep{\stack}{(\lam{\var}{\tm})\,\decode{\stackiteml}}}}
    \in \Stable{\emptyset}{\fv{\tm_0}}$. 
    Then by \cref{rem:stable_when_removing_dumps_stacks}
    $\decode{\dump}$, $\decode{\stack}$ and $\decode{\stackiteml}$ are all in 
    $\Stable{\expansion{\emptyset}{\decode{\genv}}}{\expansion{\fv{\tm_0}}{\decode{\genv}}}$,
    and $\decode{\genv}\in\Stable{\emptyset}{\fv{\tm_0}}$
  \item[2'.]
    $\stackiteml$ is $\genv$-\stabilized, and
    if $\stack$ is of the form $\stack_1 \cons \stackitemtwol \cons \stack_2$,
    then $\stackitemtwol$ is $\genv$-\stabilized
  \item[3'.]
    If $\genv$ is of the form $\genv_1 \cons \esub{\vartwo}{\stackitemtwol} \cons \genv_2$,
    then $\stackitemtwol$ is $\genv_2$-\stabilized.
  \end{enumerate}
  Let us show that the invariant holds for $\statetwo$:
  \begin{enumerate}
  \item
    $\decodep{\genv}{\decodep{\dump}{\decodep{\stack}{\tm\esub{\var}{\decode{\stackiteml}}}}}
    \in \Stable{\emptyset}{\fv{\tm_0}}$:
    this holds directly from (1').
  \item 
    If $\stack$ is of the form $\stack_1 \cons \stackitemtwol \cons \stack_2$,
    then $\stackitemtwol$ is $\genv$-\stabilized:
    this holds directly from (2').
  \item 
    $\stackiteml$ is $\genv$-\stabilized, by item (2'), and the rest of this condition holds by (3').
  \end{enumerate}
\item 
  If $\state = \glamourst{\dump}{\var}{\stackiteml\cons\stack}{\genv}
  \tomachue \glamourst{\dump}{\rename{\code}}{\stackiteml\cons\stack}{\genv} = \statetwo$, 
  where $\genv = \genv_1 \cons \esub\var{\herval\code} \cons \genv_2$,
  then $\state$ verifies the invariant:
  \begin{enumerate}
  \item[1'.]
    $\decode{\glamourst{\dump}{\var}{\stackiteml\cons\stack}{\genv_1 \cons \esub\var{\herval\code} \cons \genv_2}} 
    =
    \decodep{\genv_2}{\decodep{\dump}{\decodep{\stack}{\var \, \decode{\stackiteml}}}\decode{\genv_1\esub{\var}{\code^\alive}}}
    \in \Stable{\emptyset}{\fv{\tm_0}}$.
    Then by \cref{rem:stable_when_removing_dumps_stacks}
    $\decode{\dump}$, $\decode{\stack}$ and $\decode{\stackiteml}$ are all in 
    $\Stable{\expansion{\emptyset}{\decode{\genv}}}{\expansion{\fv{\tm_0}}{\decode{\genv}}}$,
    and $\decode{\genv} = \decode{\genv_1\cons\esub{\var}{\code^\alive}\cons\genv_2}\in\Stable{\emptyset}{\fv{\tm_0}}$
  \item[2'.]
    $\stackiteml$ is $\genv$-\stabilized
  \item[3'.]
    $\code^\alive$ is $\genv_2$-\stabilized
  \end{enumerate}
  Let us show that the invariant holds for $\statetwo$:
  \begin{enumerate}
  \item
    $\decodep{\genv_2}{\decodep{\dump}{\decodep{\stack}{\tm \, \decode{\stackiteml}}}\decode{\genv_1\esub{\var}{\code^\alive}}}
     \in \Stable{\emptyset}{\fv{\tm_0}}$: this holds directly from (1')
  \item
    $\stackiteml$ is $\genv_1\cons\esub{\var}{\code^\alive}\cons\genv_2$-\stabilized:
    this holds directly from (2')
  \item
    $\code^\alive$ is $\genv_2$-\stabilized:
    this holds directly from (3').
  \end{enumerate}
\item
  If $\state = \glamourst{\dump \cons \pair{\code}{\stack}}{\lam{\var}{\codethree}}{\estack}{\genv} \tomachctwo 
  \glamourst{\dump}{\code}{\herval{(\lam{\var}{\codethree})}\cons\stack}{\genv} = \statetwo$
  then
  \begin{align*}
    \decode{\glamourst{\dump \cons \pair{\code}{\stack}}{\lam{\var}{\codethree}}{\estack}{\genv}}
  & =
    \ctxof{(\decodep{\genv}{\decodep{\dump \cons \pair{\code}{\stack}}{\decode{\estack}}})}{\lam{\var}{\tmthree}}
  \\
  & =
    \ctxof{
      (\decodep{\genv}{\ctxof{\decodep{\dump}{\decodep{\stack}{\tm\ectx}}}{\decode{\estack}}})
    }{
      \lam{\var}{\tmthree}
    }
  \\
  & =
    \ctxof{
      (\decodep{\genv}{\ctxof{\decodep{\dump}{\decodep{\stack}{\tm\ectx}}}{\ectx}})
    }{
      \lam{\var}{\tmthree}
    }
  \\
  & =
    \ctxof{
      (\decodep{\genv}{\decodep{\dump}{\decodep{\stack}{\tm\ectx}}})
    }{
      \lam{\var}{\tmthree}
    }
  \\
  & =
    (\decodep{\genv}{\decodep{\dump}{\decodep{\stack}{\tm\,\lam{\var}{\tmthree}}}})
  \\
  & =
    \ctxof{(\decodep{\genv}{\decodep{\dump}{\decodep{\stack}{\ectx{(\lam{\var}{\tmthree})}}}})}{\tm}
  \\
  & =
    \ctxof{(\decodep{\genv}{\decodep{\dump}{\decode{\herval{(\lam{\var}{\codethree})} \cons \stack}}})}{\tm}
  \\
  & =
    \decode{\glamourst{\dump}{\code}{\herval{(\lam{\var}{\codethree})} \cons \stack}{\genv}}
  \end{align*}
  % \[
  %   \begin{array}{rcl}
  %     \decode{\glamourst{\dump \cons \pair{\code}{\stack}}{\lam{\var}{\codethree}}{\estack}{\genv}}
  %   & = &
  %     \ctxof{(\decodep{\genv}{\decodep{\dump \cons \pair{\code}{\stack}}{\decode{\estack}}})}{\lam{\var}{\tmthree}}
  %   \\
  %   & = &
  %     \ctxof{
  %       (\decodep{\genv}{\ctxof{\decodep{\dump}{\decodep{\stack}{\tm\ectx}}}{\decode{\estack}}})
  %     }{
  %       \lam{\var}{\tmthree}
  %     }
  %   \\
  %   & = &
  %     \ctxof{
  %       (\decodep{\genv}{\ctxof{\decodep{\dump}{\decodep{\stack}{\tm\ectx}}}{\ectx}})
  %     }{
  %       \lam{\var}{\tmthree}
  %     }
  %   \\
  %   & = &
  %     \ctxof{
  %       (\decodep{\genv}{\decodep{\dump}{\decodep{\stack}{\tm\ectx}}})
  %     }{
  %       \lam{\var}{\tmthree}
  %     }
  %   \\
  %   & = &
  %     (\decodep{\genv}{\decodep{\dump}{\decodep{\stack}{\tm\,\lam{\var}{\tmthree}}}})
  %   \\
  %   & = &
  %     \ctxof{(\decodep{\genv}{\decodep{\dump}{\decodep{\stack}{\ectx{(\lam{\var}{\tmthree})}}}})}{\tm}
  %   \\
  %   & = &
  %     \ctxof{(\decodep{\genv}{\decodep{\dump}{\decode{\herval{(\lam{\var}{\codethree})} \cons \stack}}})}{\tm}
  %   \\
  %   & = &
  %     \decode{\glamourst{\dump}{\code}{\herval{(\lam{\var}{\codethree})} \cons \stack}{\genv}}
  %   \end{array}
  % \]
  Therefore $\decode{\statetwo} \in \Stable{\emptyset}{\fv{\tm_0}}$ holds 
  as the translation of both states is the same.
  Item 2 of the invariant trivially holds
  for $\statetwo$ since $\lam{\var}{\tmthree} \in \HAbs{\aset}$
  for any abstraction frame $\aset$.
  Item 3 holds for $\statetwo$ since it holds for $\state$ which has the same environment $\genv$.
\item
  If
  $\state = \glamourst{\dump\cons \pair{\code}{\stack}}{\var}{\stacktwo}{\genv}
    \tomachcthree
    \glamourst{\dump}{\code}{\pair{\var}{\stacktwo}^\dead\cons\stack}{\genv} = \statetwo$,
  with $\var \notin \domSctx{\genv}$, then
  \[
    \begin{array}{rcl}
      \decode{\glamourst{\dump\cons \pair{\code}{\stack}}{\var}{\stacktwo}{\genv}}
    & = &
      \ctxof{
        (\decodep{\genv}{\decodep{\dump \cons \pair{\code}{\stack}}{\decode{\stacktwo}}})
      }{
        \var
      }
    \\
    & = &
      \ctxof{
        (\decodep{\genv}{\ctxof{\decodep{\dump}{\decodep{\stack}{\tm\ectx}}}{\decode{\stacktwo}}})
      }{
        \var
      }
    \\
    & = &
      \ctxof{
        (\decodep{\genv}{\decodep{\dump}{\decodep{\stack}{\tm \, \decode{\stacktwo}}}})
      }{
        \var
      }
    \\
    & = &
      \decodep{\genv}{\decodep{\dump}{\decodep{\stack}{\tm \, \decodep{\stacktwo}{\var}}}}
    \\
    & = &
      \ctxof{(\decodep{\genv}{\decodep{\dump}{\decodep{\stack}{\ectx{\decodep{\stacktwo}{\var}}}}})}{\tm}
    \\
    & = &
      \ctxof{(\decodep{\genv}{\decodep{\dump}{\decodep{\stack}{\ectx{\decode{\pair{\var}{\stacktwo}^\dead}}}}})}{\tm}
    \\
    & = &
      \ctxof{(\decodep{\genv}{\decodep{\dump}{\decode{\pair{\var}{\stacktwo}^\dead\cons\stack}}})}{\tm}
    \\
    & = &
      \decode{\glamourst{\dump}{\code}{\pair{\var}{\stacktwo}^\dead\cons\stack}{\genv}}
    \end{array}
  \]
  Therefore $\decode{\statetwo} \in \Stable{\emptyset}{\fv{\tm_0}}$ holds 
  as the translation of both states is the same.
  Item 2 of the invariant holds for $\statetwo$
  since $\decodep{\stacktwo}{\var} \in \Struct{\expansion{\fv{\tm_0}}{\decode{\genv}}}$,
  given the fact that $\var \in \fv{\tm_0}$.
  And item 3 holds for $\statetwo$ since it holds for $\state$ which has the same environment $\genv$.
\item
  If $\state = \glamourst{\dump\cons\pair{\code}{\stack}}{\var}{\stacktwo}{\genv_1 \cons \esub\var{\phi^\dead} \cons \genv_2}
  \tomachcfour 
  \glamourst{\dump}{\code}{\pair{\var}{\stacktwo}^\dead\cons\stack}{\genv_1 \cons \esub\var{\phi^\dead} \cons \genv_2} = \statetwo$
  then
  \[
    \begin{array}{rcl}
      \decode{
        \glamourst
          {\dump\cons\pair{\code}{\stack}}
          {\var}
          {\stacktwo}
          {\genv_1\esub\var{\phi^\dead}\genv_2}
      }
    & = &
      \ctxof{
        (\decodep{
          \genv_1\esub\var{\phi^\dead}\genv_2
        }{
          \decodep{\dump \cons \pair{\code}{\stack}}{\decode{\stacktwo}}
        })
      }{
        \var
      }
    \\
    & = &
      \ctxof{
        (\ctxof{
          (\decodep{\genv_2}{\ectx\decode{\genv_1\esub\var{\phi^\dead}}})
        }{
          \decodep{\dump \cons \pair{\code}{\stack}}{\decode{\stacktwo}}
        })
      }{
        \var
      }
    \\
    & = &
      \ctxof{
        (\decodep{
          \genv_2
        }{
          \decodep{
            \dump \cons \pair{\code}{\stack}
          }{
            \decode{\stacktwo}
          }\decode{\genv_1\esub\var{\phi^\dead}}
        })
      }{
        \var
      }
    \\
    & = &
      \ctxof{
        (\decodep{
          \genv_2
        }{
          \ctxof{
            \decodep{\dump}{\decodep{\stack}{\tm\ectx}}
          }{
            \decode{\stacktwo}
          }\decode{\genv_1\esub\var{\phi^\dead}}
        })
      }{
        \var
      }
    \\
    & = &
      \ctxof{
        (\decodep{
          \genv_2
        }{
          \decodep{
            \dump
          }{
            \decodep{\stack}{\tm \, \decode{\stacktwo}}
          }\decode{\genv_1\esub\var{\phi^\dead}}
        })
      }{
        \var
      }
    \\
    & = &
      \decodep{
        \genv_2
      }{
        \decodep{
          \dump
        }{
          \decodep{
            \stack
          }{
            \tm \, \decodep{\stacktwo}{\var}
          }
        }\decode{\genv_1\esub\var{\phi^\dead}}
      }
    \\
    & = &
      \ctxof{
        ({\decodep{
          \genv_2
        }{
          \decodep{
            \dump
          }{
            \decodep{\stack}{\ectx\decodep{\stacktwo}{\var}}
          }\decode{\genv_1\esub\var{\phi^\dead}}
        }})
      }{
        \tm
      }
    \\
    & = &
      \ctxof{
        ({\decodep{
          \genv_2
        }{
          \decodep{
            \dump
          }{
            \decodep{\stack}{\ectx\decode{\pair{\var}{\stacktwo}^\dead}}
          }\decode{\genv_1\esub\var{\phi^\dead}}
        }})
      }{
        \tm
      }
    \\
    & = &
      \ctxof{
        ({\decodep{
          \genv_2
        }{
          \decodep{
            \dump
          }{
            \decode{\pair{\var}{\stacktwo}^\dead\cons\stack}
          }\decode{\genv_1\esub\var{\phi^\dead}}
        }})
      }{
        \tm
      }
    \\
    & = &
      \ctxof{
        (\ctxof{
          {\decodep{\genv_2}{\ectx\decode{\genv_1\esub\var{\phi^\dead}}}}
        }{
          \decodep{\dump}{\decode{\pair{\var}{\stacktwo}^\dead\cons\stack}}
        })
      }{
        \tm
      }
    \\
    & = &
      \ctxof{
        (\decodep{
          {\genv_1\esub\var{\phi^\dead}\genv_2}
        }{
          \decodep{\dump}{\decode{\pair{\var}{\stacktwo}^\dead\cons\stack}}
        })
      }{
        \tm
      }
    \\
    & = &
      \decode{
        \glamourst
          {\dump}
          {\code}
          {\pair{\var}{\stacktwo}^\dead\cons\stack}
          {\genv_1\esub\var{\phi^\dead}\genv_2}
      }
    \end{array}
  \]
  Therefore $\decode{\statetwo} \in \Stable{\emptyset}{\fv{\tm_0}}$ holds 
  as the translation of both states is the same.
  Item 2 of the invariant holds for $\statetwo$
  since $\decodep{\stacktwo}{\var} \in 
  \Struct{\expansion{\fv{\tm_0}}{\decode{\genv_1\cons\esub{\var}{\stackitem^\dead}\cons\genv_2}}}$,
  since the stack item is decorated with the label $\dead$.
  And item 3 holds for $\statetwo$ since it holds for $\state$ which has the same environment $\genv$.
\item
  If 
  $\state = \glamourst{\dump\cons\pair{\code}{\stack}}{\var}{\estack}{\genv_1 \cons \esub\var{\herval{\codethree}} \cons \genv_2}
  \tomachcfive \glamourst{\dump}{\code}{\herval{\var}\cons\stack}{\genv_1 \cons \esub\var{\herval{\codethree}} \cons \genv_2} = \statetwo$
  then
  \[
    \begin{array}{rcl}
      \decode{
        \glamourst
          {\dump\cons\pair{\code}{\stack}}
          {\var}
          {\estack}
          {\genv_1\esub\var{\herval{\codethree}}\genv_2}
      }
    & = &
      \ctxof{
        (\decodep{
          \genv_1\esub{\var}{\herval{\codethree}}\genv_2
        }{
          \decodep{\dump \cons \pair{\code}{\stack}}{\decode{\estack}}
        })
      }{
        \var
      }
    \\
    & = &
      \ctxof{
        (\ctxof{
          \decodep{\genv_2}{\ectx\decode{\genv_1\esub{\var}{\herval{\codethree}}}}
        }{
          \decodep{\dump \cons \pair{\code}{\stack}}{\decode{\estack}}
        })
      }{
        \var
      }
    \\
    & = &
      \ctxof{
        (\decodep{
          \genv_2
        }{
          \decodep{
            \dump \cons \pair{\code}{\stack}
          }{
            \decode{\estack}
          }\decode{\genv_1\esub{\var}{\herval{\codethree}}}
        })
      }{
        \var
      }
    \\
    & = &
      \ctxof{
        (\decodep{
          \genv_2
        }{
          \ctxof{
            \decodep{\dump}{\decodep{\stack}{\tm\ectx}}
          }{
            \decode{\estack}
          }\decode{\genv_1\esub{\var}{\herval{\codethree}}}
        })
      }{
        \var
      }
    \\
    & = &
      \ctxof{
        (\decodep{
          \genv_2
        }{
          \ctxof{
            \decodep{\dump}{\decodep{\stack}{\tm\ectx}}
          }{
            \ectx
          }\decode{\genv_1\esub{\var}{\herval{\codethree}}}
        })
      }{
        \var
      }
    \\
    & = &
      \ctxof{
        (\decodep{
          \genv_2
        }{
          \decodep{\dump}{\decodep{\stack}{\tm\ectx}}
          \decode{\genv_1\esub{\var}{\herval{\codethree}}}
        })
      }{
        \var
      }
    \\
    & = &
      \decodep{
        \genv_2
      }{
        \decodep{
          \dump
        }{
          \decodep{
            \stack
          }{
            \tm \, \var
          }
        }\decode{\genv_1\esub{\var}{\herval{\codethree}}}
      }
    \\
    & = &
      \ctxof{
        ({\decodep{
          \genv_2
        }{
          \decodep{
            \dump
          }{
            \decodep{\stack}{\ectx\var}
          }\decode{\genv_1\esub{\var}{\herval{\codethree}}}
        }})
      }{
        \tm
      }
    \\
    & = &
      \ctxof{
        ({\decodep{
          \genv_2
        }{
          \decodep{
            \dump
          }{
            \decode{\herval{\var}\cons\stack}
          }\decode{\genv_1\esub{\var}{\herval{\codethree}}}
        }})
      }{
        \tm
      }
    \\
    & = &
      \ctxof{
        (\ctxof{
          {\decodep{\genv_2}{\ectx\decode{\genv_1\esub{\var}{\herval{\codethree}}}}}
        }{
          \decodep{\dump}{\decode{\herval{\var}\cons\stack}}
        })
      }{
        \tm
      }
    \\
    & = &
      \ctxof{
        (\decodep{
          {\genv_1\esub{\var}{\herval{\codethree}}\genv_2}
        }{
          \decodep{\dump}{\decode{\herval{\var}\cons\stack}}
        })
      }{
        \tm
      }
    \\
    & = &
      \decode{
        \glamourst
          {\dump}{\code}{\herval{\var}\cons\stack}
          {\genv_1\esub{\var}{\herval{\codethree}}\genv_2}
      }
    \end{array}
  \]
  Therefore $\decode{\statetwo} \in \Stable{\emptyset}{\fv{\tm_0}}$ holds 
  as the translation of both states is the same.
  Item 2 of the invariant holds since 
  $\var \in \HAbs{\expansion{\emptyset}{\decode{\genv_1\cons\genv_2}} \cup \set{\var}}$,
  since the code $\codethree$ is decorated with the label $\alive$.
  And item 3 for $\statetwo$ holds since it holds for $\state$ which has the same environment $\genv$.
\end{itemize}
\end{proof}

\begin{lemma}
\label{lem:equivalence_in_translations}
Let $\var \notin \fv{\decode\stack}$ and $\var \notin \fv{\decode\dump}$.
Then:
\begin{enumerate}
\item
  $\decodep\stack{\tm\esub\var\tmtwo} \equiv \decodep\stack\tm\esub\var\tmtwo$
\item
  $\decodep\dump{\tm\esub\var\tmtwo} \equiv \decodep\dump\tm\esub\var\tmtwo$
\end{enumerate}
\end{lemma}
% Label: lem:equivalence_in_translations

\begin{proof}
\quad
\begin{enumerate}
  \item
    We proceed by induction on the structure of $\stack$.
    \begin{itemize}
    \item
      If $\stack = \estack$, then
      $
          \decodep{\estack}{\tm\esub{\var}{\tmtwo}}
        = \ctxof{\ectx}{\tm\esub{\var}{\tmtwo}}
        = \tm\esub{\var}{\tmtwo}
        \equiv
          \ctxof{\ectx}{\tm}\esub{\var}{\tmtwo}
        = \decodep{\estack}{\tm}\esub{\var}{\tmtwo}
      $.
    \item
      If $\stack = \stackiteml \cons \stacktwo$, then
      by $\alpha$-equivalence, we may assume
      in particular $\var \notin \fv{\decode{\stackiteml}}$.
      Therefore
      \[
        \begin{array}{rcll}
            \decodep{\stackiteml \cons \stacktwo}{\tm\esub{\var}{\tmtwo}}
          & = &
            \ctxof{\decodep{\stacktwo}{\ectx \, \decode{\stackiteml}}}{\tm\esub{\var}{\tmtwo}}
        \\
          & = & 
            \decodep{\stacktwo}{\tm\esub{\var}{\tmtwo} \, \decode{\stackiteml}}
        \\
          & \equiv & 
            \decodep{\stacktwo}{(\tm \, \decode{\stackiteml})\esub{\var}{\tmtwo}}
          & \text{(by $\ruleEquivEsLDist$, since $\var \notin \fv{\decode{\stackiteml}}$)}
        \\
          & \equiv & 
            \decodep{\stacktwo}{\tm \, \decode{\stackiteml}}\esub{\var}{\tmtwo}
          & \text{(by \ih on $\stacktwo$)}
        \\
          & = &
            \decodep{\stackiteml \cons \stacktwo}{\tm}\esub{\var}{\tmtwo}
        \end{array}
      \]
    \end{itemize}
  \item
    We proceed by induction on the structure of $\dump$.
    \begin{itemize}
    \item
      If $\dump = \estack$, then it is analogous to the base case in the previous item.
    \item
      If $\dump = \dumptwo \cons \pair{\codethree}{\stack}$, then
      by $\alpha$-equivalence, we may assume in particular
      $\var \notin \fv{\tmthree}$.
      Therefore
      \[
        \begin{array}{rcll}
            \decodep{\dumptwo \cons \pair{\codethree}{\stack}}{\tm\esub{\var}{\tmtwo}}
          & = &
            \ctxof{\decodep{\dumptwo}{\decodep{\stack}{\tmthree \, \ectx}}}{\tm\esub{\var}{\tmtwo}}
        \\
          & = & 
            \decodep{\dumptwo}{\decodep{\stack}{\tmthree \, \tm\esub{\var}{\tmtwo}}}
        \\
          & \equiv & 
            \decodep{\dumptwo}{\decodep{\stack}{(\tmthree \, \tm)\esub{\var}{\tmtwo}}}
          & \text{(by $\ruleEquivEsRDist$, since $\var \notin \fv{\tmthree}$)}
        \\
          & \equiv & 
            \decodep{\dumptwo}{\decodep{\stack}{\tmthree \, \tm}\esub{\var}{\tmtwo}}
          & \text{(by (1))}
        \\
          & \equiv & 
            \decodep{\dumptwo}{\decodep{\stack}{\tmthree \, \tm}}\esub{\var}{\tmtwo}
          & \text{(by \ih on $\dumptwo$)}\\
          & = &
            \decodep{\dumptwo \cons \pair{\codethree}{\stack}}{\tm}\esub{\var}{\tmtwo}
        \end{array}
      \]
    \end{itemize}
\end{enumerate}
\end{proof}

\glamourSimulation*
% Label: glamour_simulation

\begin{proof}
We proceed by case analysis of $\tomachhole{}$.
To lighten the proof, we simplify
the notation $\esub{\var}{\stackiteml} \cons \genv$
by $\esub{\var}{\stackiteml}\genv$.
\begin{enumerate}
\item
  If $\glamourst{\dump}{\lam{\var}{\code}}{\stackiteml\cons\stack}{\genv} \tomachum
  \glamourst{\dump}{\code}{\stack}{\esub\var{\stackiteml} \cons \genv}$
  then
  \[
    \begin{array}{rcll}
      \decode{\glamourst{\dump}{\lam{\var}{\code}}{\stackiteml\cons\stack}{\genv}}
    & = &
      \ctxof
        {(\decodep{\genv}{\decodep{\dump}{\decode{\stackiteml\cons\stack}}})}
        {\lam{\var}{\tm}}
    \\
    & = &
      \ctxof
        {(\decodep{\genv}{\decodep{\dump}{\decodep{\stack}{\ectx\decode{\stackiteml}}}})}
        {\lam{\var}{\tm}}
    \\
    & = &
      \decodep
        {\genv}
        {\decodep{\dump}{\decodep{\stack}{(\lam{\var}{\tm})\,\decode{\stackiteml}}}}
    \\
    & & \tostable{\ruledb}{\emptyset}{\fv{\statei}}{\nonapp}
    & (\star)
    \\
    & &
      \decodep{\genv}{\decodep{\dump}{\decodep{\stack}{\tm\esub{\var}{\decode{\stackiteml}}}}}
    \\
    & \equiv &
      \decodep{\genv}{\decodep{\dump}{\decodep{\stack}{\tm}}\esub{\var}{\decode{\stackiteml}}}
    & \text{By \cref{lem:equivalence_in_translations}} (*)
    \\
    & = &
      \ctxof{(\decodep{\genv}{\decodep{\dump}{\decode{\stack}}\esub{\var}{\decode{\stackiteml}}})}{\tm}
    \\
    & = &
      \ctxof
        {(\ctxof
          {\decodep{\genv}{\ectx\esub{\var}{\decode{\stackiteml}}}}
          {\decodep{\dump}{\decode{\stack}}})}
        {\tm}
    \\
    & = &
      \ctxof
        {(\decodep{\esub{\var}{\stackiteml}\genv}{\decodep{\dump}{\decode{\stack}}})}
        {\tm}
    \\
    & = &
      \decode{\glamourst{\dump}{\code}{\stack}{\esub{\var}{\stackiteml}\genv}}
    \end{array}
  \]
  The step $(\star)$ holds by the following:
  $\decodep{\genv}{\decodep{\dump}{\decodep{\stack}{(\lam{\var}{\tm})\,\decode{\stackiteml}}}}
  \in \Stable{\emptyset}{\fv{\tm_0}}$ and
  $\stackiteml$ is $\genv$-\stabilized
  by \cref{lem:invariants_glamour} (items 1 and 2) respectively.
  Then 
  $\decode{\stackiteml} \in \HAbs{\expansion{\emptyset}{\decode{\genv}}} 
    \cup \Struct{\expansion{\fv{\tm_0}}{\decode{\genv}}}$,
  so applying congruence reduction rules accordingly to reach the redex 
  $(\lam{\var}{\tm}) \, \decode{\stackiteml}$, we can reduce this subterm with
  rule $\ruleUDbStable$, so that
  $
    (\lam{\var}{\tm}) \, \decode{\stackiteml} 
    \tostable{\ruledb}{\expansion{\emptyset}{\decode{\genv}}}{\expansion{\fv{\tm_0}}{\decode{\genv}}}{\appflag}
    \tm\esub{\var}{\decode{\stackiteml}}
  $.

  On the other hand, the step $(*)$ holds applying rule $\ruleEquivEsLDist$, 
  as $\var \notin \fv{\decode{\stack}}$ and $\var \notin \fv{\decode{\dump}}$ 
  holds by $\alpha$-conversion.
\item 
  If $\glamourst{\dump}{\var}{\stackiteml\cons\stack}{\genv_1 \cons \esub\var{\herval{\tocode{\val}}} \cons \genv_2}
  \tomachue \glamourst{\dump}{\rename{\tocode{\val}}}{\stackiteml\cons\stack}{\genv_1 \cons \esub\var{\herval{\tocode{\val}}} \cons \genv_2}$
  then
  \[
    \begin{array}{rcll}
      \decode{
        \glamourst
          {\dump}{\var}{\stackiteml\cons\stack}
          {\genv_1\esub\var{\herval{\tocode\val}}\genv_2}
      }
    & = &
      \decodep{
        \evctx_{\glamourst{\dump}{\var}{\stackiteml\cons\stack}{\genv_1\esub\var{\herval{\tocode\val}}\genv_2}}
      }{
        \var
      }
    \\
    & = &
      \ctxof{
        (\decodep{
          \genv_1\esub{\var}{\herval{\tocode\val}}\genv_2
        }{
          \decodep{\dump}{\decode{\stackiteml\cons\stack}}
        })
      }{
        \var
      }
    \\
    & = &
      \ctxof{
        (\ctxof{
          (\decodep{\genv_2}{\ectx\decode{\genv_1\esub{\var}{\herval{\tocode\val}}}})
        }{
          \decodep{\dump}{\decode{\stackiteml\cons\stack}}
        })
      }{
        \var
      }
    \\
    & = &
      \ctxof{
        (\decodep{
          \genv_2
        }{
          \decodep{\dump}{\decode{\stackiteml\cons\stack}
          }\decode{\genv_1\esub{\var}{\herval{\tocode\val}}}
        })
      }{
        \var
      }
    \\
    & = &
      \ctxof{
        (\decodep{
          \genv_2
        }{
          \decodep{\dump}{\decodep{\stack}{\ectx\decode{\stackiteml}}
          }\decode{\genv_1\esub{\var}{\herval{\tocode\val}}}
        })
      }{
        \var
      }
    \\
    & = &
      \decodep{\genv_2}{
        \decodep{\dump}{
          \decodep{\stack}{\var \, \decode{\stackiteml}}
        }\decode{\genv_1\esub{\var}{\herval{\tocode\val}}}
      }
    \\
    & & \tostable{\rulelsv}{\emptyset}{\fv{\tm_0}}{\nonapp}
    & (\star)
    \\
    & &
      \decodep{
        \genv_2
      }{
        \decodep{
          \dump
        }{
          \decodep{
            \stack
          }{
            \tm \, \decode{\stackiteml}
          }
        }\decode{\genv_1\esub{\var}{\herval{\tocode\val}}}
      }
    \\
    & = &
      \ctxof{
        ({\decodep{
          \genv_2
        }{
          \decodep{
            \dump
          }{
            \decodep{\stack}{\ectx\decode{\stackiteml}}
          }\decode{\genv_1\esub{\var}{\herval{\tocode\val}}}
        }})
      }{
        \tm
      }
    \\
    & = &
      \ctxof{
        ({\decodep{
          \genv_2
        }{
          \decodep{
            \dump
          }{
            \decode{\stackiteml\cons\stack}
          }\decode{\genv_1\esub{\var}{\herval{\tocode\val}}}
        }})
      }{
        \tm
      }
    \\
    & = &
      \ctxof{
        (\ctxof{
          {\decodep{\genv_2}{\ectx\decode{\genv_1\esub{\var}{\herval{\tocode\val}}}}}
        }{
          \decodep{\dump}{\decode{\stackiteml\cons\stack}}
        })
      }{
        \tm
      }
    \\
    & = &
      \ctxof{
        (\decodep{
          {\genv_1\esub{\var}{\herval{\tocode\val}}\genv_2}
        }{
          \decodep{\dump}{\decode{\stackiteml\cons\stack}}
        })
      }{
        \tm
      }
    \\
    & = &
      \decodep{
        \evctx_{\glamourst{\dump}{\rename{\tocode\val}}{\stackiteml\cons\stack}{\genv_1\esub\var{\herval{\tocode{\val}}}\genv_2}}
      }{
        \tm
      }
    \\
    & = &
      \decode{
        \glamourst
          {\dump}{\rename{\code}}{\stackiteml\cons\stack}
          {\genv_1\esub\var{\herval{\tocode\val}}\genv_2}
      }
    \end{array}
  \]
  The step $(\star)$ holds since
  $\decodep{\genv_2}
    {\decodep{\dump}{\decodep{\stack}{\var \, \decode{\stackiteml}}}\decode{\genv_1\esub{\var}{\herval{\tocode{\val}}}}}
  \in \Stable{\emptyset}{\fv{\tm_0}}$ 
  by item 1 of \cref{lem:invariants_glamour}.
  So applying congruence reduction
  rules and rule $\ruleULsv$ accordingly to reach the redex 
  $\var \, \decode{\stackiteml}$, we can reduce this subterm with
  rule $\ruleUAppL$, so that the reduction
  $
    \var \, \decode{\stackiteml}
    \tostable
      {\rulesub{\var}{\val}}
      {\expansion{\emptyset}{\decode{\genv_1\genv_2}} \cup \set{\var}}
      {\expansion{\fv{\tm_0}}{\decode{\genv_1\genv_2}}}{\appflag}
    \val \, \decode{\stackiteml}$,
  is derived from applying rule $\ruleUSub$:
  $
    \var
    \tostable
      {\rulesub{\var}{\val}}
      {\expansion{\emptyset}{\decode{\genv_1\genv_2}} \cup \set{\var}}
      {\expansion{\fv{\tm_0}}{\decode{\genv_1\genv_2}}}{\app}
    \val
  $.
\item
  If $\glamourst{\dump}{\code\,\codethree}{\stack}{\genv} \tomachcone
  \glamourst{\dump\cons\pair{\code}{\stack}}{\codethree}{\estack}{\genv}$
  then $\decode{\glamourst{\dump}{\code\,\codethree}{\stack}{\genv}} = 
    \decode{\glamourst{\dump\cons\pair{\code}{\stack}}{\codethree}{\estack}{\genv}}$,
    is proved as in case $\tomachcone$ in \cref{lem:invariants_glamour}.
\item
  If $\glamourst{\dump \cons \pair{\code}{\stack}}{\lam{\var}{\codethree}}{\estack}{\genv}
  \tomachctwo \glamourst{\dump}{\code}{\herval{(\lam{\var}{\codethree})}\cons\stack}{\genv}$
  then
  $\decode{\glamourst{\dump \cons \pair{\code}{\stack}}{\lam{\var}{\codethree}}{\estack}{\genv}} =
  \decode{\glamourst{\dump}{\code}{\herval{(\lam{\var}{\codethree})} \cons \stack}{\genv}}$, 
  is proved as in case $\tomachctwo$ in \cref{lem:invariants_glamour}.
\item
  If
  $\glamourst{\dump\cons \pair{\code}{\stack}}{\var}{\stacktwo}{\genv}
    \tomachcthree
    \glamourst{\dump}{\code}{\pair{\var}{\stacktwo}^\dead\cons\stack}{\genv}$,
  with $\var \notin \domSctx{\genv}$, then
  $\decode{\glamourst{\dump\cons \pair{\code}{\stack}}{\var}{\stacktwo}{\genv}}
  = \decode{\glamourst{\dump}{\code}{\pair{\var}{\stacktwo}^\dead\cons\stack}{\genv}}$,
  is proved as in case $\tomachcthree$ in \cref{lem:invariants_glamour}.
\item
  If $\glamourst{\dump\cons\pair{\code}{\stack}}{\var}{\stacktwo}{\genv_1 \cons \esub\var{\phi^\dead} \cons \genv_2}
  \tomachcfour \glamourst{\dump}{\code}{\pair{\var}{\stacktwo}^\dead\cons\stack}{\genv_1 \cons \esub\var{\phi^\dead} \cons \genv_2}$
  then
  $
    \decode{\glamourst{\dump\cons\pair{\code}{\stack}}{\var}{\stacktwo}{\genv_1\esub\var{\phi^\dead}\genv_2}} =
    \decode{\glamourst{\dump}{\code}{\pair{\var}{\stacktwo}^\dead\cons\stack}{\genv_1\esub\var{\phi^\dead}\genv_2}}
  $
  is proved as in case $\tomachcfour$ in \cref{lem:invariants_glamour}.
\item
  \sloppy
  If $\glamourst{\dump\cons\pair{\code}{\stack}}{\var}{\estack}{\genv_1 \cons \esub\var{\herval{\codethree}} \cons \genv_2}
  \tomachcfive \glamourst{\dump}{\code}{\herval{\var}\cons\stack}{\genv_1 \cons \esub\var{\herval{\codethree}} \cons \genv_2}$
  then
  $\decode{\glamourst{\dump\cons\pair{\code}{\stack}}{\var}{\estack}{\genv_1\esub\var{\herval{\codethree}}\genv_2}}\
  = \decode{\glamourst{\dump}{\code}{\herval{\var}\cons\stack}{\genv_1\esub{\var}{\herval{\codethree}}\genv_2}}$
  is proved as in case $\tomachcfive$ in \cref{lem:invariants_glamour}.
  
\end{enumerate}
\end{proof}

Now, we give the main results of this section.
On the first hand, we have the high-level implementation theorem:
\highLevelImplementation*
\begin{proof}
Let $\tm \tobetafireball^n \tm'$.
By Thm.~8, Thm.~9 and Coro.~1 in \cite{AccattoliC15},
there is a sequence of $p$ transitions
$\state \tomachhole{}^p \statetwo$
in the \glamour,
where $\state$ is an initial state such that $\decode{\state} = \tm$
and $p \in O(|\tm| \cdot (n^2+1))$.
By \cref{glamour_simulation,structeq_bisimulation}
there exists a term $\tm''$ such that
$\tm = \decode{\state} \toustablen{k} \tm'' \equiv \decode{\state'}$
where $k \leq p$, so $k \in O(|\tm| \cdot (n^2+1))$.
To conclude, we are left to show that $\unfold{\tm''} = \tm'$.
It is easy to see that $\tm'' \equiv \decode{\state'}$
implies $\unfold{\tm''} = \unfold{\decode{\state'}}$,
so it suffices to show that $\unfold{\decode{\state'}} = \tm'$.
This is a consequence of Thm.~3, Def.~1, and Thm.~8 in~\cite{AccattoliC15}.
\end{proof}

Next we move on the low-level implementation theorem:
\lowLevelImplementation*
\begin{proof}
First we claim that the \glamour terminates when starting from the initial
state $\state$.
Indeed, $\tomachhole{}$ can be written as the union of $\tomachume$
and $\tomachhole{\admsym_1,..,\admsym_5}$.
The relation $\tomachhole{\admsym_1,..,\admsym_5}$ is known to be terminating
from~\cite{AccattoliC15},
so an infinite reduction $\state \tomachhole{}\tomachhole{}\hdots$
must contain an infinite number of $\tomachume$ steps.
By \cref{glamour_simulation,structeq_bisimulation}
this means that there must exist an infinite $\toustable$ reduction
starting from $\decode{\state} = \tm$.
This is impossible because $\tm$ is known to have a normal form
and $\toustable$ has the diamond property.

Now let $\state \tomachhole{}^k \state'$ be a reduction to normal form in the
\glamour
containing $m$ multiplicative, $e$ exponential, and $c$ administrative steps,
so $k = m + e + c$.
By \cref{glamour_simulation,structeq_bisimulation}
we have $\tm = \decode{\state} \toustablen{m+e} \tm'' \equiv \decode{\state'}$.
Moreover $\tm \toustablen{n} \tm'$ by hypothesis.
Since $\state'$ is $\tomachhole{}$-normal, we know by \cref{glamour_simulation}
that $\decode{\state'}$ is $\toustable$-normal, so $\tm''$ is also $\toustable$-normal.
But $\toustable$ has the diamond property, so $\tm' = \tm''$ and $n = m + e$.
To conclude, we are left to show that $k \in O(|t|(n+1))$.
By Lemma~6 in~\cite{AccattoliC15} we know that $c \in O(|t|(e+1))$.
Since $n = m + e$, finally we have that $k = m + e + c \in O(|t|(n+1))$.
\end{proof}

\section{Proofs of Section~\ref{sec:typing} ``A Quantitative Interpretation''}
\label{app:typing}
In this section we show the results concerning the typing system $\typesystem$.
We start with general lemmas and remarks used through this section, and then
we show soundness and completeness in
\cref{sec:soundness_results} and \cref{sec:completeness_results} respectively.

\relevance*
% Label: lem:relevance

\begin{proof}
By induction on $\deriv$. We omit cases \ruleTypAppC and \ruleTypES 
as they are analogous to case \ruleTypAppP.
\begin{enumerate}
\item \ruleTypVar.
  Then, $\derivs\deriv{\judg[0, \numarrows\mtyp]{\var : \mtyp}\var\mtyp}$, 
  where $\tctx = \var : \mtyp$ and $\cm = 0$ and $\ce = \numarrows\mtyp$ 
  and $\tm = \var$.
  Therefore, $\dom{\var : \mtyp} = \set\var = \rv\var = \fv\var$.
\item \ruleTypAbs.
  Then
  \[
    \inferrule{
      \left(
        \derivs{\deriv_i}{
          \judg[\cm_i, \ce_i]{\tctx_i; \var : \optmtyptwo_i}\tmtwo{\mtypthree_i}
        }
      \right)_\iI
    }{
      \judg[+_\iI \cm_i, +_\iI \ce_i]
        {+_\iI \tctx_i}{\lam\var\tmtwo}{\mset{\optmtyptwo_i \to \mtypthree_i}_\iI}
    }\ruleTypAbs
  \]
  where $\tctx = +_\iI \tctx_i$ and $\cm = +_\iI \cm_i$ and 
  $\ce = +_\iI \ce_i$ and $\tm = \lam\var\tmtwo$ and 
  $\mtyp = \mset{\optmtyptwo_i \to \mtypthree_i}_\iI$.
  
  We can apply the \ih on $(\deriv_i)_\iI$, yielding 
  $\rv\tmtwo \subseteq \dom{\tctx_i; \var : \optmtyptwo_i} \subseteq \fv\tmtwo$
  for all $\iI$.
  For each $\iI$, we have to separate in cases, depending on whether 
  $\optmtyptwo_i = \none$.
  
  Since both cases are analogous, we only focus on
  the case in which $\optmtyptwo \neq \none$.
  So we have
  $\bigcup_\iI \rv{\tmtwo} \subseteq
  \bigcup_\iI \dom{\tctx_i; \var : \optmtyptwo_i} \subseteq
  \bigcup_\iI \dom{\tctx_i} \cup\set{\var} \subseteq \fv{\tmtwo}$.
  By removing $\var$ from the inequalities, we obtain:
  \[
      \rv{\lam{\var}{\tmtwo}}
    = \emptyset
    \subseteq
      \dom{+_\iI \tctx_i} 
    = \bigcup_\iI \dom{\tctx_i} 
    \subseteq \fv{\tmtwo} \setminus \set{\var} = \fv{\lam{\var}{\tmtwo}}
  \]
\item \ruleTypAppP.
  Then
  \[
    \inferrule{
      \derivs{\deriv_1}{\judg[\cm_1, \ce_1]{\tctx_1}\tmtwo\tightN}
      \sep
      \derivs{\deriv_2}{\judg[\cm_2, \ce_2]{\tctx_2}\tmthree\tight}
    }{
      \judg[\cm_1 + \cm_2, \ce_1 + \ce_2]{\tctx_1 + \tctx_2}{\tmtwo \, \tmthree}\tightN
    }\ruleTypAppP
  \]
  where $\tctx = \tctx_1 + \tctx_2$ and $\cm = \cm_1 + \cm_2$ and 
  $\ce = \ce_1 + \ce_2$ and $\tm = \tmtwo \, \tmthree$ and 
  $\mtyp = \tightN$.

  By the \ih on both $\deriv_1$ and $\deriv_2$, we yield 
  $\rv\tmtwo \subseteq \dom{\tctx_1} \subseteq \fv\tmtwo$ and 
  $\rv\tmthree \subseteq \dom{\tctx_2} \subseteq \fv\tmthree$, respectively.
  We conclude that $\rv{\tmtwo \, \tmthree} 
  \subseteq \dom{\tctx_1 + \tctx_2} = \dom{\tctx_1} \cup \dom{\tctx_2} 
  \subseteq \fv{\tmtwo \, \tmthree}$.
% \item \ruleTypAppC.
%   Analogous to the previous case.
% \item \ruleTypES.
%   Then
%   \[
%     \inferrule{
%       \derivs{\deriv_1}{\judg[\cm_1, \ce_1]{\tctx_1; \var : \optmtyptwo}\tmtwo\mtyp}
%       \sep
%       \optmtyptwo \mleq \mtyptwo
%       \sep
%       \derivs{\deriv_2}{\judg[\cm_2, \ce_2]{\tctx_2}\tmthree\mtyptwo}
%     }{
%       \judg[\cm_1 + \cm_2, \ce_1 + \ce_2]
%         {\tctx_1 + \tctx_2}{\tmtwo\esub\var\tmthree}\mtyp
%     }\ruleTypES
%   \]
%   where $\tctx = \tctx_1 + \tctx_2$ and $\cm = \cm_1 + \cm_2$ and 
%   $\ce = \ce_1 + \ce_2$ and $\tm = \tmtwo\esub\var\tmthree$.
  
%   By the \ih on both $\deriv_1$ and $\deriv_2$, we yield 
%   $\rv\tmtwo \subseteq \dom{\tctx_1; \var : \optmtyptwo} \subseteq \fv\tmtwo$ 
%   and $\rv\tmthree \subseteq \dom{\tctx_2} \subseteq \fv\tmthree$, 
%   respectively.
%   By removing $\var$ from the inequations, we obtain
%   $\rv\tmtwo \setminus \set\var \subseteq \dom{\tctx_1} \subseteq \fv\tmtwo \setminus \set\var$,
%   so we conclude
%   $\rv{\tmtwo\esub\var\tmthree} = (\rv\tmtwo \setminus \set\var) \cup \rv\tmthree
%   \subseteq \dom{\tctx_1 + \tctx_2} = \dom{\tctx_1} \cup \dom{\tctx_2} 
%   \subseteq (\fv{\tmtwo} \setminus \set\var) \cup \fv\tmthree = \fv{\tmtwo\esub\var\tmthree}$.
\end{enumerate}
\end{proof}

Some (simple) properties of the notion of \emph{appropriateness}, 
defined in \cref{sec:typing}, are the following:
\begin{remark}
\label{rem:isAppropriate} \mbox{}
\begin{enumerate}
\item 
  If $\isAppr{\aset}{\var: \mtyp}$ and $\optmtyp\mleq\mtyp$ 
  then $\isAppr{\aset}{\var: \optmtyp}$.
\item 
  If $\isAppr{\aset}{\tctx}$ and for all $\var \in \aset', \var \notin \dom{\tctx}$ 
  then $\isAppr{\aset \cup \aset'}{\tctx}$.
\item 
  If $\isAppr{\aset}{\tctx}$ and $\isAppr{\aset}{\tctxtwo}$ 
  then $\isAppr{\aset}{\tctx + \tctxtwo}$.
\end{enumerate}
\end{remark}

\begin{lemma}[Types of hereditary abstractions]
\label{lem:types_of_hereditary_abstractions}
Let $\judg[\cm,\ce]{\tctx}{\tm}{\mtyp}$
where $\tm \in \HAbs{\aset}$
and $\isAppr{\aset}{\tctx}$.
Then $\mtyp \neq \tightN$.
\end{lemma}
% Label: lem:types_of_hereditary_abstractions

\begin{proof}
By induction on the derivation of $\judg[m,e]{\tctx}{\tm}{\mtyp}$.
\begin{enumerate}
\item $\ruleTypVar$.
  Let $\judg[0,n]{\var : \mtyp}{\var}{\mtyp}$, with $n = \numarrows{\mtyp}$. Moreover,
  $\var \in \HAbs{\aset}$ and $\isAppr{\aset}{\var : \mtyp}$.
  By the premise of rule $\ruleHAbsVar$, we know that $\var \in \aset$, and
  thus $\mtyp \neq \tightN$
  by definition of $\isAppr{\aset}{\var : \mtyp}$.
\item $\ruleTypAbs$.
  This case is immediate since the judgement is
  $\judg[+_\iI m_i, +_\iI e_i]{+_\iI\tctx_i}{\lam{\var}{\tm}}{
      \mset{\optmtyp_i\to\mtyptwo_i}_\iI}$,
  whose type is not equal to $\tightN$.
\item $\ruleTypAppP$.
  This case is not possible, since the term is of the form $\tmtwo \, \tmthree$,
  and the hypothesis $\tmtwo \, \tmthree \in \HAbs{\aset}$ does not hold.
\item $\ruleTypAppC$.
  Analogous to the previous case.
\item $\ruleTypES$.
  Then
  \[
    \indrule{\ruleTypES}{
      \judg[m_1,e_1]{\tctx_1; \var : \optmtyptwo}{\tmtwo}{\mtyp}
      \sep
      \optmtyptwo \mleq \mtyptwo
      \sep
      \judg[m_2,e_2]{\tctx_2}{\tmthree}{\mtyptwo}
    }{
      \judg[m_1 + m_2,e_1 +e_2]{\tctx_1 + \tctx_2}{\tmtwo\esub{\var}{\tmthree}}{\mtyp}
    }
  \]
  where 
  $\tctx = \tctx_1 + \tctx_2$,
  $m = m_1 + m_2$,
  $e = e_1 + e_2$
  and $\tm = \tmtwo\esub{\var}{\tmthree}$.
  Moreover, $\tmtwo\esub{\var}{\tmthree} \in \HAbs{\aset}$, which can be derived either by rule $\ruleHAbsSubi$ or rule $\ruleHAbsSubii$:
  \begin{enumerate}
  \item $\ruleHAbsSubi$.
   Then $\tmtwo \in \HAbs{\aset}$ and $\var \notin \aset$. Given that $\tctx = \tctx_1 + \tctx_2$, then 
    $\isAppr{\aset}{\tctx_1;\var : \optmtyptwo}$.
    We can apply \ih on $\tmtwo$, yielding $\mtyp \neq \tightN$.
  \item $\ruleHAbsSubii$.
    \sloppy
    Then $\tmtwo \in \HAbs{\aset \cup \set{\var}}$, $\var \notin \aset$
    and $\tmthree \in \HAbs{\aset}$.
    Given that $\tctx = \tctx_1 + \tctx_2$, then $\isAppr{\aset}{\tctx_2}$.
    We can apply \ih on $\tmthree$, yielding $\mtyptwo \neq \tightN$.
    Since $\optmtyptwo \mleq \mtyptwo$, we can conclude that $\optmtyptwo \neq \tightN$
    so $\isAppr{\aset \cup \set{\var}}{\tctx_1; \var : \optmtyptwo}$ holds.
    By \ih on $\tmtwo$, 
    we have $\mtyp \neq \tightN$.
  \end{enumerate}
\end{enumerate}
\end{proof}

\begin{lemma}[Splitting / Merging]
\label{lem:splitting}
Let $\mtyp_1,\mtyp_2$ be two types such that 
their sum $\mtyp_1+\mtyp_2$ is well-defined.
Then the following are equivalent:
\begin{enumerate}
\item
  $\judg[m,e]{\tctx}{\val}{\mtyp_1+\mtyp_2}$
\item
  There exist $\tctx_1,\tctx_2,m_1,e_1,m_2,e_2$ such that:
  \begin{enumerate}
  \item[(a)] $\judg[m_1,e_1]{\tctx_1}{\val}{\mtyp_1}$
  \item[(b)] $\judg[m_2,e_2]{\tctx_2}{\val}{\mtyp_2}$
  \item[(c)] $\tctx = \tctx_1 + \tctx_2$
        and $m = m_1 + m_2$
        and $e = e_1 + e_2$.
  \end{enumerate}
\end{enumerate}
\end{lemma}
% Label: lem:splitting

\begin{proof}
We first show that $(1 \implies 2)$.
The judgement can be derived using rule $\ruleTypVar$ or rule $\ruleTypAbs$:
\begin{itemize}
\item $\ruleTypVar$.
  Then $\judg[0,e]{\tctx}{\var}{\mtyp_1 + \mtyp_2}$, with
  $e = \numarrows{\mtyp_1 + \mtyp_2}$.
  Since the sum of the types is well defined, there are two subcases:
  \begin{enumerate}
  \item $\mtyp_1 = \mtyp_2 = \tightN$ so that $\mtyp_1 + \mtyp_2 = \tightN$,
    hence $\numarrows{\mtyp_1 + \mtyp_2} = 0$. Taking 
    $\tctx_1 = \var : \tightN$,
    $\tctx_2 = \var : \tightN$,
    $m_1 = 0$,
    $m_2 = 0$,
    $e_1 = 0$,
    $e_2 = 0$,
    it's easy to check that the three conditions hold.
  \item $\mtyp_1 = \nityp_1 $ and $\mtyp_2 = \nityp_2$ so that $\mtyp_1 + \mtyp_2 = \nityp_1 \niunion \nityp_2$.
    We can write $\numarrows{\nityp_1 + \nityp_2}$ as 
    $\numarrows{\nityp_1} + \numarrows{\nityp_2}$. Taking 
    $\tctx_1 = \var : \nityp_1$,
    $\tctx_2 = \var : \nityp_2$,
    $m_1 = 0$,
    $m_2 = 0$,
    $e_1 = \numarrows{\nityp_1}$,
    $e_2 = \numarrows{\nityp_2}$,
    it's easy to check the that three conditions hold.
  \end{enumerate}
\item $\ruleTypAbs$.
  Then
  \[
    \deriv \eqdef \left(
    \indrule{\ruleTypAbs}{
      (\judg[m_i, e_i]
        {\tctx_i; \var : \optmtypthree_i}
        {\tm}
        {\mtyptwo_i}
      )_\iI
    }{
      \judg[+_\iI m_i, +_\iI e_i]
        {+_\iI \tctx_i}
        {\lam{\var}{\tm}}
        {\mset{\optmtypthree_i \to \mtyptwo_i}_\iI}
    }
    \right)
  \] 
  with
  $\mset{\optmtypthree_i \to \mtyptwo_i}_\iI = \mtyp_1 + \mtyp_2$.
  Since the sum of types is well defined, there are two subcases:
  \begin{enumerate}
  \item $\mtyp_1 = \mtyp_2 = \tightN$ so that $\mtyp_1 + \mtyp_2 = \tightN$.
    This case is not possible because the term has the non-idempotent type
    $\mset{\optmtypthree_i \to \mtyptwo_i}_\iI$.
  \item $\mtyp_1 = \nityp_1$ and $\mtyp_2 = \nityp_2$ so that $\mtyp_1 + \mtyp_2 = \nityp_1 \niunion \nityp_2$.
    Then we can write
    $I$ as $I = J \niunion K$,
    in such a way that
    $\mtyp_1 = \nityp_1 = \mset{\optmtypthree_j \to \mtyptwo_j}_\jJ$
    and
    $\mtyp_2 = \nityp_2 = \mset{\optmtypthree_k \to \mtyptwo_k}_\kK$.
    In particular,
    $\mset{\optmtypthree_i \to \mtyptwo_i}_\iI =
    \mset{\optmtypthree_j \to \mtyptwo_j}_\jJ + 
     \mset{\optmtypthree_k \to \mtyptwo_k}_\kK$.
    Taking 
    $\tctx_1 = +_{\jJ} \tctx_j$,
    $\tctx_2 = +_{\kK} \tctx_k$,
    $m_1 = +_\jJ m_j$,
    $m_2 = +_\kK m_k$,
    $e_1 = +_\jJ e_j$,
    $e_2 = +_\kK e_k$,
    we can check that the three conditions hold:
    \begin{enumerate}
    \item[(a)] 
      $\judg[+_\jJ m_j, +_\jJ e_j]
        {+_\jJ \tctx_j}
        {\lam{\var}{\tm}}
        {\mset{\optmtypthree_j \to \mtyptwo_j}_\jJ}$ 
      by \ruleTypAbs~rule, since 
      $(\judg[m_j, e_j]{\tctx_j; \var : \optmtyp_j}{\tm}{\mtyptwo_j})_\jJ$
      holds, as they are $j$ premises of $\deriv$
    \item[(b)]
      $\judg[+_\kK m_k, +_\kK e_k]
        {+_\kK \tctx_k}
        {\lam{\var}{\tm}}
        {\mset{\optmtypthree_k \to \mtyptwo_k}_\kK}$ 
      by \ruleTypAbs~rule, since 
      $(\judg[m_k, e_k]{\tctx_k; \var : \optmtyp_k}{\tm}{\mtyptwo_k})_\kK$
      holds, as they are $k$ premises of $\deriv$
    \item[(c)] 
      $\tctx = +_\iI \tctx_i = +_\jJ \tctx_j +_\kK \tctx_k = \tctx_1 + \tctx_2$
          and $m = m_\iI = +_\jJ m_j +_\kK m_k = m_1 + m_2$
          and $e = e_\iI = +_\jJ e_j +_\kK e_k = e_1 + e_2$
    \end{enumerate}
  \end{enumerate}
\end{itemize}
Now we can show that $(2 \implies 1)$. 
First, note that $\val$ is either a variable or an abstraction, 
so both judgements are derived using rule $\ruleTypVar$ or rule \ruleTypAbs.
\begin{enumerate}
\item \ruleTypVar.
  Then $\judg[0,e_1]{\var : \mtyp_1}{\var}{\mtyp_1}$, with
  $e_1 = \numarrows{\mtyp_1}$ and we also
  have $\judg[0,e_2]{\var : \mtyp_2}{\var}{\mtyp_2}$, with
  $e_2 = \numarrows{\mtyp_2}$.
  Using the third condition, we can conclude
  $\judg[0, e_1 + e_2]{\var : \mtyp_1 + \mtyp_2}{\var}{\mtyp_1 + \mtyp_2}$
  by applying rule $\ruleTypVar$, since 
  $e_1 + e_2 = \numarrows{\mtyp_1} + \numarrows{\mtyp_2} = \numarrows{\mtyp_1 + \mtyp_2}$.
\item \ruleTypAbs. The following conditions hold:
  \begin{enumerate}
  \item[(a)] 
    \[
      \deriv_1 \eqdef \left(
      \indrule{\ruleTypAbs}{
        (\judg[m_j, e_j]
          {\tctx_j; \var : \optmtypthree_j}
          {\tm}
          {\mtyptwo_j}
        )_\jJ
      }{
        \judg[+_\jJ m_j, +_\jJ e_j]
          {+_\jJ \tctx_j}
          {\lam{\var}{\tm}}
          {\mset{\optmtypthree_j \to \mtyptwo_j}_\jJ}
      }
      \right)
    \]
    with $\mtyp_1 = \mset{\optmtypthree_j \to \mtyptwo_j}_\jJ$, 
    $m_1 = +_\jJ m_j$, $e_1 = +_\jJ e_j$
  \item[(b)]
    \[
      \deriv_2 \eqdef \left(
      \indrule{\ruleTypAbs}{
        (\judg[m_k, e_k]
          {\tctx_k; \var : \optmtypthree_k}
          {\tm}
          {\mtyptwo_k}
        )_\kK
      }{
        \judg[+_\kK m_k, +_\kK e_k]
          {+_\kK \tctx_k}
          {\lam{\var}{\tm}}
          {\mset{\optmtypthree_k \to \mtyptwo_k}_\kK}
      }
      \right)
    \]
    with $\mtyp_2 = \mset{\optmtypthree_k \to \mtyptwo_k}_\kK$, 
    $m_2 = +_\kK m_k$, $e_2 = +_\kK e_k$.
  \item[(c)]  $\tctx = \tctx_1 + \tctx_2 = +_\jJ \tctx_j +_\kK \tctx_k$ 
    and $m = m_1 + m_2 = +_\jJ m_j +_\kK m_k$ 
    and $e = e_1 + e_2 = +_\jJ e_j +_\kK e_k$.
  \end{enumerate}
  We can conclude $\judg[m, e]{\tctx}{\lam{\var}{\tm}}{\mtyp_1 + \mtyp_2}$
  by applying rule $\ruleTypAbs$, using the derivations $\deriv_1$ and $\deriv_2$ 
  as premises.
\end{enumerate}
\end{proof}

\begin{definition}[Typing of substitution contexts]
We extend the type system for typing substitution contexts.
The typing judgements for substitution contexts are written 
$\judgSctx[\cm, \ce]\tctx\sctx\tctxtwo$, and the new typing rules are:
\[
  \indrule{\ruleTypSctxEmpty}{
    \emptyPremise
  }{
    \judgSctx[0,0]\emptyctx\ctxhole\emptyctx
  }
  \HS
  \indrule{\ruleTypSctxAdd}{
    \judgSctx[\cm_1, \ce_1]{\tctx_1; \var : \optmtyp_1}\sctx\tctxtwo
    \sep
    \optmtyp_1 + \optmtyp_2 \mleq \mtyp
    \sep
    \judg[\cm_2, \ce_2]{\tctx_2}\tm\mtyp
  }{
    \judgSctx[\cm_1 + \cm_2, \ce_1 + \ce_2]
      {\tctx_1 + \tctx_2}{\sctx\esub\var\tm}{\tctxtwo; \var : \optmtyp_2}
  }
\]
\end{definition}

\begin{example}
Let $\tm = \var\esub\var\vartwo$ with 
$\judg[0,\numarrows\mtyp]{\vartwo : \mtyp}\vartwo\mtyp$, we type 
$\esub\var\vartwo$ as follows:
\[
  \indrule{\ruleTypSctxAdd}{
    \indrule{\ruleTypSctxEmpty}{
      \emptyPremise
    }{
      \judgSctx[0, 0]{\var : \none}\ctxhole\emptyctx
    }
    \sep
    \none + \mtyp \mleq \mtyp
    \sep
    \indrule{\ruleTypVar}{
      \emptyPremise
    }{
      \judg[0, \numarrows\mtyp]{\vartwo : \mtyp}\vartwo\mtyp
    }
  }{
    \judgSctx[0, \numarrows\mtyp]
      {\vartwo : \mtyp}{\esub\var\vartwo}{\var : \mtyp}
  }
\]
\end{example}

\begin{example}
Let $\tm = \varthree\esub\var\vartwo$ with 
$\judg[0,0]{\vartwo : \tight}\vartwo\tight$, now we type
$\esub\var\vartwo$ as follows:
\[
  \indrule{\ruleTypSctxAdd}{
    \indrule{\ruleTypSctxEmpty}{
      \emptyPremise
    }{
      \judgSctx[0,0]{\var : \none}\ctxhole\emptyctx
    }
    \sep
    \none + \none \mleq \tight
    \sep
    \indrule{\ruleTypVar}{
      \emptyPremise
    }{
      \judg[0,0]{\vartwo : \tight}\vartwo\tight
    }
  }{
    \judgSctx[0,0]{\vartwo : \tight}{\esub\var\vartwo}{\var : \none}
  }
\]
\end{example}

We prove this extension of the type system still preserves the relevance property:
\begin{lemma}[Relevance for typing substitution contexts]
\label{lem:relevance_sctx}
If $\judgSctx[\cm,\ce]\tctx\sctx\tctxtwo$, 
then $\dom\tctx \subseteq \fv\sctx$
and $\dom\tctxtwo \subseteq \domSctx\sctx$.
\end{lemma}
% Label: lem:relevance_sctx

\begin{proof}
By induction on the derivation of the judgement 
$\judgSctx[\cm,\ce]\tctx\sctx\tctxtwo$.
\begin{enumerate}
\item \ruleTypSctxEmpty.
  The judgement is $\judgSctx[0,0]\emptyctx\ctxhole\emptyctx$, where 
  $\tctx = \tctxtwo = \emptyctx$ and $\cm = \ce = 0$ and $\sctx = \ctxhole$.
  Then, $\dom\emptyctx = \emptyset = \fv\ctxhole$
  and $\dom\emptyctx = \emptyset = \domSctx\ctxhole$, so we are done.
\item \ruleTypSctxAdd.
  Then the judgement is of the form
  \[
    \indrule{\ruleTypSctxAdd}{
      \judgSctx[\cm_1, \ce_1]{\tctx_1; \var : \optmtyp_1}{\sctx'}{\tctxtwo'}
      \sep
      \optmtyp_1 + \optmtyp_2 \mleq \mtyp
      \sep
      \judg[\cm_2, \ce_2]{\tctx_2}\tm\mtyp
    }{
      \judgSctx[\cm_1 + \cm_2, \ce_1 + \ce_2]
        {\tctx_1 + \tctx_2}{\sctx'\esub\var\tm}{\tctxtwo'; \var : \optmtyp_2}
    }
  \]
  where $\tctx = \tctx_1 + \tctx_2$ and $\tctxtwo = \tctxtwo'; \var : \optmtyp_2$ 
  and $m = m_1 + m_2$ and $e = e_1 + e_2$ and $\sctx = \sctx'\esub{\var}{\tm}$.
  
  On the one hand, we have 
  $\dom{\tctx_1; \var : \optmtyp_1} \subseteq \fv{\sctx'}$
  by \ih on $\sctx'$.
  Removing $\var$ from the inequation we obtain
  $\dom{\tctx_1; \var : \optmtyp_1} \setminus \set\var = \dom{\tctx_1}
   \subseteq \fv{\sctx'} \setminus \set\var$, so that:
  \[
    \begin{array}{rcll}
        \dom{\tctx_1} \cup \dom{\tctx_2}
      & \subseteq
      & (\fv{\sctx'} \setminus \set\var) \cup \dom{\tctx_2}
    \\
      & \subseteq
      & (\fv{\sctx'} \setminus \set\var) \cup \fv\tm
      & \text{(By \cref{lem:relevance})}
    \\
      & =
      & \fv{\sctx'\esub\var\tm}
    \end{array}
  \]
  On the other hand,
  \[
    \begin{array}{rcll}
        \dom{\tctxtwo'; \var : \optmtyp_2}
      & \subseteq
      & \dom{\tctxtwo'} \cup \set{\var}
    \\
      & \subseteq
      & \domSctx{\sctx'} \cup \set{\var}
      & \text{(By \ih on $\sctx'$)}
    \\
      & =
      & \domSctx{\sctx'\esub{\var}{\tm}}
    \end{array}
  \]
  so we are done.
\end{enumerate}
\end{proof}

\begin{lemma}[Composition/decomposition]
\label{lem:composition}
The following are equivalent:
\begin{enumerate}
\item
  $\judg[\cm,\ce]\tctx{\tm\sctx}\mtyp$
\item
  There exist $\tctx_\tm$, $\tctx_\sctx$, $\tctxtwo$, $\cm_\tm$, 
  $\ce_\tm$, $\cm_\sctx$, and $\ce_\sctx$ such that:
  \begin{itemize}
  \item[(a)] \label{it:comp_cond_a}
    $\judgSctx[\cm_\sctx,\ce_\sctx]{\tctx_\sctx}\sctx\tctxtwo$
  \item[(b)] \label{it:comp_cond_b}
    $\judg[\cm_\tm,\ce_\tm]{\tctx_\tm;\tctxtwo}\tm\mtyp$
  \item[(c)] \label{it:comp_cond_c}
    $\tctx = \tctx_\tm + \tctx_\sctx$ and
    $\cm = \cm_\tm + \cm_\sctx$ and
    $\ce = \ce_\tm + \ce_\sctx$.
  \end{itemize}
\end{enumerate}

Furthermore, in the $(1 \implies 2)$ direction,
if $\inv\aset\sset{\tm\sctx}$ holds, 
then $\isAppr\aset\tctx$
implies $\isAppr{\expansion\aset\sctx}\tctxtwo$.
\end{lemma}
% Label: lem:composition

\begin{proof}
Both directions of the proof are by induction on $\sctx$.

$(1 \implies 2)$
\begin{itemize}
\item $\sctx = \ctxhole$.
  The judgement is of the form $\judg[m,e]\tctx\tm\mtyp$. Taking 
  $\tctx_\tm \eqdef \tctx$, $\tctx_\sctx \eqdef \emptyctx$, 
  $\tctxtwo \eqdef \emptyctx$, $\cm_\tm \eqdef \cm$, 
  $\ce_\tm \eqdef \ce$, $\cm_\sctx \eqdef 0$, and 
  $\ce_\sctx \eqdef 0$ we obtain:
  \begin{enumerate}
  \item[(a)] 
    $\judgSctx[0,0]\emptyctx\ctxhole\emptyctx$, by rule \ruleTypSctxEmpty
  \item[(b)] 
    $\judg[\cm,\ce]\tctx\tm\mtyp$, by hypothesis
  \item[(c)] 
    $\tctx = \tctx + \emptyctx$ and $\cm = \cm + 0$ and $\ce = \ce + 0$
  \end{enumerate}
  Furthermore, it is immediate to conclude that $\isAppr\aset\emptyctx$ holds, so we are done.
\item $\sctx = \sctx'\esub\var\tmtwo$.
  The judgement is of the form $\judg[\cm,\ce]\tctx{\tm\sctx'\esub\var\tmtwo}\mtyp$, 
  so it can only be derived from rule \ruleTypES:
  \[
    \indrule{\ruleTypES}{
      \judg[\cm_1, \ce_1]{\tctx_1; \var : \optmtyptwo}{\tm\sctx'}\mtyp
      \sep
      \optmtyptwo \mleq \mtyptwo
      \sep
      \judg[\cm_2, \ce_2]{\tctx_2}\tmtwo\mtyptwo
    }{
      \judg[\cm_1 + \cm_2, \ce_1 + \ce_2]
        {\tctx_1 + \tctx_2}
        {\tm\sctx'\esub\var\tmtwo}
        \mtyp
    }
  \]
  where $\tctx = \tctx_1 + \tctx_2$, $\cm = \cm_1 + \cm_2$, and 
  $\ce = \ce_1 + \ce_2$.
  
  By \ih on $\sctx'$, there exist $\tctx'_\tm$, $\tctx'_{\sctx'}$, 
  $\tctxtwo'$, $\cm_\tm$, $\ce_\tm$, $\cm_{\sctx'}$, and $\ce_{\sctx'}$ 
  such that:
  \begin{enumerate}
  \item[(a')]\label{cond:hi_comp_1}
    $\judgSctx[\cm_{\sctx'}, \ce_{\sctx'}]
      {\tctx'_{\sctx'}}{\sctx'}{\tctxtwo'}$
  \item[(b')]\label{cond:hi_comp_2}
    $\judg[\cm_\tm,\ce_\tm]
      {\tctx'_\tm; \tctxtwo'}\tm\mtyp$
  \item[(c')]\label{cond:hi_comp_3}
    $\tctx_1; \var : \optmtyptwo = \tctx'_{\sctx'} + \tctx'_\tm$ and 
    $\cm_1 = \cm_\tm + \cm_{\sctx'}$ and 
    $\ce_1 = \ce_\tm + \ce_{\sctx'}$.
  \end{enumerate}
  Furthermore, 
  $\inv\aset\sset{\tm\sctx'\esub\var\tmtwo}$ implies
  $\inv{\tilde\aset}{\tilde\sset}{\tm\sctx'}$, 
  with $\tilde\aset = \aset \cup \set\var$ and $\tilde\sset = \sset$
  or $\tilde\aset = \aset$ and $\tilde\sset = \sset \cup \set\var$,
  so that in either case, 
  $\isAppr{\tilde\aset}{\tctx_1; \var : \optmtyptwo}$ implies 
  $\isAppr{\expansion{\tilde\aset}{\sctx'}}{\tctxtwo'}$.
  
  By statement (c'), we can write
  $\optmtyptwo$ as $\optmtyptwo_{\sctx'} + \optmtyptwo_\tm$, so that
  $\tctx'_{\sctx'} = \tctx_{\sctx'}; \var : \optmtyptwo_{\sctx'}$ and
  $\tctx'_\tm = \tctx_\tm; \var : \optmtyptwo_\tm$ and
  $\optmtyptwo_{\sctx'} + \optmtyptwo_\tm \mleq \mtyptwo$.
  We then have  $\tctx_1 = \tctx_{\sctx'} + \tctx_\tm$.
  Taking 
  $\tctx_\sctx \eqdef \tctx_{\sctx'} + \tctx_2$, 
  $\tctx_\tm$, 
  $\tctxtwo \eqdef \tctxtwo'; \var : \optmtyptwo_\tm$, 
  $\cm_\tm$,
  $\ce_\tm$, 
  $\cm_\sctx \eqdef \cm_{\sctx'} + \cm_2$, and 
  $\ce_\sctx = \ce_{\sctx'} + \ce_2$ we have:
  \begin{enumerate}
  \item[(a)]
    \[
      \indrule{\ruleTypSctxAdd}{
        \judgSctx[\cm_{\sctx'}, \ce_{\sctx'}]
          {\tctx_{\sctx'}; \var : \optmtyptwo_{\sctx'}}
          {\sctx'}
          {\tctxtwo'}
        \sep
        \optmtyptwo_{\sctx'} + \optmtyptwo_\tm \mleq \mtyptwo
        \sep
        \judg[\cm_2, \ce_2]{\tctx_2}\tmtwo\mtyptwo
      }{
        \judgSctx[\cm_{\sctx'} + \cm_2, \ce_{\sctx'} + \ce_2]
          {\tctx_{\sctx'} + \tctx_2}
          {\sctx'\esub\var\tmtwo}
          {\tctxtwo'; \var : \optmtyptwo_\tm}
      }
    \]
  \item[(b)]
    $\judg[\cm_\tm,\ce_\tm]
      {\tctx_\tm;\tctxtwo';\var:\optmtyptwo_\tm}\tm\mtyp$, 
    by condition (b')
  \item[(c)]
    $\tctx 
    = \tctx_1 + \tctx_2 
    = \tctx_\tm + \tctx_{\sctx'} + \tctx_2 
    = \tctx_\tm + \tctx_\sctx$ 
    and
    $\cm 
    = \cm_1 + \cm_2 
    = \cm_\tm + \cm_{\sctx'} + \cm_2 
    = \cm_\tm + \cm_\sctx$
    and
    $\ce 
    = \ce_1 + \ce_2 
    = \ce_\tm + \ce_{\sctx'} + \ce_2 
    = \ce_\tm + \ce_\sctx$
  \end{enumerate}
  \sloppy
  If $\inv\aset\sset{\tm\sctx'\esub\var\tmtwo}$ holds, 
  we want to show that $\isAppr\aset\tctx$ implies
  $\isAppr{\expansion\aset{\sctx'\esub\var\tmtwo}}\tctxtwo$. 
  We already know that $\isAppr{\expansion\aset{\sctx'}}{\tctxtwo'}$ 
  holds by the \ih on $\sctx'$.
  We consider two cases:
  \begin{itemize}
  \item 
    $\expansion\aset{\sctx'\esub\var\tmtwo} = \expansion\aset{\sctx'} \cup \set\var$.
    This means $\tmtwo \in \HAbs\aset$. 
    Moreover, $\isAppr\aset{\tctx_2}$ holds, as 
    $\isAppr\aset\tctx$ and $\tctx = \tctx_1 + \tctx_2$.
    Then $\mtyptwo \neq \tightN$ by \cref{lem:types_of_hereditary_abstractions}. 
    By definition 
    $\optmtyptwo_{\sctx'} + \optmtyptwo_\tm \mleq \mtyptwo$ implies 
    $\optmtyptwo_{\sctx'} \neq \tightN$ and $\optmtyptwo_\tm \neq \tightN$,
     so we can conclude that
    $\isAppr{\expansion\aset{\sctx'\esub\var\tmtwo}}\tctxtwo$ holds.
  \item
    $\expansion\aset{\sctx'\esub\var\tmtwo} = \expansion\aset{\sctx'}$.
    Immediate by the \ih on $\sctx'$.
  \end{itemize}    
\end{itemize}
$(2 \implies 1)$
\begin{itemize}
\item $\sctx = \ctxhole$.
  Then, there exist 
  $\tctx_\tm$, $\tctx_\ctxhole$, $\tctxtwo$, $\cm_\tm$, $\ce_\tm$, 
  $\cm_\ctxhole$, and $\ce_\ctxhole$ such that:
  \begin{enumerate}
  \item[(a)] 
    $\judgSctx[\cm_\ctxhole, \ce_\ctxhole]
      {\tctx_\ctxhole}\ctxhole\tctxtwo$
  \item[(b)] 
    $\judg[\cm_\tm,\ce_\tm]{\tctx_\tm;\tctxtwo}\tm\mtyp$
  \item[(c)]
    $\tctx = \tctx_\tm + \tctx_\ctxhole$ and
    $\cm = \cm_\tm + \cm_\ctxhole$ and
    $\ce = \ce_\tm + \ce_\ctxhole$
  \end{enumerate}
  The judgement from condition (a) can only be derived by rule 
  \ruleTypSctxEmpty, hence $\tctx_\ctxhole = \tctxtwo = \emptyctx$ 
  and $\cm_\ctxhole = \ce_\ctxhole = 0$.
  By conditions (b) and (c) it is immediate to conclude 
  $\judg[\cm, \ce]\tctx\tm\mtyp$.
\item $\sctx = \sctx'\esub\var\tmtwo$.
  Then, there exist 
  $\tctx_\tm$, $\tctx_\sctx$, $\tctxtwo$, $\cm_\tm$, $\ce_\tm$, 
  $\cm_\sctx$, and $\ce_\sctx$ such that:
  \begin{enumerate}
  \item[(a)] 
    $\judgSctx[\cm_\sctx, \ce_\sctx]
      {\tctx_\sctx}{\sctx'\esub\var\tmtwo}\tctxtwo$
  \item[(b)] 
    $\judg[\cm_\tm, \ce_\tm]{\tctx_\tm; \tctxtwo}\tm\mtyp$
  \item[(c)]
    $\tctx = \tctx_\sctx + \tctx_\tm$ and
    $\cm = \cm_\sctx + \cm_\tm$ and
    $\ce = \ce_\sctx + \ce_\tm$
  \end{enumerate}
  The judgement from condition (a) can only be derived by rule 
  \ruleTypSctxAdd:
  \[
    \indrule{\ruleTypSctxAdd}{
      \judgSctx[m_{\sctx'}, e_{\sctx'}]
        {\tctx_{\sctx'}; \var : \optmtyptwo_1}
        {\sctx'}
        {\tctxtwo'}
      \sep
      \optmtyptwo_1 + \optmtyptwo_2 \mleq \mtyptwo
      \sep
      \judg[m_\tmtwo, e_\tmtwo]{\tctx_\tmtwo}\tmtwo\mtyptwo
    }{
      \judgSctx[\cm_{\sctx'} + \cm_\tmtwo, \ce_{\sctx'} + \cm_\tmtwo]
        {\tctx_{\sctx'} + \tctx_\tmtwo}
        {\sctx'\esub\var\tmtwo}
        {\tctxtwo'; \var : \optmtyptwo_2}
    }
  \]
  where $\tctx_\sctx = \tctx_{\sctx'} + \tctx_\tmtwo$,
  $\tctxtwo = \tctxtwo'; \var : \optmtyptwo_2$,
  $\cm_\sctx = \cm_{\sctx'} + \cm_2$, and 
  $\ce_\sctx = \ce_{\sctx'} + e_2$.

  Then, there exist 
  $\tctxthree_\tm = \tctx_\tm; \var : \optmtyptwo_2$, 
  $\tctxthree_{\sctx'} = \tctx_{\sctx'}; \var : \optmtyptwo_1$, 
  $\tctxtwo'$, $\cm_\tm$, $\ce_\tm$, $\cm_{\sctx'}$, and $\ce_{\sctx'}$
  such that:
  \begin{enumerate}
  \item[(a')] 
    $\judgSctx[\cm_{\sctx'}, \ce_{\sctx'}]{\tctxthree_{\sctx'}}{\sctx'}{\tctxtwo'}$
  \item[(b')] 
    $\judg[\cm_\tm, \ce_\tm]{\tctxthree_\tm; \tctxtwo'}\tm\mtyp$
  \item[(c')]
    $\tctxthree = \tctxthree_{\sctx'} + \tctxthree_\tm$ and
    $\cm' = \cm_{\sctx'} + \cm_\tm$ and
    $\ce' = \ce_{\sctx'} + \ce_\tm$.
  \end{enumerate}
  We can apply the \ih on $\sctx'$, yielding 
  $\judg[\cm', \ce']\tctxthree{\tm\sctx'}\mtyp$.
  Applying rule \ruleTypES, we conclude:
  \[
    \indrule{\ruleTypES}{
      \judg[\cm_{\sctx'} + \cm_\tm, \ce_{\sctx'} + \ce_\tm]
        {\tctx_{\sctx'} + \tctx_\tm; \var : (\optmtyptwo_1 + \optmtyptwo_2)}
        {\tm\sctx'}
        \mtyp
      \sep
      \optmtyptwo_1 + \optmtyptwo_2 \mleq \mtyptwo
      \sep
      \judg[m_\tmtwo, e_\tmtwo]{\tctx_\tmtwo}{\tmtwo}{\mtyptwo}        
    }{
      \judg[m, e]{\tctx}{\tm\sctx'\esub{\var}{\tmtwo}}{\mtyp}
    }
  \]
\end{itemize}
\end{proof}

\begin{definition}
Let $X$ be a finite set of variables. 
We say that a context $\tctx$ is \defn{$X-$tight} if for all 
$\var \in X$, $\tctx(\var) = \none$ or $\tctx(\var) = \tight$.
\end{definition}

\begin{lemma}
\label{lem:tight_spreading_struct}
Let $\inv\aset\sset\tm$. 
If $\judg[\cm, \ce]\tctx\tm\mtyp$,
where $\tctx$ is $\sset$-tight and $\tm \in \Struct\sset$.
Then $\mtyp$ must be tight.
\end{lemma}
% Label: lem:tight_spreading_struct

\begin{proof}
By induction on the derivation of $\tm \in \Struct\sset$.
\begin{itemize}
\item \ruleStructVar.
  Then $\tm = \var$ and $\var \in \sset$ by premise of the rule 
  \ruleStructVar. 
  The judgement $\judg[\cm,\ce]\tctx\var\mtyp$ can only be derived 
  by rule \ruleTypVar, so $\tctx$ is of the form $\var : \mtyp$.
  The typing context $\var : \mtyp$ is $\sset-$tight by hypothesis,
  and we already know that $\var \in \sset$, hence $\tctx(\var)$ must 
  be tight, so we are done.
\item \ruleStructApp.
  Then $\tm = \tmtwo \, \tmthree$ and
  $\tmtwo \in \Struct\sset$ by premise of the rule \ruleStructApp.
  The judgement $\judg[\cm,\ce]\tctx{\tmtwo \, \tmthree}\mtyp$ can be 
  derived either by rule \ruleTypAppP or rule \ruleTypAppC.
  The former case is immediate, as the conclusion always type with a
  tight type. Hence, we look at the second case, where the typing 
  derivation is of the form:
  \[
    \indrule{\ruleTypAppC}{
      \judg[\cm_1,\ce_1]{\tctx_1}\tmtwo{\mset{\optmtyptwo \to \mtyp}}
      \sep
      \optmtyptwo \mleq \mtyptwo
      \sep
      \judg[\cm_2,\ce_2]{\tctx_2}\tmthree\mtyptwo
    }{
      \judg[1 + \cm_1 + \cm_2, \ce_1 + \ce_2]
        {\tctx_1 + \tctx_2}{\tmtwo \, \tmthree}\mtyp
    }
  \]
  where $\tctx = \tctx_1 + \tctx_2$, $\cm = 1 + \cm_1 + \cm_2$, and 
  $e = e_1 + e_2$.
  
  Moreover, $\inv\aset\sset{\tmtwo \, \tmthree}$ implies 
  $\inv\aset\sset\tmtwo$, and $\tctx_1$ is $\sset-$tight since 
  $\tctx$ is $\sset-$tight and $\tctx = \tctx_1 + \tctx_2$.
  Then we have that $\mset{\optmtyptwo \to \mtyp}$ must be tight by
  the \ih on $\tmtwo$, which is impossible and yields to a contradiction.
  Hence this case is not possible.
\item \ruleStructSubi.
  Then $\tm = \tmtwo\esub\var\tmthree$, and
  \[
    \indrule{\ruleStructSubi}{
      \tmtwo \in \Struct\sset
      \sep
      \var \notin \sset
    }{
      \tmtwo\esub\var\tmthree \in \Struct\sset
    }
  \]
  The judgement $\judg[\cm,\ce]\tctx{\tmtwo\esub\var\tmthree}\mtyp$ 
  can only be derived by rule \ruleTypES, so the typing derivation is
  of the form:
  \[
    \indrule{\ruleTypES}{
      \judg[\cm_1,\ce_1]{\tctx_1; \var : \optmtyptwo}\tmtwo\mtyp
      \sep
      \optmtyptwo \mleq \mtyptwo
      \sep
      \judg[\cm_2,\ce_2]{\tctx_2}\tmthree\mtyptwo
    }{
      \judg[\cm_1 + \cm_2, \ce_1 + \ce_2]
        {\tctx_1 + \tctx_2}{\tmtwo\esub\var\tmthree}\mtyp
    }
  \]
  where $\tctx = \tctx_1 + \tctx_2$, $\cm = \cm_1 + \cm_2$, and 
  $\ce = \ce_1 + \ce_2$.
  
  Moreover, $\inv\aset\sset{\tmtwo\esub\var\tmthree}$ implies 
  $\inv{\aset \cup \set\var}\sset\tmtwo$,
  and $\tctx_1; \var : \optmtyptwo$ is $\sset-$tight since 
  (1) $\tctx$ is $\sset-$tight and $\tctx = \tctx_1 + \tctx_2$, and 
  (2) $\var \not \in \sset$ by premise of rule \ruleStructSubi. 
  Furthermore, $\tmtwo \in \Struct\sset$ holds
  by premise of rule \ruleStructSubi.
  We can then apply the \ih on $\tmtwo$, concluding that $\mtyp$ is 
  tight, so we are done.
\item \ruleStructSubii.
  Then $\tm = \tmtwo\esub\var\tmthree$, and
  \[
    \indrule{\ruleStructSubii}{
      \tmtwo \in \Struct{\sset \cup \set\var}
      \sep
      \var \notin \sset
      \sep
      \tmthree \in \Struct\sset
    }{
      \tmtwo\esub\var\tmthree \in \Struct\sset
    }
  \]
  The judgement $\judg[\cm,\ce]\tctx{\tmtwo\esub\var\tmthree}\mtyp$ 
  can only be derived by rule \ruleTypES, so the typing derivation is 
  of the form:
  \[
    \indrule{\ruleTypES}{
      \judg[\cm_1,\ce_1]{\tctx_1; \var : \optmtyptwo}\tmtwo\mtyp
      \sep
      \optmtyptwo \mleq \mtyptwo
      \sep
      \judg[\cm_2,\ce_2]{\tctx_2}\tmthree\mtyptwo
    }{
      \judg[\cm_1 + \cm_2, \ce_1 + \ce_2]
        {\tctx_1 + \tctx_2}{\tmtwo\esub\var\tmthree}\mtyp
    }
  \]
  where $\tctx = \tctx_1 + \tctx_2$, $\cm = \cm_1 + \cm_2$, and 
  $\ce = \ce_1 + \ce_2$.
  
  By Barendregt's convention, we may assume $\var \notin \aset \cup \sset$.
  Moreover, $\inv\aset\sset{\tmtwo\esub\var\tmthree}$ implies 
  $\inv\aset\sset\tmthree$.
  Furthermore, $\tctx_2$ is $\sset-$tight, since $\tctx$ is 
  $\sset-$tight and $\tctx = \tctx_1 + \tctx_2$. 
  By premise of the rule \ruleStructSubii,
  $\tmthree \in \Struct\sset$ holds.
  We can then apply the \ih on $\tmthree$, yielding that $\mtyptwo$ is tight.

  On the other hand, $\inv\aset\sset{\tmtwo\esub\var\tmthree}$ 
  implies $\inv\aset{\sset \cup \set\var}\tmtwo$, and 
  $\tctx_1; \var : \optmtyptwo$ is $(\sset \cup \set\var)-$tight since 
  (1) $\mtyptwo$ is tight and $\optmtyptwo \mleq \mtyptwo$ by premise of rule \ruleTypES,
  so $\optmtyptwo$ must be either $\none$ or $\tight$, and
  (2) $\tctx$ is $(\sset \cup \set{\var})-$tight and $\tctx = \tctx_1 + \tctx_2$.
  By premise of rule \ruleStructSubii, 
  $\tmtwo \in \Struct{\sset \cup \set\var}$ holds.
  Applying the \ih on $\tmtwo$, we conclude that $\mtyp$ is tight,
  so we are done.
\end{itemize}
\end{proof}

\begin{lemma}
\label{lem:reducible_when_multisets}
Let $\inv\aset\sset\tm$, and suppose that the following hold:
\begin{enumerate}
\item 
  $\judg[\cm,\ce]{\tctx;\tctxtwo}\tm\tight$
\item 
  $\isAppr\aset\tctx$, with $\tctx$ tight
\item 
  $\tctxtwo \neq \emptyctx$
\item 
  $\dom\tctxtwo \subseteq \aset$
\item 
  For all $\var \in \dom\tctxtwo$, $\tctxtwo(\var)$ is a non-empty multiset.
\end{enumerate}
Then, $\tm \notin \NF\aset\sset\appflag$.
\end{lemma}
% \hiddenproof{
%   By induction on the derivation of the judgement
%   $\judg[\cm,\ce]{\tctx;\tctxtwo}\tm\tight$.
% }{
  % Label: lem:reducible_when_multisets

\begin{proof}
By induction on the derivation of the judgement 
$\judg[\cm,\ce]{\tctx;\tctxtwo}\tm\tight$.
\begin{itemize}
\item \ruleTypVar.
  Then, $\tm = \var$ and the typing judgement is of the form
  $\judg[0,0]{\var : \tight}\var\tight$, where 
  $\tctx;\tctxtwo = \var : \tight$, $m = 0$, and $e = 0$.
  Moreover, $\tctxtwo \neq \emptyctx$ by hypothesis, so 
  $\tctxtwo = \var : \tight$.
  But we reach a contradiction since $\tctxtwo(\var)$ is a non-empty multiset
  by condition 5. of the hypothesis.
  Thus, this case does not apply.
\item \ruleTypAbs.
  Then, $\tm = \lam\var\tmtwo$ and the typing judgement is of the form
  $\judg[\cm,\ce]{\tctx;\tctxtwo}{\lam\var\tmtwo}\tight$, that can only
  be derived by rule \ruleTypAbs.
  Hence, it must be the case that $\tight = \emset$ and the conclusion
  of the typing derivation holds with no premises, so we have in particular
  $\tctx;\tctxtwo = \emptyctx$.
  But we reach a contradiction since $\tctxtwo \neq \emptyctx$ by
  hypothesis.
  Thus, this case does not apply.
\item \ruleTypAppP.
  Then, $\tm = \tmtwo\,\tmthree$, and the following hold:
  \begin{enumerate}
  \item 
    \[
      \indrule{\ruleTypAppP}{
        \judg[\cm_1,\ce_1]{\tctx_1;\tctxtwo_1}\tmtwo\tightN
        \sep
        \judg[\cm_2,\ce_2]{\tctx_2;\tctxtwo_2}\tmthree\tight
      }{
        \judg[\cm_1 + \cm_2, \ce_1 + \ce_2]
          {(\tctx_1 + \tctx_2);(\tctxtwo_1 + \tctxtwo_2)}
          {\tmtwo \, \tmthree}
          \tightN
      }
    \]
    where $\tctx = \tctx_1 + \tctx_2$, $\tctxtwo = \tctxtwo_1 + \tctxtwo_2$, 
    $\cm = \cm_1 + \cm_2$, $\ce = \ce_1 + \ce_2$, and $\tight = \tightN$
  \item 
    $\isAppr\aset{\tctx_1 + \tctx_2}$, with $\tctx_1 + \tctx_2$ tight
  \item 
    $\tctxtwo_1 + \tctxtwo_2 \neq \emptyctx$
  \item 
    $\dom{\tctxtwo_1 + \tctxtwo_2} \subseteq \aset$
  \item 
    For all $\var \in \dom{\tctxtwo_1 + \tctx_2}$, 
    $(\tctxtwo_1 + \tctxtwo_2)(\var)$ is a non-empty multiset.
  \end{enumerate}
  Moreover, $\inv\aset\sset{\tmtwo\,\tmthree}$ implies $\inv\aset\sset\tmtwo$ 
  and $\inv\aset\sset\tmthree$.
  By condition 3., we have to analyze two possible cases:
  \begin{itemize}
  \item If $\tctxtwo_1 \neq \emptyctx$, we state that the following hold:
    \begin{enumerate}[label=\arabic*'.]
    \item 
      $\judg[\cm_1,\ce_1]{\tctx_1;\tctxtwo_1}\tmtwo\tightN$, by 
      premise of the typing derivation in condition 1
    \item 
      $\isAppr\aset{\tctx_1}$, with $\tctx_1$ tight, since 
      $\tctx = \tctx_1 + \tctx_2$ and condition 2
    \item 
      $\tctxtwo_1 \neq \emptyctx$, by condition 3
    \item 
      $\dom{\tctxtwo_1} \subseteq \aset$, since 
      $\tctxtwo = \tctxtwo_1 + \tctxtwo_2$ and condition 4
    \item 
      For all $\var \in \dom{\tctxtwo_1}$, $\tctxtwo_1(\var)$ is a 
      non-empty multiset, since $\tctxtwo = \tctxtwo_1 + \tctxtwo_2$ 
      and condition 5.
    \end{enumerate}
    We can apply the \ih on $\tmtwo$, yielding 
    $\tmtwo \notin \NF\aset\sset\appflag$, which in particular holds 
    when $\appflag = \app$.
    By definition, there exist $\rulename$ and $\tmtwo'$ such that 
    $\tmtwo \tov\rulename\aset\sset\app \tmtwo'$.
    We can apply rule \ruleUAppL, yielding
    $\tmtwo\,\tmthree \tov\rulename\aset\sset\appflag \tmtwo' \, \tmthree$,
    and by definition we can conclude that
    $\tmtwo\,\tmthree \notin \NF\aset\sset\appflag$.
  \item If $\tctxtwo_2 \neq \emptyctx$, we state that the following hold:
    \begin{enumerate}[label=\arabic*'.]
    \item 
      $\judg[\cm_2,\ce_2]{\tctx_2; \tctxtwo_2}\tmthree\tight$, by 
      premise of the typing derivation in condition 1
    \item 
      $\isAppr\aset{\tctx_2}$, with $\tctx_2$ tight, since 
      $\tctx = \tctx_1 + \tctx_2$ and condition 2
    \item 
      $\tctxtwo_2 \neq \emptyctx$, by condition 3
    \item 
      $\dom{\tctxtwo_2} \subseteq \aset$, since 
      $\tctxtwo = \tctxtwo_1 + \tctxtwo_2$ and condition 4
    \item 
      For all $\var \in \dom{\tctxtwo_2}$, $\tctxtwo_{2}(\var)$ is a 
      non-empty multiset, since $\tctxtwo = \tctxtwo_1 + \tctxtwo_2$ 
      and condition 5.
    \end{enumerate}
    We can apply the \ih on $\tmthree$, yielding 
    $\tmthree \notin \NF\aset\sset\appflag$, which in particular holds 
    when $\appflag = \nonapp$.
    By definition, there exist $\rulename$ and $\tmthree'$ such that 
    $\tmthree \tov\rulename\aset\sset\nonapp \tmthree'$.
    We reason by contradiction, assuming 
    $\tmtwo \, \tmthree \in \NF\aset\sset\appflag$.
    This assumption can only be derived by rule \ruleUNFApp, so we have
    $\tmtwo \in \NF\aset\sset\app$ in particular.
    By \cref{lem:nf_in_HAbs_or_Struct}, then $\tmtwo \in \Struct\sset$, 
    so we can apply rule \ruleUAppR, yielding 
    $\tmtwo \, \tmthree \tov\rulename\aset\sset\appflag \tmtwo\,\tmthree'$, 
    which contradicts the assumption.
    We conclude that $\tmtwo\,\tmthree \notin \NF\aset\sset\appflag$.
  \end{itemize}
\item \ruleTypAppC.
  Then $\tm = \tmtwo\,\tmthree$, and the following hold:
  \begin{enumerate}
  \item 
    \[
      \indrule{\ruleTypAppC}{
        \judg[\cm_1,\ce_1]
          {\tctx_1;\tctxtwo_1}\tmtwo{\mset{\optmtyp \to \tight}}
        \sep
        \optmtyp \mleq \mtyp
        \sep
        \judg[\cm_2,\ce_2]{\tctx_2;\tctxtwo_2}\tmthree\mtyp
      }{
        \judg[1 + \cm_1 + \cm_2, \ce_1 + \ce_2]
          {(\tctx_1 + \tctx_2);(\tctxtwo_1 + \tctxtwo_2)}
          {\tmtwo \, \tmthree}
          \tight
      }
    \]
    where $\tctx = \tctx_1 + \tctx_2$, $\tctxtwo = \tctxtwo_1 + \tctxtwo_2$, 
    $\cm = 1 + \cm_1 + \cm_2$, and $\ce = \ce_1 + \ce_2$
  \item 
    $\isAppr\aset{\tctx_1 + \tctx_2}$, with $\tctx_1 + \tctx_2$ tight
  \item 
    $\tctxtwo_1 + \tctxtwo_2 \neq \emptyctx$
  \item 
    $\dom{\tctxtwo_1 + \tctxtwo_2} \subseteq \aset$
  \item 
    For all $\var \in \dom{\tctxtwo_1 + \tctxtwo_2}$, 
    $(\tctxtwo_1 + \tctxtwo_2)(\var)$ is a non-empty multiset.
  \end{enumerate}
  Moreover, $\inv\aset\sset{\tmtwo\,\tmthree}$ implies 
  $\inv\aset\sset\tmtwo$ and $\inv\aset\sset\tmthree$.
  We reason by contradiction, assuming $\tmtwo\,\tmthree \in \NF\aset\sset\appflag$.
  This assumption can only be derived by rule \ruleUNFApp, so in particular 
  we have $\tmtwo \in \NF\aset\sset\app$.
  By \cref{lem:nf_in_HAbs_or_Struct}, we have $\tmtwo \in \Struct\sset$.
  Moreover, $\tctx_1; \tctxtwo_1$ is $\sset-$tight since 
  (1) $\tctx_1$ is tight by $\tctx = \tctx_1 + \tctx_2$ and condition (b), and
  (2) $\dom{\tctxtwo_1} \subseteq \aset$ by condition (d). 
  We can then apply \cref{lem:tight_spreading_struct}, yielding that $\mset{\optmtyp \to \tight}$ must be tight, 
  which gives a contradiction since it is a singleton.
  We can then conclude that $\tmtwo \, \tmthree \notin \NF{\aset}{\sset}{\appflag}$.
\item \ruleTypES.
  Then $\tm = \tmtwo\esub\var\tmthree$, and the following hold:
  \begin{enumerate}
  \item 
    \[
      \indrule{\ruleTypES}{
        \judg[\cm_1,\ce_1]
          {\tctx_1;\tctxtwo_1; \var : \optmtyp}\tmtwo\tight
        \sep
        \optmtyp \mleq \mtyp
        \sep
        \judg[\cm_2,\ce_2]{\tctx_2;\tctxtwo_2}\tmthree\mtyp
      }{
        \judg[\cm_1 + \cm_2, \ce_1 + \ce_2]
          {(\tctx_1 + \tctx_2);(\tctxtwo_1 + \tctxtwo_2)}
          {\tmtwo\esub\var\tmthree}
          \tight
      }
    \]
    where $\tctx = \tctx_1 + \tctx_2$, $\tctxtwo = \tctxtwo_1 + \tctxtwo_2$, 
    $\cm = \cm_1 + \cm_2$, and $\ce = \ce_1 + \ce_2$
  \item 
    $\isAppr\aset{\tctx_1 + \tctx_2}$, with $\tctx_1 + \tctx_2$ tight
  \item 
    $\tctxtwo_1 + \tctxtwo_2 \neq \emptyctx$
  \item 
    $\dom{\tctxtwo_1 + \tctxtwo_2} \subseteq \aset$
  \item 
    For all $\vartwo \in \dom{\tctxtwo_1 + \tctxtwo_2}$, 
    $(\tctxtwo_1 + \tctxtwo_2)(\vartwo)$ is a non-empty multiset.
  \end{enumerate}
  
  Moreover, $\inv\aset\sset{\tmtwo\esub\var\tmthree}$ implies in 
  particular $\inv\aset\sset\tmthree$.
  We reason by contradiction, assuming $\tmtwo\esub{\var}{\tmthree} \in \NF{\aset}{\sset}{\appflag}$.
  This assumption can be derived either by rule \ruleUNFEsAbs\ or by rule \ruleUNFEsStruct.
  We may assume $\var \notin (\aset \cup \sset)$ by Barendregt's convention.
  \begin{enumerate}
  \item \ruleUNFEsAbs.
    Then
    \[
      \indrule{\ruleUNFEsAbs}{
        \tmtwo \in \NF{\aset \cup \set{\var}}{\sset}{\appflag}
        \sep
        \tmthree \in \NF{\aset}{\sset}{\nonapp}
        \sep
        \tmthree \in \HAbs{\aset}
      }{
        \tmtwo\esub{\var}{\tmthree} \in \NF{\aset}{\sset}{\appflag}
      }
    \]
    By condition (c), $\tctxtwo_1 + \tctxtwo_2 \neq \emptyctx$.
    Also, $\mtyp$ can be tight or not.
    We have to analyze three possible cases:
    \begin{enumerate}
      \item $\mtyp$ not tight.
      By condition (a), then $\optmtyp \mleq \mtyp$.
      Since $\mtyp$ is not tight, then $\optmtyp = \mtyp$.
      This means that $\optmtyp \neq \none$.
      We have that the following holds:
      \begin{enumerate}
      \item[(a')]
        $\judg[m_1, e_1]{\tctx_1; (\tctxtwo_1; \var : \mtyp)}{\tmtwo}{\tight}$, by condition (a)
      \item[(b')] 
        $\isAppr{\aset \cup \set{\var}}{\tctx_1}$, with $\tctx_1$ tight, 
        since $\tctx = \tctx_1 + \tctx_2$ and condition (b)
      \item[(c')]
        $(\tctxtwo_1; \var : \mtyp) \neq \emptyctx$ since at least $(\tctxtwo_1; \var : \mtyp)(\var) \neq \none$
      \item[(d')]
        $\dom{\tctxtwo_1; \var : \mtyp} \subseteq \aset \cup \set{\var}$ 
        since $\tctxtwo = \tctxtwo_1 + \tctxtwo_2$ and  condition (d)
      \item[(e')]
        For all $\vartwo \in \dom{\tctxtwo_1; \var : \mtyp}$, we have 
        that $(\tctxtwo_1; \var : \mtyp)(\vartwo)$ is a non-empty multiset
        since $\mtyp$ is not tight, $\tctxtwo = \tctxtwo_1 + \tctxtwo_2$ and by condition (e).
      \end{enumerate}
      We can apply \ih on $\tmtwo$, yielding
      $\tmtwo \notin \NF{\aset \cup \set{\var}}{\sset}{\appflag}$,
      which leads to a contradiction with our assumption. 
      Therefore $\tmtwo\esub{\var}{\tmthree} \notin \NF{\aset}{\sset}{\appflag}$.
    \item $\mtyp = \tight$ and $\tctxtwo_2 \neq \emptyctx$.
      We have that the following holds:
      \begin{enumerate}
      \item[(a')] $\judg[m_2, e_2]{\tctx_2; \tctxtwo_2}{\tmthree}{\tight}$, by condition (a)
      \item[(b')] $\isAppr{\aset}{\tctx_2}$, with $\tctx_2$ tight since $\tctx = \tctx_1 + \tctx_2$ and  condition (b)
      \item[(c')] $\tctxtwo_2 \neq \emptyctx$ by hypothesis
      \item[(d')] $\dom{\tctxtwo_2} \subseteq \aset$ since $\tctxtwo = \tctxtwo_1 + \tctxtwo_2$ and  condition (e)
      \item[(e')] For all $\vartwo \in \dom{\tctxtwo_2}$, we have that $\tctxtwo_2(\vartwo)$ is a non-empty multiset
         since $\tctxtwo = \tctxtwo_1 + \tctxtwo_2$ and condition (e).
      \end{enumerate}
      We can apply \ih on $\tmthree$ yielding
      $\tmthree \notin \NF{\aset}{\sset}{\appflag}$, in particular $\appflag = \nonapp$, 
      which leads to a contradiction with our assumption.
      Therefore $\tmtwo\esub{\var}{\tmthree} \notin \NF{\aset}{\sset}{\appflag}$.
    \item $\mtyp = \tight, \tctxtwo_1 \neq \emptyctx$ and $\tctxtwo_2 = \emptyctx$.
      We have that $\inv{\aset}{\sset}{\tmtwo\esub{\var}{\tmthree}}$ implies $\inv{\aset \cup \set{\var}}{\sset}{\tmtwo}$,
      and that the following holds:
      \begin{enumerate}
      \item[(a')] 
        $\judg[m_1, e_1]{\tctx_1; (\tctxtwo_1; \var : \optmtyp)}{\tmtwo}{\tight}$, by condition (a)
      \item[(b')]
        $\isAppr{\aset \cup \set{\var}}{\tctx_1; \var : \optmtyp}$, as
        \begin{itemize}
        \item[(1)]
          $\isAppr{\aset \cup \set{\var}}{\tctx_1}$. 
          Indeed $\tctx = \tctx_1 + \tctx_2$ and condition (b) and \cref{rem:isAppropriate}; 
        \item[(2)]
          $\isAppr{\aset \cup \set{\var}}{\var : \optmtyp}$.
          It holds that $\isAppr{\aset}{\tctx_2;\tctxtwo_2}$, since 
          $\tctx = \tctx_1 + \tctx_2$ and $\tctxtwo = \tctxtwo_1 + \tctxtwo_2$ and conditions (b) and (e). 
          Then $\mtyp \neq \tightN$ by applying \cref{lem:types_of_hereditary_abstractions}
          on the derivation of $\tmthree$.
          Therefore we are in the case where $\optmtyp \neq \tightN$.
        \end{itemize}
        Moreover, $\tctx_1$ is tight since $\tctx = \tctx_1 + \tctx_2$ and  condition (b).
      \item[(c')] $\tctxtwo_1 \neq \emptyctx$ by hypothesis
      \item[(d')] $\dom{\tctxtwo_1} \subseteq \aset \cup \set{\var}$ 
        since $\tctxtwo = \tctxtwo_1 + \tctxtwo_2$ and  condition (d)
      \item[(e')] For all $\vartwo \in \dom{\tctxtwo_1}$, we have that $\tctxtwo_1(\vartwo)$ is a non-empty multiset
         since $\tctxtwo = \tctxtwo_1 + \tctxtwo_2$ and condition (e).
      \end{enumerate}
     We can apply \ih on $\tmtwo$ yielding
      $\tmtwo \notin \NF{\aset \cup \set{\var}}{\sset}{\appflag}$,
      which leads to a contradiction with our assumption.
      Therefore $\tmtwo\esub{\var}{\tmthree} \notin \NF{\aset}{\sset}{\appflag}$.
    \end{enumerate}
  \item \ruleUNFEsStruct.
    Then
    \[
      \indrule{\ruleUNFEsStruct}{
        \tmtwo \in \NF{\aset}{\sset \cup \set{\var}}{\appflag}
        \sep
        \tmthree \in \NF{\aset}{\sset}{\nonapp}
        \sep
        \tmthree \in \Struct{\sset}
      }{
        \tmtwo\esub{\var}{\tmthree} \in \NF{\aset}{\sset}{\appflag}
      }
    \]
    By condition (c), $\tctxtwo_1 + \tctxtwo_2 \neq \emptyctx$
    so we have to analyze two possible cases: 
    \begin{enumerate}
    \item $\tctxtwo_2 \neq \emptyctx$.
      The following holds:
      \begin{enumerate}
      \item[(a')] $\judg[m_2, e_2]{\tctx_2; \tctxtwo_2}{\tmthree}{\tight}$, by condition (a)
      \item[(b')] $\isAppr{\aset}{\tctx_2}$, with $\tctx_2$ tight since $\tctx = \tctx_1 + \tctx_2$ and by condition (b)
      \item[(c')] $\tctxtwo_2 \neq \emptyctx$ by hypothesis
      \item[(d')] $\dom{\tctxtwo_2} \subseteq \aset$ since $\tctxtwo = \tctxtwo_1 + \tctxtwo_2$ and by condition (d)
      \item[(e')] For all $\vartwo \in \dom{\tctxtwo_2}$, we have that $\tctxtwo_2(\vartwo)$ is a non-empty multiset
         since $\tctxtwo = \tctxtwo_1 + \tctxtwo_2$ and by condition (e).
      \end{enumerate}
     We can apply \ih on $\tmthree$ yielding
      $\tmthree \notin \NF{\aset}{\sset}{\appflag}$, in particular $\appflag = \nonapp$, 
      which leads to a contradiction with our assumption.
      We then conclude $\tmtwo\esub{\var}{\tmthree} \notin \NF{\aset}{\sset}{\appflag}$.     
    \item $\tctxtwo_2 = \emptyctx$ (and thus $\tctxtwo_1 \neq \emptyctx$).
      We have that $\inv{\aset}{\sset}{\tmtwo\esub{\var}{\tmthree}}$ implies $\inv{\aset}{\sset \cup \set{\var}}{\tmtwo}$,
      and that the following holds:
      \begin{enumerate}
      \item[(a')] $\judg[m_1, e_1]{\tctx_1; \tctxtwo_1; \var : \optmtyp}{\tmtwo}{\tight}$, by condition (a)
      \item[(b')] $\isAppr{\aset}{\tctx_1; \var : \optmtyp}$, with $\tctx_1$ tight, since 
        $\var \notin \aset$, $\tctx = \tctx_1 + \tctx_2$ and by condition (b)
      \item[(c')] $\tctxtwo_1 \neq \emptyctx$ by hypothesis
      \item[(d')] $\dom{\tctxtwo_1} \subseteq \aset \cup \set{\var}$ since $\tctxtwo = \tctxtwo_1 + \tctxtwo_2$ and by condition (d)
      \item[(e')] For all $\vartwo \in \dom{\tctxtwo_1}$, we have that $\tctxtwo_1(\vartwo)$ is a non-empty multiset
         since $\tctxtwo = \tctxtwo_1 + \tctxtwo_2$ and by condition (e).
      \end{enumerate}
      We can apply \ih on $\tmtwo$ yielding
      $\tmtwo \notin \NF{\aset}{\sset  \cup \set{\var}}{\appflag}$, 
      which leads to a contradiction with our assumption. 
      Therefore $\tmtwo\esub{\var}{\tmthree} \notin \NF{\aset}{\sset}{\appflag}$.
    \end{enumerate}
  \end{enumerate}
\end{itemize}
\end{proof}

% }

\begin{lemma}
\label{lem:normal_forms_zero_counter}
Let $\inv{\aset}{\sset}{\tm}$ and suppose the following hypothesis hold:
  (1) $\tm \in \NF{\aset}{\sset}{\appflag}$;
  (2) $\judg[m,e]{\tctx}{\tm}{\tight}$ is tight;
  (3) $\isAppr{\aset}{\tctx}$;
  (4) If $\appflag = \app$ then $\tight = \tightN$.
Then $(m,e) = (0,0)$.
\end{lemma}
\hiddenproof{
  By induction on the derivation
  of $\tm \in \NF{\aset}{\sset}{\appflag}$.
}{./proofs/normal_forms_zero_counter}

\subsection{Soundness of $\typesystem$}
\label{sec:soundness_results}

This subsection aims to give the main results regarding the soundness of the 
type system $\typesystem$ with respect to the $\UOCBV$ strategy.
In order to prove this result, stated in \cref{thm:soundness_typing}, we need to 
show that the subject reduction property (\cref{prop:subject_reduction}) holds.
For doing that, we start by presenting the substitution lemma for $\typesystem$ 
in \cref{lem:substitution}.

\begin{restatable}[Substitution]{lemma}{substitutionlemma}
\label{lem:substitution}
Suppose that the following conditions hold:
\begin{enumerate}
\item[(a)]
  $\tm \tov{\rulesub{\var}{\val}}{\aset\cup\set{\var}}{\sset}{\appflag} \tm'$
\item[(b)]
  $\judg[m,e]{\tctx;\var:\optmtyp}{\tm}{\mtyptwo}$
\item[(c)]
  $\isAppr{\aset\cup\set{\var}}{\tctx;\var:\optmtyp}$
\item[(d)]
  If $\appflag = \app$
  then either $\mtyptwo = \tightN$
  or $\mtyptwo$ is a singleton, \ie, of the form $[\typ]$.
\end{enumerate}
Then there exist $\typ$ and $\optmtyp_2$
such that $\optmtyp = \mset{\typ} + \optmtyp_2$
and such that,
whenever $\judg[m',e']{\tctxtwo}{\val}{\mset{\typ}}$,
we have that $e > 0$
and
$\judg[m+m',e+e'-1]{\tctx+\tctxtwo;\var:\optmtyp_2}{\tm'}{\mtyptwo}$.
\end{restatable}
% Label: lem:substitution

\begin{proof}
By induction on the derivation of
  $\tm \tov{\rulesub{\var}{\val}}{\aset\cup\set{\var}}{\sset}{\appflag} \tm'$.

\begin{enumerate}
\item $\ruleUSub$. The following conditions hold:
  \begin{enumerate}
  \item[(a)]
    $\var \tov{\rulesub{\var}{\val}}{\aset \cup \set{\var}}{\sset}{\app} \val$
  \item[(b)]
    $\judg[m,e]{\tctx;\var:\optmtyp}{\var}{\mtyptwo}$
  \item[(c)]
    $\isAppr{\aset\cup\set{\var}}{\tctx;\var:\optmtyp}$
  \item[(d)]
    Since $\appflag = \app$,
    either $\mtyptwo = \tightN$
    or $\mtyptwo$ is a singleton
  \end{enumerate}
  The judgement of condition (b)
  can only be derived by the rule $\ruleTypVar$, hence
  $\optmtyp = \mtyptwo$, $\tctx = \emptyset$, $m = 0$ and $e = \numarrows{\mtyptwo}$.
  Furthermore, 
  $\mtyptwo \neq \tightN$, 
  since $\isAppr{\aset\cup\set{\var}}{\var:\mtyptwo}$ holds. This implies by hypothesis
  that $\mtyptwo$ is a singleton $\mset{\typ}$ for some type $\typ$.
  Thus  $e = 1$.
  Taking such $\typ$ and $\optmtyp_2 = \none$, 
  whenever $\judg[m',e']{\tctxtwo}{\val}{\mset{\typ}}$, 
  we can conclude that $\judg[0 + m',1 + e' - 1]{\emptyset + \tctxtwo}{\val}{\mset{\typ}}$.
\item $\ruleUAppL$. The following conditions hold:
  \begin{enumerate}
  \item[(a)]
    $\tmtwo \, \tmthree \tov{\rulesub{\var}{\val}}{\aset \cup \set{\var}}{\sset}{\appflag} \tmtwo' \, \tmthree$,
    with $\tmtwo \tov{\rulesub{\var}{\val}}{\aset \cup \set{\var}}{\sset}{\app} \tmtwo'$.
  \item[(b)]
    $\judg[m,e]{\tctx;\var:\optmtyp}{\tmtwo \, \tmthree}{\mtyptwo}$
  \item[(c)]
    $\isAppr{\aset\cup\set{\var}}{\tctx;\var:\optmtyp}$
  \item[(d)]
    If $\appflag = \app$
    then either $\mtyptwo = \tightN$
    or $\mtyptwo$ is a singleton
  \end{enumerate}
  The judgement $\judg[m,e]{\tctx;\var:\optmtyp}{\tmtwo \, \tmthree}{\mtyptwo}$ 
  can be derived either by rule $\ruleTypAppP$ or rule $\ruleTypAppC$:
  \begin{enumerate}
  \item $\ruleTypAppP$.
    Then
    \[
      \deriv \eqdef \left(
      \indrule{\ruleTypAppP}{
        \judg[m_{\tmtwo},e_{\tmtwo}]{\tctx_{\tmtwo};\var:\optmtyp_{\tmtwo}}{\tmtwo}{\tightN}
        \sep
        \judg[m_{\tmthree},e_{\tmthree}]{\tctx_{\tmthree};\var:\optmtyp_{\tmthree}}{\tmthree}{\tight}
      }{
        \judg[m_{\tmtwo} + m_{\tmthree}, e_{\tmtwo} + e_{\tmthree}]
          {\tctx_{\tmtwo} + \tctx_{\tmthree};\var:(\optmtyp_{\tmtwo} + \optmtyp_{\tmthree})}
          {\tmtwo \, \tmthree}
          {\tightN}
      }
      \right)
    \]
    with $\tctx = \tctx_{\tmtwo} + \tctx_{\tmthree}$,
    $\optmtyp = \optmtyp_{\tmtwo} + \optmtyp_{\tmthree}$,
    $m = m_{\tmtwo} + m_{\tmthree}$ and $e = e_{\tmtwo} + e_{\tmthree}$.
    The following holds:
    \begin{enumerate}
    \item[(a')]
      $\tmtwo \tov{\rulesub{\var}{\val}}{\aset \cup \set{\var}}{\sset}{\appflag_\tmtwo} \tmtwo'$,
      with $\appflag_\tmtwo = \app$
      by condition (a)
    \item[(b')]
      $\judg[m_{\tmtwo},e_{\tmtwo}]{\tctx_{\tmtwo};\var:\optmtyp_{\tmtwo}}{\tmtwo}{\tightN}$,
      by condition (b)
    \item[(c')]
      $\isAppr{\aset\cup\set{\var}}{\tctx_{\tmtwo};\var:\optmtyp_{\tmtwo}}$,
      since $\isAppr{\aset\cup\set{\var}}{\tctx;\var:\optmtyp}$ holds, 
      and $\tctx = \tctx_{\tmtwo} + \tctx_{\tmthree}$,
      $\optmtyp = \optmtyp_{\tmtwo} + \optmtyp_{\tmthree}$,
      and by condition (c)
    \item[(d')]
      Here $\appflag_\tmtwo = \app$
      and the type of $\tmtwo$ is $\tightN$.
    \end{enumerate}
    We can apply \ih on $\tmtwo$, yielding $\typ$ and $\optmtyp_{s2}$ such that
    $\optmtyp_{\tmtwo} = \mset{\typ} + \optmtyp_{s2}$ and such that, whenever 
    $\judg[m',e']{\tctxtwo}{\val}{\mset{\typ}}$, we have that $e_{\tmtwo} > 0$ and
    $\judg[m_{\tmtwo} + m',e_{\tmtwo} + e' - 1]{\tctx_{\tmtwo} + \tctxtwo;\var:\optmtyp_{s2}}{\tmtwo'}{\tightN}$.
    Then, taking this judgement and the second premise of $\deriv$,
    we can apply rule $\ruleTypAppP$, yielding
    $\judg[m + m',e + e'  - 1]
      {\tctx + \tctxtwo; \var : (\optmtyp_{s2} + \optmtyp_{\tmthree})}
      {\tmtwo' \, \tmthree}
      {\tightN}$, 
    with $e > 0$ because $e = e_{\tmtwo} + e_{\tmthree}$, and $e_{\tmtwo} > 0$ by \ih\ 
    We take $\optmtyp_2 = \optmtyp_{s2} + \optmtyp_{\tmthree}$, 
    and we can conclude that $\optmtyp = \mset{\typ} + \optmtyp_2$ holds.
  \item $\ruleTypAppC$.
    Then
    \[
      \deriv \eqdef \left(
      \indrule{\ruleTypAppC}{
        \judg[m_{\tmtwo},e_{\tmtwo}]{\tctx_{\tmtwo};\var:\optmtyp_{\tmtwo}}{\tmtwo}{\mset{\optmtypthree \to \mtyptwo}}
        \sep
        \optmtypthree \mleq \mtypthree
        \sep
        \judg[m_{\tmthree},e_{\tmthree}]{\tctx_{\tmthree};\var:\optmtyp_{\tmthree}}{\tmthree}{\mtypthree}
      }{
        \judg[1 + m_{\tmtwo} + m_{\tmthree}, e_{\tmtwo} + e_{\tmthree}]
          {\tctx_{\tmtwo} + \tctx_{\tmthree};\var:(\optmtyp_{\tmtwo} + \optmtyp_{\tmthree})}
          {\tmtwo \, \tmthree}
          {\mtyptwo}
      }
      \right)
    \]
    with $\tctx = \tctx_{\tmtwo} + \tctx_{\tmthree}$,
    $\optmtyp = \optmtyp_{\tmtwo} + \optmtyp_{\tmthree}$,
    $m = 1 + m_{\tmtwo} + m_{\tmthree}$ and $e = e_{\tmtwo} + e_{\tmthree}$.
    The following holds:
    \begin{enumerate}
    \item[(a')]
      $\tmtwo \tov{\rulesub{\var}{\val}}{\aset \cup \set{\var}}{\sset}{\appflag_\tmtwo} \tmtwo'$
      with $\appflag_\tmtwo = \app$
      by condition (a)
    \item[(b')]
      $\judg[m_{\tmtwo},e_{\tmtwo}]
        {\tctx_{\tmtwo};\var:\optmtyp_{\tmtwo}}
        {\tmtwo}
        {\mset{\optmtypthree \to \mtyptwo}}$,
      by condition (b)
    \item[(c')]
      $\isAppr{\aset\cup\set{\var}}{\tctx_{\tmtwo};\var:\optmtyp_{\tmtwo}}$,
      since $\isAppr{\aset\cup\set{\var}}{\tctx;\var:\optmtyp}$ holds, 
      and $\tctx = \tctx_{\tmtwo} + \tctx_{\tmthree}$, and by condition (c)
      $\optmtyp = \optmtyp_{\tmtwo} + \optmtyp_{\tmthree}$
    \item[(d')]
      We have $\appflag_\tmtwo = \app$ and  the type of $\tmtwo$ is the singleton 
      $\mset{\optmtypthree \to \mtyptwo}$.
    \end{enumerate}
    We cann apply \ih on $\tmtwo$, yielding $\typ$ and $\optmtyp_{s2}$ such that
    $\optmtyp_{\tmtwo} = \mset{\typ} + \optmtyp_{s2}$ and such that, whenever 
    $\judg[m' ,e']{\tctxtwo}{\val}{\mset{\typ}}$, we have that $e_{\tmtwo} > 0$ and
    $\judg[m_{\tmtwo} + m',e_{\tmtwo} + e' - 1]
      {\tctx_{\tmtwo} + \tctxtwo;\var: \optmtyp_{s2}}
      {\tmtwo'}
      {\mset{\optmtypthree \to \mtyptwo}}$.
    Then, taking this judgement and the second premise of $\deriv$,
    we can apply rule $\ruleTypAppC$, yielding
    $\judg[m + m',e + e' - 1]
      {\tctx + \tctxtwo; \var : (\optmtyp_{s2} + \optmtyp_{\tmthree})}
      {\tmtwo' \, \tmthree}
      {\mtyptwo}$, 
    with $e > 0$ because $e = e_{\tmtwo} + e_{\tmthree}$, and $e_{\tmtwo} > 0$ by \ih 
    We take $\optmtyp_2 = \optmtyp_{s2}  + \optmtyp_{\tmthree}$, 
   and we can conclude $\optmtyp = \mset{\typ} + \optmtyp_2$ holds.
  \end{enumerate}
\HIDDENFRAGMENT{
  \item The remaining cases are similar.
}{
\item $\ruleUAppR$. The following conditions hold:
  \begin{enumerate}
  \item[(a)]
    $\tmtwo \, \tmthree \tov{\rulesub{\var}{\val}}{\aset \cup \set{\var}}{\sset}{\appflag} \tmtwo \, \tmthree'$,
    with $\tmtwo \in \Struct{\sset}$ and 
    $\tmthree \tov{\rulesub{\var}{\val}}{\aset \cup \set{\var}}{\sset}{\nonapp} \tmthree'$.
  \item[(b)]
    $\judg[m,e]{\tctx;\var:\optmtyp}{\tmtwo \, \tmthree}{\mtyptwo}$
  \item[(c)]
    $\isAppr{\aset\cup\set{\var}}{\tctx; \var : \optmtyp}$
  \item[(d)]
    If $\appflag = \app$
    then either $\mtyptwo = \tightN$
    or $\mtyptwo$ is a singleton
  \end{enumerate}
  The judgement $\judg[m,e]{\tctx;\var:\optmtyp}{\tmtwo \, \tmthree}{\mtyptwo}$ 
  can be derived either by rule $\ruleTypAppP$ 
  or by $\ruleTypAppC$ rule:
  \begin{enumerate}
  \item $\ruleTypAppP$.
    Then
    \[
      \deriv \eqdef \left(
      \indrule{\ruleTypAppP}{
        \judg[m_{\tmtwo},e_{\tmtwo}]{\tctx_{\tmtwo};\var:\optmtyp_{\tmtwo}}{\tmtwo}{\tightN}
        \sep
        \judg[m_{\tmthree},e_{\tmthree}]{\tctx_{\tmthree};\var:\optmtyp_{\tmthree}}{\tmthree}{\tight}
      }{
        \judg[m_{\tmtwo} + m_{\tmthree}, e_{\tmtwo} + e_{\tmthree}]
          {\tctx_{\tmtwo} + \tctx_{\tmthree};\var:(\optmtyp_{\tmtwo} + \optmtyp_{\tmthree})}
          {\tmtwo \, \tmthree}
          {\tightN}
      }
      \right)
    \]
    with $\tctx = \tctx_{\tmtwo} + \tctx_{\tmthree}$,
    $\optmtyp = \optmtyp_{\tmtwo} + \optmtyp_{\tmthree}$,
    $m = m_{\tmtwo} + m_{\tmthree}$ and $e = e_{\tmtwo} + e_{\tmthree}$.
    The following holds:
    \begin{enumerate}
    \item[(a')]
      $\tmthree \tov{\rulesub{\var}{\val}}{\aset \cup \set{\var}}{\sset}{\appflag_\tmthree} \tmthree'$,
      with $\appflag_\tmthree = \nonapp$
      by condition (a)
    \item[(b')]
      $\judg[m_{\tmthree},e_{\tmthree}]{\tctx_{\tmthree};\var:\optmtyp_{\tmthree}}{\tmthree}{\tight}$,
      by condition (b)
    \item[(c')]
      $\isAppr{\aset\cup\set{\var}}{\tctx_{\tmthree}; \var : \optmtyp_{\tmthree}}$, 
      since $\isAppr{\aset\cup\set{\var}}{\tctx; \var : \optmtyp}$ holds 
      and $\tctx = \tctx_{\tmtwo} + \tctx_{\tmthree}$, 
      $\optmtyp = \optmtyp_{\tmtwo} + \optmtyp_{\tmthree}$,
      and by condition (c)
    \item[(d')]
      Here $\appflag_\tmthree = \nonapp$, by condition (a).
    \end{enumerate}
    We can apply \ih on $\tmthree$, yielding $\typ$ and $\optmtyp_{s2}$ such that
    $\optmtyp_{\tmthree} = \mset{\typ} + \optmtyp_{s2}$ and such that, whenever 
    $\judg[m',e']{\tctxtwo}{\val}{\mset{\typ}}$, we have that $e_{\tmthree} > 0$ and
    $\judg[m_{\tmthree} + m',e_{\tmthree} + e' - 1]
      {\tctx_{\tmthree} + \tctxtwo;\var: \optmtyp_{s2}}
      {\tmthree'}
      {\tight}$.
    Then, taking this judgement and the first premise of $\deriv$,
    we can apply rule $\ruleTypAppP$, yielding
    $\judg[m + m',e + e' - 1]
      {\tctx + \tctxtwo; \var : (\optmtyp_{\tmtwo} + \optmtyp_{s2})}
      {\tmtwo \, \tmthree'}
      {\tightN}$, 
    with $e > 0$ because $e = e_{\tmtwo} + e_{\tmthree}$, and $e_{\tmthree} > 0$ by \ih\ 
    We take $\optmtyp_2 = \optmtyp_{\tmtwo} + \optmtyp_{s2}$, 
    and we can conclude that $\optmtyp = \mset{\typ} + \optmtyp_2$ holds.
  \item $\ruleTypAppC$.
    Then
    \[
      \deriv \eqdef \left(
      \indrule{\ruleTypAppC}{
        \judg[m_{\tmtwo},e_{\tmtwo}]{\tctx_{\tmtwo};\var:\optmtyp_{\tmtwo}}{\tmtwo}{\mset{\optmtypthree  \to \mtyptwo}}
        \sep
        \optmtypthree \mleq \mtypthree
        \sep
        \judg[m_{\tmthree},e_{\tmthree}]{\tctx_{\tmthree};\var:\optmtyp_{\tmthree}}{\tmthree}{\mtypthree}
      }{
        \judg[1 + m_{\tmtwo} + m_{\tmthree}, e_{\tmtwo} + e_{\tmthree}]
          {\tctx_{\tmtwo} + \tctx_{\tmthree};\var:(\optmtyp_{\tmtwo} + \optmtyp_{\tmthree})}
          {\tmtwo \, \tmthree}
          {\mtyptwo}
      }
      \right)
    \]
    with $\tctx = \tctx_{\tmtwo} + \tctx_{\tmthree}$,
    $\optmtyp = \optmtyp_{\tmtwo} + \optmtyp_{\tmthree}$,
    $m = 1 + m_{\tmtwo} + m_{\tmthree}$ and $e = e_{\tmtwo} + e_{\tmthree}$.
    The following holds:
    \begin{enumerate}
    \item[(a')]
      $\tmthree \tov{\rulesub{\var}{\val}}{\aset \cup \set{\var}}{\sset}{\appflag_\tmthree} \tmthree'$,
      with $\appflag_\tmthree = \nonapp$
      by condition (a)
    \item[(b')]
      $\judg[m_{\tmtwo},e_{\tmtwo}]{\tctx_{\tmtwo};\var:\optmtyp_{\tmtwo}}{\tmtwo}{\mset{\optmtypthree \to \mtyptwo}}$
     by condition (b)
    \item[(c')]
      $\isAppr{\aset\cup\set{\var}}{\tctx_{\tmthree}; \var : \optmtyp_{\tmthree}}$, 
      since $\isAppr{\aset\cup\set{\var}}{\tctx; \var : \optmtyp}$ holds 
      and $\tctx = \tctx_{\tmtwo} + \tctx_{\tmthree}$, 
      $\optmtyp = \optmtyp_{\tmtwo} + \optmtyp_{\tmthree}$,
      and by condition (c)
    \item[(d')]
      Here $\appflag_\tmthree = \nonapp$, by condition (a).
    \end{enumerate}
    We can apply \ih on $\tmthree$, yielding $\typ$ and $\optmtyp_{u2}$ such that
    $\optmtyp_{\tmthree} = \mset{\typ} + \optmtyp_{u2}$ and such that, whenever 
    $\judg[m', e']{\tctxtwo}{\val}{\mset{\typ}}$, we have that $e_{\tmthree} > 0$ and
    $\judg[m_{\tmthree} + m',e_{\tmthree} + e' - 1]
      {\tctx_{\tmthree} + \tctxtwo;\var: \optmtyp_{u2}}
      {\tmthree'}
      {\mtypthree}$.
    Then, taking this judgement and the first premise of $\deriv$,
    we can apply rule $\ruleTypAppC$ rule, yielding
    $\judg[m + m', e + e' - 1]
      {\tctx + \tctxtwo; \var : (\optmtyp_{\tmtwo} + \optmtyp_{u2})}
      {\tmtwo \, \tmthree'}
      {\mtyptwo}$,
    with $e > 0$ because $e = e_{\tmtwo} + e_{\tmthree}$, and $e_{\tmthree} > 0$ by \ih
    We take $\optmtyp_2 = \optmtyp_{\tmtwo} + \optmtyp_{u2}$, 
    and we can conclude that $\optmtyp = \mset{\typ} + \optmtyp_2$ holds.
  \end{enumerate}
\item $\ruleUEsR$. The following conditions hold:
  \begin{enumerate}
  \item[(a)]
    $\tmtwo\esub{\vartwo}{\tmthree} 
     \tov{\rulesub{\var}{\val}}{\aset \cup \set{\var}}{\sset}{\appflag} 
     \tmtwo\esub{\vartwo}{\tmthree'}$, with 
    $\tmthree 
     \tov{\rulesub{\var}{\val}}{\aset \cup \set{\var}}{\sset}{\nonapp} 
     \tmthree'$
  \item[(b)]
    $\judg[m,e]{\tctx;\var:\optmtyp}{\tmtwo\esub{\vartwo}{\tmthree}}{\mtyptwo}$
  \item[(c)]
    $\isAppr{\aset\cup\set{\var}}{\tctx; \var : \optmtyp}$
  \item[(d)]
    If $\appflag = \app$
    then either $\mtyptwo = \tightN$
    or $\mtyptwo$ is a singleton.
  \end{enumerate}
  The judgement
  $\judg[m,e]{\tctx;\var:\optmtyp}{\tmtwo\esub{\vartwo}{\tmthree}}{\mtyptwo}$
  can only be derived using rule $\ruleTypES$,
  so
  \[
    \deriv \eqdef \left(
    \indrule{\ruleTypES}{
      \judg[m_{\tmtwo},e_{\tmtwo}]{\tctx_{\tmtwo};\var:\optmtyp_{\tmtwo};\vartwo: \optmtypthree}{\tmtwo}{\mtyptwo}
      \sep
      \optmtypthree \mleq \mtypthree
      \sep
      \judg[m_{\tmthree},e_{\tmthree}]{\tctx_{\tmthree};\var:\optmtyp_{\tmthree}}{\tmthree}{\mtypthree}
    }{
      \judg[m_{\tmtwo} + m_{\tmthree}, e_{\tmtwo} + e_{\tmthree}]
        {\tctx_{\tmtwo} + \tctx_{\tmthree};\var:(\optmtyp_{\tmtwo} + \optmtyp_{\tmthree})}
        {\tmtwo\esub{\vartwo}{\tmthree}}
        {\mtyptwo}
    }
    \right)
  \]
  with $\tctx = \tctx_{\tmtwo} + \tctx_{\tmthree}$,
  $\optmtyp = \optmtyp_{\tmtwo} + \optmtyp_{\tmthree}$,
  $m = m_{\tmtwo} + m_{\tmthree}$ and $e = e_{\tmtwo} + e_{\tmthree}$.
  Notice that $\var \neq \vartwo$ since $\vartwo$ does not occur free in the term.
  The following holds:
  \begin{enumerate}
  \item[(a')]
    $\tmthree 
     \tov{\rulesub{\var}{\val}}{\aset \cup \set{\var}}{\sset}{\appflag_\tmthree}
     \tmthree'$,
    where $\appflag_\tmthree = \nonapp$
    by condition (a)
  \item[(b')]
    $\judg[m_{\tmthree},e_{\tmthree}]
      {\tctx_{\tmthree};\var:\optmtyp_{\tmthree}}
      {\tmthree}
      {\mtypthree}$,
    by condition (b)
  \item[(c')]
    $\isAppr{\aset\cup\set{\var}}{\tctx_{\tmthree}; \var : \optmtyp_{\tmthree}}$, 
    since $\isAppr{\aset\cup\set{\var}}{\tctx; \var : \optmtyp}$ holds 
    and $\tctx = \tctx_{\tmtwo} + \tctx_{\tmthree}$, 
    $\optmtyp = \optmtyp_{\tmtwo} + \optmtyp_{\tmthree}$,
    and by condition (c)
  \item[(d')]
    Here $\appflag_\tmthree = \nonapp$, by condition (a).
  \end{enumerate}
  We can apply \ih on $\tmthree$, yielding $\typ$ and $\optmtyp_{u2}$ such that
  $\optmtyp_{\tmthree} = \mset{\typ} + \optmtyp_{u2}$ and such that, whenever 
  $\judg[m', e']{\tctxtwo}{\val}{\mset{\typ}}$, we have that $e_{\tmthree} > 0$ and
  $\judg[m_{\tmthree} + m',e_{\tmthree} + e' - 1]
    {\tctx_{\tmthree} + \tctxtwo;\var: \optmtyp_{u2}}
    {\tmthree'}
    {\mtypthree}$.
  Then, taking the last judgement and the first premise of $\deriv$,
 we can apply rule $\ruleTypES$, yielding
  $\judg[m + m', e + e' - 1]
    {\tctx + \tctxtwo; \var : (\optmtyp_{\tmtwo} + \optmtyp_{u2})}
    {\tmtwo\esub{\vartwo}{\tmthree'}}
    {\mtyptwo}$,
  with $e > 0$ because $e = e_{\tmtwo} + e_{\tmthree}$, and $e_{\tmthree} > 0$ by \ih\
  We take $\optmtyp_2 = \optmtyp_{\tmtwo} + \optmtyp_{u2}$, 
  and we can conclude that $\optmtyp = \mset{\typ} + \optmtyp_2$ holds.
\item $\ruleUEsLAbs$. The following conditions hold: 
  \begin{enumerate}
  \item[(a)]
    $\tmtwo\esub{\vartwo}{\tmthree} 
     \tov{\rulesub{\var}{\val}}{\aset \cup \set{\var}}{\sset}{\appflag} 
     \tmtwo'\esub{\vartwo}{\tmthree}$,
    with 
    $\tmtwo 
     \tov{\rulesub{\var}{\val}}{\aset \cup \set{\var} \cup \set{\vartwo}}{\sset}{\appflag}
     \tmtwo'$,
    $\tmthree \in \HAbs{\aset \cup \set{\var}}$,
    $\vartwo \notin \aset \cup \set{\var} \cup \sset$
    and $\vartwo \notin \fv{\rulesub{\var}{\val}}$
  \item[(b)]
    $\judg[m,e]{\tctx;\var:\optmtyp}{\tmtwo\esub{\vartwo}{\tmthree}}{\mtyptwo}$
  \item[(c)]
    $\isAppr{\aset\cup\set{\var}}{\tctx; \var : \optmtyp}$
  \item[(d)]
    If $\appflag = \app$
    then either $\mtyptwo = \tightN$
    or $\mtyptwo$ is a singleton.
  \end{enumerate}
  Notice that $\var \neq \vartwo$ since by hypothesis 
  $\vartwo \notin \aset \cup \set{\var} \cup \sset$,
  and that the judgement 
  $\judg[m,e]{\tctx;\var:\optmtyp}{\tmtwo\esub{\vartwo}{\tmthree}}{\mtyptwo}$
  can be  only derived using rule $\ruleTypES$, so
  \[
    \deriv \eqdef \left(
    \indrule{\ruleTypES}{
      \judg[m_{\tmtwo},e_{\tmtwo}]{\tctx_{\tmtwo};\var:\optmtyp_{\tmtwo};\vartwo: \optmtypthree}{\tmtwo}{\mtyptwo}
      \sep
      \optmtypthree \mleq \mtypthree
      \sep
      \judg[m_{\tmthree},e_{\tmthree}]{\tctx_{\tmthree};\var:\optmtyp_{\tmthree}}{\tmthree}{\mtypthree}
      }
     {
      \judg[m_{\tmtwo} + m_{\tmthree}, e_{\tmtwo} + e_{\tmthree}]
        {\tctx_{\tmtwo} + \tctx_{\tmthree};\var:(\optmtyp_{\tmtwo} + \optmtyp_{\tmthree})}
        {\tmtwo\esub{\vartwo}{\tmthree}}
        {\mtyptwo}
    }
    \right)
    \]
  with $\tctx = \tctx_{\tmtwo} + \tctx_{\tmthree}$,
  $\optmtyp = \optmtyp_{\tmtwo} + \optmtyp_{\tmthree}$,
  $m = m_{\tmtwo} + m_{\tmthree}$ and $e = e_{\tmtwo} + e_{\tmthree}$.
  The following holds:
  \begin{enumerate}
  \item[(a')]
    $\tmtwo 
     \tov{\rulesub{\var}{\val}}{\aset \cup \set{\var} \cup \set{\vartwo}}{\sset}{\appflag} 
     \tmtwo'$, by condition (a)
  \item[(b')]
    $\judg[m_{\tmtwo},e_{\tmtwo}]
      {\tctx_{\tmtwo};\var:\optmtyp_{\tmtwo};\vartwo: \optmtypthree}
      {\tmtwo}
      {\mtyptwo}$,
    condition (b)
  \item[(c')]
    $\isAppr
      {\aset\cup\set{\var}\cup\set{\vartwo}}
      {\tctx_{\tmtwo}; \var : \optmtyp_{\tmtwo}; \vartwo : \optmtypthree}$: 
    on the one hand,
    $\isAppr{\aset\cup\set{\var}}{\tctx; \var : \optmtyp_{\tmtwo}}$ holds
    and $\tctx = \tctx_{\tmtwo} + \tctx_{\tmthree}$, 
    $\optmtyp = \optmtyp_{\tmtwo} + \optmtyp_{\tmthree}$;
    on  the other hand,
    $\mtypthree \neq \tightN$ by \cref{lem:types_of_hereditary_abstractions} on $\tmthree$, 
    so $\optmtypthree \neq \tightN$ since 
    $\optmtypthree \mleq \mtypthree$.
  \item[(d')]
    If $\appflag = \app$
    then either $\mtyptwo = \tightN$
    or $\mtyptwo$ is a singleton, 
    by condition (d).
  \end{enumerate}
  We can apply \ih on $\tmtwo$, yielding $\typ$ and $\optmtyp_{s2}$ such that
  $\optmtyp_{\tmtwo} = \mset{\typ} + \optmtyp_{s2}$ and such that, whenever 
  $\judg[m', e']{\tctxtwo}{\val}{\mset{\typ}}$, we have that $e_{\tmthree} > 0$ and
  $\judg[m_{\tmtwo} + m', e_{\tmtwo} + e' - 1]
    {\tctx_{\tmtwo} + \tctxtwo;\var: \optmtyp_{s2};\vartwo: \optmtypthree}
    {\tmtwo'}
    {\mtyptwo}$ .
  Then, taking the last judgement and the first premise of $\deriv$,
  we can apply rule $\ruleTypES$, yielding
  $\judg[m + m',e + e' - 1]
    {\tctx + \tctxtwo; \var : (\optmtyp_{s2} + \optmtyp_{\tmthree})}
    {\tmtwo'\esub{\vartwo}{\tmthree}}
    {\mtyptwo}$, 
  with $e > 0$ because $e = e_{\tmtwo} + e_{\tmthree}$, and $e_{\tmtwo} > 0$ by \ih\
  And since we take $\optmtyp_2 = \optmtyp_{s2} + \optmtyp_{\tmthree}$,
  we can conclude $\optmtyp = \mset{\typ} + \optmtyp_2$.
\item $\ruleUEsLStruct$. The following conditions hold: 
  \begin{enumerate}
  \item[(a)]
    $\tmtwo\esub{\vartwo}{\tmthree} 
     \tov{\rulesub{\var}{\val}}{\aset\cup\set{\var}}{\sset}{\appflag} 
     \tmtwo'\esub{\vartwo}{\tmthree}$,
    with
    $\tmtwo 
     \tov{\rulesub{\var}{\val}}{\aset\cup\set{\var}}{\sset\cup\set{\vartwo}}{\appflag} 
     \tmtwo'$,
    $\tmthree \in \Struct{\sset}$, $\vartwo \notin \aset \cup \set{\var} \cup \sset$
    and $\vartwo \notin \fv{\rulesub{\var}{\val}}$
  \item[(b)]
    $\judg[m,e]{\tctx;\var:\optmtyp}{\tmtwo\esub{\vartwo}{\tmthree}}{\mtyptwo}$
  \item[(c)]
    $\isAppr{\aset\cup\set{\var}}{\tctx; \var : \optmtyp}$
  \item[(d)]
    If $\appflag = \app$
    then either $\mtyptwo = \tightN$
    or $\mtyptwo$ is a singleton, \ie, of the form $[\typ]$.    
  \end{enumerate}
  Notice that $\var \neq \vartwo$ since by hypothesis 
  $\vartwo \notin \aset \cup \set{\var} \cup \sset$,
  and that the judgement 
  $\judg[m,e]{\tctx;\var:\optmtyp}{\tmtwo\esub{\vartwo}{\tmthree}}{\mtyptwo}$
  can be only  derived using rule $\ruleTypES$, so
  \[
    \deriv \eqdef \left(
    \indrule{\ruleTypES}{
      \judg[m_{\tmtwo},e_{\tmtwo}]{\tctx_{\tmtwo};\var:\optmtyp_{\tmtwo};\vartwo: \optmtypthree}{\tmtwo}{\mtyptwo}
      \sep
      \optmtypthree \mleq \mtypthree
      \sep
      \judg[m_{\tmthree},e_{\tmthree}]{\tctx_{\tmthree};\var:\optmtyp_{\tmthree}}{\tmthree}{\mtypthree}
    }{
      \judg[m_{\tmtwo} + m_{\tmthree}, e_{\tmtwo} + e_{\tmthree}]
        {\tctx_{\tmtwo} + \tctx_{\tmthree};\var:(\optmtyp_{\tmtwo} + \optmtyp_{\tmthree})}
        {\tmtwo\esub{\vartwo}{\tmthree}}
        {\mtyptwo}
    }
    \right)
  \]
  with $\tctx = \tctx_{\tmtwo} + \tctx_{\tmthree}$,
  $\optmtyp = \optmtyp_{\tmtwo} + \optmtyp_{\tmthree}$,
  $m = m_{\tmtwo} + m_{\tmthree}$ and $e = e_{\tmtwo} + e_{\tmthree}$.
  The following holds:
  \begin{enumerate}
  \item[(a')]
    $\tmtwo 
     \tov{\rulesub{\var}{\val}}{\aset \cup \set{\var}}{\sset \cup \set{\vartwo}}{\appflag} 
     \tmtwo'$, by condition (a)
  \item[(b')]
    $\judg[m_{\tmtwo},e_{\tmtwo}]
      {\tctx_{\tmtwo};\var:\optmtyp_{\tmtwo};\vartwo: \optmtypthree}
      {\tmtwo}
      {\mtyptwo}$,
    by condition (b)
  \item[(c')] $\isAppr {\aset\cup\set{\var}} {\tctx_{\tmtwo}; \var :
    \optmtyp_{\tmtwo}; \vartwo : \optmtypthree}$, since (1)
    $\isAppr{\aset\cup\set{\var}}{\tctx; \var : \optmtyp}$ holds and
    $\tctx = \tctx_{\tmtwo} + \tctx_{\tmthree}$, $\optmtyp =
    \optmtyp_{\tmtwo} + \optmtyp_{\tmthree}$, and condition (c),
    (2) $\vartwo \notin \aset\cup\set{\var}$ by $\alpha$-conversion 
  \item[(d')]
    If $\appflag = \app$
    then either $\mtyptwo = \tightN$
    or $\mtyptwo$ is a singleton, 
    by condition (d).
  \end{enumerate}
  We can apply \ih on $\tmtwo$, yielding $\typ$ and $\optmtyp_{s2}$ such that
  $\optmtyp_{\tmtwo} = \mset{\typ} + \optmtyp_{s2}$ and such that, whenever 
  $\judg[m',e']{\tctxtwo}{\val}{\mset{\typ}}$, we have that $e_{\tmthree} > 0$ and
  $\judg[m_{\tmtwo} + m',e_{\tmtwo} + e' - 1]
    {\tctx_{\tmtwo} + \tctxtwo;\var: \optmtyp_{s2};\vartwo: \optmtypthree}
    {\tmtwo'}
    {\mtyp_{\tmtwo}}$.
  Then, taking the last judgement and the first premise of $\deriv$,
  we can apply rule $\ruleTypES$, yielding
  $\judg[m + m',e + e' - 1]
    {\tctx + \tctxtwo; \var : (\optmtyp_{s2} + \optmtyp_{\tmthree})}
    {\tmtwo'\esub{\vartwo}{\tmthree}}
    {\mtyptwo}$, 
  with $e > 0$ because $e = e_{\tmtwo} + e_{\tmthree}$, and $e_{\tmtwo} > 0$ by \ih\
  We take $\optmtyp_2 = \optmtyp_{s2} + \optmtyp_{\tmthree}$,
  and we can conclude that $\optmtyp = \mset{\typ} + \optmtyp_2$ holds.
}
\end{enumerate}
\end{proof}

\subjectreduction*
% Label: prop:subject_reduction

\sloppy
\begin{proof}
By induction on the derivation of
$\tm \tov{\rulename}{\aset}{\sset}{\appflag} \tm'$.
\begin{enumerate}
\item \ruleUDb.
  The following conditions hold:
  \begin{enumerate}
  \item[(a)]
    $(\lam{\var}{\tmtwo})\sctx \, \tmthree \tov{\ruledb}{\aset}{\sset}{\appflag} 
     \tmtwo\esub{\var}{\tmthree}\sctx$
  \item[(b)]
    $\judg[m,e]{\tctx}{(\lam{\var}{\tmtwo})\sctx \, \tmthree}{\mtyp}$
  \item[(c)]
    $\isAppr{\aset}{\tctx}$
  \item[(d)]
    If $\appflag = \app$
    then either $\mtyp = \tightN$
    or $\mtyp$ is a singleton, \ie, of the form $[\typ]$.
  \end{enumerate}
  The judgement of condition (b) can be derived either by rule $\ruleTypAppP$
  or rule $\ruleTypAppC$.
  \begin{enumerate}
  \item $\ruleTypAppP$.
    Then
    \[
      \indrule{\ruleTypAppP}{
        \judg[m_1,e_1]{\tctx_1}{(\lam{\var}{\tmtwo})\sctx}{\tightN}
        \sep
        \judg[m_2,e_2]{\tctx_2}{\tmthree}{\tight}
      }{
        \judg[m_1 + m_2,e_1 + e_2]{\tctx_1 + \tctx_2}{(\lam{\var}{\tmtwo})\sctx \, \tmthree}{\tightN}
      }
    \]
    where $\tctx = \tctx_1+\tctx_2$, $m=m_1 + m_2$, $e=e_1 + e_2$, $\tm=(\lam{\var}{\tmtwo})\sctx\, \tmthree$ and $\mtyp=\tightN$.
    By \cref{lem:composition}, there exist
    $\tctx_{11}, \tctx_{12}, \tctxtwo, m_{11}, e_{11}, m_{12}, e_{12}$ such that
    $\judgSctx[m_{12},e_{12}]{\tctx_{12}}{\sctx}{\tctxtwo}$ and
    $\judg[m_{11},e_{11}]{\tctx_{11}; \tctxtwo}{\lam{\var}{\tmtwo}}{\tightN}$.
    This leads to a contradiction, since an abstraction cannot have type $\tightN$. 
    Then this case is not possible.
  \item $\ruleTypAppC$.
    Then
    \[
      \indrule{\ruleTypAppC}{
        \judg[m_1,e_1]{\tctx_1}{(\lam{\var}{\tmtwo})\sctx}{\mset{\optmtyptwo \to \mtyp}}
        \sep
        \optmtyptwo \mleq \mtyptwo
        \sep
        \judg[m_2,e_2]{\tctx_2}{\tmthree}{\mtyptwo}
      }{
        \judg[1 + m_1 + m_2,e_1 + e_2]{\tctx_1 + \tctx_2}{(\lam{\var}{\tmtwo})\sctx \, \tmthree}{\mtyp}
      }
    \]
    where $\tctx = \tctx_1+\tctx_2$, $m=1+m_1 + m_2$, $e=e_1 + e_2$,
    and $\tm=(\lam{\var}{\tmtwo})\sctx\, \tmthree$.
    By \cref{lem:composition}, there exist
    $\tctx_{11}, \tctx_{12}, \tctxtwo, m_{11}, e_{11}, m_{12}, e_{12}$ such that
    $\judg[m_{11},e_{11}]{\tctx_{11}; \tctxtwo}{\lam{\var}{\tmtwo}}{\mset{\optmtyptwo \to \mtyp}}$
    and $\judgSctx[m_{12},e_{12}]{\tctx_{12}}{\sctx}{\tctxtwo}$, with 
    $\tctx_1 = \tctx_{11} + \tctx_{12}$, 
    $m_1 = m_{11} + m_{12}$ and 
    $e_1 = e_{11} + e_{12}$.

    The judgement for the term $\lam{\var}{\tmtwo}$ can only be derived by rule \ruleTypAbs:
    \[
      \indrule{\ruleTypAbs}{
        \judg[m_{11},e_{11}]
          {\tctx_{11}; \tctxtwo; \var : \optmtyptwo}
          {\tmtwo}
          {\mtyp}
      }{
        \judg[m_{11},e_{11}]
          {\tctx_{11}; \tctxtwo}
          {\lam{\var}{\tmtwo}}
          {\mset{\optmtyptwo \to \mtyp}}
      }
    \]
    Then we build the following derivation:
    \[
      \indrule{\ruleTypES}{
        \judg[m_{11},e_{11}]
          {\tctx_{11}; \tctxtwo; \var : \optmtyptwo}
          {\tmtwo}
          {\mtyp}
        \sep
        \optmtyptwo \mleq \mtyptwo
        \sep
        \judg[m_2,e_2]{\tctx_2}{\tmthree}{\mtyptwo}
      }{
        \judg[m_{11} + m_2,e_{11} + e_2]
          {\tctx_{11} + \tctx_2; \tctxtwo}
          {\tmtwo\esub{\var}{\tmthree}}
          {\mtyp}
      }
    \]
    By applying \cref{lem:composition} again, we can conclude that
    $\judg[m_1 + m_2, e_1 + e_2]{\tctx_1 + \tctx_2}{\tmtwo\esub{\var}{\tmthree}\sctx}{\mtyp}$.
    In this case $\rulename = \ruledb$, and
    with
    $m > 0$, and 
    $(m',e') = (m_1 + m_2, e_1 + e_2) = (m - 1, e)$.
  \end{enumerate}
\item \ruleULsv.
  The following conditions hold:
  \begin{enumerate}
  \item[(a)]
    \[
      \indrule{\ruleULsv}{
        \tmtwo 
        \tov{\rulesub{\var}{\val}}{\aset \cup \set{\var}}{\sset}{\appflag} 
        \tmtwo' 
        \sep
        \var \notin \aset \cup \sset
        \sep
        \val\sctx \in \HAbs{\aset}        
      }{
        \tmtwo\esub{\var}{\val\sctx} 
        \tov{\rulelsv}{\aset}{\sset}{\appflag} 
        \tmtwo'\esub{\var}{\val}\sctx 
      }
    \]
  \item[(b)]
    $\judg[m,e]{\tctx}{\tmtwo\esub{\var}{\val\sctx}}{\mtyp}$
  \item[(c)]
    $\isAppr{\aset}{\tctx}$
  \item[(d)]
    If $\appflag = \app$
    then either $\mtyp = \tightN$
    or $\mtyp$ is a singleton, \ie, of the form $[\typ]$.
  \end{enumerate}
  The judgement of condition (b) can only be derived by the rule $\ruleTypES$,
  so
  \[
    \indrule{\ruleTypES}{
      \judg[m_1,e_1]{\tctx_1; \var : \optmtyptwo}{\tmtwo}{\mtyp}
      \sep
      \optmtyptwo \mleq \mtyptwo
      \sep
      \judg[m_2,e_2]{\tctx_2}{\val\sctx}{\mtyptwo}
    }{
      \judg[m_1 + m_2,e_1 + e_2]{\tctx_1 + \tctx_2}{\tmtwo\esub{\var}{\val\sctx}}{\mtyp}
    }
  \]
  where $\tctx = \tctx_1+\tctx_2$, $m=m_1 + m_2$, $e=e_1 + e_2$,
  and $\tm=\tmtwo\esub{\var}{\val\sctx}$.
  By \cref{lem:composition}, there exist
  $\tctx_{21}, \tctx_{22}, \tctxtwo, m_{21}, e_{21}, m_{22}, e_{22}$ such that
  $\judg[m_{21},e_{21}]{\tctx_{21}; \tctxtwo}{\val}{\mtyptwo}$
  and $\judgSctx[m_{22},e_{22}]{\tctx_{22}}{\sctx}{\tctxtwo}$, with 
  $\tctx_2 = \tctx_{21} + \tctx_{22}$, 
  $m_2 = m_{21} + m_{22}$ and 
  $e_2 = e_{21} + e_{22}$.

  Moreover, we can assume
  $\var \notin \dom{\tctx}$
  by $\alpha$-conversion 
  so that $\isAppr{\aset\cup \set{\var}}{\tctx}$ also holds.
  Therefore $\isAppr{\aset \cup \set{\var}}{\tctx_1; \var : \optmtyptwo}$ 
  holds since (1): $\isAppr{\aset\cup \set{\var}}{\tctx}$ and $\tctx_1$ is part
  of $\tctx$, and  (2): $\val\sctx \in \HAbs{\aset}$ by condition (a),
  so $\mtyptwo \neq \tightN$
  by \cref{lem:types_of_hereditary_abstractions} 
  and then $\optmtyptwo \mleq \mtyptwo$
  in condition (b)
  implies $\optmtyptwo \neq \tightN$.
  Then, we have all the conditions to apply \cref{lem:substitution}, 
  yielding $\typ$ and $\optmtyptwo_2$ such that
  $\optmtyptwo = \mset{\typ} + \optmtyptwo_2$,
  hence $\optmtyptwo=\mtyptwo$, thus $\mtyptwo = \mset{\typ} + \optmtyptwo_2$.
  We have two cases, depending on whether $\optmtyptwo_2$ is $\none$ or not:
  \begin{enumerate}
  \item $\optmtyptwo_2 = \none$. Then, $\mtyptwo = \mset{\typ} + \optmtyptwo_2 = \mset{\typ} = \mset{\typ} + \emset$.
    Recall that $\judg[m_{21},e_{21}]{\tctx_{21}; \tctxtwo}{\val}{\mset{\typ}}$
    holds and note that moreover $\judg[0,0]{\emptyctx}{\val}{\emset}$.
    By \cref{lem:substitution}
    we have that $e_1 > 0$ and
    $\judg[m_1 + m_{21},e_1 + e_{21} - 1]
      {\tctx_1 + (\tctx_{21}; \tctxtwo); \var : \none}
      {\tmtwo'}
      {\mtyp}$.
    Then we build the following derivation:
    \[
      \indrule{\ruleTypES}{
      \judg[m_1 + m_{21},e_1 + e_{21} - 1]
          {\tctx_1 + (\tctx_{21}; \tctxtwo); \var : \none}
          {\tmtwo'}
          {\mtyp}
        \sep
        \none \mleq \emset
        \sep
        \judg[0, 0]{\emptyctx}{\val}{\emset}
      }{
        \judg[m_1 + m_{21},e_1 + e_{21} - 1]
          {\tctx_1 + \tctx_{21}; \tctxtwo}
          {\tmtwo'\esub{\var}{\val}}
          {\mtyp}
      }
    \]
  \item $\optmtyptwo_2 \neq \none$.
    Then $\optmtyptwo_2 = \mtyptwo_2$
    so we can write $\mtyptwo = \mset{\typ} + \mtyptwo_2$.
    Hence we have $\judg[m_{21}, e_{21}]{\tctx_{21}; \tctxtwo}{\val}{\mset{\typ} + \mtyptwo_2}$.
    We can apply \cref{lem:splitting}, so that there exists
    $\tctx'_{211} = \tctx_{211}; \tctxtwo_1, 
     \tctx'_{212} = \tctx_{212}; \tctxtwo_2$, 
    $m_{211}, m_{212}, e_{211}, e_{212}$ such that
    $\judg[m_{211}, e_{211}]{\tctx'_{211}}{\val}{\mset{\typ}}$ and
    $\judg[m_{212}, e_{212}]{\tctx'_{212}}{\val}{\mtyptwo_2}$,
    with
    $\tctx_{21} = \tctx_{211} + \tctx_{212}$
    and $\tctxtwo = \tctxtwo_1 + \tctxtwo_2$, 
    and $m_{21} = m_{211} + m_{212}$ and $e_{21} = e_{211} + e_{212}$.
    So by \cref{lem:substitution},
    whenever 
    $\judg[m_{211},e_{211}]{\tctx_{211}; \tctxtwo_1}{\val}{\mset{\typ}}$, we have that $e_1 > 0$ and
    $\judg[m_1 + m_{211},e_1 + e_{211} - 1]
      {\tctx_1 + (\tctx_{211}; \tctxtwo_1); \var : \mtyptwo_2}
      {\tmtwo'}
      {\mtyp}$.
    Then we build the following derivation:
    \[
      \indrule{\ruleTypES}{
      \judg[m_1 + m_{211},e_1 + e_{211} - 1]
          {\tctx_1 + (\tctx_{211}; \tctxtwo_1); \var : \mtyptwo_2}
          {\tmtwo'}
          {\mtyp}
        \sep
        \mtyptwo_2 \mleq \mtyptwo_2
        \sep
        \judg[m_{212},e_{212}]{\tctx_{212}; \tctxtwo_2}{\val}{\mtyptwo_2}
      }{
        \judg[m_1 + m_{21},e_1 + e_{21} - 1]
          {\tctx_1 + \tctx_{21}; \tctxtwo}
          {\tmtwo'\esub{\var}{\val}}
          {\mtyp}
      }
    \]
  \end{enumerate}
  We can now apply \cref{lem:composition} using the judgements
  $\judgSctx[m_{22},e_{22}]{\tctx_{22}}{\sctx}{\tctxtwo}$ and 
  $\judg[m_1 + m_{21}, e_1 + e_{21} - 1]
    {\tctx_1 + \tctx_{21}; \tctxtwo}
    {\tmtwo'\esub{\var}{\val}}
    {\mtyp}$,
  yielding
  $\judg[m, e - 1]{\tctx}{\tmtwo'\esub{\var}{\val}\sctx}{\mtyp}$.
  In this case $\rulename = \rulelsv$ and it holds that
  $e > 0$ since $e_1 > 0$ and $e = e_1 + e_2$.
  Moreover, $(m', e') = (m, e - 1)$, so we are done.
\HIDDENFRAGMENT{
\item The congruence cases are uninteresting and are omitted here.
}{
  \item \ruleUAppL. 
    The following conditions hold:
    \begin{enumerate}
    \item[(a)]
      $\indrule{\ruleUAppL}{
        \tmtwo \tov{\rulename}{\aset}{\sset}{\app} \tmtwo'
      }{
        \tmtwo \, \tmthree \tov{\rulename}{\aset}{\sset}{\appflag} \tmtwo' \, \tmthree
      }$
      where $\rulename \in \set{\ruledb,\rulelsv}$
    \item[(b)]
      $\judg[m,e]{\tctx}{\tmtwo \, \tmthree}{\mtyp}$
    \item[(c)]
      $\isAppr{\aset}{\tctx}$
    \item[(d)]
      If $\appflag = \app$
      then either $\mtyp = \tightN$
      or $\mtyp$ is a singleton, \ie, of the form $[\typ]$.
    \end{enumerate}

    The judgement of condition (b) can be derived either by rule $\ruleTypAppP$ or rule $\ruleTypAppC$.
    \begin{enumerate}
    \item $\ruleTypAppP$.
      Then
      \[
        \indrule{\ruleTypAppP}{
          \judg[m_1,e_1]{\tctx_1}{\tmtwo}{\tightN}
          \sep
          \judg[m_2,e_2]{\tctx_2}{\tmthree}{\tight}
        }{
          \judg[m_1 + m_2,e_1 + e_2]{\tctx_1 + \tctx_2}{\tmtwo \, \tmthree}{\tightN}
        }
      \]
      where $\tctx = \tctx_1+\tctx_2$, $m=m_1 + m_2$, $e=e_1 + e_2$, $\tm=\tmtwo \, \tmthree$ and $\mtyp=\tightN$.
      We have that the following hold:
      \begin{enumerate}
      \item[(a')]
        $\tmtwo \tov{\rulename}{\aset}{\sset}{\app} \tmtwo'$
        where $\rulename \in \set{\ruledb,\rulelsv}$
      \item[(b')]
        $\judg[m_1,e_1]{\tctx_1}{\tmtwo}{\tightN}$
      \item[(c')]
        $\isAppr{\aset}{\tctx_1}$, since $\isAppr{\aset}{\tctx}$ by hypothesis and $\tctx = \tctx_1 + \tctx_2$ 
      \item[(d')]
        It holds that $\appflag = \app$, and the term $\tmtwo$ 
        is typed with $\tightN$.
      \end{enumerate}
      We can apply \ih on $\tmtwo$, yielding
      $\judg[m'_1,e'_1]{\tctx_1}{\tmtwo'}{\tightN}$,
      where,
      if $\rulename = \ruledb$  then $m_1 > 0$ and $(m'_1,e'_1) = (m_1 - 1,e_1)$, and
      if $\rulename = \rulelsv$ then $e_1 > 0$ and $(m'_1,e'_1) = (m_1,e_1-1)$.
      Hence we can apply rule $\ruleTypAppP$:
      \[
        \indrule{\ruleTypAppP}{
          \judg[m'_1,e'_1]{\tctx_1}{\tmtwo'}{\tightN}
          \sep
          \judg[m_2,e_2]{\tctx_2}{\tmthree}{\tight}
        }{
          \judg[m'_1 + m_2,e'_1 + e_2]{\tctx_1 + \tctx_2}{\tmtwo' \, \tmthree}{\tightN}
        }
      \]
      where
      if $\rulename = \ruledb$ then
      $m_1 + m_2 > 0$ since $m_1 > 0$, and 
      $(m'_1 + m_2,e'_1 + e_2) = (m_1 - 1 + m_2,e_1 + e_2)$ by \ih,
      and
      if $\rulename = \rulelsv$ then
      $e_1 + e_2 > 0$ since $e_1 > 0$,
      and 
      $(m'_1 + m_2,e'_1 + e_2) = (m_1 + m_2, e_1 - 1 + e_2)$ by \ih,
      so we are done.
    \item $\ruleTypAppC$.
      Then
      \[
        \indrule{\ruleTypAppC}{
          \judg[m_1,e_1]{\tctx_1}{\tmtwo}{\mset{\optmtyptwo \to \mtyp}}
          \sep
          \optmtyptwo \mleq \mtyptwo
          \sep
          \judg[m_2,e_2]{\tctx_2}{\tmthree}{\mtyptwo}
        }{
          \judg[1 + m_1 + m_2,e_1 + e_2]{\tctx_1 + \tctx_2}{\tmtwo \, \tmthree}{\mtyp}
        }
      \]
      where $\tctx = \tctx_1+\tctx_2$, $m=1+m_1 + m_2$, $e=e_1 + e_2$,
      and $\tm = \tmtwo \, \tmthree$. 
      We have that the following hold:
      \begin{enumerate}
      \item[(a')]
        $\tmtwo \tov{\rulename}{\aset}{\sset}{\app} \tmtwo'$
        where $\rulename \in \set{\ruledb,\rulelsv}$
      \item[(b')]
        $\judg[m_1,e_1]{\tctx_1}{\tmtwo}{\mset{\optmtyptwo \to \mtyp}}$
      \item[(c')]
        $\isAppr{\aset}{\tctx_1}$, since $\isAppr{\aset}{\tctx}$ by hypothesis and $\tctx = \tctx_1 + \tctx_2$ 
      \item[(d')]
        In case $\appflag = \app$,
        then it is not possible to have $\mset{\optmtyptwo \to \mtyp} = \tightN$,
        but it holds that $\mset{\optmtyptwo \to \mtyp}$ is a singleton.
      \end{enumerate}
      We can apply \ih on $\tmtwo$, yielding
      $\judg[m'_1,e'_1]{\tctx_1}{\tmtwo'}{\mset{\optmtyptwo \to \mtyp}}$,
      where,
      if $\rulename = \ruledb$  then $m_1 > 0$ and $(m'_1,e'_1) = (m_1 - 1,e_1)$, and
      if $\rulename = \rulelsv$ then $e_1 > 0$ and $(m'_1,e'_1) = (m_1,e_1-1)$.
      Hence we can apply rule $\ruleTypAppC$:
      \[
        \indrule{\ruleTypAppC}{
          \judg[m'_1,e'_1]{\tctx_1}{\tmtwo'}{\mset{\optmtyptwo \to \mtyp}}
          \sep
          \optmtyptwo \mleq \mtyptwo
          \sep
          \judg[m_2,e_2]{\tctx_2}{\tmthree}{\mtyptwo}
        }{
          \judg[1 + m'_1 + m_2,e'_1 + e_2]{\tctx_1 + \tctx_2}{\tmtwo' \, \tmthree}{\mtyp}
        }
      \]
      where
      if $\rulename = \ruledb$ then 
      $1 + m_1 + m_2 > 0$, 
      and 
      $ (1 + m'_1 + m_2,e'_1 + e_2) 
      = (1 + m_1 - 1 + m_2,e_1 + e_2)
      = (m_1 + m_2, e_1 + e_2$ 
      by \ih, and
      if $\rulename = \rulelsv$ then 
      $e_1 + e_2 > 0$ since $e_1 > 0$ by \ih,
      and 
      $(1 + m'_1 + m_2, e'_1 + e_2) = (1 + m_1 + m_2, e_1 - 1 + e_2)$ by \ih,
      so we are done.
    \end{enumerate}
  \item \ruleUAppR.
    The following conditions hold:
    \begin{enumerate}
    \item[(a)]
      $\indrule{\ruleUAppR}{
        \tmtwo \in \Struct{\sset}
        \sep
        \tmthree \tov{\rulename}{\aset}{\sset}{\nonapp} \tmthree'
      }{
        \tmtwo \, \tmthree \tov{\rulename}{\aset}{\sset}{\appflag} \tmtwo \, \tmthree'
      }$
      where $\rulename \in \set{\ruledb,\rulelsv}$
    \item[(b)]
      $\judg[m,e]{\tctx}{\tmtwo \, \tmthree}{\mtyp}$
    \item[(c)]
      $\isAppr{\aset}{\tctx}$
    \item[(d)]
      If $\appflag = \app$
      then either $\mtyp = \tightN$
      or $\mtyp$ is a singleton, \ie, of the form $[\typ]$.
    \end{enumerate}

    The judgement of condition (b) can be derived either by rule $\ruleTypAppP$ or rule $\ruleTypAppC$.
    \begin{enumerate}
    \item \ruleTypAppP.
      Then
      \[
        \indrule{\ruleTypAppP}{
          \judg[m_1,e_1]{\tctx_1}{\tmtwo}{\tightN}
          \sep
          \judg[m_2,e_2]{\tctx_2}{\tmthree}{\tight}
        }{
          \judg[m_1 + m_2,e_1 + e_2]{\tctx_1 + \tctx_2}{\tmtwo \, \tmthree}{\tightN}
        }
      \]
      where $\tctx = \tctx_1+\tctx_2$, $m=m_1 + m_2$, $e=e_1 + e_2$, $\tm=\tmtwo \, \tmthree$ and $\mtyp=\tightN$.
      We have that the following hold:
      \begin{enumerate}
      \item[(a')]
        $\tmthree \tov{\rulename}{\aset}{\sset}{\nonapp} \tmthree'$
        where $\rulename \in \set{\ruledb,\rulelsv}$
      \item[(b')]
        $\judg[m_2,e_2]{\tctx_2}{\tmthree}{\tight}$
      \item[(c')]
        $\isAppr{\aset}{\tctx_2}$, since $\isAppr{\aset}{\tctx}$ by hypothesis and $\tctx = \tctx_1 + \tctx_2$ 
      \item[(d')]
        Here $\appflag \neq \app$.
      \end{enumerate}
      We can apply \ih on $\tmthree$, yielding
      $\judg[m'_2,e'_2]{\tctx_2}{\tmthree'}{\tight}$,
      where,
      if $\rulename = \ruledb$  then $m_2 > 0$ and $(m'_2,e'_2) = (m_2 - 1,e_2)$, and
      if $\rulename = \rulelsv$ then $e_2 > 0$ and $(m'_2,e'_2) = (m_2,e_2-1)$.
      Hence we can apply rule $\ruleTypAppP$:
      \[
        \indrule{\ruleTypAppP}{
          \judg[m_1,e_1]{\tctx_1}{\tmtwo}{\tightN}
          \sep
          \judg[m'_2,e'_2]{\tctx_2}{\tmthree'}{\tight}
        }{
          \judg[m_1 + m'_2,e_1 + e'_2]{\tctx_1 + \tctx_2}{\tmtwo \, \tmthree'}{\tightN}
        }
      \]
      where
      if $\rulename = \ruledb$ then 
      $m_1 + m_2 > 0$ since $m_2 > 0$ 
      and
      $(m_1 + m'_2, e_1 + e'_2) = (m_1 + m_2 - 1, e_1 + e_2)$ by \ih,
      and if $\rulename = \rulelsv$ then 
      $e_1 + e_2 > 0$ since $e_2 > 0$ and 
      $(m_1 + m'_2,e_1 + e'_2) = (m_1 + m_2, e_1 + e_2 - 1)$ by \ih,
      so we are done.
    \item \ruleTypAppC.
      Then
      \[
        \indrule{\ruleTypAppC}{
          \judg[m_1,e_1]{\tctx_1}{\tmtwo}{\mset{\optmtyptwo \to \mtyp}}
          \sep
          \optmtyptwo \mleq \mtyptwo
          \sep
          \judg[m_2,e_2]{\tctx_2}{\tmthree}{\mtyptwo}
        }{
          \judg[1 + m_1 + m_2,e_1 + e_2]{\tctx_1 + \tctx_2}{\tmtwo \, \tmthree}{\mtyp}
        }
      \]
      where $\tctx = \tctx_1+\tctx_2$, $m=1+m_1 + m_2$, $e=e_1 + e_2$, and $\tm=\tmtwo \, \tmthree$.
      We have that the following hold:
      \begin{enumerate}
      \item[(a')]
        $\tmthree \tov{\rulename}{\aset}{\sset}{\nonapp} \tmthree'$
        where $\rulename \in \set{\ruledb,\rulelsv}$
      \item[(b')]
        $\judg[m_2,e_2]{\tctx_2}{\tmthree}{\mtyptwo}$
      \item[(c')]
        $\isAppr{\aset}{\tctx_2}$, since $\isAppr{\aset}{\tctx}$ by hypothesis and $\tctx = \tctx_1 + \tctx_2$ 
      \item[(d')]
        Here $\appflag \neq \app$.
      \end{enumerate}
      We can apply \ih on $\tmthree$, yielding
      $\judg[m'_2,e'_2]{\tctx_2}{\tmthree'}{\mtyptwo}$,
      where,
      if $\rulename = \ruledb$  then $m_2 > 0$ and $(m'_2,e'_2) = (m_2 - 1,e_2)$, and
      if $\rulename = \rulelsv$ then $e_2 > 0$ and $(m'_2,e'_2) = (m_2,e_2-1)$.
      Hence we can apply rule \ruleTypAppC:
      \[
        \indrule{\ruleTypAppC}{
          \judg[m_1,e_1]{\tctx_1}{\tmtwo}{\mset{\optmtyptwo \to \mtyp}}
          \sep
          \optmtyptwo \mleq \mtyptwo
          \sep
          \judg[m'_2,e'_2]{\tctx_2}{\tmthree'}{\mtyptwo}
        }{
          \judg[1 + m_1 + m'_2,e_1 + e'_2]{\tctx_1 + \tctx_2}{\tmtwo \, \tmthree'}{\mtyp}
        }
      \]
      where
      if $\rulename = \ruledb$ then
      $1 + m_1 + m_2 > 0$, and 
      $(1 + m_1 + m'_2, e_1 + e'_2) = (1 + m_1 + m_2 - 1,e_1 + e_2)$ by \ih, 
      and if $\rulename = \rulelsv$ then
      $e_1 + e_2 > 0$ since $e_2 > 0$ and 
      $(1 + m_1 + m'_2, e_1 + e'_2) = (1 + m_1 + m_2, e_1 + e_2 - 1)$ by \ih,
      so we are done.
    \end{enumerate}
  \item $\ruleUEsR$.
    This case is analogous to the previous one.
  \item $\ruleUEsLAbs$.
    The following conditions hold:
    \begin{enumerate}
    \item[(a)]
      $\indrule{\ruleUEsLAbs}{
        \tmtwo \tov{\rulename}{\aset \cup \set{\var}}{\sset}{\rulename} \tmtwo'
        \sep
        \tmthree \in \HAbs{\aset}
        \sep
        \var \notin \aset \cup \sset
        \sep
        \var \notin \fv{\rulename}
      }{
        \tmtwo\esub{\var}{\tmthree} \tov{\rulename}{\aset}{\sset}{\appflag} \tmtwo'\esub{\var}{\tmthree}
      }$
      where $\rulename \in \set{\ruledb,\rulelsv}$
    \item[(b)]
      $\judg[m,e]{\tctx}{\tmtwo\esub{\var}{\tmthree}}{\mtyp}$
    \item[(c)]
      $\isAppr{\aset}{\tctx}$
    \item[(d)]
      If $\appflag = \app$
      then either $\mtyp = \tightN$
      or $\mtyp$ is a singleton, \ie, of the form $[\typ]$.
    \end{enumerate}
    The judgement of condition (b) can be derived only by the rule $\ruleTypES$:
    \[
      \indrule{\ruleTypES}{
        \judg[m_1,e_1]{\tctx_1; \var : \optmtyptwo}{\tmtwo}{\mtyp}
        \sep
        \optmtyptwo \mleq \mtyptwo
        \sep
        \judg[m_2,e_2]{\tctx_2}{\tmthree}{\mtyptwo}
      }{
        \judg[m_1 + m_2,e_1 + e_2]{\tctx_1 + \tctx_2}{\tmtwo\esub{\var}{\tmthree}}{\mtyp}
      }
    \]
    where $\tctx = \tctx_1+\tctx_2$, $m=m_1 + m_2$, $e=e_1 + e_2$, and $\tm=\tmtwo\esub{\var}{\tmthree}$.
    We have that the following hold:
    \begin{enumerate}
    \item[(a')]
      $\tmtwo \tov{\rulename}{\aset \cup \set{\var}}{\sset}{\app} \tmtwo'$
      where $\rulename \in \set{\ruledb,\rulelsv}$
    \item[(b')]
      $\judg[m_1,e_1]{\tctx_1; \var : \optmtyptwo}{\tmtwo}{\mtyp}$
    \item[(c')]
      $\isAppr{\aset \cup \set{\var}}{\tctx_1; \var : \optmtyptwo}$, since 
      (1) $\isAppr{\aset}{\tctx}$ by hypothesis and $\tctx = \tctx_1 + \tctx_2$, and
      (2) given $\tmthree \in \HAbs{\aset}$ by hypothesis of rule $\ruleUEsLAbs$, 
      we can apply \cref{lem:types_of_hereditary_abstractions} on $\tmthree$ 
      yielding $\mtyptwo \neq \tightN$, so since $\optmtyptwo \mleq \mtyptwo$ 
      we can say that $\optmtyptwo \neq \tightN$
    \item[(d')]
      If $\appflag = \app$
      then either $\mtyp = \tightN$
      or $\mtyp$ is a singleton, \ie, of the form $[\typ]$, by condition (d).
    \end{enumerate}
    We can apply \ih on $\tmtwo$, yielding
    $\judg[m'_1,e'_1]{\tctx_1; \var : \optmtyptwo}{\tmtwo'}{\mtyp}$,
    where,
    if $\rulename = \ruledb$  then $m_1 > 0$ and $(m'_1,e'_1) = (m_1 - 1,e_1)$, and
    if $\rulename = \rulelsv$ then $e_1 > 0$ and $(m'_1,e'_1) = (m_1,e_1-1)$.
    Hence we can apply rule \ruleTypES:
    \[
      \indrule{\ruleTypES}{
        \judg[m'_1,e'_1]{\tctx_1; \var : \optmtyptwo}{\tmtwo'}{\mtyp}
        \sep
        \optmtyptwo \mleq \mtyptwo
        \sep
        \judg[m_2,e_2]{\tctx_2}{\tmthree}{\mtyptwo}
      }{
        \judg[m'_1 + m_2,e'_1 + e_2]{\tctx_1 + \tctx_2}{\tmtwo'\esub{\var}{\tmthree}}{\mtyp}
      }
    \]
    where
    if $\rulename = \ruledb$ then 
    $m_1 + m_2 > 0$ since $m_1 > 0$ and 
    $(m'_1 + m_2, e'_1 + e_2) = (m_1 - 1 + m_2, e_1 + e_2)$ by \ih, 
    and if $\rulename = \rulelsv$ then 
    $e_1 + e_2 > 0$ and 
    $(m'_1 + m_2, e'_1 + e_2) = (m_1 + m_2, e_1 - 1 + e_2)$ by \ih,
    so we are done.
  \item $\ruleUEsLStruct$.
    The following conditions hold:
    \begin{enumerate}
    \item[(a)]
      $\indrule{\ruleUEsLStruct}{
        \tmtwo \tov{\rulename}{\aset}{\sset \cup \set{\var}}{\rulename} \tmtwo'
        \sep
        \tmthree \in \Struct{\sset}
        \sep
        \var \notin \aset \cup \sset
        \sep
        \var \notin \fv{\rulename}
      }{
        \tmtwo\esub{\var}{\tmthree} \tov{\rulename}{\aset}{\sset}{\appflag} \tmtwo'\esub{\var}{\tmthree}
      }$
      where $\rulename \in \set{\ruledb,\rulelsv}$
    \item[(b)]
      $\judg[m,e]{\tctx}{\tmtwo\esub{\var}{\tmthree}}{\mtyp}$
    \item[(c)]
      $\isAppr{\aset}{\tctx}$
    \item[(d)]
      If $\appflag = \app$
      then either $\mtyp = \tightN$
      or $\mtyp$ is a singleton, \ie, of the form $[\typ]$.
    \end{enumerate}
    The judgement of condition (b) can be derived only by the rule \ruleTypES:
    \[
      \indrule{\ruleTypES}{
        \judg[m_1,e_1]{\tctx_1; \var : \optmtyptwo}{\tmtwo}{\mtyp}
        \sep
        \optmtyptwo \mleq \mtyptwo
        \sep
        \judg[m_2,e_2]{\tctx_2}{\tmthree}{\mtyptwo}
      }{
        \judg[m_1 + m_2,e_1 + e_2]{\tctx_1 + \tctx_2}{\tmtwo\esub{\var}{\tmthree}}{\mtyp}
      }
    \]
    where $\tctx = \tctx_1+\tctx_2$, $m=m_1 + m_2$, $e=e_1 + e_2$, and $\tm=\tmtwo\esub{\var}{\tmthree}$.
    We have that the following hold:
    \begin{enumerate}
    \item[(a')]
      $\tmtwo \tov{\rulename}{\aset}{\sset \cup \set{\var}}{\app} \tmtwo'$
      where $\rulename \in \set{\ruledb,\rulelsv}$
    \item[(b')]
      $\judg[m_1,e_1]{\tctx_1; \var : \optmtyptwo}{\tmtwo}{\mtyp}$
    \item[(c')]
      $\isAppr{\aset}{\tctx_1; \var : \optmtyptwo}$, since 
      (1) $\isAppr{\aset}{\tctx}$ by hypothesis and $\tctx = \tctx_1 + \tctx_2$, and
      (2) $\var \notin \aset \cup \sset$ by hypothesis
    \item[(d')]
      If $\appflag = \app$
      then either $\mtyp = \tightN$
      or $\mtyp$ is a singleton, \ie, of the form $[\typ]$, by condition (d).
    \end{enumerate}
    We can apply \ih on $\tmtwo$, yielding
    $\judg[m'_1,e'_1]{\tctx_1; \var : \optmtyptwo}{\tmtwo'}{\mtyp}$,
    where,
    if $\rulename = \ruledb$  then $m_1 > 0$ and $(m'_1,e'_1) = (m_1 - 1,e_1)$, and
    if $\rulename = \rulelsv$ then $e_1 > 0$ and $(m'_1,e'_1) = (m_1,e_1-1)$.
    Hence we can apply rule $\ruleTypES$:
    \[
      \indrule{\ruleTypES}{
        \judg[m'_1,e'_1]{\tctx_1; \var : \optmtyptwo}{\tmtwo'}{\mtyp}
        \sep
        \optmtyptwo \mleq \mtyptwo
        \sep
        \judg[m_2,e_2]{\tctx_2}{\tmthree}{\mtyptwo}
      }{
        \judg[m'_1 + m_2,e'_1 + e_2]{\tctx_1 + \tctx_2}{\tmtwo'\esub{\var}{\tmthree}}{\mtyp}
      }
    \]
    where
    if $\rulename = \ruledb$ then 
    $m_1 + m_2 > 0$ since $m_1 > 0$ and 
    $(m'_1 + m_2, e'_1 + e_2) = (m_1 - 1 + m_2, e_1 + e_2)$ by \ih, 
    and if $\rulename = \rulelsv$ then 
    $e_1 + e_2 > 0$ since $e_1 > 0$ and 
    $(m'_1 + m_2, e'_1 + e_2) = (m_1 + m_2, e_1 - 1 + e_2)$ by \ih,
    so we are done.
}
\end{enumerate}
\end{proof}

\soundnesstyping*

\begin{proof}
By induction on $\cm + \ce$, separating in cases depending of whether 
$\tm \in \NF\emptyset{\fv\tm}\nonapp$:
\begin{enumerate}
\item 
  If $\tm \in \NF{\emptyset}{\fv{\tm}}{\nonapp}$,
  we also have that the judgement $\judg[m,e]{\tctx}{\tm}{\mtyp}$ is tight, 
  $\isAppr{\emptyset}{\tctx}$, and $\appflag = \nonapp$.
  We can apply \cref{lem:normal_forms_zero_counter}, yielding $m = 0, e = 0$. 
  Then we conclude with $\tmtwo \eqdef \tm$.
\item 
  If $\tm \notin \NF{\emptyset}{\fv{\tm}}{\nonapp}$,
  then 
  by \cref{thm:characterization_of_useful_normal_forms}
  there exist a term $\tm'$ and a rule name $\rulename$ 
  such that
  $\tm \tov{\rulename}{\emptyset}{\fv{\tm}}{\nonapp} \tm'$.
  By Subject reduction (\cref{prop:subject_reduction}), we have that
  $\judg[m', e']{\tctx}{\tm'}{\mtyp}$, 
  with $m > 0$ and $(m', e') = (m - 1, e)$ if $\rulename = \ruledb$,
  and if $\rulename = \rulelsv$ then $e > 0$ and $(m', e') = (m, e - 1)$.
  Since the judgement is tight, we can apply \ih,
  so there exists a term $\tmtwo \in \NF{\emptyset}{\fv{\tm'}}{\nonapp}$ such that
  $\tm' \tov{\rulename_1}{\emptyset}{\fv{\tm'}}{\nonapp} \hdots \tov{\rulename_n}{\emptyset}{\fv{\tm'}}{\nonapp} \tmtwo$,
  where $n = m + e - 1$ and 
  \[
    m' = \#\set{i \ST 1 \leq i \leq n, \rulename_i = \ruledb}
    \HS
    e' = \#\set{i \ST 1 \leq i \leq n, \rulename_i = \rulelsv}
  \]

  Then we have a reduction sequence. If we rename $\rulename$ by $\rulename_{n + 1}$, we can conclude that
  \[
    \tm    \tov{\rulename_{n + 1}}{\emptyset}{\fv{\tm}}{\nonapp} 
    \tm'   \tov{\rulename_1}{\emptyset}{\fv{\tm}}{\nonapp}
    \hdots \tov{\rulename_n}{\emptyset}{\fv{\tm}}{\nonapp}
    \tmtwo
  \]
  where
  $1 + n = m + e$ and if
  \begin{itemize}
  \item $\rulename_{n+1} = \ruledb$:
    \[
      \begin{array}{rcl}
          m
        & = 
        & (m - 1) + 1
        \\
        & = 
        & m' + 1
        \\
        & = 
        & \#\set{i \ST 1 \leq i \leq 1, \rulename_i = \ruledb} + 1
        \\
        & = 
        & \#\set{i \ST 1 \leq i \leq 1 + n, \rulename_i = \ruledb}
      \end{array}
      \begin{array}{rcl}
          e
        & =
        & e'
        \\
        & =
        & \#\set{i \ST 1 \leq i \leq n, \rulename_i = \rulelsv}
        \\
        & = 
        & \#\set{i \ST 1 \leq i \leq 1 + n, \rulename_i = \rulelsv}
      \end{array}
    \]
  \item $\rulename_{n+1} = \rulelsv$:
    \[
      \begin{array}{rcl}
          m
        & = 
        & m'
        \\
        & = 
        & \#\set{i \ST 1 \leq i \leq n, \rulename_i = \ruledb}
        \\
        & = 
        & \#\set{i \ST 1 \leq i \leq 1 + n, \rulename_i = \ruledb}
      \end{array}
      \begin{array}{rcl}
          e
        & =
        & (e - 1) + 1
        \\
        & =
        & e' + 1
        \\
        & =
        & \#\set{i \ST 1 \leq i \leq n, \rulename_i = \rulelsv} + 1
        \\
        & = 
        & \#\set{i \ST 1 \leq i \leq 1 + n, \rulename_i = \rulelsv}
      \end{array}
    \]
  \end{itemize}
\end{enumerate}
\end{proof}

\subsection{Completeness of $\typesystem$}
\label{sec:completeness_results}

In this subsection we give the main results regarding the completeness of
the type system $\typesystem$ with respect to the $\UOCBV$ strategy.
In order to prove this result, stated in \cref{thm:completeness_typing}, we need to 
show that the subject expansion property (\cref{prop:subject_expansion}) holds.
For doing that, we start by 
presenting the anti-substitution lemma for $\typesystem$ in \cref{lem:anti_substitution}.
Moreover, given that we are working with a tight typing system, we also need 
to show that normal forms of \UOCBV are tight typable, as stated in \cref{prop:nfs_are_tight_typable}.

\nftighttypable*
% Label: prop:nfs_are_tight_typable

\begin{proof}
By induction on the derivation of $\tm \in \NF{\aset}{\sset}{\appflag}$. \\

\begin{enumerate}
\item $\ruleUNFVar$. Then
  \[
  \indrule{\ruleUNFVar}{
    \var \in \aset \implies \appflag = \nonapp
  }{
    \var \in \NF{\aset}{\sset}{\appflag}
  }
  \]
  There are two possible cases depending on 
  whether $\var \in \aset$ or $\var \in \sset$.
  \begin{enumerate}
  \item $\var \in \aset$.
    By definition $\TEnv{\aset}{\sset}{\var}(\var) = \emset$, 
    since $\var \in \aset \cap \rv{\var}$ and 
      $\TEnv{\aset}{\sset}{\var}(\vartwo)=\none$ for any other variable $\vartwo$. 
    Having $\numarrows{\emset} = 0$,
    we apply rule \ruleTypVar, yielding
    $\judg[0,0]{\var : \emset}{\var}{\emset}$,
    with $\var \in \HAbs{\aset}$ by rule \ruleHAbsVar.
  \item $\var \in \sset$.
    By definition $\TEnv{\aset}{\sset}{\var}(\var) = \tightN$, 
    since $\var \in \sset \cap \rv{\var}$ and 
      $\TEnv{\aset}{\sset}{\var}(\vartwo)=\none$ for any other variable $\vartwo$. 
    Having $\numarrows{\tightN} = 0$,
     we apply rule \ruleTypVar, yielding
    $\judg[0,0]{\var : \tightN}{\var}{\tightN}$,
    with $\var \in \Struct{\sset}$ by rule \ruleStructVar.
  \end{enumerate}
\item $\ruleUNFLam$. Then
  \[
  \indrule{\ruleUNFLam}{}{
    \lam{\var}{\tmtwo} \in \NF{\aset}{\sset}{\nonapp}
  }
  \]
  For any variable $\vartwo$, we have
  $\TEnv{\aset}{\sset}{\lam{\var}{\tmtwo}}(\vartwo) = \none$
  since $\rv{\lam{\var}{\tmtwo}} = \emptyset$.
  We apply rule \ruleTypAbs, yielding
  $\judg[0,0]{}{\lam{\var}{\tmtwo}}{\mset{}}$,
  with $\lam{\var}{\tmtwo} \in \HAbs{\aset}$ by rule \ruleHAbsLam.
\item $\ruleUNFApp$. Then
  \[
  \indrule{\ruleUNFApp}{
    \tmtwo \in \NF{\aset}{\sset}{\app}
    \sep
    \tmthree \in \NF{\aset}{\sset}{\nonapp}
  }{
    \tmtwo \, \tmthree \in \NF{\aset}{\sset}{\appflag}
  }
  \]
  Since $\inv{\aset}{\sset}{\tmtwo\, \tmthree}$ implies 
  $\inv{\aset}{\sset}{\tmtwo}$ and $\inv{\aset}{\sset}{\tmthree}$, 
  then by \ih on $\tmtwo$ and $\tmthree$ 
  there exist tight types $\tight_1$ and $\tight_2$ such that
  $\judg[0,0]{\TEnv{\aset}{\sset}{\tmtwo}}{\tmtwo}{\tight_1}$ and
  $\judg[0,0]{\TEnv{\aset}{\sset}{\tmthree}}{\tmthree}{\tight_2}$. 
  Moreover, $\tmtwo \in \Struct{\sset}$ by \cref{lem:nf_in_HAbs_or_Struct}, 
  so $\tight_1 = \tightN$.
  We apply rule \ruleTypAppP, yielding
  $\judg[0,0]
    {\TEnv{\aset}{\sset}{\tmtwo} + \TEnv{\aset}{\sset}{\tmthree}}
    {\tmtwo \, \tmthree}
    {\tightN}$, with $\tmtwo \, \tmthree \in \Struct{\sset}$ by rule \ruleStructApp.
  Notice that 
  $\TEnv{\aset}{\sset}{\tmtwo} + \TEnv{\aset}{\sset}{\tmthree} = \TEnv{\aset}{\sset}{\tmtwo \, \tmthree}$,
  so we are done.
\item $\ruleUNFEsAbs$. Then
  \[
  \indrule{\ruleUNFEsAbs}{
    \tmtwo \in \NF{\aset \cup \set{\var}}{\sset}{\appflag}
    \sep
    \tmthree \in \NF{\aset}{\sset}{\nonapp}
    \sep
    \tmthree \in \HAbs{\aset}
  }{
    \tmtwo\esub{\var}{\tmthree} \in \NF{\aset}{\sset}{\appflag}
  }
  \]
  Since $\inv{\aset}{\sset}{\tmtwo\esub{\var}{\tmthree}}$ implies 
  $\inv{\aset \cup \set{\var}}{\sset}{\tmtwo}$ and $\inv{\aset}{\sset}{\tmthree}$, 
  then by \ih on $\tmtwo$ and $\tmthree$ there exist 
  tight types $\tight_1$ and $\tight_2$ such that
  $\judg[0,0]{\TEnv{\aset \cup \set{\var}}{\sset}{\tmtwo}}{\tmtwo}{\tight_1}$ and
  $\judg[0,0]{\TEnv{\aset}{\sset}{\tmthree}}{\tmthree}{\tight_2}$. Moreover,
  $\tight_2 = \emset$ since $\tmthree \in \HAbs{\aset}$.
  Notice that $\TEnv{\aset \cup \set{\var}}{\sset}{\tmtwo} = 
  \TEnv{\aset}{\sset}{\tmtwo}; \var : \TEnv{\aset \cup \set{\var}}{\sset}{\tmtwo}(\var)$.
  Moreover, $\TEnv{\aset \cup \set{\var}}{\sset}{\tmtwo}(\var)$ is $\emset$ if $\var \in \rv{\tmtwo}$, and $\none$ otherwise.
  We build the following derivation:
  \[
    \indrule{\ruleTypES}{
      \indrule{}{
        \text{(By \ih)}
      }{
        \judg[0,0]
          {\TEnv{\aset}{\sset}{\tmtwo}; \var : \TEnv{\aset \cup \set{\var}}{\sset}{\tmtwo}(\var)}
          {\tmtwo}
          {\tight_1}
      }
      \sep
      \indrule{}{
        \text{(By \ih)}
      }{
        \judg[0,0]{\TEnv{\aset}{\sset}{\tmthree}}{\tmthree}{\emset}
      }
    }{
      \judg[0,0]{\TEnv{\aset}{\sset}{\tmtwo} + \TEnv{\aset}{\sset}{\tmthree}}{\tmtwo\esub{\var}{\tmthree}}{\tight_1}
    }
  \]
  where $\TEnv{\aset \cup \set{\var}}{\sset}{\tmtwo}(\var) \mleq \emset$ necessarily holds since 
  $\TEnv{\aset \cup \set{\var}}{\sset}{\tmtwo}(\var)$ is $\none$ or $\emset$.
  Notice that
  $\TEnv{\aset}{\sset}{\tmtwo} + \TEnv{\aset}{\sset}{\tmthree} = \TEnv{\aset}{\sset}{\tmtwo\esub{\var}{\tmthree}}$,
  so we are done.
\item $\ruleUNFEsStruct$. Then
  \[
  \indrule{\ruleUNFEsStruct}{
    \tmtwo \in \NF{\aset}{\sset \cup \set{\var}}{\appflag}
    \sep
    \tmthree \in \NF{\aset}{\sset}{\nonapp}
    \sep
    \tmthree \in \Struct{\sset}
  }{
    \tmtwo\esub{\var}{\tmthree} \in \NF{\aset}{\sset}{\appflag}
  }
  \]
  Since $\inv{\aset}{\sset}{\tmtwo\esub{\var}{\tmthree}}$ implies
  $\inv{\aset}{\sset \cup \set{\var}}{\tmtwo}$ 
  and $\inv{\aset}{\sset}{\tmthree}$, then 
  by \ih on $\tmtwo$ and $\tmthree$  there exist 
  tight types $\tight_1$ and $\tight_2$ such that
  $\judg[0,0]{\TEnv{\aset}{\sset \cup \set{\var}}{\tmtwo}}{\tmtwo}{\tight_1}$ and
  $\judg[0,0]{\TEnv{\aset}{\sset}{\tmthree}}{\tmthree}{\tight_2}$. Moreover,
  $\tight_2 = \tightN$ since $\tmthree \in \Struct{\sset}$.
  Notice that 
  $\TEnv{\aset}{\sset \cup \set{\var}}{\tmtwo} = 
   \TEnv{\aset}{\sset}{\tmtwo}; \var : \TEnv{\aset}{\sset \cup \set{\var}}{\tmtwo}(\var)$.
  Moreover, $\TEnv{\aset \cup \set{\var}}{\sset}{\tmtwo}(\var)$
  is $\tightN$ if $\var \in \rv{\tmtwo}$ and $\none$ otherwise.
  We build the following derivation:
  \[
    \indrule{\ruleTypES}{
      \indrule{}{
        \text{(By \ih)}
      }{
        \judg[0,0]{\TEnv{\aset}{\sset}{\tmtwo}; \var : \TEnv{\aset \cup \set{\var}}{\sset}{\tmtwo}(\var)}{\tmtwo}{\tight_1}
      }
      \sep
      \indrule{}{
        \text{(By \ih)}
      }{
        \judg[0,0]{\TEnv{\aset}{\sset}{\tmthree}}{\tmthree}{\tightN}
      }
    }{
      \judg[0,0]{\TEnv{\aset}{\sset}{\tmtwo} + \TEnv{\aset}{\sset}{\tmthree}}{\tmtwo\esub{\var}{\tmthree}}{\tight_1}
    }
  \]
  where $\TEnv{\aset \cup \set{\var}}{\sset}{\tmtwo}(\var) \mleq \tightN$ necessarily holds since 
  $\TEnv{\aset \cup \set{\var}}{\sset}{\tmtwo}(\var)$ is $\none$ or $\tightN$. 
  Notice that
  $\TEnv{\aset}{\sset}{\tmtwo} + \TEnv{\aset}{\sset}{\tmthree} = \TEnv{\aset}{\sset}{\tmtwo\esub{\var}{\tmthree}}$,
  so we are done.
\end{enumerate}
\end{proof}

\begin{restatable}[Anti-substitution]{lemma}{antisubstitution}
\label{lem:anti_substitution}
Let $\inv{\aset \cup \set{\var}}{\sset}{\tm}$.
Consider a subset $\aset_0 \subseteq \aset$
and let $\asettwo$ be a set of variables disjoint from $\aset$.
Suppose also that the following conditions hold:
\begin{enumerate}
\item[(a)]
  $\tm \tov{\rulesub{\var}{\val}}{\aset\cup\set{\var}}{\sset}{\appflag} \tm'$
\item[(b)]
  $\judg[m,e]{\tctx;\var:\optmtyp}{\tm'}{\mtyptwo}$
\item[(c)]
  $\isAppr{\aset\cup\asettwo\cup\set{\var}}{\tctx;\var:\optmtyp}$
\item[(d)]
  If $\appflag = \app$
  then either $\mtyptwo = \tightN$
  or $\mtyptwo$ is a singleton, \ie, of the form $\mset{\typ}$.
\item[(e)]
  $\val \in \HAbs{\aset_0\cup\asettwo}$
\end{enumerate}
Then there exist $\tctx_\tm,\tctx_\val,\typ,m_\tm,e_\tm,m_\val,e_\val$
such that $\tctx = \tctx_\tm + \tctx_\val$
and $m = m_\tm + m_\val$ and $e = e_\tm + e_\val$;
$\judg[m_\tm,e_\tm+1]{\tctx_\tm;\var:\optmtyp+\mset{\typ}}{\tm}{\mtyptwo}$;
and $\judg[m_\val,e_\val]{\tctx_\val}{\val}{\mset{\typ}}$.
\end{restatable}
% Label: lem:anti_substitution

\begin{proof}
By induction on the derivation of $\tm \tov{\rulesub{\var}{\val}}{\aset \cup \set{\var}}{\sset}{\appflag} \tm'$.
\begin{enumerate}
\item \ruleUSub.
  Then the following conditions hold:
  \begin{enumerate}
  \item[(a)] 
    $\var \tov{\rulesub{\var}{\val}}{\aset \cup \set{\var}}{\sset}{\app} \val$,
    where $\tm = \var, \appflag = \app$ and $\tm' = \val$
  \item[(b)] 
    $\judg[m, e]{\tctx; \var : \optmtyp}{\val}{\mtyptwo}$
  \item[(c)] 
    $\isAppr{\aset \cup \asettwo \cup \set{\var}}{\tctx; \var : \optmtyp}$
  \item[(d)] 
    Since $\appflag = \app$, so either $\mtyptwo = \tightN$ or $\mtyptwo$ is of the form $\mset{\typ}$.
  \item[(e)]
    $\val \in \HAbs{\aset_0 \cup \asettwo}$    
  \end{enumerate}
  Recall that $\var \notin \fv{\val}$ by the grammar of rule names. We analyze by cases the form of $\val$:  
  \begin{enumerate}
  \item $\val = \vartwo$. Then $\var \neq \vartwo$.
    The judgement of condition (b) can be derived only by the rule \ruleTypVar, 
    so it is of the form
    $\judg[0, \numarrows{\mtyptwo}]{\vartwo : \mtyptwo}{\vartwo}{\mtyptwo}$,
    with $\tctx = \vartwo : \mtyptwo, m = 0$ and $e = \numarrows{\mtyptwo}$.
    Moreover, $\optmtyp$ must be $\none$.
    And $\vartwo \in \aset_0 \cup \asettwo$,
    since the judgement of condition (e) can only be derived by the rule $\ruleHAbsVar$.
    Then
    $(\vartwo : \mtyptwo)(\vartwo) \neq \tightN$
    by condition (c), as 
    $\aset_0 \cup \asettwo \cup \set{\var} \subseteq \aset \cup \asettwo \cup \set{\var}$.
    Therefore $\mtyptwo = \mset{\typ}$ for some type $\typ$, and thus $\numarrows{\mset{\typ}} = 1$ necessarily holds.
    Taking 
    $\tctx_\tm = \emptyctx, 
     \tctx_\val = \vartwo : \mset{\typ},
     \typ,
     m_\tm = 0, 
     e_\tm = 0, 
     m_\val = 0, 
     e_\val = 1$
    the following statements hold:
    \begin{itemize}
    \item 
      $\tctx = \vartwo : \mset{\typ} = \tctx_\tm + \tctx_\val$ 
      and $m = 0 = m_\tm + m_\val$ 
      and $e = 1 = e_\tm + e_\val$
    \item 
      $\judg[0, 1]{\var : \mset{\typ}}{\var}{\mset{\typ}}$, by rule \ruleTypVar
    \item 
      $\judg[0, 1]{\vartwo : \mset{\typ}}{\vartwo}{\mset{\typ}}$, by condition (b).
    \end{itemize}
  \item $\val = \lam{\vartwo}{\tmtwo}$.
    The judgement of condition (b) can be derived only by the rule \ruleTypAbs.
    Moreover, we can derive the judgement
    $\lam{\vartwo}{\tmtwo} \in \HAbs{\aset \cup \asettwo \cup \set{\var}}$
    by rule \ruleHAbsLam,
    and along with condition (c) we conclude
    $\mtyptwo \neq \tightN$ by \cref{lem:types_of_hereditary_abstractions}.
    Then the only possible case is when $\mtyptwo$ is of the form $\mset{\typ}$
    for some $\typ = \optmtypthree \to \mtypfour$.
    Since $\var \notin \fv{\lam{\vartwo}{\tmtwo}}$
    it must be the case that
    $\var \notin \dom{\tctx; \var : \optmtyp}$,
    by \cref{lem:relevance},
    that is, $\optmtyp = \none$.
    Then the following derivation is for the judgement of condition (b):
    \[
      \indrule{\ruleTypAbs}{
        \judg[m, e]{\tctx; \vartwo : \optmtypthree}{\tmtwo}{\mtypfour}
      }{
        \judg[m, e]{\tctx}{\lam{\vartwo}{\tmtwo}}{\mset{\optmtypthree \to \mtypfour}}
      }
    \]
    Taking 
    $\tctx_\tm = \emptyctx, 
     \tctx_\val = \tctx,
     \typ = \optmtypthree \to \mtypfour,
     m_\tm = 0, 
     e_\tm = 0, 
     m_\val = m, 
     e_\val = e$
    the following statements hold:
    \begin{itemize}
    \item 
      $\tctx = \tctx_\tm + \tctx_\val$ 
      and $m = m_\tm + m_\val$ 
      and $e = e_\tm + e_\val$
    \item 
      $\judg[0, 1]{\var : \mset{\optmtypthree \to \mtypfour}}{\var}{\mset{\optmtypthree \to \mtypfour}}$, 
      by rule \ruleTypVar
    \item 
      $\judg[0, 1]{\tctx_\val}{\lam{\vartwo}{\tmtwo}}{\mset{\optmtypthree \to \mtypfour}}$, 
      by condition (b).
    \end{itemize}
  \end{enumerate}
\item \ruleUAppL.
  Then the following conditions hold:
  \begin{enumerate}
  \item[(a)]
    \[
      \indrule{\ruleUAppL}{
        \tmtwo \tov{\rulesub{\var}{\val}}{\aset \cup \set{\var}}{\sset}{\app} \tmtwo'
      }{
        \tmtwo \, \tmthree \tov{\rulesub{\var}{\val}}{\aset \cup \set{\var}}{\sset}{\appflag} \tmtwo' \, \tmthree
      }
    \]
    where $\tm = \tmtwo \, \tmthree$ and $\tm' = \tmtwo' \, \tmthree$
  \item[(b)] 
    $\judg[m, e]{\tctx; \var : \optmtyp}{\tmtwo' \, \tmthree}{\mtyptwo}$
  \item[(c)] 
    $\isAppr{\aset \cup \asettwo \cup \set{\var}}{\tctx; \var : \optmtyp}$
  \item[(d)] 
    If $\appflag = \app$ then either $\mtyptwo = \tightN$ or $\mtyptwo$ is of the form $\mset{\typ}$.
  \item[(e)] 
    $\val \in \HAbs{\aset_0 \cup \asettwo}$
  \end{enumerate}
  Since $\inv{\aset \cup \set{\var}}{\sset}{\tmtwo \, \tmthree}$
  then in particular $\inv{\aset \cup \set{\var}}{\sset}{\tmtwo}$.
  The judgement of condition (b) can be derived either by rule $\ruleTypAppP$ or by rule $\ruleTypAppC$:
  \begin{enumerate}
  \item \ruleTypAppP. Then
    \[
      \indrule{\ruleTypAppP}{
        \judg[m_1, e_1]{\tctx_1; \var : \optmtyp_1}{\tmtwo'}{\tightN}
        \sep
        \judg[m_2, e_2]{\tctx_2; \var : \optmtyp_2}{\tmthree}{\tight}
      }{
        \judg[m_1 + m_2, e_1 + e_2]{\tctx_1 + \tctx_2; \var : \optmtyp_1 + \optmtyp_2}{\tmtwo' \, \tmthree}{\tightN}
      }
    \]
    where 
    $\tctx = \tctx_1 + \tctx_2, 
     \optmtyp = \optmtyp_1 + \optmtyp_2, 
     m = m_1 + m_2, e = e_1 + e_2$ 
    and $\mtyptwo = \tightN$.
    The following conditions hold:
    \begin{enumerate}
    \item[(a')]
      $\tmtwo \tov{\rulesub{\var}{\val}}{\aset \cup \set{\var}}{\sset}{\app} \tmtwo'$, by condition (a)
    \item[(b')]
      $\judg[m_1, e_1]{\tctx_1; \var : \optmtyp_1}{\tmtwo'}{\tightN}$, by condition (b)
    \item[(c')]
      $\isAppr{\aset \cup \asettwo \cup \set{\var}}{\tctx_1; \var : \optmtyp_1}$, 
      since $\tctx = \tctx_1 + \tctx_2$ and $\optmtyp = \optmtyp_1 + \optmtyp_2$,
      and by condition (c)
    \item[(d')]
      Here $\appflag = \app$ and the term is typed with $\tightN$.
    \item[(e')]
      $\val \in \HAbs{\aset_0 \cup \asettwo}$, by condition (e)
    \end{enumerate}
    Then we can apply \ih on $\tmtwo'$, yielding 
    $\tctx_\tmtwo$, $\tctx_\val$, $\typ$, $m_\tmtwo$, $e_\tmtwo$, $m_\val$, $e_\val$
    such that
    \begin{itemize}
    \item[$\circ$]
      $\tctx_1 = \tctx_\tmtwo + \tctx_\val$ and $m_1 = m_\tmtwo + m_\val$ and $e_1 = e_\tmtwo + e_\val$
    \item[$\circ$]
      $\judg[m_\tmtwo,e_\tmtwo + 1]{\tctx_\tmtwo;\var:\optmtyp_1 + \mset{\typ}}{\tmtwo}{\tightN}$
    \item[$\circ$]
      $\judg[m_\val,e_\val]{\tctx_\val}{\val}{\mset{\typ}}$
    \end{itemize}
    Taking 
    $\tctx_\tm = \tctx_\tmtwo + \tctx_2,\tctx_\val,\typ,m_\tm = m_\tmtwo + m_2,e_\tm = e_\tmtwo + e_2,m_\val,e_\val$
    the following statements hold:
    \begin{itemize}
    \item 
      $\tctx = \tctx_1 + \tctx_2 = \tctx_\tmtwo + \tctx_\val + \tctx_2 = \tctx_\tm + \tctx_\val$ 
      and $m = m_1 + m_2 = m_\tmtwo + m_\val + m_2 = m_\tm + m_\val$ 
      and $e = e_1 + e_2 = e_\tmtwo + e_\val + e_2 = e_\tm + e_\val$
    \item
      \[
        \indrule{\ruleTypAppP}{
          \judg[m_\tmtwo,e_\tmtwo + 1]{\tctx_\tmtwo;\var:\optmtyp_1 + \mset{\typ}}{\tmtwo}{\tightN}
          \sep
          \judg[m_2, e_2]{\tctx_2; \var : \optmtyp_2}{\tmthree}{\tight}
        }{
          \judg[m_\tmtwo + m_2, e_\tmtwo + 1 + e_2]{\tctx_\tmtwo + \tctx_2; \var : (\optmtyp_1 + \mset{\typ} + \optmtyp_2)}{\tmtwo \, \tmthree}{\tightN}
        }
      \]
    \item 
      $\judg[m_\val,e_\val]{\tctx_\val}{\val}{\mset{\typ}}$
    \end{itemize}
  \item \ruleTypAppC. Then
    \[
      \indrule{\ruleTypAppC}{
        \judg[m_1, e_1]{\tctx_1; \var : \optmtyp_1}{\tmtwo'}{\mset{\optmtypthree \to \mtyptwo}}
        \sep
        \optmtypthree \mleq \mtypthree
        \sep
        \judg[m_2, e_2]{\tctx_2; \var : \optmtyp_2}{\tmthree}{\mtypthree}
      }{
        \judg[1 + m_1 + m_2, e_1 + e_2]{\tctx_1 + \tctx_2; \var : \optmtyp_1 + \optmtyp_2}{\tmtwo' \, \tmthree}{\mtyptwo}
      }
    \]
    where $\tctx = \tctx_1 + \tctx_2, \optmtyp = \optmtyp_1 + \optmtyp_2, m = 1 + m_1 + m_2$ and $e = e_1 + e_2$.
    The following conditions hold:
    \begin{enumerate}
    \item[(a')] 
      $\tmtwo \tov{\rulesub{\var}{\val}}{\aset \cup \set{\var}}{\sset}{\app} \tmtwo'$,
      by condition (a)
    \item[(b')] 
      $\judg[m_1, e_1]{\tctx_1; \var : \optmtyp_1}{\tmtwo'}{\mset{\optmtypthree \to \mtyptwo}}$, 
      by condition (b)
    \item[(c')] 
      $\isAppr{\aset \cup \asettwo \cup \set{\var}}{\tctx_1; \var : \optmtyp_1}$, 
      since $\tctx = \tctx_1 + \tctx_2$ and $\optmtyp = \optmtyp_1 + \optmtyp_2$,
      and by condition (c)
    \item[(d')] 
      Here $\appflag = \app$ and the term is the singleton $\mset{\optmtypthree \to \mtyptwo}$.
    \item[(e')] 
      $\val \in \HAbs{\aset_0 \cup \asettwo}$, by condition (e)
    \end{enumerate}
    Then we can apply \ih on $\tmtwo'$, yielding 
    $\tctx_\tmtwo$, $\tctx_\val$, $\typ$, $m_\tmtwo$, $e_\tmtwo$, $m_\val$, $e_\val$
    such that
    \begin{itemize}
    \item[$\circ$]
      $\tctx_1 = \tctx_\tmtwo + \tctx_\val$ and $m_1 = m_\tmtwo + m_\val$ and $e_1 = e_\tmtwo + e_\val$
    \item[$\circ$]
      $\judg[m_\tmtwo,e_\tmtwo + 1]{\tctx_\tmtwo;\var:\optmtyp_1 + \mset{\typ}}{\tmtwo}{\mset{\optmtypthree \to \mtyptwo}}$
    \item[$\circ$]
      $\judg[m_\val,e_\val]{\tctx_\val}{\val}{\mset{\typ}}$
    \end{itemize}
    Taking 
    $\tctx_\tm = \tctx_\tmtwo + \tctx_2,\tctx_\val,\typ,m_\tm = 1 + m_\tmtwo + m_2,e_\tm = e_\tmtwo + e_2,m_\val,e_\val$
    the following statements hold:
    \begin{itemize}
    \item 
      $\tctx = \tctx_1 + \tctx_2 = \tctx_\tmtwo + \tctx_\val + \tctx_2 = \tctx_\tm + \tctx_\val$ 
      and $m = 1 + m_1 + m_2 = 1 + m_\tmtwo + m_\val + m_2 = m_\tm + m_\val$ 
      and $e = e_1 + e_2 = e_\tmtwo + e_\val + e_2 = e_\tm + e_\val$
    \item
      \[
        \indrule{\ruleTypAppC}{
          \judg[m_\tmtwo,e_\tmtwo + 1]{\tctx_\tmtwo;\var:\optmtyp_1 + \mset{\typ}}{\tmtwo}{\mset{\optmtypthree \to \mtyptwo}}
          \sep
          \optmtypthree \mleq \mtypthree
          \sep
          \judg[m_2, e_2]{\tctx_2; \var : \optmtyp_2}{\tmthree}{\mtypthree}
        }{
          \judg[1 + m_\tmtwo + m_2, e_\tmtwo + 1 + e_2]{\tctx_\tmtwo + \tctx_2; \var : (\optmtyp_1 + \mset{\typ} + \optmtyp_2)}{\tmtwo \, \tmthree}{\mtyptwo}
        }
      \]
    \item 
      $\judg[m_\val,e_\val]{\tctx_\val}{\val}{\mset{\typ}}$
    \end{itemize}
  \end{enumerate}
\HIDDENFRAGMENT{
  \item The remaining cases are similar.
}{
\item \ruleUAppR. Then the following conditions hold:
  \begin{enumerate}
  \item[(a)]
    \[
      \indrule{\ruleUAppR}{
        \tmtwo \in \Struct{\sset}
        \sep
        \tmthree \tov{\rulesub{\var}{\val}}{\aset \cup \set{\var}}{\sset}{\nonapp} \tmthree'
      }{
        \tmtwo \, \tmthree \tov{\rulesub{\var}{\val}}{\aset \cup \set{\var}}{\sset}{\appflag} \tmtwo \, \tmthree'
      }
    \]
    where $\tm = \tmtwo \, \tmthree$ and $\tm' = \tmtwo \, \tmthree'$
  \item[(b)] 
    $\judg[m, e]{\tctx; \var : \optmtyp}{\tmtwo \, \tmthree'}{\mtyptwo}$
  \item[(c)] 
    $\isAppr{\aset \cup \asettwo \cup \set{\var}}{\tctx; \var : \optmtyp}$
  \item[(d)] 
    If $\appflag = \app$ then 
    either $\mtyptwo = \tightN$ or $\mtyptwo$ is of the form $\mset{\typ}$.
  \item[(e)] 
    $\val \in \HAbs{\aset_0 \cup \asettwo}$
  \end{enumerate}
  Since $\inv{\aset \cup \set{\var}}{\sset}{\tmtwo \, \tmthree}$
  then in particular $\inv{\aset \cup \set{\var}}{\sset}{\tmthree}$.
  The judgement of condition (b) can be derived 
  either by rule $\ruleTypAppP$ or by rule \ruleTypAppC:
  \begin{enumerate}
  \item \ruleTypAppP. Then
    \[
      \indrule{\ruleTypAppP}{
        \judg[m_1, e_1]{\tctx_1; \var : \optmtyp_1}{\tmtwo}{\tightN}
        \sep
        \judg[m_2, e_2]{\tctx_2; \var : \optmtyp_2}{\tmthree'}{\tight}
      }{
        \judg[m_1 + m_2, e_1 + e_2]{\tctx_1 + \tctx_2; \var : (\optmtyp_1 + \optmtyp_2)}{\tmtwo \, \tmthree'}{\tightN}
      }
    \]
    where $\tctx = \tctx_1 + \tctx_2, \optmtyp = \optmtyp_1 + \optmtyp_2, m = m_1 + m_2, e = e_1 + e_2$ and $\mtyptwo = \tightN$.
    The following conditions hold:
    \begin{enumerate}
    \item[(a')] 
      $\tmthree \tov{\rulesub{\var}{\val}}{\aset \cup \set{\var}}{\sset}{\nonapp} \tmthree'$, by condition (a)
    \item[(b')] 
      $\judg[m_2, e_2]{\tctx_2; \var : \optmtyp_2}{\tmthree'}{\tight}$, by condition (b)
    \item[(c')] 
      $\isAppr{\aset \cup \asettwo \cup \set{\var}}{\tctx_2; \var : \optmtyp_2}$, 
      since $\tctx = \tctx_1 + \tctx_2$ and $\optmtyp = \optmtyp_1 + \optmtyp_2$, and by condition (c)
    \item[(d')] 
      Here $\appflag = \nonapp$.
    \item[(e')] 
      $\val \in \HAbs{\aset_0 \cup \asettwo}$, by condition (e)
    \end{enumerate}
    We can apply \ih on $\tmthree'$, yielding 
    $\tctx_\tmthree$, $\tctx_\val$, $\typ$, $m_\tmthree$, $e_\tmthree$, $m_\val$, $e_\val$
    such that
    \begin{itemize}
    \item[$\circ$]
      $\tctx_2 = \tctx_\tmthree + \tctx_\val$
      and $m_2 = m_\tmthree + m_\val$
      and $e_2 = e_\tmthree + e_\val$
    \item[$\circ$]
      $\judg[m_\tmthree, e_\tmthree + 1]{\tctx_\tmthree; \var : \optmtyp_2 + \mset{\typ}}{\tmthree}{\tight}$
    \item[$\circ$]
      $\judg[m_\val, e_\val]{\tctx_\val}{\val}{\mset{\typ}}$
    \end{itemize}
    Taking
    $\tctx_\tm = \tctx_1 + \tctx_\tmthree, 
     \tctx_\val, \typ, 
     m_\tm = m_1 + m_\tmthree, e_\tm = e_1 + e_\tmthree, m_\val, e_\val$
    the following statements hold:
    \begin{itemize}
    \item 
      $\tctx = \tctx_1 + \tctx_2 = \tctx_1 + \tctx_\tmthree + \tctx_\val = \tctx_\tm + \tctx_\val$
      and $m = m_1 + m_2 = m_1 + m_\tmthree + m_\val = m_\tm + m_\val$
      and $e = e_1 + e_2 = e_1 + e_\tmthree + e_\val = e_\tm + e_\val$
    \item
      \[
        \indrule{\ruleTypAppP}{
          \judg[m_1, e_1]{\tctx_1; \var : \optmtyp_1}{\tmtwo}{\tightN}
          \sep
          \judg[m_\tmthree, e_\tmthree + 1]{\tctx_\tmthree; \var : \optmtyp_2 + \mset{\typ}}{\tmthree}{\tight}
        }{
          \judg[m_1 + m_\tmthree, e_1 + e_\tmthree + 1]
            {\tctx_1 + \tctx_\tmthree; \var : (\optmtyp_1 + \optmtyp_2 + \mset{\typ})}
            {\tmtwo \, \tmthree}
            {\tightN}
        }
      \]
    \item 
      $\judg[m_\val, e_\val]{\tctx_\val}{\val}{\mset{\typ}}$
    \end{itemize}
  \item \ruleTypAppC. Then
    \[
      \indrule{\ruleTypAppC}{
        \judg[m_1, e_1]{\tctx_1; \var : \optmtyp_1}{\tmtwo}{\mset{\optmtypthree \to \mtyptwo}}
        \sep
        \optmtypthree \mleq \mtypthree
        \sep
        \judg[m_2, e_2]{\tctx_2; \var : \optmtyp_2}{\tmthree'}{\mtypthree}
      }{
        \judg[1 + m_1 + m_2, e_1 + e_2]{\tctx_1 + \tctx_2; \var : (\optmtyp_1 + \optmtyp_2)}{\tmtwo \, \tmthree'}{\mtyptwo}
      }
    \]
    where $\tctx = \tctx_1 + \tctx_2, \optmtyp = \optmtyp_1 + \optmtyp_2, m = 1 + m_1 + m_2$ and $e = e_1 + e_2$.
    The following conditions hold:
    \begin{enumerate}
    \item[(a')] 
      $\tmthree \tov{\rulesub{\var}{\val}}{\aset \cup \set{\var}}{\sset}{\nonapp} \tmthree'$, by condition (a)
    \item[(b')] 
      $\judg[m_2, e_2]{\tctx_2; \var : \optmtyp_2}{\tmthree'}{\mtypthree}$, by condition (b)
    \item[(c')] 
      $\isAppr{\aset \cup \asettwo \cup \set{\var}}{\tctx_2; \var : \optmtyp_2}$, 
      since $\tctx = \tctx_1 + \tctx_2$ and $\optmtyp = \optmtyp_1 + \optmtyp_2$, and by condition (c)
    \item[(d')] 
      Here $\appflag = \nonapp$.
    \item[(e')] 
      $\val \in \HAbs{\aset_0 \cup \asettwo}$, by condition (e)
    \end{enumerate}
    We can apply \ih on $\tmthree'$, yielding 
    $\tctx_\tmthree$, $\tctx_\val$, $\typ$, $m_\tmthree$, $e_\tmthree$, $m_\val$, $e_\val$
    such that
    \begin{itemize}
    \item[$\circ$]
      $\tctx_2 = \tctx_\tmthree + \tctx_\val$
      and $m_2 = m_\tmthree + m_\val$
      and $e_2 = e_\tmthree + e_\val$
    \item[$\circ$]
      $\judg[m_\tmthree, e_\tmthree + 1]{\tctx_\tmthree; \var : \optmtyp_2 + \mset{\typ}}{\tmthree}{\mtypthree}$
    \item[$\circ$]
      $\judg[m_\val, e_\val]{\tctx_\val}{\val}{\mset{\typ}}$
    \end{itemize}
    Taking
    $\tctx_\tm = \tctx_1 + \tctx_\tmthree, \tctx_\val, \typ, m_\tm = 1 + m_1 + m_\tmthree, e_\tm = e_1 + e_\tmthree, m_\val, e_\val$
    the following statements holds:
    \begin{itemize}
    \item 
      $\tctx = \tctx_1 + \tctx_2 = \tctx_1 + \tctx_\tmthree + \tctx_\val = \tctx_\tm + \tctx_\val$
      and $m = 1 + m_1 + m_2 = 1 + m_1 + m_\tmthree + m_\val = m_\tm + m_\val$
      and $e = e_1 + e_2 = e_1 + e_\tmthree + e_\val = e_\tm + e_\val$
    \item
      \[
        \indrule{\ruleTypAppC}{
          \judg[m_1, e_1]{\tctx_1; \var : \optmtyp_1}{\tmtwo}{\mset{\optmtypthree \to \mtyptwo}}
          \sep
          \optmtypthree \mleq \mtypthree
          \sep
          \judg[m_\tmthree, e_\tmthree + 1]{\tctx_\tmthree; \var : \optmtyp_2 + \mset{\typ}}{\tmthree}{\mtypthree}
        }{
          \judg[1 + m_1 + m_\tmthree, e_1 + e_\tmthree + 1]
            {\tctx_1 + \tctx_\tmthree; \var : (\optmtyp_1 + \optmtyp_2 + \mset{\typ})}
            {\tmtwo \, \tmthree}
            {\mtyptwo}
        }
      \]
    \item 
      $\judg[m_\val, e_\val]{\tctx_\val}{\val}{\mset{\typ}}$
    \end{itemize}
  \end{enumerate}
\item \ruleUEsR. The following conditions hold:
  \begin{enumerate}
  \item[(a)]
    \[
      \indrule{\ruleUEsR}{
        \tmthree 
        \tov{\rulesub{\var}{\val}}{\aset \cup \set{\var}}{\sset}{\nonapp} 
        \tmthree'
      }{
        \tmtwo\esub{\vartwo}{\tmthree} 
        \tov{\rulesub{\var}{\val}}{\aset \cup \set{\var}}{\sset}{\appflag} 
        \tmtwo\esub{\vartwo}{\tmthree'}
      }
    \]
    where $\tm = \tmtwo\esub{\vartwo}{\tmthree}$ and $\tm' = \tmtwo\esub{\vartwo}{\tmthree'}$
  \item[(b)] 
    $\judg[m, e]{\tctx; \var : \optmtyp}{\tmtwo\esub{\var}{\tmthree'}}{\mtyptwo}$
  \item[(c)] 
    $\isAppr{\aset \cup \asettwo \cup \set{\var}}{\tctx; \var : \optmtyp}$
  \item[(d)] 
    If $\appflag = \app$ then either $\mtyptwo = \tightN$ or $\mtyptwo$ is of the form $\mset{\typ}$.
  \item[(e)] 
    $\val \in \HAbs{\aset_0 \cup \asettwo}$
  \end{enumerate}
  The judgement of condition (b) can be derived only by rule \ruleTypES, so
  \[
    \indrule{\ruleTypES}{
      \judg[m_1, e_1]
        {\tctx_1; \var : \optmtyp_1; \vartwo : \optmtypthree}
        {\tmtwo}
        {\mtyptwo}
      \sep
      \optmtypthree \mleq \mtypthree
      \sep
      \judg[m_2, e_2]
        {\tctx_2; \var : \optmtyp_2}
        {\tmthree'}
        {\mtypthree}
    }{
      \judg[m_1 + m_2, e_1 + e_2]{\tctx_1 + \tctx_2; \var : (\optmtyp_1 + \optmtyp_2)}{\tmtwo\esub{\var}{\tmthree'}}{\mtyptwo}
    }
  \]
  where $\tctx = \tctx_1 + \tctx_2, \optmtyp = \optmtyp_1 + \optmtyp_2; m = m_1 + m_2$ and $e = e_1 + e_2$.
  Since $\inv{\aset \cup \set{\var}}{\sset}{\tmtwo\esub{\vartwo}{\tmthree}}$
  then in particular $\inv{\aset \cup \set{\var}}{\sset}{\tmthree}$.
  The following conditions hold:
  \begin{enumerate}
  \item[(a')] 
    $\tmthree \tov{\rulesub{\var}{\val}}{\aset \cup \set{\var}}{\sset}{\nonapp} \tmthree'$, by condition (a)
  \item[(b')] 
    $\judg[m_2, e_2]{\tctx_2; \var : \optmtyp_2}{\tmthree'}{\mtypthree}$, by condition (b)
  \item[(c')] 
    $\isAppr{\aset \cup \asettwo \cup \set{\var}}{\tctx_2; \var : \optmtyp_2}$, 
    since $\tctx = \tctx_1 + \tctx_2$ and $\optmtyp = \optmtyp_1 + \optmtyp_2$, and by condition (c)
  \item[(d')] 
    Here $\appflag = \nonapp$.
  \item[(e')] 
    $\val \in \HAbs{\aset_0 \cup \asettwo}$, by condition (e)  
  \end{enumerate}
  We can apply \ih on $\tmthree'$, yielding 
  $\tctx_\tmthree$, $\tctx_\val$, $\typ$, $m_\tmthree$, $e_\tmthree$, $m_\val$, $e_\val$
  such that
  \begin{itemize}
  \item[$\circ$]
    $\tctx_2 = \tctx_\tmthree + \tctx_\val$
    and $m_2 = m_\tmthree + m_\val$
    and $e_2 = e_\tmthree + e_\val$
  \item[$\circ$]
    $\judg[m_\tmthree, e_\tmthree + 1]
      {\tctx_\tmthree; \var : \optmtyp_2 + \mset{\typ}}
      {\tmthree}
      {\mtypthree}$
  \item[$\circ$]
    $\judg[m_\val, e_\val]{\tctx_\val}{\val}{\mset{\typ}}$
  \end{itemize}
  Taking
  $\tctx_\tm = \tctx_1 + \tctx_\tmthree, \tctx_\val, \typ, m_\tm = m_1 + m_\tmthree, e_\tm = e_1 + e_\tmthree, m_\val, e_\val$
  the following statements hold:
  \begin{itemize}
  \item 
    $\tctx = \tctx_1 + \tctx_2 = \tctx_1 + \tctx_\tmthree + \tctx_\val = \tctx_\tm + \tctx_\val$
    and $m = m_1 + m_2 = m_1 + m_\tmthree + m_\val = m_\tm + m_\val$
    and $e = e_1 + e_2 = e_1 + e_\tmthree + e_\val = e_\tm + e_\val$
  \item
    \[
      \indrule{\ruleTypES}{
        \judg[m_1, e_1]
          {\tctx_1; \var : \optmtyp_1; \vartwo : \optmtypthree}
          {\tmtwo}
          {\mtyptwo}
        \sep
        \optmtypthree \mleq \mtypthree
        \sep
        \judg[m_\tmthree, e_\tmthree + 1]
          {\tctx_\tmthree; \var : \optmtyp_2 + \mset{\typ}}
          {\tmthree}
          {\mtypthree}
      }{
        \judg[m_1 + m_\tmthree, e_1 + e_\tmthree + 1]
          {\tctx_1 + \tctx_\tmthree; \var : (\optmtyp_1 + \optmtyp_2 + \mset{\typ})}
          {\tmtwo\esub{\var}{\tmthree}}
          {\mtyptwo}
      }
    \]
  \item 
    $\judg[m_\val, e_\val]{\tctx_\val}{\val}{\mset{\typ}}$
  \end{itemize}
\item \ruleUEsLAbs. The following conditions hold:
  \begin{enumerate}
  \item[(a)]
    \[
      \indrule{\ruleUEsLAbs}{
        \tmtwo \tov{\rulesub{\var}{\val}}{(\aset \cup \set{\vartwo}) \cup \set{\var}}{\sset}{\appflag} \tmtwo'
        \sep
        \tmthree \in \HAbs{\aset \cup \set{\var}}
        \sep
        \vartwo \notin \aset \cup \set{\var} \cup \sset
        \sep
        \vartwo \notin \fv{\rulesub{\var}{\val}}
      }{
        \tmtwo\esub{\vartwo}{\tmthree} \tov{\rulesub{\var}{\val}}{\aset \cup \set{\var}}{\sset}{\appflag} \tmtwo'\esub{\vartwo}{\tmthree}
      }
    \]
    where $\tm = \tmtwo\esub{\vartwo}{\tmthree}$ and $\tm' = \tmtwo'\esub{\vartwo}{\tmthree}$
  \item[(b)] 
    $\judg[m, e]{\tctx; \var : \optmtyp}{\tmtwo'\esub{\vartwo}{\tmthree}}{\mtyptwo}$
  \item[(c)] 
    $\isAppr{\aset \cup \asettwo \cup \set{\var}}{\tctx; \var : \optmtyp}$
  \item[(d)] 
    If $\appflag = \app$ then either $\mtyptwo = \tightN$ or $\mtyptwo$ is of the form $\mset{\typ}$.
  \item[(e)] 
    $\val \in \HAbs{\aset_0 \cup \asettwo}$
  \end{enumerate}
  The judgement of condition (b) can be derived only by the rule $\ruleTypES$, so
  \[
    \indrule{\ruleTypES}{
      \judg[m_1, e_1]
        {\tctx_1; \var : \optmtyp_1; \vartwo : \optmtypthree}
        {\tmtwo'}
        {\mtyptwo}
      \sep
      \optmtypthree \mleq \mtypthree
      \sep
      \judg[m_2, e_2]{\tctx_2; \var : \optmtyp_2}{\tmthree}{\mtypthree}
    }{
      \judg[m_1 + m_2, e_1 + e_2]
        {\tctx_1 + \tctx_2; \var : \optmtyp_1 + \optmtyp_2}
        {\tmtwo'\esub{\vartwo}{\tmthree}}
        {\mtyptwo}
    }
  \]
  where $\tctx = \tctx_1 + \tctx_2, \optmtyp = \optmtyp_1 + \optmtyp_2, m = m_1 + m_2$ and $e = e_1 + e_2$.
  Since $\inv{\aset \cup \set{\var}}{\sset}{\tmtwo\esub{\vartwo}{\tmthree}}$
  then in particular $\inv{\aset \cup \set{\var} \cup \set{\vartwo}}{\sset}{\tmtwo}$.
  The following conditions hold:
  \begin{enumerate}
  \item[(a')] 
    $\tmtwo \tov{\rulesub{\var}{\val}}{(\aset \cup \set{\vartwo}) \cup \set{\var}}{\sset}{\appflag} \tmtwo'$, 
    by condition (a)
  \item[(b')] 
    $\judg[m_1, e_1]{\tctx_1; \var : \optmtyp_1; \vartwo : \optmtypthree}{\tmtwo'}{\mtyptwo}$, 
    by condition (b)
  \item[(c')] 
    $\isAppr{\aset \cup \asettwo \cup \set{\var} \cup \set{\vartwo}}{\tctx_1; \var : \optmtyp_1; \vartwo : \optmtypthree}$, 
    given that
    \begin{enumerate}
    \item[(1)]
      $\isAppr
        {\aset \cup \asettwo \cup \set{\var} \cup \set{\vartwo}}
        {\tctx_1; \var : \optmtyp_1}$:
      since 
      $\isAppr{\aset \cup \asettwo \cup \set{\var}}{\tctx; \var : \optmtyp}$
      by condition (c), where $\tctx = \tctx_1 + \tctx_2$, $\optmtyp = \optmtyp_1 + \optmtyp_2$,
      and \cref{rem:isAppropriate}
    \item[(2)]
      $\isAppr
        {\aset \cup \asettwo \cup \set{\var} \cup \set{\vartwo}}
        {\vartwo : \optmtypthree}$:
      by condition (a) $\tmthree \in \HAbs{\aset \cup \set{\var}}$.
      Then $\aset \subseteq \aset \cup \asettwo$ 
      implies $\tmthree \in \HAbs{\aset \cup \asettwo \cup \set{\var}}$
      by \cref{rem:habs_st}.
      And also, $\isAppr{\aset \cup \asettwo \cup \set{\var}}{\tctx_2; \var : \optmtyp_2}$ holds
      by condition (c) as $\tctx = \tctx_1 + \tctx_2$ and $\optmtyp = \optmtyp_1 + \optmtyp_2$. 
      Then $\mtypthree \neq \tightN$
      by \cref{lem:types_of_hereditary_abstractions}. 
      And $\optmtypthree \mleq \mtypthree$ implies $\optmtypthree \neq \tightN$,
      so we are done.
    \end{enumerate}
  \item[(d')] 
    If $\appflag = \app$ then either $\mtyptwo = \tightN$ or $\mtyptwo$ is of the form $\mset{\typ}$, by condition (d).
  \item[(e')] 
      Since  $\val \in \HAbs{\aset_0 \cup \asettwo}$
      by condition (e), and
      $\aset_0 \cup \asettwo \subseteq (\aset_0 \cup \asettwo) \cup \set{\vartwo}$,
      then 
      $\val \in \HAbs{\aset_0 \cup \asettwo \cup \set{\vartwo}}$
      by \cref{rem:habs_st}
  \end{enumerate}
  We can apply \ih on $\tmtwo'$, yielding
  $\tctxtwo_\tmtwo$, $\tctxtwo_\val$, $\typ$, $m_\tmtwo$, $e_\tmtwo$, $m_\val$, $e_\val$
  such that
  \begin{itemize}
  \item[$\circ$]
    $\tctx_1; \vartwo : \optmtypthree = \tctxtwo_\tmtwo + \tctxtwo_\val$ and $m_1 = m_\tmtwo + m_\val$ and $e_1 = e_\tmtwo + e_\val$
  \item[$\circ$]
    $\judg[m_\tmtwo,e_\tmtwo + 1]{\tctxtwo_\tmtwo;\var:\optmtyp_1 + \mset{\typ}}{\tmtwo}{\mtyptwo}$
  \item[$\circ$]
    $\judg[m_\val,e_\val]{\tctxtwo_\val}{\val}{\mset{\typ}}$
  \end{itemize}
  We can write $\optmtypthree$ as $\optmtypthree_\tmtwo + \optmtypthree_\val$, so that
  $\tctxtwo_\tmtwo = \tctx_\tmtwo; \vartwo : \optmtypthree_\tmtwo$ and
  $\tctxtwo_\val = \tctx_\val; \vartwo : \optmtypthree_\val$.
  We then have $\tctx_1 = \tctx_\tmtwo + \tctx_\val$.
  Moreover, $\vartwo \notin \fv{\val}$ by condition (a), then $\vartwo \notin \dom{\tctxtwo_\val}$ by \cref{lem:relevance} 
  applied to $\judg[m_\val,e_\val]{\tctxtwo_\val}{\val}{\mset{\typ}}$,
  implying that $\optmtypthree_\val = \none$ and hence $\optmtypthree_\tmtwo = \optmtypthree$.
  Taking 
  $\tctx_\tm = \tctx_\tmtwo + \tctx_2, \tctx_\val = \tctxtwo_\val,\typ, 
   m_\tm = m_\tmtwo + m_2,
   e_\tm = e_\tmtwo + e_2,
   m_\val,e_\val$
  the following statements hold:
  \begin{itemize}
  \item 
    $\tctx = \tctx_1 + \tctx_2 = \tctx_\tmtwo + \tctx_\val + \tctx_2 = \tctx_\tm + \tctx_\val$ 
    and $m = m_1 + m_2 = m_\tmtwo + m_\val + m_2 = m_\tm + m_\val$ 
    and $e = e_1 + e_2 = e_\tmtwo + e_\val + e_2 = e_\tm + e_\val$
  \item
    \[
      \indrule{\ruleTypES}{
        \judg[m_\tmtwo,e_\tmtwo + 1]
          {\tctx_\tmtwo;\var:\optmtyp_1 + \mset{\typ}; \vartwo : \optmtypthree}
          {\tmtwo}
          {\mtyptwo}
        \sep
        \optmtypthree \mleq \mtypthree
        \sep
        \judg[m_2, e_2]{\tctx_2; \var : \optmtyp_2}{\tmthree}{\mtypthree}
      }{
        \judg[m_\tmtwo + m_2, e_\tmtwo + 1 + e_2]
          {\tctx_\tmtwo + \tctx_2; \var : (\optmtyp_1 + \mset{\typ} + \optmtyp_2)}
          {\tmtwo\esub{\vartwo}{\tmthree}}
          {\mtyptwo}
      }
      \]
  \item 
    $\judg[m_\val,e_\val]{\tctxtwo_\val}{\val}{\mset{\typ}}$
  \end{itemize}
\item \ruleUEsLStruct. The following conditions hold:
  \begin{enumerate}
  \item[(a)]
    \[
      \indrule{\ruleUEsLStruct}{
        \tmtwo \tov{\rulesub{\var}{\val}}{\aset \cup \set{\var}}{\sset \cup \set{\vartwo}}{\appflag} \tmtwo'
        \sep
        \tmthree \in \Struct{\sset}
        \sep
        \vartwo \notin \aset \cup \set{\var} \cup \sset
        \sep
        \vartwo \notin \fv{\rulesub{\var}{\val}}
      }{
        \tmtwo\esub{\vartwo}{\tmthree} \tov{\rulesub{\var}{\val}}{\aset \cup \set{\var}}{\sset}{\appflag} \tmtwo'\esub{\vartwo}{\tmthree}
      }
    \]
    where $\tm = \tmtwo\esub{\vartwo}{\tmthree}$ and $\tm' = \tmtwo'\esub{\vartwo}{\tmthree}$
  \item[(b)] 
    $\judg[m, e]{\tctx; \var : \optmtyp}{\tmtwo'\esub{\vartwo}{\tmthree}}{\mtyptwo}$
  \item[(c)] 
    $\isAppr{\aset \cup \asettwo \cup \set{\var}}{\tctx; \var : \optmtyp}$
  \item[(d)] 
    If $\appflag = \app$ then either $\mtyptwo = \tightN$ or $\mtyptwo$ is of the form $\mset{\typ}$.
  \item[(e)] 
      $\val \in \HAbs{\aset_0 \cup \asettwo}$
  \end{enumerate}
  The judgement of condition (b) can be derived only by the rule $\ruleTypES$, so
  \[
    \indrule{\ruleTypES}{
      \judg[m_1, e_1]
        {\tctx_1; \var : \optmtyp_1; \vartwo : \optmtypthree}
        {\tmtwo'}
        {\mtyptwo}
      \sep
      \optmtypthree \mleq \mtypthree
      \sep
      \judg[m_2, e_2]{\tctx_2; \var : \optmtyp_2}{\tmthree}{\mtypthree}
    }{
      \judg[m_1 + m_2, e_1 + e_2]
        {\tctx_1 + \tctx_2; \var : \optmtyp_1 + \optmtyp_2}
        {\tmtwo'\esub{\vartwo}{\tmthree}}
        {\mtyptwo}
    }
  \]
  where $\tctx = \tctx_1 + \tctx_2, \optmtyp = \optmtyp_1 + \optmtyp_2, m = m_1 + m_2$ and $e = e_1 + e_2$.
  Since $\inv{\aset \cup \set{\var}}{\sset}{\tmtwo\esub{\vartwo}{\tmthree}}$
  then in particular $\inv{\aset \cup \set{\var}}{\sset \cup \set{\vartwo}}{\tmtwo}$.
  The following conditions hold:
  \begin{enumerate}
  \item[(a')] 
    $\tmtwo \tov{\rulesub{\var}{\val}}{\aset \cup \set{\var}}{\sset \cup \set{\vartwo}}{\appflag} \tmtwo'$, 
    by condition (a)
  \item[(b')] 
    $\judg[m_1, e_1]{\tctx_1; \var : \optmtyp_1; \vartwo : \optmtypthree}{\tmtwo'}{\mtyptwo}$, 
    by condition (b)
  \item[(c')] 
    $\isAppr
      {\aset \cup \asettwo \cup \set{\var}}
      {\tctx_1; \var : \optmtyp_1; \vartwo : \optmtypthree}$, 
    given that
    \begin{enumerate}
    \item[(1)]
      $\isAppr{\aset \cup \asettwo \cup \set{\var}}{\tctx_1; \var : \optmtyp_1}$:
      since 
      $\isAppr{\aset \cup \asettwo \cup \set{\var}}{\tctx; \var : \optmtyp}$ 
      holds by condition (c) and
      $\tctx = \tctx_1 + \tctx_2$ and $\optmtyp = \optmtyp_1 + \optmtyp_2$
    \item[(2)]
      $\isAppr{\aset \cup \asettwo \cup \set{\var}}{\vartwo : \optmtypthree}$:
      by $\alpha$-conversion we assume 
      $\vartwo \notin (\aset \cup \asettwo) \cup \set{\var}$
      \end{enumerate}
  \item[(d')] 
    If $\appflag = \app$ 
    then either $\mtyptwo = \tightN$ or $\mtyptwo$ is of the form $\mset{\typ}$,
    by condition (d).
  \item[(e')] 
    $\val \in \HAbs{\aset_0 \cup \asettwo}$, by condition (e)
  \end{enumerate}
  We can apply \ih on $\tmtwo'$, yielding 
  $\tctxtwo_\tmtwo$, $\tctxtwo_\val$, $\typ$, $m_\tmtwo$, $e_\tmtwo$, $m_\val$, $e_\val$
  such that
  \begin{itemize}
  \item[$\circ$]
    $\tctx_1; \vartwo : \optmtypthree = \tctxtwo_\tmtwo + \tctxtwo_\val$ and $m_1 = m_\tmtwo + m_\val$ and $e_1 = e_\tmtwo + e_\val$
  \item[$\circ$]
    $\judg[m_\tmtwo,e_\tmtwo + 1]{\tctxtwo_\tmtwo;\var:\optmtyp_1 + \mset{\typ}}{\tmtwo}{\mtyptwo}$
  \item[$\circ$]
    $\judg[m_\val,e_\val]{\tctxtwo_\val}{\val}{\mset{\typ}}$
  \end{itemize}
  We can write $\optmtypthree$ as $\optmtypthree_\tmtwo + \optmtypthree_\val$, so that
  $\tctxtwo_\tmtwo = \tctx_\tmtwo; \vartwo : \optmtypthree_\tmtwo$ and
  $\tctxtwo_\val = \tctx_\val; \vartwo : \optmtypthree_\val$.
  We then have $\tctx_1 = \tctx_\tmtwo + \tctx_\val$.
  Moreover, $\vartwo \notin \fv{\val}$ by condition (a),
  then $\vartwo \notin \dom{\tctxtwo_\val}$ by \cref{lem:relevance} 
  applied to $\judg[m_\val,e_\val]{\tctxtwo_\val}{\val}{\mset{\typ}}$,
  implying that $\optmtypthree_\val = \none$ and 
  hence $\optmtypthree_\tmtwo = \optmtypthree$, and $\tctx_\val = \tctxtwo_\val$.
  Taking 
  $\tctx_\tm = \tctx_\tmtwo + \tctx_2, \tctx_\val = \tctxtwo_\val,\typ, 
   m_\tm = m_\tmtwo + m_2,
   e_\tm = e_\tmtwo + e_2,
   m_\val,e_\val$
  the following statements hold:
  \begin{itemize}
  \item
    $\tctx = \tctx_1 + \tctx_2 = \tctx_\tmtwo + \tctx_\val + \tctx_2 = \tctx_\tm + \tctx_\val$ 
    and $m = m_1 + m_2 = m_\tmtwo + m_\val + m_2 = m_\tm + m_\val$ 
    and $e = e_1 + e_2 = e_\tmtwo + e_\val + e_2 = e_\tm + e_\val$
  \item
    \[
      \indrule{\ruleTypES}{
        \judg[m_\tmtwo,e_\tmtwo + 1]
          {\tctx_\tmtwo;\var:\optmtyp_1 + \mset{\typ}; \vartwo : \optmtypthree}
          {\tmtwo}
          {\mtyptwo}
        \sep
        \optmtypthree \mleq \mtypthree
        \sep
        \judg[m_2, e_2]{\tctx_2; \var : \optmtyp_2}{\tmthree}{\mtypthree}
      }{
        \judg[m_\tmtwo + m_2, e_\tmtwo + 1 + e_2]
          {\tctx_\tmtwo + \tctx_2; \var : (\optmtyp_1 + \mset{\typ} + \optmtyp_2)}
          {\tmtwo\esub{\vartwo}{\tmthree}}
          {\mtyptwo}
      }
    \]
  \item
    $\judg[m_\val,e_\val]{\tctx_\val}{\val}{\mset{\typ}}$,
    by the third item obtained after applying the \ih
  \end{itemize}
}
\end{enumerate}
\end{proof}

\subjectexpansion*
% Label: prop:subject_expansion

\begin{proof}
By induction on the derivation of 
$\tm \tov\rulename\aset\sset\appflag \tm'$.
\begin{enumerate}
\item \ruleUDb.
  The following conditions hold:
  \begin{enumerate}
  \item[(a)]
    $(\lam\var\tmtwo)\sctx \, \tmthree 
     \tov\ruledb\aset\sset\appflag
     \tmtwo\esub\var\tmthree\sctx$
  \item[(b)]
    $\judg[\cm',\ce']\tctx{\tmtwo\esub\var\tmthree\sctx}\mtyp$
  \item[(c)]
    $\isAppr\aset\tctx$
  \item[(d)]
    If $\appflag = \app$, then either $\mtyp = \tightN$
    or $\mtyp$ is a singleton, \ie, of the form $\mset\typ$.
  \end{enumerate}
  Applying \cref{lem:composition} to the judgement in condition (b), 
  there exist $\tctx_{\tmtwo\esub\var\tmthree}$, $\tctx_\sctx$, 
  $\tctxtwo$, $\cm'_{\tmtwo\esub\var\tmthree}$, 
  $\ce'_{\tmtwo\esub\var\tmthree}$, $\cm'_\sctx$, and $\ce'_\sctx$ 
  such that:
  \begin{enumerate}
  \item[1.] 
    $\judgSctx[\cm'_\sctx, \ce'_\sctx]{\tctx_\sctx}\sctx\tctxtwo$
  \item[2.] 
    $\judg[\cm'_{\tmtwo\esub\var\tmthree}, \ce'_{\tmtwo\esub\var\tmthree}]
      {\tctx_{\tmtwo\esub\var\tmthree}; \tctxtwo}
      {\tmtwo\esub\var\tmthree}
      \mtyp$
  \item[3.]
    $\tctx = \tctx_\sctx + \tctx_{\tmtwo\esub\var\tmthree}$ and
    $\cm' = \cm'_\sctx + \cm'_{\tmtwo\esub\var\tmthree}$ and
    $\ce' = \ce'_\sctx + \ce'_{\tmtwo\esub\var\tmthree}$
  \end{enumerate}
  The judgement of statement 2. can be derived only by rule \ruleTypES:
  \[
    \indrule{\ruleTypES}{
      \judg[n_1, f_1]{\tctxthree_1; \tctxtwo_1; \var : \optmtyptwo}{\tmtwo}{\mtyp}
      \sep
      \optmtyptwo \mleq \mtyptwo
      \sep
      \judg[n_2, f_2]{\tctxthree_2; \tctxtwo_2}{\tmthree}{\mtyptwo}
    }{
      \judg[n_1 + n_2, f_1 + f_2]
        {\tctxthree_1 + \tctxthree_2; \tctxtwo_1 + \tctxtwo_2}
        {\tmtwo\esub\var\tmthree}
        \mtyp
    }
  \]
  where $\tctx_{\tmtwo\esub\var\tmthree} = \tctxthree_1 + \tctxthree_2,
  \tctxtwo = \tctxtwo_1 + \tctxtwo_2, 
  m'_{\tmtwo\esub{\var}{\tmthree}} = n_1 + n_2$ and 
  $e'_{\tmtwo\esub{\var}{\tmthree}} = f_1 + f_2$.
  By $\alpha-$conversion we may assume that 
  the bound variables in $(\lam{\var}{\tmtwo})\sctx$ don't occur free in $\tmthree$.
  Since $\dom{\tctxtwo} = \dom{\tctxtwo_1 + \tctxtwo_2} = \dom{\tctxtwo_1} \cup \dom{\tctxtwo_2}$ 
  is a subset of $\domSctx{\sctx}$ by \cref{lem:relevance_sctx},
  we may assume that $\dom{\tctxtwo} \cap \fv{\tmthree} = \emptyset$, 
  so in particular $\dom{\tctxtwo_2} \cap \fv{\tmthree} = \emptyset$.
  We have $\dom{\tctxthree_2} \cup \dom{\tctxtwo_2} = \dom{\tctxthree_2; \tctxtwo_2} \subseteq \fv{\tmthree}$ by \cref{lem:relevance}.
  Therefore $\tctxtwo_2 = \emptyctx$, with $\tctxtwo = \tctxtwo_1$.
  We apply rule $\ruleTypAbs$ to the judgement derivation $\judg[n_1, f_1]{\tctxthree_1; \tctxtwo_1; \var : \optmtyptwo}{\tmtwo}{\mtyp}$,
 yielding $\judg[n_1, f_1]{\tctxthree_1; \tctxtwo}{\lam{\var}{\tmtwo}}{\mset{\optmtyptwo \to \mtyp}}$.
  By \cref{lem:composition}, we compose this judgement with the one in statement 1.,
  yielding $\judg[m'_\sctx + n_1, e'_\sctx + f_1]{\tctx_\sctx + \tctxthree_1}{(\lam{\var}{\tmtwo})\sctx}{\mset{\optmtyptwo \to \mtyp}}$.
  Now we apply rule \ruleTypAppC, taking this judgement and $\judg[n_2, f_2]{\tctxthree_2; \tctxtwo_2}{\tmthree}{\mtyptwo}$ as premises:
  \[
    \indrule{\ruleTypAppC}{
      \judg[m'_\sctx + n_1, e'_\sctx + f_1]{\tctx_\sctx + \tctxthree_1}{(\lam{\var}{\tmtwo})\sctx}{\mset{\optmtyptwo \to \mtyp}}
      \sep
      \optmtyptwo \mleq \mtyptwo
      \sep
      \judg[n_2, f_2]{\tctxthree_2}{\tmthree}{\mtyptwo}
    }{
      \judg[1 + m'_\sctx + n_1 + n_2, e'_\sctx + f_1 + f_2]{\tctx_\sctx + \tctxthree_1 + \tctxthree_2}{(\lam{\var}{\tmtwo})\sctx \, \tmthree}{\mtyp}
    }
  \]
  where $\tctx_\sctx + \tctxthree_1 + \tctxthree_2 = \tctx_\sctx + \tctx_{\tmtwo\esub{\var}{\tmthree}} = \tctx$.
  Moreover
  since $\rulename = \ruledb$, we have that 
  $ (m,e) 
  = (1 + m'_\sctx + n_1 + n_2, e'_\sctx + f_1 + f_2) 
  = (1 + m'_\sctx + m'_{\tmtwo\esub{\var}{\tmthree}}, e'_\sctx + e'_{\tmtwo\esub{\var}{\tmthree}}) 
  = (m'+1,e')$.
\item \ruleULsv.
 The following conditions hold:
  \begin{enumerate}
  \item[(a)]
    \[
      \indrule{\ruleULsv}{
        \tmtwo 
        \tov{\rulesub{\var}{\val}}{\aset \cup \set{\var}}{\sset}{\appflag} 
        \tmtwo'
        \sep
        \var \notin \aset \cup \sset
        \sep
        \val\sctx \in \HAbs{\aset}
      }{
        \tmtwo\esub{\var}{\val\sctx} 
        \tov{\rulelsv}{\aset}{\sset}{\appflag} 
        \tmtwo'\esub{\var}{\val}\sctx
      }
    \]
  \item[(b)]
    $\judg[m',e']{\tctx}{\tmtwo'\esub{\var}{\val}\sctx}{\mtyp}$
  \item[(c)]
    $\isAppr{\aset}{\tctx}$
  \item[(d)]
    If $\appflag = \app$
    then either $\mtyp = \tightN$
    or $\mtyp$ is a singleton, \ie, of the form $\mset{\typ}$.
  \end{enumerate}
  By \cref{lem:composition} there exist 
  $\tctx_{\tmtwo'\esub{\var}{\val}}, \tctx_\sctx, \tctxtwo,
   m'_{\tmtwo'\esub{\var}{\val}}, e'_{\tmtwo'\esub{\var}{\val}}, m'_\sctx$
  and $e'_\sctx$ such that:
  \begin{enumerate}
  \item[1.]
    $\judgSctx[m'_\sctx, e'_\sctx]{\tctx_\sctx}{\sctx}{\tctxtwo}$
  \item[2.]
    $\judg[m'_{\tmtwo'\esub{\var}{\val}}, e'_{\tmtwo'\esub{\var}{\val}}]
      {\tctx_{\tmtwo'\esub{\var}{\val}}; \tctxtwo}
      {\tmtwo'\esub{\var}{\val}}
      {\mtyp}$
  \item[3.]
    $\tctx = \tctx_\sctx + \tctx_{\tmtwo'\esub{\var}{\val}}$, and
    $m' = m'_\sctx + m'_{\tmtwo'\esub{\var}{\val}}$, and
    $e' = e'_\sctx + e'_{\tmtwo'\esub{\var}{\val}}$.
  \end{enumerate}
  Furthermore,
  $\isAppr{\aset}{\tctx}$ by condition (c),
  and note that
  $\inv{\aset}{\sset}{\tmtwo\esub{\var}{\val\sctx}}$
  implies
  $\inv{\aset}{\sset}{\tmtwo'\esub{\var}{\val}\sctx}$,
  since
  $\aset \disj \sset$  and
  $\fv{\tmtwo'\esub{\var}{\val}\sctx} = 
   \fv{\tmtwo\esub{\var}{\val\sctx}} \subseteq \aset \cup \sset$.
  Therefore
  $\isAppr{\expansion{\aset}{\sctx}}{\tctxtwo}$.
  The judgement of statement (b) can only be derived by rule \ruleTypES:
  \[
    \indrule{\ruleTypES}{
      \judg[n_1, f_1]
        {\tctxthree_1; \tctxtwo_1; \var : \optmtyptwo}
        {\tmtwo'}
        {\mtyp}
      \sep
      \optmtyptwo \mleq \mtyptwo
      \sep
      \judg[n_2, f_2]{\tctxthree_2; \tctxtwo_2}{\val}{\mtyptwo}
    }{
      \judg[n_1 + n_2, f_1 + f_2]
        {\tctxthree_1 + \tctxthree_2; \tctxtwo_1 + \tctxtwo_2}
        {\tmtwo'\esub{\var}{\val}}
        {\mtyp}
    }
  \]
  where 
  $\tctx_{\tmtwo'\esub{\var}{\val}} = \tctxthree_1 + \tctxthree_2, 
   \tctxtwo = \tctxtwo_1 + \tctxtwo_2, 
   m'_{\tmtwo'\esub{\var}{\val}} = n_1 + n_2$ and 
  $e'_{\tmtwo'\esub{\var}{\val}} = f_1 + f_2$.
  Moreover, $\inv{\aset}{\sset}{\tmtwo\esub{\var}{\val\sctx}}$ 
  implies $\inv{\aset \cup \set{\var}}{\sset}{\tmtwo}$.
    The following conditions hold:
  \begin{enumerate}
  \item[(a')]
    $\tmtwo \tov{\rulesub{\var}{\val}}{\aset\cup\set{\var}}{\sset}{\appflag} \tmtwo'$,
    by condition (a)
  \item[(b')]
    $\judg[n_1, f_1]
      {(\tctxthree_1; \tctxtwo_1); \var : \optmtyptwo}
      {\tmtwo'}
      {\mtyp}$,
    by premise of the rule in condition (b)
  \item[(c')]
    $\isAppr
      {\expansion{\aset}{\sctx} \cup \set{\var}}
      {(\tctxthree_1; \tctxtwo_1);\var:\optmtyptwo}$ 
    since
    \begin{enumerate}
    \item[(1)]
      $\isAppr{\aset}{\tctxthree_1}$, 
      as $\tctx = \tctx_\sctx + \tctxthree_1 + \tctxthree_2$, and
      $\isAppr{\aset}{\tctx}$ by condition (c).
      We conclude by \cref{rem:isAppropriate} that
      $\isAppr{\expansion{\aset}{\sctx}\cup\set{\var}}{\tctxthree_1}$.
    \item[(2)]
      $\isAppr{\expansion{\aset}{\sctx}}{\tctxtwo_1}$, 
      as $\tctxtwo = \tctxtwo_1 + \tctxtwo_2$,
      and $\isAppr{\expansion{\aset}{\sctx}}{\tctxtwo}$, justified before.
      We conclude by \cref{rem:isAppropriate} that
      $\isAppr{\expansion{\aset}{\sctx}\cup\set{\var}}{\tctxtwo_1}$.
    \item[(3)] $\isAppr{\set{\var}}{\var : \optmtyptwo}$. 
      Indeed,
      $\val \in \HAbs{\expansion{\aset}{\sctx}}$ by
      \cref{lem:tL_hAbs_t_hAbsExp}, and
      $\isAppr{\expansion{\aset}{\sctx}}{\tctxthree_2; \tctxtwo_2}$,
      since 
      (3.1): $\isAppr{\expansion{\aset}{\sctx}}{\tctxthree_2}$, given that
      $\tctx = \tctx_\sctx + \tctxthree_1 + \tctxthree_2$, and
      $\isAppr{\aset}{\tctx}$ by condition (c); and 
      (3.2):
      $\isAppr{\expansion{\aset}{\sctx}}{\tctxtwo_2}$, 
      given $\isAppr{\expansion{\aset}{\sctx}}{\tctxtwo}$ mentioned before and
      $\tctxtwo = \tctxtwo_1 + \tctxtwo_2$.
      Hence we can apply
      \cref{lem:types_of_hereditary_abstractions} yielding
      $\mtyptwo \neq \tightN$. 
      Then $\optmtyptwo \mleq
      \mtyptwo$ implies $\optmtyptwo \neq
      \tightN$.  By \cref{rem:isAppropriate}, we conclude that
      $\isAppr{\expansion{\aset}{\sctx} \cup
        \set{\var}}{\var:\optmtyptwo}$.
    \end{enumerate}
  \item[(d')]
    If $\appflag = \app$
    then either $\mtyp = \tightN$
    or $\mtyp$ is a singleton, \ie, of the form $\mset{\typ}$,
    by condition (d).
  \item[(e')]
    $\val \in \HAbs{\expansion{\aset}{\sctx}}$,
    by \cref{lem:tL_hAbs_t_hAbsExp}
  \end{enumerate}
  Therefore by \cref{lem:anti_substitution}, we have that there exist
  $\tctxfour_\tmtwo, \tctxfour_\val, \typ, n_\tmtwo, f_\tmtwo, n_\val$ and $f_\val$
  such that:
  \begin{enumerate}
  \item[1'.]
    $\tctxthree_1; \tctxtwo_1 = \tctxfour_\tmtwo + \tctxfour_\val$ and
    $n_1 = n_\tmtwo + n_\val$ and
    $f_1 = f_\tmtwo + f_\val$
  \item[2'.]
    $\judg[n_\tmtwo, f_\tmtwo + 1]
      {\tctxfour_\tmtwo; \var : \optmtyptwo + \mset{\typ}}
      {\tmtwo}
      {\mtyp}$
  \item[3'.]
    $\judg[n_\val, f_\val]{\tctxfour_\val}{\val}{\mset{\typ}}$
  \end{enumerate}
  By statement 1'. we can write
  $\tctxfour_\tmtwo$ as $\tctxthree_{11}; \tctxtwo_{11}$ and
  $\tctxfour_\val$ as $\tctxthree_{12}; \tctxtwo_{12}$ with
  $\tctxthree_1 = \tctxthree_{11} + \tctxthree_{12}$ and
  $\tctxtwo_1 = \tctxtwo_{11} + \tctxtwo_{12}$.
  By $\alpha-$conversion
  the bound variables in $\val\sctx$ don't occur free in $\tmtwo$;
  in particular $\domSctx{\sctx} \cap \fv{\tmtwo} = \emptyset$.
  Since $\dom{\tctxtwo} = \dom{\tctxtwo_1 + \tctxtwo_2} = \dom{\tctxtwo_1} \cup \dom{\tctxtwo_2}$
  is a subset of $\domSctx{\sctx}$
  by \cref{lem:relevance_sctx}, then
  $\dom{\tctxtwo} \cap \fv{\tmtwo} = \emptyset$, so in particular
  $\dom{\tctxtwo_2} \cap \fv{\tmtwo} = \emptyset$.
  We have
  $\dom{\tctxfour_\tmtwo; \var : \optmtyptwo + \mset{\typ}} \subseteq \fv{\tmtwo}$
  by \cref{lem:relevance}. 
  Therefore $\tctxtwo_{11} = \emptyset$, with $\tctxtwo_1 = \tctxtwo_{12}$.
  We apply \cref{lem:splitting}, merging the judgement of statement 3'. and
  the judgement $\judg[n_2, f_2]{\tctxthree_2; \tctxtwo_2}{\val}{\mtyptwo}$, 
  so that we obtain
  $\judg[n_\val + n_2, f_\val + f_2]
    {(\tctxthree_{12}; \tctxtwo_1) + (\tctxthree_2; \tctxtwo_2)}
    {\val}
    {\mset{\typ} + \mtyptwo}$.
  We compose this result with the judgement from statement 1, by \cref{lem:composition}, yielding
  $\judg[m'_\sctx + n_\val + n_2, e'_\sctx + f_\val + f_2]
    {\tctx_\sctx + \tctxthree_{12} + \tctxthree_2}
    {\val\sctx}
    {\mset{\typ} + \mtyptwo}$.
  We apply rule \ruleTypES:
  \[
    \indrule{\ruleTypES}{
      \judg[n_\tmtwo, f_\tmtwo + 1]
        {\tctxthree_{11}; \var : \optmtyptwo + \mset{\typ}}
        {\tmtwo}
        {\mtyp}
      \sep
      \optmtyptwo \mleq \mtyptwo
      \sep
      \judg[m'_\sctx + n_\val + n_2, e'_\sctx + f_\val + f_2]
        {\tctx_\sctx + \tctxthree_{12} + \tctxthree_2}
        {\val\sctx}
        {\mset{\typ} + \mtyptwo}
    }{
      \judg[n_\tmtwo + m'_\sctx + n_\val + n_2, f_\tmtwo + 1 + e'_\sctx + f_\val + f_2]
        {\tctxthree_{11} + \tctx_\sctx + \tctxthree_{12} + \tctxthree_2}
        {\tmtwo\esub{\var}{\val\sctx}}
        {\mtyp}
    }
  \]
  where
  $ \tctxthree_{11} + \tctx_\sctx + \tctxthree_{12} + \tctxthree_2 
  = \tctxthree_1 + \tctx_\sctx + \tctxthree_2
  = \tctx_{\tmtwo'\esub{\var}{\val}} + \tctx_\sctx
  = \tctx$.
  Moreover, since $\rulename = \rulelsv$, we have that
  $ (m,e)
  = (n_\tmtwo + m'_\sctx + n_\val + n_2, f_\tmtwo + 1 + e'_\sctx + f_\val + f_2)
  = (m'_\sctx + n_1 + n_2, f_1 + f_2 + e'_\sctx + 1)
  = (m'_\sctx + m'_{\tmtwo'\esub{\var}{\val}}, e'_{\tmtwo'\esub{\var}{\val}} + e'_\sctx + 1)
  = (m', e' + 1)$,
  so we are done.
\HIDDENFRAGMENT{
\item The congruence cases are uninteresting and are omitted here.
}{
\item \ruleUAppL.
  The following conditions hold:
  \begin{enumerate}
  \item[(a)]
    \[
      \indrule{\ruleUAppL}{
        \tmtwo \tov{\rulename}{\aset}{\sset}{\app} \tmtwo'
      }{
        \tmtwo \, \tmthree \tov{\rulename}{\aset}{\sset}{\appflag} \tmtwo' \, \tmthree
      }
    \]
    with $\rulename \in \set{\ruledb, \rulelsv}$
  \item[(b)]
    $\judg[m',e']{\tctx}{\tmtwo' \, \tmthree}{\mtyp}$
  \item[(c)]
    $\isAppr{\aset}{\tctx}$
  \item[(d)]
    If $\appflag = \app$ then either $\mtyp = \tightN$ or $\mtyp$ is of the form $\mset{\typ}$.
  \end{enumerate}
  Note that $\inv{\aset}{\sset}{\tmtwo \, \tmthree}$
  implies in particular $\inv{\aset}{\sset}{\tmtwo}$.
  The judgement of the condition (b) can be derived either by rule $\ruleTypAppP$ or by rule \ruleTypAppC:
\begin{enumerate}
  \item \ruleTypAppP.
    Then 
    \[
      \indrule{\ruleTypAppP}{
        \judg[m'_1,e'_1]{\tctx_1}{\tmtwo'}{\tightN}
        \sep
        \judg[m'_2,e'_2]{\tctx_2}{\tmthree}{\tight}
      }{
        \judg[m'_1 + m'_2, e'_1 + e'_2]{\tctx_1 + \tctx_2}{\tmtwo' \, \tmthree}{\tightN}
      }
    \]
    where $\tctx = \tctx_1 + \tctx_2, m' = m'_1 + m'_2, e' = e'_1 + e'_2$ and $\mtyp = \tightN$.
    The following conditions hold:
    \begin{enumerate}
    \item[(a')]
      $\tmtwo \tov{\rulename}{\aset}{\sset}{\app} \tmtwo'$,
      with $\rulename \in \set{\ruledb, \rulelsv}$,
      by condition (a)
    \item[(b')]
      $\judg[m'_1,e'_1]{\tctx_1}{\tmtwo'}{\tightN}$, by premise of the judgement of the condition (b)
    \item[(c')]
      $\isAppr{\aset}{\tctx_1}$, since $\tctx = \tctx_1 + \tctx_2$ and by condition (c)
    \item[(d')]
      Here $\appflag = \app$, with $\mtyp = \tightN$.
    \end{enumerate}
    We can apply \ih on $\tmtwo'$, yielding 
    $\judg[m_1,e_1]{\tctx_1}{\tmtwo}{\tightN}$, where,
    if $\rulename = \ruledb$  then $(m_1,e_1) = (m'_1 + 1, e'_1)$, and
    if $\rulename = \rulelsv$ then $(m_1,e_1) = (m'_1, e'_1 + 1)$.
    We build the following derivation:
    \[
      \indrule{\ruleTypAppP}{
        \judg[m_1,e_1]{\tctx_1}{\tmtwo}{\tightN}
        \sep
        \judg[m'_2,e'_2]{\tctx_2}{\tmthree}{\tight}
      }{
        \judg[m_1 + m'_2, e_1 + e'_2]{\tctx_1 + \tctx_2}{\tmtwo \, \tmthree}{\tightN}
      }
    \]
    where 
    if $\rulename = \ruledb$ then $(m,e) = (m_1 + m'_2,e_1 + e'_2) = (m'_1 + 1 + m'_2, e'_1 + e'_2) = (m' + 1,e')$, and
    if $\rulename = \rulelsv$ then $(m,e) = (m_1 + m'_2,e_1 + e'_2) = (m'_1 + m'_2, e'_1 + 1 + e'_2) = (m',e' + 1)$,
    so we are done. 
  \item \ruleTypAppC. % Se puede decir que es análogo al caso de la regla \ruleTypAppP.
    Then 
    \[
      \indrule{\ruleTypAppC}{
        \judg[m'_1,e'_1]{\tctx_1}{\tmtwo'}{\mset{\optmtyptwo \to \mtyp}}
        \sep
        \optmtyptwo \mleq \mtyptwo
        \sep
        \judg[m'_2,e'_2]{\tctx_2}{\tmthree}{\mtyptwo}
      }{
        \judg[1 + m'_1 + m'_2, e'_1 + e'_2]{\tctx_1 + \tctx_2}{\tmtwo' \, \tmthree}{\mtyp}
      }
    \]
    where $\tctx = \tctx_1 + \tctx_2, m' = 1 + m'_1 + m'_2$ and $e' = e'_1 + e'_2$.
    The following conditions hold:
    \begin{enumerate}
    \item[(a')]
      $\tmtwo \tov{\rulename}{\aset}{\sset}{\app} \tmtwo'$,
      with $\rulename \in \set{\ruledb, \rulelsv}$,
      by condition (a)
    \item[(b')]
      $\judg[m'_1,e'_1]{\tctx_1}{\tmtwo'}{\mset{\optmtyptwo \to \mtyp}}$, 
      by premise of the judgement of the condition (b)
    \item[(c')]
      $\isAppr{\aset}{\tctx_1}$, since $\tctx = \tctx_1 + \tctx_2$ and by condition (c)
    \item[(d')]
      Here $\appflag = \app$, and the term is typed with the singleton $\mset{\optmtyptwo \to \mtyp}$.
    \end{enumerate}
    We can apply \ih on $\tmtwo'$, yielding 
    $\judg[m_1,e_1]{\tctx_1}{\tmtwo}{\mset{\optmtyptwo \to \mtyp}}$, where,
    if $\rulename = \ruledb$  then $(m_1,e_1) = (m'_1 + 1, e'_1)$, and
    if $\rulename = \rulelsv$ then $(m_1,e_1) = (m'_1, e'_1 + 1)$.
    We build the following derivation:
    \[
      \indrule{\ruleTypAppC}{
        \judg[m_1,e_1]{\tctx_1}{\tmtwo}{\mset{\optmtyptwo \to \mtyp}}
        \sep
        \optmtyptwo \mleq \mtyptwo
        \sep
        \judg[m'_2,e'_2]{\tctx_2}{\tmthree}{\mtyptwo}
      }{
        \judg[1 + m_1 + m'_2, e_1 + e'_2]{\tctx_1 + \tctx_2}{\tmtwo \, \tmthree}{\mtyp}
      }
    \]
    where 
    if $\rulename = \ruledb$ then $(m,e) = (1 + m_1 + m'_2,e_1 + e'_2) = (1 + m'_1 + 1 + m'_2, e'_1 + e'_2) = (m' + 1,e')$, and
    if $\rulename = \rulelsv$ then $(m,e) = (1 + m_1 + m'_2,e_1 + e'_2) = (1 + m'_1 + m'_2, e'_1 + 1 + e'_2) = (m',e' + 1)$,
    so we are done.
  \end{enumerate}
\item \ruleUAppR.
  The following conditions hold:
  \begin{enumerate}
  \item[(a)]
    \[
      \indrule{\ruleUAppR}{
        \tmtwo \in \Struct{\sset}
        \sep
        \tmthree \tov{\rulename}{\aset}{\sset}{\nonapp} \tmthree'
      }{
        \tmtwo \, \tmthree \tov{\rulename}{\aset}{\sset}{\appflag} \tmtwo \, \tmthree'
      }
    \]
    with $\rulename \in \set{\ruledb, \rulelsv}$
  \item[(b)]
    $\judg[m',e']{\tctx}{\tmtwo \, \tmthree'}{\mtyp}$
  \item[(c)]
    $\isAppr{\aset}{\tctx}$
  \item[(d)]
    If $\appflag = \app$ then either $\mtyp = \tightN$ or $\mtyp$ is of the form $\mset{\typ}$.
  \end{enumerate}
  Note that $\inv{\aset}{\sset}{\tmtwo \, \tmthree}$
  implies in particular $\inv{\aset}{\sset}{\tmthree}$.
  The judgement of the condition (b) can be derived either by rule $\ruleTypAppP$ or by rule \ruleTypAppC:
  \begin{enumerate}
  \item \ruleTypAppP.
    Then 
    \[
      \indrule{\ruleTypAppP}{
        \judg[m'_1,e'_1]{\tctx_1}{\tmtwo}{\tightN}
        \sep
        \judg[m'_2,e'_2]{\tctx_2}{\tmthree'}{\tight}
      }{
        \judg[m'_1 + m'_2, e'_1 + e'_2]{\tctx_1 + \tctx_2}{\tmtwo \, \tmthree'}{\tightN}
      }
    \]
    where $\tctx = \tctx_1 + \tctx_2, m' = m'_1 + m'_2, e' = e'_1 + e'_2$ and $\mtyp = \tightN$.
    The following conditions hold:
    \begin{enumerate}
    \item[(a')]
      $\tmthree \tov{\rulename}{\aset}{\sset}{\nonapp} \tmthree'$,
      with $\rulename \in \set{\ruledb, \rulelsv}$,
      by condition (a)
    \item[(b')]
      $\judg[m'_2,e'_2]{\tctx_2}{\tmthree'}{\tight}$, by premise of the judgement of the condition (b)
    \item[(c')]
      $\isAppr{\aset}{\tctx_2}$, since $\tctx = \tctx_1 + \tctx_2$ and by condition (c)
    \item[(d')]
      Here $\appflag = \nonapp$.
    \end{enumerate}
    We can apply \ih on $\tmthree'$, yielding 
    $\judg[m_2,e_2]{\tctx_2}{\tmthree}{\tight}$, where,
    if $\rulename = \ruledb$  then $(m_2,e_2) = (m'_2 + 1,e'_2)$, and
    if $\rulename = \rulelsv$ then $(m_2,e_2) = (m'_2,e'_2 + 1)$.
    We build the following derivation:
    \[
      \indrule{\ruleTypAppP}{
        \judg[m'_1,e'_1]{\tctx_1}{\tmtwo}{\tightN}
        \sep
        \judg[m_2,e_2]{\tctx_2}{\tmthree}{\tight}
      }{
        \judg[m'_1 + m_2, e'_1 + e_2]{\tctx_1 + \tctx_2}{\tmtwo \, \tmthree}{\tightN}
      }
    \]
    where 
    if $\rulename = \ruledb$  then $(m,e) = (m'_1 + m_2, e'_1 + e_2) = (m'_1 + m'_2 + 1, e'_1 + e'_2) = (m' + 1,e')$, and
    if $\rulename = \rulelsv$ then $(m,e) = (m'_1 + m_2, e'_1 + e_2) = (m'_1 + m'_2, e'_1 + e'_2 + 1) = (m',e' + 1)$,
    so we are done.
  \item \ruleTypAppC. %Se puede decir que es análogo al caso de la regla \ruleTypAppP.
    Then 
    \[
      \indrule{\ruleTypAppC}{
        \judg[m'_1,e'_1]{\tctx_1}{\tmtwo}{\mset{\optmtyptwo \to \mtyp}}
        \sep
        \optmtyptwo \mleq \mtyptwo
        \sep
        \judg[m'_2,e'_2]{\tctx_2}{\tmthree'}{\mtyptwo}
      }{
        \judg[1 + m'_1 + m'_2, e'_1 + e'_2]{\tctx_1 + \tctx_2}{\tmtwo \, \tmthree'}{\mtyp}
      }
    \]
    where $\tctx = \tctx_1 + \tctx_2, m' = 1 + m'_1 + m'_2$ and $e' = e'_1 + e'_2$.
    The following conditions hold:
    \begin{enumerate}
    \item[(a')]
      $\tmthree \tov{\rulename}{\aset}{\sset}{\nonapp} \tmthree'$,
      with $\rulename \in \set{\ruledb, \rulelsv}$,
      by condition (a)
    \item[(b')]
      $\judg[m'_2,e'_2]{\tctx_2}{\tmthree'}{\mtyptwo}$, by premise of the judgement of the condition (b)
    \item[(c')]
      $\isAppr{\aset}{\tctx_2}$, since $\tctx = \tctx_1 + \tctx_2$ and by condition (c)
    \item[(d')]
      Here $\appflag = \nonapp$.
    \end{enumerate}
    We can apply \ih on $\tmthree'$, yielding 
    $\judg[m_2,e_2]{\tctx_2}{\tmthree}{\mtyptwo}$, where,
    if $\rulename = \ruledb$  then $(m_2,e_2) = (m'_2 + 1, e'_2)$, and
    if $\rulename = \rulelsv$ then $(m_2,e_2) = (m'_2, e'_2 + 1)$.
    We build the following derivation:
    \[
      \indrule{\ruleTypAppC}{
        \judg[m'_1,e'_1]{\tctx_1}{\tmtwo}{\mset{\optmtyptwo \to \mtyp}}
        \sep
        \optmtyptwo \mleq \mtyptwo
        \sep
        \judg[m_2,e_2]{\tctx_2}{\tmthree}{\mtyptwo}
      }{
        \judg[1 + m'_1 + m_2, e'_1 + e_2]{\tctx_1 + \tctx_2}{\tmtwo \, \tmthree}{\mtyp}
      }
    \]
    where 
    if $\rulename = \ruledb$  then $(m,e) = (1 + m'_1 + m_2,e'_1 + e_2) = (1 + m'_1 + m'_2 + 1, e'_1 + e'_2) = (m' + 1,e')$, and
    if $\rulename = \rulelsv$ then $(m,e) = (1 + m'_1 + m_2,e'_1 + e_2) = (1 + m'_1 + m'_2, e'_1 + e'_2 + 1) = (m',e' + 1)$,
    so we are done.
  \end{enumerate}
\item \ruleUEsR. % Se puede decir que es análogo al caso \ruleUAppR.
  The following conditions hold:
  \begin{enumerate}
  \item[(a)]
    \[
      \indrule{\ruleUEsR}{
        \tmthree \tov{\rulename}{\aset}{\sset}{\nonapp} \tmthree'
      }{
        \tmtwo\esub{\var}{\tmthree} \tov{\rulename}{\aset}{\sset}{\appflag} \tmtwo\esub{\var}{\tmthree'}
      }
    \]
    with $\rulename \in \set{\ruledb, \rulelsv}$
  \item[(b)]
    $\judg[m',e']{\tctx}{\tmtwo\esub{\var}{\tmthree'}}{\mtyp}$
  \item[(c)]
    $\isAppr{\aset}{\tctx}$
  \item[(d)]
    If $\appflag = \app$ then either $\mtyp = \tightN$ or $\mtyp$ is of the form $\mset{\typ}$.
  \end{enumerate}
  The judgement of the condition (b) can be derived only by rule \ruleTypES
  \[
    \indrule{\ruleTypES}{
      \judg[m'_1,e'_1]{\tctx_1; \var : \optmtyptwo}{\tmtwo}{\mtyp}
      \sep
      \optmtyptwo \mleq \mtyptwo
      \sep
      \judg[m'_2,e'_2]{\tctx_2}{\tmthree'}{\mtyptwo}
    }{
      \judg[m'_1 + m'_2, e'_1 + e'_2]{\tctx_1 + \tctx_2}{\tmtwo\esub{\var}{\tmthree'}}{\mtyp}
    }
  \]
  where $\tctx = \tctx_1 + \tctx_2, m' = m'_1 + m'_2$ and $e' = e'_1 + e'_2$.
  Note that $\inv{\aset}{\sset}{\tmtwo\esub{\var}{\tmthree}}$
  implies in particular $\inv{\aset}{\sset}{\tmthree}$.
    The following conditions hold:
  \begin{enumerate}
  \item[(a')]
    $\tmthree \tov{\rulename}{\aset}{\sset}{\nonapp} \tmthree'$,
    with $\rulename \in \set{\ruledb, \rulelsv}$,
    by condition (a)
  \item[(b')]
    $\judg[m'_2,e'_2]{\tctx_2}{\tmthree'}{\mtyptwo}$, by premise of the judgement of the condition (b)
  \item[(c')]
    $\isAppr{\aset}{\tctx_2}$, since $\tctx = \tctx_1 + \tctx_2$ and by condition (c)
  \item[(d')]
    Here $\appflag = \nonapp$.
  \end{enumerate}
  We can apply \ih on $\tmthree'$, yielding 
  $\judg[m_2,e_2]{\tctx_2}{\tmthree}{\mtyptwo}$, where,
  if $\rulename = \ruledb$  then $(m_2,e_2) = (m'_2 + 1, e'_2)$, and
  if $\rulename = \rulelsv$ then $(m_2,e_2) = (m'_2, e'_2 + 1)$.
  We build the following derivation:
  \[
    \indrule{\ruleTypES}{
      \judg[m'_1,e'_1]{\tctx_1; \var : \optmtyptwo}{\tmtwo}{\mtyp}
      \sep
      \optmtyptwo \mleq \mtyptwo
      \sep
      \judg[m_2,e_2]{\tctx_2}{\tmthree}{\mtyptwo}
    }{
      \judg[m'_1 + m_2, e'_1 + e_2]{\tctx_1 + \tctx_2}{\tmtwo\esub{\var}{\tmthree}}{\mtyp}
    }
  \]
  where 
  if $\rulename = \ruledb$  then $(m,e) = (m'_1 + m_2,e'_1 + e_2) = (m'_1 + m'_2 + 1, e'_1 + e'_2) = (m' + 1,e')$, and
  if $\rulename = \rulelsv$ then $(m,e) = (m'_1 + m_2,e'_1 + e_2) = (m'_1 + m'_2, e'_1 + e'_2 + 1) = (m',e' + 1)$,
  so we are done.
\item \ruleUEsLAbs.
  The following conditions hold:
  \begin{enumerate}
  \item[(a)]
    \[
      \indrule{\ruleUEsLAbs}{
        \tmtwo \tov{\rulename}{\aset \cup \set{\var}}{\sset}{\appflag} \tmtwo'
        \sep
        \tmthree \in \HAbs{\aset}
        \sep
        \var \notin \aset \cup \sset
        \sep
        \var \notin \fv{\rulename}
      }{
        \tmtwo\esub{\var}{\tmthree} \tov{\rulename}{\aset}{\sset}{\appflag} \tmtwo'\esub{\var}{\tmthree}
      }
    \]
    with $\rulename \in \set{\ruledb, \rulelsv}$
  \item[(b)]
    $\judg[m',e']{\tctx}{\tmtwo'\esub{\var}{\tmthree}}{\mtyp}$
  \item[(c)]
    $\isAppr{\aset}{\tctx}$
  \item[(d)]
    If $\appflag = \app$ then either $\mtyp = \tightN$ or $\mtyp$ is of the form $\mset{\typ}$.
  \end{enumerate}
  Note that $\inv{\aset}{\sset}{\tmtwo\esub{\var}{\tmthree}}$
  implies in particular $\inv{\aset \cup \set{\var}}{\sset}{\tmtwo}$.
    The judgement of the condition (b) can be derived only by rule \ruleTypES:
  \[
    \indrule{\ruleTypES}{
      \judg[m'_1,e'_1]{\tctx_1; \var : \optmtyptwo}{\tmtwo'}{\mtyp}
      \sep
      \optmtyptwo \mleq \mtyptwo
      \sep
      \judg[m'_2,e'_2]{\tctx_2}{\tmthree}{\mtyptwo}
    }{
      \judg[m'_1 + m'_2, e'_1 + e'_2]{\tctx_1 + \tctx_2}{\tmtwo'\esub{\var}{\tmthree}}{\mtyp}
    }
  \]
  where $\tctx = \tctx_1 + \tctx_2, m' = m'_1 + m'_2$ and $e' = e'_1 + e'_2$.
  The following conditions hold:
  \begin{enumerate}
  \item[(a')]
    $\tmtwo \tov{\rulename}{\aset \cup \set{\var}}{\sset}{\appflag} \tmtwo'$,
    with $\rulename \in \set{\ruledb, \rulelsv}$,
    by condition (a)
  \item[(b')]
    $\judg[m'_1,e'_1]{\tctx_1; \var : \optmtyptwo}{\tmtwo'}{\mtyp}$, 
    by premise of the judgement of the condition (b)
  \item[(c')]
    $\isAppr{\aset \cup \set{\var}}{\tctx_1; \var : \optmtyptwo}$, since 
    (1) $\isAppr{\aset \cup \set{\var}}{\tctx_1}$
    because $\tctx = \tctx_1 + \tctx_2$, and condition (c) and 
    \cref{rem:isAppropriate}; 
    (2) $\isAppr{\set{\var}}{\var : \optmtyptwo}$ 
    since $\inv{\aset}{\sset}{\tmtwo\esub{\var}{\tmthree}}$ implies 
    in particular $\inv{\aset}{\sset}{\tmthree}$
    and thus $\tmthree \in \HAbs{\aset}$ allows to obtain by \cref{lem:types_of_hereditary_abstractions} that
    $\mtyptwo \neq \tightN$. 
    Therefore, $\optmtyptwo \mleq \mtyptwo$ implies $\optmtyptwo \neq \tightN$.
    By \cref{rem:isAppropriate} we are done.
  \item[(d')]
    If $\appflag = \app$ then either $\mtyp = \tightN$ or $\mtyp$ is of the form $\mset{\typ}$, by condition (d).
  \end{enumerate}
  We can apply \ih on $\tmtwo'$, yielding 
  $\judg[m_1,e_1]{\tctx_1; \var : \optmtyptwo}{\tmtwo}{\mtyp}$, where,
  if $\rulename = \ruledb$  then $(m_1,e_1) = (m'_1 + 1, e'_1)$, and
  if $\rulename = \rulelsv$ then $(m_1,e_1) = (m'_1, e'_1 + 1)$.
  We build the following derivation:
  \[
    \indrule{\ruleTypES}{
      \judg[m_1,e_1]{\tctx_1; \var : \optmtyptwo}{\tmtwo}{\mtyp}
      \sep
      \optmtyptwo \mleq \mtyptwo
      \sep
      \judg[m'_2,e'_2]{\tctx_2}{\tmthree}{\mtyptwo}
    }{
      \judg[m_1 + m'_2, e_1 + e'_2]{\tctx_1 + \tctx_2}{\tmtwo\esub{\var}{\tmthree}}{\mtyp}
    }
  \]
  where 
  if $\rulename = \ruledb$  then $(m,e) = (m_1 + m'_2,e_1 + e'_2) = (m'_1 + 1 + m'_2, e'_1 + e'_2) = (m' + 1,e')$, and
  if $\rulename = \rulelsv$ then $(m,e) = (m_1 + m'_2,e_1 + e'_2) = (m'_1 + m'_2, e'_1 + 1 + e'_2) = (m',e' + 1)$,
  so we are done.
\item \ruleUEsLStruct.
  The following conditions hold:
  \begin{enumerate}
  \item[(a)]
    \[
      \indrule{\ruleUEsLStruct}{
        \tmtwo \tov{\rulename}{\aset}{\sset \cup \set{\var}}{\appflag} \tmtwo'
        \sep
        \tmthree \in \Struct{\sset}
        \sep
        \var \notin \aset \cup \sset
        \sep
        \var \notin \fv{\rulename}
      }{
        \tmtwo\esub{\var}{\tmthree} \tov{\rulename}{\aset}{\sset}{\appflag} \tmtwo'\esub{\var}{\tmthree}
      }
    \]
    with $\rulename \in \set{\ruledb, \rulelsv}$
  \item[(b)]
    $\judg[m',e']{\tctx}{\tmtwo'\esub{\var}{\tmthree}}{\mtyp}$
  \item[(c)]
    $\isAppr{\aset}{\tctx}$
  \item[(d)]
    If $\appflag = \app$ then either $\mtyp = \tightN$ or $\mtyp$ is of the form $\mset{\typ}$.
  \end{enumerate}
  The judgement of the condition (b) can be derived only by rule \ruleTypES
  \[
    \indrule{\ruleTypES}{
      \judg[m'_1,e'_1]{\tctx_1; \var : \optmtyptwo}{\tmtwo'}{\mtyp}
      \sep
      \optmtyptwo \mleq \mtyptwo
      \sep
      \judg[m'_2,e'_2]{\tctx_2}{\tmthree}{\mtyptwo}
    }{
      \judg[m'_1 + m'_2, e'_1 + e'_2]{\tctx_1 + \tctx_2}{\tmtwo'\esub{\var}{\tmthree}}{\mtyp}
    }
  \]
  where $\tctx = \tctx_1 + \tctx_2, m' = m'_1 + m'_2$ and $e' = e'_1 + e'_2$.
  Note that $\inv{\aset}{\sset}{\tmtwo\esub{\var}{\tmthree}}$
  implies in particular $\inv{\aset}{\sset \cup \set{\var}}{\tmtwo}$.
  
  The following conditions hold:
  \begin{enumerate}
  \item[(a')]
    $\tmtwo \tov{\rulename}{\aset}{\sset \cup \set{\var}}{\appflag} \tmtwo'$,
    with $\rulename \in \set{\ruledb, \rulelsv}$,
    by condition (a)
  \item[(b')]
    $\judg[m'_1,e'_1]{\tctx_1; \var : \optmtyptwo}{\tmtwo'}{\mtyp}$, 
    by premise of the judgement of the condition (b)
  \item[(c')]
    $\isAppr{\aset}{\tctx_1; \var : \optmtyptwo}$, since 
    (1) $\isAppr{\aset \cup \set{\var}}{\tctx_1}$ 
    because $\tctx = \tctx_1 + \tctx_2$, and condition (c) and \cref{rem:isAppropriate}; 
    (2) $\var \notin \aset$  by condition (a)
  \item[(d')]
    If $\appflag = \app$ then either $\mtyp = \tightN$ or $\mtyp$ is of the form $\mset{\typ}$, by condition (d).
  \end{enumerate}
  We can apply \ih on $\tmtwo'$, yielding 
  $\judg[m_1,e_1]{\tctx_1; \var : \optmtyptwo}{\tmtwo}{\mtyp}$, where,
  if $\rulename = \ruledb$ then $(m_1,e_1) = (m'_1 + 1, e'_1)$, and
  if $\rulename = \rulelsv$ then $(m_1,e_1) = (m'_1, e'_1 + 1)$.
  We build the following derivation:
  \[
    \indrule{\ruleTypES}{
      \judg[m_1,e_1]{\tctx_1; \var : \optmtyptwo}{\tmtwo}{\mtyp}
      \sep
      \optmtyptwo \mleq \mtyptwo
      \sep
      \judg[m'_2,e'_2]{\tctx_2}{\tmthree}{\mtyptwo}
    }{
      \judg[m_1 + m'_2, e_1 + e'_2]{\tctx_1 + \tctx_2}{\tmtwo\esub{\var}{\tmthree}}{\mtyp}
    }
  \]
  where 
  if $\rulename = \ruledb$  then $(m,e) = (m_1 + m'_2,e_1 + e'_2) = (m'_1 + 1 + m'_2, e'_1 + e'_2) = (m' + 1,e')$, and
  if $\rulename = \rulelsv$ then $(m,e) = (m_1 + m'_2,e_1 + e'_2) = (m'_1 + m'_2, e'_1 + 1 + e'_2) = (m',e' + 1)$,
  so we are done.
}
\end{enumerate}
\end{proof}

\completenesstyping*

\begin{proof}
By induction on $n$:
\begin{enumerate}
\item $n = 0$.
  Then $\tm \in \NF{\emptyset}{\fv{\tm}}{\nonapp}$,
  and $m = e = 0$ by definition of $m$ and $e$.
  By \cref{prop:nfs_are_tight_typable}, 
  there exists a tight type $\tight$ such that $\judg[0,0]{\TEnv{\emptyset}{\fv{\tm}}{\tm}}{\tm}{\tight}$,
  with $\TEnv{\emptyset}{\fv{\tm}}{\tm}$ a tight environment,
  so we are done.
\item $n = n' + 1$. 
  Then $\tm \notin \NF{\emptyset}{\fv{\tm}}{\nonapp}$.
  The reduction sequence is then of the form:
  \[
    \tm \tov{\rulename_1}{\emptyset}{\fv{\tm}}{\nonapp} \tm' 
        \tov{\rulename_2}{\emptyset}{\fv{\tm}}{\nonapp} \hdots \tov{\rulename_{n'+1}}{\emptyset}{\fv{\tm}}{\nonapp} \tmtwo
  \]
  where $n' = m' + e'$ and
  \[
    m' = \#\set{i \ST 2 \leq i \leq n', \rulename_i = \ruledb}
    \HS
    e' = \#\set{i \ST 2 \leq i \leq n', \rulename_i = \rulelsv}
  \]
  Since $\fv{\tm} = \fv{\tm'}$, since the reduction is non-erasing when $\rulename_1 \in \set{\ruledb, \rulelsv}$,
  then we can apply \ih, yielding that there exists a tight environment $\tctx$ 
  and a tight type $\tight$ such that
  $\judg[m',e']{\tctx}{\tm'}{\tight}$.
  By \nameref{prop:subject_expansion},
  it holds that $\judg[m,e]{\tctx}{\tm}{\tight}$, where
  if $\rulename = \ruledb$  then $(m,e) = (m' + 1, e')$ and
  if $\rulename = \rulelsv$ then $(m,e) = (m', e' + 1)$,
  so we are done.
\end{enumerate}
\end{proof}

\end{document}